\documentclass[fleqn,notitlepage,thmsb,12pt]{article}%
\usepackage{amssymb}
\usepackage{amsthm}
\usepackage{graphicx}
\usepackage{amsmath}
\usepackage{xr-hyper}
\usepackage{hyperref}
\usepackage{amsfonts}
\usepackage{color}%
\usepackage{setspace,booktabs,makecell,array,graphicx,siunitx}
\usepackage{threeparttable,longtable,tabularx,ragged2e}
\usepackage[title]{appendix}
\providecommand{\U}[1]{\protect\rule{.1in}{.1in}}

\newtheorem{theorem}{Theorem}

\newtheorem{condition}{Condition}

\newtheorem{corollary}{Corollary}

\newtheorem{definition}{Definition}

\newtheorem{lemma}{Lemma}

\newtheorem{proposition}{Proposition}
\newtheorem{remark}{Remark}

\allowdisplaybreaks

\begin{document}

\title{Efficient Bias Correction for Cross-section and Panel Data\footnote{T. Rothenberg provided helpful comments.}}
\author{Jinyong Hahn\thanks{Department of Economics, UCLA, Los Angeles, CA 90095-1477. Email: hahn@econ.ucla.edu}\\UCLA
\and David W. Hughes\thanks{Boston College, Department of Economics, Chestnut Hill, MA 02467. Email: dw.hughes@bc.edu.}\\Boston College
\and Guido Kuersteiner\thanks{University of Maryland, Department of Economics, College Park, MD 20742-7211. Email:
kuersteiner@econ.umd.edu. Financial support from NSF grants SES-0095132 and
SES-0523186 is gratefully acknowledged. }\\University of Maryland
\and Whitney K. Newey\thanks{MIT, Department of Economics, Cambridge, MA 02139. Email: wnewey@mit.edu. Financial support from NSF Grant 1757140 is gratefully acknowledged.}\\MIT and NBER}
\maketitle

\begin{abstract}
Bias correction can often improve the finite sample performance of estimators. We show that the choice of bias correction method has no effect on the higher-order variance of semiparametrically efficient parametric estimators, so long as the estimate of the bias is asymptotically linear. It is also shown that bootstrap, jackknife, and analytical bias estimates are asymptotically linear for  estimators with higher-order expansions of a standard form. In particular, we find that for a variety of estimators the straightforward bootstrap bias correction gives the same higher-order variance as more complicated analytical or jackknife bias corrections. In contrast, bias corrections that do not estimate the bias at the parametric rate, such as the split-sample jackknife, result in larger higher-order variances in the i.i.d.\ setting we focus on. For both a cross-sectional MLE and a panel model with individual fixed effects, we show that the split-sample jackknife has a higher-order variance term that is twice as large as that of the `leave-one-out' jackknife.

\end{abstract}

\newpage

\section{Introduction}

Asymptotic bias corrections can be useful for centering estimators nearer to the
truth. One approach is to use analytical corrections such as the standard
textbook expansion for functions of sample means and the more complicated
formulas required for other estimators. Alternatively, we may use jackknife and
bootstrap bias corrections. To help choose a bias correction method it would
be useful to know which, if any, is preferable on asymptotic efficiency
grounds. Although the bias correction does not affect the first-order asymptotic
variance, it can affect the higher-order variance.  We show that the method of bias correction does not affect the higher-order variance of any parametric estimator that is efficient in a semiparametric model,  
as long as the bias estimator is asymptotically linear. Thus, one can choose a bias correction
for an efficient estimator based on computational convenience, or some other
criteria, without affecting its higher-order efficiency. We give a formal
expansion showing this property for a parametric estimator in a general semiparametric model, i.e. a model with a parametric component in which some other components are left unspecified. We also prove 
that the bootstrap, jackknife and analytical bias estimates are asymptotically linear, when the estimator of the parameters of interest has a standard form of stochastic expansion (which is known to exist for a large class of models). Derivations in the case of the MLE, show that the jackknife, bootstrap, and one type of analytical bias correction deliver estimators that have identical stochastic expansions to third order, and so have an even stronger equivalence property.

There are many implications of this higher-order efficiency result. One is that bias corrections that are not asymptotically linear may not have the same higher-order variance as those that are. 
For example, split-sample jackknife bias corrections are not asymptotically linear in cross-section or panel data and have a larger higher-order variance than other bias-corrected estimators. 
We find that the higher-order variance term is twice the size of that for the leave-one-out jackknife bias correction. On the other hand, the split-sample jackknife is useful in time series or panel data when the observations are not independent over time because the leave-one-out panel jackknife does not work in this case.  

Another implication of the result is that it allows researchers to choose the bias correction method that is computationally convenient. 
For example, Newey and Smith (2004) showed that the empirical likelihood estimator is higher-order efficient in moment condition models when certain analytical bias corrections are used. 
Asymptotic linearity of the bootstrap means that calculation of the bias formula can be avoided by using the bootstrap bias correction instead. 
As another example, Cattaneo, Jansson, and Ma (2019)
showed that a jackknife bias correction can be important when a first step regression with many regressors is plugged into a second step regression. Although our current results do not include asymptotics in which the number of regressors increases with the sample size, we conjecture that the bootstrap bias correction has a similar properties as the jackknife bias correction when the number of regressors increases slowly enough.

The higher-order variance concept that we consider is the $O(
n^{-1})$ variance of a third-order stochastic expansion of the
estimator. Its use for comparison of estimators was pioneered by Nagar (1959).
As shown in Pfanzagl and Wefelmeyer (1978) and Ghosh, Sinea, and Wieand (1980),
and discussed in Rothenberg (1984), under appropriate regularity conditions
rankings based on this higher-order variance correspond to rankings based on
the variance of an Edgeworth approximation. Thus, the bias and variance of leading terms in a stochastic expansion are also the leading
terms of an expansion of the bias and variance of an approximating
distribution. Furthermore, as noted by Rothenberg (1984), Akahira and
Takeuchi (1981) have shown that all well-behaved asymptotically efficient
estimators of a parametric likelihood model necessarily have the same skewness
and kurtosis to an order $n^{-1}$ approximation and that to compare the 
dispersion of second-order approximations to the distribution of efficient estimators it suffices to compare their higher-order variances.  
This motivates our focus here on the higher-order variance of bias-corrected estimators, which serves to quantify their higher-order  
efficiency. 
In having this focus we follow much of the more recent literature, such as Rilstone, Srivastava, and Ullah (1996) and Newey and Smith (2004).

We also derive asymptotic higher-order variance expressions for panel data 
models with unobserved, individual fixed 
effects as well as for cross-section models. These estimators are known to suffer from asymptotic bias under asymptotic sequences
in which $n$ and $T$ grow at the same rate. Two common methods of bias correction in this setting are the
`leave-one-out' panel jackknife (Hahn and Newey, 2004), and the split-sample jackknife (Dhaene and Jochmans, 2015).
Both deliver estimates that are asymptotically normal and centered at the truth when
$n$ and $T$ grow at the same rate, with equal first-order (asymptotic) variances.
Our analysis makes it possible to compare the two bias corrections in terms of their higher-order variance. We find that with i.i.d.\ data the split-sample correction has a higher-order variance that is twice the size of the `leave-one-out' jackknife. Although we focus on the maximum likelihood setting for panel data, the results are applicable to a broader set of moment condition estimators under suitable assumptions on the moment functions. 
Numerical comparisons in recent papers confirm that the difference in higher-order variance can be meaningful in practice, with the split-sample jackknife having larger dispersion and lower coverage than analytical or leave-one-out jackknife corrections in a variety of settings; see for example, 
Alexander and Breunig (2016), Fernandez-Val and Weidner (2018), and Czarnowske and Stammann (2019). Our own simulations of a panel probit model with one
common parameter and individual fixed effects also support this result. 
This comparison is also true of estimates of a marginal effect parameter.

\subsection{Related literature}

Higher-order efficiency of the MLE was analyzed by Pfanzagl and Wefelmeyer (1978) in terms of risk functions or Akahira and Takeuchi (1981) in terms of concentration probabilities. Ghosh (1994), and Taniguchi and Kakizawa (2000) for time series models, contain surveys of this literature. A common theme is that a higher-order bias-corrected version of the MLE is higher-order efficient in the case of higher-order squared risk. Similar results obtain for median bias-corrected MLE's in the case of concentration probabilities. The bias correction in this literature is typically of a known parametric form that only depends on the estimated parameters. A plug-in estimator is then a regular estimator for the bias term. Amari (1982) obtains similar results for curved exponential families using differential geometry that characterizes the MLE in terms of tangent spaces. Akahira (1983, 1989) shows that when instead of using parametric bias correction one relies on the jackknife the same higher-order efficiency results remain true for the jackknifed MLE. We add to this literature by analyzing the effects of bias correction for first-order, but not necessarily higher-order, semiparametrically efficient estimators.  We demonstrate that any bias-corrected semiparametrically efficient estimator using an asymptotically linear bias estimator has a higher-order variance that does not depend on the nature of the bias estimate. This result applies to cases where the higher-order bias may not be known in closed form. We show that efficient bias correction may be based on sample averages, bootstrap or jackknife methods as long as the bias estimates are asymptotically linear.

The use of a jackknife bias estimator goes back to Quenouille (1949, 1956) and
Tukey (1958). Bootstrap bias estimation was discussed by Parr (1983), Shao
(1988a,b), Hall (1992), and Horowitz (1998) in the context of nonlinear
transformations of OLS estimators of linear models and nonlinear functions of
the mean. Akahira (1983) showed that the jackknife bias-corrected maximum
likelihood estimator is higher-order efficient. Our work extends the
literature on the bootstrap and jackknife bias-corrected estimators by
analyzing any semiparametrically efficient parametric estimator rather than nonlinear transformations
of linear estimators as in Shao (1988a,b).  

In the panel data literature, methods to control for unobserved heterogeneity are well established 
(early literature includes Rasch (1960, 1961) and Andersen (1970)); see for example Chamberlain (1984), Arellano and Honor\'{e} (2001), and Arellano and Hahn (2010) for reviews. Because of the incidental parameters problem, the 
best that can be achieved in a fixed-$T$ setting is partial identification in general; 
this is especially true for policy relevant parameters such as average marginal
effects (see Chernozhukov, Newey, Hahn and Fern\'{a}ndez-Val (2013)). Even
under sequences in which $T$ grows at the same rate as $n$, fixed effects estimators may be
asymptotically biased, as discussed in Hahn and Kuersteiner (2002) and Hahn
and Newey (2004). Given the typical size of panel data sets, in which $n$ is
much larger than $T$, it is desirable to find estimators that have biases of
order $O(T^{-2})$ or smaller, rather than the typical $O(T^{-1})$ of fixed effects estimators.

For a static model, Hahn and Newey (2004) show that a `leave-one-out'
jackknife estimator is asymptotically normal and centered at the truth when
$n$ and $T$ grow at the same rate. Other styles of jackknife bias correction
are also possible. For example, Dhaene and Jochmans (2015) suggest a
split-sample bias correction that, in its simplest form, is constructed by
splitting the sample into two half-panels of length $T/2$. It should be understood that the split-sample jackknife provides valid bias corrections with autocorrelated data where the `leave-one out' jackknife does not. Thus the split-sample jackknife is preferred to the leave-one-out jackknife with autocorrelated data. Our results show that the leave-one-out jackknife is preferred to the split-sample jackknife in i.i.d.\ data in the sense that the higher-order variance is smaller and small sample performance is better.  

The remainder of the paper is set out as follows. In Section \ref{sec:semipar_result} we discuss the higher-order bias and variance of parametric estimators, and provide our main results on the higher-order efficiency of bias corrections for semiparametric efficient estimators and asymptotic linearity of analytical, bootstrap and jackknife bias corrections. In Section \ref{sec:mle_crosssec}, we provide expressions for the higher-order expansions and variances of various bias-corrections in a cross-sectional MLE. We extend the results to panel settings
in Section \ref{section-panel-model}, by deriving the higher-order variances of the leave-one-out and split-sample jackknifes in a model with individual fixed effects.
Section \ref{section-MC} provides Monte Carlo evidence to support the theory in the panel setting. Section \ref{section-summary} concludes.




\section{Higher-order efficiency of bias-corrected estimators} \label{sec:semipar_result}

\subsection{Higher-order bias and variance}
We begin with a discussion of higher-order expansions for parametric estimators, and make precise our definition of higher-order variance, following closely the exposition in Rothenberg (1984). We focus on semiparametric models in which the
data $\left\{  Z_{i}\right\}  _{i=1}^{n}$ are i.i.d.\ with some
distribution $F_{0}$ contained in a class of distributions that is the model.\footnote{By `semiparametric model', we refer to a model which has a parametric component, but leaves the functional form of some other components unspecified. See Newey (1990) for further discussion.} Let $\widehat{\theta}$ denote an estimator of a parameter
$\theta_{0}.$ When $\widehat{\theta}$ is an estimator based on parametric moment conditions,
such as GMM, and the moment conditions are smooth enough in the parameter
$\theta$ and possibly other parameters, there will be a stochastic expansion of
the form%
\begin{align} \label{eq:expan}
\sqrt{n}\left(  \widehat{\theta}-\theta_{0}\right)   &  =A_{1}+\frac{A_{2}}%
{\sqrt{n}}+\frac{A_{3}}{n}+o_{p}\left( n^{-1} \right)  ,
\\
A_{1}  &  =\frac{1}{\sqrt{n}}\sum_{i=1}^{n}\psi(Z_{i}),\quad E[\psi
(Z_{i})]=0,\nonumber
\end{align}
where $A_{2}$ and $A_{3}$ are second and third-order products of sample averages of mean-zero random variables, each multiplied by $\sqrt{n}$; see
for example, Rilstone, Srivastava, and Ullah (1996). The terms $A_{1},$ $A_{2},$
and $A_{3}$ are bounded in probability so that equation (\ref{eq:expan})
is an expansion where the stochastic order of terms is smaller as one moves to
the right in the expansion.

We define the higher-order bias and variance of estimators using the first and second moments of the leading three terms in this expansion. We also compare higher-order efficiency of bias-corrected estimators by comparing their higher-order variance, as discussed in the introduction. The higher-order bias of $\sqrt{n}(\widehat{\theta}-$ $\theta_{0})$ is given by
$E[A_{2}]/\sqrt{n}.$ This follows from the fact that $E[A_{1}]=0$, while the expectation of $A_{3}/n$
is generally of smaller order than $E[A_{2}]/\sqrt{n}$. Dividing through
by $\sqrt{n},$ the higher-order bias of $\widehat{\theta}$ is%
\[
Bias(\widehat{\theta})\approx\frac{b_{0}}{n},\text{ }b_{0}=E[A_{2}].
\]
In general $b_{0}=E[B(Z)],$ where the function $B(z)$ captures the `own observation' term in $A_{2},$ i.e.\ where the observation indices coincide in the
product of sample means that make up $A_{2}$. The term $B(z)$ may depend on the distribution of the
data through the parameter $\theta$ or in other ways. We let this dependence
be implicit for notational convenience. 

Similarly, the higher-order variance of $\sqrt{n}(\widehat{\theta}-\theta_{0})$ can be obtained from the variance of the first three expansion terms, up to order $1/n.$ This gives%
\begin{align*}
&Var\big(\sqrt{n}(\widehat{\theta}-\theta_{0})\big) \approx Var(A_1) + \frac{\upsilon}{n} \\
&\upsilon:=Var(A_{2})+2\sqrt{n}E[A_{1}%
A_{2}]+2E[A_{1}A_{3}].
\end{align*}
where $Var(A_1)$ is the asymptotic (first-order) variance of the estimator and we refer to $\upsilon$ as the higher-order variance.

\medskip
{\bf Example 1: Normally distributed data.}
In order to demonstrate ideas, we will follow a simple example. Suppose that we observe a sample of observations $Z_{i}\sim N\left(
\sqrt{\theta},1\right) $. The MLE provides an efficient estimate for $\theta$, and is given by
$\widehat{\theta}=\left(  \frac{1}{n}\sum_{i=1}^{n}Z_{i}\right)  ^{2}$. It can be shown that the asymptotic expansion for this estimator is of the form in (\ref{eq:expan}), with
\begin{align} \label{eq:exp_example}
    A_1 = 2\sqrt{\theta} \frac{1}{\sqrt{n}}\sum_i( Z_i - \sqrt{\theta}),\qquad
    A_2 = \big(\frac{1}{\sqrt{n}} \sum_i (Z_i - \sqrt{\theta}) \big)^2,
\end{align}
and $A_3=0$.\footnote{See Section \ref{sec-4.2-proofs} in the Supplementary Appendix for details of this example.} In this case, since the estimator is quadratic, the expansion is exact. We can conclude that the bias of the estimator is $E[A_{2}]/n = 1/n$. Since $E[A_1A_2]=0$ in this example, the higher-order variance is
\[
Var\big(\sqrt{n}(\widehat{\theta}-\theta)\big) = Var(A_1) + Var(A_2)/n = 4\theta + 2/n.
\]

\subsection{The effect of bias correction}
A bias-corrected estimator can be formed by subtracting off an estimator $\widehat
{b}$ of $b_{0}$ from $\widehat{\theta}$
\[
\tilde{\theta}=\widehat{\theta}-\frac{\widehat{b}}{n}%
\]
The focus of this paper is on the effect of the choice of $\widehat{b}$ on the
variance of $\tilde{\theta}.$ Generally $\widehat{b}$ has no effect on the
asymptotic variance (i.e. the first-order variance) of $\tilde{\theta}$, as long as $\widehat{b}$ is bounded in
probability, because then $\sqrt{n}(\widehat{b}/n)=\widehat{b}/\sqrt{n}=o_{p}(1).$ The
choice of $\widehat{b}$ can have an effect on the higher-order variance of
$\tilde{\theta}.$ To describe this effect note that the asymptotic expansion
for $\widehat{\theta}$ implies that, by adding and subtracting $b_{0}/\sqrt{n},$%
\begin{align}\label{eq:exp_bias_corr}
\sqrt{n}(\tilde{\theta}-\theta_{0})  &  =A_{1}+\frac{A_{2}-b_{0}}{\sqrt{n}%
}+\frac{A_{3}-\sqrt{n}(\widehat
{b}-b_{0})}{n}+o_{p}(n^{-1})\\
&  =A_{1}+\frac{\tilde{A}_{2}}{\sqrt{n}%
}+\frac{\tilde{A}_{3}}{n}+o_{p}(n^{-1}),
\nonumber
\end{align}
where $\tilde{A}_{2}=A_2-b_0$, and $\tilde{A}_{3}=A_3-
 \sqrt{n}(\widehat
{b}-b_{0})$. 
This is again an expansion whose terms are of decreasing stochastic order, so long as $\widehat{b}$ is
$\sqrt{n}$-consistent so that the third term is smaller order than the second term. Importantly, the second-order term  $\tilde{A}_{2}/\sqrt{n}$ has expectation zero, so that the bias-corrected estimator $\tilde{\theta}$ is higher-order unbiased. 

To analyze the higher-order variance of the bias-corrected estimator it is
helpful to be more specific about $\widehat{b}.$ For now we assume $\widehat{b}$
is asymptotically equivalent to a sample average, i.e.\ that there exists a $\phi(z)$
with $E[\phi(Z)]=0$, $Var(\phi(Z))<\infty,$ such that%
\[
\sqrt{n}(\widehat{b}-b_{0})=\frac{1}{\sqrt{n}}\sum_{i=1}^{n}\phi(Z_{i})+o_{p}(1).
\]
Here we are assuming that $\widehat{b}$ is asymptotically linear with influence
function $\phi(z).$ Bias corrections that are based on directly estimating
$b_{0}=E[B(Z)]$ will often have this property and we find that other bias
corrections, like the jackknife and bootstrap also have this property. In this
case the expansion in equation (\ref{eq:exp_bias_corr}) continues to hold with
$\tilde{A}_{3}=A_{3}-\Delta$ for $\Delta:=\sum_{i=1}^{n}\phi(Z_{i})/\sqrt{n}.$ 
The variance of the third-order approximation for $\tilde{\theta}$ is $Var(A_1) + \tilde{\upsilon}/n$, with higher-order variance%
\begin{align} \label{eq:higher-var}
\tilde{\upsilon}&:=Var(\tilde{A}_{2})+2\sqrt{n}E[A_{1}\tilde{A}_{2}]+2E[A_{1}\tilde{A}_{3}] \\
&= Var(A_{2}-b_0)+2\sqrt{n}E[A_{1}(A_{2}-b_0)]+2E[A_{1}(A_3 - \Delta)] \nonumber\\ 
&= \upsilon -2E[A_{1}\Delta] \nonumber.
\end{align}
Higher-order efficiency of the bias-corrected estimator refers to the size of $\tilde{\upsilon}$. One bias-corrected estimator is higher-order more efficient than another if it has smaller $\tilde{\upsilon}$.

The contribution of the bias correction to the higher-order variance $\tilde{\upsilon}$ is through the term
\begin{align} \label{eq:psi_phi_cov}
-2E[A_{1}\Delta]=-2E[\psi(Z)\phi(Z)].
\end{align}
Thus the bias correction $\widehat{b}$ affects the higher-order variance only
through the covariance of its influence function with that of $\widehat{\theta}$.
In the following theorem, we show that when $\tilde{\theta}$ is an efficient 
estimator, its  higher-order variance does not depend on the choice of $\widehat{b}$, by demonstrating that the covariance in (\ref{eq:psi_phi_cov}) is the same for any $\phi(Z)$. We adopt the notation and terminology of Newey (1990, pp 104-106) for the statement and proof of this result. Here $\mathcal{T}$ denotes the tangent set for the semiparametric model, which is the mean-square closure of the set of all scores for regular parametric submodels. A parameter is differentiable if there is a random variable $d(Z)$ such that the derivative of the parameter with respect to parameters of any regular parametric submodel exists and equals the expected product of $d(Z)$ with the submodel score.  Under additional regularity conditions, the influence function of any asymptotically linear estimator will be equal to such a $d(Z)$.

\begin{theorem} \label{thm:semip_equiv}
If $\mathcal{T}$ is linear,  $\theta_0$ and $b_0$ are differentiable parameters of the semiparametric model,  $\widehat{\theta}$ is asymptotically linear and efficient with influence function $\psi(Z)$, and $\widehat{b}$ is asymptotically linear with influence function $d(Z)$, then the higher-order variance $\tilde{\upsilon}$ does not depend on $\widehat{b}$. 
\end{theorem}

\begin{proof}

It follows from asymptotic linearity and efficiency of $\widehat{\theta}$ that its influence function $\psi(Z)$ is an element of $\mathcal{T}$. By Theorem 3.1 of Newey (1990) the efficient influence function $\delta(Z)$ for $b_0$ is equal to the projection of $d(Z)$ on $\mathcal{T}$. This allows us to decompose $d(Z)$ as,
\[
d(Z)=\delta(Z)+U(Z),\qquad E[U(Z)t(Z)]=0\text{ for all }%
t\in\mathcal{T}\text{.}%
\]
As shown above in (\ref{eq:psi_phi_cov}), the higher-order variance of $\tilde{\theta} = \widehat{\theta} - \widehat{b}/n$, depends on $\widehat{b}$ only through the covariance of $\psi(Z)$ with $d(Z)$. By the above decomposition, since $\psi(Z)$ is in the tangent set, this covariance is
\[
E[\psi(Z)d(Z)]=E[\psi(Z)\delta(Z)]+E[\psi(Z)U(Z)]=E[\psi(Z)\delta(Z)],
\]
Since the efficient influence function $\delta(Z)$ is unique, it follows
that $E[\psi(Z)d(Z)]$ does not vary with $d(Z)$, and hence $\tilde{\upsilon}$ does not vary with $\widehat{b}$.

\end{proof}

The proof uses geometry associated with semiparametric models, and relies on the efficiency of $\widehat{\theta}$. 
There is also some relatively simple intuition for this result. Consider any two semiparametric estimators $\widehat{b}_{1}$ and $\widehat{b}_{2}$ of $b_{0}.$ The random variable $\widehat{b}_{1}%
-\widehat{b}_{2}$ is a semiparametric estimator of $0$. If the asymptotic
covariance of $\widehat{\theta}$ with $\widehat{b}_{1}-\widehat{b}_{2}$ were nonzero then
for some fixed constant $C$ the asymptotic variance of $\bar{\theta}=\widehat{\theta
}+C(\widehat{b}_{1}-\widehat{b}_{2})$ would be less than the asymptotic variance of
$\widehat{\theta}$, so that $\widehat{\theta}$ could not be semiparametrically efficient.
That is, semiparametric efficiency of $\widehat{\theta}$ implies zero asymptotic
covariance of $\widehat{\theta}$ with $\widehat{b}_{1}-\widehat{b}_{2}$, which is
equivalent to the asymptotic covariance of $\widehat{\theta}$ with $\widehat{b}_{1}$
being equal to the asymptotic covariance of $\widehat{\theta}$ with $\widehat{b}_{2}.$
The asymptotic covariance of $\widehat{\theta}$ with any asymptotically linear
$\widehat{b}$ is $E[\psi(Z)\phi(Z)],$ so semiparametric efficiency of $\widehat
{\theta}$ implies that this covariance does not depend on $\phi(Z).$
This intuition extends the Hausman (1978) result that the covariance of any estimator of a parameter of interest with an efficient estimator of that parameter is equal to the variance of the efficient estimator, i.e.\ the covariance does not depend on the estimator. The intuition and result here show that the covariance of an efficient estimator of a parameter of interest with an estimator of any object does not depend on the estimator of that object.

\subsection{Asymptotically linear bias correction \label{subsec:2_3}}

There are many examples of bias corrections to semiparametric estimators that
have different influence functions. Consider a parametric model of the
conditional pdf of an outcome variable $Y$ given regressors $X$, where the pdf of
$Y$ conditional on $X$ has a parametric form $f(y|x,\beta)$ and $\theta$ is
some function of $\beta.$ This is a familiar semiparametric model where the
marginal distribution of $X$ is unspecified. A special case is a parametric
likelihood model where $X$ is a constant and $f(y|x,\beta)$ specifies the unconditional pdf
of observation $Y$. Since generally $b_{0}=E[B(Z)]$, for $Z=(Y,X)$ and for some function
$B(z)$ that may also depend on $\beta$, one could devise a number of estimates for $b_{0}$. 
For example, an analytical bias correction could be obtained by plugging in estimates $\widehat{\beta}$ and replacing the expectation operators with
sample averages. Alternative forms of analytical bias correction could make use
of the structure imposed by $f(y|x,\beta)$ to integrate over $y$ in $B(z)$, or use restrictions implied by $\int f(y|x,\beta)dy=1$ to estimate $b_{0}$ (e.g. applying the information equality). 

One could also use a nonparametric method like the bootstrap or jackknife to
estimate $b_{0}$. The bootstrap bias correction estimates the bias by $\frac{1}{n}\widehat{b}_B\equiv E^{\ast}\left[  \widehat{\theta}^{\ast}\right]
-\widehat{\theta}$, where $\widehat{\theta}^{\ast}$ are estimates obtained using bootstrap samples from the empirical distribution of $Z$. Alternatively, the jackknife estimates the bias by taking the difference between $\widehat{\theta}$ and the average of all `leave-one-out' estimates, i.e.\ estimates formed by excluding a single observation. The jackknife bias estimator is given by
$\frac{1}{n} \widehat{b}_J =(n-1)\big(\frac{1}{n}\sum_{i=1}^{n}
\widehat{\theta}_{\left(  i\right)  } - \widehat{\theta} \big)$
where $\widehat{\theta}_{(i)}$ is the estimate of the parameter that excludes observation $i$. In Section \ref{sec:mle_crosssec} we describe these bias correction techniques in detail in the context of MLE.

\medskip
{\bf Example 1, continued:}
Continuing our earlier example, we can compute the jackknife and analytical bias corrections for $\widehat{\theta}$. From the asymptotic expansion in (\ref{eq:exp_example}), we have that
\[
    E[A_2] = E[(Z_i - \sqrt{\theta})^2]
\]
which suggests the analytical bias estimate $\widehat{b}_a = \frac{1}{n}\sum_i (Z_i  - \bar{Z})^2$. Alternatively, using the known variance, we could simply use $b_a=1$. To construct a jackknife bias estimator, we use
\[
    \frac{1}{n} \widehat{b}_J = (n-1)\Big( \frac{1}{n}\sum_i \big(\frac{1}{n-1} \sum_{j\ne i} Z_j \big)^2 - \big(\frac{1}{n} \sum_{i} Z_i \big)^2 \Big)
    = \frac{1}{n(n-1)}\sum_i (Z_i  - \bar{Z})^2
\]
which is the same as the first analytical bias estimate, up to the factor $n/(n-1)$. It is straightforward to show that both bias estimates can be written as
\[
    \sqrt{n}(\widehat{b}-1) = \frac{1}{n}\sum_i \big((Z_i  - \sqrt{\theta})^2 -1\big) + o_p(1)
\]
so that they are both asymptotically linear, in this case with the same influence function $d(Z)= (Z  - \sqrt{\theta})^2 -1$. Given that the third moment of $Z$ is equal to zero, since it is normally distributed, we have $E[A_1(\widehat{b}-1)]=0$, so that the higher-order variance of the bias-corrected estimators that use either of the analytical or the jackknife bias estimates, are the same as that of $\widehat{\theta}$, i.e. $4\theta + 2/n$.
\medskip

Theorem \ref{thm:semip_equiv} implies that the various analytical bias corrections as well as both nonparametric bias correction methods
lead to the same higher-order variance as any other method when the $\widehat{b}$
is asymptotically linear. In the example above, both the analytical and jackknife bias estimates are asymptotically linear. This is not a coincidence. It turns out that for a fairly general class of models, the analytical, jackknife, and bootstrap bias estimates are asymptotically linear, so that Theorem \ref{thm:semip_equiv} can be applied.\footnote{The theorem is stated for a scalar parameter $\theta$, but straightforwardly extends to a vector-valued parameter, at the cost of more complicated notation.}

\begin{theorem} \label{thm:asym_lin}
Let $\widehat{\theta}$ be an estimator with a stochastic expansion
\begin{equation} \label{eq:expansion_asym_lin}
\widehat{\theta}-\theta_{0} = \frac{1}{\sqrt{n}}A_{1}+\frac{1}{n}A_{2}+\frac{1}{n^{3/2}}A_{3}
 +\frac{1}{n^2}A_{4}+\frac{1}{n^{5/2}}A_{5}+o_{p}(n^{5/2})
\end{equation}
where $n^{-k/2}A_{k}$ is a $k$-th order V-statistic, i.e.\
\begin{equation}
\frac{1}{n^{k/2}}A_{k}=\frac{1}{n^k}\sum_{i_1=1}^{n}\cdots \sum_{i_k=1}^{n}
g_{k,1}(z_{i_1},\theta_0)\cdots g_{k,k}(z_{i_k},\theta_0)
\end{equation}
and $g_{k,j}$ are functions of the data that are continuously differentiable in $\theta$ at $\theta_0$, mean zero, i.e.\ $E[g_{j,k}(z_{i},\theta_0)]=0$, and with $E[g_{j,k}(z_i,\theta_0)^{10}]\leq C <\infty$.

Then, the jackknife, bootstrap, and (sample average) analytical\footnote{See Section \ref{sec:mle_crosssec} for exact definitions of these three bias corrections, $\widehat{\theta}_J$, $\widehat{\theta}_B$, and $\tilde{\theta}_a$.} bias estimates of the higher-order 
bias $b_{0}=E[A_2]=E[g_{2,1}(z_{i},\theta_0)g_{2,2}(z_{i},\theta_0)]$, are asymptotically linear.
\end{theorem}

\begin{proof}
See Appendix \ref{app:main}
\end{proof}

The V-statistic structure assumed in Theorem \ref{thm:asym_lin} is shown to hold for the case of MLE under standard regularity conditions on the likelihood function (see Supplementary Appendix I \ref{Proof-Main-Results}). In Section 4 we derive such expansions for a maximum likelihood estimator. The derivation is based on an expansion of the first order condition of the MLE and so applies to any other estimator that can be characterized by a moment equation satisfying similar regularity conditions. Therefore, the structure can be shown to appear in the expansion of any M-estimator, i.e. an estimator that maximizes the sample average of some function, which need not be a log-likelihood.

An important semiparametric model is a model of unconditional moment
restrictions that motivates GMM estimation. Newey and Smith (2004) derive the asymptotic expansion for generalized empirical likelihood (GEL) estimators up to third order and examination of that expansion shows the same V-statistic structure. Under sufficient regularity conditions, the structure should also exist in the fifth-order expansion used in Theorem \ref{thm:asym_lin}, so that the equivalence result of Theorem \ref{thm:semip_equiv} would imply that jackknife and bootstrap bias corrections lead to the same higher-order variance for GEL estimators. This extends the result in Newey and Smith (2004), who found that
averaging over an efficient estimator of the distribution of the data to
obtain $\widehat{b}$ does not affect the higher-order variance of bias-corrected
GMM and generalized empirical likelihood (GEL) estimators.

Another interesting example is a nonparametric model where no restrictions are placed
on the distribution of the data. In this case $b_{0}$ will also be a
nonparametric object and there will exist only one influence function
corresponding to $b_{0}$; see Van der Vaart (1991) and Newey (1994). Here, all
asymptotically linear estimators of $b_{0}$ will have the same influence
function so that Theorem \ref{thm:semip_equiv} holds trivially. In particular, when general
misspecification is allowed for in debiasing, so that the distribution of
the data is unrestricted, there will only be one influence function for
$b_{0}$. Then any bias correction method, such as the analytical, bootstrap, or jackknife, 
must have the same higher-order variance, as long as they are asymptotically linear.

We emphasize that Theorem \ref{thm:semip_equiv} only applies to comparisons between bias
corrections for a given semiparametrically efficient estimator. The stochastic
expansion term $A_{2}$ may differ across semiparametrically efficient estimators so that their 
higher-order variances need not be the same. 
For example, Newey and Smith (2004) showed that GMM will tend to have larger second-order bias than GEL when there are many instrumental variables, so that $A_{2}$ is different for GMM
and GEL. Theorem \ref{thm:semip_equiv} has nothing to say about how the higher-order variance of 
bias-corrected GMM compares with that of bias-corrected GEL. It only implies that the form
of $\widehat{b}$ does not affect the higher-order variance of bias-corrected GMM
and separately for bias-corrected GEL.

Note that while Theorem \ref{thm:asym_lin} implies that bootstrap, jackknife, and analytical bias estimates are asymptotically linear in general, this is not true of all bias estimates. An example of this is the split-sample jackknife, which estimates the bias using $\frac{1}{n}\widehat{b}_{SS} = \frac{1}{2}(\widehat{\theta}_{(1)} + \widehat{\theta}_{(2)}) - \widehat{\theta}$, where $\widehat{\theta}_{(1)}$ and $\widehat{\theta}_{(2)}$ are estimates using two separate halves of the data set.

{\bf Example 1, continued:}
Let $N=2m$, and define $\xi^{(1)} = \frac{1}{\sqrt{m}} \sum_{i\leq m} (Z_i-\sqrt{\theta})$ and $\xi^{(2)} = \frac{1}{\sqrt{m}} \sum_{i>m} (Z_i-\sqrt{\theta})$.  Some algebra gives the expression
\begin{align*}
    \widehat{b}_{SS} &= \frac{1}{2} \big( \xi^{(1)} - \xi^{(2)} \big)^2
\end{align*}
Since $\xi^{(1)}$ and $\xi^{(2)}$ are independent standard normal variables, we have
$E[\widehat{b}_{SS}]=1$, so that the split-sample jackknife gives an unbiased estimator for $b_0$. However, $\widehat{b}_{SS}=O_p(1)$, so that the bias estimator is also inconsistent, and hence not asymptotically linear. In this case, the higher-order variance of the bias-corrected estimator can be shown to be $Var\big(\sqrt{n}(\widehat{\theta}_{SS})\big) = 4\theta + 4/n$, which is larger than that of the analytical and jackknife estimators. In Section \ref{sec:SS_var}, we give a formal result that shows that this larger higher-order variance is true in general for MLE.




\section{Bias-corrected MLE\label{sec:mle_crosssec}}

In this section, we demonstrate the higher-order equivalence of bias corrections in a fully parametric model by deriving higher-order expansions and variance expressions for several bias-corrected maximum likelihood estimators. We consider the bootstrap bias correction, jackknife bias correction,
as well as three different versions of analytical bias correction, including the bias-corrected MLE of Pfanzagl and Wefelmeyer (1978), which was shown to be third-order optimal. 
We show that all of these corrections lead to estimators with the same higher-order variance. Further, we show that the bootstrap, jackknife, and a particular analytical bias-correction, result in identical third-order asymptotic expansions. 

To describe the parametric model, let $\left\{  Z_{i}\right\}  _{i=1}^{n}$ be an i.i.d.\ sample
$Z_{i}\sim f\left(  z,\theta_{0}\right)  $, such that $f(z,\theta)$ satisfies
sufficient smoothness conditions.\footnote{The results in Section
\ref{sec:mle_crosssec} are predicated on regularity conditions,
including Conditions \ref{HH}, \ref{M} and \ref{EC}, which are presented in
Appendix \ref{sec:reg-conds}.} The density $f(z,\theta)$ is a member of
a parametric family of distributions $P_{\theta}$ indexed by $\theta\in\Theta$
with $\Theta\in\mathbb{R}$ a compact set.\footnote{For simplicity of notation,
we will assume that $p=1$, where $p=\dim\left(  \theta\right)  $. The result
is expected to hold for any finite $p$.} We consider properties of the
MLE $\widehat{\theta}$, where
\[
\widehat{\theta}\equiv\arg\sup_{\theta\in\Theta}n^{-1}\sum_{i=1}^{n}\log f\left(
Z_{i},\theta\right)  .
\]

We adopt the following set of notations in this section of the paper. We let $\ell\left(
\cdot,\theta\right)  \equiv\left.  \partial\log f\left(  \cdot,\theta\right)
\right/  \partial\theta$, $\ell^{\theta}\left(  \cdot,\theta\right)
\equiv\left.  \partial^{2}\log f\left(  \cdot,\theta\right)  \right/
\partial\theta^{2}$, $\ell^{\theta\theta}\left(  \cdot,\theta\right)
\equiv\left.  \partial^{3}\log f\left(  \cdot,\theta\right)  \right/
\partial\theta^{3}$, etc. We also define $\mathcal{I}\equiv-E\left[
\ell^{\theta}(Z_{i},\theta_{0})\right]  $, $\mathcal{Q}_{1}\left(
\theta\right)  \equiv E\left[  \ell^{\theta\theta}(Z_{i},\theta)\right]  $,
and $\mathcal{Q}_{2}\left(  \theta\right)  \equiv E\left[  \ell^{\theta
\theta\theta}(Z_{i},\theta)\right]  $. Furthermore, we let $U_{i}\left(
\theta\right)  \equiv\ell\left(  Z_{i},\theta\right)  $, $V_{i}\left(
\theta\right)  \equiv\ell^{\theta}\left(  Z_{i},\theta\right)  -E\left[
\ell^{\theta}\left(  Z_{i},\theta\right)  \right]  $, $W_{i}(\theta) \equiv\ell
^{\theta\theta}\left(  Z_{i},\theta\right)  -E\left[  \ell^{\theta\theta}\left(
Z_{i},\theta\right)  \right]  $. Finally, we define $U\left(  \theta\right)  \equiv
n^{-1/2}\sum_{i=1}^{n}U_{i}\left(  \theta\right)  $, $V\left(  \theta\right)
\equiv n^{-1/2}\sum_{i=1}^{n}V_{i}\left(  \theta\right)  $, and $W\left(
\theta\right)  \equiv n^{-1/2}\sum_{i=1}^{n}W_{i}\left(  \theta\right) $. 
We next describe the bias corrections and derive higher-order properties of the bias-corrected estimators. 

\subsection{The bootstrap and jackknife bias corrections}

The bootstrap estimator constructs an estimate of the higher-order bias as the difference between the average of bootstrap replicates $\widehat{\theta}^{\ast}$ and the MLE $\widehat{\theta}$. In particular, we first obtain bootstrapped estimates $\widehat{\theta}^{\ast}$ by sampling
$Z_{1}^{\ast},...,Z_{n}^{\ast}$ identically and independently from the
empirical distribution $\widehat{F}(z) = \frac{1}{n}\sum 1\{ Z_i \leq z\}$. Let $E^{\ast}$ be the expectation
operator with respect to $\widehat{F}.$ The idea behind the bootstrap bias
correction is to estimate $E\left[  \widehat{\theta}\right]  -\theta_{0},$ if
it exists, by $\frac{1}{n}\widehat{b}_B\equiv E^{\ast}\left[  \widehat{\theta}^{\ast}\right]
-\widehat{\theta}$. This in turn will allow us to construct the bias-corrected estimate
\[
\widehat{\theta}_{B}\equiv\widehat{\theta}-\frac{\widehat{b}_B}{n}=2\widehat{\theta}-E^{\ast
}\left[  \widehat{\theta}^{\ast}\right]  .
\]

An alternative nonparametric bias correction is a jackknife bias-corrected estimator. The jackknife estimates the higher-order bias by taking the difference between the MLE and the average of all `leave-one-out' estimates, i.e.\ estimates formed by excluding a single observation. The jackknife estimate is given by
\[
\widehat{\theta}_{J} = \widehat{\theta} - \frac{\widehat{b}_J}{n}
=n\widehat{\theta}-\frac{n-1}{n}\sum_{i=1}^{n}%
\widehat{\theta}_{\left(  i\right)  }%
\]
where $\widehat{\theta}_{(i)}$ is the estimate of the parameter that excludes observation $i$. The jackknife uses the bias estimate $\frac{1}{n}\widehat{b}_J = (n-1)\left(\frac{1}{n}\sum_{i=1}^{n}\widehat{\theta}_{\left(  i\right)} - \widehat{\theta} \right)$.

The following proposition establishes the higher-order properties of the
bootstrap and jackknife bias-corrected MLEs.

\begin{proposition}
\label{boot-jack-higher}Let $\tilde{b}$ be either the bootstrap bias estimate $\widehat{b}_B$, or the jackknife bias estimate $\widehat{b}_J$. Then, under regularity conditions stated in Appendix \ref{sec:reg-conds}, $\tilde{b}$ satisfies
\begin{align*}
    \sqrt{n}(\tilde{b} - b_0) &= \frac{1}{\sqrt{n}}\sum_{i=1}^n \mathbb{B}_i + o_p(1)
\end{align*}
where
\begin{align*}
    \mathbb{B}_i &= \left( \frac{1}{2}\mathcal{I}^{-3}\mathcal{Q}_{2} + \frac{3}{2}\mathcal{I}^{-4}\mathcal{Q}_{1}^2 + 3\mathcal{I}^{-4}\mathcal{Q}_{1}E[U_iV_i] + \mathcal{I}^{-3}\left(E[U_iW_i] + E[(\ell^\theta)^2] \right) \right) U_i \\
    &+ \left( \frac{3}{2}\mathcal{I}^{-3}\mathcal{Q}_{1} + 2\mathcal{I}^{-3}E[U_iV_i]  \right) V_i 
    + \frac{1}{2}\mathcal{I}^{-2} W_i \\
    &+ \frac{1}{2}\mathcal{I}^{-3}\mathcal{Q}_{1}\left(   \ell(Z_i,\theta_0)^2 - E[\ell^2]\right) 
    + \mathcal{I}^{-2} \left( \ell(Z_i,\theta_0) \ell^\theta(Z_i,\theta_0) - E[\ell \ell^\theta] \right)
\end{align*}
and hence, for $\tilde{\theta}$ either the bootstrap or jackknife bias-corrected MLE,
\[
\sqrt{n}\left(  \tilde{\theta}-\theta_{0}\right)  =A_{1}+\frac{1}%
{\sqrt{n}}\left(  A_{2}-b(\theta_{0})\right)  +\frac{1}{n}(A_{3}-
\mathbb{B})+o_{p}\left( n^{-1} \right)  ,
\]
where $\mathbb{B} = \sum_i \mathbb{B}_i / \sqrt{n}$.
\end{proposition}

\begin{proof}
See Supplementary Appendix I \ref{Proof-Main-Results}\footnote{The Supplementary Appendix is available at \href{http://arxiv.org/abs/2207.09943}{http://arxiv.org/abs/2207.09943}}.
\end{proof}

Here we see that the jackknife and bootstrap bias estimates are $\sqrt{n}$-consistent and asymptotically linear (note that $E[\mathbb{B}_i]=0$), which from Theorem  \ref{thm:semip_equiv} implies that their higher-order variances are the same as any other estimator that uses an asymptotically linear estimator of the bias. In fact, the jackknife and bootstrap share an even stronger equivalence, in the sense that their bias estimates have identical influence functions, which implies that their asymptotic expansions are identical up to the third order.

\subsection{Analytical bias corrections}

We next introduce and discuss three forms of analytical bias correction. An analytical expression for the higher-order bias is given by
\begin{align}
b\left(  \theta_{0}\right)  \equiv E \left[  A_{2}\right]   &
=\frac{1}{2\mathcal{I}^{3}}E\left[  \ell^{\theta\theta}\right]  E\left[
\ell^{2}\right]  +\frac{1}{\mathcal{I}^{2}}E\left[  \ell\ell^{\theta}\right]
\label{bias-no-info}
\end{align}

The first bias correction is based on the bias formula (\ref{bias-no-info}), replacing expectations with sample averages. This gives the estimator $\widetilde{\theta}_{a}\equiv\widehat{\theta}-\frac{\widetilde{b}\left(
\widehat{\theta}\right)  }{n}$, where
\begin{align} \label{analy-no-info}
\widetilde{b}\left(\widehat{\theta}\right) =
-\frac{\left(  \frac{1}{n}\sum_{i}\ell^{\theta\theta}(  Z_{i},\widehat{\theta
})  \right)  \left(  \frac{1}{n}\sum_{i} \ell(  Z_{i}%
,\widehat{\theta})   ^{2}\right)  }{2\left(  \frac{1}{n}\sum_{i}%
\ell^{\theta}(  Z_{i},\widehat{\theta})  \right)  ^{3}}%
+\frac{\left(  \frac{1}{n}\sum_{i}\ell(  Z_{i},\widehat{\theta})
\ell^{\theta}(  Z_{i},\widehat{\theta})  \right)  }{\left(
\frac{1}{n}\sum_{i}\ell^{\theta}(  Z_{i},\widehat{\theta})  \right)
^{2}} 
\end{align}

It can be shown that this form of analytical bias estimate shares the same influence function as the bootstrap and jackknife bias estimators. It follows that it has an identical third-order asymptotic expansion.

\begin{proposition}
\label{HigherOrderEfficiency3} Let the regularity conditions stated in Appendix \ref{sec:reg-conds} hold. The estimator $\widetilde{\theta}_a$ has the higher-order expansion
\[
\sqrt{n}\left(  \widetilde{\theta}_{a}-\theta_{0}\right)  = A_{1}+\frac
{1}{\sqrt{n}}\left(  A_{2}-b\left(  \theta_{0}\right)  \right)  +\frac{1}%
{n}\left(  A_{3}-\mathbb{B}\right)  +o_{p}\left(  n^{-1}\right)  .
\]

\end{proposition}

\begin{proof}
See Supplementary Appendix I \ref{Proof-Main-Results}.
\end{proof}

Alternative analytical corrections can be constructed by imposing the information equality, $E[\ell^2]=\mathcal{I}$ in the characterization of the bias, that is, using the bias formula 
\begin{align}
b\left(  \theta_{0}\right)  \equiv E \left[  A_{2}\right]  
  =\mathcal{I}^{-2} \left( \frac{1}{2}E\left[  \ell^{\theta\theta}\right]  + E\left[  \ell\ell^{\theta}\right] \right). \label{bias-info}%
\end{align}

An analytical bias correction that uses the information equality is given by
\[
\widehat{\theta}_{a}\equiv\widehat{\theta}-\frac{\widehat{b}\left(
\widehat{\theta}\right)  }{n}\equiv\widehat{\theta}-\frac{1}{n}\left(
\frac{\left(  n^{-1}\sum_{i}\ell^{\theta\theta}(  Z_{i},\widehat{\theta
})  \right)  }{2\left(  n^{-1}\sum_{i}\ell^{\theta}(
Z_{i},\widehat{\theta})  \right)  ^{2}}+\frac{\left(  n^{-1}\sum_{i}%
\ell(  Z_{i},\widehat{\theta})  \ell^{\theta}(  Z_{i}%
,\widehat{\theta})  \right)  }{\left(  n^{-1}\sum_{i}\ell^{\theta
}(  Z_{i},\widehat{\theta})  \right)  ^{2}}\right)  .
\]

It is also possible to construct a bias correction that computes the expectations in (\ref{bias-info}) in their integral form. This is the analytical bias-corrected estimator of Pfanzagl and Wefelmeyer (1978), which they showed is
higher-order efficient. 
\[
\widehat{\theta}_{c}\equiv\widehat{\theta}-\frac{b\left(  \widehat{\theta
}\right)  }{n}=\widehat{\theta}-\frac{1}{n}\left(  \frac{\int\ell
^{\theta\theta}(  z,\widehat{\theta})  f(  z,\widehat{\theta
})  dz}{2\left(  \int\ell^{\theta}(  z,\widehat{\theta})
f(  z,\widehat{\theta})  dz\right)  ^{2}}+\frac{\int\ell(  z,\widehat{\theta})  \ell^{\theta}(  z,\theta
)  f(  z,\widehat{\theta})  dz}{\left(  \int\ell^{\theta
}(  z,\widehat{\theta})  f(  z,\widehat{\theta})
dz\right)  ^{2}}\right)  .
\]

The analytical bias-corrected estimators that impose the information equality do not have identical higher-order expressions to the bias-corrected estimators that do not rely on the information equality. Instead, they differ in their third expansion terms, as can be seen in the following proposition.

\begin{proposition}
\label{HigherOrderEfficiency}Let the regularity conditions stated in Appendix \ref{sec:reg-conds} hold. The analytical  bias estimates satisfy
\begin{align*}
    \sqrt{n}(\widehat{b}(\widehat{\theta})  - b_0) &= \frac{1}{\sqrt{n}}\sum_{i=1}^n \mathbb{A}_i + o_p(1) \\
    \sqrt{n}(b(\widehat{\theta}) - b_0) &= \frac{1}{\sqrt{n}}\sum_{i=1}^n \mathbb{C}_i + o_p(1)
\end{align*}
where $\mathbb{A}_i$ and $\mathbb{C}_i$ are as defined in Supplementary Appendix I \ref{Proof-Main-Results}. Hence, the analytical bias-corrected estimators $\widehat{\theta}_{c}$ and $\widehat{\theta}_{a}$ have higher-order expansions
\begin{align*}
\sqrt{n}\left(  \widehat{\theta}_{c}-\theta_{0}\right)   
&  =A_{1}+\frac{1}{\sqrt{n}}\left(  A_{2}-b\left(  \theta_{0}\right)  \right)  +\frac{1}%
{n}\left(  A_{3}-\mathbb{C}\right)  +o_{p}\left(  n^{-1}\right)  ,\\
\sqrt{n}\left(  \widehat{\theta}_{a}-\theta_{0}\right)   
&  =A_{1}+\frac{1}{\sqrt{n}}\left(  A_{2}-b\left(  \theta_{0}\right)  \right)  +\frac{1}%
{n}\left(  A_{3}-\mathbb{A}\right)  +o_{p}\left(  n^{-1}\right)  ,
\end{align*}
where $\mathbb{A} = \sum_i \mathbb{A}_i /\sqrt{n}$ and $\mathbb{C} = \sum_i \mathbb{C}_i /\sqrt{n}$.
\end{proposition}

\begin{proof}
See Supplementary Appendix I \ref{Proof-Main-Results}.
\end{proof}

Propositions \ref{boot-jack-higher},  \ref{HigherOrderEfficiency3}, and 
\ref{HigherOrderEfficiency} imply that the
higher-order variances of $\widehat{\theta}_{B}$, $\widehat{\theta}_{J}$,
$\widetilde{\theta}_{a}$ are equal to $\operatorname{Var}\left(  A_{2}\right)
+2n^{1/2}E\left[  A_{1}A_{2}\right]  +2E\left[  A_{1}\left(  A_{3}%
-\mathbb{B}\right)  \right]  $,
while the higher-order variances of $\widehat{\theta}_{a}$, and
$\widehat{\theta}_{c}$ are $\operatorname{Var}\left(  A_{2}\right)
+2n^{1/2}E\left[  A_{1}A_{2}\right]  +2E\left[  A_{1}\left(  A_{3}%
-\mathbb{A}\right)  \right]  $ and $\operatorname{Var}\left(  A_{2}\right)
+2n^{1/2}E\left[  A_{1}A_{2}\right]  +2E\left[  A_{1}\left(  A_{3}%
-\mathbb{C}\right)  \right]  $, respectively.
Therefore, it is natural to
conjecture that the higher-order variances of $\widehat{\theta}_{B}$,
$\widehat{\theta}_{J}$, and $\widetilde{\theta}_{a}$ are different from that
of $\widehat{\theta}_{c}$ or $\widehat{\theta}_{a}$. 
However, the propositions also state that each estimator uses an estimate of the higher-order bias $b_0$ that is asymptotically linear, and hence a corollary of Theorem \ref{thm:semip_equiv} is that the higher-order variances of the bias-corrected estimators are all equal.

\begin{corollary}\label{ET1T3}
Let the regularity conditions stated in Appendix \ref{sec:reg-conds} hold. Then, 
\[
    E\left[  \mathbb{B}A_{1}\right]  =E\left[
\mathbb{A}A_{1}\right]  =E\left[  \mathbb{C}A_{1}\right]
\]
\end{corollary}

\begin{proof}
See Supplementary Appendix I \ref{Proof-Main-Results}.
\end{proof}

By Corollary \ref{ET1T3}, and the expression for the higher-order variance (\ref{eq:higher-var}), we conclude that all five estimators are 
higher-order efficient. This result follows because (i) they have identical higher-order variance; and (ii) $\widehat{\theta}_{c}$ was shown to be higher-order efficient by Pfanzagl and Wefelmeyer (1978).

The next proposition provides an expression for the higher-order variance of these five bias-corrected estimators.

\begin{proposition} \label{higher_var}
Let the regularity conditions stated in Appendix \ref{sec:reg-conds} hold, and define
\begin{align*}
    X_{i} &= \mathcal{I}^{-1}U_i \\
    Y_{i} &= \frac{1}{2}\mathcal{I}^{-2}\mathcal{Q}_1U_i + \mathcal{I}^{-1}V_i \\
    \Upsilon &= \left( E[X_{i}^2]E[Y_{i}^2] + E[X_{1}Y_{i}]^2\right) 
\end{align*}
The higher-order variance of $\widehat{\theta}_B$, $\widehat{\theta}_J$, $\widetilde{\theta}_a$, $\widehat{\theta}_a$, and $\widehat{\theta}_c$ is
\begin{align} \label{eq:higher-var-expr}
    Var(A_1) + \frac{\tilde{\upsilon}}{n} = 
    \mathcal{I}^{-1} + \frac{1}{n}\Upsilon
\end{align}
\end{proposition}

\begin{proof}
See Supplementary Appendix I \ref{Proof-Main-Results}.
\end{proof}

The above higher-order variance expression is useful since the higher-order variance formula in (\ref{eq:higher-var}) includes two covariance terms $2E\left[  A_{1}%
A_{3}\right]  +2n^{1/2}E\left[  A_{1}A_{2}\right]  $, and as such, it is not
obvious that (\ref{eq:higher-var}) is positive. On the other hand, because
$\Upsilon>0$, the expression in the proposition consists of
positive terms and eliminates this concern. According to Theorem \ref{thm:semip_equiv},
(\ref{eq:higher-var-expr}) is also the higher-order variance of any other asymptotically linear bias correction.




\subsection{Inefficient bias correction} \label{sec:SS_var}

Our discussion above shows that many common methods of bias correction use estimates of the bias that are asymptotically linear, and hence estimate the bias at the $n^{-1/2}$-rate. These bias-corrected estimators have equivalent higher-order variances.
We now show by an example that the equivalence
result does not hold in general if this rate assumption is violated. We showed
in Example 1 (see the end of Section
\ref{subsec:2_3}) that the split-sample jackknife provided an inconsistent estimate of the higher-order bias, so that its 
higher-order variance was larger than that of the leave-one-out jackknife. In the Appendix we show that this result holds more generally, for estimators with the expansion structure used in Theorem 2. Here we derive expressions for the higher-order variance in the MLE case.\footnote{Derivation of the expressions used below are available in Section \ref{HOV-SS} of Supplementary Appendix I.}

Recall that the split-sample jackknife estimator is given by \[\widehat{\theta}_{SS}\equiv2\widehat{\theta}-\left(  \widehat{\theta}_{(1)}+\widehat{\theta}_{(2)}\right)  /2,
\]
where $\widehat{\theta}_{(  1)  }$ and $\widehat{\theta}_{(  2)}$ denote the MLEs based on separate halves of the sample. The implicit bias estimate used by the split-sample jackknife is given by $\frac{1}{n}\widehat{b}_{SS}=(  \widehat{\theta}_{(1)}+\widehat{\theta}_{(2)})  /2 - \widehat{\theta}$.

\begin{proposition} \label{higher_var_ss}
Let the regularity conditions stated in Appendix \ref{sec:reg-conds} hold. The bias estimate used by the split-sample estimator can be written as
\begin{align*}
    \widehat{b}_{SS} - b_0 &= \frac{1}{2}\left( X_{(1)}Y_{(1)} + X_{(2)}Y_{(2)} - 2b_0 -\left( X_{(1)}Y_{(2)} + X_{(2)}Y_{(1)}  \right) \right) \\
    &= O_p(1)
\end{align*}
where, for $n=2m$, $X_{(1)} = \frac{1}{\sqrt{m}}\sum_{i=1}^m X_{i}$ is the scaled sum over $X_{i}$ in the first half of the sample, and $X_{(2)}$, $Y_{(1)}$, and $Y_{(2)}$ are defined similarly (for $X_i$ and $Y_i$ as defined in Proposition \ref{higher_var}).

The higher-order variance of $\widehat{\theta}_{SS}$ is
\begin{align*}
    Var(A_1) + \frac{\tilde{\upsilon}_{SS}}{n} = 
    \mathcal{I}^{-1} + \frac{2}{n}\Upsilon
\end{align*}
\end{proposition}

\begin{proof}
See Supplementary Appendix I \ref{Proof-Main-Results}.
\end{proof}

That is, the higher-order variance of $\widehat{\theta}_{SS}$ is strictly larger than that of $\widehat{\theta}_{J}$. As shown in Proposition \ref{higher_var_ss}, the split-sample bias correction provides an inconsistent estimate of the bias term $b_0$, which results in a bias correction that impacts the second-order term $A_2$ in the expansion, and hence affects the higher-order variance of the estimator.




\section{Bias correction for panel data\label{section-panel-model}}

We now extend some of these ideas to the panel data setting by comparing the jackknife and split-sample bias correction methods for a model with individual fixed effects. We begin with a description of the fixed effects maximum likelihood estimator.
Let $Z_{it}$, for $i=1,\dots,n$ and $t=1,\dots,T$, be a vector of observed
data. Denote $\theta$ a $p\times1$ parameter vector and $\alpha_{i}$ a scalar
unobserved individual effect.\footnote{As before, we will assume that $p=1$
for notational simplicity, but results should be expected to hold for any finite
$p>1$. The analysis in this section assumes that time effects are not present.} The data have density function $f(z|\theta,\alpha)$ with respect to
some measure, and so (treating the $\alpha_{i}$ as parameters to be estimated)
we may estimate $\theta$ via maximum likelihood. Assuming that the $Z_{it}$
are independent across both $i$ and $t$, the MLE solves
\[
\widehat{\theta}_{T}\equiv\arg\max_{\theta}\sum_{i=1}^{n}\sum_{t=1}^{T}\ln
f(Z_{it}|\theta,\widehat{\alpha}_{i}(\theta)),\quad\widehat{\alpha}_{i}(\theta
)\equiv\arg\max_{\alpha}\sum_{t=1}^{T}\ln f(Z_{it}|\theta,\alpha)
\]

In contrast to the higher-order bias discussed earlier for the cross-sectional model, the fixed effects panel data estimator suffers from the well-known incidental parameters problem (Neyman and Scott, 1948). For fixed $T$, the probability limit of the estimator $\theta_{T}\equiv
\operatorname*{plim}_{n\rightarrow\infty}\widehat{\theta}_{T}$ generally differs from
$\theta_{0}$.\footnote{$\theta_{T}$ is
given by $\operatorname*{argmax}_{\theta}\lim_{n\rightarrow\infty}\frac{1}%
{n}\sum_{i=1}^{n}E\left[  \sum_{t=1}^{T}\ln f(z_{it}|\theta,\widehat{\alpha}%
_{i}(\theta))\right]  $} Even when the number of time periods grow, so that $n/T\rightarrow\rho$ (which will be assumed in this
section), the estimator $\widehat{\theta}$ remains asymptotically biased, i.e.\
$\sqrt{nT}(\widehat{\theta}-\theta_{0})\Rightarrow N(\sqrt{\rho}\mathbf{B}%
,\Omega)$ for some bias term $\mathbf{B}$. The bias is of order $O(T^{-1})$,
and can be substantial if $T$ is not sufficiently large.

It is useful to think about $\theta_{0}$ and $\alpha_{i}$ as solutions to a
set of moment equations given by the score functions
\[
0=\sum_{i=1}^{n}E\left[  \frac{\partial}{\partial\theta}\ln f(z_{it}%
|\theta_{0},\alpha_{i})\right]  ,\quad0=E\left[  \frac{\partial}%
{\partial\alpha_{i}}\ln f(z_{it}|\theta_{0},\alpha_{i})\right]
\]
As earlier, we can expand these first order
conditions to produce asymptotic expansions. We may also consider other quantities of interest that can be
defined via some moment condition, for example an average effect parameter
$\mu_{0}$, defined as the solution to
\[
0=\mu_{0}-\frac{1}{n}\sum_{i=1}^{n}E\left[  m(z_{it},\theta_{0},\alpha
_{i})\right]
\]
where $m$ is some function of interest, for example, the partial derivative of a conditional expectation function. Stacking these moment conditions, we can define the common parameter to be
$(\theta,\mu)$, so that the results presented below will also apply to these types of parameters.

We use the following notation for this section. Let $u_{it}(\theta
,\alpha)\equiv\frac{\partial}{\partial\theta}\ln f(z_{it}|\theta,\alpha)$ and
$V_{it}(\theta,\alpha)\equiv\frac{\partial}{\partial\alpha_{i}}\ln
f(z_{it}|\theta,\alpha)$ be the score functions. When evaluating functions at
the true value of parameters, arguments will be dropped, e.g. $u_{it}%
=u_{it}(\theta_{0},\alpha_{i})$. Further, let $U_{it}(\theta,\alpha
)=u_{it}(\theta,\alpha)-\delta V_{it}(\theta,\alpha)$, for $\delta
=E[u_{it}V_{it}]/E[V_{it}^{2}]$, be the efficient score for $\theta$. All
expectations are taken with respect to the distribution for an individual $i$,
that is $E[h(z_{it})]=\int h(z)f(z|\theta,\alpha_{i})dz$. Also define $\mathcal{I}_n = \frac{1}{n}\sum_{i=1}^n E[U_{it}^2]$. As in the cross
sectional case, we denote partial derivatives of these functions with
superscripts, e.g. $\partial U_{it}/\partial\theta=U_{it}^{\theta}$.

\subsection{Higher-order comparison of jackknife estimators\label{jackknife-discussion}}

In this section we derive the higher-order properties of both the leave-one-out jackknife and split-sample jackknife estimators for the panel data model. 

It has be shown previously (e.g. Hahn and Kuersteiner (2002), Hahn and Newey (2004)) that the 
MLE $\widehat{\theta}$ has an expansion of the form
\begin{equation}
\sqrt{nT}(\widehat{\theta}-\theta_{0})= A_1 +\frac{\sqrt{n}}{\sqrt{T}}A_2+\frac{\sqrt{n}}{T}A_3+O_{p}(T^{-1}) \label{eq:mle_exp}%
\end{equation}
where the expansion terms $A_1,A_2,A_3$ are each $O_p(1)$.\footnote{In Supplementary Appendix IV, we provide an even
higher-order expansion than is available in Hahn and Newey (2004).} This expansion is similar in style to the asymptotic expansion for the cross-sectional model developed earlier, except that the asymptotic order of the leading terms in the expansion are in terms of $T^{-1/2}$ due to the presence of the individual fixed effects. Here the asymptotic bias is 
given by $\sqrt{\frac{n}{T}}\mathbf{B} = \sqrt{\frac{n}{T}} E[A_2]$.

The `leave-one-out' jackknife estimator is 
\[
\widehat{\theta}_{J}%
=T\widehat{\theta}-(T-1)\frac{1}{T}\sum_{t=1}^{T}\widehat{\theta}_{(t)},
\]
where $\widehat{\theta}_{(t)}$ is the estimator formed from the subsample that
excludes time period $t$.\footnote{Jackknife estimators that drop $k$ time periods, rather than just one, could also be used. However, averaging over all $\binom{T}{k}$ leave-$k$-out estimates would be computationally demanding, and so we do not pursue this idea here.}
Dhaene and Jochmans (2015) propose
the use of split-panel jackknifes that only make use of subpanels that contain  consecutive time periods. The split-sample jackknife estimator is
\[
\widehat{\theta}_{SS}=2\widehat{\theta}-\bar{\theta}_{SS},
\]
where $\bar{\theta}_{SS}=\frac{1}{2}(\widehat{\theta}_{1}+\widehat
{\theta}_{2})$, with $\widehat{\theta}_{1}$ being the estimate using observations from the first half of
time periods, and $\widehat{\theta}_{2}$ the estimate that uses the second half of
time periods. Other choices of split-sample jackknife are available; however, the results in
Dhaene and Jochmans (2015) show that non-overlapping sub-panels in general
have lower asymptotic variance, and that among the non-overlapping options,
splitting in two leads to the smallest inflation of higher-order bias. Hence,
in this paper we focus on this half sample version, and simply refer to it as
the split-sample jackknife.

Both the jackknife and split-sample corrections have no impact on the first-order term in (\ref{eq:mle_exp})  so that both  estimators have the same asymptotic variance, equal to $\lim_{n\rightarrow\infty}%
\operatorname{Var}(A_1)$ $=\lim_{n\rightarrow\infty}\mathcal{I}%
_{n}^{-1}$. To compute the higher-order variance for the estimators, we take the variance
of the first three expansions terms, retaining terms up to $O(T^{-1})$,
similar to what was done in Section \ref{sec:mle_crosssec}. The
following proposition establishes the higher-order variance expressions for the two 
estimators, from which we can conclude that the higher-order variance of the split-sample
estimator is larger than that of the jackknife.

\begin{proposition}
\label{thm:variance} Let the regularity conditions stated in Appendix \ref{sec:reg-conds-panel} hold. The higher-order variances of the jackknife and split-sample bias-corrected estimators are given by
\[
\operatorname{Var}(\widehat{\theta}_{J}) \approx \operatorname{Var}(A_1)+\frac{1}{T-1}\tilde{\upsilon},\quad \operatorname{Var}(\widehat{\theta}_{SS})\approx \operatorname{Var}(A_1)%
+\frac{2}{T}\tilde{\upsilon},
\]
where
\begin{align*}
\operatorname{Var}(A_1)  &  =\mathcal{I}_{n}^{-1},\\
\tilde{\upsilon}  &  =\mathcal{I}_{n}^{-2}\frac{1}{n}\sum_{i}
\frac{\frac{1}{2}E[U_{it}^{\alpha\alpha}]^{2}+2E[U_{it}^{\alpha\alpha}%
]E[V_{it}U_{it}^{\alpha}] 
+E[V_{it}^{2}]E[(U_{it}^{\alpha})^{2}]+E[V_{it}U_{it}^{\alpha}]^{2}}{E[V_{it}^{2}%
]^{2}}.
\end{align*}

\end{proposition}

\begin{proof}
See Supplementary Appendix I \ref{proof-for-panel-models}.
\end{proof}

\subsection{Accuracy in estimating the bias}

In Section \ref{sec:SS_var}, it was demonstrated that while the jackknife uses a consistent estimate of the bias term, the split-sample jackknife uses an unbiased, but inconsistent estimate. A similar situation arises in the panel model. For the `leave-one-out' jackknife, the implicit bias estimate is equal to $\frac{1}{T}%
\widehat{b}_{J}=(T-1)\left(  \frac{1}{T}\sum_{t=1}^{T}\widehat{\theta}_{(t)}%
-\widehat{\theta}\right) $, while the split-sample jackknife uses the bias
estimate $\frac{1}{T}\widehat{b}_{SS}=(\bar{\theta}_{SS}-\widehat{\theta})$.
The following
proposition establishes the accuracy of $\widehat{b}_{J}$ and $\widehat{b}_{SS}$
as estimators for $\mathbf{B}=\lim_{n\rightarrow\infty}E[A_2]$.

\begin{proposition}
\label{Thm:bias_est} Let the regularity conditions stated in Appendix \ref{sec:reg-conds-panel} hold. Let $\frac{1}{T}\widehat{b}_{J}=(T-1)(\frac{1}{T}\sum_{t=1}^{T}%
\widehat{\theta}_{(t)}-\widehat{\theta})$ and $\frac{1}{T}\widehat{b}_{SS}=(\bar{\theta}%
_{SS}-\widehat{\theta})$ be the jackknife and split-sample estimators for the
bias term $\mathbf{B}$. Then, 
\begin{align*}
\sqrt{nT}\frac{1}{T}(\widehat{b}_{J} -\mathbf{B})=O_{p}(T^{-1}) \\
\sqrt{nT}\frac{1}{T}(\widehat{b}_{SS}-\mathbf{B})=O_{p}(T^{-1/2}).
\end{align*}
\end{proposition}

\begin{proof}
See Supplementary Appendix  I \ref{proof-for-panel-models}.
\end{proof}

Similarly to the cross-sectional analysis, the jackknife estimates the bias term at a faster rate than the split-sample estimator.
In contrast, in the panel setting the split-sample bias estimate is in fact consistent; it is an average over $n$ unbiased, but inconsistent, estimates of individual-level bias terms.
Nonetheless, Proposition \ref{Thm:bias_est} shows that the jackknife bias
correction affects the third-order, $O_{p}(T^{-1})$, part of the expansion,
while the split-sample bias correction appears as a second-order,
$O_{p}(T^{-1/2})$, term. This implies that the jackknife bias estimate will
only impact the higher-order variance through its covariance with the
first-order term $A_1$, i.e.\ through the term $\operatorname{Cov}\left(
\sqrt{n}A_{1},\sqrt{nT}(\widehat{b}_{J}-\mathbf{B})\right)$. In contrast,
the split-sample bias estimate appears in the higher-order variance both
through its covariance with $A_1$, as well as through its own variance, $\operatorname{Var}\left(\sqrt{n}(\widehat{b}_{SS}-\mathbf{B})\right)$.

It should be noted that the results in Propositions \ref{thm:variance} and \ref{Thm:bias_est} are derived under i.i.d sampling over both individuals and time. In this setting, the results mirror those for the cross sectional analysis. In settings where serial correlation exists, the leave-one-out style jackknife does not provide a valid estimate of the bias, while the  split-sample estimator continues to produce a consistent estimate of the bias, albeit at a slower rate. This is an important advantage of  the method in Dhaene and Jochmans (2015). Although a general result on higher-order properties under serial
correlation is not yet available, it may be reasonable to speculate that the
convergence rate of the estimator of the bias will continue to play an important role. This would suggest that analytical bias corrections involving parametric estimators, when available, are likely to be superior to nonparametric estimators that require estimates of long-run variances.

\subsection{Example 2: panel data\label{sec:example}}

A simple example may help to highlight the results. Suppose that $z_{it}\sim
N\left(  \alpha_{i},\theta\right)  $ are independent across $i$ and $t$. This model was studied by Neyman and Scott (1948). The MLE is given by $\widehat{\theta}=\frac{1}{nT}\sum_{i=1}^{n}\sum_{t=1}^{T}(z_{it}%
-\bar{z}_{i})^{2}$, and a standard calculation gives that $E[\widehat{\theta
}]=\frac{T-1}{T}\theta$ so that the MLE has a bias of $\frac{1}{T}%
\mathbf{B}=-\frac{1}{T}\theta$. It can be shown that the jackknife
bias-corrected estimator has the form
\[
\tilde{\theta}_{J}=\frac{1}{n(T-1)}\sum_{i=1}^{n}\sum_{t=1}^{T}(z_{it}-\bar
{z}_{i})^{2}%
\]
while the split-sample estimator is
\[
\tilde{\theta}_{SS}=2\frac{1}{nT}\sum_{i=1}^{n}\sum_{t=1}^{T}(z_{it}-\bar
{z}_{i})^{2}-\frac{1}{2}\left(  \frac{1}{nM}\sum_{i=1}^{n}\sum_{t=1}%
^{M}(z_{it}-\bar{z}_{i,1})^{2}+\frac{1}{nM}\sum_{i=1}^{n}\sum_{t=M+1}%
^{T}(z_{it}-\bar{z}_{i,2})^{2}\right)
\]
where $T=2M$, and $\bar{z}_{i,1}$ and $\bar{z}_{i,2}$ are the sample means in
the first and second halves of the sample time period. In this simple model,
the formula given in Proposition \ref{thm:variance} gives $V_{J}%
=\frac{2T}{T-1}\theta^{2}$ and $V_{SS}=\frac{2T+4}{T}\theta^{2}$, which we
can easily confirm are also the exact finite sample variances of
the two estimators in this case. Considering the estimation of the bias
itself, we can see that the jackknife bias-correction estimates the bias as
$\frac{1}{T}\widehat{b}_{J}=-\frac{1}{T}\frac{1}{n(T-1)}\sum_{i=1}^{n}%
\sum_{t=1}^{T}(z_{it}-\bar{z}_{i})^{2}$, while the split-sample correction
estimates the bias using $\frac{1}{T}\widehat{b}_{SS}=-\frac{1}{2n}\sum
_{i=1}^{n}  \left(  (\bar{z}_{i,1}-\bar{z}_{i})^{2}+(\bar{z}_{i,2}%
-\bar{z}_{i})^{2}\right)   $. Both are unbiased estimators of the bias
term in this case, i.e.\ $E[\widehat{b}_{J}]=E[\widehat{b}_{SS}]=-\theta$. It
is straightforward to show $\operatorname{Var}\left(  \sqrt{\frac{n}{T}}%
(\widehat{b}_{J}-\mathbf{B})\right)  =\frac{2\theta^{2}}{T(T-1)}$, whereas
$\operatorname{Var}\left(  \sqrt{\frac{n}{T}}(\widehat{b}_{SS}-\mathbf{B}%
)\right)  =\frac{2\theta^{2}}{T}$, so that the variance of the split-sample bias
estimate is larger by a factor $T$, as predicted by Proposition \ref{Thm:bias_est}.




\subsection{Extension to time series data}%

In Proposition \ref{higher_var}, we noted that the higher-order variance of the bias
corrected MLE can be expressed in terms of $\Upsilon$ in (\ref{eq:higher-var-expr}). This
characterization is useful because it is intuitively positive, which is not
obvious from the definition in (\ref{eq:higher-var}). A natural question is
whether an analytical bias correction in a general time series environment
would lead to a higher-order variance of the form $E\left[  A_{1}^{2}\right]
+\frac{1}{n}\left(  \Upsilon+o\left(  n^{-1}\right)  \right)  $, where
$\Upsilon$ is appropriately replaced by the long-run
counterpart of%
\[
E\left[  \left(  \frac{1}{\sqrt{T}}\sum_{t=1}^{T}X_{t}\right)  ^{2}\right]
E\left[  \left(  \frac{1}{\sqrt{T}}\sum_{t=1}^{T}Y_{t}\right)  ^{2}\right]
+\left(  E\left[  \left(  \frac{1}{\sqrt{T}}\sum_{t=1}^{T}X_{t}\right)
\left(  \frac{1}{\sqrt{T}}\sum_{t=1}^{T}Y_{t}\right)  \right]  \right)
^{2}.
\]
Although we are unable to answer the question in its general form, it can be
shown to be valid for a strictly stationary AR(1) model with normally
distributed errors. To be more precise, we go through the formal expansion of
the MLE of the AR(1) model $y_{t}=\theta y_{t-1}+\varepsilon_{t}$, where
$\varepsilon_{t}\sim N\left(  0,\sigma^{2}\right)  $ and $y_{0}\sim N\left(
0,\left.  \sigma^{2}\right/  \left(  1-\theta^{2}\right)  \right)  $.\footnote{The result is available in Supplementary Appendix III.} The
higher-order bias based on the asymptotic expansion for the MLE
yields
$\lim_{T\rightarrow\infty}E\left[  A_{2}\right]  =-2\theta$, so the analytical
bias-corrected estimator takes the form $\widehat{\theta}+\left.
2\widehat{\theta}\right/  T$. This means that the higher-order variance of
the bias-corrected estimator is $\operatorname{Var}\left(  A_{1}\right)
+\frac{1}{T}\operatorname{Var}\left(  A_{2}\right)  +\frac{2}{\sqrt{T}%
}E\left[  A_{1}A_{2}\right]  +\frac{2}{T}E\left[  A_{1}\left(  A_{3}%
+2A_{1}\right)  \right]  $, which is not obviously positive. We show that this
higher-order variance is equal to $\operatorname{Var}\left(  A_{1}\right)
+\frac{1}{T}\left(  \operatorname{Var}\left(  \mathcal{X}\right)
\operatorname{Var}\left(  \mathcal{Y}\right)  +\left(  E\left[
\mathcal{X}\mathcal{Y}\right]  \right)  ^{2}\right)  $ for
$\mathcal{X}=\left(  1-\theta^{2}\right)  \frac{1}{\sqrt{T}}\sum_{t=1}%
^{T}\frac{y_{t-1}\varepsilon_{t}}{\sigma^{2}}$ and $\mathcal{Y}=-\left(
1-\theta^{2}\right)  \frac{1}{\sqrt{T}}\sum_{t=1}^{T}\left(  \frac{y_{t-1}%
^{2}}{\sigma^{2}}-\frac{1}{1-\theta^{2}}\right)$. Analogously to the 
result in Proposition \ref{higher_var_ss}, we show that the split-sample 
jackknife bias-corrected estimator for the AR(1) model has higher-order variance 
that is given by $\operatorname{Var}\left(  A_{1}\right)
+\frac{2}{T}\left(  \operatorname{Var}\left(  \mathcal{X}\right)
\operatorname{Var}\left(  \mathcal{Y}\right)  +\left(  E\left[
\mathcal{X}\mathcal{Y}\right]  \right)  ^{2}\right)$, and so 
the higher-order part of its variance is larger than that of the 
analytical bias-corrected estimator by a factor of two. See 
Propositions \ref{prop-AR1-higher-var} and \ref{prop-AR1-higher-var-SS} in Supplementary Appendix III.
These results are
admittedly specific to the AR(1) model. We conjecture that they carry over to
more general parametric time series models as long as the bias has a closed
form parametric expression that can be estimated at parametic rates. On the
other hand, estimators with non-parametrically estimated bias corrections that
typically involve estimated long run variances may not share the same
efficiency properties. We conjecture the same is true for non-parametric
block bootstrap procedures and other subsampling techniques used to estimate
the higher-order bias.%




\section{Monte Carlo analysis\label{section-MC}}

To highlight the relevance of the results in a more practical setting,
we conduct two Monte Carlo exercises. In the cross-sectional setting we estimate a marginal treatment effect (MTE) model as in Heckman and Vytlacil (2005). The simulation design follows that used in Catteneo et.\ al.\ (2019) and is a simplified model with no covariates. The treatment $T_i$ is assigned according to $T_i = 1\{ P(Z_i) \geq V_i \}$, where $P(Z_i)$ is a propensity score (which is a function of observed instrumental variables $Z_i$), and $V_i \sim \text{Uniform}[0,1]$ is an unobservable shock that is correlated with potential outcomes and generates selection. Potential outcomes under treated and control states are generated according to
\begin{align*}
    Y_i(0) &= U_{0i}, \quad U_{0i} \vert Z_i, V_i \sim \text{Uniform}[-1,1],\\
    Y_i(1) &= 0.5 + U_{1i}, \quad U_{1i} \vert Z_i, V_i  \sim \text{Uniform}[-0.5,1.5-2V_i].
\end{align*}
The observed outcome variable is given by $Y_i = T_i Y_i(1) + (1-T_i)Y_i(0)$. We create a set of 19 potential instruments, $Z_{j,i}\sim \text{Uniform}[0,1]$ for $j=1,\dots,19$, to be used in estimation, although the true propensity score only depends on the first four of these (in addition to the constant) $P(Z_i)  = 0.1 + Z_{1,i}+ Z_{2,i}+ Z_{3,i}+ Z_{4,i}$.

The MTE function is defined as $MTE(v) = E[Y_1 - Y_0 \vert V=v]$, and measures the average treatment effect for individuals with a given level of unobserved resistance to treatment $V_i=v$. Many objects of interest can be represented as weighted averages of the MTE function, although for the purposes of this simulation we follow Catteneo et.\ al.\ (2019) and focus on $\theta=MTE(0.5)$. The MTE function can be identified as $MTE(v)=\partial E[Y_i \vert P(Z_i)=p]/\partial p \vert_{p=v}$.

To estimate the marginal treatment effect parameter, the propensity score $P(z)$ is first estimated via regression of the treatment dummy variable $T_i$ on $Z_i(k)=(1, Z_{1,i},\dots,Z_{k-1,i})$, where $k$ ranges from 5 to 20. The second step then regresses the outcome $Y_i$ on a quadratic in the propensity score, $Y_i = \beta_1 + \beta_2\widehat{P}(z) + \beta_3\widehat{P}(z)^2$. The scalar parameter of interest $\theta$ is the derivative of this function at $p=0.5$, i.e. $\widehat{\theta} = \widehat{\beta}_2+\widehat{\beta}_3$.

\begin{table}[!h]
\caption{Simulation of marginal treatment effect estimates}%
\label{tab:sims_mte}
\centering \small
\begin{threeparttable}

\begin{tabular}{cccccccc}
\hline 
 &  & \multicolumn{3}{c}{Conventional} & \multicolumn{3}{c}{Bootstrap} \tabularnewline
n & k & Bias & SD & RMSE & Bias & SD & RMSE \tabularnewline
\hline 
1000 & 5  & 0.441 & 4.922 & 4.941 & 0.099 & 5.135 & 5.135 \tabularnewline
     & 10 & 1.059 & 4.712 & 4.828 & 0.184 & 5.381 & 5.383 \tabularnewline
     & 15 & 1.570 & 4.460 & 4.727 & 0.326 & 5.447 & 5.456 \tabularnewline
     & 20 & 1.952 & 4.231 & 4.659 & 0.425 & 5.435 & 5.450 \tabularnewline
\hline 
2000 & 5  & 0.317 & 4.730 & 4.740 & 0.074 & 4.830 & 4.830 \tabularnewline
     & 10 & 0.798 & 4.636 & 4.703 & 0.122 & 4.964 & 4.964 \tabularnewline
     & 15 & 1.226 & 4.505 & 4.668 & 0.191 & 5.021 & 5.024 \tabularnewline
     & 20 & 1.604 & 4.421 & 4.702 & 0.264 & 5.105 & 5.110 \tabularnewline
\hline 
3000 & 5  & 0.155 & 4.756 & 4.757 & -0.047 & 4.822 & 4.821 \tabularnewline
     & 10 & 0.582 & 4.666 & 4.701 & 0.008 & 4.881 & 4.880 \tabularnewline
     & 15 & 0.971 & 4.574 & 4.675 & 0.069 & 4.924 & 4.923 \tabularnewline
     & 20 & 1.300 & 4.517 & 4.700 & 0.103 & 4.995 & 4.994 \tabularnewline
\hline 
&  & \multicolumn{3}{c}{Jackknife} & \multicolumn{3}{c}{Split-Sample} \tabularnewline
n & k & Bias & SD & RMSE & Bias & SD & RMSE \tabularnewline
\hline 
1000 & 5  & 0.105 & 5.146 & 5.146 & 0.092 & 5.426 & 5.426 \tabularnewline
     & 10 & 0.254 & 5.320 & 5.325 & 0.243 & 5.683 & 5.686 \tabularnewline
     & 15 & 0.511 & 5.274 & 5.297 & 0.459 & 5.611 & 5.628 \tabularnewline
     & 20 & 0.713 & 5.167 & 5.215 & 0.716 & 5.563 & 5.607 \tabularnewline
\hline 
2000 & 5  & 0.075 & 4.840 & 4.839 & 0.087 & 5.012 & 5.012 \tabularnewline
     & 10 & 0.154 & 4.953 & 4.954 & 0.175 & 5.185 & 5.187 \tabularnewline
     & 15 & 0.275 & 4.972 & 4.979 & 0.260 & 5.204 & 5.210 \tabularnewline
     & 20 & 0.423 & 5.014 & 5.031 & 0.395 & 5.267 & 5.281 \tabularnewline
\hline 
3000 & 5  & -0.044 & 4.827 & 4.826 & -0.068 & 4.924 & 4.923 \tabularnewline
     & 10 & 0.025 & 4.879 & 4.877 & -0.024 & 5.006 & 5.005 \tabularnewline
     & 15 & 0.122 & 4.899 & 4.899 & 0.068 & 5.093 & 5.092 \tabularnewline
     & 20 & 0.202 & 4.951 & 4.953 & 0.243 & 5.180 & 5.184 \tabularnewline
\hline 
\end{tabular}

\begin{tablenotes} \footnotesize
\item[*] Results of estimators over 2000 simulations. \emph{Conventional} denotes the standard two-step estimator described in the text; \emph{Jackknife} denotes the leave-one-out jackknife bias-corrected estimator; \emph{Bootstrap} denotes the bootstrap bias-corrected estimator based on $n/2$ bootstrap draws. 
\end{tablenotes}
\end{threeparttable}
\end{table}

We investigate the performance of the bootstrap, jackknife, and split-sample bias corrections. As predicted by the theory, the standard deviation of the bootstrap and jackknife bias-corrected estimators are very similar across the different simulation settings. This is particularly true as the sample size grows large; for $n=3000$ the jackknife and bootstrap bias corrections are very close in terms of standard deviation. The split-sample estimator has larger standard deviation than the other two bias corrections, as expected given its larger higher-order variance. As the sample size grows, the difference decreases; this is expected given that the estimators all have the same asymptotic variance. 

As a panel data example, we estimate a probit model with strictly exogenous covariates and individual fixed effects.
\begin{align*}
y_{it}  &  = 1\{\theta_{0}'x_{it}+\alpha_{i}+\varepsilon_{it}>0\},\quad
\varepsilon_{it}\sim N(0,1)
\end{align*}
The simulation is calibrated to the female labor force participation application of Fernandez-Val (2009), and is the same as that used in Fernandez-Val and Weidner (2018). Here the outcome is an indicator for participation in the labor force, and the covariates include three measures of fertility, the number of children aged 0-2, 3-5, and 6-17 years, as well as the log of husband's income, and a quadratic in age. We focus on the coefficients on the three fertility variables. Below we report the results of simulations drawn from a sample of $n=500$ individuals and $T=\{4,8\}$ time periods.\footnote{Fernandez-Val and Weidner (2018) report results using the $n=664$ and $T=9$, which matches the sample size in the PSID data set, and find similar results.}

Table \ref{tab:sims_beta} reports the results for the biased MLE fixed effects estimator as well as three bias corrections: the analytical correction, leave-one-out jackknife, and the split-sample jackknife. We report the bias, standard deviation and root mean-squared error as a percentage of the true coefficient values, and the rejection rate for a test with 5\% significance. The MLE fixed effects estimator has large bias and rejection rates as large as 54\%. As is evident from the theory, the size of the bias is decreasing with the number of time periods. Both the jackknife and analytical bias corrections lead to significant reductions in the bias with no cost in precision; in fact, both bias corrected estimators have smaller standard deviations than the MLE. In contrast, while the split-sample jackknife also reduces bias (although to a lesser degree than the other corrections), it has substantially larger variance and mean-squared error. It is evident that, even for $T=8$, the impact of higher-order differences in the bias corrections remains important for the finite sample properties of the estimator.

\begin{table}[t]
\caption{Simulation of probit model with individual fixed effects}%
\label{tab:sims_beta}
\centering \small
\begin{threeparttable}
\begin{tabularx}{0.93\textwidth}{@{}
                       l l S[table-format=4.1,table-column-width=1cm]
                         S[table-format=3.1,table-column-width=1cm]
                         S[table-format=3.1,table-column-width=1cm]
                         S[table-format=1.2,table-column-width=1cm] l 
                         S[table-format=3.1,table-column-width=1cm]
                         S[table-format=3.1,table-column-width=1cm]
                         S[table-format=3.1,table-column-width=1cm]
                         S[table-format=1.2,table-column-width=1cm] @{}}
\toprule
&  & \multicolumn{4}{c}{MLE} & & \multicolumn{4}{c}{Analytical}\\
& & {Bias} & {SD} & {RMSE} & {Rej 5\%} & & {Bias} & {SD} & {RMSE} & {Rej 5\%} \\ 
\hline
T=4 & Ages 0-2 & -41.9 & 24.7 & 48.6 & 0.54 & & -8.2 & 19.0 & 20.7 & 0.06\\
& Ages 3-5 & -42.4 & 47.8 & 63.9 & 0.24 & & -8.2 & 37.4 & 38.3 & 0.05\\
& Ages 6-17 & -42.7 & 132.1 & 138.8 & 0.12 & & -1.1 & 102.9 & 102.9 & 0.03\\ \\[-0.8em]

T=8 & Ages 0-2 & -16.9 & 11.2 & 20.3 & 0.36 & & -3.9 & 10.0 & 10.7 & 0.05\\
& Ages 3-5 & -16.8 & 18.6 & 25.1 & 0.18 & & -3.7 & 16.6 & 17.0 & 0.05\\
& Ages 6-17 & -18.6 & 50.7 & 54.0 & 0.09 & & -4.8 & 45.2 & 45.5 & 0.05\\
\midrule
&  & \multicolumn{4}{c}{Jackknife} & & \multicolumn{4}{c}{Split-sample}\\
&  & {Bias} & {SD} & {RMSE} & {Rej 5\%} & & {Bias} & {SD} & {RMSE} & {Rej 5\%} \\ 
\hline
T=4 & Ages 0-2 & 21.2 & 16.7 & 27.0 & 0.14 & & 37.7 & 49.5 & 62.2 & 0.53\\
& Ages 3-5 & 20.9 & 30.6 & 37.1 & 0.03 & & 35.2 & 99.5 & 105.5 & 0.46\\
& Ages 6-17 & 22.4 & 88.8 & 91.6 & 0.02 & & 31.2 & 303.7 & 305.3 & 0.44\\  \\[-0.8em]
T=8 & Ages 0-2 & 4.3 & 9.3 & 10.2 & 0.04 & & 7.1 & 16.3 & 17.8 & 0.25\\
& Ages 3-5 & 4.4 & 15.6 & 16.2 & 0.04 & & 7.1 & 30.6 & 31.4 & 0.27\\
& Ages 6-17 & 3.1 & 42.2 & 42.3 & 0.03 & & 8.6 & 82.1 & 82.5 & 0.28\\
\bottomrule
\end{tabularx}

\begin{tablenotes} \footnotesize
\item[*] Results of estimators over 1000 simulations. Bias, SD, and RMSE are percentages of the true parameter values. 
\end{tablenotes}
\end{threeparttable}
\end{table}

\section{Summary\label{section-summary}}

We show that the choice of bias correction method does not affect the higher-order
variance of any parametric estimator that is semiparametrically efficient, as long as the bias
estimator is asymptotically linear, i.e. asymptotically equivalent to a sample average. We give a formal expansion
showing this property in a general semiparametric model. We also prove that the bootstrap, jackknife and a version of the analytical bias estimates are asymptotically linear, when the estimator of the parameters of
interest has a standard form of stochastic expansion (which is known to exist for a large class
of models). The result implies that a researcher
may choose a bias correction for an efficient estimator based on computational convenience, or
some other criteria, without affecting its higher-order efficiency.  

We have verified this result using derivations of the asymptotic expansion and higher-order variance for maximum likelihood
estimation of a parametric model, using a bootstrap, jackknife, and three forms of analytical bias corrections. Furthermore, we found that the third-order stochastic expansion of the bootstrap, jackknife, and one type of analytical bias-corrected MLE are identical, and hence have an even stronger higher-order equivalence property.

These results show that the higher-order efficiency of bias-corrected efficient estimators does
not depend on the form of the bias correction, as long as the estimate of the bias term is asymptotically linear (and hence $\sqrt{n}$-consistent). Thus, in practice one might use
whatever bias correction method is most convenient. An important caveat
is that the split-sample jackknife estimator does not estimate the bias at the $\sqrt{n}$-rate, and so is not asymptotically linear, and we show the resultant higher-order variance to be
strictly larger in an i.i.d.\ setting, suggesting the importance of the accuracy in estimating the bias.

We generalized the result to the analysis of a panel data model with fixed effects, and established that
the split-sample bias-corrected estimator has larger higher-order variance
than the jackknife estimator, confirming the importance of the accuracy in
estimation of the bias even in panel settings. In non-i.i.d.\ settings, the
standard jackknife cannot be used. Comparison of the split-sample correction
with alternatives, such as the analytical correction given in Hahn and
Kuersteiner (2002, 2011), is a topic that we leave for future research.

\begin{appendices}
    
\section{Proof of Theorem 2} \label{app:main}

We assume that $\widehat{\theta}$ is an estimator with a stochastic expansion
\begin{align}\label{eq:expansion}
\sqrt{n}(\widehat{\theta}-\theta_{0}) & =A_{1}+\frac{1}{\sqrt{n}}A_{2}+\frac{1}{n}A_{3}\\
 & \quad+\frac{1}{n^{3/2}}A_{4}+\frac{1}{n^{2}}A_{5}+o_{p}(n^{-2})
\end{align}
where $n^{-k/2}A_{k}$ is a $k$-th order V-statistic, i.e.
\begin{equation}
\frac{1}{n^{k/2}}A_{k}=\big(\frac{1}{n}\sum_{i=1}^{n}g_{k,1}(z_{i})\big)\cdots\big(\frac{1}{n}\sum_{i=1}^{n}g_{k,k}(z_{i})\big)\label{eq:V-stat}
\end{equation}
and $g_{k,j}$ are functions of the data, evaluated at $\theta_{0}$
with expectation zero, i.e $E[g_{j,k}(z_{i})]=0$.

The following lemma derives the first-order expansions of the jackknife, bootstrap and (sample-average) analytical bias estimators, from which asymptotic linearity follows.

\begin{lemma}
Define $b_{0}=E[g_{2,1}(z_{i})g_{2,2}(z_{i})]$ as the higher-order
bias of $\widehat{\theta}$. Then, the jackknife bias estimate $\widehat{b}_{J}$
satisfies
\begin{align*}
\sqrt{n}(\widehat{b}_{J}-b_{0}) & =\frac{1}{\sqrt{n}}\sum_{i}\big(g_{2,1}(z_{i})g_{2,2}(z_{i})-E[g_{2,1}(z_{i})g_{2,2}(z_{i})]\big)\\
 & \quad+E[g_{3,1}(z_{i})g_{3,2}(z_{i})]\frac{1}{\sqrt{n}}\sum_{i}g_{3,3}(z_{i})+E[g_{3,1}(z_{i})g_{3,3}(z_{i})]\frac{1}{\sqrt{n}}\sum_{i}g_{3,2}(z_{i})\\
 & \quad+E[g_{3,2}(z_{i})g_{3,3}(z_{i})]\frac{1}{\sqrt{n}}\sum_{i}g_{3,1}(z_{i})+o_{p}(1)
\end{align*}

and the bootstrap and analytical bias corrections satisfy
\begin{align*}
\sqrt{n}(\widehat{b}_{B}-b_{0}) & =\frac{1}{\sqrt{n}}\sum_{i}\big(g_{2,1}(z_{i})g_{2,2}(z_{i})-E[g_{2,1}(z_{i})g_{2,2}(z_{i})]\big)\\
 & \quad+\big(E[g_{2,1}(z_{i})h_{2}(z_{i})]+E[g_{2,2}(z_{i})h_{1}(z_{i})]\big)\frac{1}{\sqrt{n}}\sum_{i}g_{1,1}(z_{i})+o_{p}(1)
\end{align*}
where $h_{1}$ and $h_{2}$ are derivatives of $g_{2,1}$ and $g_{2,2}$
respectively. That is, they are all asymptotically linear estimators
for $b_{0}$.
\end{lemma}

\begin{proof}

\emph{\bf Jackknife bias correction}
A jackknife estimate of the higher-order bias of $\widehat{\theta}$
is given by
\[
\frac{\widehat{b}_{J}}{n}=(n-1)(\frac{1}{n}\sum_{i=1}^{n}\widehat{\theta}_{(i)}-\widehat{\theta})
\]
where $\widehat{\theta}_{(i)}$ is the estimator that excludes observation
$i$. Note that $\widehat{\theta}_{(i)}$ has an equivalent expansion
to $\widehat{\theta}$, with terms $A_{k,(i)}$ that are the same
as the terms in the original expansion, simply dropping observation
$i$.

We may write
\begin{align*}
\frac{\widehat{b}_{J}}{n} & =(n-1)\Big((\frac{1}{n}\sum_{i=1}^{n}\frac{1}{\sqrt{n-1}}A_{1,(i)}-\frac{1}{\sqrt{n}}A_{1})
 +(\frac{1}{n}\sum_{i=1}^{n}\frac{1}{n-1}A_{2,(i)}-\frac{1}{n}A_{2})\\
 & \qquad+(\frac{1}{n}\sum_{i=1}^{n}\frac{1}{(n-1)^{3/2}}A_{3,(i)}-\frac{1}{n^{3/2}}A_{3})\Big)
 +(\frac{1}{n}\sum_{i=1}^{n}\frac{1}{(n-1)^{2}}A_{4,(i)}-\frac{1}{n^{2}}A_{4})\Big)\\
 & \qquad+(\frac{1}{n}\sum_{i=1}^{n}\frac{1}{(n-1)^{5/2}}A_{5,(i)}-\frac{1}{n^{5/2}}A_{5})\Big)+o_{p}(n^{-5/2})\\
 & =\frac{1}{n}\Big(\tilde{B}_{1}+\frac{1}{n^{1/2}}\tilde{B}_{2}+\frac{1}{n}\tilde{B}_{3}+\frac{1}{n^{3/2}}\tilde{B}_{4}+\frac{1}{n^{2}}\tilde{B}_{5}\Big)+o_{p}(n^{-5/2}).
\end{align*}

By Lemma 21 in Supplementary Appendix II F we have $\tilde{B}_{1} =0$, while Lemma 22 gives
\begin{align*}
\tilde{B}_{2} & =\frac{1}{\sqrt{n}}\sum_{i}g_{2,1}(z_{i})g_{2,2}(z_{i})-\frac{1}{\sqrt{n}(n-1)}\sum_{i}\sum_{j\ne i}g_{2,1}(z_{i})g_{2,2}(z_{j})\\
 & =\frac{1}{\sqrt{n}}\sum_{i}g_{2,1}(z_{i})g_{2,2}(z_{i})+ o_p(1).
\end{align*}
For the next term, using Lemma 23 from Supplementary Appendix II F, it is straightforward to show that
\begin{align*}
\frac{1}{\sqrt{n}}\tilde{B}_{3} & =\frac{2n-1}{n^{3/2}(n-1)}\sum_{i}g_{3,1}(z_{i})g_{3,2}(z_{i})g_{3,3}(z_{i})\\
 & +\frac{n^{2}-3n+1}{n^{3/2}(n-1)^{2}}\sum_{i}\sum_{j\ne i}\Big(g_{3,1}(z_{i})g_{3,2}(z_{i})g_{3,3}(z_{j})+g_{3,1}(z_{i})g_{3,2}(z_{j})g_{3,3}(z_{i})\\
 & \qquad+g_{3,1}(z_{j})g_{3,2}(z_{i})g_{3,3}(z_{i})\Big)\\
 & -\frac{3n-1}{n^{3/2}(n-1)^{2}}\sum_{i}\sum_{j\ne i}\sum_{k\ne\{i,j\}}g_{3,1}(z_{i})g_{3,2}(z_{j})g_{3,3}(z_{k})\\
 & =\frac{1}{n^{3/2}}\sum_{i}\sum_{j\ne i}\Big(g_{3,1}(z_{i})g_{3,2}(z_{i})g_{3,3}(z_{j})\\
 & \qquad+g_{3,1}(z_{i})g_{3,2}(z_{j})g_{3,3}(z_{i})+g_{3,1}(z_{j})g_{3,2}(z_{i})g_{3,3}(z_{i})\Big)+o_{p}(1)\\
 & =E[g_{3,1}(z_{i})g_{3,2}(z_{i})]\frac{1}{\sqrt{n}}\sum_{i}g_{3,3}(z_{i})+E[g_{3,1}(z_{i})g_{3,3}(z_{i})]\frac{1}{\sqrt{n}}\sum_{i}g_{3,2}(z_{i})\\
 & \quad+E[g_{3,2}(z_{i})g_{3,3}(z_{i})]\frac{1}{\sqrt{n}}\sum_{i}g_{3,1}(z_{i})+o_{p}(1).
\end{align*}

Similar results show that $\tilde{B}_{4}$ and $\tilde{B}_{5}$ are
also $O_{P}(1)$ (see for example Section \ref{subsec:nomV-def} in Supplementary Appendix
IV for results on jackknife V-statistics up to sixth order).

Using these results, we may then write
\begin{align*}
\sqrt{n}(\widehat{b}_{J}-b_{0}) & =\tilde{B}_{2}-\sqrt{n}b_{0}+\frac{1}{\sqrt{n}}\tilde{B}_{3}+o_{p}(1)\\
 & =\frac{1}{\sqrt{n}}\sum_{i}\big(g_{2,1}(z_{i})g_{2,2}(z_{i})-E[g_{2,1}(z_{i})g_{2,2}(z_{i})]\big)\\
 & \quad+E[g_{3,1}(z_{i})g_{3,2}(z_{i})]\frac{1}{\sqrt{n}}\sum_{i}g_{3,3}(z_{i})+E[g_{3,1}(z_{i})g_{3,3}(z_{i})]\frac{1}{\sqrt{n}}\sum_{i}g_{3,2}(z_{i})\\
 & \quad+E[g_{3,2}(z_{i})g_{3,3}(z_{i})]\frac{1}{\sqrt{n}}\sum_{i}g_{3,1}(z_{i})+o_{p}(1)
\end{align*}
which gives the result in the proposition.

\emph{\bf Bootstrap bias correction}

Given the expansion in (\ref{eq:expansion}), the bootstrap estimate
$\widehat{\theta}^{*}$ has an equivalent expansion with reference
to the empirical distribution $\widehat{F}_{n}$, rather than the
population distribution $F_{0}$.
\begin{equation}
\begin{aligned}\sqrt{n}(\widehat{\theta}^{*}-\widehat{\theta}) & =\widehat{A}_{1}+\frac{1}{\sqrt{n}}\widehat{A}_{2}+\frac{1}{n}\widehat{A}_{3}+o_{p}(n^{-1})\end{aligned}
\label{eq:exp_boot}
\end{equation}
where
\begin{equation}
\frac{1}{n^{k/2}}\widehat{A}_{k}=\big(\frac{1}{n}\sum_{i=1}^{n}\widehat{g}_{k,1}(z_{i}^{*})\big)\cdots\big(\frac{1}{n}\sum_{i=1}^{n}\widehat{g}_{k,k}(z_{i}^{*})\big)\label{eq:V-stat-1}
\end{equation}
with $\widehat{g}_{k,j}$ the same function as in (\ref{eq:V-stat}),
evaluated at $\widehat{\theta}$ and $\widehat{F}_{n}$ rather than
$\theta_{0}$ and $F_{0}$. Note that in this expansion terms are zero mean with respect
to $\widehat{F}_{n}$, i.e. $\sum_{i}\widehat{g}_{k,j}(z_{i})=0$.

We can define a bootstrap bias estimate as $\widehat{b}_{B}/n=E^{*}[\widehat{\theta}^{*}-\widehat{\theta}]$,
where $E^{*}$ is expectation over the bootstrap distribution (i.e.
$\widehat{F}_{n}$). We may write
\begin{align*}
\frac{\widehat{b}_{B}}{n} & =\frac{1}{\sqrt{n}}E^{*}[\widehat{A}_{1}]+\frac{1}{n}E^{*}[\widehat{A}_{2}]+\frac{1}{n^{3/2}}E^{*}[\widehat{A}_{3}]+o_{p}(n^{-3/2}).
\end{align*}

From Lemma 17 in Supplementary Appendix II F we have that
\begin{align*}
E^{*}[\widehat{A}_{1}] & =E^{*}[\frac{1}{\sqrt{n}}\sum_{i=1}^{n}\widehat{g}_{1,1}(z_{i}^{*})]
 =\frac{1}{\sqrt{n}}\sum_{i=1}^{n}\widehat{g}_{1,1}(z_{i})
 =0.
\end{align*}
Lemma 18 gives
\begin{align*}
E^{*}[\widehat{A}_{2}] & =E^{*}[\frac{1}{n}\sum_{i}\sum_{j}\widehat{g}_{2,1}(z_{i}^{*})\widehat{g}_{2,2}(z_{j}^{*})]\\
 & =\frac{1}{n}\sum_{i}\widehat{g}_{2,1}(z_{i})\widehat{g}_{2,2}(z_{i})+\frac{n-1}{n^{2}}\sum_{i}\widehat{g}_{2,1}(z_{i})\sum_{j}\widehat{g}_{2,2}(z_{j})\\
 & =\frac{1}{n}\sum_{i}\widehat{g}_{2,1}(z_{i})\widehat{g}_{2,2}(z_{i})
\end{align*}
and similarly, Lemma 19 gives
\begin{align*}
E^{*}[\widehat{A}_{3}] & =E^{*}[\frac{1}{n^{3/2}}\sum_{i}\sum_{j}\sum_{k}g_{3,1}(z_{i}^{*})g_{3,2}(z_{j}^{*})g_{3,3}(z_{k}^{*})]\\
 & =\frac{1}{n^{3/2}}\sum_{i}\widehat{g}_{3,1}(z_{i})\widehat{g}_{3,2}(z_{i})\widehat{g}_{3,3}(z_{i}).
\end{align*}
Using these results we can write
\begin{align*}
\sqrt{n}(\widehat{b}_{B}-b_{0}) & =\frac{1}{\sqrt{n}}\sum_{i}\big(\widehat{g}_{2,1}(z_{i})\widehat{g}_{2,2}(z_{i})-E[g_{2,1}(z_{i})g_{2,2}(z_{i})]\big)+o_{p}(1).
\end{align*}
Next, note that
\begin{align*}
\frac{1}{\sqrt{n}} & \sum_{i}\big(\widehat{g}_{2,1}(z_{i})\widehat{g}_{2,2}(z_{i})-g_{2,1}(z_{i})g_{2,2}(z_{i})\big)\\
 & =\frac{1}{\sqrt{n}}\sum_{i}g_{2,1}(z_{i})\big(\widehat{g}_{2,2}(z_{i})-g_{2,2}(z_{i})\big)\\
 & +\frac{1}{\sqrt{n}}\sum_{i}\big(\widehat{g}_{2,1}(z_{i})-g_{2,1}(z_{i})\big)g_{2,2}(z_{i})\\
 & +\frac{1}{\sqrt{n}}\sum_{i}\big(\widehat{g}_{2,1}(z_{i})-g_{2,1}(z_{i})\big)\big(\widehat{g}_{2,2}(z_{i})-g_{2,2}(z_{i})\big).
\end{align*}

Let $h_{1}(z_{i})$ and $h_{2}(z_{i})$ be first derivatives of $g_{2,1}(z)$
and $g_{2,2}(z_{i})$ with respect to $\theta$ (evaluated at $\theta_{0}$),
so that first-order expansions of $\widehat{g}_{2,1}$ and $\widehat{g}_{2,2}$
are given by
\begin{align*}
\sqrt{n}\big(\widehat{g}_{2,1}(z_{i})-g_{2,1}(z_{i})\big) & =h_{1}(z_{i})\sqrt{n}(\widehat{\theta}-\theta_{0})+o_{p}(1)\\
\sqrt{n}\big(\widehat{g}_{2,2}(z_{i})-g_{2,2}(z_{i})\big) & =h_{2}(z_{i})\sqrt{n}(\widehat{\theta}-\theta_{0})+o_{p}(1).
\end{align*}
Then, we can write
\begin{align*}
\frac{1}{n}\sum_{i}g_{2,1}(z_{i})\big(\widehat{g}_{2,2}(z_{i})-g_{2,2}(z_{i})\big) & =\frac{1}{n}\sum_{i}g_{2,1}(z_{i})h_{2}(z_{i})\sqrt{n}(\widehat{\theta}-\theta_{0})+o_{p}(1)\\
 & =E[g_{2,1}(z_{i})h_{2}(z_{i})]\frac{1}{\sqrt{n}}\sum_{i}g_{1,1}(z_{i})+o_{p}(1)
\end{align*}
and similarly for $\frac{1}{\sqrt{n}}\sum_{i}\big(\widehat{g}_{2,1}(z_{i})-g_{2,1}(z_{i})\big)g_{2,2}(z_{i})$
so that
\begin{align*}
\sqrt{n}(\widehat{b}_{B}-b_{0}) & =\frac{1}{\sqrt{n}}\sum_{i}\big(g_{2,1}(z_{i})g_{2,2}(z_{i})-E[g_{2,1}(z_{i})g_{2,2}(z_{i})]\big)\\
 & \quad+\big(E[g_{2,1}(z_{i})h_{2}(z_{i})]+E[g_{2,2}(z_{i})h_{1}(z_{i})]\big)\frac{1}{\sqrt{n}}\sum_{i}g_{1,1}(z_{i})+o_{p}(1)
\end{align*}
giving the result of the proposition.

\emph{\bf Analytical bias correction}

Under (\ref{eq:expansion}) and (\ref{eq:V-stat}) the bias term has
the form $b_{0}=E[g_{2,1}(z_{i})g_{2,2}(z_{i})]$, 
for some functions $g_{2,1}$ and $g_{2,2}$. Assume that we can construct
consistent estimators of these functions, $\widehat{g}_{2,1}$ and
$\widehat{g}_{2,2}$ by plugging in $\widehat{\theta}$ in place of
$\theta_{0}$ and replacing expectations with sample means. This implies
$\widehat{g}_{2,1}$ and $\widehat{g}_{2,2}$ are the same functions
as in the bootstrap expansion above. We can form a bias estimate using
\[
\widehat{b}_{a}=\frac{1}{n}\sum_{i}\widehat{g}_{2,1}(z_{i})\widehat{g}_{2,2}(z_{i}).
\]
This then gives
\begin{align*}
\sqrt{n}(\widehat{b}_{a}-b_{0}) & =\frac{1}{\sqrt{n}}\sum_{i}\big(\widehat{g}_{2,1}(z_{i})\widehat{g}_{2,2}(z_{i})-E[g_{2,1}(z_{i})g_{2,2}(z_{i})]\big)
\end{align*}

and the result follows from the bootstrap result above.

\end{proof}

\subsection{Split-sample bias correction}
Here we show that, under the same asymptotic expansion structure used above, the split-sample bias estimate is not asymptotically linear. The split-sample bias estimate is given by
$
\frac{\widehat{b}_{ss}}{n}=\frac{1}{2}(\widehat{\theta}_{1}+\widehat{\theta}_{2})-\widehat{\theta}
$. 
Again, we can construct an expansion for this bias estimate from the
expansion of $\widehat{\theta}$. 
\[
\frac{\widehat{b}_{ss}}{n}=\frac{1}{\sqrt{n}}\tilde{B}_{1}+\frac{1}{n}\tilde{B}_{2}+o_{p}(n^{-1})
\]
Let $m=n/2$. We have
\begin{align*}
\tilde{B}_{1}&=\frac{1}{2}(\widehat{A}_{1,1}+\widehat{A}_{1,2})-\widehat{A}_{1}=0 \\
\tilde{B}_{2} & =\frac{1}{2}(\widehat{A}_{2,1}+\widehat{A}_{2,2})-\widehat{A}_{2}\\
 & =\frac{1}{2}\Big(\frac{4}{n}\sum_{i=1}^{m}\sum_{j=1}^{m}g_{2,1}(z_{i})g_{2,2}(z_{j})+\frac{4}{n}\sum_{i=m+1}^{n}\sum_{j=m+1}^{n}g_{2,1}(z_{i})g_{2,2}(z_{j})\Big)\\
 & \quad-\frac{1}{n}\sum_{i=1}^{n}\sum_{j=1}^{n}g_{2,1}(z_{i})g_{2,2}(z_{j})\\
 & =\frac{1}{n}\sum_{i=1}^{m}\sum_{j=1}^{m}g_{2,1}(z_{i})g_{2,2}(z_{j})+\frac{1}{n}\sum_{i=m+1}^{n}\sum_{j=m+1}^{n}g_{2,1}(z_{i})g_{2,2}(z_{j})\\
 & \quad-\frac{1}{n}\sum_{i=1}^{m}\sum_{j=m+1}^{n}g_{2,1}(z_{i})g_{2,2}(z_{j})-\frac{1}{n}\sum_{i=m+1}^{n}\sum_{j=1}^{m}g_{2,1}(z_{i})g_{2,2}(z_{j}).
\end{align*}

This then gives 
\begin{align*}
\widehat{b}_{ss}-b_{0} & =\frac{1}{n}\sum_{i=1}^{m}\sum_{j=1,j\ne i}^{m}g_{2,1}(z_{i})g_{2,2}(z_{j})+\frac{1}{n}\sum_{i=m+1}^{n}\sum_{j=m+1,j\ne i}^{n}g_{2,1}(z_{i})g_{2,2}(z_{j})\\
 & \quad-\frac{1}{n}\sum_{i=1}^{m}\sum_{j=m+1}^{n}g_{2,1}(z_{i})g_{2,2}(z_{j})-\frac{1}{n}\sum_{i=m+1}^{n}\sum_{j=1}^{m}g_{2,1}(z_{i})g_{2,2}(z_{j})\\
 & +o_{p}(1).
\end{align*}
The terms on the RHS are each $O_{p}(1)$ so that the split-sample
bias estimator is inconsistent, and has a U-statistic structure to
first-order, and hence cannot be asymptotically linear.

\section{Regularity conditions} 
\subsection{Conditions for cross-sectional models} \label{sec:reg-conds}

\begin{condition}
\label{HH}(i) The function $\log f\left(  z,\theta\right)  $ is $7$ times
continuously differentiable on $\Theta$ for each $z$; (ii) The parameter space
$\Theta\subset\mathbb{R}$ is a compact set, $\theta_{0}\in\operatorname{int}
(\Theta)$; (iii) there exists a function $M\left(  z\right)  $
such that for all $\theta\in\Theta$
\[
\left\vert \frac{\partial^{m}\log f\left(  z,\theta\right)  }{\partial
\theta^{m}}\right\vert \leq M\left(  z\right)  \qquad0\leq m\leq7
\]
and $E\left[  M\left(  Z_{i}\right)  ^{Q}\right]  <\infty$ for some $Q>16$;
(iv) If $\theta\neq\theta_{0}$ then $f\left(  Z_{i},\theta\right)  \neq
f\left(  Z_{i},\theta_{0}\right)  .$
\end{condition}

\begin{condition}
\label{M}For each $\theta\in\Theta$ and for $m\leq7,$ $\partial^{m}\log
f\left(  z,\theta\right)  /\partial\theta^{m}$ is a $P$-measurable function of
$z$.
\end{condition}

\begin{condition}
\label{EC}Let $\mathfrak{F}$ be the class of functions $\partial^{m}\log
f\left(  z,\theta\right)  /\partial\theta^{m}$ indexed by $\theta\in\Theta$
for $m=1,..,7$ with envelope $M\left(  z\right)  $. Then,
\begin{equation}
\int_{0}^{1}\sup_{\mathcal{Q}\in\mathfrak{P}}\sqrt{\log N\left(
\varepsilon\left(  \int M^{2}d\mathcal{Q}\right)  ^{1/2},\mathfrak{F}%
,L_{2}(\mathcal{Q})\right)  }d\varepsilon<\infty, \label{Entropy}%
\end{equation}
where $\mathfrak{P}$ is the class of probability measures on $\mathbb{R}$ that
concentrate on a finite set and $N$ is the cover number defined in van der
Vaart and Wellner (1996, p.90).
\end{condition}

Condition \ref{HH} is a standard condition guaranteeing identification of the
model and imposing sufficient smoothness conditions as well as existence of
higher moments to allow for a higher order stochastic expansion of the
estimator. Condition \ref{M} together with separability of the parameter space
guarantees measurability of suprema of our empirical processes. As is well
known from the probability literature, measurability conditions could be
relaxed somewhat at the expense of more refined convergence arguments. We are
abstracting from such refinements for the purpose of this paper.

\subsection{Conditions for panel models} \label{sec:reg-conds-panel}

The assumptions for the panel results follow those
used in Hahn and Newey (2004).

\begin{condition}
$n,T\rightarrow\infty$, with $n/T\rightarrow\rho$ for $0<\rho<\infty$.
\end{condition}

\begin{condition}
(i) The data $z_{it}$ are independent over $i$ and $t$ and identically
distributed over $t$ according to the density $f(z|\theta,\alpha)$; (ii) the
log density $\ln f(z|\theta,\alpha)$ is continuous in both $\theta$ and
$\alpha$; (iii) there exists a function $M(z_{it})$ such that $|\ln
f(z_{it}|\theta,\alpha_{i})|\leq M(z_{it})$, $\left\vert \left.  \partial\ln
f(z_{it}|\theta,\alpha_{i})\right/  \partial(\theta,\alpha_{i})\right\vert
\leq M(z_{it})$ and $\sup_{i}E[M(z_{it})^{33}]<\infty$.
\end{condition}

\begin{condition}
For each $\eta>0$, $\inf_{i}\left[  G_{i}(\theta_{0},\alpha_{i})-\sup
_{\{(\theta,\alpha):|(\theta,\alpha)-(\theta_{0},\alpha_{i})|>\eta\}}%
G_{i}(\theta,\alpha)\right]  >0$ where $G_{i}(\theta,\alpha)\equiv E[\ln
f(z_{it}|\theta,\alpha)]$.
\end{condition}

\begin{condition}
(i) There exists some $M(z_{it})$ such that $\left\vert \left.  \partial
^{m_{1}+m_{2}}\ln f(z_{it}|\theta,\alpha)\right/  \partial\theta^{m_{1}%
}\partial\alpha^{m_{2}}\right\vert \leq M(z_{it})$ for $0\leq m_{1}+m_{2}%
\leq7$, and $\sup_{i}E[M(z_{it})^{Q}]<\infty$ for some $Q>64$; (ii)
$\lim_{n\rightarrow\infty}\mathcal{I}_{n}>0$, where $\mathcal{I}_{n}%
\equiv\frac{1}{n}\sum_{i}E[U_{it}^{2}]$; (iii) $\min_{i}E[V_{it}^{2}]>0$.
\end{condition}

\end{appendices}

\newpage

\appendix

\begin{center}
{\LARGE Supplementary Appendix: Efficient Bias Correction for Cross-section and Panel Data}
\end{center}

The Supplementary Appendix contains four parts. Part I contains proofs for the main results in the Sections 3 and 4 of the paper. Part II contains a number of supporting results and technical lemmas for the proofs related to the cross-sectional MLE asymptotic expansions. In Part III we provide some additional analysis of an AR(1) model. Finally, Part IV contains supporting results for the panel data expansions used in Section 4 of the paper. An outline of the sections is listed below.

{\bf Supplementary Appendix 1:}
\begin{itemize}
    \item[A] Proofs for the cross-sectional MLE results of Section \ref{sec:mle_crosssec}
    \item[B] Proofs for the panel data results in Section \ref{section-panel-model}
\end{itemize}

{\bf Supplementary Appendix 2:} supporting results for cross-sectional models.
\begin{itemize}
    \item[A] Proof of Proposition 8 (MLE expansion)
    \item[B] Proof of Proposition 9 (Bootstrap expansion)
    \item[C] Consistency of bias estimator
    \item[D] Technical lemmas
    \item[E] Proofs of technical lemmas
    \item[F] Moments of bootstrapped and jackknifed statistics
    \item[G] Proofs for F
    \item[H] Proofs for Example 1
    \item[I] Higher-order variance of general k-split-sample bias correction
\end{itemize}

{\bf Supplementary Appendix 3:} time series models.
\begin{itemize}
    \item[A] Higher-order analysis of an AR(1) model
    \item[B] Technical results for part A
    \item[C] Further details
\end{itemize}

{\bf Supplementary Appendix 4:} supporting results for panel models.
\begin{itemize}
    \item[A] properties of V-statistics
    \item[B] Jackknifed normalized V-statistics
    \item[C] Split-sample V-statistics
    \item[D] Expansions for panel data models
    \item[E] Additional results for the proof of Proposition 9
    \item[F] Higher-order bias results
    \item[G] Derivation of bias for average effect estimates
\end{itemize}

\newpage
\begin{center}
{\LARGE Supplementary Appendix I: Efficient Bias Correction for Cross-section and Panel Data}
\end{center}

\section{Proofs for Cross Section Models}\label{Proof-Main-Results}

\subsection{Preliminaries}

The following proposition establishes an asymptotic expansion for the MLE, which forms the basis of the results in Section \ref{sec:mle_crosssec}.

\begin{proposition}
\label{ApproxThetaHat-strong}Under Condition \ref{HH}, with probability
tending to one, $\widehat{\theta}$ satisfies the expansion%
\[
\sqrt{n}\left(  \widehat{\theta}-\theta_{0}\right)  =T_{1}+\frac{T_{2}}%
{\sqrt{n}}+\frac{T_{3}}{n}+\frac{T_{4}}{n\sqrt{n}}+\frac{T_{5}}{n^{2}}%
+O_{p}\left(  \frac{1}{n^{2}\sqrt{n}}\right)  ,
\]
where
\begin{align}
T_{1}  &  =\mathcal{I}^{-1}U\left(  \theta_{0}\right)  ,\label{V1}\\
T_{2}  &  =\frac{1}{2}\mathcal{I}^{-3}\mathcal{Q}_{1}\left(  \theta
_{0}\right)  U\left(  \theta_{0}\right)  ^{2}+\mathcal{I}^{-2}U\left(
\theta_{0}\right)  V\left(  \theta_{0}\right)  , \label{V2}\\
T_{3}  &  =\frac{1}{6}\mathcal{I}^{-4}\mathcal{Q}_{2}\left(  \theta
_{0}\right)  U\left(  \theta_{0}\right)  ^{3}+\frac{1}{2}\mathcal{I}%
^{-5}\mathcal{Q}_{1}\left(  \theta_{0}\right)  ^{2}U\left(  \theta_{0}\right)
^{3}+\frac{3}{2}\mathcal{I}^{-4}\mathcal{Q}_{1}\left(  \theta\right)  U\left(
\theta_{0}\right)  ^{2}V\left(  \theta_{0}\right) \nonumber \\
&  +\frac{1}{2}\mathcal{I}^{-3}U\left(  \theta_{0}\right)  ^{2}W\left(
\theta_{0}\right)  +\mathcal{I}^{-3}U\left(  \theta_{0}\right)  V\left(
\theta_{0}\right)  ^{2}, \label{V3}\\
T_{4}  &  =O_{p}\left(  1\right)  , \nonumber\\
T_{5}  &  =O_{p}\left(  1\right)  . \nonumber
\end{align}

\end{proposition}

\begin{proof}
Available in Supplementary Appendix II.
\end{proof}

In order to approximate the bias of the bootstrapped estimate $\widehat{\theta
}^{\ast}$ we need a similar higher order expansion as in the case of the MLE.
Here, however, the reference point around which we develop our approximation
is the empirical distribution $\widehat{F}$ rather than the original
distribution $F$. The convergence of $\widehat{F}$ to $F$ then guarantees that
bootstrapped statistics are close to the original statistics.

\begin{definition}
$\widehat{\mathcal{I}}\left(  \theta\right)  \mathcal{\equiv}-n^{-1}\sum
_{i=1}^{n}\ell^{\theta}(Z_{i},\theta),$ $\widehat{\mathcal{Q}}_{1}\left(
\theta\right)  \equiv n^{-1}%
{\textstyle\sum_{i=1}^{n}}
\ell^{\theta\theta}(Z_{i},\theta)$ and $\widehat{\mathcal{Q}}_{2}\left(
\theta\right)  \equiv n^{-1}%
{\textstyle\sum_{i=1}^{n}}
\ell^{\theta\theta\theta}\left(  Z_{i},\theta\right)  $.
\end{definition}

\begin{definition}
$U_{i}^{\ast}\left(  \theta\right)  \equiv\ell\left(  Z_{i}^{\ast}%
,\theta\right)  $, $V_{i}^{\ast}\left(  \theta\right)  \equiv\ell^{\theta
}\left(  Z_{i}^{\ast},\theta\right)  -n^{-1}\sum_{i=1}^{n}\ell^{\theta}\left(
Z_{i},\theta\right)  $, $W_{i}^{\ast}\equiv\ell^{\theta\theta}\left(
Z_{i}^{\ast}\right)  -n^{-1}\sum_{i=1}^{n}\ell^{\theta\theta}\left(
Z_{i},\theta\right)  $.
\end{definition}

\begin{definition}
$U^{\ast}\left(  \theta\right)  =n^{-1/2}\sum_{i=1}^{n}U_{i}^{\ast}\left(
\theta\right)  $, $V^{\ast}\left(  \theta\right)  =n^{-1/2}\sum_{i=1}^{n}%
V_{i}^{\ast}\left(  \theta\right)  $ and $W^{\ast}\left(  \theta\right)
=n^{-1/2}\sum_{i=1}^{n}W_{i}^{\ast}\left(  \theta\right)  $.
\end{definition}

The roles of $\mathcal{I}$, $\mathcal{Q}_{1}$ and $\mathcal{Q}_{2}$ are played
by $\widehat{\mathcal{I}}\mathcal{\equiv}\widehat{\mathcal{I}}\left(
\widehat{\theta}\right)  $, $\widehat{\mathcal{Q}}_{1}\left(  \widehat{\theta
}\right)  $, and $\widehat{\mathcal{Q}}_{2}\left(  \widehat{\theta}\right)  $.
We obtain for the following result for the bootstrapped estimate
$\widehat{\theta}^{\ast}$.

\begin{proposition}
\label{ApproxThetaHatStar}Under Conditions \ref{HH},\ref{M} and \ref{EC}, with
probability tending to one, $\widehat{\theta}^{\ast}$ satisfies the expansion
\[
\sqrt{n}\left(  \widehat{\theta}^{\ast}-\widehat{\theta}\right)
=\widehat{T}_{1}+\frac{\widehat{T}_{2}}{\sqrt{n}}+\frac{\widehat{T}_{3}}%
{n}+\frac{\widehat{T}_{4}}{n\sqrt{n}}+O_{p}\left(  \frac{1}{n^{2}}\right)
\text{ a.s.}%
\]
where $\widehat{T}_{1}=\widehat{\mathcal{I}}^{-1}U^{\ast}\left(
\widehat{\theta}\right)  ,$ $\widehat{T}_{2}=\frac{1}{2}\widehat{\mathcal{I}%
}^{-3}\widehat{\mathcal{Q}}_{1}\left(  \widehat{\theta}\right)  U^{\ast
}\left(  \widehat{\theta}\right)  ^{2}+\widehat{\mathcal{I}}^{-2}U^{\ast
}\left(  \widehat{\theta}\right)  V^{\ast}\left(  \widehat{\theta}\right)  $ etc.
\end{proposition}

\begin{proof}
Available in Supplementary Appendix II.
\end{proof}

\subsection{Proof of Proposition \ref{boot-jack-higher} - Bootstrap}

We denote the empirical distribution of $Z_{1}^{\ast},...,Z_{n}^{\ast}$ by
$\widehat{F}^{\ast}\left(  z\right)  =n^{-1}\sum_{i=1}^{n}\mathbf{1}\left\{
Z_{i}^{\ast}\leq z\right\}  $. Our technical conditions in the previous
section guarantee the validity of our stochastic expansions. From Gine and
Zinn (1990, Theorem 2.4) and Conditions \ref{HH},\ref{M} and \ref{EC} it
follows that, almost surely, $n^{1/2}\left(  \hat{F}^{\ast}-\hat{F}\right)
\rightarrow T$ weakly in $l^{\infty}\left(  \mathfrak{F}\right)  $ where $T$
is a Brownian Bridge Process. We use the result on the convergence of the
empirical processes to obtain an expansion of the estimators $\widehat{\theta
}$ and $\widehat{\theta}^{\ast}$.

Introduce the truncation function $h_{n}(x)$, where
\[
h_{n}\left(  x\right)  =\left\{
\begin{array}
[c]{cc}%
-n^{\alpha} & \text{if }x<-n^{\alpha}\\
x & \text{if }\left\vert x\right\vert <n^{\alpha}\\
n^{\alpha} & \text{if }x>n^{\alpha}%
\end{array}
\right.
\]
with $\alpha\in\left(  0,\frac{41}{30}\right)  $. Also, let
\[
\widehat{\theta}_{aa}^{\ast}\equiv\frac{1}{n^{1/2}}\widehat{\theta}^{\epsilon
}\left(  0\right)  +\frac{1}{2}\frac{1}{n}\widehat{\theta}^{\epsilon\epsilon
}\left(  0\right)  +\frac{1}{6}\frac{1}{n^{3/2}}\widehat{\theta}%
^{\epsilon\epsilon\epsilon}\left(  0\right)  .
\]
It can be shown\footnote{See (\ref{hstar1}), (\ref{hstar2}), and
(\ref{hstar3}) in the proof of Proposition \ref{BBC} in the Supplementary
Appendix I.} that $E^{\ast}\left[  h_{n}\left(  \widehat{\theta}^{\ast
}-\widehat{\theta}\right)  \right]  =E^{\ast}\left[  \widehat{\theta}%
_{aa}^{\ast}\right]  +o_{p}\left(  n^{-3/2}\right)  $, from which it follows
that
\begin{align*}
\sqrt{n}\left(  \widehat{\theta}-E^{\ast}\left[  h_{n}\left(  \widehat{\theta
}^{\ast}-\widehat{\theta}\right)  \right]  -\theta_{0}\right)   &  =\sqrt
{n}\left(  \widehat{\theta}-E^{\ast}\left[  \theta_{aa}^{\ast}\right]
-\theta_{0}\right)  \\ &+\sqrt{n}\left(  E^{\ast}\left[  \theta_{aa}^{\ast
}\right]  -E^{\ast}\left[  h_{n}\left(  \widehat{\theta}^{\ast}%
-\widehat{\theta}\right)  \right]  \right) \\
&  =\sqrt{n}\left(  \widehat{\theta}-E^{\ast}\left[  \theta_{aa}^{\ast
}\right]  -\theta_{0}\right)  +o_{p}(n^{-1}).
\end{align*}
It can also be shown\footnote{See Lemma \ref{BExp} Section \ref{sec-bootm} of
Supplementary Appendix II.}%
\[
E^{\ast}\left[  \widehat{\theta}^{\epsilon}\left(  0\right)  \right]  =0,
\]%
\begin{align}
E^{\ast}\left[  \widehat{\theta}^{\epsilon\epsilon}\left(  0\right)  \right]
&  =\widehat{\mathcal{I}}^{-3}\widehat{\mathcal{Q}}_{1}\left(  \widehat{\theta
}\right)  \tfrac{1}{n}%
{\textstyle\sum_{i=1}^{n}}
\ell\left(  Z_{i},\widehat{\theta}\right)  ^{2}+2\widehat{\mathcal{I}}%
^{-2}\tfrac{1}{n}%
{\textstyle\sum_{i=1}^{n}}
\ell\left(  Z_{i},\widehat{\theta}\right)  \ell^{\theta}\left(  Z_{i}%
,\widehat{\theta}\right) \nonumber\\
&  =\mathcal{I}^{-2}\mathcal{Q}_{1}\left(  \theta_{0}\right)  +2\mathcal{I}%
^{-2}E\left[  \ell\left(  Z_{i},\theta_{0}\right)  \ell^{\theta}\left(
Z_{i},\theta_{0}\right)  \right]  +B_{n}+o_{p}\left(  n^{-1/2}\right)
,\nonumber
\end{align}
and
\begin{align*}
E^{\ast}\left[  \widehat{\theta}^{\epsilon\epsilon\epsilon}\left(  0\right)
\right]   &  =\widehat{\mathcal{I}}^{-4}\widehat{\mathcal{Q}}_{2}\left(
\widehat{\theta}\right)  \left(  \tfrac{1}{\sqrt{n}}\tfrac{1}{n}%
{\textstyle\sum_{j=1}^{n}}
\ell\left(  Z_{i},\widehat{\theta}\right)  ^{3}\right)  +3\widehat{\mathcal{I}%
}^{-5}\widehat{\mathcal{Q}}_{1}\left(  \widehat{\theta}\right)  ^{2}\left(
\tfrac{1}{\sqrt{n}}\tfrac{1}{n}%
{\textstyle\sum_{j=1}^{n}}
\ell\left(  Z_{i},\widehat{\theta}\right)  ^{3}\right) \\
&  +9\widehat{\mathcal{I}}^{-4}\widehat{\mathcal{Q}}_{1}\left(
\widehat{\theta}\right)  \left(  \tfrac{1}{\sqrt{n}}\tfrac{1}{n}%
{\textstyle\sum_{j=1}^{n}}
\ell\left(  Z_{i},\widehat{\theta}\right)  ^{2}\left(  \ell^{\theta}\left(
Z_{i},\widehat{\theta}\right)  -\overline{\ell^{\theta}\left(  \cdot
,\widehat{\theta}\right)  }\right)  \right) \\
&  +3\widehat{\mathcal{I}}^{-3}\left(  \tfrac{1}{\sqrt{n}}\tfrac{1}{n}%
{\textstyle\sum_{j=1}^{n}}
\ell\left(  Z_{i},\widehat{\theta}\right)  ^{2}\left(  \ell^{\theta\theta
}\left(  Z_{i},\widehat{\theta}\right)  -\overline{\ell^{\theta\theta}\left(
\cdot,\widehat{\theta}\right)  }\right)  \right) \\
&  +6\widehat{\mathcal{I}}^{-3}\left(  \tfrac{1}{\sqrt{n}}\tfrac{1}{n}%
{\textstyle\sum_{j=1}^{n}}
\ell\left(  Z_{i},\widehat{\theta}\right)  \left(  \ell^{\theta}\left(
Z_{i},\widehat{\theta}\right)  -\overline{\ell^{\theta}\left(  \cdot
,\widehat{\theta}\right)  }\right)  ^{2}\right) \\
&  =O_{p}\left(  n^{-1/2}\right)  .
\end{align*}
Here,
\begin{align*}
\sqrt{n}B_{n}  &  \equiv-3\mathcal{I}^{-3}\sqrt{n}\left(  \widehat{\mathcal{I}%
}-\mathcal{I}\right)  \mathcal{Q}_{1}\left(  \theta_{0}\right) \\
&  +\mathcal{I}^{-2}\sqrt{n}\left(  \widehat{\mathcal{Q}}_{1}\left(
\widehat{\theta}\right)  -\mathcal{Q}\left(  \theta_{0}\right)  \right) \\
&  +\mathcal{I}^{-3}\mathcal{Q}_{1}\left(  \theta_{0}\right)  \left(  \frac
{1}{\sqrt{n}}%
{\textstyle\sum_{i=1}^{n}}
\left[  \ell\left(  Z_{i},\theta_{0}\right)  ^{2}-E\left[  \ell\left(
Z_{i},\theta_{0}\right)  ^{2}\right]  \right]  +\sqrt{n}\left(  \bar{m}%
_{4}\left(  \widehat{\theta}\right)  -\bar{m}_{4}\left(  \theta_{0}\right)
\right)  \right) \\
&  -4\mathcal{I}^{-3}\sqrt{n}\left(  \widehat{\mathcal{I}}-\mathcal{I}\right)
E\left[  \ell\left(  Z_{i},\theta_{0}\right)  \ell^{\theta}\left(
Z_{i},\theta\right)  \right] \\
&  +2\mathcal{I}^{-2}\left(  \sqrt{n}\left(  \bar{m}_{3}\left(
\widehat{\theta}\right)  -\bar{m}_{3}\left(  \theta_{0}\right)  \right)
\right. \\
&  \left.  +\sqrt{n}\left[  n^{-1}%
{\textstyle\sum_{i=1}^{n}}
\ell\left(  Z_{i},\theta_{0}\right)  \ell^{\theta}\left(  Z_{i},\theta
_{0}\right)  -E\left[  \ell\left(  Z_{i},\theta_{0}\right)  \ell^{\theta
}\left(  Z_{i},\theta_{0}\right)  \right]  \right]  \right)
\end{align*}
It follows that
\begin{align}
E^{\ast}\left[  \widehat{\theta}_{aa}^{\ast}\right]   &  =\frac{1}{2}\frac
{1}{n}E^{\ast}\left[  \widehat{\theta}^{\epsilon\epsilon}\left(  0\right)
\right]  +O_{p}\left(  n^{-2}\right) \nonumber\\
&  =\frac{1}{2}\frac{1}{n}\left(  \mathcal{I}^{-2}\mathcal{Q}_{1}\left(
\theta_{0}\right)  +2\mathcal{I}^{-2}E\left[  \ell\left(  Z_{i},\theta
_{0}\right)  \ell^{\theta}\left(  Z_{i},\theta_{0}\right)  \right]  \right)
\nonumber\\
&  +\frac{1}{2}\frac{1}{n^{3/2}}\left(  \sqrt{n}B_{n}\right)  +O_{p}\left(
n^{-2}\right)  . \label{E*2}%
\end{align}
It can further be shown\footnote{Proo is available in Section \ref{sec-E*3} of
Supplementary Appendix II.}, we obtain
\begin{equation}
\frac{\sqrt{n}B_{n}}{2}=\mathbb{B}+o_{p}\left(  1\right)  , \label{E*3}%
\end{equation}
where%
\begin{align}
2\mathbb{B}  &  \equiv3\mathcal{I}^{-4}\mathcal{Q}_{1}\left(  \theta
_{0}\right)  ^{2}U\left(  \theta_{0}\right)  +\mathcal{I}^{-3}\mathcal{Q}%
_{2}\left(  \theta_{0}\right)  U\left(  \theta_{0}\right)  +6\mathcal{I}%
^{-4}\mathcal{Q}_{1}\left(  \theta_{0}\right)  E\left[  U_{i}\left(
\theta_{0}\right)  V_{i}\left(  \theta_{0}\right)  \right]  U\left(
\theta_{0}\right) \nonumber\\
&  +2\mathcal{I}^{-3}E\left[  U_{i}\left(  \theta_{0}\right)  W_{i}\left(
\theta_{0}\right)  \right]  U\left(  \theta_{0}\right)  +2\mathcal{I}%
^{-3}E\left[  \ell^{\theta}\left(  Z_{i},\theta_{0}\right)  ^{2}\right]
U\left(  \theta_{0}\right) \nonumber\\
&  +3\mathcal{I}^{-3}\mathcal{Q}_{1}\left(  \theta_{0}\right)  V\left(
\theta_{0}\right)  +4\mathcal{I}^{-3}E\left[  U_{i}\left(  \theta_{0}\right)
V_{i}\left(  \theta_{0}\right)  \right]  V\left(  \theta_{0}\right)
\nonumber\\
&  +\mathcal{I}^{-2}W\left(  \theta_{0}\right) \nonumber\\
&  +2\mathcal{I}^{-2}n^{-1/2}\left(
{\textstyle\sum_{i=1}^{n}}
\ell\left(  Z_{i},\theta_{0}\right)  \ell^{\theta}\left(  Z_{i},\theta
_{0}\right)  -E\left[  \ell\left(  Z_{i},\theta_{0}\right)  \ell^{\theta
}\left(  Z_{i},\theta_{0}\right)  \right]  \right) \nonumber\\
&  +\mathcal{I}^{-3}\mathcal{Q}_{1}\left(  \theta_{0}\right)  n^{-1/2}\left(
{\textstyle\sum_{i=1}^{n}}
\left[  \ell\left(  Z_{i},\theta_{0}\right)  ^{2}-E\left[  \ell\left(
Z_{i},\theta_{0}\right)  ^{2}\right]  \right]  \right)  \label{def-boldB}%
\end{align}

Combining (\ref{E*2}) and (\ref{E*3}), we obtain
\begin{align*}
\left(  \widehat{\theta}-E^{\ast}\left[  h_{n}\left(  \widehat{\theta}^{\ast
}-\widehat{\theta}\right)  \right]  -\theta_{0}\right)   &  =\widehat{\theta
}-\theta_{0}\\
&  -\frac{1}{2}\frac{1}{n}\left(  \mathcal{I}^{-2}\mathcal{Q}_{1}\left(
\theta_{0}\right)  +2\mathcal{I}^{-2}E\left[  \ell\left(  Z_{i},\theta
_{0}\right)  \ell^{\theta}\left(  Z_{i},\theta_{0}\right)  \right]  \right) \\
&  -\frac{1}{n^{3/2}}\mathbb{B}+o_{p}\left(  \frac{1}{n^{3/2}}\right)  ,
\end{align*}
from which the conclusion follows.

\subsection{Proof of Proposition \ref{boot-jack-higher} - Jackknife}

Write $\theta^{\epsilon}=\theta^{\epsilon}\left(  0\right)  $, etc, for
notational simplicity. Because
\begin{align*}
\widehat{\theta}  &  =\theta_{0}+\frac{1}{\sqrt{n}}\theta^{\epsilon}+\frac
{1}{2}\frac{1}{n}\theta^{\epsilon\epsilon}+\frac{1}{6}\frac{1}{n\sqrt{n}%
}\theta^{\epsilon\epsilon\epsilon}\\
&  +\frac{1}{24}\frac{1}{n^{2}}\theta^{\epsilon\epsilon\epsilon\epsilon}%
+\frac{1}{120}\frac{1}{n^{2}\sqrt{n}}\theta^{\epsilon\epsilon\epsilon
\epsilon\epsilon}+\frac{1}{720}\frac{1}{n^{3}}\theta^{\epsilon\epsilon
\epsilon\epsilon\epsilon\epsilon}\left(  \widetilde{\epsilon}\right)  ,
\end{align*}
we should have
\begin{align*}
\widehat{\theta}_{\left(  j\right)  }  &  =\theta_{0}+\frac{1}{\sqrt{n-1}%
}\theta_{\left(  j\right)  }^{\epsilon}+\frac{1}{2}\frac{1}{n-1}%
\theta_{\left(  j\right)  }^{\epsilon\epsilon}+\frac{1}{6}\frac{1}{\left(
n-1\right)  \sqrt{n-1}}\theta_{\left(  j\right)  }^{\epsilon\epsilon\epsilon
}\\
&  +\frac{1}{24}\frac{1}{\left(  n-1\right)  ^{2}}\theta_{\left(  j\right)
}^{\epsilon\epsilon\epsilon\epsilon}+\frac{1}{120}\frac{1}{\left(  n-1\right)
^{2}\sqrt{n-1}}\theta_{\left(  j\right)  }^{\epsilon\epsilon\epsilon
\epsilon\epsilon}+\frac{1}{720}\frac{1}{\left(  n-1\right)  ^{3}}%
\theta_{\left(  j\right)  }^{\epsilon\epsilon\epsilon\epsilon\epsilon\epsilon
}\left(  \widetilde{\epsilon}_{\left(  j\right)  }\right)  .
\end{align*}
Therefore,
\begin{align}
\lefteqn{\sqrt{n}\left(  \widetilde{\theta}-\theta_{0}\right)  =\sqrt
{n}\left(  n\widehat{\theta}-\frac{n-1}{n}\sum_{j=1}^{n}\widehat{\theta
}_{\left(  j\right)  }-\theta_{0}\right)  }\nonumber\\
&  =\left(  n\theta^{\epsilon}-\sqrt{n}\sqrt{n-1}\frac{1}{n}\sum_{j=1}%
^{n}\theta_{\left(  j\right)  }^{\epsilon}\right)  +\frac{1}{2}\frac{1}%
{\sqrt{n}}\left(  n\theta^{\epsilon\epsilon}-\sum_{j=1}^{n}\theta_{\left(
j\right)  }^{\epsilon\epsilon}\right) \nonumber\\ 
& +\frac{1}{6}\frac{1}{n}\left(
n\theta^{\epsilon\epsilon\epsilon}-\sqrt{\frac{n}{n-1}}\sum_{j=1}^{n}%
\theta_{\left(  j\right)  }^{\epsilon\epsilon\epsilon}\right) +\frac{1}{24}\frac{1}{n\sqrt{n}}\left(  n\theta^{\epsilon\epsilon
\epsilon\epsilon}-\frac{n^{2}}{n-1}\frac{1}{n}\sum_{j=1}^{n}\theta_{\left(
j\right)  }^{\epsilon\epsilon\epsilon\epsilon}\right) \nonumber\\
&   +\frac{1}{120}\frac
{1}{n^{2}}n\left(  \theta^{\epsilon\epsilon\epsilon\epsilon\epsilon
}\allowbreak-\frac{n\sqrt{n}}{\left(  n-1\right)  \sqrt{n-1}}\frac{1}{n}%
\sum_{j=1}^{n}\theta_{\left(  j\right)  }^{\epsilon\epsilon\epsilon
\epsilon\epsilon}\right) \nonumber\\
&  +\frac{1}{720}\frac{1}{n\sqrt{n}}\theta^{\epsilon\epsilon\epsilon
\epsilon\epsilon\epsilon}\left(  \widetilde{\epsilon}\right)  -\frac{1}%
{720}\frac{1}{\left(  n-1\right)  ^{2}\sqrt{n}}\sum_{i=1}^{n}\theta_{\left(
j\right)  }^{\epsilon\epsilon\epsilon\epsilon\epsilon\epsilon}\left(
\widetilde{\epsilon}_{\left(  j\right)  }\right)  . \label{J-expansion}%
\end{align}
It can be shown\footnote{See Section \ref{sec-r6} in Supplementary Appendix
I.} that
\begin{align}
\frac{1}{720}\frac{1}{n\sqrt{n}}\theta^{\epsilon\epsilon\epsilon
\epsilon\epsilon\epsilon}\left(  \widetilde{\epsilon}\right)  -\frac{1}%
{720}\frac{1}{\left(  n-1\right)  ^{2}\sqrt{n}}\sum_{i=1}^{n}\theta_{\left(
j\right)  }^{\epsilon\epsilon\epsilon\epsilon\epsilon\epsilon}\left(
\widetilde{\epsilon}_{\left(  j\right)  }\right)   &  =o_{p}\left(  \frac
{1}{n}\right)  ,\label{r6}\\
\frac{1}{24}\frac{1}{n\sqrt{n}}\left(  n\theta^{\epsilon\epsilon
\epsilon\epsilon}-\frac{n^{2}}{n-1}\frac{1}{n}\sum_{j=1}^{n}\theta_{\left(
j\right)  }^{\epsilon\epsilon\epsilon\epsilon}\right)   &  =o_{p}\left(
\frac{1}{n}\right)  ,\label{r4}\\
\frac{1}{120}\frac{1}{n^{2}}n\left(  \theta^{\epsilon\epsilon\epsilon
\epsilon\epsilon}-\frac{n\sqrt{n}}{\left(  n-1\right)  \sqrt{n-1}}\frac{1}%
{n}\sum_{j=1}^{n}\theta_{\left(  j\right)  }^{\epsilon\epsilon\epsilon
\epsilon\epsilon}\right)   &  =o_{p}\left(  \frac{1}{n}\right)  . \label{r5}%
\end{align}
It can also be shown\footnote{It can be established by combining (\ref{V1}),
(\ref{V2}), (\ref{V3}) with Lemmas \ref{Vjack1}, \ref{Vjack2}, \ref{Vjack3} in
Section \ref{sec-bootm} of Supplementary Appendix II} that%
\begin{align}
\lefteqn{\left(  n\theta^{\epsilon}-\sqrt{n}\sqrt{n-1}\frac{1}{n}\sum
_{j=1}^{n}\theta_{\left(  j\right)  }^{\epsilon}\right)  +\frac{1}{2}\frac
{1}{\sqrt{n}}\left(  n\theta^{\epsilon\epsilon}-\sum_{j=1}^{n}\theta_{\left(
j\right)  }^{\epsilon\epsilon}\right)  }\nonumber\\
&  +\frac{1}{6}\frac{1}{n}\left(  n\theta^{\epsilon\epsilon\epsilon}%
-\sqrt{\frac{n}{n-1}}\sum_{j=1}^{n}\theta_{\left(  j\right)  }^{\epsilon
\epsilon\epsilon}\right) \nonumber\\
&  =\theta^{\epsilon}+\frac{1}{\sqrt{n}}\left(  \frac{1}{2}\theta
^{\epsilon\epsilon}-b\left(  \theta_{0}\right)  \right)  +\frac{1}{6}\frac
{1}{n}\theta^{\epsilon\epsilon\epsilon}-n^{-1}\left(  \mathbb{J}+o_{p}\left(
1\right)  \right)  , \label{r1-3}%
\end{align}
where
\begin{align}
2\mathbb{J}  &  =\left(  \mathcal{I}^{-3}E\left[  \ell^{\theta\theta\theta
}\right]  +3\mathcal{I}^{-4}\left(  E\left[  \ell^{\theta\theta}\right]
\right)  ^{2}\right)  U\left(  \theta_{0}\right)  +6\mathcal{I}^{-4}E\left[
\ell^{\theta\theta}\right]  E\left[  U_{i}V_{i}\right]  U\left(  \theta
_{0}\right) \nonumber\\
&  +\frac{3}{\mathcal{I}^{3}}E\left[  \ell^{\theta\theta}\right]  V\left(
\theta_{0}\right)  +\frac{2}{\mathcal{I}^{3}}E\left[  U_{i}W_{i}\right]
U\left(  \theta_{0}\right)  +\frac{1}{\mathcal{I}^{2}}W(\theta_{0})\nonumber\\
&  +\frac{2}{\mathcal{I}^{3}}E\left[  V_{i}^{2}\right]  U\left(  \theta
_{0}\right)  +\frac{4}{\mathcal{I}^{3}}E\left[  U_{i}V_{i}\right]
V(\theta_{0})\nonumber\\
&  +\mathcal{I}^{-3}E\left[  \ell^{\theta\theta}\right]  n^{-1/2}\left(
{\textstyle\sum_{i=1}^{n}}
\left[  \ell\left(  Z_{i},\theta_{0}\right)  ^{2}-E\left[  \ell\left(
Z_{i},\theta_{0}\right)  ^{2}\right]  \right]  \right) \nonumber\\
&  +2\mathcal{I}^{-2}n^{-1/2}\left(
{\textstyle\sum_{i=1}^{n}}
\ell\left(  Z_{i},\theta_{0}\right)  \ell^{\theta}\left(  Z_{i},\theta
_{0}\right)  -E\left[  \ell\left(  Z_{i},\theta_{0}\right)  \ell^{\theta
}\left(  Z_{i},\theta_{0}\right)  \right]  \right)  . \label{def-boldJ}%
\end{align}
Combining (\ref{J-expansion}) with (\ref{r6}), (\ref{r5}), (\ref{r4}), and
(\ref{r1-3}), we obtain
\begin{align*}
\sqrt{n}\left(  \widetilde{\theta}-\theta_{0}\right)   &  =\theta^{\epsilon
}+\frac{1}{2}\frac{1}{\sqrt{n}}\left\{  \mathcal{I}^{-3}E\left[  \ell
^{\theta\theta}\right]  \left(  U(\theta_{0})^{2}-E\left[  U_{i}\left(
\theta_{0}\right)  ^{2}\right]  \right)  \right. \\
&  \left.  +\frac{2}{\mathcal{I}^{2}}\left(  U\left(  \theta_{0}\right)
V\left(  \theta_{0}\right)  -E\left[  U_{i}\left(  \theta_{0}\right)
V_{i}\left(  \theta_{0}\right)  \right]  \right)  \right\} \\
&  +\frac{1}{6}\frac{1}{n}\theta^{\epsilon\epsilon\epsilon}-\frac{1}%
{n}\mathbb{J}+o_{p}\left(  \frac{1}{n}\right)  .
\end{align*}
The conclusion follows by comparing (\ref{def-boldB}) and (\ref{def-boldJ}).

\subsection{Proof of Proposition \ref{HigherOrderEfficiency}}

For the proofs in this section it is useful to note that the bias term can be represented as
\begin{equation}
b\left(  \theta\right)  \equiv\tau\left(  \int m\left(  z,\theta\right)
f\left(  z,\theta\right)  dz\right)  \label{bias-PW}%
\end{equation}
where $\tau\left(  t_{1},t_{2},t_{3},t_{4}\right)  \equiv\frac{1}{2t_{1}^{2}%
}t_{2}+\frac{1}{t_{1}^{2}}t_{3}$, $t_{1}\equiv-\int\ell^{\theta}\left(
z,\theta\right)  f\left(  z,\theta\right)  dz,$ $t_{2}\equiv\int\ell
^{\theta\theta}\left(  z,\theta\right)  f\left(  z,\theta\right)  dz,$
$t_{3}\equiv\int\ell\left(  z,\theta\right)  \ell^{\theta}\left(
z,\theta\right)  f\left(  z,\theta\right)  dz,$ $t_{4}\equiv\int\ell\left(
z,\theta\right)  ^{2}f\left(  z,\theta\right)  dz,$ and\footnote{It is not
necessary to introduce the 4th component of $m\left(  z,\theta\right)  $ here,
but it makes it convenient to discuss $\widetilde{\theta}_{a}$ later.}
\[
m\left(  z,\theta\right)  \equiv\left(  -\ell^{\theta}\left(  z,\theta\right)
,\ell^{\theta\theta}\left(  z,\theta\right)  ,\ell\left(  z,\theta\right)
\ell^{\theta}\left(  z,\theta\right)  ,\left(  \ell\left(  z,\theta\right)
\right)  ^{2}\right)  ^{\prime}.
\]
This leads to the two analytical bias corrected estimators
\[
\widehat{\theta}_{c}\equiv\widehat{\theta}-\frac{b\left(  \widehat{\theta
}\right)  }{n}=\widehat{\theta}-\frac{1}{n}\left(  \frac{\int\ell
^{\theta\theta}\left(  z,\widehat{\theta}\right)  f\left(  z,\widehat{\theta
}\right)  dz}{2\left(  \int\ell^{\theta}\left(  z,\widehat{\theta}\right)
f\left(  z,\widehat{\theta}\right)  dz\right)  ^{2}}+\frac{\int\ell^{\theta
}\left(  z,\widehat{\theta}\right)  \ell^{\theta\theta}\left(  z,\theta
\right)  f\left(  z,\widehat{\theta}\right)  dz}{\left(  \int\ell^{\theta
}\left(  z,\widehat{\theta}\right)  f\left(  z,\widehat{\theta}\right)
dz\right)  ^{2}}\right)  .
\]

and, for $\widehat{b}\left(  \theta\right)
\equiv\tau\left(  n^{-1}\sum_{i}m\left(  Z_{i},\theta\right)  \right)  $, 
\[
\widehat{\theta}_{a}\equiv\widehat{\theta}-\frac{\widehat{b}\left(
\widehat{\theta}\right)  }{n}=\widehat{\theta}-\frac{1}{n}\left(
\frac{\left(  n^{-1}\sum_{i}\ell^{\theta\theta}\left(  Z_{i},\widehat{\theta
}\right)  \right)  }{2\left(  n^{-1}\sum_{i}\ell^{\theta}\left(
Z_{i},\widehat{\theta}\right)  \right)  ^{2}}+\frac{\left(  n^{-1}\sum_{i}%
\ell\left(  Z_{i},\widehat{\theta}\right)  \ell^{\theta}\left(  Z_{i}%
,\widehat{\theta}\right)  \right)  }{\left(  n^{-1}\sum_{i}\ell^{\theta
}\left(  Z_{i},\widehat{\theta}\right)  \right)  ^{2}}\right)  .
\]
Let
\begin{equation}
\overline{m}\left(  \theta\right)  \equiv E\left[  m\left(  Z_{i}%
,\theta\right)  \right]  =\int m\left(  z,\theta\right)  f\left(  z,\theta
_{0}\right)  dz \label{m-bar}%
\end{equation}
with j-th element $\overline{m}_{j}\left(  \theta\right)  $ and write
$\overline{m}=\overline{m}\left(  \theta_{0}\right)  $, $\tau_{m}%
\equiv\partial\tau\left(  \bar{m}\right)  /\partial m^{\prime}$, $M\equiv
E\left[  \frac{\partial m\left(  Z_{i},\theta_{0}\right)  }{\partial\theta
}\right]  =\int\frac{\partial m\left(  z,\theta_{0}\right)  }{\partial\theta
}f\left(  z,\theta_{0}\right)  dz$, and $\Lambda\equiv E\left[  m\left(
Z_{i},\theta_{0}\right)  \ell\left(  Z_{i},\theta_{0}\right)  \right]  $.

An expansion of $b\left(  \widehat{\theta}\right)  $ gives
\begin{align*}
b\left(  \widehat{\theta}\right)  -b\left(  \theta_{0}\right)   &
=n^{-1/2}\tau_{m}\left(  \left.  \frac{\partial\left(  \int m\left(
z,\theta\right)  f\left(  z,\theta\right)  dz\right)  }{\partial\theta
}\right\vert _{\theta=\theta_{0}}\right)  \mathcal{I}^{-1}U\left(  \theta
_{0}\right)  +O_{p}\left(  n^{-1}\right) \\
&  =n^{-1/2}\tau_{m}\left(  M+\Lambda\right)  \mathcal{I}^{-1}U\left(
\theta_{0}\right)  +O_{p}\left(  n^{-1}\right)  .
\end{align*}
A similar expansion for $\widehat{b}\left(  \widehat{\theta}\right)  $ gives
\begin{align*}
\widehat{b}\left(  \widehat{\theta}\right)  -b\left(  \theta_{0}\right)   &
=\widehat{b}\left(  \widehat{\theta}\right)  -\widehat{b}\left(  \theta
_{0}\right)  +\widehat{b}\left(  \theta_{0}\right)  -b\left(  \theta
_{0}\right) \\
&  =\tau_{m}\left(  \frac{1}{n}\sum_{i}m\left(  z_{i},\widehat{\theta}\right)
-\frac{1}{n}\sum_{i}m\left(  z_{i},\theta_{0}\right)  \right) \\
&  +\tau_{m}\left(  \frac{1}{n}\sum_{i}m\left(  z_{i},\theta_{0}\right)
-E\left[  m\left(  z_{i},\theta_{0}\right)  \right]  \right)  +O_{p}\left(
n^{-1}\right) \\
&  =n^{-1/2}\tau_{m}M\mathcal{I}^{-1}U\left(  \theta_{0}\right)  +\tau
_{m}\left(  n^{-1}%
{\textstyle\sum\nolimits_{i}}
\left(  m\left(  z_{i},\theta_{0}\right)  -\overline{m}\right)  \right)
+O_{p}\left(  n^{-1}\right)  .
\end{align*}
Plugging these expansions into that for $\widehat{\theta}$ gives
\begin{align*}
\sqrt{n}\left(  \widehat{\theta}_{c}-\theta_{0}\right)   &  =\sqrt{n}\left(
\widehat{\theta}-\theta_{0}\right)  -\frac{1}{\sqrt{n}}b\left(
\widehat{\theta}\right) \\
&  =\mathcal{I}^{-1}U\left(  \theta_{0}\right)  +\frac{1}{\sqrt{n}}\left(
\frac{1}{2}\theta^{\epsilon\epsilon}(0)-b\left(  \theta_{0}\right)  \right) \\
&  +\frac{1}{n}\left(  \frac{1}{6}\theta^{\epsilon\epsilon\epsilon}%
(0)-\tau_{m}\left(  M+\Lambda\right)  \mathcal{I}^{-1}U\left(  \theta
_{0}\right)  \right)  +O_{p}\left(  \frac{1}{n^{3/2}}\right)
\end{align*}
and
\begin{align*}
\sqrt{n}\left(  \widehat{\theta}_{a}-\theta_{0}\right)   &  =\sqrt{n}\left(
\widehat{\theta}-\theta_{0}\right)  -\frac{1}{\sqrt{n}}\widehat{b}\left(
\widehat{\theta}\right) \\
&  =\mathcal{I}^{-1}U\left(  \theta_{0}\right)  +\frac{1}{\sqrt{n}}\left(
\frac{1}{2}\theta^{\epsilon\epsilon}(0)-b\left(  \theta_{0}\right)  \right) \\
&  +\frac{1}{n}\left(  \frac{1}{6}\theta^{\epsilon\epsilon\epsilon}%
(0)-\tau_{m}\left(  M\mathcal{I}^{-1}U\left(  \theta_{0}\right)  +n^{-1/2}%
{\textstyle\sum_{i}}
\left(  m\left(  z_{i},\theta_{0}\right)  -\overline{m}\right)  \right)
\right)  +O_{p}\left(  \frac{1}{n^{3/2}}\right)  .
\end{align*}

\subsection{Proof of Proposition \ref{HigherOrderEfficiency3}}

Let $\zeta\left(  t_{1},t_{2},t_{3},t_{4}\right)  \equiv\frac
{1}{2t_{1}^{3}}t_{2}t_{4}+\frac{1}{t_{1}^{2}}t_{3}$, $\widetilde{b}\left(
\theta\right)  \equiv\zeta\left(  n^{-1}\sum_{i}m\left(  Z_{i},\theta\right)
\right)  $, so that%
\begin{align*}
\widetilde{\theta}_{a}\equiv\widehat{\theta}-\frac{\widetilde{b}\left(
\widehat{\theta}\right)  }{n} &=\widehat{\theta}-\frac{1}{n}\bigg[
-\frac{\left(  n^{-1}\sum_{i}\ell^{\theta\theta}\left(  Z_{i},\widehat{\theta
}\right)  \right)  \left(  n^{-1}\sum_{i}\left(  \ell\left(  Z_{i}%
,\widehat{\theta}\right)  \right)  ^{2}\right)  }{2\left(  n^{-1}\sum_{i}%
\ell^{\theta}\left(  Z_{i},\widehat{\theta}\right)  \right)  ^{3}}%
\\ &+\frac{\left(  n^{-1}\sum_{i}\ell\left(  Z_{i},\widehat{\theta}\right)
\ell^{\theta}\left(  Z_{i},\widehat{\theta}\right)  \right)  }{\left(
n^{-1}\sum_{i}\ell^{\theta}\left(  Z_{i},\widehat{\theta}\right)  \right)
^{2}}\bigg].
\end{align*}

An expansion for $\widetilde{b}\left(  \widehat{\theta}\right)  $ gives
\begin{align*}
\widetilde{b}\left(  \widehat{\theta}\right)  -b\left(  \theta_{0}\right)   &
=\widetilde{b}\left(  \widehat{\theta}\right)  -\widetilde{b}\left(
\theta_{0}\right)  +\widetilde{b}\left(  \theta_{0}\right)  -b\left(
\theta_{0}\right) \\
&  =\zeta_{m}\left(  \frac{1}{n}\sum_{i}m\left(  z_{i},\widehat{\theta
}\right)  -\frac{1}{n}\sum_{i}m\left(  z_{i},\theta_{0}\right)  \right) \\
&  +\zeta_{m}\left(  \frac{1}{n}\sum_{i}m\left(  z_{i},\theta_{0}\right)
-E\left[  m\left(  z_{i},\theta_{0}\right)  \right]  \right)  +O_{p}\left(
n^{-1}\right) \\
&  =n^{-1/2}\zeta_{m}M\mathcal{I}^{-1}U\left(  \theta_{0}\right)  +\zeta
_{m}\left(  n^{-1}%
{\textstyle\sum\nolimits_{i}}
\left(  m\left(  z_{i},\theta_{0}\right)  -\overline{m}\right)  \right)
+O_{p}\left(  n^{-1}\right)  ,
\end{align*}
and%
\begin{align*}
\sqrt{n}\left(  \widetilde{\theta}_{a}-\theta_{0}\right)   &  =\sqrt{n}\left(
\widehat{\theta}-\theta_{0}\right)  -\frac{1}{\sqrt{n}}\widetilde{b}\left(
\widehat{\theta}\right) \\
&  =\mathcal{I}^{-1}U\left(  \theta_{0}\right)  +\frac{1}{\sqrt{n}}\left(
\frac{1}{2}\theta^{\epsilon\epsilon}(0)-b\left(  \theta_{0}\right)  \right) \\
&  +\frac{1}{n}\left(  \frac{1}{6}\theta^{\epsilon\epsilon\epsilon}%
(0)-\zeta_{m}\left(  M\mathcal{I}^{-1}U\left(  \theta_{0}\right)  +n^{-1/2}%
{\textstyle\sum_{i}}
\left(  m\left(  z_{i},\theta_{0}\right)  -\overline{m}\right)  \right)
\right)  +O_{p}\left(  \frac{1}{n^{3/2}}\right)
\end{align*}
where $\zeta_{m}\equiv\partial\zeta\left(  \bar{m}\right)  /\partial
m^{\prime}$. We note that
\[
M=\left(  -E\left[  \ell^{\theta\theta}\right]  ,E\left[  \ell^{\theta
\theta\theta}\right]  ,E\left[  \left(  \ell^{\theta}\right)  ^{2}+\ell
\ell^{\theta\theta}\right]  ,2E\left[  \ell\ell^{\theta}\right]  \right)
^{\prime},
\]
and%
\begin{align*}
\zeta_{m}  &  =\left(  -\frac{3E\left[  \ell^{\theta\theta}\right]  E\left[
\ell^{2}\right]  }{2\mathcal{I}^{4}}-\frac{2E\left[  \ell\ell^{\theta}\right]
}{\mathcal{I}^{3}},\frac{E\left[  \ell^{2}\right]  }{2\mathcal{I}^{3}}%
,\frac{1}{\mathcal{I}^{2}},\frac{E\left[  \ell^{\theta\theta}\right]
}{2\mathcal{I}^{3}}\right)  ^{\prime}\\
&  =\left(  -\frac{3E\left[  \ell^{\theta\theta}\right]  }{2\mathcal{I}^{3}%
}-\frac{2E\left[  \ell\ell^{\theta}\right]  }{\mathcal{I}^{3}},\frac
{1}{2\mathcal{I}^{2}},\frac{1}{\mathcal{I}^{2}},\frac{E\left[  \ell
^{\theta\theta}\right]  }{2\mathcal{I}^{3}}\right)  ,
\end{align*}
so%
\[
\zeta_{m}M\mathcal{I}^{-1}U\left(  \theta_{0}\right)  =\left(  \frac{3\left(
E\left[  \ell^{\theta\theta}\right]  \right)  ^{2}}{2\mathcal{I}^{4}}%
+\frac{3E\left[  \ell^{\theta\theta}\right]  E\left[  \ell\ell^{\theta
}\right]  }{\mathcal{I}^{4}}+\frac{E\left[  \ell^{\theta\theta\theta}\right]
}{2\mathcal{I}^{3}}+\frac{E\left[  \left(  \ell^{\theta}\right)  ^{2}\right]
}{\mathcal{I}^{3}}+\frac{E\left[  \ell\ell^{\theta\theta}\right]
}{\mathcal{I}^{3}}\right)  U\left(  \theta_{0}\right)  ,
\]%
\begin{align*}
\zeta_{m}\left(  n^{-1}%
{\textstyle\sum\nolimits_{i}}
\left(  m\left(  z_{i},\theta_{0}\right)  -\overline{m}\right)  \right)   &
=\left(  \frac{3E\left[  \ell^{\theta\theta}\right]  }{2\mathcal{I}^{3}}%
+\frac{2E\left[  \ell\ell^{\theta}\right]  }{\mathcal{I}^{3}}\right)  V\left(
\theta_{0}\right)  +\frac{1}{2\mathcal{I}^{2}}W\left(  \theta_{0}\right) \\
&  +\frac{1}{\mathcal{I}^{2}}n^{-1/2}\left(
{\textstyle\sum_{i=1}^{n}}
\ell\left(  Z_{i},\theta_{0}\right)  \ell^{\theta}\left(  Z_{i},\theta
_{0}\right)  -E\left[  \ell\left(  Z_{i},\theta_{0}\right)  \ell^{\theta
}\left(  Z_{i},\theta_{0}\right)  \right]  \right) \\
&  +\frac{E\left[  \ell^{\theta\theta}\right]  }{2\mathcal{I}^{3}}%
n^{-1/2}\left(
{\textstyle\sum_{i=1}^{n}}
\left[  \ell\left(  Z_{i},\theta_{0}\right)  ^{2}-E\left[  \ell\left(
Z_{i},\theta_{0}\right)  ^{2}\right]  \right]  \right)  ,
\end{align*}
and their sum can be easily verified to be equal to $\mathbb{J}=\mathbb{B}$.

\subsection{Proof of Corollary \ref{ET1T3}}

Note first that
\begin{align*}
\lefteqn{E\left[  \mathbb{A}\theta^{\epsilon}(0)\right]  =E\left[  \tau
_{m}\left(  M\mathcal{I}^{-1}U\left(  \theta_{0}\right)  +n^{-1/2}%
{\textstyle\sum_{i}}
\left(  m\left(  z_{i},\theta_{0}\right)  -\overline{m}\right)  \right)
\cdot\left(  \mathcal{I}^{-1}U\left(  \theta_{0}\right)  \right)  \right]  }\\
&  =\tau_{m}M\mathcal{I}^{-1}E\left[  U\left(  \theta_{0}\right)  ^{2}\right]
\mathcal{I}^{-1}+\tau_{m}E\left[  \left(  n^{-1/2}%
{\textstyle\sum_{i}}
\left(  m\left(  z_{i},\theta_{0}\right)  -\overline{m}\right)  \right)
U\left(  \theta_{0}\right)  \right]  \mathcal{I}^{-1}\\
&  =\tau_{m}\left(  M+\Lambda\right)  \mathcal{I}^{-1},
\end{align*}
and%
\[
E\left[  \mathbb{C}\theta^{\epsilon}(0)\right]  =E\left[  \left(  \tau
_{m}\left(  M+\Lambda\right)  \mathcal{I}^{-1}U\left(  \theta_{0}\right)
\right)  \left(  \mathcal{I}^{-1}U\left(  \theta_{0}\right)  \right)  \right]
=\tau_{m}\left(  M+\Lambda\right)  \mathcal{I}^{-1}.
\]

To show that $E\left[  \mathbb{B}\theta^{\epsilon}\left(  0\right)  \right]
=\tau_{m}\left(  M+\Lambda\right)  \mathcal{I}^{-1}$, it suffices to prove
that $E\left[  \mathbb{B}U\left(  \theta_{0}\right)  \right]  =\tau_{m}\left(
M+\Lambda\right)  $. We first note that%
\begin{align*}
E\left[  2\mathbb{B}U\left(  \theta_{0}\right)  \right]   &  =3\mathcal{I}%
^{-3}\left(  E\left[  \ell^{\theta\theta}\right]  \right)  ^{2}+\mathcal{I}%
^{-2}E\left[  \ell^{\theta\theta\theta}\right]  +6\mathcal{I}^{-3}E\left[
\ell^{\theta\theta}\right]  E\left[  \ell\ell^{\theta}\right] \\
&  +2\mathcal{I}^{-2}E\left[  \ell\ell^{\theta\theta}\right]  +2\mathcal{I}%
^{-2}E\left[  \left(  \ell^{\theta}\right)  ^{2}\right]  +3\mathcal{I}%
^{-3}E\left[  \ell^{\theta\theta}\right]  E\left[  \ell\ell^{\theta}\right]
+4\mathcal{I}^{-3}\left(  E\left[  \ell\ell^{\theta}\right]  \right)  ^{2}\\
&  +\mathcal{I}^{-2}E\left[  \ell\ell^{\theta\theta}\right]  +2\mathcal{I}%
^{-2}E\left[  \ell^{2}\ell^{\theta}\right]  +\mathcal{I}^{-3}E\left[
\ell^{\theta\theta}\right]  E\left[  \ell^{3}\right]  .
\end{align*}
We can then use the information equality $E\left[  \ell^{3}\right]  =-E\left[
\ell^{\theta\theta}\right]  -3E\left[  \ell\ell^{\theta}\right]  $, along with
the characterizations
\begin{align*}
M  &  =\left(  -E\left[  \ell^{\theta\theta}\right]  ,E\left[  \ell
^{\theta\theta\theta}\right]  ,E\left[  \left(  \ell^{\theta}\right)
^{2}+\ell\ell^{\theta\theta}\right]  ,2E\left[  \ell\ell^{\theta}\right]
\right)  ^{\prime},\\
\Lambda &  =\left(  -E\left[  \ell\ell^{\theta}\right]  ,E\left[  \ell
\ell^{\theta\theta}\right]  ,E\left[  \ell^{2}\ell^{\theta}\right]  ,E\left[
\ell^{3}\right]  \right)  ^{\prime},\\
\tau_{m}  &  =\left(  -\frac{E\left[  \ell^{\theta\theta}\right]  +2E\left[
\ell\ell^{\theta}\right]  }{\left(  E\left[  -\ell^{\theta}\right]  \right)
^{3}},\frac{1}{2\left(  E\left[  -\ell^{\theta}\right]  \right)  ^{2}}%
,\frac{1}{\left(  E\left[  -\ell^{\theta}\right]  \right)  ^{2}},0\right)
^{\prime}\\
&  =\left(  -\frac{E\left[  \ell^{\theta\theta}\right]  +2E\left[  \ell
\ell^{\theta}\right]  }{\mathcal{I}^{3}},\frac{1}{2\mathcal{I}^{2}},\frac
{1}{\mathcal{I}^{2}},0\right)  ^{\prime}%
\end{align*}
to conclude that $E\left[  \mathbb{B}U\left(  \theta_{0}\right)  \right]
=\tau_{m}\left(  M+\Lambda\right)  $.

\subsection{Proof of Proposition \ref{higher_var_ss} \label{HOV-SS}}

Here we derive the result for the two-sample version of the split-sample jackknife. Suppose that
\[
n=2m
\]
for some integer $k$ and $m$. From%
\[
E\left[  \widehat{\theta}\right]  =\theta+\frac{1}{n}b
\]
we see that%
\[
E\left[  \widehat{\theta}^{\left(  s\right)  }\right]  =\theta+\frac{b}%
{m}=\theta+\frac{1}{n/2}b=\theta+\frac{2}{n}b
\]
for $s=1,2$, where $\widehat{\theta}^{\left(  1\right)  }$ and
$\widehat{\theta}^{\left(  2\right)  }$ denote the estimators based on the
first and second subsamples of size $m$. We therefore have%
\begin{align*}
E\left[  2\widehat{\theta}\right]   &  =2\theta+\frac{2}{n}b\\
E\left[  \frac{1}{2}\left(  \widehat{\theta}^{\left(  1\right)  }%
+\widehat{\theta}^{\left(  2\right)  }\right)  \right]   &  =\theta+\frac
{2}{n}b
\end{align*}
so%
\[
E\left[  2\widehat{\theta}-\frac{1}{2}\left(  \widehat{\theta}^{\left(
1\right)  }+\widehat{\theta}^{\left(  2\right)  }\right)  \right]  =\theta
\]
which can be used as a basis of bias correction. We let
\[
\widehat{\theta}_{SS}=2\widehat{\theta}-\frac{1}{2}\left(  \widehat{\theta
}^{\left(  1\right)  }+\widehat{\theta}^{\left(  2\right)  }\right)  .
\]

From (\ref{V1}), (\ref{V2}), and (\ref{V3}), we can write%
\[
\sqrt{n}\left(  \widehat{\theta}-\theta_{0}\right)  =T_{1}+\frac{1}{\sqrt{n}%
}T_{2}+\frac{1}{n}T_{3}+o_{p}\left(  n^{-1}\right)
\]
where%
\begin{align}
T_{1}=  &  \frac{1}{\sqrt{n}}\sum_{i=1}^{n}Y_{1,i},\nonumber\\
T_{2}=  &  \left(  \frac{1}{\sqrt{n}}\sum_{i=1}^{n}Y_{1,i}\right)  \left(
\frac{1}{\sqrt{n}}\sum_{i=1}^{n}Y_{2,i}\right)  ,\label{T2-X1X2}\\
T_{3}=  &  \left(  \frac{1}{\sqrt{n}}\sum_{i=1}^{n}Y_{1,i}\right)  \left(
\frac{1}{\sqrt{n}}\sum_{i=1}^{n}Y_{3,i}\right)  \left(  \frac{1}{\sqrt{n}}%
\sum_{i=1}^{n}Y_{4,i}\right) \label{T3-X1X3X4}\\
&  +\left(  \frac{1}{\sqrt{n}}\sum_{i=1}^{n}Y_{1,i}\right)  \left(  \frac
{1}{\sqrt{n}}\sum_{i=1}^{n}Y_{5,i}\right)  \left(  \frac{1}{\sqrt{n}}%
\sum_{i=1}^{n}Y_{6,i}\right)  . \label{T3-X1X5X6}%
\end{align}
Here, we defined%
\begin{align*}
Y_{1,i}  &  \equiv\mathcal{I}^{-1}U_{i}\left(  \theta_{0}\right)  ,\\
Y_{2,i}  &  \equiv\frac{1}{2}\mathcal{I}^{-2}\mathcal{Q}_{1}\left(  \theta
_{0}\right)  U_{i}\left(  \theta_{0}\right)  +\mathcal{I}^{-1}V_{i}\left(
\theta_{0}\right)  ,\\
Y_{3,i}  &  \equiv Y_{1,i},\\
Y_{4,i}  &  \equiv\left(  \frac{1}{6}\mathcal{I}^{-2}\mathcal{Q}_{2}\left(
\theta_{0}\right)  +\frac{1}{2}\mathcal{I}^{-3}\mathcal{Q}_{1}\left(
\theta_{0}\right)  ^{2}\right)  U_{i}\left(  \theta_{0}\right)  +\frac{3}%
{2}\mathcal{I}^{-2}\mathcal{Q}_{1}\left(  \theta_{0}\right)  V_{i}\left(
\theta_{0}\right)  +\frac{1}{2}\mathcal{I}^{-1}W_{i}\left(  \theta_{0}\right)
,\\
Y_{5,i}  &  \equiv Y_{6,i}=\mathcal{I}^{-1}V_{i}\left(  \theta_{0}\right)  .
\end{align*}
Note that they all have zero means.

Let
\begin{align*}
\mathcal{Y}_{1}  &  =\frac{1}{\sqrt{n}}\sum_{i=1}^{n}Y_{1,i},\quad
\mathcal{Y}_{1}^{\left(  1\right)  }=\frac{1}{\sqrt{m}}\sum_{i=1}^{m}%
Y_{1,i},\quad\mathcal{Y}_{1}^{\left(  2\right)  }=\frac{1}{\sqrt{m}}%
\sum_{i=m+1}^{n}Y_{1,i}\\
\mathcal{Y}_{2}  &  =\frac{1}{\sqrt{n}}\sum_{i=1}^{n}Y_{2,i},\quad
\mathcal{Y}_{2}^{\left(  1\right)  }=\frac{1}{\sqrt{m}}\sum_{i=1}^{m}%
Y_{2,i},\quad\mathcal{Y}_{2}^{\left(  2\right)  }=\frac{1}{\sqrt{m}}%
\sum_{i=m+1}^{n}Y_{2,i}\\
\mathcal{Y}_{3}  &  =\frac{1}{\sqrt{n}}\sum_{i=1}^{n}Y_{3,i},\quad
\mathcal{Y}_{3}^{\left(  1\right)  }=\frac{1}{\sqrt{m}}\sum_{i=1}^{m}%
Y_{3,i},\quad\mathcal{Y}_{3}^{\left(  2\right)  }=\frac{1}{\sqrt{m}}%
\sum_{i=m+1}^{n}Y_{3,i}\\
\mathcal{Y}_{4}  &  =\frac{1}{\sqrt{n}}\sum_{i=1}^{n}Y_{4,i},\quad
\mathcal{Y}_{4}^{\left(  1\right)  }=\frac{1}{\sqrt{m}}\sum_{i=1}^{m}%
Y_{4,i},\quad\mathcal{Y}_{4}^{\left(  2\right)  }=\frac{1}{\sqrt{m}}%
\sum_{i=m+1}^{n}Y_{4,i}\\
\mathcal{Y}_{5}  &  =\frac{1}{\sqrt{n}}\sum_{i=1}^{n}Y_{5,i},\quad
\mathcal{Y}_{5}^{\left(  1\right)  }=\frac{1}{\sqrt{m}}\sum_{i=1}^{m}%
Y_{5,i},\quad\mathcal{Y}_{5}^{\left(  2\right)  }=\frac{1}{\sqrt{m}}%
\sum_{i=m+1}^{n}Y_{5,i}\\
\mathcal{Y}_{6}  &  =\frac{1}{\sqrt{n}}\sum_{i=1}^{n},\quad\mathcal{Y}%
_{6}^{\left(  1\right)  }=\frac{1}{\sqrt{m}}\sum_{i=1}^{m}Y_{6,i}%
,\quad\mathcal{Y}_{6}^{\left(  2\right)  }=\frac{1}{\sqrt{m}}\sum_{i=m+1}%
^{n}Y_{6,i}%
\end{align*}
Note that we can write%
\begin{equation}
\sqrt{n}\left(  \widehat{\theta}-\theta_{0}\right)  =\mathcal{Y}_{1}+\frac
{1}{\sqrt{n}}\mathcal{Y}_{1}\mathcal{Y}_{2}+\frac{1}{n}\left(  \mathcal{Y}%
_{1}\mathcal{Y}_{3}\mathcal{Y}_{4}+\mathcal{Y}_{1}\mathcal{Y}_{5}%
\mathcal{Y}_{6}\right)  +o_{p}\left(  n^{-1}\right)  \label{theta-abcdfg}%
\end{equation}

We prove

\begin{lemma}
\label{lem-ss-expansion}%
\[
\sqrt{n}\left(  \widehat{\theta}_{SS}-\theta\right)  =T_{1,SS}+\frac{1}%
{\sqrt{n}}T_{2,SS}+\frac{1}{n}T_{3,SS}+o_{p}\left(  n^{-1}\right)  ,
\]
where%
\begin{align*}
T_{1,SS}  &  =\frac{\mathcal{Y}_{1}^{\left(  1\right)  }+\mathcal{Y}%
_{1}^{\left(  2\right)  }}{\sqrt{2}}=\mathcal{Y}_{1}=T_{1},\\
T_{2,SS}  &  =\left(  \mathcal{Y}_{1}^{\left(  1\right)  }+\mathcal{Y}%
_{1}^{\left(  2\right)  }\right)  \left(  \mathcal{Y}_{2}^{\left(  1\right)
}+\mathcal{Y}_{2}^{\left(  2\right)  }\right)  -\left(  \mathcal{Y}%
_{1}^{\left(  1\right)  }\mathcal{Y}_{2}^{\left(  1\right)  }+\mathcal{Y}%
_{1}^{\left(  2\right)  }\mathcal{Y}_{2}^{\left(  2\right)  }\right)
=\mathcal{Y}_{1}^{\left(  1\right)  }\mathcal{Y}_{2}^{\left(  2\right)
}+\mathcal{Y}_{1}^{\left(  2\right)  }\mathcal{Y}_{2}^{\left(  1\right)  },
\end{align*}
and%
\begin{align}
T_{3,SS}=  &  \frac{1}{\sqrt{2}}\left(  \mathcal{Y}_{1}^{\left(  1\right)
}+\mathcal{Y}_{1}^{\left(  2\right)  }\right)  \left(  \mathcal{Y}%
_{3}^{\left(  1\right)  }+\mathcal{Y}_{3}^{\left(  2\right)  }\right)  \left(
\mathcal{Y}_{4}^{\left(  1\right)  }+\mathcal{Y}_{4}^{\left(  2\right)
}\right)   \nonumber\\
& -\sqrt{2}\left(  \mathcal{Y}_{1}^{\left(  1\right)  }\mathcal{Y}%
_{3}^{\left(  1\right)  }\mathcal{Y}_{4}^{\left(  1\right)  }+\mathcal{Y}%
_{1}^{\left(  2\right)  }\mathcal{Y}_{3}^{\left(  2\right)  }\mathcal{Y}%
_{4}^{\left(  2\right)  }\right) \nonumber\\
&  +\frac{1}{\sqrt{2}}\left(  \mathcal{Y}_{1}^{\left(  1\right)  }%
+\mathcal{Y}_{1}^{\left(  2\right)  }\right)  \left(  \mathcal{Y}_{5}^{\left(
1\right)  }+\mathcal{Y}_{5}^{\left(  2\right)  }\right)  \left(
\mathcal{Y}_{6}^{\left(  1\right)  }+\mathcal{Y}_{6}^{\left(  2\right)
}\right)   \nonumber\\
& -\sqrt{2}\left(  \mathcal{Y}_{1}^{\left(  1\right)  }\mathcal{Y}%
_{5}^{\left(  1\right)  }\mathcal{Y}_{6}^{\left(  1\right)  }+\mathcal{Y}%
_{1}^{\left(  2\right)  }\mathcal{Y}_{5}^{\left(  2\right)  }\mathcal{Y}%
_{6}^{\left(  2\right)  }\right) \nonumber\\
=  &  2\mathcal{Y}_{1}\mathcal{Y}_{3}\mathcal{Y}_{4}-\sqrt{2}\left(
\mathcal{Y}_{1}^{\left(  1\right)  }\mathcal{Y}_{3}^{\left(  1\right)
}\mathcal{Y}_{4}^{\left(  1\right)  }+\mathcal{Y}_{1}^{\left(  2\right)
}\mathcal{Y}_{3}^{\left(  2\right)  }\mathcal{Y}_{4}^{\left(  2\right)
}\right) \nonumber\\
&  +2\mathcal{Y}_{1}\mathcal{Y}_{5}\mathcal{Y}_{6}-\sqrt{2}\left(
\mathcal{Y}_{1}^{\left(  1\right)  }\mathcal{Y}_{5}^{\left(  1\right)
}\mathcal{Y}_{6}^{\left(  1\right)  }+\mathcal{Y}_{1}^{\left(  2\right)
}\mathcal{Y}_{5}^{\left(  2\right)  }\mathcal{Y}_{6}^{\left(  2\right)
}\right)  . \label{T3SS}%
\end{align}

\end{lemma}

\begin{proof}
We first prove%
\begin{align}
&  \sqrt{n}\left(  \widehat{\theta}_{SS}-\theta\right) \nonumber\\
&  =\frac{\mathcal{Y}_{1}^{\left(  1\right)  }+\mathcal{Y}_{1}^{\left(
2\right)  }}{\sqrt{2}}\nonumber\\
&  +\frac{1}{\sqrt{n}}\left(  \left(  \mathcal{Y}_{1}^{\left(  1\right)
}+\mathcal{Y}_{1}^{\left(  2\right)  }\right)  \left(  \mathcal{Y}%
_{2}^{\left(  1\right)  }+\mathcal{Y}_{2}^{\left(  2\right)  }\right)
-\left(  \mathcal{Y}_{1}^{\left(  1\right)  }\mathcal{Y}_{2}^{\left(
1\right)  }+\mathcal{Y}_{1}^{\left(  2\right)  }\mathcal{Y}_{2}^{\left(
2\right)  }\right)  \right) \nonumber\\
&  +\frac{1}{\sqrt{2}}\frac{1}{n}\left(  \mathcal{Y}_{1}^{\left(  1\right)
}+\mathcal{Y}_{1}^{\left(  2\right)  }\right)  \left(  \mathcal{Y}%
_{3}^{\left(  1\right)  }+\mathcal{Y}_{3}^{\left(  2\right)  }\right)  \left(
\mathcal{Y}_{4}^{\left(  1\right)  }+\mathcal{Y}_{4}^{\left(  2\right)
}\right)   \nonumber\\
& -\frac{\sqrt{2}}{n}\left(  \mathcal{Y}_{1}^{\left(  1\right)
}\mathcal{Y}_{3}^{\left(  1\right)  }\mathcal{Y}_{4}^{\left(  1\right)
}+\mathcal{Y}_{1}^{\left(  2\right)  }\mathcal{Y}_{3}^{\left(  2\right)
}\mathcal{Y}_{4}^{\left(  2\right)  }\right) \nonumber\\
&  +\frac{1}{\sqrt{2}}\frac{1}{n}\left(  \mathcal{Y}_{1}^{\left(  1\right)
}+\mathcal{Y}_{1}^{\left(  2\right)  }\right)  \left(  \mathcal{Y}%
_{5}^{\left(  1\right)  }+\mathcal{Y}_{5}^{\left(  2\right)  }\right)  \left(
\mathcal{Y}_{6}^{\left(  1\right)  }+\mathcal{Y}_{6}^{\left(  2\right)
}\right)   \nonumber\\
& -\frac{\sqrt{2}}{n}\left(  \mathcal{Y}_{1}^{\left(  1\right)
}\mathcal{Y}_{5}^{\left(  1\right)  }\mathcal{Y}_{6}^{\left(  1\right)
}+\mathcal{Y}_{1}^{\left(  2\right)  }\mathcal{Y}_{5}^{\left(  2\right)
}\mathcal{Y}_{6}^{\left(  2\right)  }\right)  \nonumber \\
&+o_{p}\left(  n^{-1}\right)
\label{SS-abcdfg}%
\end{align}
For this purpose, we write (\ref{theta-abcdfg}) as%
\begin{align*}
\widehat{\theta}-\theta_{0}  &  =\frac{1}{\sqrt{n}}\mathcal{Y}_{1}+\frac{1}%
{n}\mathcal{Y}_{1}\mathcal{Y}_{2}+\frac{1}{n\sqrt{n}}\left(  \mathcal{Y}%
_{1}\mathcal{Y}_{3}\mathcal{Y}_{4}+\mathcal{Y}_{1}\mathcal{Y}_{5}%
\mathcal{Y}_{6}\right)  +o_{p}\left(  n^{-3/2}\right) \\
&  =\frac{1}{\sqrt{2m}}\mathcal{Y}_{1}+\frac{1}{2m}\mathcal{Y}_{1}%
\mathcal{Y}_{2}+\left(  \frac{1}{\sqrt{2m}}\right)  ^{3}\left(  \mathcal{Y}%
_{1}\mathcal{Y}_{3}\mathcal{Y}_{4}+\mathcal{Y}_{1}\mathcal{Y}_{5}%
\mathcal{Y}_{6}\right)  +o_{p}\left(  n^{-3/2}\right)
\end{align*}
and use
\[
\mathcal{Y}_{1}=\frac{1}{\sqrt{2}}\left(  \mathcal{Y}_{1}^{\left(  1\right)
}+\mathcal{Y}_{1}^{\left(  2\right)  }\right)
\]
etc., to write
\begin{align*}
\widehat{\theta}-\theta_{0}=  &  \frac{1}{2\sqrt{m}}\left(  \mathcal{Y}%
_{1}^{\left(  1\right)  }+\mathcal{Y}_{1}^{\left(  2\right)  }\right)
+\left(  \frac{1}{2\sqrt{m}}\right)  ^{2}\left(  \mathcal{Y}_{1}^{\left(
1\right)  }+\mathcal{Y}_{1}^{\left(  2\right)  }\right)  \left(
\mathcal{Y}_{2}^{\left(  1\right)  }+\mathcal{Y}_{2}^{\left(  2\right)
}\right) \\
&  +\left(  \frac{1}{2\sqrt{m}}\right)  ^{3}\left(  \mathcal{Y}_{1}^{\left(
1\right)  }+\mathcal{Y}_{1}^{\left(  2\right)  }\right)  \left(
\mathcal{Y}_{3}^{\left(  1\right)  }+\mathcal{Y}_{3}^{\left(  2\right)
}\right)  \left(  \mathcal{Y}_{4}^{\left(  1\right)  }+\mathcal{Y}%
_{4}^{\left(  2\right)  }\right) \\
&  +\left(  \frac{1}{2\sqrt{m}}\right)  ^{3}\left(  \mathcal{Y}_{1}^{\left(
1\right)  }+\mathcal{Y}_{1}^{\left(  2\right)  }\right)  \left(
\mathcal{Y}_{5}^{\left(  1\right)  }+\mathcal{Y}_{5}^{\left(  2\right)
}\right)  \left(  \mathcal{Y}_{6}^{\left(  1\right)  }+\mathcal{Y}%
_{6}^{\left(  2\right)  }\right) \\
&  +o_{p}\left(  n^{-3/2}\right)  .
\end{align*}
Using (\ref{theta-abcdfg}), we also obtain%
\begin{align*}
\widehat{\theta}^{\left(  1\right)  }-\theta_{0}=  &  \frac{1}{\sqrt{m}%
}\mathcal{Y}_{1}^{\left(  1\right)  }+\left(  \frac{1}{\sqrt{m}}\right)
^{2}\mathcal{Y}_{1}^{\left(  1\right)  }\mathcal{Y}_{2}^{\left(  1\right)
}+\left(  \frac{1}{\sqrt{m}}\right)  ^{3}\mathcal{Y}_{1}^{\left(  1\right)
}\mathcal{Y}_{3}^{\left(  1\right)  }\mathcal{Y}_{4}^{\left(  1\right)  }\\
&  +\left(  \frac{1}{\sqrt{m}}\right)  ^{3}\mathcal{Y}_{1}^{\left(  1\right)
}\mathcal{Y}_{5}^{\left(  1\right)  }\mathcal{Y}_{6}^{\left(  1\right)
}+o_{p}\left(  n^{-3/2}\right)  ,\\
\widehat{\theta}^{\left(  2\right)  }-\theta_{0}=  &  \frac{1}{\sqrt{m}%
}\mathcal{Y}_{1}^{\left(  2\right)  }+\left(  \frac{1}{\sqrt{m}}\right)
^{2}\mathcal{Y}_{1}^{\left(  2\right)  }\mathcal{Y}_{2}^{\left(  2\right)
}+\left(  \frac{1}{\sqrt{m}}\right)  ^{3}\mathcal{Y}_{1}^{\left(  2\right)
}\mathcal{Y}_{3}^{\left(  2\right)  }\mathcal{Y}_{4}^{\left(  2\right)  }\\
&  +\left(  \frac{1}{\sqrt{m}}\right)  ^{3}\mathcal{Y}_{1}^{\left(  2\right)
}\mathcal{Y}_{5}^{\left(  2\right)  }\mathcal{Y}_{6}^{\left(  2\right)
}+o_{p}\left(  n^{-3/2}\right)  ,
\end{align*}
Combining them, we obtain%
\begin{align*}
\widehat{\theta}_{SS}-\theta &  =2\left(  \widehat{\theta}-\theta_{0}\right)
-\frac{\left(  \widehat{\theta}^{\left(  1\right)  }-\theta_{0}\right)
+\left(  \widehat{\theta}^{\left(  2\right)  }-\theta_{0}\right)  }{2}\\
&  =\frac{1}{\sqrt{m}}\left(  \mathcal{Y}_{1}^{\left(  1\right)  }%
+\mathcal{Y}_{1}^{\left(  2\right)  }\right)  -\frac{1}{2}\frac{1}{\sqrt{m}%
}\left(  \mathcal{Y}_{1}^{\left(  1\right)  }+\mathcal{Y}_{1}^{\left(
2\right)  }\right) \\
&  +\frac{1}{2}\left(  \frac{1}{\sqrt{m}}\right)  ^{2}\left(  \mathcal{Y}%
_{1}^{\left(  1\right)  }+\mathcal{Y}_{1}^{\left(  2\right)  }\right)  \left(
\mathcal{Y}_{2}^{\left(  1\right)  }+\mathcal{Y}_{2}^{\left(  2\right)
}\right)  -\frac{1}{2}\left(  \frac{1}{\sqrt{m}}\right)  ^{2}\left(
\mathcal{Y}_{1}^{\left(  1\right)  }\mathcal{Y}_{2}^{\left(  1\right)
}+\mathcal{Y}_{1}^{\left(  2\right)  }\mathcal{Y}_{2}^{\left(  2\right)
}\right) \\
&  +\frac{1}{4}\left(  \frac{1}{\sqrt{m}}\right)  ^{3}\left(  \mathcal{Y}%
_{1}^{\left(  1\right)  }+\mathcal{Y}_{1}^{\left(  2\right)  }\right)  \left(
\mathcal{Y}_{3}^{\left(  1\right)  }+\mathcal{Y}_{3}^{\left(  2\right)
}\right)  \left(  \mathcal{Y}_{4}^{\left(  1\right)  }+\mathcal{Y}%
_{4}^{\left(  2\right)  }\right) \\
&  -\frac{1}{2}\left(  \frac{1}{\sqrt{m}}\right)  ^{3}\left(  \mathcal{Y}%
_{1}^{\left(  1\right)  }\mathcal{Y}_{3}^{\left(  1\right)  }\mathcal{Y}%
_{4}^{\left(  1\right)  }+\mathcal{Y}_{1}^{\left(  2\right)  }\mathcal{Y}%
_{3}^{\left(  2\right)  }\mathcal{Y}_{4}^{\left(  2\right)  }\right) \\
&  +\frac{1}{4}\left(  \frac{1}{\sqrt{m}}\right)  ^{3}\left(  \mathcal{Y}%
_{1}^{\left(  1\right)  }+\mathcal{Y}_{1}^{\left(  2\right)  }\right)  \left(
\mathcal{Y}_{5}^{\left(  1\right)  }+\mathcal{Y}_{5}^{\left(  2\right)
}\right)  \left(  \mathcal{Y}_{6}^{\left(  1\right)  }+\mathcal{Y}%
_{6}^{\left(  2\right)  }\right) \\
&  -\frac{1}{2}\left(  \frac{1}{\sqrt{m}}\right)  ^{3}\left(  \mathcal{Y}%
_{1}^{\left(  1\right)  }\mathcal{Y}_{5}^{\left(  1\right)  }\mathcal{Y}%
_{6}^{\left(  1\right)  }+\mathcal{Y}_{1}^{\left(  2\right)  }\mathcal{Y}%
_{5}^{\left(  2\right)  }\mathcal{Y}_{6}^{\left(  2\right)  }\right)
+o_{p}\left(  n^{-3/2}\right)
\end{align*}
which can be rewritten as%
\begin{align*}
\widehat{\theta}_{SS}-\theta &  =\frac{1}{2}\frac{\sqrt{2}}{\sqrt{n}}\left(
\mathcal{Y}_{1}^{\left(  1\right)  }+\mathcal{Y}_{1}^{\left(  2\right)
}\right) \\
&  +\frac{1}{2}\left(  \frac{\sqrt{2}}{\sqrt{n}}\right)  ^{2}\left(
\mathcal{Y}_{1}^{\left(  1\right)  }+\mathcal{Y}_{1}^{\left(  2\right)
}\right)  \left(  \mathcal{Y}_{2}^{\left(  1\right)  }+\mathcal{Y}%
_{2}^{\left(  2\right)  }\right)  -\frac{1}{2}\left(  \frac{\sqrt{2}}{\sqrt
{n}}\right)  ^{2}\left(  \mathcal{Y}_{1}^{\left(  1\right)  }\mathcal{Y}%
_{2}^{\left(  1\right)  }+\mathcal{Y}_{1}^{\left(  2\right)  }\mathcal{Y}%
_{2}^{\left(  2\right)  }\right) \\
&  +\frac{1}{4}\left(  \frac{\sqrt{2}}{\sqrt{n}}\right)  ^{3}\left(
\mathcal{Y}_{1}^{\left(  1\right)  }+\mathcal{Y}_{1}^{\left(  2\right)
}\right)  \left(  \mathcal{Y}_{3}^{\left(  1\right)  }+\mathcal{Y}%
_{3}^{\left(  2\right)  }\right)  \left(  \mathcal{Y}_{4}^{\left(  1\right)
}+\mathcal{Y}_{4}^{\left(  2\right)  }\right) \\
&  -\frac{1}{2}\left(  \frac{\sqrt{2}}{\sqrt{n}}\right)  ^{3}\left(
\mathcal{Y}_{1}^{\left(  1\right)  }\mathcal{Y}_{3}^{\left(  1\right)
}\mathcal{Y}_{4}^{\left(  1\right)  }+\mathcal{Y}_{1}^{\left(  2\right)
}\mathcal{Y}_{3}^{\left(  2\right)  }\mathcal{Y}_{4}^{\left(  2\right)
}\right) \\
&  +\frac{1}{4}\left(  \frac{\sqrt{2}}{\sqrt{n}}\right)  ^{3}\left(
\mathcal{Y}_{1}^{\left(  1\right)  }+\mathcal{Y}_{1}^{\left(  2\right)
}\right)  \left(  \mathcal{Y}_{5}^{\left(  1\right)  }+\mathcal{Y}%
_{5}^{\left(  2\right)  }\right)  \left(  \mathcal{Y}_{6}^{\left(  1\right)
}+\mathcal{Y}_{6}^{\left(  2\right)  }\right) \\
&  -\frac{1}{2}\left(  \frac{\sqrt{2}}{\sqrt{n}}\right)  ^{3}\left(
\mathcal{Y}_{1}^{\left(  1\right)  }\mathcal{Y}_{5}^{\left(  1\right)
}\mathcal{Y}_{6}^{\left(  1\right)  }+\mathcal{Y}_{1}^{\left(  2\right)
}\mathcal{Y}_{5}^{\left(  2\right)  }\mathcal{Y}_{6}^{\left(  2\right)
}\right)  +o_{p}\left(  n^{-3/2}\right)
\end{align*}
which proves (\ref{SS-abcdfg}).
\end{proof}

Note that the higher order variance of $\widehat{\theta}_{SS}$ is equal to
\begin{equation}
E\left[  T_{1,SS}^{2}\right]  +\frac{1}{n}\left(  E\left[  \left(
T_{2,SS}\right)  ^{2}\right]  +2\sqrt{n}E\left[  T_{1,SS}T_{2,SS}\right]
+2E\left[  T_{1,SS}T_{3,SS}\right]  \right)  .\nonumber
\end{equation}

\begin{lemma}
\label{lem-ss-first-two-terms}%
\begin{align*}
E\left[  \left(  T_{2,SS}\right)  ^{2}\right]   &  =2E\left[  Y_{1,i}%
^{2}\right]  E\left[  Y_{2,i}^{2}\right]  +2\left(  E\left[  Y_{1,i}%
Y_{2,i}\right]  \right)  ^{2},\\
E\left[  T_{1}T_{2,SS}\right]   &  =0.
\end{align*}

\end{lemma}

\begin{proof}
It follows from%
\begin{align*}
E\left[  \left(  T_{2,SS}\right)  ^{2}\right]   &  =E\left[  \left(
\mathcal{Y}_{1}^{\left(  1\right)  }\mathcal{Y}_{2}^{\left(  2\right)
}+\mathcal{Y}_{1}^{\left(  2\right)  }\mathcal{Y}_{2}^{\left(  1\right)
}\right)  ^{2}\right] \\
&  =E\left[  \left(  \mathcal{Y}_{1}^{\left(  1\right)  }\right)  ^{2}\left(
\mathcal{Y}_{2}^{\left(  2\right)  }\right)  ^{2}+\left(  \mathcal{Y}%
_{1}^{\left(  2\right)  }\right)  ^{2}\left(  \mathcal{Y}_{2}^{\left(
1\right)  }\right)  ^{2}+2\mathcal{Y}_{1}^{\left(  1\right)  }\mathcal{Y}%
_{2}^{\left(  1\right)  }\mathcal{Y}_{1}^{\left(  2\right)  }\mathcal{Y}%
_{2}^{\left(  2\right)  }\right] \\
&  =E\left[  \left(  \mathcal{Y}_{1}^{\left(  1\right)  }\right)  ^{2}\right]
E\left[  \left(  \mathcal{Y}_{2}^{\left(  2\right)  }\right)  ^{2}\right]
+E\left[  \left(  \mathcal{Y}_{1}^{\left(  2\right)  }\right)  ^{2}\right]
E\left[  \left(  \mathcal{Y}_{2}^{\left(  1\right)  }\right)  ^{2}\right] \\
& 
+2E\left[  \mathcal{Y}_{1}^{\left(  1\right)  }\mathcal{Y}_{2}^{\left(
1\right)  }\right]  E\left[  \mathcal{Y}_{1}^{\left(  2\right)  }%
\mathcal{Y}_{2}^{\left(  2\right)  }\right] \\
&  =2E\left[  Y_{1,i}^{2}\right]  E\left[  Y_{2,i}^{2}\right]  +2\left(
E\left[  Y_{1,i}Y_{2,i}\right]  \right)  ^{2},
\end{align*}
and%
\[
E\left[  T_{1}T_{2,SS}\right]  =E\left[  \left(  \frac{1}{\sqrt{2}}%
\mathcal{Y}_{1}^{\left(  1\right)  }+\frac{1}{\sqrt{2}}\mathcal{Y}%
_{1}^{\left(  2\right)  }\right)  \left(  \mathcal{Y}_{1}^{\left(  1\right)
}\mathcal{Y}_{2}^{\left(  2\right)  }+\mathcal{Y}_{1}^{\left(  2\right)
}\mathcal{Y}_{2}^{\left(  1\right)  }\right)  \right]  =0,
\]
where we used the fact that the two subsamples are independent of each other.
\end{proof}

\begin{lemma}
\label{lem-fourth-moment-Y}%
\begin{align*}
E\left[  \mathcal{Y}_{1}\mathcal{Y}_{2}\mathcal{Y}_{3}\mathcal{Y}_{4}\right]
&  =\frac{1}{n}E\left[  Y_{1,i}Y_{2,i}Y_{3,i}Y_{4,i}\right] \\
&  +\frac{n\left(  n-1\right)  }{n^{2}}\left(
\begin{array}
[c]{c}%
E\left[  Y_{1,j}Y_{2,j}\right]  E\left[  Y_{3,i}Y_{4,i}\right]  +E\left[
Y_{1,j}Y_{3,j}\right]  E\left[  Y_{2,i}Y_{4,i}\right] \\
+E\left[  Y_{1,j}Y_{4,j}\right]  E\left[  Y_{2,i}Y_{3,i}\right]
\end{array}
\right)
\end{align*}

\end{lemma}

\begin{proof}
The result follows from%
\begin{align*}
E\left[  \mathcal{Y}_{1}\mathcal{Y}_{2}\mathcal{Y}_{3}\mathcal{Y}_{4}\right]
&  =E\left[  \left(  \frac{1}{\sqrt{n}}\sum_{i=1}^{n}Y_{1,i}\right)  \left(
\frac{1}{\sqrt{n}}\sum_{i=1}^{n}Y_{2,i}\right)  \left(  \frac{1}{\sqrt{n}}%
\sum_{i=1}^{n}Y_{3,i}\right)  \left(  \frac{1}{\sqrt{n}}\sum_{i=1}^{n}%
Y_{4,i}\right)  \right] \\
&  =\frac{1}{n^{2}}\sum_{i=1}^{n}E\left[  Y_{1,i}Y_{2,i}Y_{3,i}Y_{4,i}\right]
\\
&  +\frac{1}{n^{2}}\sum_{i\neq j}\left(
\begin{array}
[c]{c}%
E\left[  Y_{1,j}Y_{2,j}Y_{3,i}Y_{4,i}\right]  +E\left[  Y_{1,j}Y_{2,i}%
Y_{3,j}Y_{4,i}\right]  +E\left[  Y_{1,j}Y_{2,i}Y_{3,i}Y_{4,j}\right] \\
E\left[  Y_{1,i}Y_{2,j}Y_{3,j}Y_{4,i}\right]  +E\left[  Y_{1,i}Y_{2,j}%
Y_{3,i}Y_{4,j}\right]  +E\left[  Y_{1,i}Y_{2,i}Y_{3,j}Y_{4,j}\right]
\end{array}
\right) \\
&  =\frac{1}{n}E\left[  Y_{1,i}Y_{2,i}Y_{3,i}Y_{4,i}\right] \\
&  +\frac{n\left(  n-1\right)  }{n^{2}}\left(  E\left[  Y_{1,j}Y_{2,j}%
Y_{3,i}Y_{4,i}\right]  +E\left[  Y_{1,j}Y_{2,i}Y_{3,j}Y_{4,i}\right]
+E\left[  Y_{1,j}Y_{2,i}Y_{3,i}Y_{4,j}\right]  \right)  .
\end{align*}

\end{proof}

\begin{lemma}
\label{lem-SS-3-direct}%
\[
E\left[  T_{1}T_{3,SS}\right]  =O\left(  \frac{1}{n}\right)
\]

\end{lemma}

\begin{proof}
We use $\mathcal{Y}_{1}=\frac{\mathcal{Y}_{1}^{\left(  1\right)  }%
+\mathcal{Y}_{1}^{\left(  2\right)  }}{\sqrt{2}}$ and write
\begin{align*}
&  E\left[  \mathcal{Y}_{1}\left(  2\mathcal{Y}_{1}\mathcal{Y}_{3}%
\mathcal{Y}_{4}-\sqrt{2}\left(  \mathcal{Y}_{1}^{\left(  1\right)
}\mathcal{Y}_{3}^{\left(  1\right)  }\mathcal{Y}_{4}^{\left(  1\right)
}+\mathcal{Y}_{1}^{\left(  2\right)  }\mathcal{Y}_{3}^{\left(  2\right)
}\mathcal{Y}_{4}^{\left(  2\right)  }\right)  \right)  \right] \\
&  =E\left[  2\mathcal{Y}_{1}^{2}\mathcal{Y}_{3}\mathcal{Y}_{4}-\sqrt
{2}\mathcal{Y}_{1}\left(  \mathcal{Y}_{1}^{\left(  1\right)  }\mathcal{Y}%
_{3}^{\left(  1\right)  }\mathcal{Y}_{4}^{\left(  1\right)  }+\mathcal{Y}%
_{1}^{\left(  2\right)  }\mathcal{Y}_{3}^{\left(  2\right)  }\mathcal{Y}%
_{4}^{\left(  2\right)  }\right)  \right] \\
&  =E\left[  2\mathcal{Y}_{1}^{2}\mathcal{Y}_{3}\mathcal{Y}_{4}-\sqrt
{2}\left(  \frac{\mathcal{Y}_{1}^{\left(  1\right)  }+\mathcal{Y}_{1}^{\left(
2\right)  }}{\sqrt{2}}\right)  \left(  \mathcal{Y}_{1}^{\left(  1\right)
}\mathcal{Y}_{3}^{\left(  1\right)  }\mathcal{Y}_{4}^{\left(  1\right)
}+\mathcal{Y}_{1}^{\left(  2\right)  }\mathcal{Y}_{3}^{\left(  2\right)
}\mathcal{Y}_{4}^{\left(  2\right)  }\right)  \right] \\
&  =E\left[  2\mathcal{Y}_{1}^{2}\mathcal{Y}_{3}\mathcal{Y}_{4}-\left(
\mathcal{Y}_{1}^{\left(  1\right)  }+\mathcal{Y}_{1}^{\left(  2\right)
}\right)  \left(  \mathcal{Y}_{1}^{\left(  1\right)  }\mathcal{Y}_{3}^{\left(
1\right)  }\mathcal{Y}_{4}^{\left(  1\right)  }+\mathcal{Y}_{1}^{\left(
2\right)  }\mathcal{Y}_{3}^{\left(  2\right)  }\mathcal{Y}_{4}^{\left(
2\right)  }\right)  \right] \\
&  =2E\left[  \mathcal{Y}_{1}^{2}\mathcal{Y}_{3}\mathcal{Y}_{4}\right]
-E\left[  \left(  \mathcal{Y}_{1}^{\left(  1\right)  }+\mathcal{Y}%
_{1}^{\left(  2\right)  }\right)  \left(  \mathcal{Y}_{1}^{\left(  1\right)
}\mathcal{Y}_{3}^{\left(  1\right)  }\mathcal{Y}_{4}^{\left(  1\right)
}+\mathcal{Y}_{1}^{\left(  2\right)  }\mathcal{Y}_{3}^{\left(  2\right)
}\mathcal{Y}_{4}^{\left(  2\right)  }\right)  \right]
\end{align*}
but using the fact that the two subsamples are independent of each other for
the second equality, and the fact that $\mathcal{Y}_{1}^{\left(  1\right)  }$
and $\mathcal{Y}_{1}^{\left(  2\right)  }$ have identical distribution, along
with the modification of Lemma \ref{lem-fourth-moment-Y} replacing
$\mathcal{Y}_{2}$ by $\mathcal{Y}_{1}$, , e.g. we get
\begin{align*}
E\left[  \left(  \mathcal{Y}_{1}^{\left(  1\right)  }+\mathcal{Y}_{1}^{\left(
2\right)  }\right)  \left(  \mathcal{Y}_{1}^{\left(  1\right)  }%
\mathcal{Y}_{3}^{\left(  1\right)  }\mathcal{Y}_{4}^{\left(  1\right)
}+\mathcal{Y}_{1}^{\left(  2\right)  }\mathcal{Y}_{3}^{\left(  2\right)
}\mathcal{Y}_{4}^{\left(  2\right)  }\right)  \right]   &  =2E\left[  \left(
\mathcal{Y}_{1}^{\left(  1\right)  }\right)  ^{2}\mathcal{Y}_{3}^{\left(
1\right)  }\mathcal{Y}_{4}^{\left(  1\right)  }\right] \\
&  =2E\left[  \mathcal{Y}_{1}^{2}\mathcal{Y}_{3}\mathcal{Y}_{4}\right]
+O\left(  \frac{1}{n}\right)  ,
\end{align*}
from which we obtain%
\[
E\left[  \mathcal{Y}_{1}\left(  2\mathcal{Y}_{1}\mathcal{Y}_{3}\mathcal{Y}%
_{4}-\sqrt{2}\left(  \mathcal{Y}_{1}^{\left(  1\right)  }\mathcal{Y}%
_{3}^{\left(  1\right)  }\mathcal{Y}_{4}^{\left(  1\right)  }+\mathcal{Y}%
_{1}^{\left(  2\right)  }\mathcal{Y}_{3}^{\left(  2\right)  }\mathcal{Y}%
_{4}^{\left(  2\right)  }\right)  \right)  \right]  =O\left(  \frac{1}%
{n}\right)
\]
Likewise, we have%
\[
E\left[  \mathcal{Y}_{1}\left(  2\mathcal{Y}_{1}\mathcal{Y}_{5}\mathcal{Y}%
_{6}-\sqrt{2}\left(  \mathcal{Y}_{1}^{\left(  1\right)  }\mathcal{Y}%
_{5}^{\left(  1\right)  }\mathcal{Y}_{6}^{\left(  1\right)  }+\mathcal{Y}%
_{1}^{\left(  2\right)  }\mathcal{Y}_{5}^{\left(  2\right)  }\mathcal{Y}%
_{6}^{\left(  2\right)  }\right)  \right)  \right]  =O\left(  \frac{1}%
{n}\right)
\]
which gives us the conclusion that%
\[
E\left[  T_{1}T_{3,SS}\right]  =O\left(  \frac{1}{n}\right)
\]

\end{proof}

\subsection{Proof of Proposition \ref{higher_var} \label{HOV-jack}}
Here we derive the higher-order variance for the leave-one-out jackknife estimator. By the equivalence results in the main paper, this is also the higher-order variance of the bootstrap and analytical bias-corrected estimators. 
We present an alternative expansion of the jackknife estimator.

\begin{lemma}
\label{lem-J-two-terms}%
\[
\sqrt{n}\left(  \widehat{\theta}_{J}-\theta_{0}\right)  =T_{1}+\frac{1}%
{\sqrt{n}}T_{2,J}+\frac{1}{n}T_{3,J}+o_{p}\left(  n^{-1}\right)  ,
\]
where%
\begin{align*}
T_{2,J}  &  =\frac{1}{n-1}\sum_{i\neq j}Y_{1,i}Y_{2,j},\\
T_{3,J}  &  =n\mathcal{Y}_{1}\mathcal{Y}_{3}\mathcal{Y}_{4}-\sqrt{\frac
{n}{n-1}}\sum_{j=1}^{n}\mathcal{Y}_{1,\left(  i\right)  }\mathcal{Y}%
_{3,\left(  i\right)  }\mathcal{Y}_{4,\left(  i\right)  }\\
&  +n\mathcal{Y}_{1}\mathcal{Y}_{5}\mathcal{Y}_{6}-\sqrt{\frac{n}{n-1}}%
\sum_{j=1}^{n}\mathcal{Y}_{1,\left(  i\right)  }\mathcal{Y}_{5,\left(
i\right)  }\mathcal{Y}_{6,\left(  i\right)  },
\end{align*}
and%
\[
\mathcal{Y}_{1,\left(  i\right)  }=\frac{1}{\sqrt{n-1}}\sum_{j\neq i}%
^{n}Y_{1j}=\frac{\sqrt{n}\mathcal{Y}_{1}-Y_{1,i}}{\sqrt{n-1}}%
\]
etc.
\end{lemma}

\begin{proof}
Combining (\ref{J-expansion}) with (\ref{r6}), (\ref{r5}), (\ref{r4}), the LHS
of (\ref{r1-3}) as well as, we obtain%
\begin{align*}
\lefteqn{\sqrt{n}\left(  \widehat{\theta}_{J}-\theta_{0}\right)
=T_{1,J}+\frac{1}{\sqrt{n}}T_{2,J}+\frac{1}{n}T_{3,J}+o_{p}\left(
n^{-1}\right)  }\\
&  =T_{1}+\frac{1}{2}\frac{1}{\sqrt{n}}\left(  n\theta^{\epsilon\epsilon}%
-\sum_{j=1}^{n}\theta_{\left(  j\right)  }^{\epsilon\epsilon}\right)
+\frac{1}{6}\frac{1}{n}\left(  n\theta^{\epsilon\epsilon\epsilon}-\sqrt
{\frac{n}{n-1}}\sum_{j=1}^{n}\theta_{\left(  j\right)  }^{\epsilon
\epsilon\epsilon}\right)  +o_{p}\left(  n^{-1}\right)  .
\end{align*}
By (\ref{T2-X1X2}) and Lemma \ref{Vjack2}, we see that%
\[
T_{2,J}=\frac{1}{2}\left(  n\theta^{\epsilon\epsilon}-\sum_{j=1}^{n}%
\theta_{\left(  j\right)  }^{\epsilon\epsilon}\right)  =\frac{1}{n-1}%
\sum_{i\neq j}Y_{1,i}Y_{2,j}.
\]
Also by definition, we see that%
\begin{align*}
T_{3,J}  &  =\frac{1}{6}\left(  n\theta^{\epsilon\epsilon\epsilon}-\sqrt
{\frac{n}{n-1}}\sum_{j=1}^{n}\theta_{\left(  j\right)  }^{\epsilon
\epsilon\epsilon}\right) \\
&  =n\mathcal{Y}_{1}\mathcal{Y}_{3}\mathcal{Y}_{4}-\sqrt{\frac{n}{n-1}}%
\sum_{j=1}^{n}\mathcal{Y}_{1,\left(  i\right)  }\mathcal{Y}_{3,\left(
i\right)  }\mathcal{Y}_{4,\left(  i\right)  }\\
&  +n\mathcal{Y}_{1}\mathcal{Y}_{5}\mathcal{Y}_{6}-\sqrt{\frac{n}{n-1}}%
\sum_{j=1}^{n}\mathcal{Y}_{1,\left(  i\right)  }\mathcal{Y}_{5,\left(
i\right)  }\mathcal{Y}_{6,\left(  i\right)  }.
\end{align*}

\end{proof}

The higher order variance of the jackknife estimator is%
\begin{equation}
E\left[  T_{1,J}^{2}\right]  +\frac{1}{n}\left(  E\left[  \left(
T_{2,J}\right)  ^{2}\right]  +2\sqrt{n}E\left[  T_{1}T_{2,J}\right]
+2E\left[  T_{1}T_{3,J}\right]  \right)  . \label{J-high-var}%
\end{equation}

\begin{lemma}
\label{lem-J-first-two-terms}%
\begin{align*}
E\left[  \left(  T_{2,J}\right)  ^{2}\right]   &  =\frac{1}{2}E\left[  \left(
T_{2,SS}\right)  ^{2}\right]  +O\left(  n^{-1}\right)  ,\\
E\left[  T_{1}T_{2,J}\right]   &  =0.
\end{align*}

\end{lemma}

\begin{proof}
We note that%
\begin{align}
E\left[  \left(  T_{2,J}\right)  ^{2}\right]   &  =\frac{1}{\left(
n-1\right)  ^{2}}E\left[  \sum_{i_{1}\neq j_{1}}\sum_{i_{2}\neq j_{2}%
}Y_{1,i_{1}}Y_{2,j_{1}}Y_{1,i_{2}}Y_{2,j_{2}}\right] \nonumber\\
&  =\frac{1}{\left(  n-1\right)  ^{2}}\sum_{i_{1}=i_{2}\neq j_{1}=j_{2}%
}E\left[  Y_{1,i_{1}}Y_{2,j_{1}}Y_{1,i_{2}}Y_{2,j_{2}}\right] \nonumber\\
&  +\frac{1}{\left(  n-1\right)  ^{2}}\sum_{i_{1}=j_{2}\neq j_{1}=i_{2}%
}E\left[  Y_{1,i_{1}}Y_{2,j_{1}}Y_{1,i_{2}}Y_{2,j_{2}}\right] \nonumber\\
&  =\frac{n\left(  n-1\right)  }{\left(  n-1\right)  ^{2}}\left(  E\left[
Y_{1,i}^{2}\right]  E\left[  Y_{2,i}^{2}\right]  +\left(  E\left[
Y_{1,i}Y_{2,i}\right]  \right)  ^{2}\right) \nonumber\\
&  =E\left[  Y_{1,i}^{2}\right]  E\left[  Y_{2,i}^{2}\right]  +\left(
E\left[  Y_{1,i}Y_{2,i}\right]  \right)  ^{2}+O\left(  n^{-1}\right)
\nonumber\\
&  =\frac{1}{2}E\left[  \left(  T_{2,SS}\right)  ^{2}\right]  +O\left(
n^{-1}\right)  , \label{J1}%
\end{align}
and%
\begin{equation}
E\left[  T_{1}T_{2,J}\right]  =E\left[  \left(  \frac{1}{\sqrt{n}}\sum
_{i=1}^{n}Y_{1,i}\right)  \left(  \frac{1}{n-1}\sum_{i\neq j}Y_{1,i}%
Y_{2,j}\right)  \right]  =0=E\left[  T_{1}T_{2,SS}\right]  , \label{J2}%
\end{equation}
where we used Lemma \ref{lem-ss-first-two-terms} in the last equalities of
(\ref{J1}) and (\ref{J2}).
\end{proof}

\begin{lemma}
\label{lem-J-third-term-details}%
\[
E\left[  T_{1}T_{3,J}\right]  =O\left(  \frac{1}{n}\right)  .
\]

\end{lemma}

\begin{proof}
We note that%
\begin{align*}
&  E\left[  \mathcal{Y}_{1}\left(  n\mathcal{Y}_{1}\mathcal{Y}_{3}%
\mathcal{Y}_{4}-\sqrt{\frac{n}{n-1}}\sum_{i=1}^{n}\mathcal{Y}_{1,\left(
i\right)  }\mathcal{Y}_{3,\left(  i\right)  }\mathcal{Y}_{4,\left(  i\right)
}\right)  \right] \\
&  =nE\left[  \mathcal{Y}_{1}^{2}\mathcal{Y}_{3}\mathcal{Y}_{4}\right]
-\sqrt{\frac{n}{n-1}}\sum_{i=1}^{n}E\left[  \mathcal{Y}_{1}\mathcal{Y}%
_{1,\left(  i\right)  }\mathcal{Y}_{3,\left(  i\right)  }\mathcal{Y}%
_{4,\left(  i\right)  }\right] \\
&  =nE\left[  \mathcal{Y}_{1}^{2}\mathcal{Y}_{3}\mathcal{Y}_{4}\right]
-\sqrt{\frac{n}{n-1}}\sum_{i=1}^{n}E\left[  \frac{\sqrt{n-1}\mathcal{Y}%
_{1,\left(  i\right)  }+Y_{1,i}}{\sqrt{n}}\mathcal{Y}_{1,\left(  i\right)
}\mathcal{Y}_{3,\left(  i\right)  }\mathcal{Y}_{4,\left(  i\right)  }\right]
\\
&  =nE\left[  \mathcal{Y}_{1}^{2}\mathcal{Y}_{3}\mathcal{Y}_{4}\right]
-\sum_{i=1}^{n}E\left[  \mathcal{Y}_{1,\left(  i\right)  }^{2}\mathcal{Y}%
_{3,\left(  i\right)  }\mathcal{Y}_{4,\left(  i\right)  }\right]  -\frac
{1}{\sqrt{n-1}}\sum_{i=1}^{n}E\left[  Y_{1,i}\mathcal{Y}_{1,\left(  i\right)
}\mathcal{Y}_{3,\left(  i\right)  }\mathcal{Y}_{4,\left(  i\right)  }\right]
\\
&  =nE\left[  \mathcal{Y}_{1}^{2}\mathcal{Y}_{3}\mathcal{Y}_{4}\right]
-nE\left[  \mathcal{Y}_{1,\left(  i\right)  }^{2}\mathcal{Y}_{3,\left(
i\right)  }\mathcal{Y}_{4,\left(  i\right)  }\right]  +0\\
&  =n\left(  \frac{1}{n}C_{1}+\frac{n\left(  n-1\right)  }{n^{2}}C_{2}\right)
-n\left(  \frac{1}{n-1}C_{1}+\frac{\left(  n-1\right)  \left(  n-2\right)
}{\left(  n-1\right)  ^{2}}C_{2}\right) \\
&  =-\frac{1}{n-1}C_{1}+\frac{C_{2}}{n-1}\\
&  =O\left(  \frac{1}{n}\right)  ,
\end{align*}
where we used the modification of Lemma \ref{lem-fourth-moment-Y} replacing
$\mathcal{Y}_{2}$ by $\mathcal{Y}_{1}$ in the third last equality. Here, we
defined
\begin{align*}
C_{1}  &  =E\left[  Y_{1,i}Y_{2,i}Y_{3,i}Y_{4,i}\right] \\
C_{2}  &  =2\left(  E\left[  Y_{1,j}Y_{2,j}\right]  E\left[  Y_{3,i}%
Y_{4,i}\right]  +E\left[  Y_{1,j}Y_{3,j}\right]  E\left[  Y_{2,i}%
Y_{4,i}\right]  +E\left[  Y_{1,j}Y_{4,j}\right]  E\left[  Y_{2,i}%
Y_{3,i}\right]  \right)  .
\end{align*}

\end{proof}

From Lemmas \ref{lem-ss-first-two-terms} and \ref{lem-SS-3-direct} along with
Lemmas \ref{lem-J-first-two-terms} and \ref{lem-J-third-term-details}, we
conclude that%
\begin{align*}
&  E\left[  T_{1,SS}^{2}\right]  +\frac{1}{n}\left(  E\left[  \left(
T_{2,SS}\right)  ^{2}\right]  +2\sqrt{n}E\left[  T_{1,SS}T_{2,SS}\right]
+2E\left[  T_{1,SS}T_{3,SS}\right]  \right) \\
&  -\left(  E\left[  T_{1,J}^{2}\right]  +\frac{1}{n}\left(  E\left[  \left(
T_{2,J}\right)  ^{2}\right]  +2\sqrt{n}E\left[  T_{1}T_{2,J}\right]
+2E\left[  T_{1}T_{3,J}\right]  \right)  \right) \\
&  =\frac{1}{n}\left(  E\left[  Y_{1,i}^{2}\right]  E\left[  Y_{2,i}%
^{2}\right]  +\left(  E\left[  Y_{1,i}Y_{2,i}\right]  \right)  ^{2}+O\left(
\frac{1}{n}\right)  \right)  .
\end{align*}
In other words, the split-sample jackknife has higher order variance that is
strictly larger than jackknife bias correction by $\frac{1}{n}\left(  E\left[
Y_{1,i}^{2}\right]  E\left[  Y_{2,i}^{2}\right]  +\left(  E\left[
Y_{1,i}Y_{2,i}\right]  \right)  ^{2}\right)  $.

\section{Proofs for Panel Models \label{proof-for-panel-models}}

\subsection{Expansions and V-statistic forms}

The method of deriving the expansions in this paper follows that used in Hahn
and Newey (2004). For notational simplicity the proofs are done for scalar
$\theta$, with the vector case expected to follow similarly. To describe the
expansions, first note that the MLE solves the moment conditions
\begin{align*}
0  &  =\sum_{i=1}^{n}\sum_{t=1}^{T}U(z_{it};\hat{\theta},\hat{\alpha}_{i}%
(\hat{\theta}))\\
0  &  =\sum_{t=1}^{T}V(z_{it};\hat{\theta},\hat{\alpha}_{i}(\hat{\theta
}))\quad\text{for all }i
\end{align*}
for $U(z_{it};\theta,\alpha)=u(z_{it};\theta,\alpha)-\delta V(z_{it}%
;\theta,\alpha)$, where $\delta=E[u_{it}V_{it}]/E[V_{it}^{2}]$. Let
$F\equiv(F_{1},\dots,F_{n})$ be the collection of distribution functions and
$\hat{F}$ the corresponding collection of empirical distributions. Next,
define $F(\epsilon)=F+\epsilon\sqrt{T}(\hat{F}-F)$ for $\epsilon\in
\lbrack0,T^{-1/2}]$, and let $\alpha_{i}(\theta,F_{i}(\epsilon))$ and
$\theta(F(\epsilon))$ be solutions to the moment equations
\begin{align*}
0  &  =\int V\left(  z_{it};\theta,\alpha_{i}(\theta,F_{i}(\epsilon))\right)
dF_{i}(\epsilon),\quad\text{for all }i\\
0  &  =\sum_{i=1}^{n}\int U\left(  z_{it};\theta(F(\epsilon)),\alpha
_{i}\left(  \theta(F(\epsilon)),F_{i}(\epsilon)\right)  \right)
dF_{i}(\epsilon)
\end{align*}

The expansion is generated via Taylor series expansion of $\theta
(F(\epsilon))$
\[
\hat{\theta}-\theta_{0}=\frac{1}{\sqrt{T}}\theta^{\epsilon}(0)+\frac{1}%
{2}\frac{1}{T}\theta^{\epsilon\epsilon}(0)+\frac{1}{6}\frac{1}{T^{3/2}}%
\theta^{\epsilon\epsilon\epsilon}(0)+\frac{1}{24}\frac{1}{T^{2}}%
\theta^{\epsilon\epsilon\epsilon\epsilon}(0)+R_{n,T}%
\]
where $\theta^{\epsilon}(0)=d\theta(F(\epsilon))/d\epsilon$, $\theta
^{\epsilon\epsilon}(0)=d^{2}\theta(F(\epsilon))/d\epsilon^{2}$ etc. evaluated
at $\epsilon=0$. In the main paper we use the notation $A_{n}=\theta
^{\epsilon}(0)$, $B_{n}=\frac{1}{2}\theta^{\epsilon\epsilon}(0)$ and so on to
avoid introducing complicated notation. Lemmas 1-5 of the Supplementary
Appendix II show that $E[A_{n}]=O(1)$, $E[B_{n}]=O(1)$, and so on, while
$\sqrt{n}(A_{n}-E[A_{n}])=O_{p}(1)$, $\sqrt{n}(B_{n}-E[B_{n}])=O_{p}(1)$
etc.\footnote{The results are established up to sixth-order terms} Further, in
Section D of the Supplementary Appendix II we show that $R_{n,T}=o_{p}%
(T^{-1})$. These results allow us to write the expansion
\begin{equation}
\sqrt{nT}(\hat{\theta}-\theta_{0})=\sqrt{n}A_{n}+\frac{\sqrt{n}}{\sqrt{T}%
}B_{n}+\frac{\sqrt{n}}{T}C_{n}+\frac{\sqrt{n}}{T\sqrt{T}}D_{n}+o_{p}(T^{-1})
\label{eq:expansionA}%
\end{equation}

\subsubsection{Normalized V-statistics}

The proofs of most of the results in the paper make use of a particular
structure for the terms in the above expansion. Consider a statistic of the
form
\begin{align*}
W_{i,T,m}  &  \equiv\frac{1}{T^{m/2}}\sum_{t_{1}=1}^{T}\sum_{t_{2}=1}%
^{T}\cdots\sum_{t_{m}=1}^{T}k_{1}\left(  x_{i,t_{1}}\right)  k_{2}\left(
x_{i,t_{2}}\right)  \cdots k_{m}\left(  x_{i,t_{m}}\right) \\
&  \equiv T^{m/2}\overline{k}_{i,1}\overline{k}_{i,2}\cdots\overline{k}_{i,m}%
\end{align*}
where $E\left[  k_{j}\left(  x_{i,t}\right)  \right]  =0$. We will call the
average
\[
W_{T,(m)}=\frac{1}{n}\sum_{i=1}^{n}W_{i,T,m}%
\]
a normalized V-statistic of order $m$. The elements of the expansion can all
be shown to be products of such V-statistics
\[
W_{T,(m_{1},\dots,m_{L})}=W_{T,m_{1}}\cdots W_{T,m_{L}}%
\]
Inspection of the expansion terms (see appendix to Hahn and Newey (2004))
shows the first-order term $A_{n}$ is a V-statistic with $L=1$ and $m_{1}=1$,
while the second-order term $B_{n}$ contains V-statistics with $L\leq2$ and
$\sum_{l}m_{l}=2$, the third-order term $C_{n}$ contains V-statistics with
$L\leq3$ and $\sum_{l}m_{l}=3$, and so on. Below we prove results for these
general statistics, which subsequently imply the results for the expansion terms.

\subsection{Jackknifing V-statistics}

Here we show the impact of the jackknife and split-sample bias corrections on
V-statistics up to third-order. We focus on $W_{i,T,m}$ rather than averages
over $i$, since the bias corrections only act on the time series dimension.
Recall that an $m$-th order expansion term is an average of terms
$T^{-m/2}W_{i,T,m}$. We can write the leave-one-out statistics as
\begin{align*}
W_{i,T,m}^{(-t)}  &  =\frac{1}{(T-1)^{m/2}}\sum_{t_{1}\neq t}^{T}\sum
_{t_{2}\neq t}^{T}\cdots\sum_{t_{m}\neq t}^{T}k_{1}\left(  x_{i,t_{1}}\right)
k_{2}\left(  x_{i,t_{2}}\right)  \cdots k_{m}\left(  x_{i,t_{m}}\right) \\
&  =(T-1)^{m/2}\frac{T\bar{k}_{i,1}-k_{1}(x_{i,t})}{T-1}\cdots\frac{T\bar
{k}_{i,m}-k_{m}(x_{i,t})}{T-1}\\
&  =(T-1)^{-m/2}\left(  T\bar{k}_{i,1}-k_{1}(x_{i,t})\right)  \cdots\left(
T\bar{k}_{i,m}-k_{m}(x_{i,t})\right)
\end{align*}
The corresponding jackknifed statistic is then
\[
T^{-\frac{m}{2}}\tilde{W}_{i,T,m}=T\cdot T^{-\frac{m}{2}}W_{i,T,m}%
-(T-1)\cdot(T-1)^{-\frac{m}{2}}\frac{1}{T}\sum_{t=1}^{T}W_{i,T,m}^{(-t)}%
\]
The following Lemma outlines the impact of the jackknife for $m=1,2,3$.

\begin{lemma}
\label{lem:WiT_jack} Let $\tilde{W}_{i,T,m}=TW_{i,T,m}-(T-1)(\frac{T}%
{T-1})^{m/2}\frac{1}{T}\sum_{t}W_{i,T,m}^{(-t)}$. Then we have:
\begin{align*}
\tilde{W}_{i,T,1}  &  =W_{i,T,1}\\
\tilde{W}_{i,T,2}  &  =\frac{1}{T-1}\sum_{t_{1}\neq t_{2}}k_{1}(x_{i,t_{1}%
})k_{2}(x_{i,t_{2}})\\
\tilde{W}_{i,T,3}  &  =\frac{-1}{T^{1/2}(T-1)}\frac{1}{T}\sum_{t}k_{1}%
(x_{i,t})k_{2}(x_{i,t})k_{3}(x_{i,t})\\
&  \quad+\frac{1}{T^{1/2}(T-1)^{2}}\sum_{t_{1}\neq t_{2}}\left(
\begin{array}
[c]{c}%
k_{1}(x_{i,t_{1}})k_{2}(x_{i,t_{2}})k_{3}(x_{i,t_{2}})\\
+k_{1}(x_{i,t_{2}})k_{2}(x_{i,t_{1}})k_{3}(x_{i,t_{2}})\\
+k_{1}(x_{i,t_{2}})k_{2}(x_{i,t_{2}})k_{3}(x_{i,t_{1}})
\end{array}
\right) \\
&  \quad+\frac{T^{2}-T-2}{T^{1/2}(T-1)^{2}(T-2)}\sum_{t_{1}\neq t_{2}\neq
t_{3}}k_{1}(x_{i,t_{1}})k_{2}(x_{i,t_{2}})k_{3}(x_{i,t_{3}})
\end{align*}

\end{lemma}

The proof of this Lemma is simple algebra and is left for the Supplementary
Appendix II, which also details results for higher-order terms.

For the split-sample correction we write $W_{i,T,m}^{(1)},W_{i,T,m}^{(2)}$ for
the V-statistics using the first and second halves of time periods, so that
the expansion terms take the form
\[
T^{-m/2}\check{W}_{i,T,m}=2\times T^{-m/2}W_{i,T,m}-\frac{1}{2}\times(\frac
{T}{2})^{-m/2}\left(  W_{i,T,m}^{(1)}+W_{i,T,m}^{(2)}\right)
\]

\begin{lemma}
\label{lem:wit_half-1} Let $\check{W}_{i,T,m}=2W_{i,T,m}-2^{m/2}\frac{1}%
{2}(W_{i,T,m}^{(1)}+W_{i,T,m}^{(2)})$, and let $\tau_{1}=\{1,2,\dots,M\}$ and
$\tau_{2}=\{M+1,\dots,T\}$ for $T=2M$. Then:
\begin{align*}
\check{W}_{i,T,1}  &  =W_{i,T,1}\\
\check{W}_{i,T,2}  &  =2\left(  \frac{1}{T}\sum_{t_{1}\in\tau_{1}}\sum
_{t_{2}\in\tau_{2}}k_{1}(x_{it_{1}})k_{2}(x_{it_{2}})+\frac{1}{T}\sum
_{t_{1}\in\tau_{2}}\sum_{t_{2}\in\tau_{1}}k_{1}(x_{it_{1}})k_{2}(x_{it_{2}%
})\right) \\
\check{W}_{i,T,3}  &  =2W_{i,T,3}-4\left(
\begin{array}
[c]{c}%
\frac{1}{T^{3/2}}\sum_{t_{1},t_{2},t_{3}\in\tau_{1}}k_{1}(x_{it_{1}}%
)k_{2}(x_{it_{2}})k_{3}(x_{it_{3}})\\
+\frac{1}{T^{3/2}}\sum_{t_{1},t_{2},t_{3}\in\tau_{2}}k_{1}(x_{it_{1}}%
)k_{2}(x_{it_{2}})k_{3}(x_{it_{3}})
\end{array}
\right)
\end{align*}

\end{lemma}

Again, the proof is left for the Supplementary Appendix IV.

\subsection{Proof of Proposition \ref{thm:variance}}

We first note that from Lemmas \ref{lem:WiT_jack} and \ref{lem:wit_half-1},
neither the jackknife nor the split-sample correction impact the first-order
expansion term so that $A_{n}=\tilde{A}_{n,J}=\tilde{A}_{n,1/2}$ and the
first-order variance $V_{1,n}=\operatorname{Var}(\sqrt{n}A_{n})$ is the same
for both estimators.

Next, we consider the variance of the second order term. There are two forms
of V-statistic contained in $B_{n}$, $W_{T,(2)}$ and $W_{T,(1,1)}$ and we
examine each in turn.

For the first case, we have
\begin{align*}
\operatorname{Var}(\sqrt{n}\tilde{W}_{T,(2)})  &  =\operatorname{Var}\left(
\frac{1}{(T-1)}\frac{1}{\sqrt{n}}\sum_{i}\sum_{t_{1}\neq t_{2}}k_{1}\left(
x_{i,t_{1}}\right)  k_{2}\left(  x_{i,t_{2}}\right)  \right) \\
&  =\frac{1}{n}\sum_{i}\frac{1}{(T-1)^{2}}\sum_{t_{1}\neq t_{2}}\left(
\begin{array}
[c]{c}%
E\left[  k_{1}(x_{i,t_{1}})^{2}\right]  E\left[  k_{2}(x_{i,t_{2}})^{2}\right]
\\
+E\left[  k_{1}(x_{i,t_{1}})k_{2}(x_{i,t_{1}})\right]  E\left[  k_{1}%
(x_{i,t_{2}})k_{2}(x_{i,t_{2}})\right]
\end{array}
\right) \\
&  =\frac{T}{T-1}\frac{1}{n}\sum_{i}\left(
\begin{array}
[c]{c}%
E\left[  k_{1}(x_{i,t})^{2}\right]  E\left[  k_{2}(x_{i,t})^{2}\right] \\
+E\left[  k_{1}(x_{i,t})k_{2}(x_{i,t})\right]  E\left[  k_{1}(x_{i,t}%
)k_{2}(x_{i,t})\right]
\end{array}
\right)
\end{align*}

\begin{align*}
&  \operatorname{Var}(\sqrt{n}\check{W}_{T,(2)})\\
&  =\operatorname{Var}\left(  \frac{1}{\sqrt{n}}\sum_{i}\frac{2}{T}\sum
_{t_{1}\in\tau_{1}}\sum_{t_{2}\in\tau_{2}}\left(  k_{1}(x_{it_{1}}%
)k_{2}(x_{it_{2}})+k_{1}(x_{it_{2}})k_{2}(x_{it_{1}})\right)  \right) \\
&  =\frac{1}{n}\sum_{i}\frac{4}{T^{2}}\sum_{t_{1}\in\tau_{1}}\sum_{t_{2}%
\in\tau_{2}}2\left(
\begin{array}
[c]{c}%
E[k_{1}(x_{it_{1}})^{2}]E[k_{2}(x_{it_{2}})^{2}]\\
+E[k_{1}(x_{it_{1}})k_{2}(x_{it_{1}})]E[k_{1}(x_{it_{2}})k_{2}(x_{it_{2}})]
\end{array}
\right) \\
&  =\frac{2}{n}\sum_{i}\left(  E[k_{1}(x_{i,t})^{2}]E[k_{2}(x_{i,t}%
)^{2}]+E[k_{1}(x_{i,t})k_{2}(x_{i,t})]^{2}\right) \\
&  =2\frac{T-1}{T}\operatorname{Var}(\sqrt{n}\tilde{W}_{T,(2)})
\end{align*}
For the second case we have
\begin{align*}
\operatorname{Var}(\sqrt{n}\tilde{W}_{T,(1,1)})  &  =\operatorname{Var}\left(
\frac{1}{n^{3/2}}\sum_{i_{1},i_{2}}\frac{1}{T-1}\sum_{t_{1},t_{2}}%
k_{1}(x_{i_{1},t_{1}})k_{2}(x_{i_{2},t_{2}})\right) \\
&  =\frac{T-1}{T}\frac{1}{n^{3}}\sum_{i_{1},i_{2}}\left(
\begin{array}
[c]{c}%
E\left[  k_{1}(x_{i_{1},t})^{2}\right]  E\left[  k_{2}(x_{i_{2},t})^{2}\right]
\\
+E\left[  k_{1}(x_{i_{1},t})k_{2}(x_{i_{1},t})\right]  E\left[  k_{1}%
(x_{i_{2},t})k_{2}(x_{i_{2},t})\right]
\end{array}
\right) \\
&  =O(n^{-1})\\
\operatorname{Var}(\sqrt{n}\check{W}_{T,(1,1)})  &  =2\frac{T-1}%
{T}\operatorname{Var}(\sqrt{n}\tilde{W}_{T,(1,1)})=O(n^{-1})
\end{align*}
Finally, similar calculations give $Cov(\sqrt{n}\tilde{W}_{T,(2)},\sqrt
{n}\tilde{W}_{T,(1,1)})=O(n^{-1})$ as well as \\ $Cov(\sqrt{n}\check{W}_{T,(2)},\sqrt
{n}\check{W}_{T,(1,1)})=O(n^{-1})$.

Next, we establish that $Cov(A_{n},\tilde{B}_{n,J})=Cov(A_{n},\tilde
{B}_{n,1/2})=0$. This result is immediate since the bias-corrected second
order terms are sums over $k_{1}(x_{i,t_{1}})k_{2}(x_{i,t_{2}})$, for which
$t_{1}\ne t_{2}$. The product of these second-order terms with the first order
terms in $A_{n}$ will then contain triplets $k_{1}(x_{i,t_{1}})k_{2}%
(x_{i,t_{2}})k_{3}(x_{i,t_{3}})$ which must have mean zero by $t_{1}\ne t_{2}$.

The final part of the higher-order variance comes from $Cov(\sqrt{n}%
A_{n},\sqrt{n}\tilde{C}_{n,J})$ and \\ $Cov(\sqrt{n}A_{n},\sqrt{n}\tilde
{C}_{n,1/2})$. Here we show that these terms are both $o(1)$. Recall from
Lemma \ref{lem:WiT_jack} the third-order jackknife term has the form
\begin{align*}
\tilde{W}_{i,T,3}  &  =\frac{-1}{T^{1/2}(T-1)}\frac{1}{T}\sum_{t}k_{1}%
(x_{i,t})k_{2}(x_{i,t})k_{3}(x_{i,t})\\
&  \quad+\frac{1}{T^{1/2}(T-1)^{2}}\sum_{t_{1}\neq t_{2}}\left(
\begin{array}
[c]{c}%
k_{1}(x_{i,t_{1}})k_{2}(x_{i,t_{2}})k_{3}(x_{i,t_{2}})+k_{1}(x_{i,t_{2}}%
)k_{2}(x_{i,t_{1}})k_{3}(x_{i,t_{2}})\\
+k_{1}(x_{i,t_{2}})k_{2}(x_{i,t_{2}})k_{3}(x_{i,t_{!}})
\end{array}
\right) \\
&  \quad+\frac{T^{2}-T-2}{T^{1/2}(T-1)^{2}(T-2)}\sum_{t_{1}\neq t_{2}\neq
t_{3}}k_{1}(x_{i,t_{1}})k_{2}(x_{i,t_{2}})k_{3}(x_{i,t_{3}})\\
&  =J_{i,1}+J_{i,2}+J_{i,3}%
\end{align*}
First consider $\tilde{W}_{T,(3)}=\frac{1}{n}\sum_{i}(J_{i,1}+J_{i,2}%
+J_{i,3})$ and its covariance with a first order statistic $\tilde{W}_{T,(1)}%
$.
\begin{align*}
&  Cov(\sqrt{n}\tilde{W}_{T,(1)},\sqrt{n}J_{1})\\
&  =\frac{-1}{T^{1/2}(T-1)}\frac{1}{n}\sum_{i_{1},i_{2}}\frac{1}{T^{3/2}}%
\sum_{t_{1},t_{2}}E\left[  k_{1}(x_{i_{1},t})k_{2}(x_{i_{2},t})k_{3}%
(x_{i_{3},t})k_{4}(x_{i_{4},t})\right] \\
&  =-\frac{1}{T-1}\frac{1}{n}\sum_{i}E\left[  k_{1}(x_{i,t})k_{2}%
(x_{i,t})k_{3}(x_{i,t})k_{4}(x_{i,t})\right] \\
&  =O(T^{-1})
\end{align*}

\begin{align*}
&  Cov(\sqrt{n}\tilde{W}_{T,(1)},\sqrt{n}J_{2})\\
&  =\frac{1}{T(T-1)^{2}}\frac{1}{n}\sum_{i_{1},i_{2}}\sum_{t}\sum_{t_{2}\neq
t_{3}}E\left[  k_{1}(x_{i_{1},t_{1}})\left(
\begin{array}
[c]{c}%
k_{2}(x_{i_{2},t_{3}})k_{3}(x_{i_{2},t_{2}})k_{4}(x_{i_{2},t_{2}})\\
+k_{1}(x_{i_{2},t_{2}})k_{2}(x_{i_{2},t_{3}})k_{3}(x_{i_{2},t_{2}})\\
+k_{1}(x_{i_{2},t_{2}})k_{2}(x_{i_{2},t_{2}})k_{3}(x_{i_{2},t_{3}})
\end{array}
\right)  \right] \\
&  =\frac{1}{n}\sum_{i}\frac{1}{T(T-1)^{2}}\sum_{t_{1}\neq t_{2}}\left(
\begin{array}
[c]{c}%
E\left[  k_{1}(x_{i,t_{1}})k_{2}(x_{i,t_{1}})\right]  E\left[  k_{3}%
(x_{i,t_{2}})k_{4}(x_{i,t_{2}})\right] \\
+E\left[  k_{1}(x_{i,t_{1}})k_{3}(x_{i,t_{1}})\right]  E\left[  k_{2}%
(x_{i,t_{2}})k_{4}(x_{i,t_{2}})\right] \\
+E\left[  k_{1}(x_{i,t_{1}})k_{4}(x_{i,t_{1}})\right]  E\left[  k_{2}%
(x_{i,t_{2}})k_{3}(x_{i,t_{2}})\right]
\end{array}
\right) \\
&  =O(T^{-1})
\end{align*}
Since $J_{3}$ contains terms with $t_{1}\neq t_{2}\neq t_{3}$, the covariance
involving these terms is zero. Identical steps applied to $\sqrt{n}\tilde
{W}_{T,(2,1)}$ and $\sqrt{n}\tilde{W}_{T,(1,1,1)}$ will confirm that the
covariances of these terms with $\sqrt{n}\tilde{W}_{T,(1)}$ are also
$O(T^{-1})$ or lower order, so that $Cov(\sqrt{n}A_{n},\sqrt{n}\tilde{C}%
_{n,J})=o(1)$, as required.

For the split-sample case, we have
\begin{align*}
\check{W}_{i,T,3}  &  =2W_{i,T,3}-4\left(
\begin{array}
[c]{c}%
\frac{1}{T^{3/2}}\sum_{t_{1},t_{2},t_{3}\in\tau_{1}}k_{1}(x_{it_{1}}%
)k_{2}(x_{it_{2}})k_{3}(x_{it_{3}})\\
+\frac{1}{T^{3/2}}\sum_{t_{1},t_{2},t_{3}\in\tau_{2}}k_{1}(x_{it_{1}}%
)k_{2}(x_{it_{2}})k_{3}(x_{it_{3}})
\end{array}
\right) \\
&  =2W_{i,T,3}-4S_{1}-4S_{2}%
\end{align*}
The covariance with the first term is
\begin{align*}
&  Cov(\sqrt{n}\check{W}_{T,(1)},\sqrt{n}W_{T,(3)})\\
&  =\frac{1}{n}\sum_{i}\frac{1}{T^{2}}\sum_{t_{1},t_{2},t_{3},t_{4}}E\left[
k_{1}(x_{i,t_{1}})k_{2}(x_{i,t_{2}})k_{3}(x_{i,t_{3}})k_{4}(x_{i,t_{4}%
})\right] \\
&  =\frac{1}{n}\sum_{i}\frac{1}{T^{2}}\left\{
\begin{array}
[c]{c}%
\sum_{t}E\left[  k_{1}(x_{i,t})k_{2}(x_{i,t})k_{3}(x_{i,t})k_{4}%
(x_{i,t})\right] \\
+\sum_{t_{1}\neq t_{2}}\left(
\begin{array}
[c]{c}%
E\left[  k_{1}(x_{i,t_{1}})k_{2}(x_{i,t_{1}})\right]  E\left[  k_{3}%
(x_{i,t_{2}})k_{4}(x_{i,t_{2}})\right] \\
+E\left[  k_{1}(x_{i,t_{1}})k_{3}(x_{i,t_{1}})\right]  E\left[  k_{2}%
(x_{i,t_{2}})k_{4}(x_{i,t_{2}})\right] \\
+E\left[  k_{1}(x_{i,t_{1}})k_{4}(x_{i,t_{1}})\right]  E\left[  k_{2}%
(x_{i,t_{2}})k_{3}(x_{i,t_{2}})\right]
\end{array}
\right)
\end{array}
\right\} \\
&  =\frac{1}{T}\frac{1}{n}\sum_{i}E\left[  k_{1}(x_{i,t})k_{2}(x_{i,t}%
)k_{3}(x_{i,t})k_{4}(x_{i,t})\right] \\
&  \quad+\frac{T-1}{T}\frac{1}{n}\sum_{i}\left(
\begin{array}
[c]{c}%
E\left[  k_{1}(x_{i,t})k_{2}(x_{i,t})\right]  E\left[  k_{3}(x_{i,t}%
)k_{4}(x_{i,t})\right] \\
+E\left[  k_{1}(x_{i,t})k_{3}(x_{i,t})\right]  E\left[  k_{2}(x_{i,t}%
)k_{4}(x_{i,t})\right] \\
+E\left[  k_{1}(x_{i,t})k_{4}(x_{i,t})\right]  E\left[  k_{2}(x_{i,t}%
)k_{3}(x_{i,t})\right]
\end{array}
\right)
\end{align*}
And for $S_{1}$
\begin{align*}
&  Cov(\sqrt{n}\check{W}_{T,(1)},\sqrt{n}S_{1})\\
&  =\frac{1}{n}\sum_{i}\frac{1}{T^{2}}\sum_{t_{1},t_{2},t_{3},t_{4}\in\tau
_{1}}E\left[  k_{1}(x_{it_{1}})k_{2}(x_{it_{2}})k_{3}(x_{it_{3}}%
)k_{4}(x_{it_{4}})\right] \\
&  =\frac{1}{2}\frac{1}{T}\frac{1}{n}\sum_{i}E\left[  k_{1}(x_{i,t}%
)k_{2}(x_{i,t})k_{3}(x_{i,t})k_{4}(x_{i,t})\right] \\
&  \quad+\frac{1}{4}\frac{T-2}{T}\frac{1}{n}\sum_{i}\left(
\begin{array}
[c]{c}%
E\left[  k_{1}(x_{i,t})k_{2}(x_{i,t})\right]  E\left[  k_{3}(x_{i,t}%
)k_{4}(x_{i,t})\right] \\
+E\left[  k_{1}(x_{i,t})k_{3}(x_{i,t})\right]  E\left[  k_{2}(x_{i,t}%
)k_{4}(x_{i,t})\right] \\
+E\left[  k_{1}(x_{i,t})k_{4}(x_{i,t})\right]  E\left[  k_{2}(x_{i,t}%
)k_{3}(x_{i,t})\right]
\end{array}
\right)
\end{align*}
and similarly for $S_{2}$. Combining, gives
\begin{align*}
&  Cov(\sqrt{n}\check{W}_{T,(1)},\sqrt{n}\check{W}_{T,(3)})\\
&  =-2\frac{1}{T}\frac{1}{n}\sum_{i}E\left[  k_{1}(x_{i,t})k_{2}(x_{i,t}%
)k_{3}(x_{i,t})k_{4}(x_{i,t})\right] \\
&  \quad+2\frac{1}{T}\frac{1}{n}\sum_{i}\left(
\begin{array}
[c]{c}%
E\left[  k_{1}(x_{i,t})k_{2}(x_{i,t})\right]  E\left[  k_{3}(x_{i,t}%
)k_{4}(x_{i,t})\right] \\
+E\left[  k_{1}(x_{i,t})k_{3}(x_{i,t})\right]  E\left[  k_{2}(x_{i,t}%
)k_{4}(x_{i,t})\right] \\
+E\left[  k_{1}(x_{i,t})k_{4}(x_{i,t})\right]  E\left[  k_{2}(x_{i,t}%
)k_{3}(x_{i,t})\right]
\end{array}
\right) \\
&  =O(T^{-1})
\end{align*}
Similar statements can be made for the terms $\check{W}_{T,(2,1)}$ and
$\check{W}_{T,(1,1,1)}$, giving $Cov(\sqrt{n}A_{n},\sqrt{n}\tilde{C}%
_{n,1/2})=o(1)$ as before. We can now derive the expression for $V_{2,n}$ in
the theorem. From Hahn and Newey (2004), we have
\begin{align*}
B_{n}  &  =\frac{1}{2}\theta^{\epsilon\epsilon}(0)=\frac{1}{n}\sum_{i=1}%
^{n}\frac{1}{T}\sum_{s,t}k_{1}(x_{i,s})k_{2}(x_{i,t})+o_{p}(1)\\
k_{1}(x_{i,t})  &  =\mathcal{I}_{n}^{-1}\left(  \frac{E[U_{it}^{\alpha\alpha
}]}{2E[V_{it}^{2}]}V_{it}+U_{it}^{\alpha}\right) \\
k_{2}(x_{i,t})  &  =\frac{1}{E[V_{it}^{2}]}V_{it}%
\end{align*}
We can then apply the form of $\operatorname{Var}(\tilde{W}_{T,(2)})$ and
$\operatorname{Var}(\check{W}_{T,(2)})$ above to give the result.

\subsection{Proof of Proposition \ref{Thm:bias_est}\label{sec:Thm_bias_proof}}

We begin by establishing an expansion for the bias estimators of the form
\begin{align*}
\sqrt{nT}\frac{1}{T}(\hat{\beta}_{J}-\mathbf{B})  &  =\frac{\sqrt{n}}{\sqrt
{T}}(B_{n}-\tilde{B}_{n,J}-\mathbf{B})+O_{p}(T^{-1})\\
\sqrt{nT}\frac{1}{T}(\hat{\beta}_{1/2}-\mathbf{B})  &  =\frac{\sqrt{n}}%
{\sqrt{T}}(B_{n}-\tilde{B}_{n,1/2}-\mathbf{B})+O_{p}(T^{-1/2})
\end{align*}
Since the implicit bias estimates are given by $\frac{1}{T}\hat{\beta}%
=\hat{\theta}-\tilde{\theta}$, the expansion in (\ref{eq:expansionA}), and the
equivalent expansions for the bias-corrected estimators, imply the result as
long as $\frac{\sqrt{n}}{T}(C_{n}-\tilde{C}_{n,J})=O_{p}(T^{-1})$ and
$\frac{\sqrt{n}}{T}(C_{n}-\tilde{C}_{n,1/2})=O_{p}(T^{-1/2})$. This follows
straightforwardly from Lemmas \ref{lem:WiT_jack} and \ref{lem:wit_half-1}, and
is shown in the Supplementary Appendix IV.

Given this expansion for the bias estimates, we analyze the V-statistic forms
that are present in $\frac{\sqrt{n}}{\sqrt{T}}(B_{n}-\tilde{B}_{n,J}%
-\mathbf{B})$
\begin{align*}
&  \frac{\sqrt{n}}{\sqrt{T}}\left(  W_{T,(2)}-\tilde{W}_{T,(2)}-E[W_{T,(2)}%
]\right) \\
&  =\frac{1}{\sqrt{n}}\sum_{i}\frac{1}{T\sqrt{T}}\sum_{t_{1},t_{2}}\left(
k_{1}(x_{i,t_{1}})k_{2}(x_{i,t_{2}})-E[k_{1}(x_{i,t_{1}})k_{2}(x_{i,t_{2}%
})]\right) \\
&  \quad-\frac{1}{\sqrt{n}}\sum_{i}\frac{1}{(T-1)\sqrt{T}}\sum_{t_{1}\neq
t_{2}}k_{1}(x_{i,t_{1}})k_{2}(x_{i,t_{2}})\\
&  =\frac{1}{T}\frac{1}{\sqrt{nT}}\sum_{i}\sum_{t}\left(  k_{1}(x_{i,t}%
)k_{2}(x_{i,t})-E[k_{1}(x_{i,t})k_{2}(x_{i,t})]\right) \\
&  \quad-\frac{1}{\sqrt{T}(T-1)}\frac{1}{\sqrt{n}}\sum_{i}\frac{1}{T}%
\sum_{t_{1}\neq t_{2}}k_{1}(x_{i,t_{1}})k_{2}(x_{i,t_{2}})\\
&  =O_{p}(T^{-1})
\end{align*}
And similarly, we can show $\frac{\sqrt{n}}{\sqrt{T}}\left(  W_{T,(1,1)}%
-\tilde{W}_{T,(1,1)}-E[W_{T,(1,1)}]\right)  =O_{p}(T^{-1})$. Consequently,
$\frac{\sqrt{n}}{\sqrt{T}}(B_{n}-\tilde{B}_{n,J}-\mathbf{B})=O_{p}(T^{-1})$.

For the split-sample version of the statistics, let $1_{t_{1},t_{2}}$ be an
indicator that is equal to one whenever $t_{1}$ and $t_{2}$ are in the same
half of time periods, and zero otherwise, and $\tilde{1}_{t_{1},t_{2}%
}=(-1)^{(1-1_{t_{1},t_{2}})}$. Then we have
\begin{align*}
&  \frac{\sqrt{n}}{\sqrt{T}}\left(  W_{T,(2)}-\check{W}_{T,(2)}-E[W_{T,(2)}%
]\right) \\
&  =\frac{1}{\sqrt{n}}\sum_{i}\frac{1}{T\sqrt{T}}\sum_{t_{1},t_{2}}\left(
k_{1}(x_{i,t_{1}})k_{2}(x_{i,t_{2}})-E[k_{1}(x_{i,t_{1}})k_{2}(x_{i,t_{2}%
})]\right) \\
&  \quad-\frac{1}{\sqrt{n}}\sum_{i}\frac{2}{T\sqrt{T}}\sum_{t_{1}\in\tau_{1}%
}\sum_{t_{2}\in\tau_{2}}\left(  k_{1}(x_{it_{1}})k_{2}(x_{it_{2}}%
)+k_{1}(x_{it_{2}})k_{2}(x_{it_{1}})\right) \\
&  =\frac{1}{T\sqrt{nT}}\sum_{i}\sum_{t_{1},t_{2}}\tilde{1}_{t_{1},t_{2}%
}\left(  k_{1}(x_{i,t_{1}})k_{2}(x_{i,t_{2}})-E[k_{1}(x_{i,t_{1}}%
)k_{2}(x_{i,t_{2}})]\right)
\end{align*}
Consider the variance of the final term
\begin{align*}
&  \frac{1}{nT^{3}}\operatorname{Var}\left(  \sum_{i}\sum_{t_{1},t_{2}}%
\tilde{1}_{t_{1},t_{2}}\left(  k_{1}(x_{i,t_{1}})k_{2}(x_{i,t_{2}}%
)-E[k_{1}(x_{i,t_{1}})k_{2}(x_{i,t_{2}})]\right)  \right) \\
&  =\frac{1}{nT^{3}}\sum_{i_{1},i_{2}}\sum_{t_{1},t_{2},t_{3},t_{4}}E\left[
\begin{array}
[c]{c}%
\left(  k_{1}(x_{i_{1},t_{1}})k_{2}(x_{i_{1},t_{2}})-E[k_{1}(x_{i_{1},t_{1}%
})k_{2}(x_{i_{1},t_{2}})]\right) \\
\times\left(  k_{1}(x_{i_{2},t_{3}})k_{2}(x_{i_{2},t_{4}})-E[k_{1}%
(x_{i_{2},t_{3}})k_{2}(x_{i_{2},t_{4}})]\right)
\end{array}
\right]  \tilde{1}_{t_{1},t_{2}}\tilde{1}_{t_{3},t_{4}}\\
&  =\frac{1}{nT^{3}}\sum_{i}\sum_{t}E\left[  \left(  k_{1}(x_{i,t}%
)k_{2}(x_{i,t})-E[k_{1}(x_{i,t})k_{2}(x_{i,t})]\right)  ^{2}\right] \\
&  \quad+\frac{1}{nT^{3}}\sum_{i}\sum_{t_{1}\neq t_{2}}\left(
\begin{array}
[c]{c}%
E[k_{1}(x_{i,t_{1}})^{2}]E[k_{2}(x_{i,t_{2}})^{2}]\\
+E[k_{1}(x_{i,t_{1}})k_{2}(x_{i,t_{1}})]E[k_{1}(x_{i,t_{2}})k_{2}(x_{i,t_{2}%
})]
\end{array}
\right) \\
&  =O(T^{-1})
\end{align*}
Consequently, $\frac{\sqrt{n}}{\sqrt{T}}\left(  W_{T,(2)}-\check{W}%
_{T,(2)}-E[W_{T,(2)}]\right)  =O_{p}(T^{-1/2})$ as required. Again, the same
steps can also be used to show $\frac{\sqrt{n}}{\sqrt{T}}\left(
W_{T,(1,1)}-\check{W}_{T,(1,1)}-E[W_{T,(1,1)}]\right)  =O_{p}(T^{-1/2})$ , and
hence that $\frac{\sqrt{n}}{\sqrt{T}}(B_{n}-\tilde{B}_{n,1/2}-\mathbf{B}%
)=O_{p}(T^{-1/2})$, giving the result of the theorem.

\newpage

\begin{center}
{\LARGE Supplementary Appendix II: Efficient Bias Correction for Cross-section and Panel Data}
\end{center}
\setcounter{section}{0}

\section{Proof of Proposition \ref{ApproxThetaHat-strong}}

It is convenient to understand $\widehat{\theta}\equiv\widehat{\theta}\left(
\frac{1}{\sqrt{n}}\right)  $, where $\theta\left(  \epsilon\right)  $ denotes
the solution
\[
\hat{\theta}(\epsilon)=\sup_{\theta\in\Theta}\int\log f\left(  \cdot
,\theta\right)  dF_{\epsilon}\left(  z\right)  .
\]
Here, $F_{\epsilon}\equiv F+\epsilon\Delta\equiv F+\epsilon\sqrt{n}\left(
\widehat{F}-F\right)  $, $\epsilon\in\left[  0,\frac{1}{\sqrt{n}}\right]  $,
and $F$ and $\widehat{F}$ denote the underlying cumulative distribution
function and the empirical distribution function $\widehat{F}\left(  z\right)
\equiv n^{-1}\sum_{i=1}^{n}\mathbf{1}\left\{  Z_{i}\leq z\right\}  $.
Proposition \ref{ApproxThetaHat-strong} will be established later
using this device. The bootstrap counterpart will be established with a
similar device. Using the above notation, it follows that $\widehat{\theta
}^{\ast}\equiv\widehat{\theta}^{\ast}\left(  \frac{1}{\sqrt{n}}\right)  $ is
the solution
\[
\hat{\theta}^{\ast}\left(  \epsilon\right)  =\sup_{\theta\in\Theta}\int\log
f\left(  \cdot,\theta\right)  d\widehat{F}_{\epsilon}\left(  z\right)  ,
\]
where
\[
\widehat{F}_{\epsilon}\equiv\widehat{F}+\epsilon\widehat{\Delta}%
\equiv\widehat{F}+\epsilon\sqrt{n}\left(  \widehat{F}^{\ast}-\widehat{F}%
\right)  ,\qquad\epsilon\in\left[  0,\frac{1}{\sqrt{n}}\right]  .
\]
Here, $\widehat{\Delta}$ is the bootstrap empirical process $\widehat{\Delta
}\equiv\sqrt{n}\left(  \widehat{F}^{\ast}-\widehat{F}\right)  $.

Let $\hat{Q}(\theta)=\int\log f\left(  \cdot,\theta\right)  d\hat{F}\left(
z\right)  ,Q_{\epsilon}(\theta)=\int\log f\left(  \cdot,\theta\right)
dF_{\epsilon}\left(  z\right)  $ and $Q(\theta)=\int\log f\left(  \cdot
,\theta\right)  dQ\left(  z\right)  $ such that $Q_{\epsilon}(\theta
)-Q(\theta)=\epsilon\sqrt{n}\left(  \hat{Q}(\theta)-Q(\theta)\right)  .$ By
Conditions \ref{HH}, \ref{M} and \ref{EC} and van der Vaart and Wellner (1996,
Theorem 2.4.3) it follows that $\sup_{\theta}\left\vert Q_{\epsilon}%
(\theta)-Q(\theta)\right\vert \leq\sup_{\theta}\left\vert \hat{Q}%
(\theta)-Q(\theta)\right\vert \rightarrow0$ in probability. By van der Vaart
and Wellner (1996, Corollary 3.2.3), it follows that uniformly in $\epsilon
\in\left[  -n^{-1/2},n^{-1/2}\right]  ,$ $\hat{\theta}\left(  \epsilon\right)
\overset{p}{\rightarrow}\theta_{0}$. This implies that for any compact set
$K\subset\Theta$ with $\theta_{0}\in K,$ $P(\hat{\theta}\left(  \epsilon
\right)  \in K)\rightarrow1$, as $n\rightarrow\infty.$ Consider the function
$G\left(  \epsilon,\theta\right)  \equiv\int\ell\left(  \cdot,\theta\right)
dF_{\epsilon}\left(  z\right)  .$ If $\partial\Theta$ is the boundary of
$\Theta$ then $P(G\left(  \epsilon,\hat{\theta}\left(  \epsilon\right)
\right)  \neq0)\leq P\left(  \hat{\theta}\left(  \epsilon\right)  \in
\partial\Theta\right)  \leq1-P\left(  \theta\left(  \epsilon\right)  \in
K\right)  \rightarrow0$. We now condition on the event $\left\{  G\left(
\epsilon,\hat{\theta}\left(  \epsilon\right)  \right)  =0\right\}  $.

By Taylor's theorem there exists some$\ \tilde{\epsilon}\in\left[
0,1/\sqrt{n}\right]  $ such that
\[
\widehat{\theta}\left(  n^{-1/2}\right)  =\theta(0)+\sum_{k=1}^{m-1}\frac
{1}{k!n^{k/2}}\theta^{\left(  k\right)  }\left(  0\right)  +\frac{1}%
{m!n^{m/2}}\theta^{\left(  m\right)  }\left(  \tilde{\epsilon}\right)  .
\]
By Lemmas \ref{Lem-consistency-of-theta-epsilon} and \ref{multi:k} (in
Appendix \ref{sec-lfbs}), it follows that $\max_{0\leq\epsilon\leq n^{-1/2}%
}\theta^{\left(  k\right)  }\left(  \epsilon\right)  =O_{p}(1)$ such that the
remainder term $\frac{1}{m!n^{m/2}}\theta^{\left(  m\right)  }\left(
\tilde{\epsilon}\right)  =O_{p}\left(  n^{-m/2}\right)  $ for $m\leq6$. To
find the derivatives $\theta^{\left(  k\right)  }$, let
\[
h\left(  z,\epsilon\right)  \equiv\ell\left(  z,\theta\left(  \epsilon\right)
\right)  ,
\]
and rewrite the first order condition as
\[
0=\int h\left(  z,\epsilon\right)  dF_{\epsilon}\left(  z\right)
\]
Differentiating repeatedly with respect to $\epsilon$, we obtain
\begin{align}
0  &  =\int\frac{dh\left(  z,\epsilon\right)  }{d\epsilon}dF_{\epsilon}\left(
z\right)  +\int h\left(  z,\epsilon\right)  d\Delta\left(  z\right)
\label{h1}\\
0  &  =\int\frac{d^{2}h\left(  z,\epsilon\right)  }{d\epsilon^{2}}%
dF_{\epsilon}\left(  z\right)  +2\int\frac{dh\left(  z,\epsilon\right)
}{d\epsilon}d\Delta\left(  z\right) \label{h2}\\
0  &  =\int\frac{d^{3}h\left(  z,\epsilon\right)  }{d\epsilon^{3}}%
dF_{\epsilon}\left(  z\right)  +3\int\frac{d^{2}h\left(  z,\epsilon\right)
}{d\epsilon^{2}}d\Delta\left(  z\right) \label{h3}\\
0  &  =\int\frac{d^{4}h\left(  z,\epsilon\right)  }{d\epsilon^{4}}%
dF_{\epsilon}\left(  z\right)  +4\int\frac{d^{3}h\left(  z,\epsilon\right)
}{d\epsilon^{3}}d\Delta\left(  z\right) \label{h4}\\
0  &  =\int\frac{d^{5}h\left(  z,\epsilon\right)  }{d\epsilon^{5}}%
dF_{\epsilon}\left(  z\right)  +5\int\frac{d^{4}h\left(  z,\epsilon\right)
}{d\epsilon^{4}}d\Delta\left(  z\right) \label{h5}\\
0  &  =\int\frac{d^{6}h\left(  z,\epsilon\right)  }{d\epsilon^{6}}%
dF_{\epsilon}\left(  z\right)  +6\int\frac{d^{5}h\left(  z,\epsilon\right)
}{d\epsilon^{5}}d\Delta\left(  z\right)  \label{h6}%
\end{align}
Note that
\begin{equation}
\frac{dh\left(  \epsilon\right)  }{d\epsilon}=\ell^{\theta}\theta^{\epsilon}
\label{dh1}%
\end{equation}%
\begin{equation}
\frac{d^{2}h\left(  \epsilon\right)  }{d\epsilon^{2}}=\ell^{\theta\theta
}\left(  \theta^{\epsilon}\right)  ^{2}+\ell^{\theta}\theta^{\epsilon\epsilon}
\label{dh2}%
\end{equation}%
\begin{equation}
\frac{d^{3}h\left(  \epsilon\right)  }{d\epsilon^{3}}=\ell^{\theta\theta
\theta}\left(  \theta^{\epsilon}\right)  ^{3}+3\ell^{\theta\theta}%
\theta^{\epsilon}\theta^{\epsilon\epsilon}+\ell^{\theta}\theta^{\epsilon
\epsilon\epsilon} \label{dh3}%
\end{equation}%
\begin{equation}
\frac{d^{4}h\left(  \epsilon\right)  }{d\epsilon^{4}}=\ell^{\theta\theta
\theta\theta}\left(  \theta^{\epsilon}\right)  ^{4}+\allowbreak6\ell
^{\theta\theta\theta}\left(  \theta^{\epsilon}\right)  ^{2}\theta
^{\epsilon\epsilon}+3\ell^{\theta\theta}\left(  \theta^{\epsilon\epsilon
}\right)  ^{2}+4\ell^{\theta\theta}\theta^{\epsilon}\theta^{\epsilon
\epsilon\epsilon}+\ell^{\theta}\theta^{\epsilon\epsilon\epsilon\epsilon}
\label{dh4}%
\end{equation}%
\begin{align}
\frac{d^{5}h\left(  \epsilon\right)  }{d\epsilon^{5}}  &  =\allowbreak
\ell^{\theta\theta\theta\theta\theta}\left(  \theta^{\epsilon}\right)
^{5}+\allowbreak10\ell^{\theta\theta\theta\theta}\left(  \theta^{\epsilon
}\right)  ^{3}\allowbreak\theta^{\epsilon\epsilon}+15\ell^{\theta\theta\theta
}\allowbreak\theta^{\epsilon}\left(  \theta^{\epsilon\epsilon}\right)
^{2}\label{dh5}\\
&  +10\ell^{\theta\theta\theta}\left(  \theta^{\epsilon}\right)  ^{2}%
\theta^{\epsilon\epsilon\epsilon}+10\ell^{\theta\theta}\theta^{\epsilon
\epsilon}\theta^{\epsilon\epsilon\epsilon}+5\ell^{\theta\theta}\theta
^{\epsilon}\theta^{\epsilon\epsilon\epsilon\epsilon}+\ell^{\theta}%
\theta^{\epsilon\epsilon\epsilon\epsilon\epsilon}\nonumber
\end{align}%
\begin{align}
\frac{d^{6}h\left(  \epsilon\right)  }{d\epsilon^{6}}  &  =\allowbreak
\ell^{\theta\theta\theta\theta\theta\theta}\left(  \theta^{\epsilon}\right)
^{6}+\allowbreak15\ell^{\theta\theta\theta\theta\theta}\left(  \theta
^{\epsilon}\right)  ^{4}\theta^{\epsilon\epsilon}+45\ell^{\theta\theta
\theta\theta}\allowbreak\left(  \theta^{\epsilon}\right)  ^{2}\left(
\theta^{\epsilon\epsilon}\right)  ^{2}\label{dh6}\\
&  +20\allowbreak\ell^{\theta\theta\theta\theta}\left(  \theta^{\epsilon
}\right)  ^{3}\theta^{\epsilon\epsilon\epsilon}+15\ell^{\theta\theta\theta
}\left(  \theta^{\epsilon\epsilon}\right)  ^{3}+60\ell^{\theta\theta\theta
}\allowbreak\theta^{\epsilon}\theta^{\epsilon\epsilon}\theta^{\epsilon
\epsilon\epsilon}\nonumber\\
&  +\allowbreak15\ell^{\theta\theta\theta}\left(  \theta^{\epsilon}\right)
^{2}\theta^{\epsilon\epsilon\epsilon\epsilon}+10\ell^{\theta\theta}\left(
\theta^{\epsilon\epsilon\epsilon}\right)  ^{2}+15\ell^{\theta\theta}%
\theta^{\epsilon\epsilon}\theta^{\epsilon\epsilon\epsilon\epsilon}%
+6\ell^{\theta\theta}\theta^{\epsilon}\theta^{\epsilon\epsilon\epsilon
\epsilon\epsilon}\nonumber\\
&  +\ell^{\theta}\theta^{\epsilon\epsilon\epsilon\epsilon\epsilon\epsilon
}\nonumber
\end{align}
Here, $\theta^{\epsilon}$ denotes the derivative of $\theta$ with respect to
$\epsilon$. Combining (\ref{h1}) - (\ref{h4}) with (\ref{dh1}) - (\ref{dh4}),
we obtain
\begin{equation}
0=E_{\epsilon}\left[  \ell^{\theta}\left(  Z_{i},\epsilon\right)  \right]
\theta^{\epsilon}\left(  \epsilon\right)  +\int\ell\left(  z,\epsilon\right)
d\Delta\left(  z\right)  \label{multi:alt-h1}%
\end{equation}%
\begin{equation}
0=E_{\epsilon}\left[  \ell^{\theta\theta}\left(  Z_{i},\epsilon\right)
\right]  \left(  \theta^{\epsilon}\left(  \epsilon\right)  \right)
^{2}+E_{\epsilon}\left[  \ell^{\theta}\left(  Z_{i},\epsilon\right)  \right]
\theta^{\epsilon\epsilon}\left(  \epsilon\right)  +2\left(  \int\ell^{\theta
}\left(  z,\epsilon\right)  d\Delta\left(  z\right)  \right)  \theta
^{\epsilon}\left(  \epsilon\right)  \label{multi:alt-h2}%
\end{equation}%
\begin{align}
0  &  =E_{\epsilon}\left[  \ell^{\theta\theta\theta}\left(  Z_{i}%
,\epsilon\right)  \right]  \left(  \theta^{\epsilon}\left(  \epsilon\right)
\right)  ^{3}+3E_{\epsilon}\left[  \ell^{\theta\theta}\left(  Z_{i}%
,\epsilon\right)  \right]  \theta^{\epsilon}\left(  \epsilon\right)
\theta^{\epsilon\epsilon}\left(  \epsilon\right)  +E_{\epsilon}\left[
\ell^{\theta}\left(  Z_{i},\epsilon\right)  \right]  \theta^{\epsilon
\epsilon\epsilon}\left(  \epsilon\right) \nonumber\\
&  +3\left(  \int\ell^{\theta\theta}\left(  z,\epsilon\right)  d\Delta\left(
z\right)  \right)  \left(  \theta^{\epsilon}\left(  \epsilon\right)  \right)
^{2}+3\left(  \int\ell^{\theta}\left(  z,\epsilon\right)  d\Delta\left(
z\right)  \right)  \theta^{\epsilon\epsilon}\left(  \epsilon\right)
\label{multi:alt-h3}%
\end{align}%
\begin{align}
0  &  =E_{\epsilon}\left[  \ell^{\theta\theta\theta\theta}\left(
Z_{i},\epsilon\right)  \right]  \left(  \theta^{\epsilon}\left(
\epsilon\right)  \right)  ^{4}+6E_{\epsilon}\left[  \ell^{\theta\theta\theta
}\left(  Z_{i},\epsilon\right)  \right]  \left(  \theta^{\epsilon}\left(
\epsilon\right)  \right)  ^{2}\theta^{\epsilon\epsilon}\left(  \epsilon
\right)  \nonumber \\
&+3E_{\epsilon}\left[  \ell^{\theta\theta}\left(  Z_{i},\epsilon
\right)  \right]  \left(  \theta^{\epsilon\epsilon}\left(  \epsilon\right)
\right)  ^{2} +4E_{\epsilon}\left[  \ell^{\theta\theta}\left(  Z_{i},\epsilon\right)
\right]  \theta^{\epsilon}\left(  \epsilon\right)  \theta^{\epsilon
\epsilon\epsilon}\left(  \epsilon\right)  +E_{\epsilon}\left[  \ell^{\theta
}\left(  Z_{i},\epsilon\right)  \right]  \theta^{\epsilon\epsilon
\epsilon\epsilon}\left(  \epsilon\right)  \nonumber\\
&  +4\left(  \theta^{\epsilon}\left(
\epsilon\right)  \right)  ^{3}\left(  \int\ell^{\theta\theta\theta}\left(
z,\epsilon\right)  d\Delta\left(  z\right)  \right) 
  +12\theta^{\epsilon}\left(  \epsilon\right)  \theta^{\epsilon\epsilon
}\left(  \epsilon\right)  \left(  \int\ell^{\theta\theta}\left(
z,\epsilon\right)  d\Delta\left(  z\right)  \right)  \nonumber\\
&  +4\theta^{\epsilon
\epsilon\epsilon}\left(  \epsilon\right)  \left(  \int\ell^{\theta}\left(
z,\epsilon\right)  d\Delta\left(  z\right)  \right)  \label{multi:alt-h4}%
\end{align}%
\begin{align}
0  &  =E_{\epsilon}\left[  \allowbreak\ell^{\theta\theta\theta\theta\theta
}\left(  Z_{i},\epsilon\right)  \right]  \left(  \theta^{\epsilon}\left(
\epsilon\right)  \right)  ^{5}+\allowbreak10E_{\epsilon}\left[  \ell
^{\theta\theta\theta\theta}\left(  Z_{i},\epsilon\right)  \right]  \left(
\theta^{\epsilon}\left(  \epsilon\right)  \right)  ^{3}\allowbreak
\theta^{\epsilon\epsilon}\left(  \epsilon\right) \nonumber\\
&  +15E_{\epsilon}\left[  \ell^{\theta\theta\theta}\left(  Z_{i}%
,\epsilon\right)  \right]  \allowbreak\theta^{\epsilon}\left(  \epsilon
\right)  \left(  \theta^{\epsilon\epsilon}\left(  \epsilon\right)  \right)
^{2}\\
&  +10E_{\epsilon}\left[  \ell^{\theta\theta\theta}\left(  Z_{i}%
,\epsilon\right)  \right]  \left(  \theta^{\epsilon}\left(  \epsilon\right)
\right)  ^{2}\theta^{\epsilon\epsilon\epsilon}\left(  \epsilon\right)
+10E_{\epsilon}\left[  \ell^{\theta\theta}\left(  Z_{i},\epsilon\right)
\right]  \theta^{\epsilon\epsilon}\left(  \epsilon\right)  \theta
^{\epsilon\epsilon\epsilon}\left(  \epsilon\right) \nonumber\\
&  +5E_{\epsilon}\left[  \ell^{\theta\theta}\left(  Z_{i},\epsilon\right)
\right]  \theta^{\epsilon}\left(  \epsilon\right)  \theta^{\epsilon
\epsilon\epsilon\epsilon}\left(  \epsilon\right)  +E_{\epsilon}\left[
\ell^{\theta}\left(  Z_{i},\epsilon\right)  \right]  \theta^{\epsilon
\epsilon\epsilon\epsilon\epsilon}\left(  \epsilon\right)   \nonumber\\
& +5\left(
\theta^{\epsilon}\left(  \epsilon\right)  \right)  ^{4}\left(  \int%
\ell^{\theta\theta\theta\theta}\left(  z,\epsilon\right)  d\Delta\left(
z\right)  \right) +30\left(  \theta^{\epsilon}\left(  \epsilon\right)  \right)  ^{2}%
\theta^{\epsilon\epsilon}\left(  \epsilon\right)  \left(  \int\ell
^{\theta\theta\theta}\left(  z,\epsilon\right)  d\Delta\left(  z\right)
\right)  \nonumber\\
&  +15\left(  \theta^{\epsilon\epsilon}\left(  \epsilon\right)  \right)
^{2}\left(  \int\ell^{\theta\theta}\left(  z,\epsilon\right)  d\Delta\left(
z\right)  \right) +20\theta^{\epsilon}\left(  \epsilon\right)  \theta^{\epsilon
\epsilon\epsilon}\left(  \epsilon\right)  \left(  \int\ell^{\theta\theta
}\left(  z,\epsilon\right)  d\Delta\left(  z\right)  \right)   \nonumber\\
& +5\theta
^{\epsilon\epsilon\epsilon\epsilon}\left(  \epsilon\right)  \left(  \int%
\ell^{\theta}\left(  z,\epsilon\right)  d\Delta\left(  z\right)  \right)
\label{multi:alt-h5}%
\end{align}
and
\begin{align}
0  &  =E_{\epsilon}\left[  \allowbreak\ell^{\theta\theta\theta\theta
\theta\theta}\left(  Z_{i},\epsilon\right)  \right]  \left(  \theta^{\epsilon
}\left(  \epsilon\right)  \right)  ^{6}+\allowbreak15E_{\epsilon}\left[
\ell^{\theta\theta\theta\theta\theta}\left(  Z_{i},\epsilon\right)  \right]
\left(  \theta^{\epsilon}\left(  \epsilon\right)  \right)  ^{4}\theta
^{\epsilon\epsilon}\left(  \epsilon\right) \nonumber\\
&  +45E_{\epsilon}\left[  \ell^{\theta\theta\theta\theta}\left(
Z_{i},\epsilon\right)  \right]  \allowbreak\left(  \theta^{\epsilon}\left(
\epsilon\right)  \right)  ^{2}\left(  \theta^{\epsilon\epsilon}\left(
\epsilon\right)  \right)  ^{2}+20\allowbreak E_{\epsilon}\left[  \ell
^{\theta\theta\theta\theta}\left(  Z_{i},\epsilon\right)  \right]  \left(
\theta^{\epsilon}\left(  \epsilon\right)  \right)  ^{3}\theta^{\epsilon
\epsilon\epsilon}\left(  \epsilon\right) \nonumber\\
&  +15E_{\epsilon}\left[  \ell^{\theta\theta\theta}\left(  Z_{i}%
,\epsilon\right)  \right]  \left(  \theta^{\epsilon\epsilon}\left(
\epsilon\right)  \right)  ^{3}+60E_{\epsilon}\left[  \ell^{\theta\theta\theta
}\left(  Z_{i},\epsilon\right)  \right]  \allowbreak\theta^{\epsilon}\left(
\epsilon\right)  \theta^{\epsilon\epsilon}\left(  \epsilon\right)
\theta^{\epsilon\epsilon\epsilon}\left(  \epsilon\right) \nonumber\\
&  +\allowbreak15E_{\epsilon}\left[  \ell^{\theta\theta\theta}\left(
Z_{i},\epsilon\right)  \right]  \left(  \theta^{\epsilon}\left(
\epsilon\right)  \right)  ^{2}\theta^{\epsilon\epsilon\epsilon\epsilon}\left(
\epsilon\right)  +10E_{\epsilon}\left[  \ell^{\theta\theta}\left(
Z_{i},\epsilon\right)  \right]  \left(  \theta^{\epsilon\epsilon\epsilon
}\left(  \epsilon\right)  \right)  ^{2}\nonumber\\
&  +15E_{\epsilon}\left[  \ell^{\theta\theta}\left(  Z_{i},\epsilon\right)
\right]  \theta^{\epsilon\epsilon}\left(  \epsilon\right)  \theta
^{\epsilon\epsilon\epsilon\epsilon}\left(  \epsilon\right)  +6E_{\epsilon
}\left[  \ell^{\theta\theta}\left(  Z_{i},\epsilon\right)  \right]
\theta^{\epsilon}\left(  \epsilon\right)  \theta^{\epsilon\epsilon
\epsilon\epsilon\epsilon}\left(  \epsilon\right) \nonumber\\
&  +E_{\epsilon}\left[  \ell^{\theta}\left(  Z_{i},\epsilon\right)  \right]
\theta^{\epsilon\epsilon\epsilon\epsilon\epsilon\epsilon}\left(
\epsilon\right)  +6\left(  \theta^{\epsilon}\left(  \epsilon\right)  \right)
^{5}\left(  \int\ell^{\theta\theta\theta\theta\theta}\left(  z,\epsilon
\right)  d\Delta\left(  z\right)  \right) \nonumber\\
&  +60\left(  \theta^{\epsilon}\left(  \epsilon\right)  \right)
^{3}\allowbreak\theta^{\epsilon\epsilon}\left(  \epsilon\right)  \left(
\int\ell^{\theta\theta\theta\theta}\left(  z,\epsilon\right)  d\Delta\left(
z\right)  \right)  +90\theta^{\epsilon}\left(  \theta^{\epsilon\epsilon
}\left(  \epsilon\right)  \right)  ^{2}\left(  \int\ell^{\theta\theta\theta
}\allowbreak\left(  z,\epsilon\right)  d\Delta\left(  z\right)  \right)
\nonumber\\
&  +60\left(  \theta^{\epsilon}\left(  \epsilon\right)  \right)  ^{2}%
\theta^{\epsilon\epsilon\epsilon}\left(  \epsilon\right)  \left(  \int%
\ell^{\theta\theta\theta}\left(  z,\epsilon\right)  d\Delta\left(  z\right)
\right)  +60\theta^{\epsilon\epsilon}\left(  \epsilon\right)  \theta
^{\epsilon\epsilon\epsilon}\left(  \epsilon\right)  \left(  \int\ell
^{\theta\theta}\left(  z,\epsilon\right)  d\Delta\left(  z\right)  \right)
\nonumber\\
&  +30\theta^{\epsilon}\left(  \epsilon\right)  \theta^{\epsilon
\epsilon\epsilon\epsilon}\left(  \epsilon\right)  \left(  \int\ell
^{\theta\theta}\left(  z,\epsilon\right)  d\Delta\left(  z\right)  \right)
+6\theta^{\epsilon\epsilon\epsilon\epsilon\epsilon}\left(  \epsilon\right)
\left(  \int\ell^{\theta}\left(  z,\epsilon\right)  d\Delta\left(  z\right)
\right)  \label{multi:alt-h6}%
\end{align}
Here, $E_{\epsilon}\left[  \cdot\right]  $ is defined such that
\[
E_{\epsilon}\left[  g\left(  Z_{i},\epsilon\right)  \right]  \equiv\int
g\left(  z,\epsilon\right)  dF_{\epsilon}\left(  z\right)
\]
Evaluating expressions (\ref{multi:alt-h1}) - (\ref{multi:alt-h4}) at
$\epsilon=0$, we obtain
\begin{equation}
\theta^{\epsilon}=\frac{1}{-E\left[  \ell^{\theta}\right]  }\left(  \int\ell
d\Delta\right)  =\frac{1}{\mathcal{I}}\int\ell d\Delta, \label{theta-e1}%
\end{equation}%
\begin{align}
\theta^{\epsilon\epsilon}  &  =\frac{1}{-E\left[  \ell^{\theta}\right]
}\left(  E\left[  \ell^{\theta\theta}\right]  \left(  \theta^{\epsilon
}\right)  ^{2}+2\left(  \int\ell^{\theta}d\Delta\right)  \theta^{\epsilon
}\right) \nonumber\\
&  =\frac{E\left[  \ell^{\theta\theta}\right]  }{-E\left[  \ell^{\theta
}\right]  }\left(  \theta^{\epsilon}\right)  ^{2}+2\frac{1}{-E\left[
\ell^{\theta}\right]  }\left(  \int\ell^{\theta}d\Delta\right)  \theta
^{\epsilon}\nonumber\\
&  =\frac{E\left[  \ell^{\theta\theta}\right]  }{\mathcal{I}^{3}}\left(
\int\ell d\Delta\right)  ^{2}+\frac{2}{\mathcal{I}^{2}}\left(  \int%
\ell^{\theta}d\Delta\right)  \left(  \int\ell d\Delta\right)  ,
\label{theta-e2}%
\end{align}%
\begin{align}
\theta^{\epsilon\epsilon\epsilon}  &  =\frac{E\left[  \ell^{\theta\theta
\theta}\right]  }{-E\left[  \ell^{\theta}\right]  }\left(  \theta^{\epsilon
}\right)  ^{3}+3\frac{E\left[  \ell^{\theta\theta}\right]  }{-E\left[
\ell^{\theta}\right]  }\theta^{\epsilon}\theta^{\epsilon\epsilon}+3\frac
{1}{-E\left[  \ell^{\theta}\right]  }\left(  \int\ell^{\theta\theta}%
d\Delta\right)  \left(  \theta^{\epsilon}\right)  ^{2}  \nonumber\\
& +3\frac{1}{-E\left[
\ell^{\theta}\right]  }\left(  \int\ell^{\theta}d\Delta\right)  \theta
^{\epsilon\epsilon}\nonumber\\
&  =\left(  \frac{E\left[  \ell^{\theta\theta\theta}\right]  }{\mathcal{I}%
^{4}}+\frac{3\left(  E\left[  \ell^{\theta\theta}\right]  \right)  ^{2}%
}{\mathcal{I}^{5}}\right)  \left(  \int\ell d\Delta\right)  ^{3}%
+\frac{9E\left[  \ell^{\theta\theta}\right]  }{\mathcal{I}^{4}}\left(
\int\ell d\Delta\right)  ^{2}\left(  \int\ell^{\theta}d\Delta\right)
\nonumber\\
&  +\frac{3}{\mathcal{I}^{3}}\left(  \int\ell d\Delta\right)  ^{2}\left(
\int\ell^{\theta\theta}d\Delta\right)  +\frac{6}{\mathcal{I}^{3}}\left(
\int\ell d\Delta\right)  \left(  \int\ell^{\theta}d\Delta\right)  ^{2}
\label{theta-e3}%
\end{align}

\begin{align}
\theta^{\epsilon\epsilon\epsilon\epsilon}  &  =\frac{E\left[  \ell
^{\theta\theta\theta\theta}\right]  }{-E\left[  \ell^{\theta}\right]  }\left(
\theta^{\epsilon}\right)  ^{4}+\allowbreak6\frac{E\left[  \ell^{\theta
\theta\theta}\right]  }{-E\left[  \ell^{\theta}\right]  }\left(
\theta^{\epsilon}\right)  ^{2}\theta^{\epsilon\epsilon}+3\frac{E\left[
\ell^{\theta\theta}\right]  }{-E\left[  \ell^{\theta}\right]  }\left(
\theta^{\epsilon\epsilon}\right)  ^{2}\nonumber\\
&  +4\frac{E\left[  \ell^{\theta\theta}\right]  }{-E\left[  \ell^{\theta
}\right]  }\theta^{\epsilon}\theta^{\epsilon\epsilon\epsilon}+4\frac
{1}{-E\left[  \ell^{\theta}\right]  }\left(  \theta^{\epsilon}\right)
^{3}\left(  \int\ell^{\theta\theta\theta}d\Delta\right) \nonumber\\
&  +12\frac{1}{-E\left[  \ell^{\theta}\right]  }\theta^{\epsilon}%
\theta^{\epsilon\epsilon}\left(  \int\ell^{\theta\theta}d\Delta\right)
+4\frac{1}{-E\left[  \ell^{\theta}\right]  }\theta^{\epsilon\epsilon\epsilon
}\left(  \int\ell^{\theta}d\Delta\right)  \label{theta-e4}%
\end{align}%
\begin{align}
\theta^{\epsilon\epsilon\epsilon\epsilon\epsilon}  &  =\frac{E\left[
\allowbreak\ell^{\theta\theta\theta\theta\theta}\right]  }{-E\left[
\ell^{\theta}\right]  }\left(  \theta^{\epsilon}\right)  ^{5}+\allowbreak
10\frac{E\left[  \ell^{\theta\theta\theta\theta}\right]  }{-E\left[
\ell^{\theta}\right]  }\left(  \theta^{\epsilon}\right)  ^{3}\allowbreak
\theta^{\epsilon\epsilon}+15\frac{E\left[  \ell^{\theta\theta\theta}\right]
}{-E\left[  \ell^{\theta}\right]  }\allowbreak\theta^{\epsilon}\left(
\theta^{\epsilon\epsilon}\right)  ^{2}\nonumber\\
&  +10\frac{E\left[  \ell^{\theta\theta\theta}\right]  }{-E\left[
\ell^{\theta}\right]  }\left(  \theta^{\epsilon}\right)  ^{2}\theta
^{\epsilon\epsilon\epsilon}+10\frac{E\left[  \ell^{\theta\theta}\right]
}{-E\left[  \ell^{\theta}\right]  }\theta^{\epsilon\epsilon}\theta
^{\epsilon\epsilon\epsilon}\nonumber\\
&  +5\frac{E\left[  \ell^{\theta\theta}\right]  }{-E\left[  \ell^{\theta
}\right]  }\theta^{\epsilon}\theta^{\epsilon\epsilon\epsilon\epsilon}%
+5\frac{1}{-E\left[  \ell^{\theta}\right]  }\left(  \theta^{\epsilon}\right)
^{4}\left(  \int\ell^{\theta\theta\theta\theta}d\Delta\right) \nonumber\\
&  +30\frac{1}{-E\left[  \ell^{\theta}\right]  }\left(  \theta^{\epsilon
}\right)  ^{2}\theta^{\epsilon\epsilon}\left(  \int\ell^{\theta\theta\theta
}d\Delta\right)  +15\frac{1}{-E\left[  \ell^{\theta}\right]  }\left(
\theta^{\epsilon\epsilon}\right)  ^{2}\left(  \int\ell^{\theta\theta}%
d\Delta\right) \nonumber\\
&  +20\frac{1}{-E\left[  \ell^{\theta}\right]  }\theta^{\epsilon}%
\theta^{\epsilon\epsilon\epsilon}\left(  \int\ell^{\theta\theta}%
d\Delta\right)  +5\frac{1}{-E\left[  \ell^{\theta}\right]  }\theta
^{\epsilon\epsilon\epsilon\epsilon}\left(  \int\ell^{\theta}d\Delta\right)  .
\label{theta-e5}%
\end{align}

\section{Proof of Proposition \ref{ApproxThetaHatStar}}

Let $\hat{Q}^{\ast}(\theta)=\int\log f\left(  \cdot,\theta\right)  d\hat
{F}^{\ast}\left(  z\right)  ,\hat{Q}_{\epsilon}(\theta)=\int\log f\left(
\cdot,\theta\right)  d\hat{F}_{\epsilon}\left(  z\right)  $ and $\hat
{Q}(\theta)=\int\log f\left(  \cdot,\theta\right)  d\hat{F}\left(  z\right)  $
such that $\hat{Q}_{\epsilon}(\theta)-\hat{Q}(\theta)=\epsilon\sqrt{n}\left(
\hat{Q}^{\ast}(\theta)-\hat{Q}(\theta)\right)  .$ By Conditions (\ref{HH}),
(\ref{M}) and (\ref{EC}) and Gin\'{e} and Zinn (1996, Theorem 2.6) it follows
that $\sup_{\theta}\left\vert \hat{Q}_{\epsilon}(\theta)-\hat{Q}%
(\theta)\right\vert \leq\sup_{\theta}\left\vert \hat{Q}^{\ast}(\theta)-\hat
{Q}(\theta)\right\vert \rightarrow0$ in probability, $P^{\mathbb{N}}$a.s. By
standard arguments such as Arcones and Gin\'{e} (1992), it follows that
uniformly in $\epsilon\in\left[  -n^{-1/2},n^{-1/2}\right]  ,$ $\hat{\theta
}^{\ast}\left(  \epsilon\right)  \overset{P^{\ast}}{\rightarrow}\hat{\theta},$
$P^{\mathbb{N}}$a.s. This implies that for any compact set $K\subset\Theta$
with $\theta_{0}\in K,$ $P^{\ast}(\hat{\theta}^{\ast}\left(  \epsilon\right)
\in K)\rightarrow1,P^{\mathbb{N}}$a.s., as $n\rightarrow\infty.$ Consider the
function $\hat{G}\left(  \epsilon,\theta\right)  \equiv\int\ell\left(
\cdot,\theta\right)  d\hat{F}_{\epsilon}\left(  z\right)  .$ If $\partial
\Theta$ is the boundary of $\Theta$ then $P^{\ast}(\hat{G}\left(
\epsilon,\hat{\theta}^{\ast}\left(  \epsilon\right)  \right)  \neq0)\leq
P^{\ast}\left(  \hat{\theta}^{\ast}\left(  \epsilon\right)  \in\partial
\Theta\right)  \leq1-P^{\ast}\left(  \theta^{\ast}\left(  \epsilon\right)  \in
K\right)  \rightarrow0,.P^{\mathbb{N}}$a.s. We now condition on the event
$\left\{  \hat{G}\left(  \epsilon,\hat{\theta}^{\ast}\left(  \epsilon\right)
\right)  =0\right\}  $. By the same arguments as in the proof of Proposition
\ref{ApproxThetaHat-strong} it follows that there exists some $\tilde
{\epsilon}\in\left[  0,n^{-1/2}\right]  $ such that $\sqrt{n}\left(
\widehat{\theta}^{\ast}-\widehat{\theta}\right)  =\widehat{\theta}^{\epsilon
}\left(  0\right)  +\sum_{k=1}^{m-1}\frac{1}{k!n^{k/2}}\widehat{\theta
}^{\left(  k\right)  }\left(  0\right)  +\frac{1}{m!n^{m/2}}\widehat{\theta
}^{\left(  m\right)  }\left(  \tilde{\epsilon}\right)  $ $P^{\mathbb{N}}$a.s.,
where $\widehat{\theta}^{\epsilon}\left(  0\right)  $ is obtained from
evaluating $\int\frac{dh\left(  z,\epsilon\right)  }{d\epsilon}d\widehat{F}%
_{\epsilon}\left(  z\right)  +\int h\left(  z,\epsilon\right)
d\widehat{\Delta}\left(  z\right)  $ at $\epsilon=0$. We obtain
\[
o_{p}(n^{-m/2})=\int\ell^{\theta}(z,\widehat{\theta})d\widehat{F}\left(
z\right)  \widehat{\theta}^{\epsilon}(0)+\int\ell\left(  z,\widehat{\theta
}\right)  d\widehat{\Delta}\left(  z\right)  ,
\]
where $\int\ell^{\theta}(z,\widehat{\theta})d\widehat{F}\left(  z\right)
\equiv n^{-1}\sum_{i=1}^{n}\ell^{\theta}(Z_{i},\widehat{\theta})$ and
$\widehat{\Delta}\left(  z\right)  \equiv\sqrt{n}\left(  \widehat{F}^{\ast
}\left(  z\right)  -\widehat{F}\left(  z\right)  \right)  $. Similar
expressions can be found for higher order derivatives of $\widehat{\theta
}(\epsilon).$ These expressions depend on $n^{-1}\sum_{i=1}^{n}\ell^{\left(
k\right)  }\left(  Z_{i},\widehat{\theta}\right)  $ and $\int\ell^{\left(
k\right)  }\left(  z,\widehat{\theta}\right)  d\widehat{\Delta}\left(
z\right)  $ for $k=0,1,...,6$. By Condition \ref{HH} and Lemma
\ref{Lem-consistency-of-theta-epsilon} (in Appendix \ref{sec-lfbs}), it
follows that $n^{-1}\sum_{i=1}^{n}\ell^{\left(  k\right)  }\left(
Z_{i},\widehat{\theta}\right)  =E\left[  \ell^{\left(  k\right)  }\left(
Z_{i},\theta_{0}\right)  \right]  +o_{p}\left(  1\right)  $ by a uniform law
of large numbers. By Proposition \ref{Donsker} (in Appendix \ref{sec-lfbs}),
the class $\mathfrak{F}$ is Donsker. By the proof of Theorem 2.4 in Gine and
Zinn (1990) it follows that the following conditional stochastic
equicontinuity property
\[
P^{\mathbb{N}}\text{a.s., }\lim_{\delta\downarrow0}\underset{n\rightarrow
\infty}{\lim\sup}P^{\ast}\left(  \sup_{\left\vert t-s\right\vert <\delta
}\left\vert \int\left(  \ell^{\left(  k\right)  }\left(  z,t\right)
-\ell^{\left(  k\right)  }(z,s)\right)  d\widehat{\Delta}\left(  z\right)
\right\vert >\eta\right)  =0
\]
holds. Then
\begin{align}
\lefteqn{P^{\ast}\left(  \left\vert \int\left(  \ell^{\left(  k\right)
}\left(  z,\widehat{\theta}\right)  -\ell^{\left(  k\right)  }\left(
z,\theta_{0}\right)  \right)  d\widehat{\Delta}\left(  z\right)  \right\vert
>\eta\right)  }\nonumber\\
&  \leq P^{\ast}\left(  \sup_{\left\vert \theta-\theta_{0}\right\vert <\delta
}\left\vert \int\left(  \ell^{\left(  k\right)  }\left(  z,\theta\right)
-\ell^{\left(  k\right)  }(z,\theta_{0})\right)  d\widehat{\Delta}\left(
z\right)  \right\vert >\eta/2\right)  +P^{\ast}\left(  \left\vert
\widehat{\theta}-\theta_{0}\right\vert \geq\eta/2\right)  \label{StochEqBound}%
\end{align}
or
\[
\int\ell^{\left(  k\right)  }\left(  z,\widehat{\theta}\right)
d\widehat{\Delta}\left(  z\right)  =\int\ell^{\left(  k\right)  }\left(
z,\theta_{0}\right)  d\widehat{\Delta}\left(  z\right)  +o_{p}(1)\text{
}P^{\mathbb{N}}\text{a.s.}%
\]
It now follows from Proposition \ref{Donsker} and Theorem 2.4 of Gine and Zinn
(1990) that \\ $\int\ell^{\left(  k\right)  }\left(  z,\theta_{0}\right)
d\widehat{\Delta}\left(  z\right)  \rightsquigarrow\int\ell^{\left(  k\right)
}\left(  z,\theta_{0}\right)  dT\left(  z\right)  $ almost surely, where
$T\left(  z\right)  $ is a Brownian Bridge process. We finally have to analyze
the term $\widehat{\theta}^{\left(  m\right)  }\left(  \tilde{\epsilon
}\right)  $ which contains expressions of the form $\int\ell^{\left(
k\right)  }(z,\widehat{\theta}^{\ast}(\epsilon))d\widehat{F}_{\epsilon}\left(
z\right)  $ and $\int\ell^{\left(  k\right)  }\left(  z,\widehat{\theta}%
^{\ast}(\epsilon)\right)  d\widehat{\Delta}\left(  z\right)  .$ For $\int%
\ell^{\left(  k\right)  }\left(  z,\widehat{\theta}^{\ast}(\epsilon)\right)
d\widehat{\Delta}\left(  z\right)  $ we use the same inequality as in
(\ref{StochEqBound}) together with Lemma \ref{BootstrapThetaEpsilonConv} (in
Appendix \ref{sec-lfbs}) to show that
\[
\int\ell^{\left(  k\right)  }\left(  z,\widehat{\theta}^{\ast}(\epsilon
)\right)  d\widehat{\Delta}\left(  z\right)  =\int\ell^{\left(  k\right)
}\left(  z,\theta_{0}\right)  d\widehat{\Delta}\left(  z\right)
+o_{p}(1)\text{ }P^{\mathbb{N}}\text{a.s.}%
\]
Next consider
\begin{align*}
&  \left\vert \int\ell^{\left(  k\right)  }(z,\widehat{\theta}^{\ast}%
(\epsilon))d\widehat{F}_{\epsilon}\left(  z\right)  -\ell^{\left(  k\right)
}(z,\theta_{0})dF\left(  z\right)  \right\vert \\
&  \leq\left\vert \epsilon\right\vert \left\vert \int\ell^{\left(  k\right)
}(z,\widehat{\theta}^{\ast}(\epsilon))d\widehat{\Delta}\left(  z\right)
\right\vert +\left\vert \int\ell^{\left(  k\right)  }(z,\theta_{0})d\left(
F\left(  z\right)  -\widehat{F}\left(  z\right)  \right)  \right\vert \\
&  +\left\vert \int\left[  \ell^{\left(  k\right)  }(z,\widehat{\theta}^{\ast
}(\epsilon))-\ell^{\left(  k\right)  }(z,\theta_{0})\right]  d\widehat{F}%
\left(  z\right)  \right\vert
\end{align*}
where $\int\ell^{\left(  k\right)  }(z,\widehat{\theta}^{\ast}(\epsilon
))d\widehat{\Delta}\left(  z\right)  =O_{p}(1)$ $P^{\mathbb{N}}$a.s. by
Proposition \ref{Donsker} and $\sup\left\vert \epsilon\right\vert
=O(n^{-1/2})$. The second term is $o_{p}(1)$ by a law of large numbers.
Finally,
\begin{multline*}
P^{\ast}\left(  \sup_{\epsilon}\left\vert \int\left[  \ell^{\left(  k\right)
}(z,\widehat{\theta}^{\ast}(\epsilon))-\ell^{\left(  k\right)  }(z,\theta
_{0})\right]  d\widehat{F}\left(  z\right)  \right\vert >\eta\right) \\
\leq P^{\ast}\left(  \sup_{\left\vert \theta-\theta_{0}\right\vert <\delta
}\left\vert \int\left[  \ell^{\left(  k\right)  }(z,\theta)-\ell^{\left(
k\right)  }(z,\theta_{0})\right]  d\widehat{F}\left(  z\right)  \right\vert
>\eta\right)  +P^{\ast}\left(  \sup_{0\leq\epsilon\leq1/\sqrt{n}}\left\vert
\widehat{\theta}^{\ast}(\epsilon)-\theta_{0}\right\vert \geq\delta\right)
\end{multline*}
where the first probability is zero with $P^{\mathbb{N}}$-probability tending
to one by stochastic equicontinuity and the second probability goes to zero by
Lemma \ref{BootstrapThetaEpsilonConv} (in Appendix \ref{sec-lfbs}). It follows
that $\int\ell^{\left(  k\right)  }(z,\widehat{\theta}^{\ast}(\epsilon
))d\widehat{F}_{\epsilon}\left(  z\right)  \overset{p}{\rightarrow}%
E\ell^{\left(  k\right)  }(z,\theta_{0})$ $P^{\mathbb{N}}$a.s. Together, these
results imply that $\sup_{\epsilon}\left\vert \widehat{\theta}^{\left(
k\right)  }(\epsilon)\right\vert =O_{p}(1)$ $P^{\mathbb{N}}$a.s. for $k\leq6.$
This establishes the validity of the expansion.

\begin{proof}
See Appendix \ref{Proof-Main-Results}.
\end{proof}

While this result establishes that we can consistently estimate the higher
order bias it is not sufficient to guarantee good higher order properties of
the bias corrected estimator. For this reason we establish the next result.

\section{Consistency of \texorpdfstring{$b^{\ast}$}{b*}}

We will establish that $b^{\ast}$ estimates the higher order bias $b\left(
\theta\right)  $ consistently.

\begin{proposition}
\label{BBC}Assume Conditions \ref{HH},\ref{M} and \ref{EC} hold. Then
$b^{\ast}=\left.  b\left(  \theta_{0}\right)  \right/  n+o_{p}\left(
n^{-1}\right)  $.
\end{proposition}

\begin{proof}
Introduce the truncation function $h_{n}(x)$, where
\begin{equation}
h_{n}\left(  x\right)  =\left\{
\begin{array}
[c]{cc}%
-n^{\alpha} & \text{if }x<-n^{\alpha}\\
x & \text{if }\left\vert x\right\vert <n^{\alpha}\\
n^{\alpha} & \text{if }x>n^{\alpha}%
\end{array}
\right.  \label{hn}%
\end{equation}
with $\alpha\in\left(  0,\frac{41}{30}\right)  .$ We first show that
\[
E^{\ast}\left[  \widehat{\theta}^{\ast}-\widehat{\theta}\right]  -E^{\ast
}\left[  h_{n}\left(  \widehat{\theta}^{\ast}-\widehat{\theta}\right)
\right]  =o_{p}(n^{-3/2}).
\]
Note that $\left(  \widehat{\theta}^{\ast}-\widehat{\theta}\right)
-h_{n}\left(  \widehat{\theta}^{\ast}-\widehat{\theta}\right)  \leq\left\vert
\widehat{\theta}^{\ast}-\widehat{\theta}\right\vert 1\left\{  \left\vert
\widehat{\theta}^{\ast}-\widehat{\theta}\right\vert >n^{\alpha}\right\}  $. By
compactness of $\Theta$ there exists a constant $C$ such that $\left\vert
\widehat{\theta}^{\ast}-\widehat{\theta}\right\vert <C$ such that
\[
\left\vert E^{\ast}\left[  \widehat{\theta}^{\ast}-\widehat{\theta}\right]
-E^{\ast}\left[  h_{n}\left(  \widehat{\theta}^{\ast}-\widehat{\theta}\right)
\right]  \right\vert \leq CP^{\ast}\left(  \sqrt{n}\left\vert \widehat{\theta
}^{\ast}-\widehat{\theta}\right\vert >n^{\alpha+1/2}\right)  .
\]
Using the expansion for $\sqrt{n}\left(  \widehat{\theta}^{\ast}%
-\widehat{\theta}\right)  $ from Proposition \ref{ApproxThetaHatStar} together
with Lemma \ref{boot:bound:theta} (in Appendix \ref{sec-lfbs}), it follows
that $P^{\ast}\left(  \sqrt{n}\left\vert \widehat{\theta}^{\ast}%
-\widehat{\theta}\right\vert >n^{\alpha+1/2}\right)  =o_{p}(n^{-20/3}).$ This
shows that we can replace $E^{\ast}\left[  \widehat{\theta}^{\ast
}-\widehat{\theta}\right]  $ with a truncated integral $E^{\ast}\left[
h_{n}\left(  \widehat{\theta}^{\ast}-\widehat{\theta}\right)  \right]  $. Let
\begin{equation}
\widehat{\theta}_{a}^{\ast}\equiv n^{-1/2}\widehat{\theta}^{\epsilon}\left(
0\right)  +\frac{1}{2}\frac{1}{n}\widehat{\theta}^{\epsilon\epsilon}\left(
0\right)  +\frac{1}{6}\frac{1}{n^{3/2}}\widehat{\theta}^{\epsilon
\epsilon\epsilon}\left(  0\right)  +\frac{1}{24}\frac{1}{n^{2}}\widehat{\theta
}^{\epsilon\epsilon\epsilon\epsilon}\left(  0\right)  . \label{Ca}%
\end{equation}
Because $\left\vert h_{n}(x)-h_{n}(y)\right\vert \leq2n^{\alpha}%
\wedge\left\Vert x-y\right\Vert $, we have
\[
\left\vert h_{n}\left(  \widehat{\theta}^{\ast}-\widehat{\theta}\right)
-h_{n}\left(  \widehat{\theta}_{a}^{\ast}\right)  \right\vert \leq\min\left(
2n^{\alpha},\;\frac{1}{96n^{5/2}}\sup_{0\leq\epsilon\leq1/\sqrt{n}}\left\Vert
\widehat{\theta}^{\epsilon\epsilon\epsilon\epsilon\epsilon}\left(
\epsilon\right)  \right\Vert \right)  .
\]
Fix $\varepsilon>0$ and $\frac{7}{96}<\delta<\frac{1}{2}$ arbitrary. Taking
expectations with respect to the measure $\widehat{F}$ leads to
\begin{multline*}
\left\vert E^{\ast}\left[  h_{n}\left(  \widehat{\theta}^{\ast}%
-\widehat{\theta}\right)  \right]  -E^{\ast}\left[  h_{n}\left(
\widehat{\theta}_{a}^{\ast}\right)  \right]  \right\vert \\
\leq\varepsilon/n^{2-\delta}+2n^{\alpha}\cdot P^{\ast}\left[  \frac
{1}{96n^{5/2}}\sup_{0\leq\epsilon\leq1/\sqrt{n}}\left\Vert \widehat{\theta
}^{\epsilon\epsilon\epsilon\epsilon\epsilon}\left(  \epsilon\right)
\right\Vert >\varepsilon/n^{2-\delta}\right]  .
\end{multline*}
Use the fact that $P^{\ast}\left[  \frac{1}{96n^{5/2}}\sup_{0\leq\epsilon
\leq1/\sqrt{n}}\left\Vert \widehat{\theta}^{\epsilon\epsilon\epsilon
\epsilon\epsilon}\left(  \epsilon\right)  \right\Vert >\varepsilon
/n^{2-\delta}\right]  =o_{p}\left(  n^{-76/60-\left(  16/5\right)  \delta
}\right)  $ by setting $-v=1/60+\delta/5$ in Lemma \ref{boot:bound:theta} (in
Appendix \ref{sec-lfbs}). Choose $\delta\in\left(  7/96+\left(  5/16\right)
\alpha,1/2\right)  $. It follows that
\begin{align}
\left\vert E^{\ast}\left[  h_{n}\left(  \widehat{\theta}^{\ast}%
-\widehat{\theta}\right)  \right]  -E^{\ast}\left[  h_{n}\left(
\widehat{\theta}_{a}^{\ast}\right)  \right]  \right\vert  &  \leq
\varepsilon/n^{2-\delta}+2o_{p}\left(  n^{-76/60-\left(  16/5\right)
\delta+\alpha}\right) \nonumber\\
&  =O_{p}(n^{\delta-2})=o_{p}(n^{-3/2}) \label{hstar1}%
\end{align}
Next, we show that
\begin{equation}
E^{\ast}\left[  h_{n}\left(  \widehat{\theta}_{a}^{\ast}\right)  \right]
-E^{\ast}\left[  \widehat{\theta}_{a}^{\ast}\right]  =o_{p}\left(
n^{-3/2}\right)  . \label{hstar2}%
\end{equation}
Note that
\begin{align*}
\left\vert E^{\ast}\left[  h_{n}\left(  \widehat{\theta}_{a}^{\ast}\right)
\right]  -E^{\ast}\left[  \widehat{\theta}_{a}^{\ast}\right]  \right\vert  &
\leq E^{\ast}\left[  \left\vert n^{\alpha}-\widehat{\theta}_{a}^{\ast
}\right\vert \mathbf{1}\left\{  \left\vert \widehat{\theta}_{a}^{\ast
}\right\vert \geq n^{\alpha}\right\}  \right] \\
&  \leq E^{\ast}\left[  \left\vert \widehat{\theta}_{a}^{\ast}\right\vert
\mathbf{1}\left\{  \left\vert \widehat{\theta}_{a}^{\ast}\right\vert \geq
n^{\alpha}\right\}  \right] \\
&  +n^{\alpha}E^{\ast}\left[  \mathbf{1}\left\{  \left\vert \widehat{\theta
}_{a}^{\ast}\right\vert \geq n^{\alpha}\right\}  \right] \\
&  \leq2E^{\ast}\left[  \frac{\left\vert \widehat{\theta}_{a}^{\ast
}\right\vert ^{4}}{\left(  n^{\alpha}\right)  ^{3}}\right]  .
\end{align*}
Here, $\left\vert \widehat{\theta}_{a}^{\ast}\right\vert ^{4}$ is a forth
order polynomial in $a=\widehat{\theta}^{\epsilon}\left(  0\right)  $,
$b=\frac{1}{2}\widehat{\theta}^{\epsilon\epsilon}\left(  0\right)  $,
$c=\frac{1}{6}\widehat{\theta}^{\epsilon\epsilon\epsilon}\left(  0\right)  $,
and $d=\frac{1}{24}\widehat{\theta}^{\epsilon\epsilon\epsilon\epsilon}\left(
0\right)  .$ Expectations of all terms of the from $E^{\ast}\left[  a^{i}%
b^{j}c^{k}d^{l}\right]  $ where $i,j,k,l\in\left\{  0,1,2,3,4\right\}  $ and
$i+j+k+l=4$ are bounded in probability such that $E^{\ast}\left[  a^{i}%
b^{j}c^{k}d^{l}\right]  =O_{p}(1)$ where $E^{\ast}\left[  \frac{1}{n^{2}}%
a^{4}\right]  =O_{p}(n^{-2})$ is the largest term. It follows that $\left\vert
E^{\ast}\left[  h_{n}\left(  \widehat{\theta}_{a}^{\ast}\right)  \right]
-E^{\ast}\left[  \widehat{\theta}_{a}^{\ast}\right]  \right\vert =O_{p}\left(
n^{-2-3\alpha}\right)  =o_{p}\left(  n^{-3/2}\right)  $. Because $E^{\ast
}\left[  \frac{1}{24n^{2}}\widehat{\theta}^{\epsilon\epsilon\epsilon\epsilon
}\left(  0\right)  \right]  =O_{p}\left(  n^{-2}\right)  $, we have
\begin{equation}
E^{\ast}\left[  \widehat{\theta}_{a}^{\ast}\right]  =E^{\ast}\left[
\widehat{\theta}_{aa}^{\ast}\right]  +o_{p}(n^{-3/2}), \label{hstar3}%
\end{equation}
where
\begin{equation}
\widehat{\theta}_{aa}^{\ast}\equiv\frac{1}{n^{1/2}}\widehat{\theta}^{\epsilon
}\left(  0\right)  +\frac{1}{2}\frac{1}{n}\widehat{\theta}^{\epsilon\epsilon
}\left(  0\right)  +\frac{1}{6}\frac{1}{n^{3/2}}\widehat{\theta}%
^{\epsilon\epsilon\epsilon}\left(  0\right)  . \label{Caa}%
\end{equation}
In order to evaluate $E^{\ast}\left[  \widehat{\theta}_{aa}^{\ast}\right]  $
we use Proposition \ref{ApproxThetaHatStar} by which $\widehat{\theta
}^{\epsilon}\left(  0\right)  =\widehat{\mathcal{I}}^{-1}U^{\ast}\left(
\widehat{\theta}\right)  ,$ $\widehat{\theta}^{\epsilon\epsilon}\left(
0\right)  =\widehat{\mathcal{I}}^{-3}\widehat{\mathcal{Q}}_{1}\left(
\widehat{\theta}\right)  U^{\ast}\left(  \widehat{\theta}\right)
^{2}+2\widehat{\mathcal{I}}^{-2}U^{\ast}\left(  \widehat{\theta}\right)
V^{\ast}\left(  \widehat{\theta}\right)  $ and
\begin{align*}
\widehat{\theta}^{\epsilon\epsilon\epsilon}\left(  0\right)   &
=\widehat{\mathcal{I}}^{-4}\widehat{\mathcal{Q}}_{2}\left(  \widehat{\theta
}\right)  U^{\ast}\left(  \widehat{\theta}\right)  ^{3}+3\widehat{\mathcal{I}%
}^{-5}\widehat{\mathcal{Q}}_{1}\left(  \widehat{\theta}\right)  ^{2}U^{\ast
}\left(  \widehat{\theta}\right)  ^{3}+9\widehat{\mathcal{I}}^{-4}%
\widehat{\mathcal{Q}}_{1}\left(  \widehat{\theta}\right)  U^{\ast}\left(
\widehat{\theta}\right)  ^{2}V^{\ast}\left(  \widehat{\theta}\right) \\
&  +3\widehat{\mathcal{I}}^{-3}U^{\ast}\left(  \widehat{\theta}\right)
^{2}W^{\ast}\left(  \widehat{\theta}\right)  +6\widehat{\mathcal{I}}%
^{-3}U^{\ast}\left(  \widehat{\theta}\right)  V^{\ast}\left(  \widehat{\theta
}\right)  ^{2}.
\end{align*}
Note that $\widehat{\mathcal{I}},$ $\widehat{\mathcal{Q}}_{1}$ and
$\widehat{\mathcal{Q}}_{2}$ are constants with respect to $E^{\ast}$. It thus
follows that
\begin{equation}
E^{\ast}\left[  \widehat{\theta}^{\epsilon}\left(  0\right)  \right]
=\widehat{\mathcal{I}}^{-1}E^{\ast}\left[  U\left(  \widehat{\theta}\right)
\right]  =0 \label{ECaa1}%
\end{equation}
by Lemma \ref{BExp} (in Appendix \ref{sec-bootm}). We consider $E^{\ast
}\left[  U\left(  \widehat{\theta}\right)  ^{2}\right]  =\tfrac{1}{n}%
{\textstyle\sum_{i=1}^{n}}
\ell\left(  Z_{i},\widehat{\theta}\right)  ^{2}$. By Proposition \ref{Donsker}
and van der Waart and Wellner (1996, Theorem 1.5.7) it follows that
\[
\underset{n\rightarrow\infty}{\lim\sup}P\left(  \sup_{\left\vert \theta
-\theta_{0}\right\vert <\delta}\left\vert \tfrac{1}{n}%
{\textstyle\sum_{i=1}^{n}}
\ell\left(  Z_{i},\theta\right)  ^{2}-\tfrac{1}{n}%
{\textstyle\sum_{i=1}^{n}}
\ell\left(  Z_{i},\theta_{0}\right)  ^{2}\right\vert >\varepsilon\right)  =0
\]
such that by Lemma \ref{Lem-consistency-of-theta-epsilon} (in Appendix
\ref{sec-lfbs}) it follows that
\[
E^{\ast}\left[  U\left(  \widehat{\theta}\right)  ^{2}\right]  =\tfrac{1}{n}%
{\textstyle\sum_{i=1}^{n}}
\ell\left(  Z_{i},\theta_{0}\right)  ^{2}+o_{p}(1).
\]
Similar results can be established for the other expressions of Lemma
\ref{BExp} (in Appendix \ref{sec-bootm}). It therefore follows that
\begin{align}
E^{\ast}\left[  \widehat{\theta}^{\epsilon\epsilon}\left(  0\right)  \right]
&  =\mathcal{I}^{-3}\mathcal{Q}_{1}\left(  \theta_{0}\right)  \tfrac{1}{n}%
{\textstyle\sum_{i=1}^{n}}
\ell\left(  Z_{i},\theta_{0}\right)  ^{2}+2\mathcal{I}^{-2}\tfrac{1}{n}%
{\textstyle\sum_{i=1}^{n}}
\ell\left(  Z_{i},\theta_{0}\right)  \ell^{\theta}\left(  Z_{i},\theta
_{0}\right)  +o_{p}\left(  1\right) \nonumber\\
&  =\mathcal{I}^{-2}\mathcal{Q}_{1}\left(  \theta_{0}\right)  +2\mathcal{I}%
^{-2}E\left[  \ell\ell^{\theta}\right]  +o_{p}\left(  1\right) \nonumber\\
&  =2b\left(  \theta_{0}\right)  +o_{p}(1). \label{ECaa2}%
\end{align}
By the same arguments, we also obtain
\begin{equation}
E^{\ast}\left[  \widehat{\theta}^{\epsilon\epsilon\epsilon}\left(  0\right)
\right]  =O_{p}\left(  n^{-1/2}\right)  . \label{ECaa3}%
\end{equation}
Combining (\ref{Caa}), (\ref{ECaa1}), (\ref{ECaa2}), and (\ref{ECaa3}), we
obtain
\[
E^{\ast}\left[  \widehat{\theta}_{aa}^{\ast}\right]  =\frac{b\left(
\theta_{0}\right)  }{n}+o_{p}\left(  n^{-1}\right)  ,
\]
which establishes the result.
\end{proof}

\section{Technical Lemmas\label{sec-lfbs}}

\begin{lemma}
\label{Lem-consistency-of-theta-epsilon}Under Condition \ref{HH}, we have
\[
\Pr\left[  \max_{0\leq\epsilon\leq\frac{1}{\sqrt{n}}}\left\vert \theta\left(
\epsilon\right)  -\theta_{0}\right\vert \geq\eta\right]  =o\left(
n^{-\frac{23}{3}}\right)
\]
for every $\eta>0$.
\end{lemma}

\begin{lemma}
\label{multi:k}Assume that Condition \ref{HH}\ holds. Suppose that $K\left(
z;\theta\left(  \epsilon\right)  \right)  $ is equal to
\[
\frac{\partial^{m}\log f\left(  z;\theta\left(  \epsilon\right)  \right)
}{\partial\theta^{m}}\qquad
\]
for some $m\leq6$. Then, for any $\eta>0$, we have
\[
\Pr\left[  \max_{0\leq\epsilon\leq\frac{1}{\sqrt{n}}}\left\vert \int K\left(
z;\theta\left(  \epsilon\right)  \right)  dF_{\epsilon}\left(  z\right)
-E\left[  K\left(  Z_{i};\theta_{0}\right)  \right]  \right\vert >\eta\right]
=o\left(  n^{-\frac{23}{3}}\right)  .
\]
Also,
\[
\Pr\left[  \max_{0\leq\epsilon\leq\frac{1}{\sqrt{n}}}\left\vert \int K\left(
\cdot;\theta\left(  \epsilon\right)  \right)  d\Delta\right\vert >Cn^{\frac
{1}{12}-\upsilon}\right]  =o\left(  n^{-1+16\upsilon}\right)
\]
for some constant $C>0$ and for every $\upsilon$ such that $\upsilon<\frac
{1}{16}$. If $\upsilon<\frac{1}{48}$ then the above order is $o\left(
n^{-1}\right)  $.
\end{lemma}

\begin{lemma}
\label{multi:bound:theta-epsilon-1}Suppose that Condition \ref{HH} holds.
Then, we have
\begin{align*}
\Pr\left[  \max_{0\leq\epsilon\leq\frac{1}{\sqrt{n}}}\left\vert \theta
^{\epsilon}\left(  \epsilon\right)  \right\vert >Cn^{\frac{1}{12}-\upsilon
}\right]   &  =o\left(  n^{-1+16\upsilon}\right) \\
\Pr\left[  \max_{0\leq\epsilon\leq\frac{1}{\sqrt{n}}}\left\vert \theta
^{\epsilon\epsilon}\left(  \epsilon\right)  \right\vert >C\left(  n^{\frac
{1}{12}-\upsilon}\right)  ^{2}\right]   &  =o\left(  n^{-1+16\upsilon}\right)
\\
&  \vdots\\
\Pr\left[  \max_{0\leq\epsilon\leq\frac{1}{\sqrt{n}}}\left\vert \theta
^{\epsilon\epsilon\epsilon\epsilon\epsilon\epsilon}\left(  \epsilon\right)
\right\vert >C\left(  n^{\frac{1}{12}-\upsilon}\right)  ^{6}\right]   &
=o\left(  n^{-1+16\upsilon}\right)
\end{align*}
for some constant $C>0$ and for every $\upsilon$ such that $\upsilon<\frac
{1}{16}$. If $\upsilon<\frac{1}{48}$ then the above orders are $o\left(
n^{-1}\right)  $.
\end{lemma}

\begin{lemma}
\label{equiLemma} Suppose that Condition \ref{HH} holds. Let let $\bar{m}%
_{3}\left(  \theta\right)  $ and $\bar{m}_{4}\left(  \theta\right)  $ be as
defined in \ref{m-bar}. Then
\[
\sqrt{n}\left(  \widehat{\mathcal{I}}-\mathcal{I}\right)  =-V\left(
\theta_{0}\right)  -\mathcal{Q}_{1}\left(  \theta_{0}\right)  \mathcal{I}%
^{-1}U\left(  \theta_{0}\right)  +o_{p}\left(  1\right)  ,
\]%
\[
\sqrt{n}\left(  \widehat{\mathcal{Q}}_{1}\left(  \widehat{\theta}\right)
-\mathcal{Q}_{1}\left(  \theta_{0}\right)  \right)  =W\left(  \theta
_{0}\right)  +\mathcal{Q}_{2}\left(  \theta_{0}\right)  \mathcal{I}%
^{-1}U\left(  \theta_{0}\right)  +o_{p}\left(  1\right)  ,
\]%
\[
\sqrt{n}\left(  \overline{m}_{3}\left(  \widehat{\theta}\right)  -\overline
{m}_{3}\left(  \theta_{0}\right)  \right)  =\left(  E\left[  V_{i}\left(
\theta_{0}\right)  ^{2}\right]  +\left(  E\left[  \ell^{\theta}\left(
Z_{i},\theta_{0}\right)  \right]  \right)  ^{2}+E\left[  U_{i}\left(
\theta_{0}\right)  W_{i}\left(  \theta_{0}\right)  \right]  \right)
\mathcal{I}^{-1}U\left(  \theta_{0}\right)  ,
\]%
\[
\sqrt{n}\left(  \overline{m}_{4}\left(  \widehat{\theta}\right)  -\overline
{m}_{4}\left(  \theta_{0}\right)  \right)  =2E\left[  U_{i}\left(  \theta
_{0}\right)  V_{i}\left(  \theta_{0}\right)  \right]  \mathcal{I}^{-1}U\left(
\theta_{0}\right)  +o_{p}\left(  1\right)  .
\]

\end{lemma}

\begin{proposition}
\label{Donsker}Assume that Conditions \ref{HH},\ref{M} and \ref{EC} hold. Let
$\mathfrak{F}$ be the class of measurable functions defined in Condition
\ref{EC}. Let $\rightsquigarrow$ denote weak convergence. Let $\left(
\Omega,\mathcal{F},P\right)  $ be a probability space such that $Z_{i}:$
$\left(  \Omega^{\mathbb{N}},\mathcal{F}^{\mathbb{N}},P^{\mathbb{N}}\right)
\rightarrow\left(  \Omega,\mathcal{F},P\right)  $ are coordinate projections.
Then, for $f\in\mathfrak{F}$ , $\sqrt{n}\left(  \widehat{F}-F\right)
f\rightsquigarrow Tf$ where $T$ is a tight Brownian bridge with variance
covariance function $F\left(  t\wedge s\right)  -F\left(  s\right)  F\left(
t\right)  $. Let $BL_{1}$ be the set of all function $h:l^{\infty}\left(
\mathfrak{F}\right)  \mapsto\left[  0,1\right]  $ such that $\left\vert
h(z_{1})-h(z_{2})\right\vert \leq\left\Vert z_{1}-z_{2}\right\Vert
_{\mathfrak{F}}$ for every $z_{1}$and $z_{2}$ where $l^{\infty}\left(
\mathfrak{F}\right)  $ is the set of uniformly bounded real functions on
$\mathfrak{F}$ and $\left\Vert .\right\Vert _{\mathfrak{F}}$ is the uniform
norm for maps from $\mathfrak{F}$ to $\mathbb{R}$. Then
\[
\sup_{h\in BL_{1}}\left\vert E^{\ast}h\left[  \sqrt{n}\left(  \widehat{F}%
^{\ast}-\widehat{F}\right)  f\right]  -Eh\left[  Tf\right]  \right\vert
\rightarrow0,\quad P^{\mathbb{N}}-a.s.
\]

\end{proposition}

\begin{lemma}
\label{P-rate}Assume that Condition \ref{HH} is satisfied. Suppose that, for
each $i$, $\xi_{i}^{\ast}\left(  \phi\right)  =\tau\left(  Z_{i}^{\ast}%
,\phi\right)  -\frac{1}{n}\sum_{i=1}^{n}\tau\left(  Z_{i},\phi\right)  $,
$i=\left\{  1,2,\ldots\right\}  $\ is a sequence of bootstrapped
transformations of random variables indexed by some parameter $\phi\in\Phi$
with $E^{\ast}\left[  \xi_{i}^{\ast}\left(  \phi\right)  \right]  =0$. We also
assume that $\sup_{\phi\in\Phi}\left\vert \tau\left(  Z_{i},\phi\right)
\right\vert \leq B_{i}$ for some sequence of random variables $B_{i}$ that is
i.i.d. Finally, we assume that $E\left[  \left\vert B_{i}\right\vert
^{16}\right]  <\infty$. We then have
\[
P^{\ast}\left[  \left\vert \frac{1}{\sqrt{n}}\sum_{i=1}^{n}\xi_{i}^{\ast
}\left(  \phi_{n}\right)  \right\vert >n^{\frac{1}{12}-\upsilon}\right]
=o_{p}\left(  n^{-1+16\upsilon}\right)
\]
for every $\upsilon$ such that $\upsilon<\frac{1}{16}$. Moreover,
\[
P^{\ast}\left[  \left\vert \frac{1}{n}\sum_{i=1}^{n}\xi_{i}^{\ast}\left(
\phi_{n}\right)  \right\vert >n^{\frac{1}{12}-\upsilon}\right]  =o_{p}\left(
n^{-\frac{23}{3}}\right)  .
\]
Here, $\phi_{n}$ is an arbitrary sequence in $\Phi$ and $P^{\ast}$ is the
conditional probability measure of $Z_{i}^{\ast}$ given $Z_{i}$.
\end{lemma}

\begin{lemma}
\label{BootstrapThetaEpsilonConv}Under Condition \ref{HH}, we have
\[
P^{\ast}\left[  \max_{0\leq\epsilon\leq\frac{1}{\sqrt{n}}}\left\vert
\widehat{\theta}^{\ast}\left(  \epsilon\right)  -\widehat{\theta}\right\vert
\geq\eta\right]  =o_{p}\left(  n^{-\frac{23}{3}}\right)  .
\]

\end{lemma}

\begin{lemma}
\label{multi:k-boot}Assume that Condition \ref{HH} is satisfied. Let $K\left(
\cdot;\theta\left(  \epsilon\right)  \right)  $ be defined as in Lemma
\ref{multi:k}. Then, for any $\eta>0$, we have
\[
P^{\ast}\left[  \max_{0\leq\epsilon\leq\frac{1}{\sqrt{n}}}\left\vert \int
K\left(  z;\theta^{\ast}\left(  \epsilon\right)  \right)  d\widehat{F}%
_{\epsilon}\left(  z\right)  -E\left[  K\left(  z;\theta_{0}\right)  \right]
\right\vert >\eta\right]  =o_{p}\left(  n^{-\frac{23}{3}}\right)  .
\]
Also,
\[
P^{\ast}\left[  \max_{0\leq\epsilon\leq\frac{1}{\sqrt{n}}}\left\vert \int
K\left(  \cdot;\widehat{\theta}^{\ast}\left(  \epsilon\right)  \right)
d\widehat{\Delta}\right\vert >Cn^{\frac{1}{12}-\upsilon}\right]  =o_{p}\left(
\max\left(  n^{-\frac{23}{3}},n^{-1+16\upsilon}\right)  \right)
\]
for some constant $C>0$ and for every $\upsilon$ such that $\upsilon<\frac
{1}{16}$.
\end{lemma}

\begin{lemma}
\label{boot:bound:theta}Suppose that Condition \ref{HH} holds. Then, we have
\begin{align*}
P^{\ast}\left[  \max_{0\leq\epsilon\leq\frac{1}{\sqrt{n}}}\left\vert
\widehat{\theta}^{\epsilon}\left(  \epsilon\right)  \right\vert >Cn^{\frac
{1}{12}-\upsilon}\right]   &  =o_{p}\left(  \max n^{-\frac{23}{3}%
},n^{-1+16\upsilon}\right) \\
P^{\ast}\left[  \max_{0\leq\epsilon\leq\frac{1}{\sqrt{n}}}\left\vert
\widehat{\theta}^{\epsilon\epsilon}\left(  \epsilon\right)  \right\vert
>C\left(  n^{\frac{1}{12}-\upsilon}\right)  ^{2}\right]   &  =o_{p}\left(
\max n^{-\frac{23}{3}},n^{-1+16\upsilon}\right) \\
&  \vdots\\
P^{\ast}\left[  \max_{0\leq\epsilon\leq\frac{1}{\sqrt{n}}}\left\vert
\widehat{\theta}^{\epsilon\epsilon\epsilon\epsilon\epsilon\epsilon}\left(
\epsilon\right)  \right\vert >C\left(  n^{\frac{1}{12}-\upsilon}\right)
^{6}\right]   &  =o_{p}\left(  \max n^{-\frac{23}{3}},n^{-1+16\upsilon
}\right)
\end{align*}
for some constant $C>0$ and for every $\upsilon$ such that $\upsilon<\frac
{1}{16}$.
\end{lemma}

\begin{lemma}
\label{sum-bound}Assume that $W_{i}$ are iid with $E\left[  W_{i}\right]  =0$
and $E\left[  W_{i}^{2k}\right]  <\infty.$ Then,
\[
E\left[  \left(
{\textstyle\sum\nolimits_{i=1}^{n}}
W_{i}\right)  ^{2k}\right]  =C(k)n^{k}+o(n^{k})
\]
for some constant $C(k)$.
\end{lemma}

\begin{proof}
By adopting an argument in the proof of Lemma 5.1 in Lahiri (1992), we have
\begin{equation}
E\left[  \left(
{\textstyle\sum\nolimits_{i=1}^{n}}
W_{i}\right)  ^{2k}\right]  =%
{\textstyle\sum_{j=1}^{2k}}
{\textstyle\sum_{\alpha}}
C(\alpha_{1},...,\alpha_{j})%
{\textstyle\sum_{\emph{I}}}
E\left[
{\textstyle\prod_{s=1}^{j}}
W_{i_{s}}^{\alpha_{s}}\right]  , \label{MomentBound}%
\end{equation}
where for each fixed $j\in\left\{  1,...,2k\right\}  ,$ $\sum_{\alpha}$
extends over all $j$-tuples of positive integers $\left(  \alpha
_{1},...,\alpha_{j}\right)  $ such that $\alpha_{1}+...+\alpha_{j}=2k$ and
$\sum_{\emph{I}}$ extends over all ordered $j$-tuples $\left(  i_{1}%
,...,i_{j}\right)  $ of integers such that $1\leq i_{j}\leq n.$ Also,
$C(\alpha_{1},...,\alpha_{j})$ stands for a bounded constant. Note, that if
$j>k$ then at least one of the indices $\alpha_{j}=1.$ By independence and the
fact that $EW_{i}=0$ it follows that $E%
{\textstyle\prod_{s=1}^{j}}
W_{i_{s}}^{\alpha_{s}}=0$ whenever $j>k.$ This shows that $E\left(
{\textstyle\sum\nolimits_{i=1}^{n}}
W_{i}\right)  ^{2k}=C(k)n^{k}+o(n^{k})$ for some constant $C(k).$
\end{proof}

\begin{lemma}
\label{modifiedHH:Lem1}Suppose that $\left\{  \xi_{i},i=1,2,\ldots\right\}  $
\ is a sequence of zero mean i.i.d. random variables. We also assume that
$E\left[  \left|  \xi_{i}\right|  ^{16}\right]  <\infty$. We then have
\[
\Pr\left[  \left|  \frac{1}{n}\sum_{i=1}^{n}\xi_{i}\right|  >\eta\right]
=O\left(  n^{-8}\right)
\]
for every $\eta>0$.
\end{lemma}

\begin{proof}
Using Lemma \ref{sum-bound}, we obtain
\[
E\left[  \left|  \sum_{i=1}^{n}\xi_{i}\right|  ^{16}\right]  \leq
Cn^{8}+o(n^{8}),
\]
where $C>0$ is a constant. Therefore, we have
\[
n^{8}\Pr\left[  \left|  \frac{1}{n}\sum_{i=1}^{n}\xi_{i}\right|  >\eta\right]
\leq O(n^{8}\frac{Cn^{8}}{n^{16}\eta^{16}})=O\left(  1\right)  .
\]

\end{proof}

\begin{lemma}
\label{modifiedHH:Lem3}Suppose that, for each $i$, $\left\{  \xi_{i}\left(
\phi\right)  ,i=1,2,\ldots\right\}  $ \ is a sequence of zero mean i.i.d.
random variables indexed by some parameter $\phi\in\Phi$. We also assume that
$\sup_{\phi\in\Phi}\left|  \xi_{i}\left(  \phi\right)  \right|  \leq B_{i}$
for some sequence of random variables $B_{i}$ that is i.i.d. Finally, we
assume that $E\left[  \left|  B_{i}\right|  ^{16}\right]  <\infty$. We then
have
\[
\Pr\left[  \left|  \frac{1}{\sqrt{n}}\sum_{i=1}^{n}\xi_{i}\left(  \phi
_{n}\right)  \right|  >n^{\frac{1}{12}-\upsilon}\right]  =o\left(
n^{-1+16v}\right)
\]
for every $\upsilon$ such that $\upsilon<\frac{1}{16}$. For $\upsilon<\frac
{1}{48}$ we have
\[
\Pr\left[  \left|  \frac{1}{\sqrt{n}}\sum_{i=1}^{n}\xi_{i}\left(  \phi
_{n}\right)  \right|  >n^{\frac{1}{12}-\upsilon}\right]  =o(n^{-1}).
\]
Here, $\phi_{n}$ is an arbitrary sequence in $\Phi$.
\end{lemma}

\begin{proof}
By Markov's inequality, we have
\begin{align*}
\Pr\left[  \sup_{\phi\in\Phi}\left|  \frac{1}{\sqrt{n}}\sum_{i=1}^{n}\xi
_{i}\left(  \phi\right)  \right|  >n^{\frac{1}{12}-\upsilon}\right]   &
=\Pr\left[  \sup_{\phi\in\Phi}\left|  \sum_{i=1}^{n}\xi_{i}\left(  \phi
_{n}\right)  \right|  >n^{\frac{7}{12}-\upsilon}\right] \\
&  \leq\frac{E\left[  \sup_{\phi\in\Phi}\left(  \sum_{i=1}^{n}\xi_{i}\left(
\phi\right)  \right)  ^{16}\right]  }{n^{\frac{28}{3}-16\upsilon}\eta^{16}}\\
&  =\frac{\sup_{\phi\in\Phi}E\left[  \left(  \sum_{i=1}^{n}\xi_{i}\left(
\phi\right)  \right)  ^{16}\right]  }{n^{\frac{28}{3}-16\upsilon}\eta^{16}},
\end{align*}
where the last equality is based on dominated convergence. By Lemma
\ref{sum-bound}, we have
\[
E\left[  \left(  \sum_{i=1}^{n}\xi_{i}\left(  \phi\right)  \right)
^{16}\right]  \leq Cn^{8},
\]
where $C>0$ is a constant. Therefore, we have
\[
\Pr\left[  \sup_{\phi\in\Phi}\left|  \frac{1}{\sqrt{n}}\sum_{i=1}^{n}\xi
_{i}\left(  \phi\right)  \right|  >n^{\frac{1}{12}-\upsilon}\right]  \leq
\frac{Cn^{8}}{n^{28/3-16\upsilon}\eta^{16}}=O\left(  n^{-4/3+16\upsilon
}\right)  .
\]

\end{proof}

\begin{lemma}
\label{Lem-uniform-convergence}Let $\widehat{G}\left(  \theta\right)
\equiv\frac{1}{n}\sum_{i=1}^{n}\log f\left(  Z_{i},\theta\right)  $. Suppose
that Condition \ref{HH} holds. We then have for all $\eta>0$ that
\[
\Pr\left[  \sup_{\theta}\left|  \widehat{G}\left(  \theta\right)  -G\left(
\theta\right)  \right|  \geq\eta\right]  =o\left(  n^{-\frac{23}{3}}\right)
\]

\end{lemma}

\begin{proof}
Note that
\[
\Pr\left[  \sup_{\theta}\left\vert \widehat{G}\left(  \theta\right)  -G\left(
\theta\right)  \right\vert \geq\eta\right]  =\Pr\left[  \sup_{\theta}\sqrt
{n}\left\vert \widehat{G}\left(  \theta\right)  -G\left(  \theta\right)
\right\vert \geq\eta n^{\frac{1}{12}-\upsilon}\right]
\]
where $\upsilon=-\frac{5}{12}$. Then the result follows by Lemma
\ref{modifiedHH:Lem3}.
\end{proof}

\begin{lemma}
\label{Vjack4}Let
\begin{align*}
W  &  =\left(  \frac{1}{\sqrt{n}}\sum_{i=1}^{n}X_{1,i}\right)  \left(
\frac{1}{\sqrt{n}}\sum_{i=1}^{n}X_{2,i}\right)  \left(  \frac{1}{\sqrt{n}}%
\sum_{i=1}^{n}X_{3,i}\right)  \left(  \frac{1}{\sqrt{n}}\sum_{i=1}^{n}%
X_{4,i}\right)  ,\\
W_{\left(  j\right)  }  &  =\left(  \frac{1}{\sqrt{n-1}}\sum_{i\neq j}%
X_{1,i}\right)  \left(  \frac{1}{\sqrt{n-1}}\sum_{i\neq j}X_{2,i}\right)
\left(  \frac{1}{\sqrt{n-1}}\sum_{i\neq j}X_{3,i}\right) \\
&  \times\left(  \frac{1}{\sqrt{n-1}}\sum_{i\neq j}X_{4,i}\right)
\end{align*}
Then,
\begin{align*}
&  nW-\frac{n^{2}}{n-1}\frac{1}{n}\sum_{j=1}^{n}W_{\left(  j\right)  }\\
&  =\frac{n\left(  n^{2}+3n-1\right)  }{\left(  n-1\right)  ^{3}}W\\
&  -\frac{n^{3}}{\left(  n-1\right)  ^{3}}\left(  \sqrt{n}\overline{X}%
_{1}\right)  \left(  \sqrt{n}\overline{X}_{2}\right)  \left(  \frac{1}{n}%
\sum_{i=1}^{n}X_{3,i}X_{4,i}\right)   \\
& -\frac{n^{3}}{\left(  n-1\right)  ^{3}%
}\left(  \sqrt{n}\overline{X}_{1}\right)  \left(  \sqrt{n}\overline{X}%
_{3}\right)  \left(  \frac{1}{n}\sum_{i=1}^{n}X_{2,i}X_{4,i}\right) \\
&  -\frac{n^{3}}{\left(  n-1\right)  ^{3}}\left(  \sqrt{n}\overline{X}%
_{1}\right)  \left(  \sqrt{n}\overline{X}_{4}\right)  \left(  \frac{1}{n}%
\sum_{i=1}^{n}X_{2,i}X_{3,i}\right)  \\
&  -\frac{n^{3}}{\left(  n-1\right)  ^{3}%
}\left(  \sqrt{n}\overline{X}_{2}\right)  \left(  \sqrt{n}\overline{X}%
_{3}\right)  \left(  \frac{1}{n}\sum_{i=1}^{n}X_{1,i}X_{4,i}\right) \\
&  -\frac{n^{3}}{\left(  n-1\right)  ^{3}}\left(  \sqrt{n}\overline{X}%
_{2}\right)  \left(  \sqrt{n}\overline{X}_{4}\right)  \left(  \frac{1}{n}%
\sum_{i=1}^{n}X_{1,i}X_{3,i}\right)  \\
&  -\frac{n^{3}}{\left(  n-1\right)  ^{3}%
}\left(  \sqrt{n}\overline{X}_{3}\right)  \left(  \sqrt{n}\overline{X}%
_{4}\right)  \left(  \frac{1}{n}\sum_{i=1}^{n}X_{1,i}X_{2,i}\right) \\
&  +\frac{n^{3}}{\sqrt{n}\left(  n-1\right)  ^{3}}\left(  \sqrt{n}\overline
{X}_{1}\right)  \left(  \frac{1}{n}\sum_{i=1}^{n}X_{2,i}X_{3,i}X_{4,i}\right)  \\
& 
+\frac{n^{3}}{\sqrt{n}\left(  n-1\right)  ^{3}}\left(  \sqrt{n}\overline
{X}_{2}\right)  \left(  \frac{1}{n}\sum_{i=1}^{n}X_{1,i}X_{3,i}X_{4,i}\right)
\\
&  +\frac{n^{3}}{\sqrt{n}\left(  n-1\right)  ^{3}}\left(  \sqrt{n}\overline
{X}_{3}\right)  \left(  \frac{1}{n}\sum_{i=1}^{n}X_{1,i}X_{2,i}X_{4,i}\right)
\\
& +\frac{n^{3}}{\sqrt{n}\left(  n-1\right)  ^{3}}\left(  \sqrt{n}\overline
{X}_{4}\right)  \left(  \frac{1}{n}\sum_{i=1}^{n}X_{1,i}X_{2,i}X_{3,i}\right)
\\
&  -\frac{n^{2}}{\left(  n-1\right)  ^{3}}\left(  \frac{1}{n}\sum_{i=1}%
^{n}X_{1,i}X_{2,i}X_{3,i}X_{4,i}\right)
\end{align*}

\end{lemma}

\begin{lemma}
\label{Vjack5}Let
\begin{align*}
W  &  =\left(  \frac{1}{\sqrt{n}}\sum_{i=1}^{n}X_{1,i}\right)  \left(
\frac{1}{\sqrt{n}}\sum_{i=1}^{n}X_{2,i}\right)  \left(  \frac{1}{\sqrt{n}}%
\sum_{i=1}^{n}X_{3,i}\right)  \left(  \frac{1}{\sqrt{n}}\sum_{i=1}^{n}%
X_{4,i}\right)  \left(  \frac{1}{\sqrt{n}}\sum_{i=1}^{n}X_{5,i}\right)  ,\\
W_{\left(  j\right)  }  &  =\left(  \frac{1}{\sqrt{n-1}}\sum_{i\neq j}%
X_{1,i}\right)  \left(  \frac{1}{\sqrt{n-1}}\sum_{i\neq j}X_{2,i}\right)
\left(  \frac{1}{\sqrt{n-1}}\sum_{i\neq j}X_{3,i}\right) \\
&  \times\left(  \frac{1}{\sqrt{n-1}}\sum_{i\neq j}X_{4,i}\right)  \left(
\frac{1}{\sqrt{n}}\sum_{i=1}^{n}X_{5,i}\right)
\end{align*}
Then, $\allowbreak$%
\begin{align*}
&  W-\frac{n\sqrt{n}}{\left(  n-1\right)  \sqrt{n-1}}\frac{1}{n}\sum_{j=1}%
^{n}W_{\left(  j\right)  }\\
&  =\allowbreak\frac{n^{3}+6n^{2}-4n+1}{\left(  n-1\right)  ^{4}}\left(
\sqrt{n}\overline{X}_{1}\right)  \left(  \sqrt{n}\overline{X}_{2}\right)
\left(  \sqrt{n}\overline{X}_{3}\right)  \left(  \sqrt{n}\overline{X}%
_{4}\right)  \left(  \sqrt{n}\overline{X}_{5}\right)
\end{align*}

\begin{align*}
&  -\frac{n^{2}}{\left(  n-1\right)  ^{4}}\left(  \sqrt{n}\overline{X}%
_{1}\right)  \left(  \sqrt{n}\overline{X}_{2}\right)  \left(  \sqrt
{n}\overline{X}_{5}\right)  \left(  \frac{1}{n}\sum_{j=1}^{n}X_{3,j}%
X_{4,j}\right) \\
&  -\frac{n^{2}}{\left(  n-1\right)  ^{4}}\left(  \sqrt{n}\overline{X}%
_{1}\right)  \left(  \sqrt{n}\overline{X}_{2}\right)  \left(  \sqrt
{n}\overline{X}_{3}\right)  \left(  \frac{1}{n}\sum_{j=1}^{n}X_{4,j}%
X_{5,j}\right) \\
&  -\frac{n^{2}}{\left(  n-1\right)  ^{4}}\left(  \sqrt{n}\overline{X}%
_{1}\right)  \left(  \sqrt{n}\overline{X}_{3}\right)  \left(  \sqrt
{n}\overline{X}_{5}\right)  \left(  \frac{1}{n}\sum_{j=1}^{n}X_{2,j}%
X_{4,j}\right) \\
&  -\frac{n^{2}}{\left(  n-1\right)  ^{4}}\left(  \sqrt{n}\overline{X}%
_{1}\right)  \left(  \sqrt{n}\overline{X}_{4}\right)  \left(  \sqrt
{n}\overline{X}_{5}\right)  \left(  \frac{1}{n}\sum_{j=1}^{n}X_{2,j}%
X_{3,j}\right) \\
&  -\frac{n^{2}}{\left(  n-1\right)  ^{4}}\left(  \sqrt{n}\overline{X}%
_{1}\right)  \left(  \sqrt{n}\overline{X}_{3}\right)  \left(  \sqrt
{n}\overline{X}_{4}\right)  \left(  \frac{1}{n}\sum_{j=1}^{n}X_{2,j}%
X_{5,j}\right) \\
&  -\frac{n^{2}}{\left(  n-1\right)  ^{4}}\left(  \sqrt{n}\overline{X}%
_{3}\right)  \left(  \sqrt{n}\overline{X}_{4}\right)  \left(  \sqrt
{n}\overline{X}_{5}\right)  \left(  \frac{1}{n}\sum_{j=1}^{n}X_{1,j}%
X_{2,j}\right) \\
&  -\frac{n^{2}}{\left(  n-1\right)  ^{4}}\left(  \sqrt{n}\overline{X}%
_{2}\right)  \left(  \sqrt{n}\overline{X}_{3}\right)  \left(  \sqrt
{n}\overline{X}_{5}\right)  \left(  \frac{1}{n}\sum_{j=1}^{n}X_{1,j}%
X_{4,j}\right) \\
&  -\frac{n^{2}}{\left(  n-1\right)  ^{4}}\left(  \sqrt{n}\overline{X}%
_{2}\right)  \left(  \sqrt{n}\overline{X}_{4}\right)  \left(  \sqrt
{n}\overline{X}_{5}\right)  \left(  \frac{1}{n}\sum_{j=1}^{n}X_{1,j}%
X_{3,j}\right) \\
&  -\frac{n^{2}}{\left(  n-1\right)  ^{4}}\left(  \sqrt{n}\overline{X}%
_{2}\right)  \left(  \sqrt{n}\overline{X}_{3}\right)  \left(  \sqrt
{n}\overline{X}_{4}\right)  \left(  \frac{1}{n}\sum_{j=1}^{n}X_{1,j}%
X_{5,j}\right) \\
&  -\frac{n^{2}}{\left(  n-1\right)  ^{4}}\left(  \sqrt{n}\overline{X}%
_{1}\right)  \left(  \sqrt{n}\overline{X}_{2}\right)  \left(  \sqrt
{n}\overline{X}_{4}\right)  \left(  \frac{1}{n}\sum_{j=1}^{n}X_{3,j}%
X_{5,j}\right)
\end{align*}

\begin{align*}
&  +\frac{n\sqrt{n}}{\left(  n-1\right)  ^{4}}\left(  \sqrt{n}\overline{X}%
_{3}\right)  \left(  \sqrt{n}\overline{X}_{5}\right)  \left(  \frac{1}{n}%
\sum_{j=1}^{n}X_{1,j}X_{2,j}X_{4,j}\right) \\
&  +\frac{n\sqrt{n}}{\left(  n-1\right)  ^{4}}\left(  \sqrt{n}\overline{X}%
_{1}\right)  \left(  \sqrt{n}\overline{X}_{5}\right)  \left(  \frac{1}{n}%
\sum_{j=1}^{n}X_{2,j}X_{3,j}X_{4,j}\right) \\
&  +\frac{n\sqrt{n}}{\left(  n-1\right)  ^{4}}\left(  \sqrt{n}\overline{X}%
_{1}\right)  \left(  \sqrt{n}\overline{X}_{3}\right)  \left(  \frac{1}{n}%
\sum_{j=1}^{n}X_{2,j}X_{4,j}X_{5,j}\right) \\
&  +\frac{n\sqrt{n}}{\left(  n-1\right)  ^{4}}\left(  \sqrt{n}\overline{X}%
_{1}\right)  \left(  \sqrt{n}\overline{X}_{4}\right)  \left(  \frac{1}{n}%
\sum_{j=1}^{n}X_{2,j}X_{3,j}X_{5,j}\right) \\
&  +\frac{n\sqrt{n}}{\left(  n-1\right)  ^{4}}\left(  \sqrt{n}\overline{X}%
_{1}\right)  \left(  \sqrt{n}\overline{X}_{2}\right)  \left(  \frac{1}{n}%
\sum_{j=1}^{n}X_{3,j}X_{4,j}X_{5,j}\right) \\
&  +\frac{n\sqrt{n}}{\left(  n-1\right)  ^{4}}\left(  \sqrt{n}\overline{X}%
_{2}\right)  \left(  \sqrt{n}\overline{X}_{5}\right)  \left(  \frac{1}{n}%
\sum_{j=1}^{n}X_{1,j}X_{3,j}X_{4,j}\right) \\
&  +\frac{n\sqrt{n}}{\left(  n-1\right)  ^{4}}\left(  \sqrt{n}\overline{X}%
_{2}\right)  \left(  \sqrt{n}\overline{X}_{3}\right)  \left(  \frac{1}{n}%
\sum_{j=1}^{n}X_{1,j}X_{4,j}X_{5,j}\right) \\
&  +\frac{n\sqrt{n}}{\left(  n-1\right)  ^{4}}\left(  \sqrt{n}\overline{X}%
_{2}\right)  \left(  \sqrt{n}\overline{X}_{4}\right)  \left(  \frac{1}{n}%
\sum_{j=1}^{n}X_{1,j}X_{3,j}X_{5,j}\right) \\
&  +\frac{n\sqrt{n}}{\left(  n-1\right)  ^{4}}\left(  \sqrt{n}\overline{X}%
_{4}\right)  \left(  \sqrt{n}\overline{X}_{5}\right)  \left(  \frac{1}{n}%
\sum_{j=1}^{n}X_{1,j}X_{2,j}X_{3,j}\right) \\
&  +\frac{n\sqrt{n}}{\left(  n-1\right)  ^{4}}\left(  \sqrt{n}\overline{X}%
_{3}\right)  \left(  \sqrt{n}\overline{X}_{4}\right)  \left(  \frac{1}{n}%
\sum_{j=1}^{n}X_{1,j}X_{2,j}X_{5,j}\right)
\end{align*}

\begin{align*}
&  -\frac{n}{\left(  n-1\right)  ^{4}}\left(  \sqrt{n}\overline{X}_{1}\right)
\left(  \frac{1}{n}\sum_{j=1}^{n}X_{2,j}X_{3,j}X_{4,j}X_{5,j}\right) \\
&  -\frac{n}{\left(  n-1\right)  ^{4}}\left(  \sqrt{n}\overline{X}_{2}\right)
\left(  \frac{1}{n}\sum_{j=1}^{n}X_{1,j}X_{3,j}X_{4,j}X_{5,j}\right) \\
&  -\frac{n}{\left(  n-1\right)  ^{4}}\left(  \sqrt{n}\overline{X}_{3}\right)
\left(  \frac{1}{n}\sum_{j=1}^{n}X_{1,j}X_{2,j}X_{4,j}X_{5,j}\right) \\
&  -\frac{n}{\left(  n-1\right)  ^{4}}\left(  \sqrt{n}\overline{X}_{4}\right)
\left(  \frac{1}{n}\sum_{j=1}^{n}X_{1,j}X_{2,j}X_{3,j}X_{5,j}\right) \\
&  -\frac{n}{\left(  n-1\right)  ^{4}}\left(  \sqrt{n}\overline{X}_{5}\right)
\left(  \frac{1}{n}\sum_{j=1}^{n}X_{1,j}X_{2,j}X_{3,j}X_{4,j}\right) \\
&  +\frac{n\sqrt{n}}{\left(  n-1\right)  ^{4}}\frac{1}{n}\sum_{j=1}^{n}%
X_{1,j}X_{2,j}X_{3,j}X_{4,j}X_{5,j}%
\end{align*}

\end{lemma}

\section{Proofs for Section \ref{sec-lfbs}}

\subsection{Proof of Lemma \ref{Lem-consistency-of-theta-epsilon}}

Let $\eta$ be given, and let $\varepsilon\equiv G\left(  \theta_{0}\right)
-\sup_{\left\{  \theta:\left\vert \theta-\theta_{0}\right\vert >\eta\right\}
}G\left(  \theta\right)  >0$. Letting $g\left(  z,\theta\right)  \equiv\log
f\left(  z,\theta\right)  $, we have
\[
\int g\left(  z,\theta\right)  dF_{\epsilon}\left(  z\right)  =\left(
1-\epsilon\sqrt{n}\right)  G\left(  \theta\right)  +\epsilon\sqrt
{n}\widehat{G}\left(  \theta\right)
\]
and
\[
\left\vert \int g\left(  z,\theta\right)  dF_{\epsilon}\left(  z\right)
-G\left(  \theta\right)  \right\vert \leq\left(  1-\epsilon\sqrt{n}\right)
\left\vert \widehat{G}\left(  \theta\right)  -G\left(  \theta\right)
\right\vert \leq\left\vert \widehat{G}\left(  \theta\right)  -G\left(
\theta\right)  \right\vert .
\]
Here, the last inequality is based on the fact that $0\leq\epsilon\leq\frac
{1}{\sqrt{n}}$. By Lemma \ref{Lem-uniform-convergence}, we have
\[
\Pr\left[  \max_{0\leq\epsilon\leq\frac{1}{\sqrt{n}}}\sup_{\theta}\left\vert
\int g\left(  z,\theta\right)  dF_{\epsilon}\left(  z\right)  -G\left(
\theta\right)  \right\vert \geq\eta\right]  =o\left(  n^{-\frac{23}{3}%
}\right)
\]
Therefore, for every $0\leq\epsilon\leq\frac{1}{\sqrt{n}}$ with probability
equal to $1-o\left(  n^{-\frac{23}{3}}\right)  $, we have
\begin{align*}
\max_{\left\vert \theta-\theta_{0}\right\vert >\eta}\int g\left(
z,\theta\right)  dF_{\epsilon}\left(  z\right)   &  \leq\max_{\left\vert
\theta-\theta_{0}\right\vert >\eta}G\left(  \theta\right)  +\frac{1}%
{3}\varepsilon\\
&  <G\left(  \theta_{0}\right)  -\frac{2}{3}\varepsilon\\
&  <\int g\left(  z,\theta_{0}\right)  dF_{\epsilon}\left(  z\right)
-\frac{1}{3}\varepsilon.
\end{align*}
We also have
\[
\max_{\theta}\int g\left(  z,\theta\right)  dF_{\epsilon}\left(  z\right)
\geq\int g\left(  z,\theta_{0}\right)  dF_{\epsilon}\left(  z\right)
\]
by definition. It follows that
\[
\max_{\left\vert \theta-\theta_{0}\right\vert >\eta}\int g\left(
z,\theta\right)  dF_{\epsilon}\left(  z\right)  <\max_{\theta}\int g\left(
z,\theta\right)  dF_{\epsilon}\left(  z\right)  -\frac{1}{3}\varepsilon
\]
for every $0\leq\epsilon\leq\frac{1}{\sqrt{n}}$. We therefore obtain that
$\Pr\left[  \max_{0\leq\epsilon\leq\frac{1}{\sqrt{n}}}\left\vert \theta\left(
\epsilon\right)  -\theta_{0}\right\vert \geq\eta\right]  =o\left(
n^{-\frac{23}{3}}\right)  $.

\subsection{Proof of Lemma \ref{multi:k}}

Note that we may write
\begin{align*}
&  \int K\left(  z;\theta\left(  \epsilon\right)  \right)  dF_{\epsilon
}\left(  z\right)  -E\left[  K\left(  Z_{i};\theta_{0}\right)  \right] \\
&  =\int K\left(  z;\theta\left(  \epsilon\right)  \right)  dF_{\epsilon
}\left(  z\right)  -\int K\left(  z;\theta_{0}\right)  dF\left(  z\right) \\
&  =\int K\left(  z;\theta\left(  \epsilon\right)  \right)  dF_{\epsilon
}\left(  z\right)  -\int K\left(  z;\theta_{0}\right)  dF_{\epsilon}\left(
z\right)  +\int K\left(  z;\theta_{0}\right)  dF_{\epsilon}\left(  z\right)
-\int K\left(  z;\theta\left(  \epsilon\right)  \right)  dF\left(  z\right) \\
&  =\int\frac{\partial K\left(  z;\theta^{\ast}\right)  }{\partial\theta
}\left(  \theta\left(  \epsilon\right)  -\theta_{0}\right)  dF_{\epsilon
}\left(  z\right)  +\epsilon\sqrt{n}\int K\left(  z;\theta_{0}\right)
d\left(  \widehat{F}-F\right)  \left(  z\right)
\end{align*}
where $\theta^{\ast}$ is between $\theta_{0}$ and $\theta\left(
\epsilon\right)  $. Therefore, we have
\begin{align*}
\left\vert \int K\left(  z;\theta\left(  \epsilon\right)  \right)
dF_{\epsilon}\left(  z\right)  -E\left[  K\left(  Z_{i};\theta_{0}\right)
\right]  \right\vert  &  \leq\left\vert \theta\left(  \epsilon\right)
-\theta_{0}\right\vert \cdot\left(  E\left[  M\left(  Z_{i}\right)  \right]
+\frac{1}{n}\sum_{i=1}^{n}M\left(  Z_{i}\right)  \right) \\
&  +\left\vert \frac{1}{n}\sum_{i=1}^{n}\left(  M\left(  Z_{i}\right)
-E\left[  M\left(  Z_{i}\right)  \right]  \right)  \right\vert
\end{align*}
where $M\left(  \cdot\right)  $ is defined in Condition \ref{HH}. Using Lemma
\ref{Lem-consistency-of-theta-epsilon}, we can bound
\[
\max_{0\leq\epsilon\leq\frac{1}{\sqrt{T}}}\left\vert \int K\left(
z;\theta\left(  \epsilon\right)  \right)  dF_{\epsilon}\left(  z\right)
-E\left[  K\left(  Z_{i};\theta_{0}\right)  \right]  \right\vert
\]
in absolute value by some $\eta>0$ with probability $1-o\left(  n^{-\frac
{23}{3}}\right)  $.

Using Condition \ref{HH} and Lemmas \ref{modifiedHH:Lem3}, we can also show
that $\left\vert \int K\left(  \cdot;\theta\left(  \epsilon\right)  \right)
d\Delta\right\vert $ can be bounded by $Cn^{\frac{1}{12}-\upsilon}$ for some
constant $C>0$ and $\upsilon$ such that $\upsilon<\frac{1}{16}$ with
probability $1-o\left(  n^{-1+16\upsilon}\right)  $. Similarly, if
$\upsilon<\frac{1}{48}$, then the statement holds with probability
$o(n^{-1}).$

\subsection{Proof of Lemma \ref{multi:bound:theta-epsilon-1}}

From (\ref{multi:alt-h1}), we have
\[
\theta^{\epsilon}\left(  \epsilon\right)  =-\left[  \int\ell^{\theta}\left(
z,\epsilon\right)  dF_{\epsilon}\left(  z\right)  \right]  ^{-1}\left[
\int\ell\left(  \cdot,\epsilon\right)  d\Delta\right]
\]
Using Lemma \ref{multi:k}, we can bound the denominator by some $C>0$, and the
numerator by some $Cn^{\frac{1}{12}-\upsilon}$ with probability $1-o\left(
n^{-1+16\upsilon}\right)  $, from which the first conclusion follows. As for
the second conclusion, we note from (\ref{multi:alt-h2}) that we have
\[
0=E_{\epsilon}\left[  \ell^{\theta\theta}\left(  Z_{i},\epsilon\right)
\right]  \left(  \theta^{\epsilon}\left(  \epsilon\right)  \right)
^{2}+E_{\epsilon}\left[  \ell^{\theta}\left(  Z_{i},\epsilon\right)  \right]
\theta^{\epsilon\epsilon}\left(  \epsilon\right)  +2\left(  \int\ell^{\theta
}\left(  z,\epsilon\right)  d\Delta\left(  z\right)  \right)  \theta
^{\epsilon}\left(  \epsilon\right)
\]
The second conclusion follows by using Lemmas \ref{multi:k} along with the
first conclusion. The rest of the Lemmas can be established similarly. Note
that if $\upsilon<\frac{1}{48}$ then we can apply the specialized result of
Lemma \ref{multi:k} in the same way as before.

\subsection{Proof of Lemma \ref{equiLemma}}

Let $\overline{m}_{0}\left(  \theta\right)  \equiv\int\ell^{\theta}\left(
z,\theta\right)  f\left(  z,\theta_{0}\right)  dz$. Note that
\begin{align*}
\widehat{\mathcal{I}}-\mathcal{I}  &  =-n^{-1}%
{\textstyle\sum_{i=1}^{n}}
\ell^{\theta}\left(  Z_{i},\widehat{\theta}\right)  +E\left[  \ell^{\theta
}\left(  Z_{i},\theta_{0}\right)  \right] \\
&  =-n^{-1}%
{\textstyle\sum_{i=1}^{n}}
\left(  \ell^{\theta}\left(  Z_{i},\theta_{0}\right)  -\overline{m}_{0}\left(
\theta_{0}\right)  \right)  +o_{p}\left(  n^{-1/2}\right)  -\left(
\overline{m}_{0}\left(  \widehat{\theta}\right)  -\overline{m}_{0}\left(
\theta_{0}\right)  \right)  ,
\end{align*}
where the last equality is based on the usual stochastic equicontinuity. Also
note that $\left.  \partial\overline{m}_{0}\left(  \theta\right)  \right/
\partial\theta=\int\ell^{\theta\theta}\left(  z,\theta\right)  f\left(
z,\theta_{0}\right)  dz$ by dominated convergence. We therefore obtain
\begin{align*}
\sqrt{n}\left(  \widehat{\mathcal{I}}-\mathcal{I}\right)   &  =-n^{-1/2}%
{\textstyle\sum_{i=1}^{n}}
\left(  \ell^{\theta}\left(  Z_{i},\theta_{0}\right)  -E\left[  \ell^{\theta
}\left(  Z_{i},\theta_{0}\right)  \right]  \right)  \\
&-E\left[  \ell
^{\theta\theta}\left(  Z_{i},\theta_{0}\right)  \right]  \sqrt{n}\left(
\widehat{\theta}-\theta_{0}\right)  +o_{p}\left(  1\right) \\
&  =-V\left(  \theta_{0}\right)  -\mathcal{Q}_{1}\left(  \theta_{0}\right)
\mathcal{I}^{-1}U\left(  \theta_{0}\right)  +o_{p}\left(  1\right)  ,
\end{align*}
Likewise, we obtain
\begin{align*}
\sqrt{n}\left(  \widehat{\mathcal{Q}}_{1}\left(  \widehat{\theta}\right)
-\mathcal{Q}_{1}\left(  \theta_{0}\right)  \right)   &  =n^{-1/2}%
{\textstyle\sum_{i=1}^{n}}
\left(  \ell^{\theta\theta}\left(  Z_{i},\theta_{0}\right)  -E\left[
\ell^{\theta\theta}\left(  Z_{i},\theta_{0}\right)  \right]  \right) \\
&  +E\left[  \ell^{\theta\theta\theta}\left(  Z_{i},\theta_{0}\right)
\right]  \sqrt{n}\left(  \widehat{\theta}-\theta_{0}\right)  +o_{p}\left(
1\right) \\
&  =W\left(  \theta_{0}\right)  +\mathcal{Q}_{2}\left(  \theta_{0}\right)
\mathcal{I}^{-1}U\left(  \theta_{0}\right)  +o_{p}\left(  1\right)  ,
\end{align*}%
\begin{align*}
\sqrt{n}\left(  \overline{m}_{1}\left(  \widehat{\theta}\right)  -\overline
{m}_{1}\left(  \theta_{0}\right)  \right)   &  =-E\left[  \ell^{\theta\theta
}\left(  Z_{i},\theta_{0}\right)  \right]  \sqrt{n}\left(  \widehat{\theta
}-\theta_{0}\right)  +o_{p}\left(  1\right) \\
&  =-\mathcal{Q}_{1}\left(  \theta_{0}\right)  \mathcal{I}^{-1}U\left(
\theta_{0}\right)  +o_{p}\left(  1\right)  ,
\end{align*}%
\begin{align*}
\sqrt{n}\left(  \overline{m}_{3}\left(  \widehat{\theta}\right)  -\overline
{m}_{3}\left(  \theta_{0}\right)  \right)   &  =\left(  E\left[  \ell^{\theta
}\left(  Z_{i},\theta_{0}\right)  ^{2}\right]  +E\left[  \ell\left(
Z_{i},\theta_{0}\right)  \ell^{\theta\theta}\left(  Z_{i},\theta_{0}\right)
\right]  \right)  \sqrt{n}\left(  \widehat{\theta}-\theta_{0}\right)
\\ &+o_{p}\left(  1\right) \\
&  =\Big(  E\left[  V_{i}\left(  \theta_{0}\right)  ^{2}\right]  +\left(
E\left[  \ell^{\theta}\left(  Z_{i},\theta_{0}\right)  \right]  \right)
^{2} \\
&\quad +E\left[  U_{i}\left(  \theta_{0}\right)  W_{i}\left(  \theta_{0}\right)
\right]  \Big)  \mathcal{I}^{-1}U\left(  \theta_{0}\right) + o_p(1)
\end{align*}

\subsection{Proof of Proposition \ref{Donsker}}

We first show that for $f\in\mathfrak{F}$ , $\sqrt{n}\left(  \widehat{F}%
-F\right)  f\rightsquigarrow Tf$ or in other words that $\mathfrak{F}$ is a
Donsker class. Define $\mathfrak{F}_{\delta}=\left\{  f-g:f,g\in
\mathfrak{F,}E\left[  \left\Vert f-g\right\Vert ^{2}\right]  <\delta\right\}
$, $\mathfrak{F}_{\infty}=\left\{  f-g:f,g\in\mathfrak{F}\right\}  $ and
$\mathfrak{F}_{\infty}^{2}=\left\{  f^{2}:f\in\mathfrak{F}_{\infty}\right\}
.$ In light of van der Vaart and Wellner (1996, Theorem 2.5.2), it is enough
to show that $\mathfrak{F}_{\delta}$ and $\mathfrak{F}_{\infty}^{2}$ are $F $
measurable classes for every $\delta>0$ and $E\left[  M\left(  z\right)
^{2}\right]  <\infty$. The second requirement is satisfied by Condition
\ref{HH}. Since $\mathfrak{F}_{\delta}\subset\mathfrak{F}_{\infty}$ the first
condition holds if for $f\in\mathfrak{F}_{\infty}^{2}$ and any vector
$a\in\mathbb{R}^{n}$ and any $n$ the function $s(Z_{1},...,Z_{n})=\sup
_{\theta_{1},\theta_{2}\in\Theta}\left\vert \sum_{i}^{n}a_{i}\left(
\ell^{\left(  k\right)  }(Z_{i},\theta_{1})-\ell^{\left(  k\right)  }%
(Z_{i},\theta_{2})\right)  ^{2}\right\vert $ is measurable. Let $\Theta_{k}$
be an increasing sequence of countable subsets of $\Theta$ whose limit is
dense in $\Theta.$ Then
\[
s_{k}(Z_{1},...,Z_{n})=\sup_{\theta_{1},\theta_{2}\in\Theta_{k}}\left\vert
\sum_{i}^{n}a_{i}\left(  \ell^{\left(  k\right)  }(Z_{i},\theta_{1}%
)-\ell^{\left(  k\right)  }(Z_{i},\theta_{2})\right)  ^{2}\right\vert
\]
is measurable by Condition \ref{M}. By continuity of $\ell^{\left(  k\right)
}(Z_{i},\theta)$ in $\theta$ it follows that
\[
\lim\inf_{k}s_{k}(Z_{1},...,Z_{n})=s(Z_{1},...,Z_{n})
\]
such that measurability of $s$ follows from Royden (1988, Theorem 20, p.68).
Conditional weak convergence of $\widehat{\Delta}$ follows from Gine and Zinn
(1990, Theorem 2.4). Note that by measurability of $\sqrt{n}\left(  \hat
{F}-F\right)  f$ and Gine and Zinn (1990, p861) the convergence of $\sup_{h\in
BL_{1}}\left\vert E^{\ast}h\left[  \sqrt{n}\left(  \widehat{F}^{\ast
}-\widehat{F}\right)  f\right]  -Eh\left[  Tf\right]  \right\vert $ is a.s.

\subsection{Proof of Lemma \ref{P-rate}}

Note that $\sum_{i=1}^{n}\xi_{i}^{\ast}(\phi)=\sum_{i=1}^{n}\left(
N_{ni}-1\right)  \tau(Z_{i},\phi)$ where $N_{n1},...,N_{nn}$ is multinomially
distributed with parameters $\left(  n,1/n,...,1/n\right)  =\left(
k,p_{1},...,p_{n}\right)  $ and independent of $Z_{i}$ such that $\Pr\left(
\cap_{i=1}^{n}\left\{  N_{ni}=n_{i}\right\}  \right)  =n!/\left(  \prod
_{i}n_{i}!\right)  \prod_{i}n^{-n_{i}}$ where $\sum_{i}^{n}n_{i}=n,$
$n_{i}\geq0.$ Let $\kappa_{r_{1}r_{2}....r_{n}}$ be the mixed higher order
cumulant of $N_{n1},...,N_{nn}$ of order $r=r_{1}+...+r_{n}$ for $r_{i}\geq0,$
$r_{i}$ integer. Mixed higher order cumulants can be obtained from Guldberg's
(1935) recurrence relation $\kappa_{r_{1}r_{2}..r_{i}+1...r_{n}}=a_{i}%
\partial\left(  \kappa_{r_{1}r_{2}..r_{i}...r_{n}}\right)  /\partial a_{i}$
where $a_{i}=p_{i}/p_{1} $. Let $b$ be the number of non zero indices $r_{i}.$
The arguments in Wishart (1949) imply that for $p_{i}=n^{-1}$ we have
$\kappa_{r_{1}r_{2}....r_{n}}\leq cn^{-b+1}$ for some constant $c.$ For
notational convenience we will represent cumulants with zero indices as lower
order cumulants of the variables with non-zero indices, i.e. write
$\kappa_{...r_{i\neq j}..}=\kappa_{r_{1}r_{2}....r_{n}}$ where $r_{j}=0.$

Consider
\begin{align*}
P^{\ast}\left[  \sup_{\phi\in\Phi}\left\vert \frac{1}{\sqrt{n}}\sum_{i=1}%
^{n}\xi_{i}^{\ast}\left(  \phi\right)  \right\vert >n^{\frac{1}{12}-\upsilon
}\right]   &  =P^{\ast}\left[  \sup_{\phi\in\Phi}\left\vert \sum_{i=1}^{n}%
\xi_{i}^{\ast}\left(  \phi_{n}\right)  \right\vert >n^{\frac{7}{12}-\upsilon
}\right] \\
&  \leq\frac{E^{\ast}\left[  \sup_{\phi\in\Phi}\left(  \sum_{i=1}^{n}\xi
_{i}^{\ast}\left(  \phi\right)  \right)  ^{16}\right]  }{n^{\frac{28}%
{3}-16\upsilon}\eta^{16}}\\
&  =\frac{\sup_{\phi\in\Phi}E^{\ast}\left[  \left(  \sum_{i=1}^{n}\xi
_{i}^{\ast}\left(  \phi\right)  \right)  ^{16}\right]  }{n^{\frac{28}%
{3}-16\upsilon}\eta^{16}},
\end{align*}
where the last equality uses the fact that $\sup_{\phi\in\Phi}$ does not
involve $N_{n1},...,N_{nn}.$ By adopting an argument in the proof of Lemma 5.1
in Lahiri (1992), we have
\begin{equation}
E^{\ast}\left(
{\textstyle\sum\nolimits_{i=1}^{n}}
\xi_{i}^{\ast}\left(  \phi\right)  \right)  ^{2k}=%
{\textstyle\sum_{j=1}^{2k}}
{\textstyle\sum_{\alpha}}
C(\alpha_{1},...,\alpha_{j})%
{\textstyle\sum_{\emph{I}}}
{\textstyle\prod_{t=1}^{j}}
\tau(Z_{i_{t}},\phi)^{\alpha_{t}}E^{\ast}%
{\textstyle\prod_{s=1}^{j}}
\left(  N_{ni_{s}}-1\right)  ^{\alpha_{s}}, \label{MomentBound1}%
\end{equation}
where for each fixed $j\in\left\{  1,...,2k\right\}  ,$ $\sum_{\alpha}$
extends over all $j$-tuples of positive integers $\left(  \alpha
_{1},...,\alpha_{j}\right)  $ such that $\alpha_{1}+...+\alpha_{j}=2k$ and
$\sum_{\emph{I}}$ extends over all ordered $j$-tuples $\left(  i_{1}%
,...,i_{j}\right)  $ of integers such that $1\leq i_{j}\leq n.$ Also,
$C(\alpha_{1},...,\alpha_{j})$ stands for a bounded constant. Next we consider
the mixed central moments $\mu(\alpha_{1},...,\alpha_{j})=E^{\ast}\prod
_{s=1}^{j}\left(  N_{ni_{s}}-1\right)  ^{\alpha_{s}}.$ From Shiryaev (1989,
Theorem 6, p.290) we obtain a relationship between cumulants and mixed
moments. Let $\alpha=\left(  \alpha_{1},...,\alpha_{j}\right)  ^{\prime
},r^{(p)}=\left(  r_{1}^{(p)},...,r_{j}^{(p)}\right)  ,\left\vert
r^{(p)}\right\vert =r_{1}^{(p)}+...+r_{j}^{(p)}$ and $r^{(p)}!=r_{1}%
^{(p)}!...r_{j}^{(p)}!$ such that
\[
\mu(\alpha_{1},...,\alpha_{j})=\sum_{r^{(1)}+...+r^{(q)}=\alpha}\frac{1}%
{q!}\frac{\alpha!}{r^{(1)}!...r^{(q)}!}\prod_{p=1}^{q}\kappa_{r_{1}^{(p)}%
r_{2}^{(p)}....r_{j}^{(p)}}%
\]
where $\sum_{r^{(1)}+...+r^{(q)}=\alpha}$ indicates the sum over all ordered
sets of nonnegative integral vectors $r^{(p)},\left\vert r^{(p)}\right\vert
>0,$whose sum is $\alpha.$ Since the order of \ref{MomentBound1} depends both
on the number of nonzero terms in $\sum_{\emph{I}}$ and the size of
$\mu(\alpha_{1},...,\alpha_{j})$ for each $j,$ we analyze the term
\[
S(n,j)=\sum_{\emph{I}}\prod_{t=1}^{j}\tau(Z_{i_{t}},\phi)^{\alpha_{s}}E^{\ast
}\prod_{s=1}^{j}\left(  N_{ni_{s}}-1\right)  ^{\alpha_{s}}%
\]
for each $j.$ Note that $\left\vert \prod_{t=1}^{j}\tau(Z_{i_{t}}%
,\phi)^{\alpha_{s}}\right\vert $ is bounded almost surely and therefore does
not affect the analysis. Also, $\sum_{\emph{I}}$ is a sum over $n^{j}$ terms
and thus is $O(n^{j})$ if all these terms are nonzero. The crucial factor in
determining the overall order is therefore $E^{\ast}\prod_{s=1}^{j}\left(
N_{ni_{s}}-1\right)  ^{\alpha_{s}}.$ We start with $j=1.$ Then $\alpha
_{1}=2k,$ $q=1...2k$ and $r^{(p)}$ are scalars. Consequently, $\kappa
_{r_{1}^{(p)}}=c_{1}$ where $c_{1}$ is some constant and $S(n,1)\leq c_{2}%
\sum_{\emph{i=1}}^{n}\left\vert \tau(Z_{i_{t}},\phi)\right\vert ^{2k}$ for
some other constant $c_{2}.$ If $j\leq k$ then for $q=1...2q,$ $r^{(p)}$ are
vectors with possibly only one element different from zero. Again, $S(n,j)\leq
c_{2}\sum_{\emph{I}}\prod_{t=1}^{j}\left\vert \tau(Z_{i_{t}},\phi)\right\vert
^{\alpha_{s}}$ for $j\leq k.$ If $j\geq k$ then $\alpha$ contains at least
$2(j-k)$ elements $\alpha_{i}=1.$ Now assume that for some $p,$ $r_{i}%
^{(p)}=1$ and $r_{j}^{(p)}=0$ for $i\neq j.$ Then $\kappa_{r_{i}^{(p)}%
}=E\left(  N_{ni_{s}}-1\right)  =0$ and thus $\prod_{p=1}^{q}\kappa
_{r_{1}^{(p)}r_{2}^{(p)}....r_{j}^{(p)}}=0.$ On the other hand if $r_{i}%
^{(p)}=1$ and $r_{j}^{(p)}\neq0$ for at least one $j\neq i$ then
$\kappa_{r_{1}^{(p)}r_{2}^{(p)}....r_{n}^{(p)}}\leq c_{1}n^{-1}.$ Since there
must exists $p^{\prime}$ corresponding to the other $\alpha_{i^{\prime}}=1$
such that either $r_{i^{^{\prime}}}^{(p^{\prime})}=1$ and $r_{j}^{(p^{\prime
})}=0$ for $i^{\prime}\neq j$ or $r_{i^{\prime}}^{(p^{\prime})}=1$ and
$r_{j}^{(p^{\prime})}\neq0$ for at least one $j\neq i^{\prime},$ it follows
that $\prod_{p=1}^{q}\kappa_{r_{1}^{(p)}r_{2}^{(p)}....r_{j}^{(p)}}%
=c_{3}n^{-2(j-k)},$ at most. It now follows that $S(n,j)\leq c_{2}%
n^{-2(j-k)}\sum_{\emph{I}}\prod_{t=1}^{j}\left\vert \tau(Z_{i_{t}}%
,\phi)\right\vert ^{\alpha_{s}}$ for all $j>k.$ Then,
\[
E\left\vert S(n,j)\right\vert \leq c_{2}\sum_{\emph{I}}E\left(  \prod
_{t=1}^{j}\left\vert \tau(Z_{i_{t}},\phi)\right\vert ^{\alpha_{s}}\right)
\leq c_{2}n^{j}E\left\vert \tau(Z_{i_{t}},\phi)\right\vert ^{2k}%
\]
for $j\leq k$ and
\[
E\left\vert S(n,j)\right\vert \leq c_{2}n^{-2(j-k)}\sum_{\emph{I}}E\left(
\prod_{t=1}^{j}\left\vert \tau(Z_{i_{t}},\phi)\right\vert ^{\alpha_{s}%
}\right)  \leq c_{2}n^{2k-j}E\left\vert \tau(Z_{i_{t}},\phi)\right\vert
^{2k}\leq c_{2}n^{k}E\left\vert \tau(Z_{i_{t}},\phi)\right\vert ^{2k}%
\]
for $j>k.$ Together these results imply that
\[
E\left\vert E^{\ast}\left(
{\textstyle\sum\nolimits_{i=1}^{n}}
\xi_{i}^{\ast}\left(  \phi\right)  \right)  ^{2k}\right\vert \leq
C(k)n^{k}E\left\vert \tau(Z_{i_{t}},\phi)\right\vert ^{2k}%
\]
where $C(k)$ is a constant that depends on $k.$ By the Markov inequality it
follows that $E^{\ast}\left(
{\textstyle\sum\nolimits_{i=1}^{n}}
\xi_{i}^{\ast}\left(  \phi\right)  \right)  ^{2k}=O_{p}(n^{k})$. We conclude
that
\[
P^{\ast}\left[  \left\vert \frac{1}{\sqrt{n}}\sum_{i=1}^{n}\xi_{i}\left(
\phi_{n}\right)  \right\vert >n^{\frac{1}{12}-\upsilon}\right]  \leq
\frac{O_{p}(n^{8})}{n^{\frac{28}{3}-16\upsilon}\eta^{16}}=O_{p}\left(
n^{-\frac{4}{3}+16\upsilon}\right)  .
\]
The second result follows immediately from
\[
P^{\ast}\left[  \left\vert \frac{1}{n}\sum_{i=1}^{n}\xi_{i}\left(  \phi
_{n}\right)  \right\vert >\eta\right]  =P^{\ast}\left[  \left\vert \frac
{1}{\sqrt{n}}\sum_{i=1}^{n}\xi_{i}\left(  \phi_{n}\right)  \right\vert >\eta
n^{1/2}\right]  \leq o_{p}\left(  n^{-\frac{23}{3}}\right)
\]
by the previous result.

\subsection{Proof of Lemma \ref{BootstrapThetaEpsilonConv}}

For any $\eta>0$, there exists some $\delta>0$ such that $\left\vert
\theta-\theta_{0}\right\vert >\eta/2$ implies $\left\vert G\left(
\theta\right)  -G\left(  \theta_{0}\right)  \right\vert >\delta$. Let
$\widehat{G}^{\ast}\left(  \theta\right)  \equiv\int g(z,\theta)d\widehat{F}%
^{\ast}\left(  z\right)  $ and $\widehat{G}_{\epsilon}^{\ast}\left(
\theta\right)  \equiv\int g\left(  z,\theta\right)  d\widehat{F}_{\epsilon
}\left(  z\right)  $. Then,
\[
P^{\ast}\left[  \max_{0\leq\epsilon\leq\frac{1}{\sqrt{n}}}\left\vert
\widehat{\theta}^{\ast}\left(  \epsilon\right)  -\widehat{\theta}\right\vert
\geq\eta\right]  \leq P^{\ast}\left[  \max_{0\leq\epsilon\leq\frac{1}{\sqrt
{n}}}\left\vert G(\widehat{\theta}^{\ast}(\epsilon))-G\left(  \widehat{\theta
}\right)  \right\vert >\delta\right]  .
\]
Because
\begin{align*}
G\left(  \widehat{\theta}^{\ast}\left(  \epsilon\right)  \right)  -G\left(
\widehat{\theta}\right)   &  =\left(  G\left(  \widehat{\theta}^{\ast}\left(
\epsilon\right)  \right)  -\widehat{G}\left(  \widehat{\theta}^{\ast}\left(
\epsilon\right)  \right)  \right)  +\left(  \widehat{G}\left(  \widehat{\theta
}^{\ast}\left(  \epsilon\right)  \right)  -\widehat{G}_{\epsilon}^{\ast
}\left(  \widehat{\theta}^{\ast}\left(  \epsilon\right)  \right)  \right) \\
&  +\left(  \widehat{G}_{\epsilon}^{\ast}\left(  \widehat{\theta}^{\ast
}\left(  \epsilon\right)  \right)  -\widehat{G}_{\epsilon}^{\ast}\left(
\widehat{\theta}\right)  \right)  +\left(  \widehat{G}_{\epsilon}^{\ast
}\left(  \widehat{\theta}\right)  -G\left(  \widehat{\theta}\right)  \right)
\end{align*}
and
\[
\left\vert \widehat{G}_{\epsilon}^{\ast}\left(  \theta\right)  -\widehat{G}%
\left(  \theta\right)  \right\vert \leq\left\vert \widehat{G}^{\ast}\left(
\theta\right)  -\widehat{G}\left(  \theta\right)  \right\vert ,
\]
we obtain
\begin{align}
\lefteqn{\max_{0\leq\epsilon\leq\frac{1}{\sqrt{n}}}\left\vert G\left(
\widehat{\theta}^{\ast}\left(  \epsilon\right)  \right)  -G\left(
\widehat{\theta}\right)  \right\vert }\nonumber\\
&  \leq\sup_{\theta\in\Theta}\left\vert \widehat{G}^{\ast}\left(
\theta\right)  -\widehat{G}\left(  \theta\right)  \right\vert +\sup_{\theta
\in\Theta}\left\vert \widehat{G}\left(  \theta\right)  -G\left(
\theta\right)  \right\vert \nonumber\\
&  +\max_{0\leq\epsilon\leq\frac{1}{\sqrt{n}}}\left\vert \widehat{G}%
_{\epsilon}^{\ast}\left(  \widehat{\theta}^{\ast}\left(  \epsilon\right)
\right)  -\widehat{G}_{\epsilon}^{\ast}\left(  \widehat{\theta}\right)
\right\vert +\max_{0\leq\epsilon\leq\frac{1}{\sqrt{n}}}\left\vert
\widehat{G}_{\epsilon}^{\ast}\left(  \widehat{\theta}\right)  -G\left(
\widehat{\theta}\right)  \right\vert \nonumber\\
&  \leq\sup_{\theta\in\Theta}\left\vert \widehat{G}^{\ast}\left(
\theta\right)  -\widehat{G}\left(  \theta\right)  \right\vert +\sup_{\theta
\in\Theta}\left\vert \widehat{G}\left(  \theta\right)  -G\left(
\theta\right)  \right\vert \nonumber\\
&  +\max_{0\leq\epsilon\leq\frac{1}{\sqrt{n}}}\left\vert \widehat{G}%
_{\epsilon}^{\ast}\left(  \widehat{\theta}^{\ast}\left(  \epsilon\right)
\right)  -\widehat{G}_{\epsilon}^{\ast}\left(  \widehat{\theta}\right)
\right\vert +\max_{0\leq\epsilon\leq\frac{1}{\sqrt{n}}}\left\vert
\widehat{G}_{\epsilon}^{\ast}\left(  \widehat{\theta}\right)  -\widehat{G}%
\left(  \widehat{\theta}\right)  \right\vert +\left\vert \widehat{G}\left(
\widehat{\theta}\right)  -G\left(  \widehat{\theta}\right)  \right\vert
\nonumber\\
&  \leq\sup_{\theta\in\Theta}\left\vert \widehat{G}^{\ast}\left(
\theta\right)  -\widehat{G}\left(  \theta\right)  \right\vert +\sup_{\theta
\in\Theta}\left\vert \widehat{G}\left(  \theta\right)  -G\left(
\theta\right)  \right\vert \nonumber\\
&  +\max_{0\leq\epsilon\leq\frac{1}{\sqrt{n}}}\left\vert \widehat{G}%
_{\epsilon}^{\ast}\left(  \widehat{\theta}^{\ast}\left(  \epsilon\right)
\right)  -\widehat{G}_{\epsilon}^{\ast}\left(  \widehat{\theta}\right)
\right\vert +\left\vert \widehat{G}^{\ast}\left(  \widehat{\theta}\right)
-\widehat{G}\left(  \widehat{\theta}\right)  \right\vert +\left\vert
\widehat{G}\left(  \widehat{\theta}\right)  -G\left(  \widehat{\theta}\right)
\right\vert \nonumber\\
&  \leq2\sup_{\theta\in\Theta}\left\vert \widehat{G}^{\ast}\left(
\theta\right)  -\widehat{G}\left(  \theta\right)  \right\vert +2\sup
_{\theta\in\Theta}\left\vert \widehat{G}\left(  \theta\right)  -G\left(
\theta\right)  \right\vert +\max_{0\leq\epsilon\leq\frac{1}{\sqrt{n}}%
}\left\vert \widehat{G}_{\epsilon}^{\ast}\left(  \widehat{\theta}^{\ast
}\left(  \epsilon\right)  \right)  -\widehat{G}_{\epsilon}^{\ast}\left(
\widehat{\theta}\right)  \right\vert \label{G-approx}%
\end{align}
By Lemma \ref{P-rate}, we have
\begin{equation}
P^{\ast}\left[  \sup_{\theta\in\Theta}\left\vert \widehat{G}^{\ast}\left(
\theta\right)  -\widehat{G}\left(  \theta\right)  \right\vert >\frac{\delta
}{6}\right]  =o_{p}\left(  n^{-\frac{23}{3}}\right)  \label{G-approx-1}%
\end{equation}
Conditional on data, $\sup_{\theta\in\Theta}\left\vert \widehat{G}\left(
\theta\right)  -G\left(  \theta\right)  \right\vert >\delta$ is a
non-stochastic event. Therefore, we can write
\[
P^{\ast}\left[  \sup_{\theta\in\Theta}\left\vert \widehat{G}\left(
\theta\right)  -G\left(  \theta\right)  \right\vert >\delta\right]  =1\left\{
\sup_{\theta\in\Theta}\left\vert \widehat{G}\left(  \theta\right)  -G\left(
\theta\right)  \right\vert >\delta\right\}  ,
\]
where $1\left\{  \cdot\right\}  $ denotes an indicator function. For every
$\sigma>0$, we have
\begin{align}
&  \Pr\left[  P^{\ast}\left[  \sup_{\theta\in\Theta}\left\vert \widehat{G}%
\left(  \theta\right)  -G\left(  \theta\right)  \right\vert >\frac{\delta}%
{6}\right]  >\sigma n^{-\frac{23}{3}}\right] \label{P*-bound}\\
&  =\Pr\left[  1\left\{  \sup_{\theta\in\Theta}\left\vert \widehat{G}\left(
\theta\right)  -G\left(  \theta\right)  \right\vert >\frac{\delta}{6}\right\}
>0\right] \nonumber\\
&  =\Pr\left[  \sup_{\theta\in\Theta}\left\vert \widehat{G}\left(
\theta\right)  -G\left(  \theta\right)  \right\vert >\frac{\delta}{6}\right]
=o(1)\nonumber
\end{align}
where the last equality is implied by Lemma \ref{Lem-uniform-convergence}. It
therefore follows that
\begin{equation}
P^{\ast}\left[  \sup_{\theta\in\Theta}\left\vert \widehat{G}\left(
\theta\right)  -G\left(  \theta\right)  \right\vert >\frac{\delta}{6}\right]
=o_{p}\left(  n^{-\frac{23}{3}}\right)  . \label{G-approx-2}%
\end{equation}
Finally,
\begin{align*}
\max_{0\leq\epsilon\leq\frac{1}{\sqrt{n}}}&\left\vert \widehat{G}_{\epsilon
}^{\ast}\left(  \widehat{\theta}^{\ast}\left(  \epsilon\right)  \right)
-\widehat{G}_{\epsilon}^{\ast}\left(  \widehat{\theta}\right)  \right\vert  \\
&\leq\max_{0\leq\epsilon\leq\frac{1}{\sqrt{n}}}\left\vert \widehat{G}%
_{\epsilon}^{\ast}\left(  \widehat{\theta}^{\ast}\left(  \epsilon\right)
\right)  -\widehat{G}\left(  \widehat{\theta}\right)  \right\vert +\max
_{0\leq\epsilon\leq\frac{1}{\sqrt{n}}}\left\vert \widehat{G}_{\epsilon}^{\ast
}\left(  \widehat{\theta}\right)  -\widehat{G}\left(  \widehat{\theta}\right)
\right\vert \\
&  =\max_{0\leq\epsilon\leq\frac{1}{\sqrt{n}}}\left\vert \sup_{\theta
}\widehat{G}_{\epsilon}^{\ast}\left(  \theta\right)  -\sup_{\theta}%
\widehat{G}\left(  \theta\right)  \right\vert +\max_{0\leq\epsilon\leq\frac
{1}{\sqrt{n}}}\left\vert \widehat{G}_{\epsilon}^{\ast}\left(  \widehat{\theta
}\right)  -\widehat{G}\left(  \widehat{\theta}\right)  \right\vert \\
&  \leq\max_{0\leq\epsilon\leq\frac{1}{\sqrt{n}}}\left\vert \sup_{\theta
}\widehat{G}_{\epsilon}^{\ast}\left(  \theta\right)  -\sup_{\theta}%
\widehat{G}\left(  \theta\right)  \right\vert +\left\vert \widehat{G}^{\ast
}\left(  \widehat{\theta}\right)  -\widehat{G}\left(  \widehat{\theta}\right)
\right\vert \\
&  \leq\max_{0\leq\epsilon\leq\frac{1}{\sqrt{n}}}\sup_{\theta}\left\vert
\widehat{G}_{\epsilon}^{\ast}\left(  \theta\right)  -\widehat{G}\left(
\theta\right)  \right\vert +\sup_{\theta}\left\vert \widehat{G}^{\ast}%
(\theta)-\widehat{G}\left(  \theta\right)  \right\vert \\
&  \leq\max_{0\leq\epsilon\leq\frac{1}{\sqrt{n}}}\sup_{\theta}\left\vert
\widehat{G}^{\ast}\left(  \theta\right)  -\widehat{G}\left(  \theta\right)
\right\vert +\sup_{\theta}\left\vert \widehat{G}^{\ast}(\theta)-\widehat{G}%
\left(  \theta\right)  \right\vert \\
&  \leq\sup_{\theta}\left\vert \widehat{G}^{\ast}\left(  \theta\right)
-\widehat{G}\left(  \theta\right)  \right\vert +\sup_{\theta}\left\vert
\widehat{G}^{\ast}(\theta)-\widehat{G}\left(  \theta\right)  \right\vert \\
&  =2\sup_{\theta}\left\vert \widehat{G}^{\ast}\left(  \theta\right)
-\widehat{G}\left(  \theta\right)  \right\vert
\end{align*}
Here, the first equality is based on the definitions of $\widehat{\theta
}^{\ast}\left(  \epsilon\right)  $ and $\widehat{\theta}$. Because
\[
P^{\ast}\left[  \sup_{\theta\in\Theta}\left\vert \widehat{G}^{\ast}\left(
\theta\right)  -\widehat{G}\left(  \theta\right)  \right\vert >\delta\right]
=o_{p}\left(  n^{-\frac{23}{3}}\right)
\]
we can conclude that
\begin{equation}
P^{\ast}\left[  \max_{0\leq\epsilon\leq\frac{1}{\sqrt{n}}}\left\vert
\widehat{G}_{\epsilon}^{\ast}\left(  \widehat{\theta}^{\ast}\left(
\epsilon\right)  \right)  -\widehat{G}_{\epsilon}^{\ast}\left(
\widehat{\theta}\right)  \right\vert >\frac{\delta}{3}\right]  =o_{p}\left(
n^{-\frac{23}{3}}\right)  . \label{G-approx-3}%
\end{equation}
The conclusion follows by combining (\ref{G-approx}) - (\ref{G-approx-3}).

\subsection{Proof of Lemma \ref{multi:k-boot}}

In the same way as in the proof of Lemma \ref{multi:k}
\begin{align*}
&  \int K\left(  z;\widehat{\theta}^{\ast}\left(  \epsilon\right)  \right)
d\widehat{F}_{\epsilon}\left(  z\right)  -\int K\left(  z;\theta_{0}\right)
d\widehat{F}\left(  z\right) \\
&  =\int\frac{\partial K\left(  z;\theta^{\ast}\right)  }{\partial\theta
}\left(  \widehat{\theta}^{\ast}\left(  \epsilon\right)  -\theta_{0}\right)
d\widehat{F}_{\epsilon}\left(  z\right)  +\epsilon\sqrt{n}\int K\left(
z;\theta_{0}\right)  d\left(  \widehat{F}^{\ast}-\widehat{F}\right)  \left(
z\right)
\end{align*}
where $\theta^{\ast}$ is between $\theta_{0}$ and $\widehat{\theta}^{\ast
}\left(  \epsilon\right)  $. Therefore, we have
\begin{align*}
\bigg\vert \int K\left(  z;\theta\left(  \epsilon\right)  \right)
dF_{\epsilon}\left(  z\right)  &-\int K\left(  z;\theta_{0}\right)
d\widehat{F}\left(  z\right)  \bigg\vert  \\
&  \leq\left\vert \widehat{\theta
}^{\ast}\left(  \epsilon\right)  -\theta_{0}\right\vert \cdot\left(  \frac
{1}{n}\sum_{i=1}^{n}M\left(  Z_{i}\right)  +\frac{1}{n}\sum_{i=1}^{n}M\left(
Z_{i}^{\ast}\right)  \right) \\
&  +\left\vert \frac{1}{n}\sum_{i=1}^{n}M\left(  Z_{i}^{\ast}\right)
-\frac{1}{n}\sum_{i=1}^{n}M\left(  Z_{i}\right)  \right\vert
\end{align*}
where $M\left(  \cdot\right)  $ is defined in Condition \ref{HH}. Let $\bar
{M}=\frac{1}{n}\sum_{i=1}^{n}M\left(  Z_{i}\right)  $ and $\bar{M}^{\ast
}=\frac{1}{n}\sum_{i=1}^{n}M\left(  Z_{i}^{\ast}\right)  $. Then, for any
$\eta$ and some $c$
\[
P^{\ast}\left[  \left\vert \widehat{\theta}^{\ast}\left(  \epsilon\right)
-\theta_{0}\right\vert \bar{M}>\eta\right]  \leq P^{\ast}\left[  \left\vert
\widehat{\theta}^{\ast}\left(  \epsilon\right)  -\theta_{0}\right\vert
>\eta/c\right]  +P^{\ast}\left[  \left\vert \bar{M}-E\left[  M\left(
Z_{i}\right)  \right]  \right\vert >c\right]  =o_{p}\left(  n^{-\frac{23}{3}%
}\right)
\]
since $P^{\ast}\left[  \left\vert \bar{M}-E\left[  M\left(  Z_{i}\right)
\right]  \right\vert >c\right]  =1$ with probability equal to $P\left[
\left\vert \bar{M}-E\left[  M\left(  Z_{i}\right)  \right]  \right\vert
>c\right]  =o\left(  n^{-\frac{23}{3}}\right)  $ by Lemma
\ref{modifiedHH:Lem1} and zero otherwise for some $c.$ Then, $P^{\ast}\left[
\left\vert \bar{M}-E\left[  M\left(  Z_{i}\right)  \right]  \right\vert
>c\right]  =o_{p}\left(  n^{-\frac{23}{3}}\right)  $ by the same argument as
in \ref{P*-bound} Moreover,
\[
P^{\ast}\left[  \left\vert \widehat{\theta}^{\ast}\left(  \epsilon\right)
-\theta_{0}\right\vert \left\vert \bar{M}^{\ast}-\bar{M}\right\vert
>\eta\right]  \leq P^{\ast}\left[  \left\vert \widehat{\theta}^{\ast}\left(
\epsilon\right)  -\theta_{0}\right\vert >\eta/c\right]  +P^{\ast}\left[
\left\vert \bar{M}^{\ast}-\bar{M}\right\vert >c\right]  =o_{p}\left(
n^{-\frac{23}{3}}\right)
\]
by Lemmas \ref{P-rate} and \ref{BootstrapThetaEpsilonConv}. It thus follows
that for any $\eta>0,$%
\[
P^{\ast}\left(  \left\vert \int K\left(  z;\theta\left(  \epsilon\right)
\right)  dF_{\epsilon}\left(  z\right)  -\int K\left(  z;\theta_{0}\right)
d\widehat{F}\left(  z\right)  \right\vert >\eta\right)  =o_{p}\left(
n^{-\frac{23}{3}}\right)  .
\]
Finally note that $P^{\ast}\left(  \left\vert \int K\left(  z;\theta
_{0}\right)  d\widehat{F}\left(  z\right)  -EK\left(  z;\theta_{0}\right)
\right\vert >\eta\right)  =1$ with probability
\[
P\left(  \left\vert \int K\left(  z;\theta_{0}\right)  d\widehat{F}\left(
z\right)  -E\left[  K\left(  z;\theta_{0}\right)  \right]  \right\vert
>\eta\right)  =o(n^{-\frac{23}{3}})
\]
by Lemma \ref{modifiedHH:Lem1}. Thus, by the same argument as in
\ref{P*-bound}
\[
P^{\ast}\left(  \left\vert \int K\left(  z;\theta_{0}\right)  d\widehat{F}%
\left(  z\right)  -EK\left(  z;\theta_{0}\right)  \right\vert >\eta\right)
=o_{p}\left(  n^{-\frac{23}{3}}\right)  .
\]

For the second result fix $\delta>0$ arbitrary. Then
\begin{align*}
P^{\ast}\left[  \max_{0\leq\epsilon\leq\frac{1}{\sqrt{n}}}\left\vert \int
K\left(  \cdot;\widehat{\theta}^{\ast}\left(  \epsilon\right)  \right)
d\widehat{\Delta}\right\vert >Cn^{\frac{1}{12}-\upsilon}\right]   &  \leq
P^{\ast}\left[  \sup_{\left\vert \theta-\widehat{\theta}\right\vert <\delta
}\left\vert \int K\left(  \cdot;\theta\right)  d\widehat{\Delta}\right\vert
>Cn^{\frac{1}{12}-\upsilon}\right] \\
&  +P^{\ast}\left[  \max_{0\leq\epsilon\leq\frac{1}{\sqrt{n}}}\left\vert
\widehat{\theta}^{\ast}\left(  \epsilon\right)  -\widehat{\theta}\right\vert
\geq\delta\right]
\end{align*}
where
\[
P^{\ast}\left[  \sup_{\left\vert \theta-\widehat{\theta}\right\vert <\delta
}\left\vert \int K\left(  \cdot;\theta\right)  d\widehat{\Delta}\right\vert
>Cn^{\frac{1}{12}-\upsilon}\right]  =o_{p}\left(  n^{-1+16\upsilon}\right)
\]
follows directly from Lemma \ref{P-rate} and
\[
P^{\ast}\left[  \max_{0\leq\epsilon\leq\frac{1}{\sqrt{n}}}\left\vert
\widehat{\theta}^{\ast}\left(  \epsilon\right)  -\widehat{\theta}\right\vert
\geq\delta\right]  =o_{p}\left(  n^{-\frac{23}{3}}\right)
\]
follows from Lemma \ref{BootstrapThetaEpsilonConv}.

\subsection{Proof of Lemma \ref{boot:bound:theta}}

Let $\bar{M}_{\epsilon}=\int\ell^{\theta}\left(  z,\epsilon\right)
d\widehat{F}_{\epsilon}\left(  z\right)  $ such that
\[
\widehat{\theta}^{\epsilon}\left(  \epsilon\right)  =-\bar{M}_{\epsilon}%
^{-1}\int\ell\left(  \cdot,\epsilon\right)  d\widehat{\Delta}%
\]
and for any $\delta>0$ some $C>0$ and for every $\upsilon$ such that
$\upsilon<\frac{1}{16}$
\begin{align*}
P^{\ast}\left[  \left\vert \widehat{\theta}^{\epsilon}\left(  \epsilon\right)
\right\vert >Cn^{\frac{1}{12}-\upsilon}\right]   &  \leq P^{\ast}\left[
\sup_{\epsilon}\left\vert \int\ell\left(  \cdot,\epsilon\right)
d\widehat{\Delta}\right\vert >\delta Cn^{\frac{1}{12}-\upsilon}\right] \\
&  +P^{\ast}\left[  \sup_{\epsilon}\left\vert \bar{M}_{\epsilon}-E\left[
\ell^{\theta}\left(  z,\theta_{0}\right)  \right]  \right\vert \geq
\delta\right] \\
&  =o_{p}\left(  \max n^{-\frac{23}{3}},n^{-1+16\upsilon}\right)
\end{align*}
by Lemma \ref{multi:k-boot}. The rest of the Lemma can be established similarly.

\subsection{Detailed Derivation for (\ref{E*3})\label{sec-E*3}}

We have%
\begin{align*}
\sqrt{n}B_{n}  &  \equiv-3\mathcal{I}^{-3}\sqrt{n}\left(  \widehat{\mathcal{I}%
}-\mathcal{I}\right)  \mathcal{Q}_{1}\left(  \theta_{0}\right) \\
&  +\mathcal{I}^{-2}\sqrt{n}\left(  \widehat{\mathcal{Q}}_{1}\left(
\widehat{\theta}\right)  -\mathcal{Q}\left(  \theta_{0}\right)  \right) \\
&  +\mathcal{I}^{-3}\mathcal{Q}_{1}\left(  \theta_{0}\right)  \left(  \frac
{1}{\sqrt{n}}%
{\textstyle\sum_{i=1}^{n}}
\left[  \ell\left(  Z_{i},\theta_{0}\right)  ^{2}-E\left[  \ell\left(
Z_{i},\theta_{0}\right)  ^{2}\right]  \right]  +\sqrt{n}\left(  \bar{m}%
_{4}\left(  \widehat{\theta}\right)  -\bar{m}_{4}\left(  \theta_{0}\right)
\right)  \right) \\
&  -4\mathcal{I}^{-3}\sqrt{n}\left(  \widehat{\mathcal{I}}-\mathcal{I}\right)
E\left[  \ell\left(  Z_{i},\theta_{0}\right)  \ell^{\theta}\left(
Z_{i},\theta\right)  \right] \\
&  +2\mathcal{I}^{-2}\left(  \sqrt{n}\left(  \bar{m}_{3}\left(
\widehat{\theta}\right)  -\bar{m}_{3}\left(  \theta_{0}\right)  \right)
\right. \\
&  \left.  +\sqrt{n}\left[  n^{-1}%
{\textstyle\sum_{i=1}^{n}}
\ell\left(  Z_{i},\theta_{0}\right)  \ell^{\theta}\left(  Z_{i},\theta
_{0}\right)  -E\left[  \ell\left(  Z_{i},\theta_{0}\right)  \ell^{\theta
}\left(  Z_{i},\theta_{0}\right)  \right]  \right]  \right)
\end{align*}
From Lemma \ref{equiLemma}, we have
\begin{align*}
\sqrt{n}\left(  \widehat{\mathcal{I}}-\mathcal{I}\right)   &  =-V\left(
\theta_{0}\right)  -\mathcal{Q}_{1}\left(  \theta_{0}\right)  \mathcal{I}%
^{-1}U\left(  \theta_{0}\right)  +o_{p}\left(  1\right)  ,\\
\sqrt{n}\left(  \widehat{\mathcal{Q}}_{1}\left(  \widehat{\theta}\right)
-\mathcal{Q}_{1}\left(  \theta_{0}\right)  \right)   &  =W\left(  \theta
_{0}\right)  +\mathcal{Q}_{2}\left(  \theta_{0}\right)  \mathcal{I}%
^{-1}U\left(  \theta_{0}\right)  +o_{p}\left(  1\right)  ,\\
\sqrt{n}\left(  \overline{m}_{4}\left(  \widehat{\theta}\right)  -\overline
{m}_{4}\left(  \theta_{0}\right)  \right)   &  =2E\left[  \ell\left(
Z_{i},\theta_{0}\right)  \ell^{\theta}\left(  Z_{i},\theta_{0}\right)
\right]  \mathcal{I}^{-1}U\left(  \theta_{0}\right)  +o_{p}\left(  1\right)
,\\
\sqrt{n}\left(  \overline{m}_{3}\left(  \widehat{\theta}\right)  -\overline
{m}_{3}\left(  \theta_{0}\right)  \right)   &  =E\left[  \left(  \ell^{\theta
}\left(  Z_{i},\theta_{0}\right)  \right)  ^{2}\right]  \mathcal{I}%
^{-1}U\left(  \theta_{0}\right)  \\
&+E\left[  \ell\left(  Z_{i},\theta
_{0}\right)  \ell^{\theta\theta}\left(  Z_{i},\theta_{0}\right)  \right]
\mathcal{I}^{-1}U\left(  \theta_{0}\right)  +o_{p}\left(  1\right)  ,
\end{align*}

\[
E\left[  \ell^{\theta}\left(  Z_{i},\theta\right)  \right]
\]
Therefore, we have%
\begin{align*}
\sqrt{n}B_{n}  &  =-3\mathcal{I}^{-3}\mathcal{Q}_{1}\left(  \theta_{0}\right)
\left(  -V\left(  \theta_{0}\right)  -\mathcal{Q}_{1}\left(  \theta
_{0}\right)  \mathcal{I}^{-1}U\left(  \theta_{0}\right)  \right) \\
&  +\mathcal{I}^{-2}\left(  W_{i}\left(  \theta_{0}\right)  +\mathcal{Q}%
_{2}\left(  \theta_{0}\right)  \mathcal{I}^{-1}U\left(  \theta_{0}\right)
\right) \\
&  +\mathcal{I}^{-3}\mathcal{Q}_{1}\left(  \theta_{0}\right)  \left(  \frac
{1}{\sqrt{n}}%
{\textstyle\sum_{i=1}^{n}}
\left[  \ell\left(  Z_{i},\theta_{0}\right)  ^{2}-E\left[  \ell\left(
Z_{i},\theta_{0}\right)  ^{2}\right]  \right]  \right) \\
&  +2\mathcal{I}^{-3}\mathcal{Q}_{1}\left(  \theta_{0}\right)  \left(
E\left[  \ell\left(  Z_{i},\theta_{0}\right)  \ell^{\theta}\left(
Z_{i},\theta_{0}\right)  \right]  \mathcal{I}^{-1}U\left(  \theta_{0}\right)
\right) \\
&  -4\mathcal{I}^{-3}E\left[  \ell\left(  Z_{i},\theta_{0}\right)
\ell^{\theta}\left(  Z_{i},\theta\right)  \right]  \left(  -V\left(
\theta_{0}\right)  -\mathcal{Q}_{1}\left(  \theta_{0}\right)  \mathcal{I}%
^{-1}U\left(  \theta_{0}\right)  \right) \\
&  +2\mathcal{I}^{-2}\left(  E\left[  \left(  \ell^{\theta}\left(
Z_{i},\theta_{0}\right)  \right)  ^{2}\right]  \mathcal{I}^{-1}U\left(
\theta_{0}\right)  +E\left[  \ell\left(  Z_{i},\theta_{0}\right)  \ell
^{\theta\theta}\left(  Z_{i},\theta_{0}\right)  \right]  \mathcal{I}%
^{-1}U\left(  \theta_{0}\right)  \right) \\
&  +2\mathcal{I}^{-2}\sqrt{n}\left(  n^{-1}%
{\textstyle\sum_{i=1}^{n}}
\ell\left(  Z_{i},\theta_{0}\right)  \ell^{\theta}\left(  Z_{i},\theta
_{0}\right)  -E\left[  \ell\left(  Z_{i},\theta_{0}\right)  \ell^{\theta
}\left(  Z_{i},\theta_{0}\right)  \right]  \right) \\
&  +o_{p}\left(  1\right)
\end{align*}
from which we obtain%
\begin{align*}
2\mathbb{B}  &  \equiv3\mathcal{I}^{-4}\mathcal{Q}_{1}\left(  \theta
_{0}\right)  ^{2}U\left(  \theta_{0}\right)  +\mathcal{I}^{-3}\mathcal{Q}%
_{2}\left(  \theta_{0}\right)  U\left(  \theta_{0}\right) \\
&  -\mathcal{I}^{-4}\left(  \mathcal{Q}_{1}\left(  \theta_{0}\right)  \right)
^{2}U\left(  \theta_{0}\right)  +4\mathcal{I}^{-4}\mathcal{Q}_{1}\left(
\theta_{0}\right)  E\left[  \ell\left(  Z_{i},\theta_{0}\right)  \ell^{\theta
}\left(  Z_{i},\theta\right)  \right]  U\left(  \theta_{0}\right) \\
&  +2\mathcal{I}^{-3}E\left[  U_{i}\left(  \theta_{0}\right)  W_{i}\left(
\theta_{0}\right)  \right]  U\left(  \theta_{0}\right)  +2\mathcal{I}%
^{-3}E\left[  \ell^{\theta}\left(  Z_{i},\theta_{0}\right)  ^{2}\right]
U\left(  \theta_{0}\right) \\
&  +3\mathcal{I}^{-3}\mathcal{Q}_{1}\left(  \theta_{0}\right)  V\left(
\theta_{0}\right)  +4\mathcal{I}^{-3}E\left[  U_{i}\left(  \theta_{0}\right)
V_{i}\left(  \theta_{0}\right)  \right]  V\left(  \theta_{0}\right) \\
&  +\mathcal{I}^{-2}W\left(  \theta_{0}\right) \\
&  +2\mathcal{I}^{-2}n^{-1/2}\left(
{\textstyle\sum_{i=1}^{n}}
\ell\left(  Z_{i},\theta_{0}\right)  \ell^{\theta}\left(  Z_{i},\theta
_{0}\right)  -E\left[  \ell\left(  Z_{i},\theta_{0}\right)  \ell^{\theta
}\left(  Z_{i},\theta_{0}\right)  \right]  \right) \\
&  +\mathcal{I}^{-3}\mathcal{Q}_{1}\left(  \theta_{0}\right)  n^{-1/2}\left(
{\textstyle\sum_{i=1}^{n}}
\left[  \ell\left(  Z_{i},\theta_{0}\right)  ^{2}-E\left[  \ell\left(
Z_{i},\theta_{0}\right)  ^{2}\right]  \right]  \right)  .
\end{align*}

\subsection{Proofs of (\ref{r6}), (\ref{r4}) and (\ref{r5})\label{sec-r6}}

From Lemma \ref{multi:bound:theta-epsilon-1}, we have
\[
\Pr\left[  \left\vert \frac{1}{n^{\frac{1}{2}-6\upsilon}}\theta^{\epsilon
\epsilon\epsilon\epsilon\epsilon\epsilon}\left(  \widetilde{\epsilon}\right)
\right\vert >C\right]  \leq\Pr\left[  \max_{0\leq\epsilon\leq\frac{1}{\sqrt
{n}}}\left\vert \theta^{\epsilon\epsilon\epsilon\epsilon\epsilon\epsilon
}\left(  \epsilon\right)  \right\vert >Cn^{\frac{1}{2}-6\upsilon}\right]
=o\left(  1\right)
\]
for every $\upsilon$ such that $\upsilon<\frac{1}{48}$. In particular, we
have
\begin{equation}
\frac{1}{\sqrt{n}}\theta^{\epsilon\epsilon\epsilon\epsilon\epsilon\epsilon
}\left(  \widetilde{\epsilon}\right)  =o_{p}\left(  1\right)  . \label{r6-1}%
\end{equation}
By Lemma \ref{multi:bound:theta-epsilon-1} again, we obtain
\begin{align*}
\Pr\left[  \left\vert \frac{1}{n^{\frac{1}{2}-6\upsilon}}\frac{1}{n}\sum
_{j=1}^{n}\theta_{\left(  j\right)  }^{\epsilon\epsilon\epsilon\epsilon
\epsilon\epsilon}\left(  \widetilde{\epsilon}_{\left(  j\right)  }\right)
\right\vert >C\right]   &  \leq\sum_{j=1}^{n}\Pr\left[  \left\vert \frac
{1}{n^{\frac{1}{2}-6\upsilon}}\theta_{\left(  j\right)  }^{\epsilon
\epsilon\epsilon\epsilon\epsilon\epsilon}\left(  \widetilde{\epsilon}_{\left(
j\right)  }\right)  \right\vert >C\right] \\
&  \leq\sum_{j=1}^{n}\Pr\left[  \max_{0\leq\epsilon\leq\frac{1}{\sqrt{n}}%
}\left\vert \frac{1}{n^{\frac{1}{2}-6\upsilon}}\theta_{\left(  j\right)
}^{\epsilon\epsilon\epsilon\epsilon\epsilon\epsilon}\left(  \epsilon\right)
\right\vert >C\right] \\
&  =n\Pr\left[  \max_{0\leq\epsilon\leq\frac{1}{\sqrt{n}}}\left\vert \frac
{1}{n^{\frac{1}{2}-6\upsilon}}\theta_{\left(  j\right)  }^{\epsilon
\epsilon\epsilon\epsilon\epsilon\epsilon}\left(  \epsilon\right)  \right\vert
>C\right] \\
&  =o\left(  1\right)  .
\end{align*}
Here, the first equality is based on the fact that $Z_{i}$ are i.i.d., so that
$\theta_{\left(  j\right)  }^{\epsilon\epsilon\epsilon\epsilon\epsilon
\epsilon}\left(  \epsilon\right)  $ are identically distributed for
$j=1,\ldots,n$. In particular, we have
\begin{equation}
\frac{1}{\left(  n-1\right)  \sqrt{n}}\sum_{j=1}^{n}\theta_{\left(  j\right)
}^{\epsilon\epsilon\epsilon\epsilon\epsilon\epsilon}\left(
\widetilde{\epsilon}_{\left(  j\right)  }\right)  =o_{p}\left(  1\right)  .
\label{r6-2}%
\end{equation}
Combining (\ref{r6-1}) and (\ref{r6-2}), we obtain (\ref{r6}).

Note that $\theta^{\epsilon\epsilon\epsilon\epsilon}$ is a sum of V-statistic
of order 4 as considered in Lemma \ref{Vjack4}. Likewise, $\theta
^{\epsilon\epsilon\epsilon\epsilon\epsilon}$ is a sum of V-statistic of order
5 as considered in Lemma \ref{Vjack5}. Therefore, combining (\ref{theta-e4})
and (\ref{theta-e5}) with Lemmas \ref{Vjack4} and \ref{Vjack5}, we obtain
\begin{align*}
&  n\theta^{\epsilon\epsilon\epsilon\epsilon}-\frac{n^{2}}{n-1}\frac{1}{n}%
\sum_{j=1}^{n}\theta_{\left(  j\right)  }^{\epsilon\epsilon\epsilon\epsilon
}=O_{p}\left(  1\right) \\
&  n\left(  \theta^{\epsilon\epsilon\epsilon\epsilon\epsilon}-\frac{n\sqrt{n}%
}{\left(  n-1\right)  \sqrt{n-1}}\frac{1}{n}\sum_{j=1}^{n}\theta_{\left(
j\right)  }^{\epsilon\epsilon\epsilon\epsilon\epsilon}\right)  =O_{p}\left(
1\right)
\end{align*}
from which we further obtain (\ref{r4}) and (\ref{r5}).

\subsection{Detailed Derivation for (\ref{r1-3})}

We have
\[
\theta^{\epsilon}=T_{1}=\mathcal{I}^{-1}U\left(  \theta_{0}\right)
=n^{-1/2}\sum_{i=1}^{n}\mathcal{I}^{-1}U_{i}\left(  \theta\right)  ,
\]
so it takes the form of $n^{-1/2}\sum_{i=1}^{n}X_{i}$. We can therefore apply
Lemma \ref{Vjack1} and conclude that%
\begin{equation}
n\theta^{\epsilon}-\sqrt{n}\sqrt{n-1}\frac{1}{n}\sum_{j=1}^{n}\theta_{\left(
j\right)  }^{\epsilon}=\theta^{\epsilon}. \label{r1-3-1}%
\end{equation}

We also have%
\begin{align*}
T_{2}  &  =\frac{1}{2}\theta^{\epsilon\epsilon}=\frac{1}{2}\mathcal{I}%
^{-3}\mathcal{Q}_{1}\left(  \theta_{0}\right)  U\left(  \theta_{0}\right)
^{2}+\mathcal{I}^{-2}U\left(  \theta_{0}\right)  V\left(  \theta_{0}\right) \\
&  =\frac{1}{2}\mathcal{I}^{-3}E\left[  \ell^{\theta\theta}\right]  \left(
n^{-1/2}\sum_{i=1}^{n}U_{i}\left(  \theta\right)  \right)  ^{2}+\mathcal{I}%
^{-2}\left(  n^{-1/2}\sum_{i=1}^{n}U_{i}\left(  \theta\right)  \right)
\left(  n^{-1/2}\sum_{i=1}^{n}V_{i}\left(  \theta\right)  \right) \\
&  =\left(  n^{-1/2}\sum_{i=1}^{n}\mathcal{I}^{-2}U_{i}\left(  \theta\right)
\right)  \left(  n^{-1/2}\sum_{i=1}^{n}\left(  \frac{1}{2}\mathcal{I}%
^{-1}E\left[  \ell^{\theta\theta}\right]  U_{i}\left(  \theta\right)
+V_{i}\left(  \theta\right)  \right)  \right)  ,
\end{align*}
so it takes the form of $\left(  n^{-1/2}\sum_{i=1}^{n}X_{1,i}\right)  \left(
n^{-1/2}\sum_{i=1}^{n}X_{2,i}\right)  $. Because%
\begin{align*}
\frac{1}{n-1}\sum_{i\neq j}X_{1,i}X_{2,j}  &  =\frac{1}{n}\sum_{i\neq
j}X_{1,i}X_{2,j}+\left(  \frac{1}{n-1}-\frac{1}{n}\right)  \sum_{i\neq
j}X_{1,i}X_{2,j}\\
&  =\left(  n^{-1/2}\sum_{i=1}^{n}X_{1,i}\right)  \left(  n^{-1/2}\sum
_{i=1}^{n}X_{2,i}\right) \\
&-\frac{1}{n}\sum_{i=1}^{n}X_{1,i}X_{2,i}+\frac
{1}{n\left(  n-1\right)  }\sum_{i\neq j}X_{1,i}X_{2,j}%
\end{align*}
and if $X_{1,i}$ and $X_{2,j}$ have zero means, we would have%
\begin{align*}
E\left[  \frac{1}{n\left(  n-1\right)  }\sum_{i\neq j}X_{1,i}X_{2,j}\right]
&  =0\\
E\left[  \left(  \frac{1}{n\left(  n-1\right)  }\sum_{i\neq j}X_{1,i}%
X_{2,j}\right)  ^{2}\right]   &  =\frac{n\left(  n-1\right)  E\left[
X_{1,i}^{2}\right]  E\left[  X_{2,j}^{2}\right]  }{n^{2}\left(  n-1\right)
^{2}}=O\left(  \frac{1}{n^{2}}\right)  .
\end{align*}
In other words, we have%
\[
\frac{1}{n-1}\sum_{i\neq j}X_{1,i}X_{2,j}=\left(  n^{-1/2}\sum_{i=1}%
^{n}X_{1,i}\right)  \left(  n^{-1/2}\sum_{i=1}^{n}X_{2,i}\right)  -\frac{1}%
{n}\sum_{i=1}^{n}X_{1,i}X_{2,i}+O_{p}\left(  n^{-1}\right)
\]
under the zero mean assumption. The zero mean assumption is satisfied by
$\mathcal{I}^{-2}U_{i}\left(  \theta\right)  $ and $\frac{1}{2}\mathcal{I}%
^{-1}E\left[  \ell^{\theta\theta}\right]  U_{i}\left(  \theta\right)
+V_{i}\left(  \theta\right)  $, so we can apply Lemma \ref{Vjack2} and
conclude that%
\begin{align*}
\frac{1}{2}\frac{1}{\sqrt{n}}\left(  n\theta^{\epsilon\epsilon}-\sum_{j=1}%
^{n}\theta_{\left(  j\right)  }^{\epsilon\epsilon}\right)   &  =\frac{1}%
{2}\theta^{\epsilon\epsilon}-\frac{1}{n}\sum_{i=1}^{n}\mathcal{I}^{-2}%
U_{i}\left(  \theta\right)  \left(  \frac{1}{2}\mathcal{I}^{-1}E\left[
\ell^{\theta\theta}\right]  U_{i}\left(  \theta\right)  +V_{i}\left(
\theta\right)  \right)  +O_{p}\left(  n^{-1}\right) \\
&  =\frac{1}{2}\theta^{\epsilon\epsilon}\\
&  -\frac{1}{2}\mathcal{I}^{-3}E\left[  \ell^{\theta\theta}\right]  \frac
{1}{n}\sum_{i=1}^{n}\ell\left(  Z_{i},\theta_{0}\right)  ^{2}\\
&  -\mathcal{I}^{-2}\frac{1}{n}\sum_{i=1}^{n}\ell\left(  Z_{i},\theta
_{0}\right)  \left(  \ell^{\theta}\left(  Z_{i},\theta_{0}\right)  -E\left[
\ell^{\theta}\left(  Z_{i},\theta_{0}\right)  \right]  \right)  +O_{p}\left(
n^{-1}\right)  ,
\end{align*}
which we rewrite as%
\begin{align*}
&  \frac{1}{2}\frac{1}{\sqrt{n}}\left(  n\theta^{\epsilon\epsilon}-\sum
_{j=1}^{n}\theta_{\left(  j\right)  }^{\epsilon\epsilon}\right) \\
&  =\frac{1}{2}\theta^{\epsilon\epsilon}-\frac{1}{n}\sum_{i=1}^{n}%
\mathcal{I}^{-2}U_{i}\left(  \theta\right)  \left(  \frac{1}{2}\mathcal{I}%
^{-1}E\left[  \ell^{\theta\theta}\right]  U_{i}\left(  \theta\right)
+V_{i}\left(  \theta\right)  \right)  +O_{p}\left(  n^{-1}\right) \\
&  =\frac{1}{2}\theta^{\epsilon\epsilon}-\frac{1}{2}\mathcal{I}^{-3}E\left[
\ell^{\theta\theta}\right]  E\left[  \ell\left(  Z_{i},\theta_{0}\right)
^{2}\right] \\
&  -\mathcal{I}^{-2}E\left[  \ell\left(  Z_{i},\theta_{0}\right)  \left(
\ell^{\theta}\left(  Z_{i},\theta_{0}\right)  -E\left[  \ell^{\theta}\left(
Z_{i},\theta_{0}\right)  \right]  \right)  \right] \\
&  -\frac{1}{2}\mathcal{I}^{-3}E\left[  \ell^{\theta\theta}\right]  \frac
{1}{n}\sum_{i=1}^{n}\left(  \ell\left(  Z_{i},\theta_{0}\right)  ^{2}-E\left[
\ell\left(  Z_{i},\theta_{0}\right)  ^{2}\right]  \right) \\
&  -\mathcal{I}^{-2}\frac{1}{n}\sum_{i=1}^{n}\left(
\begin{array}
[c]{c}%
\ell\left(  Z_{i},\theta_{0}\right)  \left(  \ell^{\theta}\left(  Z_{i}%
,\theta_{0}\right)  +\mathcal{I}\right) \\
-E\left[  \ell\left(  Z_{i},\theta_{0}\right)  \left(  \ell^{\theta}\left(
Z_{i},\theta_{0}\right)  +\mathcal{I}\right)  \right]
\end{array}
\right)  +O_{p}\left(  n^{-1}\right)  .
\end{align*}
Recalling that $E\left[  \ell\left(  Z_{i},\theta_{0}\right)  \right]  =0$ and
using the information equality, we obtain the simplification%
\begin{align*}
\frac{1}{2}\frac{1}{\sqrt{n}}\left(  n\theta^{\epsilon\epsilon}-\sum_{j=1}%
^{n}\theta_{\left(  j\right)  }^{\epsilon\epsilon}\right)   &  =\frac{1}%
{2}\theta^{\epsilon\epsilon}-\frac{1}{2}\mathcal{I}^{-2}E\left[  \ell
^{\theta\theta}\right]  -\mathcal{I}^{-2}E\left[  \ell\left(  Z_{i},\theta
_{0}\right)  \ell^{\theta}\left(  Z_{i},\theta_{0}\right)  \right] \\
&  -\frac{1}{2}\mathcal{I}^{-3}E\left[  \ell^{\theta\theta}\right]  \frac
{1}{n}\sum_{i=1}^{n}\left(  \ell\left(  Z_{i},\theta_{0}\right)  ^{2}-E\left[
\ell\left(  Z_{i},\theta_{0}\right)  ^{2}\right]  \right) \\
&  -\mathcal{I}^{-2}\frac{1}{n}\sum_{i=1}^{n}\left(  \ell\left(  Z_{i}%
,\theta_{0}\right)  \ell^{\theta}\left(  Z_{i},\theta_{0}\right)  -E\left[
\ell\left(  Z_{i},\theta_{0}\right)  \ell^{\theta}\left(  Z_{i},\theta
_{0}\right)  \right]  \right) \\
&  -\mathcal{I}^{-1}\frac{1}{n}\sum_{i=1}^{n}\ell\left(  Z_{i},\theta
_{0}\right) \\
&  +O_{p}\left(  n^{-1}\right)  ,
\end{align*}
or%
\begin{align}
\frac{1}{2}\frac{1}{\sqrt{n}}&\left(  n\theta^{\epsilon\epsilon}-\sum_{j=1}%
^{n}\theta_{\left(  j\right)  }^{\epsilon\epsilon}\right)  \\
&  =\frac{1}%
{2}\theta^{\epsilon\epsilon}-b\left(  \theta_{0}\right) \nonumber\\
&  -\frac{1}{2}\mathcal{I}^{-3}E\left[  \ell^{\theta\theta}\right]  \frac
{1}{\sqrt{n}}\left(  \frac{1}{\sqrt{n}}\sum_{i=1}^{n}\left(  \ell\left(
Z_{i},\theta_{0}\right)  ^{2}-E\left[  \ell\left(  Z_{i},\theta_{0}\right)
^{2}\right]  \right)  \right) \nonumber\\
&  -\mathcal{I}^{-2}\frac{1}{\sqrt{n}}\left(  \frac{1}{\sqrt{n}}\sum_{i=1}%
^{n}\left(  \ell\left(  Z_{i},\theta_{0}\right)  \ell^{\theta}\left(
Z_{i},\theta_{0}\right)  -E\left[  \ell\left(  Z_{i},\theta_{0}\right)
\ell^{\theta}\left(  Z_{i},\theta_{0}\right)  \right]  \right)  \right)
\nonumber\\
&  -\mathcal{I}^{-1}\frac{1}{\sqrt{n}}U\left(  \theta_{0}\right)
+O_{p}\left(  n^{-1}\right)  . \label{r1-3-2}%
\end{align}

As for
\begin{align*}
\frac{1}{6}\theta^{\epsilon\epsilon\epsilon}  &  =T_{3}=\frac{1}{6}%
\mathcal{I}^{-4}\mathcal{Q}_{2}\left(  \theta_{0}\right)  U\left(  \theta
_{0}\right)  ^{3}+\frac{1}{2}\mathcal{I}^{-5}\mathcal{Q}_{1}\left(  \theta
_{0}\right)  ^{2}U\left(  \theta_{0}\right)  ^{3}+\frac{3}{2}\mathcal{I}%
^{-4}\mathcal{Q}_{1}\left(  \theta\right)  U\left(  \theta_{0}\right)
^{2}V\left(  \theta_{0}\right) \\
&  +\frac{1}{2}\mathcal{I}^{-3}U\left(  \theta_{0}\right)  ^{2}W\left(
\theta_{0}\right)  +\mathcal{I}^{-3}U\left(  \theta_{0}\right)  V\left(
\theta_{0}\right)  ^{2}.
\end{align*}
we see that all the terms on the right takes the form of \\ $\left(  n^{-1/2}%
\sum_{i=1}^{n}X_{1,i}\right)  \left(  n^{-1/2}\sum_{i=1}^{n}X_{2,i}\right)
\left(  n^{-1/2}\sum_{i=1}^{n}X_{3,i}\right)  $ where $X_{1,i}$, $X_{2,i}$
and$X_{3,i}$ have zero means. The zero mean assumption implies that%
\[
\frac{n^{2}}{\left(  n-1\right)  ^{2}}\left(  \sqrt{n}\overline{X}_{1}\right)
\left(  \frac{1}{n}\sum_{i=1}^{n}X_{2,i}X_{3,i}\right)  =E\left[
X_{2,i}X_{3,i}\right]  n^{-1/2}\sum_{i=1}^{n}X_{1,i}+o_{p}\left(  1\right)
\]
and likewise%
\begin{align*}
\frac{n^{2}}{\left(  n-1\right)  ^{2}}\left(  \sqrt{n}\overline{X}_{2}\right)
\left(  \frac{1}{n}\sum_{i=1}^{n}X_{3,i}X_{1,i}\right)   &  =E\left[
X_{3,i}X_{1,i}\right]  n^{-1/2}\sum_{i=1}^{n}X_{2,i}+o_{p}\left(  1\right)
,\\
\frac{n^{2}}{\left(  n-1\right)  ^{2}}\left(  \sqrt{n}\overline{X}_{3}\right)
\left(  \frac{1}{n}\sum_{i=1}^{n}X_{1,i}X_{2,i}\right)   &  =E\left[
X_{1,i}X_{2,i}\right]  n^{-1/2}\sum_{i=1}^{n}X_{3,i}+o_{p}\left(  1\right)  .
\end{align*}
Finally,%
\[
\frac{n\sqrt{n}}{\left(  n-1\right)  ^{2}}\left(  \frac{1}{n}\sum_{i=1}%
^{n}X_{1,i}X_{2,i}X_{3,i}\right)  =o_{p}\left(  1\right)  .
\]
It follows that Lemma \ref{Vjack3} applied to $\frac{1}{6}\frac{1}{n}\left(
n\theta^{\epsilon\epsilon\epsilon}-\sqrt{\frac{n}{n-1}}\sum_{j=1}^{n}%
\theta_{\left(  j\right)  }^{\epsilon\epsilon\epsilon}\right)  $ results in%
\begin{align*}
&  \frac{1}{6}\frac{1}{n}\left(  n\theta^{\epsilon\epsilon\epsilon}%
-\sqrt{\frac{n}{n-1}}\sum_{j=1}^{n}\theta_{\left(  j\right)  }^{\epsilon
\epsilon\epsilon}\right) \\
&  =\frac{1}{6}\frac{1}{n}\theta^{\epsilon\epsilon\epsilon}-\frac{1}%
{2}\mathcal{I}^{-4}\mathcal{Q}_{2}\left(  \theta_{0}\right)  E\left[
U_{i}^{2}\right]  U\left(  \theta_{0}\right) \\
&  -\frac{3}{2}\mathcal{I}^{-5}\mathcal{Q}_{1}\left(  \theta_{0}\right)
^{2}E\left[  U_{i}^{2}\right]  U\left(  \theta_{0}\right) \\
&  -3\mathcal{I}^{-4}\mathcal{Q}_{1}\left(  \theta\right)  E\left[  U_{i}%
V_{i}\right]  U\left(  \theta_{0}\right)  -\frac{3}{2}\mathcal{I}%
^{-4}\mathcal{Q}_{1}\left(  \theta\right)  E\left[  U_{i}^{2}\right]  V\left(
\theta_{0}\right) \\
&  -\mathcal{I}^{-3}E\left[  U_{i}W_{i}\right]  U\left(  \theta_{0}\right)
-\frac{1}{2}\mathcal{I}^{-3}E\left[  U_{i}^{2}\right]  W\left(  \theta
_{0}\right) \\
&  -2\mathcal{I}^{-3}E\left[  U_{i}V_{i}\right]  V\left(  \theta_{0}\right)
-\mathcal{I}^{-3}E\left[  V_{i}^{2}\right]  U\left(  \theta_{0}\right)
+o_{p}\left(  1\right)  .
\end{align*}
Using information equality, we simplify a bit and obtain
\begin{align*}
&  \frac{1}{6}\frac{1}{n}\left(  n\theta^{\epsilon\epsilon\epsilon}%
-\sqrt{\frac{n}{n-1}}\sum_{j=1}^{n}\theta_{\left(  j\right)  }^{\epsilon
\epsilon\epsilon}\right) \\
&  =\frac{1}{6}\frac{1}{n}\theta^{\epsilon\epsilon\epsilon}-\frac{1}%
{2}\mathcal{I}^{-3}E\left[  \ell^{\theta\theta\theta}\right]  U\left(
\theta_{0}\right) \\
&  -\frac{3}{2}\mathcal{I}^{-4}\left(  E\left[  \ell^{\theta\theta}\right]
\right)  ^{2}U\left(  \theta_{0}\right) \\
&  -3\mathcal{I}^{-4}E\left[  \ell^{\theta\theta}\right]  E\left[  \ell
\ell^{\theta}\right]  U\left(  \theta_{0}\right)  -\frac{3}{2}\mathcal{I}%
^{-3}E\left[  \ell^{\theta\theta}\right]  V\left(  \theta_{0}\right) \\
&  -\mathcal{I}^{-3}E\left[  \ell\ell^{\theta\theta}\right]  U\left(
\theta_{0}\right)  -\frac{1}{2}\mathcal{I}^{-2}W\left(  \theta_{0}\right) \\
&  -2\mathcal{I}^{-3}E\left[  \ell\ell^{\theta}\right]  V\left(  \theta
_{0}\right)  -\mathcal{I}^{-3}\left(  E\left[  \left(  \ell^{\theta}\right)
^{2}\right]  -\mathcal{I}^{2}\right)  U\left(  \theta_{0}\right)
+o_{p}\left(  1\right)
\end{align*}
or%
\begin{align}
&  \frac{1}{6}\frac{1}{n}\left(  n\theta^{\epsilon\epsilon\epsilon}%
-\sqrt{\frac{n}{n-1}}\sum_{j=1}^{n}\theta_{\left(  j\right)  }^{\epsilon
\epsilon\epsilon}\right) \nonumber\\
&  =\frac{1}{6}\frac{1}{n}\theta^{\epsilon\epsilon\epsilon}-\frac{1}{2}\left(
\mathcal{I}^{-3}E\left[  \ell^{\theta\theta\theta}\right]  +3\mathcal{I}%
^{-4}\left(  E\left[  \ell^{\theta\theta}\right]  \right)  ^{2}\right)
U\left(  \theta_{0}\right) \nonumber\\
&  -3\mathcal{I}^{-4}E\left[  \ell^{\theta\theta}\right]  E\left[  \ell
\ell^{\theta}\right]  U\left(  \theta_{0}\right)  -\frac{3}{2}\mathcal{I}%
^{-3}E\left[  \ell^{\theta\theta}\right]  V\left(  \theta_{0}\right)
\nonumber\\
&  -\mathcal{I}^{-3}E\left[  \ell\ell^{\theta\theta}\right]  U\left(
\theta_{0}\right)  -\frac{1}{2}\mathcal{I}^{-2}W\left(  \theta_{0}\right)
\nonumber\\
&  -2\mathcal{I}^{-3}E\left[  \ell\ell^{\theta}\right]  V\left(  \theta
_{0}\right)  -\mathcal{I}^{-3}E\left[  \left(  \ell^{\theta}\right)
^{2}\right]  U\left(  \theta_{0}\right)  +\mathcal{I}^{-1}U\left(  \theta
_{0}\right)  +o_{p}\left(  1\right)  . \label{r1-3-3}%
\end{align}

Combining (\ref{r1-3-1}), (\ref{r1-3-2}), and (\ref{r1-3-3}), we obtain
\begin{align*}
\lefteqn{\left(  n\theta^{\epsilon}-\sqrt{n}\sqrt{n-1}\frac{1}{n}\sum
_{j=1}^{n}\theta_{\left(  j\right)  }^{\epsilon}\right)  +\frac{1}{2}\frac
{1}{\sqrt{n}}\left(  n\theta^{\epsilon\epsilon}-\sum_{j=1}^{n}\theta_{\left(
j\right)  }^{\epsilon\epsilon}\right)  }\\
&  +\frac{1}{6}\frac{1}{n}\left(  n\theta^{\epsilon\epsilon\epsilon}%
-\sqrt{\frac{n}{n-1}}\sum_{j=1}^{n}\theta_{\left(  j\right)  }^{\epsilon
\epsilon\epsilon}\right) \\
&  =\theta^{\epsilon}+\frac{1}{\sqrt{n}}\left(  \frac{1}{2}\theta
^{\epsilon\epsilon}-b\left(  \theta_{0}\right)  \right)  +\frac{1}{6}\frac
{1}{n}\theta^{\epsilon\epsilon\epsilon}-n^{-1}\left(  \mathbb{J}+o_{p}\left(
1\right)  \right)  ,
\end{align*}
where%
\begin{align*}
\mathbb{J}  &  =\frac{1}{2}\mathcal{I}^{-3}E\left[  \ell^{\theta\theta
}\right]  \left(  \frac{1}{\sqrt{n}}\sum_{i=1}^{n}\left(  \ell\left(
Z_{i},\theta_{0}\right)  ^{2}-E\left[  \ell\left(  Z_{i},\theta_{0}\right)
^{2}\right]  \right)  \right) \\
&  +\mathcal{I}^{-2}\left(  \frac{1}{\sqrt{n}}\sum_{i=1}^{n}\left(
\ell\left(  Z_{i},\theta_{0}\right)  \ell^{\theta}\left(  Z_{i},\theta
_{0}\right)  -E\left[  \ell\left(  Z_{i},\theta_{0}\right)  \ell^{\theta
}\left(  Z_{i},\theta_{0}\right)  \right]  \right)  \right) \\
&  +\mathcal{I}^{-1}U\left(  \theta_{0}\right) \\
&  +\frac{1}{2}\left(  \mathcal{I}^{-3}E\left[  \ell^{\theta\theta\theta
}\right]  +3\mathcal{I}^{-4}\left(  E\left[  \ell^{\theta\theta}\right]
\right)  ^{2}\right)  U\left(  \theta_{0}\right) \\
&  +3\mathcal{I}^{-4}E\left[  \ell^{\theta\theta}\right]  E\left[  \ell
\ell^{\theta}\right]  U\left(  \theta_{0}\right)  +\frac{3}{2}\mathcal{I}%
^{-3}E\left[  \ell^{\theta\theta}\right]  V\left(  \theta_{0}\right) \\
&  +\mathcal{I}^{-3}E\left[  \ell\ell^{\theta\theta}\right]  U\left(
\theta_{0}\right)  +\frac{1}{2}\mathcal{I}^{-2}W\left(  \theta_{0}\right) \\
&  +2\mathcal{I}^{-3}E\left[  \ell\ell^{\theta}\right]  V\left(  \theta
_{0}\right)  +\mathcal{I}^{-3}E\left[  \left(  \ell^{\theta}\right)
^{2}\right]  U\left(  \theta_{0}\right)  -\mathcal{I}^{-1}U\left(  \theta
_{0}\right)  .
\end{align*}
Organizing a bit, we obtain
\begin{align*}
2\mathbb{J}  &  =\left(  \mathcal{I}^{-3}E\left[  \ell^{\theta\theta\theta
}\right]  +3\mathcal{I}^{-4}\left(  E\left[  \ell^{\theta\theta}\right]
\right)  ^{2}\right)  U\left(  \theta_{0}\right)  +6\mathcal{I}^{-4}E\left[
\ell^{\theta\theta}\right]  E\left[  U_{i}V_{i}\right]  U\left(  \theta
_{0}\right) \\
&  +\frac{3}{\mathcal{I}^{3}}E\left[  \ell^{\theta\theta}\right]  V\left(
\theta_{0}\right)  +\frac{2}{\mathcal{I}^{3}}E\left[  U_{i}W_{i}\right]
U\left(  \theta_{0}\right)  +\frac{1}{\mathcal{I}^{2}}W(\theta_{0})\\
&  +\frac{2}{\mathcal{I}^{3}}E\left[  V_{i}^{2}\right]  U\left(  \theta
_{0}\right)  +\frac{4}{\mathcal{I}^{3}}E\left[  U_{i}V_{i}\right]
V(\theta_{0})\\
&  +\mathcal{I}^{-3}E\left[  \ell^{\theta\theta}\right]  n^{-1/2}\left(
{\textstyle\sum_{i=1}^{n}}
\left[  \ell\left(  Z_{i},\theta_{0}\right)  ^{2}-E\left[  \ell\left(
Z_{i},\theta_{0}\right)  ^{2}\right]  \right]  \right) \\
&  +2\mathcal{I}^{-2}n^{-1/2}\left(
{\textstyle\sum_{i=1}^{n}}
\ell\left(  Z_{i},\theta_{0}\right)  \ell^{\theta}\left(  Z_{i},\theta
_{0}\right)  -E\left[  \ell\left(  Z_{i},\theta_{0}\right)  \ell^{\theta
}\left(  Z_{i},\theta_{0}\right)  \right]  \right)  .
\end{align*}

\section{Moments of Bootstrapped and Jackknifed Statistics\label{sec-bootm}}

\begin{lemma}
\label{L1}Let $X_{i}^{\ast}=\tau\left(  Z_{i}^{\ast},\widehat{\theta}\right)
$ be some transformation of $Z_{i}^{\ast}$, where $\tau$ possibly depends on
the sample $\left\{  Z_{i}\right\}  _{i=1}^{n}$ through $\widehat{\theta}$.
Then
\[
E^{\ast}\left[  \tfrac{1}{\sqrt{n}}%
{\textstyle\sum_{i=1}^{n}}
X_{i}^{\ast}\right]  =\tfrac{1}{\sqrt{n}}%
{\textstyle\sum_{i=1}^{n}}
X_{i}%
\]
where $X_{i}=\tau\left(  Z_{i},\widehat{\theta}\right)  $.
\end{lemma}

\begin{lemma}
\label{L2}Let $X_{k,i}^{\ast}=\tau_{k}\left(  Z_{i}^{\ast},\widehat{\theta
}\right)  $ for $k=1,2$ be some transformation of $Z_{i}^{\ast}$, where
$\tau_{k}$ possibly depends on the sample $\left\{  Z_{i}\right\}  _{i=1}^{n}$
through $\widehat{\theta}$. Then
\begin{align*}
&  E^{\ast}\left[  \left(  \tfrac{1}{\sqrt{n}}%
{\textstyle\sum_{i=1}^{n}}
X_{1,i}^{\ast}\right)  \left(  \tfrac{1}{\sqrt{n}}%
{\textstyle\sum_{i=1}^{n}}
X_{2,i}^{\ast}\right)  \right] \\
&  =\tfrac{1}{n}%
{\textstyle\sum_{i=1}^{n}}
X_{1,i}X_{2,i}+\dfrac{n-1}{n}\left(  \tfrac{1}{\sqrt{n}}%
{\textstyle\sum_{i=1}^{n}}
X_{1,i}\right)  \left(  \tfrac{1}{\sqrt{n}}%
{\textstyle\sum_{i=1}^{n}}
X_{2,i}\right)
\end{align*}
where $X_{k,i}=\tau_{k}(Z_{i},\widehat{\theta}).$
\end{lemma}

\begin{lemma}
\label{L3}Let $X_{k,i}^{\ast}=\tau_{k}\left(  Z_{i}^{\ast},\widehat{\theta
}\right)  $ for $k=1,2$ be some transformation of $Z_{i}^{\ast}$, where
$\tau_{k}$ possibly depends on the sample $\left\{  Z_{i}\right\}  _{i=1}^{n}$
through $\widehat{\theta}.$ Then
\begin{align*}
&  E^{\ast}\left[  \left(  \tfrac{1}{\sqrt{n}}%
{\textstyle\sum_{i=1}^{n}}
X_{1,i}^{\ast}\right)  \left(  \tfrac{1}{\sqrt{n}}%
{\textstyle\sum_{i=1}^{n}}
X_{2,i}^{\ast}\right)  \left(  \tfrac{1}{\sqrt{n}}%
{\textstyle\sum_{i=1}^{n}}
X_{3,i}^{\ast}\right)  \right] \\
&  =\tfrac{1}{\sqrt{n}}\tfrac{1}{n}%
{\textstyle\sum_{j=1}^{n}}
X_{1,j}X_{2,j}X_{3,j}+\frac{n-1}{n}\left(  \tfrac{1}{n}%
{\textstyle\sum_{j=1}^{n}}
X_{1,j}X_{2,j}\right)  \left(  \tfrac{1}{\sqrt{n}}%
{\textstyle\sum_{j=1}^{n}}
X_{3,j}\right) \\
&  +\frac{n-1}{n}\left(  \tfrac{1}{n}%
{\textstyle\sum_{j=1}^{n}}
X_{3,j}X_{1,j}\right)  \left(  \tfrac{1}{\sqrt{n}}%
{\textstyle\sum_{j=1}^{n}}
X_{2,j}\right)  +\frac{n-1}{n}\left(  \tfrac{1}{n}%
{\textstyle\sum_{j=1}^{n}}
X_{2,j}X_{3,j}\right)  \left(  \tfrac{1}{\sqrt{n}}%
{\textstyle\sum_{j=1}^{n}}
X_{1,j}\right) \\
&  +\frac{n^{2}-3n+2}{n^{2}}\left(  \tfrac{1}{\sqrt{n}}%
{\textstyle\sum_{j=1}^{n}}
X_{1,j}\right)  \left(  \tfrac{1}{\sqrt{n}}%
{\textstyle\sum_{j=1}^{n}}
X_{2,j}\right)  \left(  \tfrac{1}{\sqrt{n}}%
{\textstyle\sum_{j=1}^{n}}
X_{3,j}\right)  .
\end{align*}

\end{lemma}

\begin{lemma}
\label{BExp}Let $U_{i}^{\ast}\left(  \theta\right)  \equiv\ell\left(
Z_{i}^{\ast},\theta\right)  $, $V_{i}^{\ast}\left(  \theta\right)  \equiv
\ell^{\theta}\left(  Z_{i}^{\ast},\theta\right)  -\overline{\ell^{\theta
}\left(  \cdot,\theta\right)  }\equiv\ell^{\theta}\left(  Z_{i}^{\ast}%
,\theta\right)  -n^{-1}\sum_{i=1}^{n}\ell^{\theta}\left(  Z_{i},\theta\right)
$, $W_{i}^{\ast}\left(  \theta\right)  \equiv\ell^{\theta\theta}\left(
Z_{i}^{\ast}\right)  -\overline{\ell^{\theta\theta}\left(  \cdot
,\theta\right)  }\equiv\ell^{\theta\theta}\left(  Z_{i}^{\ast}\right)
-n^{-1}\sum_{i=1}^{n}\ell^{\theta\theta}\left(  Z_{i},\theta\right)  $ and let
$U^{\ast}\left(  \theta\right)  =n^{-1/2}\sum_{i=1}^{n}U_{i}^{\ast}\left(
\theta\right)  $, $V^{\ast}\left(  \theta\right)  =n^{-1/2}\sum_{i=1}^{n}%
V_{i}^{\ast}\left(  \theta\right)  $ and $W^{\ast}\left(  \theta\right)
=n^{-1/2}\sum_{i=1}^{n}W_{i}^{\ast}\left(  \theta\right)  $. Then \newline(a)
\begin{align*}
E^{\ast}\left[  U^{\ast}\left(  \widehat{\theta}\right)  \right]   &  =0,\\
E^{\ast}\left[  V^{\ast}\left(  \widehat{\theta}\right)  \right]   &  =0,\\
E^{\ast}\left[  W^{\ast}\left(  \widehat{\theta}\right)  \right]   &  =0,
\end{align*}
(b)
\begin{align*}
E^{\ast}\left[  U^{\ast}\left(  \widehat{\theta}\right)  ^{2}\right]   &
=\tfrac{1}{n}%
{\textstyle\sum_{i=1}^{n}}
\ell\left(  Z_{i},\widehat{\theta}\right)  ^{2}\\
E^{\ast}\left[  U^{\ast}\left(  \widehat{\theta}\right)  V^{\ast}\left(
\widehat{\theta}\right)  \right]   &  =\tfrac{1}{n}%
{\textstyle\sum_{i=1}^{n}}
\ell\left(  Z_{i},\widehat{\theta}\right)  \ell^{\theta}\left(  Z_{i}%
,\widehat{\theta}\right)
\end{align*}
(c)
\begin{align*}
E^{\ast}\left[  U^{\ast}\left(  \theta\right)  ^{3}\right]   &  =\tfrac
{1}{\sqrt{n}}\tfrac{1}{n}%
{\textstyle\sum_{j=1}^{n}}
\ell\left(  Z_{i},\widehat{\theta}\right)  ^{3},\\
E^{\ast}\left[  U^{\ast}\left(  \widehat{\theta}\right)  ^{2}V^{\ast}\left(
\widehat{\theta}\right)  \right]   &  =\tfrac{1}{\sqrt{n}}\tfrac{1}{n}%
{\textstyle\sum_{j=1}^{n}}
\ell\left(  Z_{i},\widehat{\theta}\right)  ^{2}\left(  \ell^{\theta}\left(
Z_{i},\widehat{\theta}\right)  -\overline{\ell^{\theta}\left(  \cdot
,\widehat{\theta}\right)  }\right) \\
E^{\ast}\left[  U^{\ast}\left(  \widehat{\theta}\right)  ^{2}W^{\ast}\left(
\widehat{\theta}\right)  \right]   &  =\tfrac{1}{\sqrt{n}}\tfrac{1}{n}%
{\textstyle\sum_{j=1}^{n}}
\ell\left(  Z_{i},\widehat{\theta}\right)  ^{2}\left(  \ell^{\theta\theta
}\left(  Z_{i},\widehat{\theta}\right)  -\overline{\ell^{\theta\theta}\left(
\cdot,\widehat{\theta}\right)  }\right) \\
E^{\ast}\left[  U^{\ast}\left(  \widehat{\theta}\right)  V^{\ast}\left(
\widehat{\theta}\right)  ^{2}\right]   &  =\tfrac{1}{\sqrt{n}}\tfrac{1}{n}%
{\textstyle\sum_{j=1}^{n}}
\ell\left(  Z_{i},\widehat{\theta}\right)  \left(  \ell^{\theta}\left(
Z_{i},\widehat{\theta}\right)  -\overline{\ell^{\theta}\left(  \cdot
,\widehat{\theta}\right)  }\right)  ^{2}%
\end{align*}

\end{lemma}

\begin{lemma}
\label{Vjack1}Let
\[
W=\frac{1}{\sqrt{n}}\sum_{i=1}^{n}X_{i},\qquad W_{\left(  j\right)  }=\frac
{1}{\sqrt{n-1}}\sum_{i\neq j}X_{i}%
\]
Then, we have
\[
nW-\sqrt{n}\sqrt{n-1}\frac{1}{n}\sum_{j=1}^{n}W_{\left(  j\right)  }=W
\]

\end{lemma}

\begin{lemma}
\label{Vjack2}Let
\[
W=\left(  \frac{1}{\sqrt{n}}\sum_{i=1}^{n}X_{1,i}\right)  \left(  \frac
{1}{\sqrt{n}}\sum_{i=1}^{n}X_{2,i}\right)  ,\qquad W_{\left(  j\right)
}=\left(  \frac{1}{\sqrt{n-1}}\sum_{i\neq j}X_{1,i}\right)  \left(  \frac
{1}{\sqrt{n-1}}\sum_{i\neq j}X_{2,i}\right)
\]
Then,
\[
nW-\sum_{j=1}^{n}W_{\left(  j\right)  }=\frac{1}{n-1}\sum_{i\neq j}%
X_{1,i}X_{2,j}%
\]

\end{lemma}

\begin{lemma}
\label{Vjack3}Let
\begin{align*}
W  &  =\left(  \frac{1}{\sqrt{n}}\sum_{i=1}^{n}X_{1,i}\right)  \left(
\frac{1}{\sqrt{n}}\sum_{i=1}^{n}X_{2,i}\right)  \left(  \frac{1}{\sqrt{n}}%
\sum_{i=1}^{n}X_{3,i}\right)  ,\\
W_{\left(  j\right)  }  &  =\left(  \frac{1}{\sqrt{n-1}}\sum_{i\neq j}%
X_{1,i}\right)  \left(  \frac{1}{\sqrt{n-1}}\sum_{i\neq j}X_{2,i}\right)
\left(  \frac{1}{\sqrt{n-1}}\sum_{i\neq j}X_{3,i}\right)
\end{align*}
Then,
\begin{align*}
&  nW-\sqrt{\frac{n}{n-1}}\sum_{j=1}^{n}W_{\left(  j\right)  }\\
&  =\frac{n^{2}+n}{\left(  n-1\right)  ^{2}}W\\
&  -\frac{n^{2}}{\left(  n-1\right)  ^{2}}\left(  \sqrt{n}\overline{X}%
_{1}\right)  \left(  \frac{1}{n}\sum_{i=1}^{n}X_{2,i}X_{3,i}\right)
-\frac{n^{2}}{\left(  n-1\right)  ^{2}}\left(  \sqrt{n}\overline{X}%
_{2}\right)  \left(  \frac{1}{n}\sum_{i=1}^{n}X_{3,i}X_{1,i}\right) \\
&  -\frac{n^{2}}{\left(  n-1\right)  ^{2}}\left(  \sqrt{n}\overline{X}%
_{3}\right)  \left(  \frac{1}{n}\sum_{i=1}^{n}X_{1,i}X_{2,i}\right)
+\frac{n\sqrt{n}}{\left(  n-1\right)  ^{2}}\left(  \frac{1}{n}\sum_{i=1}%
^{n}X_{1,i}X_{2,i}X_{3,i}\right)
\end{align*}

\end{lemma}

\section{Proofs for Bootstrap and Jackknife Moments}

\subsection{Proof of Lemma \ref{L1}}

Note that
\[
E^{\ast}\left[  X_{i}^{\ast}\right]  =\tfrac{1}{n}%
{\textstyle\sum_{j=1}^{n}}
X_{j}%
\]
which in turn implies
\[
E^{\ast}\left(  \tfrac{1}{\sqrt{n}}%
{\textstyle\sum_{i=1}^{n}}
X_{i}^{\ast}\right)  =\tfrac{1}{\sqrt{n}}%
{\textstyle\sum_{i=1}^{n}}
\left(  \tfrac{1}{n}%
{\textstyle\sum_{j=1}^{n}}
X_{j}\right)  =\tfrac{1}{n}%
{\textstyle\sum_{j=1}^{n}}
X_{j}.
\]

\subsection{Proof of Lemma \ref{L2}}

This is because
\[
E^{\ast}\left[  X_{1,i}^{\ast}X_{2,i}^{\ast}\right]  =\tfrac{1}{n}%
{\textstyle\sum_{j=1}^{n}}
X_{1,j}X_{2,j},
\]
which in turn implies that
\begin{align*}
&  E^{\ast}\left[  \left(  \tfrac{1}{\sqrt{n}}%
{\textstyle\sum_{i=1}^{n}}
X_{1,i}^{\ast}\right)  \left(  \tfrac{1}{\sqrt{n}}%
{\textstyle\sum_{i=1}^{n}}
X_{2,i}^{\ast}\right)  \right] \\
&  =\tfrac{1}{n}%
{\textstyle\sum_{i=1}^{n}}
E^{\ast}\left[  X_{1,i}^{\ast}X_{2,i}^{\ast}\right]  +\tfrac{1}{n}%
{\textstyle\sum_{i\neq i^{\prime}}^{n}}
E^{\ast}\left[  X_{1,i}^{\ast}\right]  E\left[  X_{2,i^{\prime}}^{\ast}\right]
\\
&  =\tfrac{1}{n}%
{\textstyle\sum_{i=1}^{n}}
\left(  \tfrac{1}{n}%
{\textstyle\sum_{j=1}^{n}}
X_{1,j}X_{2,j}\right)  +\tfrac{1}{n}%
{\textstyle\sum_{i\neq i^{\prime}}^{n}}
\left(  \tfrac{1}{n}%
{\textstyle\sum_{j=1}^{n}}
X_{1,j}\right)  \left(  \tfrac{1}{n}%
{\textstyle\sum_{j=1}^{n}}
X_{2,j}\right) \\
&  =\tfrac{1}{n}%
{\textstyle\sum_{j=1}^{n}}
X_{1,j}X_{2,j}+\frac{n\left(  n-1\right)  }{n}\left(  \tfrac{1}{n}%
{\textstyle\sum_{j=1}^{n}}
X_{1,j}\right)  \left(  \tfrac{1}{n}%
{\textstyle\sum_{j=1}^{n}}
X_{2,j}\right) \\
&  =\tfrac{1}{n}%
{\textstyle\sum_{j=1}^{n}}
X_{1,j}X_{2,j}+\frac{n-1}{n}\left(  \tfrac{1}{\sqrt{n}}%
{\textstyle\sum_{j=1}^{n}}
X_{1,j}\right)  \left(  \tfrac{1}{\sqrt{n}}%
{\textstyle\sum_{j=1}^{n}}
X_{2,j}\right)  .
\end{align*}
The conclusion follows easily from this result.

\subsection{Proof of Lemma \ref{L3}}

Note that
\begin{align*}
&  \left(
{\textstyle\sum_{i=1}^{n}}
X_{1,i}^{\ast}\right)  \left(
{\textstyle\sum_{i=1}^{n}}
X_{2,i}^{\ast}\right)  \left(
{\textstyle\sum_{i=1}^{n}}
X_{3,i}^{\ast}\right) \\
&  =%
{\textstyle\sum_{i=1}^{n}}
X_{1,i}^{\ast}X_{2,i}^{\ast}X_{3,i}^{\ast}\\
&  +%
{\textstyle\sum_{i\neq i^{\prime}}}
X_{1,i}^{\ast}X_{2,i}^{\ast}X_{3,i^{\prime}}^{\ast}+%
{\textstyle\sum_{i\neq i^{\prime}}}
X_{1,i}^{\ast}X_{2,i^{\prime}}^{\ast}X_{3,i}^{\ast}+%
{\textstyle\sum_{i\neq i^{\prime}}}
X_{1,i^{\prime}}^{\ast}X_{2,i}^{\ast}X_{3,i}^{\ast}\\
&  +%
{\textstyle\sum_{i\neq i^{\prime}\neq i^{\prime\prime}\neq i}}
X_{1,i}^{\ast}X_{2,i^{\prime}}^{\ast}X_{3,i^{\prime\prime}}^{\ast}%
\end{align*}
Therefore, we have
\begin{align*}%
{\textstyle\sum_{i=1}^{n}}
E^{\ast}\left[  X_{1,i}^{\ast}X_{2,i}^{\ast}X_{3,i}^{\ast}\right]   &  =%
{\textstyle\sum\nolimits_{i=1}^{n}}
\left(  \tfrac{1}{n}%
{\textstyle\sum_{j=1}^{n}}
X_{1,j}X_{2,j}X_{3,j}\right) \\
&  =%
{\textstyle\sum_{j=1}^{n}}
X_{1,j}X_{2,j}X_{3,j}%
\end{align*}%
\begin{align*}%
{\textstyle\sum_{i\neq i^{\prime}}}
E^{\ast}\left[  X_{1,i}^{\ast}X_{2,i}^{\ast}X_{3,i^{\prime}}^{\ast}\right]
&  =%
{\textstyle\sum_{i\neq i^{\prime}}}
E^{\ast}\left[  X_{1,i}^{\ast}X_{2,i}^{\ast}\right]  E^{\ast}\left[
X_{3,i^{\prime}}^{\ast}\right] \\
&  =%
{\textstyle\sum_{i\neq i^{\prime}}}
\left(  \tfrac{1}{n}%
{\textstyle\sum_{j=1}^{n}}
X_{1,j}X_{2,j}\right)  \left(  \tfrac{1}{n}%
{\textstyle\sum_{j=1}^{n}}
X_{3,j}\right) \\
&  =\frac{n-1}{n}\left(
{\textstyle\sum_{j=1}^{n}}
X_{1,j}X_{2,j}\right)  \left(
{\textstyle\sum_{j=1}^{n}}
X_{3,j}\right)
\end{align*}%
\begin{align*}%
{\textstyle\sum_{i\neq i^{\prime}}}
E^{\ast}\left[  X_{1,i}^{\ast}X_{2,i^{\prime}}^{\ast}X_{3,i}^{\ast}\right]
&  =\frac{n-1}{n}\left(
{\textstyle\sum_{j=1}^{n}}
X_{3,j}X_{1,j}\right)  \left(
{\textstyle\sum_{j=1}^{n}}
X_{2,j}\right) \\%
{\textstyle\sum_{i\neq i^{\prime}}}
E^{\ast}\left[  X_{1,i^{\prime}}^{\ast}X_{2,i}^{\ast}X_{3,i}^{\ast}\right]
&  =\frac{n-1}{n}\left(
{\textstyle\sum_{j=1}^{n}}
X_{2,j}X_{3,j}\right)  \left(
{\textstyle\sum_{j=1}^{n}}
X_{1,j}\right)
\end{align*}
and
\begin{align*}%
{\textstyle\sum_{i\neq i^{\prime}\neq i^{\prime\prime}\neq i}}
E^{\ast}\left[  X_{1,i}^{\ast}X_{2,i^{\prime}}^{\ast}X_{3,i^{\prime\prime}%
}^{\ast}\right]   &  =%
{\textstyle\sum_{i\neq i^{\prime}\neq i^{\prime\prime}\neq i}}
E^{\ast}\left[  X_{1,i}^{\ast}\right]  E^{\ast}\left[  X_{2,i^{\prime}}^{\ast
}\right]  E^{\ast}\left[  X_{3,i^{\prime\prime}}^{\ast}\right] \\
&  =%
{\textstyle\sum_{i\neq i^{\prime}\neq i^{\prime\prime}\neq i}}
\left(  \tfrac{1}{n}%
{\textstyle\sum_{j=1}^{n}}
X_{1,j}\right)  \left(  \tfrac{1}{n}%
{\textstyle\sum_{j=1}^{n}}
X_{2,j}\right)  \left(  \tfrac{1}{n}%
{\textstyle\sum_{j=1}^{n}}
X_{3,j}\right) \\
&  =\frac{n^{2}-3n+2}{n^{2}}\left(
{\textstyle\sum_{j=1}^{n}}
X_{1,j}\right)  \left(
{\textstyle\sum_{j=1}^{n}}
X_{2,j}\right)  \left(
{\textstyle\sum_{j=1}^{n}}
X_{3,j}\right)
\end{align*}
from which we obtain
\begin{align*}
&  E^{\ast}\left[  \left(
{\textstyle\sum_{i=1}^{n}}
X_{1,i}^{\ast}\right)  \left(
{\textstyle\sum_{i=1}^{n}}
X_{2,i}^{\ast}\right)  \left(
{\textstyle\sum_{i=1}^{n}}
X_{3,i}^{\ast}\right)  \right] \\
&  =%
{\textstyle\sum_{j=1}^{n}}
X_{1,j}X_{2,j}X_{3,j}+\frac{n-1}{n}\left(
{\textstyle\sum_{j=1}^{n}}
X_{1,j}X_{2,j}\right)  \left(
{\textstyle\sum_{j=1}^{n}}
X_{3,j}\right) \\
&  +\frac{n-1}{n}\left(
{\textstyle\sum_{j=1}^{n}}
X_{3,j}X_{1,j}\right)  \left(
{\textstyle\sum_{j=1}^{n}}
X_{2,j}\right)  +\frac{n-1}{n}\left(
{\textstyle\sum_{j=1}^{n}}
X_{2,j}X_{3,j}\right)  \left(
{\textstyle\sum_{j=1}^{n}}
X_{1,j}\right) \\
&  +\frac{n^{2}-3n+2}{n^{2}}\left(
{\textstyle\sum_{j=1}^{n}}
X_{1,j}\right)  \left(
{\textstyle\sum_{j=1}^{n}}
X_{2,j}\right)  \left(
{\textstyle\sum_{j=1}^{n}}
X_{3,j}\right)  .
\end{align*}

\subsection{Proof of Lemma \ref{BExp}}

For part $\left(  a\right)  $ it follows from Lemma \ref{L1} that
\[
E^{\ast}\left[  U_{i}^{\ast}\left(  \widehat{\theta}\right)  \right]  =n^{-1}%
{\textstyle\sum_{i=1}^{n}}
\ell\left(  Z_{i},\widehat{\theta}\right)  =0,
\]
where the last inequality follows from the definition of $\widehat{\theta}.$
The remaining results can be established in the same way. For part $\left(
b\right)  $ we note that it follows from Lemma \ref{L2} that
\begin{align*}
E^{\ast}\left[  U^{\ast}\left(  \widehat{\theta}\right)  ^{2}\right]   &
=\tfrac{1}{n}%
{\textstyle\sum_{i=1}^{n}}
\ell\left(  Z_{i},\widehat{\theta}\right)  ^{2}+\frac{n-1}{n}\left(  \tfrac
{1}{\sqrt{n}}%
{\textstyle\sum_{i=1}^{n}}
\ell\left(  Z_{i},\widehat{\theta}\right)  \right)  ^{2}\\
&  =\tfrac{1}{n}%
{\textstyle\sum_{i=1}^{n}}
\ell\left(  Z_{i},\widehat{\theta}\right)  ^{2},
\end{align*}
where the last equality is based on $%
{\textstyle\sum_{i=1}^{n}}
\ell\left(  Z_{i},\widehat{\theta}\right)  =0$. For the second part we use
\begin{align*}
E^{\ast}&\left[  U^{\ast}\left(  \widehat{\theta}\right)  V^{\ast}\left(
\widehat{\theta}\right)  \right]  \\
&  =\tfrac{1}{n}%
{\textstyle\sum_{i=1}^{n}}
\ell\left(  Z_{i},\widehat{\theta}\right)  \left(  \ell^{\theta}\left(
Z_{i},\widehat{\theta}\right)  -\left(  \tfrac{1}{n}%
{\textstyle\sum_{i=1}^{n}}
\ell^{\theta}\left(  Z_{i},\widehat{\theta}\right)  \right)  \right) \\
&  +\frac{n-1}{n}\left(  \tfrac{1}{\sqrt{n}}%
{\textstyle\sum_{i=1}^{n}}
\ell\left(  Z_{i},\widehat{\theta}\right)  \right)  \left(  \tfrac{1}{\sqrt
{n}}%
{\textstyle\sum_{i=1}^{n}}
\left(  \ell^{\theta}\left(  Z_{i},\widehat{\theta}\right)  -\left(  \tfrac
{1}{n}%
{\textstyle\sum_{i=1}^{n}}
\ell^{\theta}\left(  Z_{j},\widehat{\theta}\right)  \right)  \right)  \right)
\\
&  =\tfrac{1}{n}%
{\textstyle\sum_{i=1}^{n}}
\ell\left(  Z_{i},\widehat{\theta}\right)  \ell^{\theta}\left(  Z_{i}%
,\widehat{\theta}\right)  .
\end{align*}
For part $\left(  c\right)  $ we use Lemma \ref{L3} to obtain
\begin{align*}
E^{\ast}\left[  U^{\ast}\left(  \widehat{\theta}\right)  ^{3}\right]   &
=\tfrac{1}{\sqrt{n}}\tfrac{1}{n}%
{\textstyle\sum_{j=1}^{n}}
\ell\left(  Z_{i},\widehat{\theta}\right)  ^{3}+3\frac{n-1}{n}\left(
\tfrac{1}{n}%
{\textstyle\sum_{j=1}^{n}}
\ell\left(  Z_{i},\widehat{\theta}\right)  ^{2}\right)  \left(  \tfrac
{1}{\sqrt{n}}%
{\textstyle\sum_{j=1}^{n}}
\ell\left(  Z_{i},\widehat{\theta}\right)  \right) \\
&  +\frac{n^{2}-3n+2}{n^{2}}\left(  \tfrac{1}{\sqrt{n}}%
{\textstyle\sum_{i=1}^{n}}
\ell\left(  Z_{i},\widehat{\theta}\right)  \right)  ^{3},
\end{align*}%
\begin{align*}
E^{\ast}\left[  U^{\ast}\left(  \widehat{\theta}\right)  ^{2}V^{\ast}\left(
\widehat{\theta}\right)  \right]   &  =\tfrac{1}{\sqrt{n}}\tfrac{1}{n}%
{\textstyle\sum_{j=1}^{n}}
U_{j}\left(  \widehat{\theta}\right)  ^{2}V_{j}\left(  \widehat{\theta
}\right)  \\
&+\frac{n-1}{n}\left(  \tfrac{1}{n}%
{\textstyle\sum_{j=1}^{n}}
U_{j}\left(  \widehat{\theta}\right)  ^{2}\right)  \left(  \tfrac{1}{\sqrt{n}}%
{\textstyle\sum_{j=1}^{n}}
V_{j}\left(  \widehat{\theta}\right)  \right) \\
&  +2\frac{n-1}{n}\left(  \tfrac{1}{n}%
{\textstyle\sum_{j=1}^{n}}
U_{j}\left(  \widehat{\theta}\right)  V_{j}\left(  \widehat{\theta}\right)
\right)  \left(  \tfrac{1}{\sqrt{n}}%
{\textstyle\sum_{j=1}^{n}}
U_{j}\left(  \widehat{\theta}\right)  \right) \\
&  +\frac{n^{2}-3n+2}{n^{2}}\left(  \tfrac{1}{\sqrt{n}}%
{\textstyle\sum_{j=1}^{n}}
U_{j}\left(  \widehat{\theta}\right)  \right)  ^{2}\left(  \tfrac{1}{\sqrt{n}}%
{\textstyle\sum_{j=1}^{n}}
V_{j}\left(  \widehat{\theta}\right)  \right)  ,
\end{align*}%
\begin{align*}
E^{\ast}\left[  U^{\ast}\left(  \widehat{\theta}\right)  ^{2}W^{\ast}\left(
\widehat{\theta}\right)  \right]   &  =\tfrac{1}{\sqrt{n}}\tfrac{1}{n}%
{\textstyle\sum_{j=1}^{n}}
U_{j}\left(  \widehat{\theta}\right)  ^{2}W_{j}\left(  \widehat{\theta
}\right)  \\
&+\frac{n-1}{n}\left(  \tfrac{1}{n}%
{\textstyle\sum_{j=1}^{n}}
U_{j}\left(  \widehat{\theta}\right)  ^{2}\right)  \left(  \tfrac{1}{\sqrt{n}}%
{\textstyle\sum_{j=1}^{n}}
W_{j}\left(  \widehat{\theta}\right)  \right) \\
&  +2\frac{n-1}{n}\left(  \tfrac{1}{n}%
{\textstyle\sum_{j=1}^{n}}
U_{j}\left(  \widehat{\theta}\right)  W_{j}\left(  \widehat{\theta}\right)
\right)  \left(  \tfrac{1}{\sqrt{n}}%
{\textstyle\sum_{j=1}^{n}}
U_{j}\left(  \widehat{\theta}\right)  \right) \\
&  +\frac{n^{2}-3n+2}{n^{2}}\left(  \tfrac{1}{\sqrt{n}}%
{\textstyle\sum_{j=1}^{n}}
U_{j}\left(  \widehat{\theta}\right)  \right)  ^{2}\left(  \tfrac{1}{\sqrt{n}}%
{\textstyle\sum_{j=1}^{n}}
W_{j}\left(  \widehat{\theta}\right)  \right)  ,
\end{align*}
and
\begin{align*}
E^{\ast}&\left[  U^{\ast}\left(  \widehat{\theta}\right)  V^{\ast}\left(
\widehat{\theta}\right)  ^{2}\right]  \\
&  =\tfrac{1}{\sqrt{n}}\tfrac{1}{n}%
{\textstyle\sum_{j=1}^{n}}
U_{j}\left(  \widehat{\theta}\right)  V_{j}\left(  \widehat{\theta}\right)
^{2}+\frac{n-1}{n}\left(  \tfrac{1}{n}%
{\textstyle\sum_{j=1}^{n}}
V_{j}\left(  \widehat{\theta}\right)  ^{2}\right)  \left(  \tfrac{1}{\sqrt{n}}%
{\textstyle\sum_{j=1}^{n}}
U_{j}\left(  \widehat{\theta}\right)  \right) \\
&  +2\frac{n-1}{n}\left(  \tfrac{1}{n}%
{\textstyle\sum_{j=1}^{n}}
U_{j}\left(  \widehat{\theta}\right)  V_{j}\left(  \widehat{\theta}\right)
\right)  \left(  \tfrac{1}{\sqrt{n}}%
{\textstyle\sum_{j=1}^{n}}
V_{j}\left(  \widehat{\theta}\right)  \right) \\
&  +\frac{n^{2}-3n+2}{n^{2}}\left(  \tfrac{1}{\sqrt{n}}%
{\textstyle\sum_{j=1}^{n}}
U_{j}\left(  \widehat{\theta}\right)  \right)  \left(  \tfrac{1}{\sqrt{n}}%
{\textstyle\sum_{j=1}^{n}}
V_{j}\left(  \widehat{\theta}\right)  \right)  ^{2}%
\end{align*}
such that the result follows.

\subsection{Proof of Lemma \ref{Vjack1}}

Note that
\[
\frac{1}{n}\sum_{j=1}^{n}W_{\left(  j\right)  }=\frac{1}{n\sqrt{n-1}}%
\sum_{j=1}^{n}\sum_{i\neq j}X_{i}=\frac{n-1}{n\sqrt{n-1}}\sum_{i=1}^{n}%
X_{i}=\frac{\sqrt{n-1}}{\sqrt{n}}W
\]
It therefore follows that
\[
nW-\sqrt{n}\sqrt{n-1}\frac{1}{n}\sum_{j=1}^{n}W_{\left(  j\right)  }%
=nW-\sqrt{n}\sqrt{n-1}\frac{\sqrt{n-1}}{\sqrt{n}}W=W
\]

\subsection{Proof of Lemma \ref{Vjack2}}

We first prove that
\[
\frac{1}{n}\sum_{j=1}^{j}W_{\left(  j\right)  }=\frac{n-2}{n-1}W+\frac
{1}{n\left(  n-1\right)  }\sum_{i=1}^{n}X_{1,i}X_{2,i}%
\]
For this purpose, note that
\begin{align*}
\frac{1}{n-1}\frac{1}{n}\sum_{j=1}^{n}W_{\left(  j\right)  }  &  =\frac{1}%
{n}\sum_{j=1}^{n}\left(  \frac{1}{n-1}\sum_{i\neq j}X_{1,i}\right)  \left(
\frac{1}{n-1}\sum_{i\neq j}X_{2,i}\right) \\
&  =\frac{1}{n}\sum_{j=1}^{n}\left(  \frac{n\overline{X}_{1}-X_{1,j}}%
{n-1}\right)  \left(  \frac{n\overline{X}_{2}-X_{2,j}}{n-1}\right) \\
&  =\frac{1}{n}\frac{n^{3}\overline{X}_{1}\overline{X}_{2}-2n^{2}\overline
{X}_{1}\overline{X}_{2}+\sum_{i=1}^{n}X_{1,i}X_{2,i}}{\left(  n-1\right)
^{2}}\\
&  =\frac{n^{2}-2n}{\left(  n-1\right)  ^{2}}\overline{X}_{1}\overline{X}%
_{2}+\frac{1}{n\left(  n-1\right)  ^{2}}\sum_{i=1}^{n}X_{1,i}X_{2,i}\\
&  =\frac{n-2}{\left(  n-1\right)  ^{2}}W+\frac{1}{n\left(  n-1\right)  ^{2}%
}\sum_{i=1}^{n}X_{1,i}X_{2,i},
\end{align*}
where $\overline{X}_{1}$ denotes the sample average of $X_{1,i}$. Therefore,
we have
\[
\frac{1}{n}\sum_{j=1}^{n}W_{\left(  j\right)  }=\frac{n-2}{n-1}W+\frac
{1}{n\left(  n-1\right)  }\sum_{i=1}^{n}X_{1,i}X_{2,i}%
\]
and
\begin{align*}
nW-\sum_{j=1}^{n}W_{\left(  j\right)  }  &  =\left(  n-\frac{n\left(
n-2\right)  }{n-1}\right)  W-\frac{1}{n-1}\sum_{i=1}^{n}X_{1,i}X_{2,i}\\
&  =\allowbreak\frac{n}{n-1}W-\frac{1}{n-1}\sum_{i=1}^{n}X_{1,i}X_{2,i}\\
&  =\allowbreak\frac{1}{n-1}\left(  \sum_{i=1}^{n}X_{1,i}\right)  \left(
\sum_{i=1}^{n}X_{2,i}\right)  -\frac{1}{n-1}\sum_{i=1}^{n}X_{1,i}X_{2,i}\\
&  =\frac{1}{n-1}\sum_{i\neq j}X_{1,i}X_{2,j}%
\end{align*}

\subsection{Proof of Lemma \ref{Vjack3}}

Observe that
\begin{align*}
&  \frac{1}{n}\sum_{j=1}^{n}\left(  \frac{1}{n-1}\sum_{i\neq j}X_{1,i}\right)
\left(  \frac{1}{n-1}\sum_{i\neq j}X_{2,i}\right)  \left(  \frac{1}{n-1}%
\sum_{i\neq j}X_{3,i}\right) \\
&  =\frac{1}{n}\sum_{j=1}^{n}\left(  \frac{n\overline{X}_{1}-X_{1,j}}%
{n-1}\right)  \left(  \frac{n\overline{X}_{2}-X_{2,j}}{n-1}\right)  \left(
\frac{n\overline{X}_{3}-X_{3,j}}{n-1}\right) \\
&  =\frac{n^{3}-3n^{2}}{\left(  n-1\right)  ^{3}}\overline{X}_{1}\overline
{X}_{2}\overline{X}_{3}\\
&  +\frac{1}{\left(  n-1\right)  ^{3}}\overline{X}_{1}\sum_{i=1}^{n}%
X_{2,i}X_{3,i}+\frac{1}{\left(  n-1\right)  ^{3}}\overline{X}_{2}\sum
_{t=1}^{T}X_{3,i}X_{1,i}\\
&  +\frac{1}{\left(  n-1\right)  ^{3}}\overline{X}_{3}\sum_{t=1}^{T}%
X_{1,i}X_{2,i}-\frac{1}{n\left(  n-1\right)  ^{3}}\sum_{t=1}^{T}X_{1,i}%
X_{2,i}X_{3,i}%
\end{align*}
It therefore follows that
\begin{align*}
\frac{1}{n}\sum_{j=1}^{n}W_{\left(  j\right)  }  &  =\frac{n-3}{\left(
n-1\right)  ^{3/2}}\sqrt{n}\left(  \sqrt{n}\overline{X}_{1}\right)  \left(
\sqrt{n}\overline{X}_{2}\right)  \left(  \sqrt{n}\overline{X}_{3}\right) \\
&  +\frac{\sqrt{n}}{\left(  n-1\right)  ^{3/2}}\left(  \sqrt{n}\overline
{X}_{1}\right)  \left(  \frac{1}{n}\sum_{i=1}^{n}X_{2,i}X_{3,i}\right)
+\frac{\sqrt{n}}{\left(  n-1\right)  ^{3/2}}\left(  \sqrt{n}\overline{X}%
_{2}\right)  \left(  \frac{1}{n}\sum_{i=1}^{n}X_{3,i}X_{1,i}\right) \\
&  +\frac{\sqrt{n}}{\left(  n-1\right)  ^{3/2}}\left(  \sqrt{n}\overline
{X}_{3}\right)  \left(  \frac{1}{n}\sum_{i=1}^{n}X_{1,i}X_{2,i}\right)
-\frac{1}{\left(  n-1\right)  ^{3/2}}\left(  \frac{1}{n}\sum_{t=1}^{T}%
X_{1,i}X_{2,i}X_{3,i}\right)
\end{align*}
and therefore,
\begin{align*}
&  nW-\sqrt{\frac{n}{n-1}}\sum_{j=1}^{n}W_{\left(  j\right)  }\\
&  =\allowbreak\frac{n^{2}+n}{\left(  n-1\right)  ^{2}}W\\
&  -\frac{n^{2}}{\left(  n-1\right)  ^{2}}\left(  \sqrt{n}\overline{X}%
_{1}\right)  \left(  \frac{1}{n}\sum_{i=1}^{n}X_{2,i}X_{3,i}\right)
-\frac{n^{2}}{\left(  n-1\right)  ^{2}}\left(  \sqrt{n}\overline{X}%
_{2}\right)  \left(  \frac{1}{n}\sum_{i=1}^{n}X_{3,i}X_{1,i}\right) \\
&  -\frac{n^{2}}{\left(  n-1\right)  ^{2}}\left(  \sqrt{n}\overline{X}%
_{3}\right)  \left(  \frac{1}{n}\sum_{i=1}^{n}X_{1,i}X_{2,i}\right)
+\frac{n\sqrt{n}}{\left(  n-1\right)  ^{2}}\left(  \frac{1}{n}\sum_{t=1}%
^{T}X_{1,i}X_{2,i}X_{3,i}\right)
\end{align*}

\subsection{Proof of Lemma \ref{Vjack4}}

Observe that
\begin{align*}
&  \frac{1}{n}\sum_{j=1}^{n}\left(  \frac{1}{n-1}\sum_{i\neq j}X_{1,i}\right)
\left(  \frac{1}{n-1}\sum_{i\neq j}X_{2,i}\right)  \left(  \frac{1}{n-1}%
\sum_{i\neq j}X_{3,i}\right)  \left(  \frac{1}{n-1}\sum_{i\neq j}%
X_{4,i}\right) \\
&  =\frac{1}{n}\sum_{j=1}^{n}\left(  \frac{n\overline{X}_{1}-X_{1,j}}%
{n-1}\right)  \left(  \frac{n\overline{X}_{2}-X_{2,j}}{n-1}\right)  \left(
\frac{n\overline{X}_{3}-X_{3,j}}{n-1}\right)  \left(  \frac{n\overline{X}%
_{4}-X_{4,j}}{n-1}\right) \\
&  =\frac{n^{4}-4n^{3}}{\left(  n-1\right)  ^{4}}\overline{X}_{1}\overline
{X}_{2}\overline{X}_{3}\overline{X}_{4}\\
&  +\frac{n}{\left(  n-1\right)  ^{4}}\overline{X}_{1}\overline{X}_{2}%
\sum_{i=1}^{n}X_{3,i}X_{4,i}+\frac{n}{\left(  n-1\right)  ^{4}}\overline
{X}_{1}\overline{X}_{3}\sum_{i=1}^{n}X_{2,i}X_{4,i}+\frac{n}{\left(
n-1\right)  ^{4}}\overline{X}_{1}\overline{X}_{4}\sum_{i=1}^{n}X_{2,i}%
X_{3,i}\\
&  +\frac{n}{\left(  n-1\right)  ^{4}}\overline{X}_{2}\overline{X}_{3}%
\sum_{i=1}^{n}X_{1,i}X_{4,i}+\frac{n}{\left(  n-1\right)  ^{4}}\overline
{X}_{2}\overline{X}_{4}\sum_{i=1}^{n}X_{1,i}X_{3,i}+\frac{n}{\left(
n-1\right)  ^{4}}\overline{X}_{3}\overline{X}_{4}\sum_{i=1}^{n}X_{1,i}%
X_{2,i}\\
&  -\frac{1}{\left(  n-1\right)  ^{4}}\overline{X}_{1}\sum_{i=1}^{n}%
X_{2,i}X_{3,i}X_{4,i}-\frac{1}{\left(  n-1\right)  ^{4}}\overline{X}_{2}%
\sum_{i=1}^{n}X_{1,i}X_{3,i}X_{4,i}\\
&  -\frac{1}{\left(  n-1\right)  ^{4}}\overline{X}_{3}\sum_{i=1}^{n}%
X_{1,i}X_{2,i}X_{4,i}-\frac{1}{\left(  n-1\right)  ^{4}}\overline{X}_{4}%
\sum_{i=1}^{n}X_{1,i}X_{2,i}X_{3,i}\\
&  +\frac{1}{n\left(  n-1\right)  ^{4}}\sum_{i=1}^{n}X_{1,i}X_{2,i}%
X_{3,i}X_{4,i}%
\end{align*}
It therefore follows that
\begin{align*}
\frac{1}{n}\sum_{j=1}^{n}W_{\left(  j\right)  }  &  =\frac{n^{2}-4n}{\left(
n-1\right)  ^{2}}\left(  \sqrt{n}\overline{X}_{1}\right)  \left(  \sqrt
{n}\overline{X}_{2}\right)  \left(  \sqrt{n}\overline{X}_{3}\right)  \left(
\sqrt{n}\overline{X}_{4}\right) \\
&  +\frac{n}{\left(  n-1\right)  ^{2}}\left(  \sqrt{n}\overline{X}_{1}\right)
\left(  \sqrt{n}\overline{X}_{2}\right)  \left(  \frac{1}{n}\sum_{i=1}%
^{n}X_{3,i}X_{4,i}\right) \\
& +\frac{n}{\left(  n-1\right)  ^{2}}\left(  \sqrt
{n}\overline{X}_{1}\right)  \left(  \sqrt{n}\overline{X}_{3}\right)  \left(
\frac{1}{n}\sum_{i=1}^{n}X_{2,i}X_{4,i}\right) \\
&  +\frac{n}{\left(  n-1\right)  ^{2}}\left(  \sqrt{n}\overline{X}_{1}\right)
\left(  \sqrt{n}\overline{X}_{4}\right)  \left(  \frac{1}{n}\sum_{i=1}%
^{n}X_{2,i}X_{3,i}\right) \\
& +\frac{n}{\left(  n-1\right)  ^{2}}\left(  \sqrt
{n}\overline{X}_{2}\right)  \left(  \sqrt{n}\overline{X}_{3}\right)  \left(
\frac{1}{n}\sum_{i=1}^{n}X_{1,i}X_{4,i}\right) \\
&  +\frac{n}{\left(  n-1\right)  ^{2}}\left(  \sqrt{n}\overline{X}_{2}\right)
\left(  \sqrt{n}\overline{X}_{4}\right)  \left(  \frac{1}{n}\sum_{i=1}%
^{n}X_{1,i}X_{3,i}\right) \\
& +\frac{n}{\left(  n-1\right)  ^{2}}\left(  \sqrt
{n}\overline{X}_{3}\right)  \left(  \sqrt{n}\overline{X}_{4}\right)  \left(
\frac{1}{n}\sum_{i=1}^{n}X_{1,i}X_{2,i}\right) \\
&  -\frac{n}{\sqrt{n}\left(  n-1\right)  ^{2}}\left(  \sqrt{n}\overline{X}%
_{1}\right)  \left(  \frac{1}{n}\sum_{i=1}^{n}X_{2,i}X_{3,i}X_{4,i}\right)
\\
&-\frac{n}{\sqrt{n}\left(  n-1\right)  ^{2}}\left(  \sqrt{n}\overline{X}%
_{2}\right)  \left(  \frac{1}{n}\sum_{i=1}^{n}X_{1,i}X_{3,i}X_{4,i}\right) \\
&  -\frac{n}{\sqrt{n}\left(  n-1\right)  ^{2}}\left(  \sqrt{n}\overline{X}%
_{3}\right)  \left(  \frac{1}{n}\sum_{i=1}^{n}X_{1,i}X_{2,i}X_{4,i}\right)
\\
&-\frac{n}{\sqrt{n}\left(  n-1\right)  ^{2}}\left(  \sqrt{n}\overline{X}%
_{4}\right)  \left(  \frac{1}{n}\sum_{i=1}^{n}X_{1,i}X_{2,i}X_{3,i}\right) \\
&  +\frac{1}{\left(  n-1\right)  ^{2}}\left(  \frac{1}{n}\sum_{i=1}^{n}%
X_{1,i}X_{2,i}X_{3,i}X_{4,i}\right)  ,
\end{align*}
from which the conclusion follows.

\subsection{Proof of Lemma \ref{Vjack5}}

Observe that
\begin{align*}
&  \frac{1}{n}\sum_{j=1}^{n}\left(  \frac{1}{n-1}\sum_{i\neq j}X_{1,i}\right)
\left(  \frac{1}{n-1}\sum_{i\neq j}X_{2,i}\right) \times  \\
&\qquad\left(  \frac{1}{n-1}%
\sum_{i\neq j}X_{3,i}\right)  \left(  \frac{1}{n-1}\sum_{i\neq j}%
X_{4,i}\right)  \left(  \frac{1}{n-1}\sum_{i\neq j}X_{5,i}\right) \\
&  =\frac{1}{n}\sum_{j=1}^{n}\left(  \frac{n\overline{X}_{1}-X_{1,j}}%
{n-1}\right)  \left(  \frac{n\overline{X}_{2}-X_{2,j}}{n-1}\right)  \left(
\frac{n\overline{X}_{3}-X_{3,j}}{n-1}\right)  \left(  \frac{n\overline{X}%
_{4}-X_{4,j}}{n-1}\right)  \left(  \frac{n\overline{X}_{5}-X_{5,j}}%
{n-1}\right) \\
&  =\frac{n^{5}-5n^{4}}{\left(  n-1\right)  ^{5}}\overline{X}_{1}\overline
{X}_{2}\overline{X}_{3}\overline{X}_{4}\overline{X}_{5}%
\end{align*}

\begin{align*}
&  +\frac{n^{2}}{\left(  n-1\right)  ^{5}}\overline{X}_{1}\overline{X}%
_{2}\overline{X}_{5}\left(  \frac{1}{n}\sum_{j=1}^{n}X_{3,j}X_{4,j}\right)
+\frac{n^{2}}{\left(  n-1\right)  ^{5}}\overline{X}_{1}\overline{X}%
_{2}\overline{X}_{3}\left(  \frac{1}{n}\sum_{j=1}^{n}X_{4,j}X_{5,j}\right) \\
&  +\frac{n^{2}}{\left(  n-1\right)  ^{5}}\overline{X}_{1}\overline{X}%
_{3}\overline{X}_{5}\left(  \frac{1}{n}\sum_{j=1}^{n}X_{2,j}X_{4,j}\right)
+\frac{n^{2}}{\left(  n-1\right)  ^{5}}\overline{X}_{1}\overline{X}%
_{4}\overline{X}_{5}\left(  \frac{1}{n}\sum_{j=1}^{n}X_{2,j}X_{3,j}\right) \\
&  +\frac{n^{2}}{\left(  n-1\right)  ^{5}}\overline{X}_{1}\overline{X}%
_{3}\overline{X}_{4}\left(  \frac{1}{n}\sum_{j=1}^{n}X_{2,j}X_{5,j}\right)
+\frac{n^{2}}{\left(  n-1\right)  ^{5}}\overline{X}_{3}\overline{X}%
_{4}\overline{X}_{5}\left(  \frac{1}{n}\sum_{j=1}^{n}X_{1,j}X_{2,j}\right) \\
&  +\frac{n^{2}}{\left(  n-1\right)  ^{5}}\overline{X}_{2}\overline{X}%
_{3}\overline{X}_{5}\left(  \frac{1}{n}\sum_{j=1}^{n}X_{1,j}X_{4,j}\right)
+\frac{n^{2}}{\left(  n-1\right)  ^{5}}\overline{X}_{2}\overline{X}%
_{4}\overline{X}_{5}\left(  \frac{1}{n}\sum_{j=1}^{n}X_{1,j}X_{3,j}\right) \\
&  +\frac{n^{2}}{\left(  n-1\right)  ^{5}}\overline{X}_{2}\overline{X}%
_{3}\overline{X}_{4}\left(  \frac{1}{n}\sum_{j=1}^{n}X_{1,j}X_{5,j}\right)
+\frac{n^{2}}{\left(  n-1\right)  ^{5}}\overline{X}_{1}\overline{X}%
_{2}\overline{X}_{4}\left(  \frac{1}{n}\sum_{j=1}^{n}X_{3,j}X_{5,j}\right)
\end{align*}

\begin{align*}
&  -\frac{n}{\left(  n-1\right)  ^{5}}\overline{X}_{3}\overline{X}_{5}\left(
\frac{1}{n}\sum_{j=1}^{n}X_{1,j}X_{2,j}X_{4,j}\right)  -\frac{n}{\left(
n-1\right)  ^{5}}\overline{X}_{1}\overline{X}_{5}\left(  \frac{1}{n}\sum
_{j=1}^{n}X_{2,j}X_{3,j}X_{4,j}\right) \\
&  -\frac{n}{\left(  n-1\right)  ^{5}}\overline{X}_{1}\overline{X}_{3}\left(
\frac{1}{n}\sum_{j=1}^{n}X_{2,j}X_{4,j}X_{5,j}\right)  -\frac{n}{\left(
n-1\right)  ^{5}}\overline{X}_{1}\overline{X}_{4}\left(  \frac{1}{n}\sum
_{j=1}^{n}X_{2,j}X_{3,j}X_{5,j}\right) \\
&  -\frac{n}{\left(  n-1\right)  ^{5}}\overline{X}_{1}\overline{X}_{2}\left(
\frac{1}{n}\sum_{j=1}^{n}X_{3,j}X_{4,j}X_{5,j}\right)  -\frac{n}{\left(
n-1\right)  ^{5}}\overline{X}_{2}\overline{X}_{5}\left(  \frac{1}{n}\sum
_{j=1}^{n}X_{1,j}X_{3,j}X_{4,j}\right) \\
&  -\frac{n}{\left(  n-1\right)  ^{5}}\overline{X}_{2}\overline{X}_{3}\left(
\frac{1}{n}\sum_{j=1}^{n}X_{1,j}X_{4,j}X_{5,j}\right)  -\frac{n}{\left(
n-1\right)  ^{5}}\overline{X}_{2}\overline{X}_{4}\left(  \frac{1}{n}\sum
_{j=1}^{n}X_{1,j}X_{3,j}X_{5,j}\right) \\
&  -\frac{n}{\left(  n-1\right)  ^{5}}\overline{X}_{4}\overline{X}_{5}\left(
\frac{1}{n}\sum_{j=1}^{n}X_{1,j}X_{2,j}X_{3,j}\right)  -\frac{n}{\left(
n-1\right)  ^{5}}\overline{X}_{3}\overline{X}_{4}\left(  \frac{1}{n}\sum
_{j=1}^{n}X_{1,j}X_{2,j}X_{5,j}\right)
\end{align*}

\begin{align*}
&  +\frac{1}{\left(  n-1\right)  ^{5}}\overline{X}_{1}\left(  \frac{1}{n}%
\sum_{j=1}^{n}X_{2,j}X_{3,j}X_{4,j}X_{5,j}\right)  +\frac{1}{\left(
n-1\right)  ^{5}}\overline{X}_{2}\left(  \frac{1}{n}\sum_{j=1}^{n}%
X_{1,j}X_{3,j}X_{4,j}X_{5,j}\right) \\
&  +\frac{1}{\left(  n-1\right)  ^{5}}\overline{X}_{3}\left(  \frac{1}{n}%
\sum_{j=1}^{n}X_{1,j}X_{2,j}X_{4,j}X_{5,j}\right)  +\frac{1}{\left(
n-1\right)  ^{5}}\overline{X}_{4}\left(  \frac{1}{n}\sum_{j=1}^{n}%
X_{1,j}X_{2,j}X_{3,j}X_{5,j}\right) \\
&  +\frac{1}{\left(  n-1\right)  ^{5}}\overline{X}_{5}\left(  \frac{1}{n}%
\sum_{j=1}^{n}X_{1,j}X_{2,j}X_{3,j}X_{4,j}\right) \\
&  -\frac{1}{\left(  n-1\right)  ^{5}}\frac{1}{n}\sum_{j=1}^{n}X_{1,j}%
X_{2,j}X_{3,j}X_{4,j}X_{5,j}%
\end{align*}
Therefore, we have
\begin{align*}
&  \frac{1}{n}\sum_{j=1}^{n}W_{\left(  j\right)  }\\
&  =\frac{n^{3}-5n^{2}}{\left(  n-1\right)  ^{2}\sqrt{n-1}\sqrt{n}}\left(
\sqrt{n}\overline{X}_{1}\right)  \left(  \sqrt{n}\overline{X}_{2}\right)
\left(  \sqrt{n}\overline{X}_{3}\right)  \left(  \sqrt{n}\overline{X}%
_{4}\right)  \left(  \sqrt{n}\overline{X}_{5}\right) \\
&  +\frac{\sqrt{n}}{\left(  n-1\right)  ^{2}\sqrt{n-1}}\left(  \sqrt
{n}\overline{X}_{1}\right)  \left(  \sqrt{n}\overline{X}_{2}\right)  \left(
\sqrt{n}\overline{X}_{5}\right)  \left(  \frac{1}{n}\sum_{j=1}^{n}%
X_{3,j}X_{4,j}\right) \\
&  +\frac{\sqrt{n}}{\left(  n-1\right)  ^{2}\sqrt{n-1}}\left(  \sqrt
{n}\overline{X}_{1}\right)  \left(  \sqrt{n}\overline{X}_{2}\right)  \left(
\sqrt{n}\overline{X}_{3}\right)  \left(  \frac{1}{n}\sum_{j=1}^{n}%
X_{4,j}X_{5,j}\right) \\
&  +\frac{\sqrt{n}}{\left(  n-1\right)  ^{2}\sqrt{n-1}}\left(  \sqrt
{n}\overline{X}_{1}\right)  \left(  \sqrt{n}\overline{X}_{3}\right)  \left(
\sqrt{n}\overline{X}_{5}\right)  \left(  \frac{1}{n}\sum_{j=1}^{n}%
X_{2,j}X_{4,j}\right) \\
&  +\frac{\sqrt{n}}{\left(  n-1\right)  ^{2}\sqrt{n-1}}\left(  \sqrt
{n}\overline{X}_{1}\right)  \left(  \sqrt{n}\overline{X}_{4}\right)  \left(
\sqrt{n}\overline{X}_{5}\right)  \left(  \frac{1}{n}\sum_{j=1}^{n}%
X_{2,j}X_{3,j}\right) \\
&  +\frac{\sqrt{n}}{\left(  n-1\right)  ^{2}\sqrt{n-1}}\left(  \sqrt
{n}\overline{X}_{1}\right)  \left(  \sqrt{n}\overline{X}_{3}\right)  \left(
\sqrt{n}\overline{X}_{4}\right)  \left(  \frac{1}{n}\sum_{j=1}^{n}%
X_{2,j}X_{5,j}\right) \\
&  +\frac{\sqrt{n}}{\left(  n-1\right)  ^{2}\sqrt{n-1}}\left(  \sqrt
{n}\overline{X}_{3}\right)  \left(  \sqrt{n}\overline{X}_{4}\right)  \left(
\sqrt{n}\overline{X}_{5}\right)  \left(  \frac{1}{n}\sum_{j=1}^{n}%
X_{1,j}X_{2,j}\right) \\
&  +\frac{\sqrt{n}}{\left(  n-1\right)  ^{2}\sqrt{n-1}}\left(  \sqrt
{n}\overline{X}_{2}\right)  \left(  \sqrt{n}\overline{X}_{3}\right)  \left(
\sqrt{n}\overline{X}_{5}\right)  \left(  \frac{1}{n}\sum_{j=1}^{n}%
X_{1,j}X_{4,j}\right) \\
&  +\frac{\sqrt{n}}{\left(  n-1\right)  ^{2}\sqrt{n-1}}\left(  \sqrt
{n}\overline{X}_{2}\right)  \left(  \sqrt{n}\overline{X}_{4}\right)  \left(
\sqrt{n}\overline{X}_{5}\right)  \left(  \frac{1}{n}\sum_{j=1}^{n}%
X_{1,j}X_{3,j}\right) \\
&  +\frac{\sqrt{n}}{\left(  n-1\right)  ^{2}\sqrt{n-1}}\left(  \sqrt
{n}\overline{X}_{2}\right)  \left(  \sqrt{n}\overline{X}_{3}\right)  \left(
\sqrt{n}\overline{X}_{4}\right)  \left(  \frac{1}{n}\sum_{j=1}^{n}%
X_{1,j}X_{5,j}\right) \\
&  +\frac{\sqrt{n}}{\left(  n-1\right)  ^{2}\sqrt{n-1}}\left(  \sqrt
{n}\overline{X}_{1}\right)  \left(  \sqrt{n}\overline{X}_{2}\right)  \left(
\sqrt{n}\overline{X}_{4}\right)  \left(  \frac{1}{n}\sum_{j=1}^{n}%
X_{3,j}X_{5,j}\right)
\end{align*}

\begin{align*}
&  -\frac{1}{\left(  n-1\right)  ^{2}\sqrt{n-1}}\left(  \sqrt{n}\overline
{X}_{3}\right)  \left(  \sqrt{n}\overline{X}_{5}\right)  \left(  \frac{1}%
{n}\sum_{j=1}^{n}X_{1,j}X_{2,j}X_{4,j}\right) \\
&  -\frac{1}{\left(  n-1\right)  ^{2}\sqrt{n-1}}\left(  \sqrt{n}\overline
{X}_{1}\right)  \left(  \sqrt{n}\overline{X}_{5}\right)  \left(  \frac{1}%
{n}\sum_{j=1}^{n}X_{2,j}X_{3,j}X_{4,j}\right) \\
&  -\frac{1}{\left(  n-1\right)  ^{2}\sqrt{n-1}}\left(  \sqrt{n}\overline
{X}_{1}\right)  \left(  \sqrt{n}\overline{X}_{3}\right)  \left(  \frac{1}%
{n}\sum_{j=1}^{n}X_{2,j}X_{4,j}X_{5,j}\right) \\
&  -\frac{1}{\left(  n-1\right)  ^{2}\sqrt{n-1}}\left(  \sqrt{n}\overline
{X}_{1}\right)  \left(  \sqrt{n}\overline{X}_{4}\right)  \left(  \frac{1}%
{n}\sum_{j=1}^{n}X_{2,j}X_{3,j}X_{5,j}\right) \\
&  -\frac{1}{\left(  n-1\right)  ^{2}\sqrt{n-1}}\left(  \sqrt{n}\overline
{X}_{1}\right)  \left(  \sqrt{n}\overline{X}_{2}\right)  \left(  \frac{1}%
{n}\sum_{j=1}^{n}X_{3,j}X_{4,j}X_{5,j}\right) \\
&  -\frac{1}{\left(  n-1\right)  ^{2}\sqrt{n-1}}\left(  \sqrt{n}\overline
{X}_{2}\right)  \left(  \sqrt{n}\overline{X}_{5}\right)  \left(  \frac{1}%
{n}\sum_{j=1}^{n}X_{1,j}X_{3,j}X_{4,j}\right) \\
&  -\frac{1}{\left(  n-1\right)  ^{2}\sqrt{n-1}}\left(  \sqrt{n}\overline
{X}_{2}\right)  \left(  \sqrt{n}\overline{X}_{3}\right)  \left(  \frac{1}%
{n}\sum_{j=1}^{n}X_{1,j}X_{4,j}X_{5,j}\right) \\
&  -\frac{1}{\left(  n-1\right)  ^{2}\sqrt{n-1}}\left(  \sqrt{n}\overline
{X}_{2}\right)  \left(  \sqrt{n}\overline{X}_{4}\right)  \left(  \frac{1}%
{n}\sum_{j=1}^{n}X_{1,j}X_{3,j}X_{5,j}\right) \\
&  -\frac{1}{\left(  n-1\right)  ^{2}\sqrt{n-1}}\left(  \sqrt{n}\overline
{X}_{4}\right)  \left(  \sqrt{n}\overline{X}_{5}\right)  \left(  \frac{1}%
{n}\sum_{j=1}^{n}X_{1,j}X_{2,j}X_{3,j}\right) \\
&  -\frac{1}{\left(  n-1\right)  ^{2}\sqrt{n-1}}\left(  \sqrt{n}\overline
{X}_{3}\right)  \left(  \sqrt{n}\overline{X}_{4}\right)  \left(  \frac{1}%
{n}\sum_{j=1}^{n}X_{1,j}X_{2,j}X_{5,j}\right)
\end{align*}

\begin{align*}
&  +\frac{1}{\left(  n-1\right)  ^{2}\sqrt{n-1}\sqrt{n}}\left(  \sqrt
{n}\overline{X}_{1}\right)  \left(  \frac{1}{n}\sum_{j=1}^{n}X_{2,j}%
X_{3,j}X_{4,j}X_{5,j}\right) \\
&  +\frac{1}{\left(  n-1\right)  ^{2}\sqrt{n-1}\sqrt{n}}\left(  \sqrt
{n}\overline{X}_{2}\right)  \left(  \frac{1}{n}\sum_{j=1}^{n}X_{1,j}%
X_{3,j}X_{4,j}X_{5,j}\right) \\
&  +\frac{1}{\left(  n-1\right)  ^{2}\sqrt{n-1}\sqrt{n}}\left(  \sqrt
{n}\overline{X}_{3}\right)  \left(  \frac{1}{n}\sum_{j=1}^{n}X_{1,j}%
X_{2,j}X_{4,j}X_{5,j}\right) \\
&  +\frac{1}{\left(  n-1\right)  ^{2}\sqrt{n-1}\sqrt{n}}\left(  \sqrt
{n}\overline{X}_{4}\right)  \left(  \frac{1}{n}\sum_{j=1}^{n}X_{1,j}%
X_{2,j}X_{3,j}X_{5,j}\right) \\
&  +\frac{1}{\left(  n-1\right)  ^{2}\sqrt{n-1}\sqrt{n}}\left(  \sqrt
{n}\overline{X}_{5}\right)  \left(  \frac{1}{n}\sum_{j=1}^{n}X_{1,j}%
X_{2,j}X_{3,j}X_{4,j}\right) \\
&  -\frac{1}{\left(  n-1\right)  ^{2}\sqrt{n-1}}\frac{1}{n}\sum_{j=1}%
^{n}X_{1,j}X_{2,j}X_{3,j}X_{4,j}X_{5,j}%
\end{align*}
from which the conclusion follows.

\section{Proofs for Example 1\label{sec-4.2-proofs}}

It is straightforward to show that $E\left[  \widehat{\theta}\right]
=\theta+\frac{1}{n}$, so the higher order bias is equal to 1. Consider two
versions of bias correction. The first one is a split-sample jackknife. (For
simplicity, suppose that $n=2m$ for some integer $m$.) It is given by the
formula
\[
\widehat{\theta}_{SS}=2\left(  \frac{1}{n}\sum_{i=1}^{n}Z_{i}\right)
^{2}-\frac{1}{2}\left(  \left(  \frac{1}{m}\sum_{i=1}^{m}Z_{i}\right)
^{2}+\left(  \frac{1}{m}\sum_{i=m+1}^{n}Z_{i}\right)  ^{2}\right)  ,
\]
which can be shown to be equivalent to%
\begin{equation}
\widehat{\theta}_{SS}=\left(  \frac{1}{m}\sum_{i=1}^{m}Z_{i}\right)  \left(
\frac{1}{m}\sum_{i=m+1}^{n}Z_{i}\right)  . \label{SS-equiv}%
\end{equation}
The second one is the jackknife estimator, which can be shown to be equivalent
to%
\begin{equation}
\widehat{\theta}_{J}=\frac{1}{n\left(  n-1\right)  }\sum_{i\neq j}Z_{i}Z_{j}.
\label{jack-equiv}%
\end{equation}

In order to understand this issue, it is useful to define $\varepsilon
_{i}=Z_{i}-\sqrt{\theta}\sim N\left(  0,1\right)  $. It is straightforward to
show that%
\begin{align*}
\widehat{\theta}  &  =\theta+\frac{2\sqrt{\theta}}{n}\sum_{i=1}^{n}%
\varepsilon_{i}+\left(  \frac{1}{n}\sum_{i=1}^{n}\varepsilon_{i}\right)
^{2},\\
\widehat{\theta}_{SS}  &  =\theta+\frac{2\sqrt{\theta}}{n}\sum_{i=1}%
^{n}\varepsilon_{i}+\frac{1}{m}\left(  \frac{1}{\sqrt{m}}\sum_{i=1}%
^{m}\varepsilon_{i}\right)  \left(  \frac{1}{\sqrt{m}}\sum_{i=m+1}%
^{n}\varepsilon_{i}\right)  ,\\
\widehat{\theta}_{J}  &  =\theta+\frac{2\sqrt{\theta}}{n}\sum_{i=1}%
^{n}\varepsilon_{i}+\frac{1}{n\left(  n-1\right)  }\sum_{i\neq j}%
\varepsilon_{i}\varepsilon_{j}.
\end{align*}
Note that the first two terms are identical across the three estimators, i.e.,
the differences are all in the second order. The difference of the second
order terms reflect the difference of the implicit bias estimators.

\subsection{Proof of (\ref{SS-equiv})}%

\begin{align*}
\widehat{\theta}_{SS}  &  =2\left(  \frac{1}{n}\sum_{i=1}^{n}Z_{i}\right)
^{2}-\frac{1}{2}\left(  \left(  \frac{1}{m}\sum_{i=1}^{m}Z_{i}\right)
^{2}+\left(  \frac{1}{m}\sum_{i=m+1}^{n}Z_{i}\right)  ^{2}\right) \\
&  =2\left(  \frac{\frac{1}{m}\sum_{i=1}^{m}Z_{i}+\frac{1}{m}\sum_{i=m+1}%
^{n}Z_{i}}{2}\right)  ^{2}-\frac{1}{2}\left(  \left(  \frac{1}{m}\sum
_{i=1}^{m}Z_{i}\right)  ^{2}+\left(  \frac{1}{m}\sum_{i=m+1}^{n}Z_{i}\right)
^{2}\right) \\
&  =\frac{1}{2}\left(  \frac{1}{m}\sum_{i=1}^{m}Z_{i}\right)  ^{2}+\left(
\frac{1}{m}\sum_{i=1}^{m}Z_{i}\right)  \left(  \frac{1}{m}\sum_{i=m+1}%
^{n}Z_{i}\right)  +\frac{1}{2}\left(  \frac{1}{m}\sum_{i=m+1}^{n}Z_{i}\right)
^{2}\\
&  -\frac{1}{2}\left(  \frac{1}{m}\sum_{i=1}^{m}Z_{i}\right)  ^{2}-\frac{1}%
{2}\left(  \frac{1}{m}\sum_{i=m+1}^{n}Z_{i}\right)  ^{2}\\
&  =\left(  \frac{1}{m}\sum_{i=1}^{m}Z_{i}\right)  \left(  \frac{1}{m}%
\sum_{i=m+1}^{n}Z_{i}\right)  .
\end{align*}

\subsection{Proof of (\ref{jack-equiv})}

We can write%
\begin{align*}
\widehat{\theta}_{J}  &  =n\left(  \overline{Z}\right)  ^{2}-\frac{n-1}{n}%
\sum_{i=1}^{n}\left(  \frac{1}{n-1}\left(  n\overline{Z}-Z_{i}\right)
\right)  ^{2}\\
&  =n\left(  \overline{Z}\right)  ^{2}-\frac{n-1}{n}\frac{1}{\left(
n-1\right)  ^{2}}\sum_{i=1}^{n}\left(  \left(  n\overline{Z}\right)
^{2}-2n\overline{Z}Z_{i}+Z_{i}^{2}\right) \\
&  =n\left(  \overline{Z}\right)  ^{2}-\frac{1}{n\left(  n-1\right)  }\left(
n^{3}\left(  \overline{Z}\right)  ^{2}-2n\overline{Z}\sum_{i=1}^{n}Z_{i}%
+\sum_{i=1}^{n}Z_{i}^{2}\right) \\
&  =n\left(  \overline{Z}\right)  ^{2}-\frac{1}{n\left(  n-1\right)  }\left(
n^{3}\left(  \overline{Z}\right)  ^{2}-2n\overline{Z}\left(  n\overline
{Z}\right)  +\sum_{i=1}^{n}Z_{i}^{2}\right) \\
&  =\left(  n-\frac{n^{3}-2n^{2}}{n\left(  n-1\right)  }\right)  \left(
\overline{Z}\right)  ^{2}-\frac{1}{n\left(  n-1\right)  }\sum_{i=1}^{n}%
Z_{i}^{2}\\
&  =\frac{n}{n-1}\left(  \overline{Z}\right)  ^{2}-\frac{1}{n\left(
n-1\right)  }\sum_{i=1}^{n}Z_{i}^{2}\\
&  =\frac{1}{n\left(  n-1\right)  }\left(  \sum_{i=1}^{n}Z_{i}\right)
^{2}-\frac{1}{n\left(  n-1\right)  }\sum_{i=1}^{n}Z_{i}^{2}\\
&  =\frac{1}{n\left(  n-1\right)  }\sum_{i\neq j}Z_{i}Z_{j}%
\end{align*}

\subsection{Variance of split-sample\label{proof-SS-var}}

Letting
\begin{align*}
S_{1}  &  =\frac{1}{\sqrt{m}}\sum_{i=1}^{m}\varepsilon_{i}\sim N\left(
0,1\right)  ,\\
S_{2}  &  =\frac{1}{\sqrt{m}}\sum_{i=m+1}^{n}\varepsilon_{i}\sim N\left(
0,1\right)  ,
\end{align*}
we can write%
\begin{align*}
\widehat{\theta}_{SS}  &  =\left(  \frac{1}{m}\sum_{i=1}^{m}Z_{i}\right)
\left(  \frac{1}{m}\sum_{i=m+1}^{n}Z_{i}\right) \\
&  =\left(  \sqrt{\theta}+\frac{1}{\sqrt{m}}\left(  \frac{1}{\sqrt{m}}%
\sum_{i=1}^{m}\varepsilon_{i}\right)  \right)  \left(  \sqrt{\theta}+\frac
{1}{\sqrt{m}}\left(  \frac{1}{\sqrt{m}}\sum_{i=m+1}^{n}\varepsilon_{i}\right)
\right) \\
&  =\left(  \sqrt{\theta}+\frac{1}{\sqrt{m}}S_{1}\right)  \left(  \sqrt
{\theta}+\frac{1}{\sqrt{m}}S_{2}\right) \\
&  =\theta+\frac{1}{\sqrt{m}}\left(  S_{1}+S_{2}\right)  \sqrt{\theta}%
+\frac{1}{m}S_{1}S_{2}.
\end{align*}
Therefore,
\[
\left(  \widehat{\theta}_{SS}-\theta\right)  ^{2}=\frac{1}{m}\left(
S_{1}+S_{2}\right)  ^{2}\theta+\frac{1}{m^{2}}S_{1}^{2}S_{2}^{2}+\frac
{2\sqrt{\theta}}{m\sqrt{m}}S_{1}S_{2}\left(  S_{1}+S_{2}\right)
\]
Because $S_{1}$ and $S_{2}$ are independent, we obtain that%
\[
E\left[  \left(  \widehat{\theta}_{SS}-\theta\right)  ^{2}\right]  =\frac
{2}{m}\theta+\frac{1}{m^{2}}=\frac{4}{n}\theta+\frac{4}{n^{2}}%
\]

\subsection{Variance of jackknife\label{proof-jack-var}}

We have%
\begin{align*}
\widehat{\theta}_{J}  &  =\frac{1}{n\left(  n-1\right)  }\sum_{i\neq j}%
Z_{i}Z_{j}=\frac{1}{n\left(  n-1\right)  }\sum_{i\neq j}\left(  \sqrt{\theta
}+\varepsilon_{i}\right)  \left(  \sqrt{\theta}+\varepsilon_{j}\right) \\
&  =\theta+2\sqrt{\theta}\frac{1}{n\left(  n-1\right)  }\sum_{i\neq j}\left(
\varepsilon_{i}+\varepsilon_{j}\right)  +\frac{1}{n\left(  n-1\right)  }%
\sum_{i\neq j}\varepsilon_{i}\varepsilon_{j}\\
&  =\theta+2\sqrt{\theta}\frac{\sum_{i}\varepsilon_{i}}{n}+\frac{1}{n\left(
n-1\right)  }\sum_{i\neq j}\varepsilon_{i}\varepsilon_{j}.
\end{align*}
Regarding the third equality, I used%
\begin{align*}
\sum_{i\neq j}\left(  \varepsilon_{i}+\varepsilon_{j}\right)   &  =\sum
_{i,j}\left(  \varepsilon_{i}+\varepsilon_{j}\right)  -\sum_{i}\left(
\varepsilon_{i}+\varepsilon_{i}\right)  =\sum_{i}\left(  \varepsilon
_{i}+n\overline{\varepsilon}\right)  -2n\overline{\varepsilon}\\
&  =n\overline{\varepsilon}+n^{2}\overline{\varepsilon}-2n\overline
{\varepsilon}\\
&  =\left(  n^{2}-n\right)  \overline{\varepsilon}.
\end{align*}
It follows that
\[
\left(  \widehat{\theta}_{J}-\theta\right)  ^{2}=4\theta\left(  \frac{\sum
_{i}\varepsilon_{i}}{n}\right)  ^{2}+\left(  \frac{1}{n\left(  n-1\right)
}\sum_{i\neq j}\varepsilon_{i}\varepsilon_{j}\right)  ^{2}+4\sqrt{\theta
}\left(  \frac{\sum_{i}\varepsilon_{i}}{n}\right)  \left(  \frac{1}{n\left(
n-1\right)  }\sum_{i\neq j}\varepsilon_{i}\varepsilon_{j}\right)
\]
Because the third term consists of third moments, we can see that it has a
zero expectation. As for the second term, we see that%
\begin{align*}
E\left[  \left(  \sum_{i\neq j}\varepsilon_{i}\varepsilon_{j}\right)
^{2}\right]   &  =E\left[  \sum_{i_{1}\neq j_{1}}\sum_{i_{2}\neq j_{2}%
}\varepsilon_{i_{1}}\varepsilon_{j_{1}}\varepsilon_{i_{2}}\varepsilon_{j_{2}%
}\right] \\
&  =E\left[  \sum_{i_{1}=i_{2}\neq j_{1}=j_{2}}\varepsilon_{i_{1}}%
\varepsilon_{j_{1}}\varepsilon_{i_{2}}\varepsilon_{j_{2}}\right]  +E\left[
\sum_{i_{1}=j_{2}\neq j_{1}=i_{2}}\varepsilon_{i_{1}}\varepsilon_{j_{1}%
}\varepsilon_{i_{2}}\varepsilon_{j_{2}}\right] \\
&  =E\left[  \sum_{i_{1}\neq j_{1}}\left(  \varepsilon_{i_{1}}\varepsilon
_{j_{1}}\right)  ^{2}\right]  +E\left[  \sum_{i_{1}\neq j_{1}}\left(
\varepsilon_{i_{1}}\varepsilon_{j_{1}}\right)  ^{2}\right] \\
&  =2n\left(  n-1\right)  E\left[  \left(  \varepsilon_{i_{1}}\varepsilon
_{j_{1}}\right)  ^{2}\right]  =2n\left(  n-1\right)  ,
\end{align*}
so%
\[
E\left[  \left(  \frac{1}{n\left(  n-1\right)  }\sum_{i\neq j}\varepsilon
_{i}\varepsilon_{j}\right)  ^{2}\right]  =\frac{2}{n\left(  n-1\right)  }.
\]
Finally, the first term is such that%
\[
E\left[  4\theta\left(  \frac{\sum_{i}\varepsilon_{i}}{n}\right)  ^{2}\right]
=\frac{4\theta}{n}.
\]

\subsection{Accuracy of jackknife bias estimate)\label{proof-jack-bias}}

We have%
\begin{align*}
\widehat{\theta}-\widehat{\theta}_{J}  &  =\left(  \frac{1}{n}\sum_{i=1}%
^{n}\varepsilon_{i}\right)  ^{2}-\frac{1}{n\left(  n-1\right)  }\sum_{i\neq
j}\varepsilon_{i}\varepsilon_{j}\\
&  =\frac{1}{n^{2}}\left(  \sum_{i=1}^{n}\varepsilon_{i}^{2}+\sum_{i\neq
j}\varepsilon_{i}\varepsilon_{j}\right)  -\frac{1}{n\left(  n-1\right)  }%
\sum_{i\neq j}\varepsilon_{i}\varepsilon_{j}\\
&  =\frac{1}{n^{2}}\sum_{i=1}^{n}\varepsilon_{i}^{2}+\left(  \frac{1}{n^{2}%
}-\frac{1}{n\left(  n-1\right)  }\right)  \sum_{i\neq j}\varepsilon
_{i}\varepsilon_{j}\\
&  =\frac{1}{n}\left(  \frac{1}{n}\sum_{i=1}^{n}\varepsilon_{i}^{2}\right)
-\frac{1}{n^{2}\left(  n-1\right)  }\sum_{i\neq j}\varepsilon_{i}%
\varepsilon_{j}.
\end{align*}
Because $E\left[  \left(  \sum_{i\neq j}\varepsilon_{i}\varepsilon_{j}\right)
^{2}\right]  =2n\left(  n-1\right)  $,\footnote{See Section
\ref{proof-jack-var}.} we have%
\[
\frac{1}{n^{2}\left(  n-1\right)  }\sum_{i\neq j}\varepsilon_{i}%
\varepsilon_{j}=\frac{O_{p}\left(  n\right)  }{n^{2}\left(  n-1\right)
}=O_{p}\left(  \frac{1}{n^{2}}\right)  .
\]
We also have%
\[
\frac{1}{n}\sum_{i=1}^{n}\varepsilon_{i}^{2}=E\left[  \varepsilon_{i}%
^{2}\right]  +O_{p}\left(  \frac{1}{\sqrt{n}}\right)  ,
\]
so%
\begin{align*}
\left(  \frac{1}{n}\sum_{i=1}^{n}\varepsilon_{i}\right)  ^{2}-\frac
{1}{n\left(  n-1\right)  }\sum_{i\neq j}\varepsilon_{i}\varepsilon_{j}%
&=\frac{1}{n}\left(  E\left[  \varepsilon_{i}^{2}\right]  +O_{p}\left(
\frac{1}{\sqrt{n}}\right)  \right)  +O_{p}\left(  \frac{1}{n^{2}}\right) \\
&=\frac{1}{n}\left(  E\left[  \varepsilon_{i}^{2}\right]  +O_{p}\left(
\frac{1}{\sqrt{n}}\right)  \right)  .
\end{align*}

\subsection{Accuracy of split-sample bias estimate)\label{proof-SS-bias}}

We have%
\begin{align*}
\widehat{\theta}-\widehat{\theta}_{SS}  &  =\left(  \frac{1}{n}\sum_{i=1}%
^{n}\varepsilon_{i}\right)  ^{2}-\frac{1}{m}\left(  \frac{1}{\sqrt{m}}%
\sum_{i=1}^{m}\varepsilon_{i}\right)  \left(  \frac{1}{\sqrt{m}}\sum
_{i=m+1}^{n}\varepsilon_{i}\right) \\
&  =\frac{1}{n^{2}}\left(  \sum_{i=1}^{n}\varepsilon_{i}\right)  ^{2}-\frac
{1}{m^{2}}\left(  \sum_{i=1}^{m}\varepsilon_{i}\right)  \left(  \sum
_{i=m+1}^{n}\varepsilon_{i}\right) \\
&  =\frac{1}{n^{2}}\left(  \left(  \sum_{i=1}^{m}\varepsilon_{i}\right)
+\left(  \sum_{i=m+1}^{n}\varepsilon_{i}\right)  \right)  ^{2}-\frac{1}{m^{2}%
}\left(  \sum_{i=1}^{m}\varepsilon_{i}\right)  \left(  \sum_{i=m+1}%
^{n}\varepsilon_{i}\right) \\
&  =\frac{1}{n^{2}}\left(  \sum_{i=1}^{m}\varepsilon_{i}\right)  ^{2}+\frac
{1}{n^{2}}\left(  \sum_{i=m+1}^{n}\varepsilon_{i}\right)  ^{2}+\left(
\frac{2}{n^{2}}-\frac{1}{m^{2}}\right)  \left(  \sum_{i=1}^{m}\varepsilon
_{i}\right)  \left(  \sum_{i=m+1}^{n}\varepsilon_{i}\right) \\
&  =\frac{m}{n^{2}}\left(  \frac{1}{\sqrt{m}}\sum_{i=1}^{m}\varepsilon
_{i}\right)  ^{2}+\frac{m}{n^{2}}\left(  \frac{1}{\sqrt{m}}\sum_{i=m+1}%
^{n}\varepsilon_{i}\right)  ^{2} \\
&+m\left(  \frac{2}{n^{2}}-\frac{1}{m^{2}%
}\right)  \left(  \frac{1}{\sqrt{m}}\sum_{i=1}^{m}\varepsilon_{i}\right)
\left(  \frac{1}{\sqrt{m}}\sum_{i=m+1}^{n}\varepsilon_{i}\right) \\
&  =\frac{1}{2n}\left(  \frac{1}{\sqrt{m}}\sum_{i=1}^{m}\varepsilon
_{i}\right)  ^{2}+\frac{1}{2n}\left(  \frac{1}{\sqrt{m}}\sum_{i=m+1}%
^{n}\varepsilon_{i}\right)  ^{2}-\frac{1}{n}\left(  \frac{1}{\sqrt{m}}%
\sum_{i=1}^{m}\varepsilon_{i}\right)  \left(  \frac{1}{\sqrt{m}}\sum
_{i=m+1}^{n}\varepsilon_{i}\right)
\end{align*}
Because $\frac{1}{\sqrt{m}}\sum_{i=1}^{m}\varepsilon_{i}$ and $\frac{1}%
{\sqrt{m}}\sum_{i=m+1}^{n}\varepsilon_{i}$ are independent $N\left(
0,1\right)  $, we conclude that%
\[
\widehat{\theta}-\widehat{\theta}_{SS}=\frac{1}{n}\left(  E\left[
\varepsilon_{i}^{2}\right]  +O_{p}\left(  1\right)  \right)  .
\]

\section{Higher Order Variance of Split Sample Bias Corrected
Estimator - Generalization \label{app-HOV-SS-gen}}

Suppose that
\[
n=km
\]
for some integer $k$ and $m$. From%
\[
E\left[  \widehat{\theta}\right]  =\theta+\frac{1}{n}b
\]
we see that%
\[
E\left[  \widehat{\theta}^{\left(  s\right)  }\right]  =\theta+\frac{b}%
{m}=\theta+\frac{1}{n/k}b=\theta+\frac{k}{n}b
\]
for $s=1,\cdots,k$. Here, $\widehat{\theta}^{\left(  s\right)  }$ denotes the
estimator based on the subsample $\left\{  Z_{\left(  m-1\right)  s+1}%
,\ldots,Z_{m}\right\}  $ of size $m$. For simplicity, let $\sum_{\left(
s\right)  }$ denote the sum over the indices $\left(  m-1\right)
s+1,\ldots,,m$. We therefore have%
\begin{align*}
E\left[  k\widehat{\theta}\right]   &  =k\theta+\frac{k}{n}b\\
E\left[  \frac{1}{k}\sum_{s=1}^{k}\widehat{\theta}^{\left(  s\right)
}\right]   &  =\theta+\frac{k}{n}b
\end{align*}
so%
\[
E\left[  \frac{1}{k-1}\left(  k\widehat{\theta}-\left(  \frac{1}{k}\sum
_{s=1}^{k}\widehat{\theta}^{\left(  s\right)  }\right)  \right)  \right]
=\theta
\]
which can be used as a basis of bias correction. We let
\[
\widehat{\theta}_{SS}=\frac{k}{k-1}\widehat{\theta}-\frac{1}{k\left(
k-1\right)  }\sum_{s=1}^{k}\widehat{\theta}^{\left(  s\right)  }%
\]

We can write%
\begin{equation}
\sqrt{n}\left(  \widehat{\theta}-\theta_{0}\right)  =T_{1}+\frac{1}{\sqrt{n}%
}T_{2}+\frac{1}{n}T_{3}+o_{p}\left(  n^{-1}\right)
\label{eq:generic-expansion}%
\end{equation}
where%
\begin{align}
T_{1}=  &  \frac{1}{\sqrt{n}}\sum_{i=1}^{n}Y_{1,i},\nonumber\\
T_{2}=  &  \left(  \frac{1}{\sqrt{n}}\sum_{i=1}^{n}Y_{1,i}\right)  \left(
\frac{1}{\sqrt{n}}\sum_{i=1}^{n}Y_{2,i}\right)  ,\label{T2-X1X2g}\\
T_{3}=  &  \left(  \frac{1}{\sqrt{n}}\sum_{i=1}^{n}Y_{1,i}\right)  \left(
\frac{1}{\sqrt{n}}\sum_{i=1}^{n}Y_{3,i}\right)  \left(  \frac{1}{\sqrt{n}}%
\sum_{i=1}^{n}Y_{4,i}\right) \label{T3-X1X3X4g}\\
&  +\left(  \frac{1}{\sqrt{n}}\sum_{i=1}^{n}Y_{1,i}\right)  \left(  \frac
{1}{\sqrt{n}}\sum_{i=1}^{n}Y_{5,i}\right)  \left(  \frac{1}{\sqrt{n}}%
\sum_{i=1}^{n}Y_{6,i}\right)  . \label{T3-X1X5X6g}%
\end{align}
Here, we defined%
\begin{align*}
Y_{1,i}  &  \equiv\mathcal{I}^{-1}U_{i}\left(  \theta_{0}\right)  ,\\
Y_{2,i}  &  \equiv\frac{1}{2}\mathcal{I}^{-2}\mathcal{Q}_{1}\left(  \theta
_{0}\right)  U_{i}\left(  \theta_{0}\right)  +\mathcal{I}^{-2}V_{i}\left(
\theta_{0}\right)  ,\\
Y_{3,i}  &  \equiv Y_{1,i},\\
Y_{4,i}  &  \equiv\frac{1}{6}\mathcal{I}^{-2}\mathcal{Q}_{2}\left(  \theta
_{0}\right)  U_{i}\left(  \theta_{0}\right)  +\frac{1}{2}\mathcal{I}%
^{-3}\mathcal{Q}_{1}\left(  \theta_{0}\right)  ^{2}V_{i}\left(  \theta
_{0}\right)  +\frac{1}{2}\mathcal{I}^{-1}W_{i}\left(  \theta_{0}\right)  ,\\
Y_{5,i}  &  \equiv Y_{6,i}=\mathcal{I}^{-1}V_{i}\left(  \theta_{0}\right)  .
\end{align*}
Note that they all have zero means.

Let
\begin{align*}
\mathcal{Y}_{1}  &  =\frac{1}{\sqrt{n}}\sum_{i=1}^{n}Y_{1,i},\quad
\mathcal{Y}_{1}^{\left(  s\right)  }=\frac{1}{\sqrt{m}}\sum_{\left(  s\right)
}Y_{1,i},\\
\mathcal{Y}_{2}  &  =\frac{1}{\sqrt{n}}\sum_{i=1}^{n}Y_{2,i},\quad
\mathcal{Y}_{2}^{\left(  s\right)  }=\frac{1}{\sqrt{m}}\sum_{\left(  s\right)
}Y_{2,i},\\
\mathcal{Y}_{3}  &  =\frac{1}{\sqrt{n}}\sum_{i=1}^{n}Y_{3,i},\quad
\mathcal{Y}_{3}^{\left(  s\right)  }=\frac{1}{\sqrt{m}}\sum_{\left(  s\right)
}Y_{3,i},\\
\mathcal{Y}_{4}  &  =\frac{1}{\sqrt{n}}\sum_{i=1}^{n}Y_{4,i},\quad
\mathcal{Y}_{4}^{\left(  s\right)  }=\frac{1}{\sqrt{m}}\sum_{\left(  s\right)
}Y_{4,i},\\
\mathcal{Y}_{5}  &  =\frac{1}{\sqrt{n}}\sum_{i=1}^{n}Y_{5,i},\quad
\mathcal{Y}_{5}^{\left(  s\right)  }=\frac{1}{\sqrt{m}}\sum_{\left(  s\right)
}Y_{5,i},\\
\mathcal{Y}_{6}  &  =\frac{1}{\sqrt{n}}\sum_{i=1}^{n},\quad\mathcal{Y}%
_{6}^{\left(  s\right)  }=\frac{1}{\sqrt{m}}\sum_{\left(  s\right)  }Y_{6,i},
\end{align*}
Note that we can write%
\begin{equation}
\sqrt{n}\left(  \widehat{\theta}-\theta_{0}\right)  =\mathcal{Y}_{1}+\frac
{1}{\sqrt{n}}\mathcal{Y}_{1}\mathcal{Y}_{2}+\frac{1}{n}\left(  \mathcal{Y}%
_{1}\mathcal{Y}_{3}\mathcal{Y}_{4}+\mathcal{Y}_{1}\mathcal{Y}_{5}%
\mathcal{Y}_{6}\right)  +o_{p}\left(  n^{-1}\right)  \label{theta-abcdfgg}%
\end{equation}

We prove

\begin{lemma}
\label{lem-ss-expansiong}%
\[
\sqrt{n}\left(  \widehat{\theta}_{SS}-\theta\right)  =T_{1,SS}+\frac{1}%
{\sqrt{n}}T_{2,SS}+\frac{1}{n}T_{3,SS}+o_{p}\left(  n^{-1}\right)  ,
\]
where%
\begin{align*}
T_{1,SS}  &  =\mathcal{Y}_{1}=T_{1},\\
T_{2,SS}  &  =\frac{1}{k-1}\left(  \sum_{s=1}^{k}\mathcal{Y}_{1}^{\left(
s\right)  }\right)  \left(  \sum_{s=1}^{k}\mathcal{Y}_{2}^{\left(  s\right)
}\right)  -\frac{1}{k-1}\sum_{s=1}^{k}\mathcal{Y}_{1}^{\left(  s\right)
}\mathcal{Y}_{2}^{\left(  s\right)  }\\
&  =\frac{1}{k-1}\sum_{s\neq s^{\prime}}\mathcal{Y}_{1}^{\left(  s\right)
}\mathcal{Y}_{2}^{\left(  s^{\prime}\right)  }\\
T_{3,SS}  &  =T_{3,1,SS}+T_{3,2,SS}%
\end{align*}
and%
\begin{align}
T_{3,1,SS}  &  =\frac{1}{\left(  k-1\right)  \sqrt{k}}\left(  \sum_{s=1}%
^{k}\mathcal{Y}_{1}^{\left(  s\right)  }\right)  \left(  \sum_{s=1}%
^{k}\mathcal{Y}_{3}^{\left(  s\right)  }\right)  \left(  \sum_{s=1}%
^{k}\mathcal{Y}_{4}^{\left(  s\right)  }\right)  -\frac{\sqrt{k}}{k-1}\left(
\sum_{s=1}^{k}\mathcal{Y}_{1}^{\left(  s\right)  }\mathcal{Y}_{3}^{\left(
s\right)  }\mathcal{Y}_{4}^{\left(  s\right)  }\right) \nonumber\\
T_{3,2,SS}  &  =\frac{1}{\left(  k-1\right)  \sqrt{k}}\left(  \sum_{s=1}%
^{k}\mathcal{Y}_{1}^{\left(  s\right)  }\right)  \left(  \sum_{s=1}%
^{k}\mathcal{Y}_{5}^{\left(  s\right)  }\right)  \left(  \sum_{s=1}%
^{k}\mathcal{Y}_{6}^{\left(  s\right)  }\right)  -\frac{\sqrt{k}}{k-1}\left(
\sum_{s=1}^{k}\mathcal{Y}_{1}^{\left(  s\right)  }\mathcal{Y}_{5}^{\left(
s\right)  }\mathcal{Y}_{6}^{\left(  s\right)  }\right)  \label{T3SSg}%
\end{align}

\end{lemma}

\begin{remark}
We may alternatively write%
\begin{align*}
T_{3,1,SS}  &  =\frac{k}{\left(  k-1\right)  }\mathcal{Y}_{1}\mathcal{Y}%
_{3}\mathcal{Y}_{4}-\frac{\sqrt{k}}{k-1}\left(  \sum_{s=1}^{k}\mathcal{Y}%
_{1}^{\left(  s\right)  }\mathcal{Y}_{3}^{\left(  s\right)  }\mathcal{Y}%
_{4}^{\left(  s\right)  }\right) \\
T_{3,2,SS}  &  =\frac{k}{\left(  k-1\right)  }\mathcal{Y}_{1}\mathcal{Y}%
_{5}\mathcal{Y}_{6}-\frac{\sqrt{k}}{k-1}\left(  \sum_{s=1}^{k}\mathcal{Y}%
_{1}^{\left(  s\right)  }\mathcal{Y}_{5}^{\left(  s\right)  }\mathcal{Y}%
_{6}^{\left(  s\right)  }\right)
\end{align*}

\end{remark}

\begin{proof}
We write (\ref{theta-abcdfgg}) as%
\begin{align*}
\widehat{\theta}-\theta_{0}  &  =\frac{1}{\sqrt{n}}\mathcal{Y}_{1}+\frac{1}%
{n}\mathcal{Y}_{1}\mathcal{Y}_{2}+\frac{1}{n\sqrt{n}}\left(  \mathcal{Y}%
_{1}\mathcal{Y}_{3}\mathcal{Y}_{4}+\mathcal{Y}_{1}\mathcal{Y}_{5}%
\mathcal{Y}_{6}\right)  +o_{p}\left(  n^{-3/2}\right) \\
&  =\frac{1}{\sqrt{km}}\mathcal{Y}_{1}+\frac{1}{km}\mathcal{Y}_{1}%
\mathcal{Y}_{2}+\left(  \frac{1}{\sqrt{km}}\right)  ^{3}\left(  \mathcal{Y}%
_{1}\mathcal{Y}_{3}\mathcal{Y}_{4}+\mathcal{Y}_{1}\mathcal{Y}_{5}%
\mathcal{Y}_{6}\right)  +o_{p}\left(  n^{-3/2}\right)
\end{align*}
and use
\[
\mathcal{Y}_{1}=\frac{1}{\sqrt{k}}\sum_{s=1}^{k}\mathcal{Y}_{1}^{\left(
s\right)  }%
\]
etc., to write
\begin{align*}
\widehat{\theta}-\theta_{0}=  &  \frac{1}{k\sqrt{m}}\left(  \sum_{s=1}%
^{k}\mathcal{Y}_{1}^{\left(  s\right)  }\right)  +\left(  \frac{1}{k\sqrt{m}%
}\right)  ^{2}\left(  \sum_{s=1}^{k}\mathcal{Y}_{1}^{\left(  s\right)
}\right)  \left(  \sum_{s=1}^{k}\mathcal{Y}_{2}^{\left(  s\right)  }\right) \\
&  +\left(  \frac{1}{k\sqrt{m}}\right)  ^{3}\left(  \sum_{s=1}^{k}%
\mathcal{Y}_{1}^{\left(  s\right)  }\right)  \left(  \sum_{s=1}^{k}%
\mathcal{Y}_{3}^{\left(  s\right)  }\right)  \left(  \sum_{s=1}^{k}%
\mathcal{Y}_{4}^{\left(  s\right)  }\right) \\
&  +\left(  \frac{1}{k\sqrt{m}}\right)  ^{3}\left(  \sum_{s=1}^{k}%
\mathcal{Y}_{1}^{\left(  s\right)  }\right)  \left(  \sum_{s=1}^{k}%
\mathcal{Y}_{5}^{\left(  s\right)  }\right)  \left(  \sum_{s=1}^{k}%
\mathcal{Y}_{6}^{\left(  s\right)  }\right) \\
&  +o_{p}\left(  n^{-3/2}\right)  .
\end{align*}
Using (\ref{theta-abcdfgg}), we also obtain%
\begin{align*}
\widehat{\theta}^{\left(  s\right)  }-\theta_{0}=  &  \frac{1}{\sqrt{m}%
}\mathcal{Y}_{1}^{\left(  s\right)  }+\left(  \frac{1}{\sqrt{m}}\right)
^{2}\mathcal{Y}_{1}^{\left(  s\right)  }\mathcal{Y}_{2}^{\left(  s\right)
}+\left(  \frac{1}{\sqrt{m}}\right)  ^{3}\mathcal{Y}_{1}^{\left(  s\right)
}\mathcal{Y}_{3}^{\left(  s\right)  }\mathcal{Y}_{4}^{\left(  s\right)  }\\
&  +\left(  \frac{1}{\sqrt{m}}\right)  ^{3}\mathcal{Y}_{1}^{\left(  s\right)
}\mathcal{Y}_{5}^{\left(  s\right)  }\mathcal{Y}_{6}^{\left(  s\right)
}+o_{p}\left(  n^{-3/2}\right)  .
\end{align*}
Combining them, we obtain
\begin{align*}
\widehat{\theta}_{SS}-\theta &  =\frac{k}{k-1}\left(  \widehat{\theta}%
-\theta\right)  -\frac{1}{k\left(  k-1\right)  }\sum_{s=1}^{k}\left(
\widehat{\theta}^{\left(  s\right)  }-\theta\right) \\
&  =\frac{1}{\left(  k-1\right)  }\frac{1}{\sqrt{m}}\left(  \sum_{s=1}%
^{k}\mathcal{Y}_{1}^{\left(  s\right)  }\right)  -\frac{1}{k\left(
k-1\right)  }\frac{1}{\sqrt{m}}\left(  \sum_{s=1}^{k}\mathcal{Y}_{1}^{\left(
s\right)  }\right) \\
&  +\frac{k}{k-1}\left(  \frac{1}{k\sqrt{m}}\right)  ^{2}\left(  \sum
_{s=1}^{k}\mathcal{Y}_{1}^{\left(  s\right)  }\right)  \left(  \sum_{s=1}%
^{k}\mathcal{Y}_{2}^{\left(  s\right)  }\right)  -\frac{1}{k\left(
k-1\right)  }\left(  \frac{1}{\sqrt{m}}\right)  ^{2}\left(  \sum_{s=1}%
^{k}\mathcal{Y}_{1}^{\left(  s\right)  }\mathcal{Y}_{2}^{\left(  s\right)
}\right) \\
&  +\frac{k}{k-1}\left(  \frac{1}{k\sqrt{m}}\right)  ^{3}\left(  \sum
_{s=1}^{k}\mathcal{Y}_{1}^{\left(  s\right)  }\right)  \left(  \sum_{s=1}%
^{k}\mathcal{Y}_{3}^{\left(  s\right)  }\right)  \left(  \sum_{s=1}%
^{k}\mathcal{Y}_{4}^{\left(  s\right)  }\right)  \\
&-\frac{1}{k\left(
k-1\right)  }\left(  \frac{1}{\sqrt{m}}\right)  ^{3}\left(  \sum_{s=1}%
^{k}\mathcal{Y}_{1}^{\left(  s\right)  }\mathcal{Y}_{3}^{\left(  s\right)
}\mathcal{Y}_{4}^{\left(  s\right)  }\right) \\
&  +\frac{k}{k-1}\left(  \frac{1}{k\sqrt{m}}\right)  ^{3}\left(  \sum
_{s=1}^{k}\mathcal{Y}_{1}^{\left(  s\right)  }\right)  \left(  \sum_{s=1}%
^{k}\mathcal{Y}_{5}^{\left(  s\right)  }\right)  \left(  \sum_{s=1}%
^{k}\mathcal{Y}_{6}^{\left(  s\right)  }\right)  \\
&-\frac{1}{k\left(
k-1\right)  }\left(  \frac{1}{\sqrt{m}}\right)  ^{3}\left(  \sum_{s=1}%
^{k}\mathcal{Y}_{1}^{\left(  s\right)  }\mathcal{Y}_{5}^{\left(  s\right)
}\mathcal{Y}_{6}^{\left(  s\right)  }\right) \\
&  +o_{p}\left(  n^{-3/2}\right)
\end{align*}
So, we obtain%
\begin{align*}
T_{1,SS}  &  =\sqrt{n}\left(  \frac{1}{\left(  k-1\right)  }\frac{1}{\sqrt{m}%
}\left(  \sum_{s=1}^{k}\mathcal{Y}_{1}^{\left(  s\right)  }\right)  -\frac
{1}{k\left(  k-1\right)  }\frac{1}{\sqrt{m}}\left(  \sum_{s=1}^{k}%
\mathcal{Y}_{1}^{\left(  s\right)  }\right)  \right) \\
&  =\sqrt{n}\frac{1}{k}\frac{1}{\sqrt{m}}\left(  \sum_{s=1}^{k}\mathcal{Y}%
_{1}^{\left(  s\right)  }\right)  =\frac{1}{\sqrt{k}}\left(  \sum_{s=1}%
^{k}\mathcal{Y}_{1}^{\left(  s\right)  }\right)  =\mathcal{Y}_{1},
\end{align*}%
\begin{align*}
T_{2,SS}  &  =n\frac{k}{k-1}\left(  \frac{1}{k\sqrt{m}}\right)  ^{2}\left(
\sum_{s=1}^{k}\mathcal{Y}_{1}^{\left(  s\right)  }\right)  \left(  \sum
_{s=1}^{k}\mathcal{Y}_{2}^{\left(  s\right)  }\right)  -\frac{1}{k\left(
k-1\right)  }n\left(  \frac{1}{\sqrt{m}}\right)  ^{2}\left(  \sum_{s=1}%
^{k}\mathcal{Y}_{1}^{\left(  s\right)  }\mathcal{Y}_{2}^{\left(  s\right)
}\right) \\
&  =\frac{1}{\left(  k-1\right)  }\left(  \left(  \sum_{s=1}^{k}%
\mathcal{Y}_{1}^{\left(  s\right)  }\right)  \left(  \sum_{s=1}^{k}%
\mathcal{Y}_{2}^{\left(  s\right)  }\right)  -\sum_{s=1}^{k}\mathcal{Y}%
_{1}^{\left(  s\right)  }\mathcal{Y}_{2}^{\left(  s\right)  }\right)
\end{align*}
and%
\begin{align*}
T_{3,1,SS}  &  =n\sqrt{n}\frac{k}{k-1}\left(  \frac{1}{k\sqrt{m}}\right)
^{3}\left(  \sum_{s=1}^{k}\mathcal{Y}_{1}^{\left(  s\right)  }\right)  \left(
\sum_{s=1}^{k}\mathcal{Y}_{3}^{\left(  s\right)  }\right)  \left(  \sum
_{s=1}^{k}\mathcal{Y}_{4}^{\left(  s\right)  }\right) \\
&  -n\sqrt{n}\frac{1}{k\left(  k-1\right)  }\left(  \frac{1}{\sqrt{m}}\right)
^{3}\left(  \sum_{s=1}^{k}\mathcal{Y}_{1}^{\left(  s\right)  }\mathcal{Y}%
_{3}^{\left(  s\right)  }\mathcal{Y}_{4}^{\left(  s\right)  }\right) \\
&  =\frac{1}{\left(  k-1\right)  \sqrt{k}}\left(  \sum_{s=1}^{k}%
\mathcal{Y}_{1}^{\left(  s\right)  }\right)  \left(  \sum_{s=1}^{k}%
\mathcal{Y}_{3}^{\left(  s\right)  }\right)  \left(  \sum_{s=1}^{k}%
\mathcal{Y}_{4}^{\left(  s\right)  }\right) \\
&  -\frac{\sqrt{k}}{\left(  k-1\right)  }\left(  \sum_{s=1}^{k}\mathcal{Y}%
_{1}^{\left(  s\right)  }\mathcal{Y}_{3}^{\left(  s\right)  }\mathcal{Y}%
_{4}^{\left(  s\right)  }\right)
\end{align*}
and likewise for $T_{3,1,SS}$.
\end{proof}

Note that the higher order variance of $\widehat{\theta}_{SS}$ is equal to
\begin{align}
&  E\left[  T_{1,SS}^{2}\right]  +\frac{1}{n}\left(  E\left[  \left(
T_{2,SS}\right)  ^{2}\right]  +2\sqrt{n}E\left[  T_{1,SS}T_{2,SS}\right]
+2E\left[  T_{1,SS}T_{3,SS}\right]  \right) \nonumber\\
&  =E\left[  T_{1}^{2}\right]  +\frac{1}{n}\left(
\begin{array}
[c]{c}%
E\left[  \left(  T_{2,SS}\right)  ^{2}\right]  +2\sqrt{n}E\left[
T_{1}T_{2,SS}\right] \\
+2E\left[  T_{1}T_{3,1,SS}\right]  +2E\left[  T_{1}T_{3,2,SS}\right]
\end{array}
\right)  . \label{SS-high-varg}%
\end{align}
We will show that it is equal to
\[
E\left[  T_{1}^{2}\right]  +\frac{1}{n}\frac{k}{k-1}\left(  E\left[
Y_{1,i}^{2}\right]  E\left[  Y_{2,i}^{2}\right]  +\left(  E\left[
Y_{1,i}Y_{2,i}\right]  \right)  ^{2}+O\left(  n^{-1}\right)  \right)
\]

\begin{lemma}
\label{lem-ss-first-two-termsg}%
\begin{align*}
E\left[  \left(  T_{2,SS}\right)  ^{2}\right]   &  =\frac{k}{k-1}\left(
E\left[  Y_{1,i}^{2}\right]  E\left[  Y_{2,i}^{2}\right]  +\left(  E\left[
Y_{1,i}Y_{2,i}\right]  \right)  ^{2}\right)  ,\\
E\left[  T_{1}T_{2,SS}\right]   &  =0.
\end{align*}

\end{lemma}

\begin{proof}
It follows from%
\begin{align*}
E\left[  \left(  T_{2,SS}\right)  ^{2}\right]   &  =E\left[  \left(  \frac
{1}{k-1}\sum_{s\neq s^{\prime}}\mathcal{Y}_{1}^{\left(  s\right)  }%
\mathcal{Y}_{2}^{\left(  s^{\prime}\right)  }\right)  ^{2}\right] \\
&  =\frac{1}{\left(  k-1\right)  ^{2}}k\left(  k-1\right)  E\left[  \left(
\mathcal{Y}_{1}^{\left(  s\right)  }\right)  ^{2}\left(  \mathcal{Y}%
_{2}^{\left(  s^{\prime}\right)  }\right)  ^{2}\right]  \\ &+\frac{k\left(
k-1\right)  }{\left(  k-1\right)  ^{2}}E\left[  \mathcal{Y}_{1}^{\left(
s\right)  }\mathcal{Y}_{2}^{\left(  s\right)  }\right]  E\left[
\mathcal{Y}_{1}^{\left(  s^{\prime}\right)  }\mathcal{Y}_{2}^{\left(
s^{\prime}\right)  }\right] \\
&  =\frac{k}{k-1}\left(  E\left[  Y_{1,i}^{2}\right]  E\left[  Y_{2,i}%
^{2}\right]  +\left(  E\left[  Y_{1,i}Y_{2,i}\right]  \right)  ^{2}\right)
\end{align*}
and%
\begin{align*}
E\left[  T_{1}T_{2,SS}\right]   &  =E\left[  \left(  \frac{1}{\sqrt{k}}%
\sum_{s=1}^{k}\mathcal{Y}_{1}^{\left(  s\right)  }\right)  \left(  \frac
{1}{k-1}\sum_{s\neq s^{\prime}}\mathcal{Y}_{1}^{\left(  s\right)  }%
\mathcal{Y}_{2}^{\left(  s^{\prime}\right)  }\right)  \right] \\
&  =\frac{1}{\sqrt{k}}\frac{1}{k-1}\sum_{s\neq s^{\prime}}\sum_{s^{\prime
\prime}=1}^{k}E\left[  \mathcal{Y}_{1}^{\left(  s\right)  }\mathcal{Y}%
_{2}^{\left(  s^{\prime}\right)  }\mathcal{Y}_{1}^{\left(  s^{\prime\prime
}\right)  }\right]  ,
\end{align*}
and either $s\neq s^{\prime\prime}$ or $s^{\prime}\neq s^{\prime\prime}$ or
both such that $s=s^{\prime}=s^{\prime\prime}$ is excluded. Then there is at
least one index different from the others, which by independence leads to the
zero expectation.
\end{proof}

\begin{lemma}
\label{lem-SS-3-directg}%
\[
E\left[  T_{1}T_{3,SS}\right]  =O\left(  n^{-1}\right)
\]

\end{lemma}

\begin{proof}%
\begin{align*}
E\left[  T_{1}T_{3,1,SS}\right]   &  =E\left[  \mathcal{Y}_{1}\left(  \frac
{k}{\left(  k-1\right)  }\mathcal{Y}_{1}\mathcal{Y}_{3}\mathcal{Y}_{4}%
-\frac{\sqrt{k}}{k-1}\left(  \sum_{s=1}^{k}\mathcal{Y}_{1}^{\left(  s\right)
}\mathcal{Y}_{3}^{\left(  s\right)  }\mathcal{Y}_{4}^{\left(  s\right)
}\right)  \right)  \right] \\
&  =\frac{k}{\left(  k-1\right)  }E\left[  \mathcal{Y}_{1}^{2}\mathcal{Y}%
_{3}\mathcal{Y}_{4}\right]  -\frac{\sqrt{k}}{\left(  k-1\right)  }E\left[
\mathcal{Y}_{1}\left(  \sum_{s=1}^{k}\mathcal{Y}_{1}^{\left(  s\right)
}\mathcal{Y}_{3}^{\left(  s\right)  }\mathcal{Y}_{4}^{\left(  s\right)
}\right)  \right] \\
&  =\frac{k}{\left(  k-1\right)  }E\left[  \mathcal{Y}_{1}^{2}\mathcal{Y}%
_{3}\mathcal{Y}_{4}\right]  -\frac{\sqrt{k}}{k-1}E\left[  \left(  \frac
{1}{\sqrt{k}}\sum_{s=1}^{k}\mathcal{Y}_{1}^{\left(  s\right)  }\right)
\left(  \sum_{s=1}^{k}\mathcal{Y}_{1}^{\left(  s\right)  }\mathcal{Y}%
_{3}^{\left(  s\right)  }\mathcal{Y}_{4}^{\left(  s\right)  }\right)  \right]
\\
&  =\frac{k}{\left(  k-1\right)  }E\left[  \mathcal{Y}_{1}^{2}\mathcal{Y}%
_{3}\mathcal{Y}_{4}\right]  -\frac{1}{k-1}\sum_{s,s^{\prime}}E\left[
\mathcal{Y}_{1}^{\left(  s^{\prime}\right)  }\mathcal{Y}_{1}^{\left(
s\right)  }\mathcal{Y}_{3}^{\left(  s\right)  }\mathcal{Y}_{4}^{\left(
s\right)  }\right] \\
&  =\frac{k}{\left(  k-1\right)  }E\left[  \mathcal{Y}_{1}^{2}\mathcal{Y}%
_{3}\mathcal{Y}_{4}\right]  -\frac{1}{k-1}\sum_{s}E\left[  \mathcal{Y}%
_{1}^{\left(  s\right)  }\mathcal{Y}_{1}^{\left(  s\right)  }\mathcal{Y}%
_{3}^{\left(  s\right)  }\mathcal{Y}_{4}^{\left(  s\right)  }\right] \\
&  =\frac{k}{\left(  k-1\right)  }E\left[  \mathcal{Y}_{1}^{2}\mathcal{Y}%
_{3}\mathcal{Y}_{4}\right]  -\frac{1}{k-1}\sum_{s}E\left[  \left(
\mathcal{Y}_{1}^{\left(  s\right)  }\right)  ^{2}\mathcal{Y}_{3}^{\left(
s\right)  }\mathcal{Y}_{4}^{\left(  s\right)  }\right] \\
&  =\frac{k}{\left(  k-1\right)  }E\left[  \mathcal{Y}_{1}^{2}\mathcal{Y}%
_{3}\mathcal{Y}_{4}\right]  -\frac{k}{k-1}E\left[  \mathcal{Y}_{1}%
^{2}\mathcal{Y}_{3}\mathcal{Y}_{4}\right]  +O\left(  n^{-1}\right) \\
&  =O\left(  n^{-1}\right)
\end{align*}
Likewise, we get $E\left[  T_{1}T_{3,1,SS}\right]  =O\left(  n^{-1}\right)  $,
from which we obtain the conclusion.
\end{proof}

\newpage
\begin{center}
{\LARGE Supplementary Appendix III: Efficient Bias Correction for Cross-section and Panel Data}
\end{center}
\setcounter{section}{0}

\section{Higher Order Analysis of the MLE for AR(1)}

We will go through the formal expansion of the MLE of the AR(1) model%
\begin{equation}
y_{t}=\theta y_{t-1}+\varepsilon_{t}, \label{AR1}%
\end{equation}
where we assume that $\varepsilon_{t}\sim N\left(  0,\sigma^{2}\right)  $. In
order to impose stationarity, we will assume that $y_{0}=\sum_{j=0}^{\infty
}\theta^{j}\varepsilon_{-j}$\textbf{, }when we establish distributional
results.%
\[
\ell\left(  \cdot,\theta\right)  \equiv\left.  \partial\log f\left(
\cdot,\theta\right)  \right/  \partial\theta
\]

The log likelihood is equal to some constant plus%
\[
-\frac{\left(  y_{t}-\theta y_{t-1}\right)  ^{2}}{2\sigma^{2}}%
\]
so the notation from the paper is
\begin{align*}
\ell &  =\frac{y_{t-1}\left(  y_{t}-\theta y_{t-1}\right)  }{\sigma^{2}}%
=\frac{y_{t-1}\varepsilon_{t}}{\sigma^{2}}\\
\ell^{\theta}  &  =-\frac{y_{t-1}^{2}}{\sigma^{2}}\\
\ell^{\theta\theta}  &  =0\\
\mathcal{I}  &  \equiv\frac{1}{\sigma^{2}}E\left[  y_{t-1}^{2}\right]
=\frac{1}{1-\theta^{2}}%
\end{align*}
so%
\begin{align}
T_{1}  &  =\left(  1-\theta^{2}\right)  \frac{1}{\sqrt{T}}\sum_{t=1}^{T}%
\frac{y_{t-1}\varepsilon_{t}}{\sigma^{2}}\label{AR1-T1}\\
T_{2}  &  =-\left(  1-\theta^{2}\right)  ^{2}\left(  \frac{1}{\sqrt{T}}%
\sum_{t=1}^{T}\left(  \frac{y_{t-1}^{2}}{\sigma^{2}}-\frac{1}{1-\theta^{2}%
}\right)  \right)  \left(  \frac{1}{\sqrt{T}}\sum_{t=1}^{T}\frac
{y_{t-1}\varepsilon_{t}}{\sigma^{2}}\right) \label{AR1-T2}\\
T_{3}  &  =\left(  1-\theta^{2}\right)  ^{3}\left(  \frac{1}{\sqrt{T}}%
\sum_{t=1}^{T}\left(  \frac{y_{t-1}^{2}}{\sigma^{2}}-\frac{1}{1-\theta^{2}%
}\right)  \right)  ^{2}\left(  \frac{1}{\sqrt{T}}\sum_{t=1}^{T}\frac
{y_{t-1}\varepsilon_{t}}{\sigma^{2}}\right)  \label{AR1-T3}%
\end{align}

In this section all $o\left(  1\right)  $ and $o_{p}\left(  1\right)  $ terms
are understood to be negligible as $T\rightarrow\infty$. By Proposition
\ref{prop-second-order bias} in Section \ref{app-sec-tech-results}, we have
$\lim_{T\rightarrow\infty}E\left[  T_{2}\right]  =-2\theta$, so we have%
\[
\sqrt{T}\left(  \widehat{\theta}-\theta+\frac{2\theta}{T}\right)  =T_{1}%
+\frac{1}{\sqrt{T}}\left(  T_{2}+2\theta\right)  +\frac{1}{T}T_{3}%
+o_{p}\left(  T^{-1}\right)  ,
\]
so we have%
\begin{align*}
\sqrt{T}\left(  \widehat{\theta}+\frac{2\widehat{\theta}}{T}-\theta\right)
&  =\sqrt{T}\left(  \widehat{\theta}-\theta+\frac{2\theta}{T}\right)
+\sqrt{T}\left(  \frac{2\widehat{\theta}}{T}-\frac{2\theta}{T}\right) \\
&  =T_{1}+\frac{1}{\sqrt{T}}\left(  T_{2}+2\theta\right)  +\frac{1}{T}%
T_{3}+\frac{2}{\sqrt{T}}\left(  \widehat{\theta}-\theta\right)  +o_{p}\left(
T^{-1}\right) \\
&  =T_{1}+\frac{1}{\sqrt{T}}\left(  T_{2}+2\theta\right)  +\frac{1}{T}\left(
T_{3}+2T_{1}\right)  +o_{p}\left(  T^{-1}\right)  .
\end{align*}

This means that the higher order variance of the bias corrected estimator
$\widehat{\theta}+2\widehat{\theta}/T$ is
\[
\operatorname{Var}\left(  T_{1}\right)  +\frac{1}{T}\operatorname{Var}\left(
T_{2}\right)  +\frac{2}{\sqrt{T}}E\left[  T_{1}T_{2}\right]  +\frac{2}%
{T}E\left[  T_{1}\left(  T_{3}+2T_{1}\right)  \right]  .
\]
In order to relate to a result in the iid setting (i.e., Proposition
\ref{higher_var}), write
\begin{align}
\mathcal{Y}_{1}  &  =\left(  1-\theta^{2}\right)  \frac{1}{\sqrt{T}}\sum
_{t=1}^{T}\frac{y_{t-1}\varepsilon_{t}}{\sigma^{2}},\label{AR1-calY1}\\
\mathcal{Y}_{2}  &  =-\left(  1-\theta^{2}\right)  \frac{1}{\sqrt{T}}%
\sum_{t=1}^{T}\left(  \frac{y_{t-1}^{2}}{\sigma^{2}}-\frac{1}{1-\theta^{2}%
}\right)  , \label{AR1-calY2}%
\end{align}
so that
\begin{equation}
T_{2}=\mathcal{Y}_{1}\mathcal{Y}_{2}. \label{AR1-T2-alt}%
\end{equation}
In Proposition \ref{higher_var}, the higher order variance of the bias
corrected estimator in the iid setting was written as%
\[
\mathcal{I}^{-1}+\frac{1}{n}\left(  E[X_{i}^{2}]E[Y_{i}^{2}]+E[X_{i}Y_{i}%
]^{2}\right)  .
\]
It can be recalled from the expression in Proposition (\ref{higher_var}) in the main paper, that $T_{2}$ in the iid
setting was equal to $\left(  n^{-1/2}\sum_{i=1}^{n}X_{i}\right)  \left(
n^{-1/2}\sum_{i=1}^{n}Y_{i}\right)  $.
Also from equation (\ref{V1}), we can see that $\mathcal{I}^{-1}$ in the iid setting is equal to
$\operatorname{Var}\left(  T_{1}\right)  $. So by understanding $n^{-1/2}%
\sum_{i=1}^{n}X_{i}=\mathcal{Y}_{1}$ and $n^{-1/2}\sum_{i=1}^{n}%
Y_{i}=\mathcal{Y}_{2}$, the higher order variance in the iid setting can be
understood as $\operatorname{Var}\left(  T_{1}\right)  +\frac{1}{n}\left(
E\left[  \mathcal{Y}_{1}^{2}\right]  E\left[  \mathcal{Y}_{2}^{2}\right]
+E\left[  \mathcal{Y}_{1}\mathcal{Y}_{2}\right]  ^{2}\right)  $. We prove that
a similar result holds for the AR(1) model. The result for iid samples and for
the AR(1)\ model are slightly different because in the time series case the
terms $E\left[  \mathcal{Y}_{1}^{2}\right]  ,$ $E\left[  \mathcal{Y}_{2}%
^{2}\right]  $ and $E\left[  \mathcal{Y}_{1}\mathcal{Y}_{2}\right]  $ making
up the higher order variance contain long run variance expressions reflecting
the temporal dependence of the data. On the other hand, the expansions for the
AR(1) case are formally the same as for the iid case because the Gaussian
AR(1) model is Markovian and has a likelihood that is formally equivalent to
the likelihood in an iid setting.%

\begin{proposition}
\label{prop-AR1-higher-var}The higher order variance of the bias corrected
estimator $\widehat{\theta}+2\widehat{\theta}/T$ for the AR(1) model is
\[
\operatorname{Var}\left(  T_{1}\right)  +\frac{1}{T}\left(  E\left[
\mathcal{Y}_{1}^{2}\right]  E\left[  \mathcal{Y}_{2}^{2}\right]  +E\left[
\mathcal{Y}_{1}\mathcal{Y}_{2}\right]  ^{2}\right)  ,
\]
where $T_{1}$, $\mathcal{Y}_{1}$, and $\mathcal{Y}_{2}$ are given by
(\ref{AR1-T1}), (\ref{AR1-calY1}), and (\ref{AR1-calY2}).
\end{proposition}

\begin{proof}
It is a consequence of Lemmas \ref{simplification-AR1}, \ref{var-intuition},
and \ref{var-T2-analytic} below.
\end{proof}

\begin{proposition}
\label{prop-AR1-higher-var-SS}The higher order variance of the split sample
estimator for the AR(1) model is
\[
\operatorname{Var}\left(  T_{1}\right)  +\frac{2}{T}\left(  E\left[
\mathcal{Y}_{1}^{2}\right]  E\left[  \mathcal{Y}_{2}^{2}\right]  +E\left[
\mathcal{Y}_{1}\mathcal{Y}_{2}\right]  ^{2}\right)  ,
\]
where $T_{1}$, $\mathcal{Y}_{1}$, and $\mathcal{Y}_{2}$ are given by
(\ref{AR1-T1}), (\ref{AR1-calY1}), and (\ref{AR1-calY2}).
\end{proposition}

\begin{proof}
We work with the generic expression (\ref{eq:generic-expansion}) and
(\ref{SS-high-varg}), which hold whether the environment is iid or the AR(1)
time series model. We present analogs of Lemmas \ref{lem-ss-first-two-termsg}
and \ref{lem-SS-3-directg}.

Lemmas \ref{Chanda} and \ref{cumsum} establish that the conditions of Lemma
\ref{lem-y1y2_T/2} hold for the Gaussian AR(1) model.\textbf{ }As for the
first implication of Lemma \ref{lem-ss-first-two-termsg}, we see that in the
time series with $k=2$, we have%
\begin{align*}
E\left[  \left(  T_{2,SS}\right)  ^{2}\right]   &  =E\left[  \left(
\mathcal{Y}_{1}^{\left(  1\right)  }\mathcal{Y}_{2}^{\left(  2\right)
}+\mathcal{Y}_{1}^{\left(  2\right)  }\mathcal{Y}_{2}^{\left(  1\right)
}\right)  ^{2}\right] \\
&  =E\left[  \left(  \mathcal{Y}_{1}^{\left(  1\right)  }\right)  ^{2}\left(
\mathcal{Y}_{2}^{\left(  2\right)  }\right)  ^{2}\right]  +E\left[  \left(
\mathcal{Y}_{1}^{\left(  2\right)  }\right)  ^{2}\left(  \mathcal{Y}%
_{2}^{\left(  1\right)  }\right)  ^{2}\right]  +2E\left[  \mathcal{Y}%
_{1}^{\left(  1\right)  }\mathcal{Y}_{2}^{\left(  1\right)  }\mathcal{Y}%
_{1}^{\left(  2\right)  }\mathcal{Y}_{2}^{\left(  2\right)  }\right] \\
&  =2E\left[  \mathcal{Y}_{1}^{2}\right]  E\left[  \mathcal{Y}_{2}^{2}\right]
+o\left(  1\right)  +2E\left[  \mathcal{Y}_{1}^{\left(  1\right)  }%
\mathcal{Y}_{2}^{\left(  1\right)  }\mathcal{Y}_{1}^{\left(  2\right)
}\mathcal{Y}_{2}^{\left(  2\right)  }\right]
\end{align*}
where the third equality is based on (\ref{y1y2_T/2_2}) in Lemma
\ref{lem-y1y2_T/2}. By (\ref{y1y2_T/2_4}) in Lemma \ref{lem-y1y2_T/2} we have%
\begin{equation}
E\left[  \mathcal{Y}_{1}^{\left(  1\right)  }\mathcal{Y}_{2}^{\left(
1\right)  }\mathcal{Y}_{1}^{\left(  2\right)  }\mathcal{Y}_{2}^{\left(
2\right)  }\right]  =E\left[  \mathcal{Y}_{1}\mathcal{Y}_{2}\right]
^{2}+o\left(  1\right)  . \label{TBD1}%
\end{equation}
This leads to%
\[
E\left[  \left(  T_{2,SS}\right)  ^{2}\right]  =2\left(  E\left[
\mathcal{Y}_{1}^{2}\right]  E\left[  \mathcal{Y}_{2}^{2}\right]  +E\left[
\mathcal{Y}_{1}\mathcal{Y}_{2}\right]  ^{2}\right)  +o\left(  1\right)  .
\]
As for the second implication of Lemma \ref{lem-ss-first-two-termsg}, we see
that the time series counterpart with $k=2$ is%
\begin{align*}
E\left[  T_{1}T_{2,SS}\right]   &  =E\left[  \left(  \frac{1}{\sqrt{2}%
}\mathcal{Y}_{1}^{\left(  1\right)  }+\frac{1}{\sqrt{2}}\mathcal{Y}%
_{1}^{\left(  2\right)  }\right)  \left(  \mathcal{Y}_{1}^{\left(  1\right)
}\mathcal{Y}_{2}^{\left(  2\right)  }+\mathcal{Y}_{1}^{\left(  2\right)
}\mathcal{Y}_{2}^{\left(  1\right)  }\right)  \right] \\
&  =\frac{1}{\sqrt{2}}E\left[  \left(  \mathcal{Y}_{1}^{\left(  1\right)
}\right)  ^{2}\mathcal{Y}_{2}^{\left(  2\right)  }\right]  +\frac{1}{\sqrt{2}%
}E\left[  \mathcal{Y}_{1}^{\left(  1\right)  }\mathcal{Y}_{1}^{\left(
2\right)  }\mathcal{Y}_{2}^{\left(  2\right)  }\right]
\end{align*}
By (\ref{y1y2_T/2_5}) and (\ref{y1y2_T/2_6}) of Lemma \ref{lem-y1y2_T/2} it
follows that
\begin{align}
E\left[  \left(  \mathcal{Y}_{1}^{\left(  1\right)  }\right)  ^{2}%
\mathcal{Y}_{2}^{\left(  2\right)  }\right]   &  =o\left(  1\right)
,\label{TBD2}\\
E\left[  \mathcal{Y}_{1}^{\left(  1\right)  }\mathcal{Y}_{1}^{\left(
2\right)  }\mathcal{Y}_{2}^{\left(  2\right)  }\right]   &  =o\left(
1\right)  . \label{TBD3}%
\end{align}
As a result, the time series counterpart of Lemma
\ref{lem-ss-first-two-termsg} for $k=2$ states that\textbf{ }%
\begin{align*}
E\left[  \left(  T_{2,SS}\right)  ^{2}\right]   &  =2\left(  E\left[
\mathcal{Y}_{1}^{2}\right]  E\left[  \mathcal{Y}_{2}^{2}\right]  +E\left[
\mathcal{Y}_{1}\mathcal{Y}_{2}\right]  ^{2}\right)  +o\left(  1\right)  ,\\
E\left[  T_{1}T_{2,SS}\right]   &  =o\left(  1\right)  .
\end{align*}
As for the counterpart of Lemma \ref{lem-SS-3-directg}, we have%
\[
T_{3}=\mathcal{Y}_{1}\mathcal{Y}_{2}^{2}%
\]
from (\ref{AR1-T3}), and
\begin{align*}
E\left[  T_{1}T_{3,SS}\right]   &  =E\left[  \mathcal{Y}_{1}\left(
2\mathcal{Y}_{1}\mathcal{Y}_{2}^{2}-\sqrt{2}\left(  \mathcal{Y}_{1}^{\left(
1\right)  }\left(  \mathcal{Y}_{2}^{\left(  1\right)  }\right)  ^{2}%
+\mathcal{Y}_{1}^{\left(  2\right)  }\left(  \mathcal{Y}_{2}^{\left(
2\right)  }\right)  ^{2}\right)  \right)  \right] \\
&  =2E\left[  \mathcal{Y}_{1}^{2}\mathcal{Y}_{2}^{2}\right]  -\sqrt{2}E\left[
\mathcal{Y}_{1}\left(  \mathcal{Y}_{1}^{\left(  1\right)  }\left(
\mathcal{Y}_{2}^{\left(  1\right)  }\right)  ^{2}+\mathcal{Y}_{1}^{\left(
2\right)  }\left(  \mathcal{Y}_{2}^{\left(  2\right)  }\right)  ^{2}\right)
\right] \\
&  =2E\left[  \mathcal{Y}_{1}^{2}\mathcal{Y}_{2}^{2}\right]  -E\left[  \left(
\mathcal{Y}_{1}^{\left(  1\right)  }+\mathcal{Y}_{1}^{\left(  2\right)
}\right)  \left(  \mathcal{Y}_{1}^{\left(  1\right)  }\left(  \mathcal{Y}%
_{2}^{\left(  1\right)  }\right)  ^{2}+\mathcal{Y}_{1}^{\left(  2\right)
}\left(  \mathcal{Y}_{2}^{\left(  2\right)  }\right)  ^{2}\right)  \right] \\
&  =2E\left[  \mathcal{Y}_{1}^{2}\mathcal{Y}_{2}^{2}\right]  -E\left[  \left(
\mathcal{Y}_{1}^{\left(  1\right)  }\right)  ^{2}\left(  \mathcal{Y}%
_{2}^{\left(  1\right)  }\right)  ^{2}\right]  -E\left[  \left(
\mathcal{Y}_{1}^{\left(  2\right)  }\right)  ^{2}\left(  \mathcal{Y}%
_{2}^{\left(  2\right)  }\right)  ^{2}\right] \\
&  -E\left[  \mathcal{Y}_{1}^{\left(  1\right)  }\mathcal{Y}_{1}^{\left(
2\right)  }\left(  \mathcal{Y}_{2}^{\left(  2\right)  }\right)  ^{2}\right]
-E\left[  \mathcal{Y}_{1}^{\left(  1\right)  }\mathcal{Y}_{1}^{\left(
2\right)  }\left(  \mathcal{Y}_{2}^{\left(  1\right)  }\right)  ^{2}\right] \\
&  =o\left(  1\right)  -E\left[  \mathcal{Y}_{1}^{\left(  1\right)
}\mathcal{Y}_{1}^{\left(  2\right)  }\left(  \mathcal{Y}_{2}^{\left(
2\right)  }\right)  ^{2}\right]  -E\left[  \mathcal{Y}_{1}^{\left(  1\right)
}\mathcal{Y}_{1}^{\left(  2\right)  }\left(  \mathcal{Y}_{2}^{\left(
1\right)  }\right)  ^{2}\right] \\
&  =o\left(  1\right)  ,
\end{align*}
where the last equality is based on (\ref{y1y2_T/2_1}) of Lemma
\ref{lem-y1y2_T/2}, and the last equality is based on Lemma
\ref{lem-AR1-4-terms}. Therefore,the higher order variance of the split sample
estimator is%
\[
\operatorname{Var}\left(  T_{1}\right)  +\frac{2}{T}\operatorname{Var}\left(
T_{2}\right)  .
\]

\end{proof}

\begin{lemma}
\label{simplification-AR1}The higher order variance of the bias corrected
estimator $\widehat{\theta}+2\widehat{\theta}/T$ for the AR(1) model is%
\[
\operatorname{Var}\left(  T_{1}\right)  +\frac{1}{T}\operatorname{Var}\left(
T_{2}\right)  +\frac{2}{T}\left(  \operatorname{Var}\left(  T_{2}\right)
-6\theta^{2}-2\right)  ,
\]
where $T_{1}$ and $T_{2}$ are given by (\ref{AR1-T1}) and (\ref{AR1-T2}).
\end{lemma}

\begin{proof}
The higher order variance of the bias corrected estimator $\widehat{\theta
}+2\widehat{\theta}/T$ is
\[
\operatorname{Var}\left(  T_{1}\right)  +\frac{1}{T}\operatorname{Var}\left(
T_{2}\right)  +\frac{2}{\sqrt{T}}E\left[  T_{1}T_{2}\right]  +\frac{2}%
{T}E\left[  T_{1}\left(  T_{3}+2T_{1}\right)  \right]  .
\]
Note that%
\begin{align*}
E\left[  T_{1}T_{3}\right]   &  =\frac{\left(  1-\theta^{2}\right)  ^{4}%
}{\left(  \sigma^{2}\right)  ^{2}}E\left[  \left(  \frac{1}{\sqrt{T}}%
\sum_{t=1}^{T}\left(  \frac{y_{t-1}^{2}}{\sigma^{2}}-\frac{1}{1-\theta^{2}%
}\right)  \right)  ^{2}\left(  \frac{1}{\sqrt{T}}\sum_{t=1}^{T}y_{t-1}%
\varepsilon_{t}\right)  ^{2}\right] \\
&  =E\left[  T_{2}^{2}\right]  =\operatorname{Var}\left(  T_{2}\right)
+\left(  E\left[  T_{2}\right]  \right)  ^{2}=\operatorname{Var}\left(
T_{2}\right)  +\left(  2\theta\right)  ^{2}+O\left(  T^{-1}\right)
\end{align*}
so the higher order variance is%
\begin{align*}
&  \operatorname{Var}\left(  T_{1}\right)  +\frac{1}{T}\operatorname{Var}%
\left(  T_{2}\right)  +\frac{2}{\sqrt{T}}E\left[  T_{1}T_{2}\right]  +\frac
{2}{T}E\left[  T_{1}\left(  T_{3}+2T_{1}\right)  \right] \\
&  =\operatorname{Var}\left(  T_{1}\right)  +\frac{1}{T}\operatorname{Var}%
\left(  T_{2}\right)  \\
&+\frac{2}{T}\left(  \operatorname{Var}\left(
T_{2}\right)  +\left(  2\theta\right)  ^{2}\right)  +\frac{2}{T}\left(
\sqrt{T}E\left[  T_{1}T_{2}\right]  +2\operatorname{Var}\left(  T_{1}\right)
\right)  +O\left(  T^{-2}\right)  .
\end{align*}
Using Proposition \ref{prop-second-order cross-product} in Section
\ref{app-sec-tech-results} along with
\[
\operatorname{Var}\left(  T_{1}\right)  =E\left[  \left(  \left(  1-\theta
^{2}\right)  \frac{1}{\sqrt{T}}\sum_{t=1}^{T}\frac{y_{t-1}\varepsilon_{t}%
}{\sigma^{2}}\right)  ^{2}\right]  =\left(  1-\theta^{2}\right)  ^{2}%
\frac{\frac{\sigma^{2}}{1-\theta^{2}}\sigma^{2}}{\left(  \sigma^{2}\right)
^{2}}=1-\theta^{2}%
\]
we get
\[
\sqrt{T}E\left[  T_{1}T_{2}\right]  +2\operatorname{Var}\left(  T_{1}\right)
=-8\theta^{2}-4+2\left(  1-\theta^{2}\right)  =-8\theta^{2}-2\theta
^{2}-2=-10\theta^{2}-2
\]
so we get that the higher order variance is%
\begin{align*}
&  \operatorname{Var}\left(  T_{1}\right)  +\frac{1}{T}\operatorname{Var}%
\left(  T_{2}\right)  +\frac{2}{T}\left(  \operatorname{Var}\left(
T_{2}\right)  +\left(  2\theta\right)  ^{2}-10\theta^{2}-2\right) \\
&  =\operatorname{Var}\left(  T_{1}\right)  +\frac{1}{T}\operatorname{Var}%
\left(  T_{2}\right)  +\frac{2}{T}\left(  \operatorname{Var}\left(
T_{2}\right)  -6\theta^{2}-2\right)  .
\end{align*}

\end{proof}

\begin{lemma}
\label{var-intuition}%
\[
\operatorname{Var}\left(  T_{2}\right)  =\operatorname{Var}\left(
\mathcal{Y}_{1}\mathcal{Y}_{2}\right)  =E\left[  \mathcal{Y}_{1}^{2}\right]
E\left[  \mathcal{Y}_{2}^{2}\right]  +\left(  E\left[  \mathcal{Y}%
_{1}\mathcal{Y}_{2}\right]  \right)  ^{2}+o\left(  1\right)  .
\]

\end{lemma}

\begin{proof}
By (\ref{eq-AR1-split-alt}) and Lemma \ref{lem-y1y2_T/2}, we get
\begin{align*}
E\left[  \mathcal{Y}_{1}^{2}\mathcal{Y}_{2}^{2}\right]   &  =\frac{1}%
{4}\left(
\begin{array}
[c]{c}%
E\left[  \mathcal{Y}_{1}^{2}\mathcal{Y}_{2}^{2}\right]  +E\left[
\mathcal{Y}_{1}^{2}\right]  E\left[  \mathcal{Y}_{2}^{2}\right] \\
+4E\left[  \mathcal{Y}_{1}\mathcal{Y}_{2}\right]  ^{2}\\
+E\left[  \mathcal{Y}_{1}^{2}\right]  E\left[  \mathcal{Y}_{2}^{2}\right]
+E\left[  \mathcal{Y}_{1}^{2}\mathcal{Y}_{2}^{2}\right]
\end{array}
\right)  +o\left(  1\right) \\
&  =\frac{1}{2}E\left[  \mathcal{Y}_{1}^{2}\mathcal{Y}_{2}^{2}\right]
+\frac{1}{2}E\left[  \mathcal{Y}_{1}^{2}\right]  E\left[  \mathcal{Y}_{2}%
^{2}\right]  +E\left[  \mathcal{Y}_{1}\mathcal{Y}_{2}\right]  ^{2}+o\left(
1\right)  ,
\end{align*}
which can be solved for $E\left[  \mathcal{Y}_{1}^{2}\mathcal{Y}_{2}%
^{2}\right]  $, implying%
\[
E\left[  \mathcal{Y}_{1}^{2}\mathcal{Y}_{2}^{2}\right]  =E\left[
\mathcal{Y}_{1}^{2}\right]  E\left[  \mathcal{Y}_{2}^{2}\right]  +2E\left[
\mathcal{Y}_{1}\mathcal{Y}_{2}\right]  ^{2}+o\left(  1\right)  .
\]
Therefore, we have%
\begin{align*}
\operatorname{Var}\left(  \mathcal{Y}_{1}\mathcal{Y}_{2}\right)   &  =E\left[
\mathcal{Y}_{1}^{2}\mathcal{Y}_{2}^{2}\right]  -E\left[  \mathcal{Y}%
_{1}\mathcal{Y}_{2}\right]  ^{2}\\
&  =E\left[  \mathcal{Y}_{1}^{2}\right]  E\left[  \mathcal{Y}_{2}^{2}\right]
+\left(  E\left[  \mathcal{Y}_{1}\mathcal{Y}_{2}\right]  \right)
^{2}+o\left(  1\right)  ,
\end{align*}
which proves Lemma \ref{var-intuition}.
\end{proof}

\begin{lemma}
\label{var-T2-analytic}%
\[
\operatorname{Var}\left(  T_{2}\right)  -6\theta^{2}-2=o\left(  1\right)
\]

\end{lemma}

\begin{proof}
By Lemma \ref{lem-var-Y2} in Section \ref{app-sec-tech-results}, we have
\begin{align*}
\operatorname{Var}\left(  \mathcal{Y}_{1}\right)  \operatorname{Var}\left(
\mathcal{Y}_{2}\right)   &  =\left(  1-\theta^{2}\right)  ^{4}%
\operatorname{Var}\left(  \frac{1}{\sqrt{T}}\sum_{t=1}^{T}\left(
\frac{y_{t-1}^{2}}{\sigma^{2}}-\frac{1}{1-\theta^{2}}\right)  \right)
\operatorname{Var}\left(  \frac{1}{\sqrt{T}}\sum_{t=1}^{T}\frac{y_{t-1}%
\varepsilon_{t}}{\sigma^{2}}\right) \\
&  =\left(  1-\theta^{2}\right)  ^{4}\frac{2\left(  1+\theta^{2}\right)
}{\left(  1-\theta^{2}\right)  ^{3}}\operatorname{Var}\left(  \frac{1}%
{\sqrt{T}}\sum_{t=1}^{T}\frac{y_{t-1}\varepsilon_{t}}{\sigma^{2}}\right)
+O\left(  T^{-1}\right) \\
&  =\left(  1-\theta^{2}\right)  ^{4}\frac{2\left(  1+\theta^{2}\right)
}{\left(  1-\theta^{2}\right)  ^{3}}\frac{1}{1-\theta^{2}}+O\left(
T^{-1}\right)  =2\left(  1+\theta^{2}\right)  +o\left(  1\right)
\end{align*}
Also,%
\[
E\left[  \mathcal{Y}_{1}\mathcal{Y}_{2}\right]  =-2\theta
\]
by Proposition \ref{prop-second-order bias} in Section
\ref{app-sec-tech-results}, so we have%
\[
\operatorname{Var}\left(  T_{2}\right)  =2\left(  1+\theta^{2}\right)
+\left(  -2\theta\right)  ^{2}+o\left(  1\right)  =2+6\theta^{2}+o\left(
1\right)
\]

\end{proof}

\section{Technical Results\label{app-sec-tech-results}}

\begin{remark}
In this section, it is assumed below that $s<t<u$.
\end{remark}

\begin{proposition}
\label{prop-second-order bias}The second order bias is%
\[
E\left[  T_{2}\right]  =E\left[  -\left(  1-\theta^{2}\right)  ^{2}\left(
\frac{1}{\sqrt{T}}\sum_{t=1}^{T}\left(  \frac{y_{t-1}^{2}}{\sigma^{2}}%
-\frac{1}{1-\theta^{2}}\right)  \right)  \left(  \frac{1}{\sqrt{T}}\sum
_{t=1}^{T}\frac{y_{t-1}\varepsilon_{t}}{\sigma^{2}}\right)  \right]
=-2\theta+O\left(  T^{-1}\right)
\]

\end{proposition}

\begin{proof}
We will write%
\begin{align*}
&  \left(  \frac{1}{\sqrt{T}}\sum_{t=1}^{T}\left(  \frac{y_{t-1}^{2}}%
{\sigma^{2}}-\frac{1}{1-\theta^{2}}\right)  \right)  \left(  \frac{1}{\sqrt
{T}}\sum_{t=1}^{T}\frac{y_{t-1}\varepsilon_{t}}{\sigma^{2}}\right) \\
&  =\frac{1}{T}\sum_{t=1}^{T}\left(  \frac{y_{t-1}^{2}}{\sigma^{2}}-\frac
{1}{1-\theta^{2}}\right)  \frac{y_{t-1}\varepsilon_{t}}{\sigma^{2}}\\
&  +\frac{1}{T}\sum_{s=1}^{T}\sum_{t>s}\left(  \frac{y_{s-1}^{2}}{\sigma^{2}%
}-\frac{1}{1-\theta^{2}}\right)  \frac{y_{t-1}\varepsilon_{t}}{\sigma^{2}}\\
&  +\frac{1}{T}\sum_{s=1}^{T}\sum_{t>s}\left(  \frac{y_{t-1}^{2}}{\sigma^{2}%
}-\frac{1}{1-\theta^{2}}\right)  \frac{y_{s-1}\varepsilon_{s}}{\sigma^{2}}%
\end{align*}
and note that the first two sums on the right should have zero expectations.
The conclusion is obtained by using Lemma \ref{lem-prop1-step3}.
\end{proof}

\begin{proposition}
\label{prop-second-order cross-product}%
\[
\sqrt{T}E\left[  T_{1}T_{2}\right]  =-8\theta^{2}-4+O\left(  \frac{1}%
{T}\right)
\]

\end{proposition}

\begin{proof}
We write%
\begin{align*}
&  \left(  \sum_{t=1}^{T}\left(  y_{t-1}^{2}-\frac{\sigma^{2}}{1-\theta^{2}%
}\right)  \right)  \left(  \sum_{t=1}^{T}y_{t-1}\varepsilon_{t}\right)  ^{2}\\
&  =\sum_{t=1}^{T}\left(  y_{t-1}^{2}-\frac{\sigma^{2}}{1-\theta^{2}}\right)
\left(  y_{t-1}\varepsilon_{t}\right)  ^{2}\\
&  +\sum_{s=1}^{T}\sum_{t>s}\left(
\begin{array}
[c]{c}%
\left(  y_{s-1}^{2}-\frac{\sigma^{2}}{1-\theta^{2}}\right)  \left(
y_{t-1}\varepsilon_{t}\right)  \left(  y_{t-1}\varepsilon_{t}\right)
+2\left(  y_{t-1}^{2}-\frac{\sigma^{2}}{1-\theta^{2}}\right)  \left(
y_{s-1}\varepsilon_{s}\right)  \left(  y_{t-1}\varepsilon_{t}\right) \\
+\left(  y_{t-1}^{2}-\frac{\sigma^{2}}{1-\theta^{2}}\right)  \left(
y_{s-1}\varepsilon_{t}\right)  \left(  y_{s-1}\varepsilon_{s}\right)
+2\left(  y_{s-1}^{2}-\frac{\sigma^{2}}{1-\theta^{2}}\right)  \left(
y_{s-1}\varepsilon_{s}\right)  \left(  y_{t-1}\varepsilon_{t}\right)
\end{array}
\right) \\
&  +\sum_{s<t<u}\left(
\begin{array}
[c]{c}%
\left(  y_{s-1}^{2}-\frac{\sigma^{2}}{1-\theta^{2}}\right)  \left(
y_{t-1}\varepsilon_{t}\right)  \left(  y_{u-1}\varepsilon_{u}\right)  +\left(
y_{s-1}^{2}-\frac{\sigma^{2}}{1-\theta^{2}}\right)  \left(  y_{u-1}%
\varepsilon_{u}\right)  \left(  y_{t-1}\varepsilon_{t}\right) \\
+\left(  y_{u-1}^{2}-\frac{\sigma^{2}}{1-\theta^{2}}\right)  \left(
y_{s-1}\varepsilon_{s}\right)  \left(  y_{t-1}\varepsilon_{t}\right)  +\left(
y_{u-1}^{2}-\frac{\sigma^{2}}{1-\theta^{2}}\right)  \left(  y_{t-1}%
\varepsilon_{t}\right)  \left(  y_{s-1}\varepsilon_{s}\right) \\
+\left(  y_{t-1}^{2}-\frac{\sigma^{2}}{1-\theta^{2}}\right)  \left(
y_{s-1}\varepsilon_{s}\right)  \left(  y_{u-1}\varepsilon_{u}\right)  +\left(
y_{t-1}^{2}-\frac{\sigma^{2}}{1-\theta^{2}}\right)  \left(  y_{u-1}%
\varepsilon_{u}\right)  \left(  y_{s-1}\varepsilon_{s}\right)
\end{array}
\right)
\end{align*}
Using Lemmas \ref{lem-prop2-stepI-summary}, , \ref{lem-prop2-stepII-summary},
and \ref{lem-prop2-stepIII-summary}, we conclude that the expectation of the
above is equal to%
\begin{align*}
&  \frac{1}{T}E\left[  \left(  \sum_{t=1}^{T}\left(  y_{t-1}^{2}-\frac
{\sigma^{2}}{1-\theta^{2}}\right)  \right)  \left(  \sum_{t=1}^{T}%
y_{t-1}\varepsilon_{t}\right)  ^{2}\right] \\
&  =\frac{2\left(  \sigma^{2}\right)  ^{3}}{\left(  1-\theta^{2}\right)  ^{2}%
}+\frac{2\left(  \sigma^{2}\right)  ^{3}}{\left(  1-\theta^{2}\right)  ^{3}%
}\left(  \theta^{2}+1\right)  +\frac{8\theta^{2}\left(  \sigma^{2}\right)
^{3}}{\left(  1-\theta^{2}\right)  ^{3}}+O\left(  \frac{1}{T}\right)
\end{align*}
The conclusion follows from%
\begin{align*}
\sqrt{T}E\left[  T_{1}T_{2}\right]   &  =-\frac{\left(  1-\theta^{2}\right)
^{3}}{\left(  \sigma^{2}\right)  ^{3}}\frac{1}{T}E\left[  \left(  \sum
_{t=1}^{T}\left(  y_{t-1}^{2}-\frac{\sigma^{2}}{1-\theta^{2}}\right)  \right)
\left(  \sum_{t=1}^{T}y_{t-1}\varepsilon_{t}\right)  ^{2}\right] \\
&  =-8\theta^{2}-4+O\left(  \frac{1}{T}\right)
\end{align*}

\end{proof}

\begin{lemma}
\label{lem-prop1-step2}%
\begin{align*}
E\left[  y_{s-1}\varepsilon_{s}y_{s+j}^{2}\right]   &  =\frac{2\theta
^{2j+1}\sigma^{4}}{1-\theta^{2}}\\
E\left[  y_{s-1}\varepsilon_{s}\left(  \frac{y_{s+j}^{2}}{\sigma^{2}}-\frac
{1}{1-\theta^{2}}\right)  \right]   &  =\frac{2\theta^{2j+1}\sigma^{2}%
}{1-\theta^{2}}%
\end{align*}

\end{lemma}

\begin{proof}
We have%
\begin{align*}
y_{s-1}\varepsilon_{s}y_{s}^{2}  &  =y_{s-1}\varepsilon_{s}\left(  \theta
y_{s-1}+\varepsilon_{s}\right)  ^{2}\\
&  =y_{s-1}\varepsilon_{s}\left(  \theta^{2}y_{s-1}^{2}+2\theta y_{s-1}%
\varepsilon_{s}+\varepsilon_{s}^{2}\right) \\
&  =\theta^{2}y_{s-1}^{3}\varepsilon_{s}+2\theta y_{s-1}^{2}\varepsilon
_{s}^{2}+y_{s-1}\varepsilon_{s}^{3}%
\end{align*}
so
\[
E\left[  y_{s-1}\varepsilon_{s}y_{s}^{2}\right]  =2\theta E\left[  y_{s-1}%
^{2}\right]  \sigma^{2}=2\theta\frac{\sigma^{2}}{1-\theta^{2}}\sigma^{2}%
=\frac{2\theta\sigma^{4}}{1-\theta^{2}}.
\]
We also have%
\begin{align*}
y_{s-1}\varepsilon_{s}y_{s+1}^{2}  &  =y_{s-1}\varepsilon_{s}\left(  \theta
y_{s}+\varepsilon_{s+1}\right)  ^{2}\\
&  =y_{s-1}\varepsilon_{s}\left(  \theta^{2}y_{s}^{2}+2\theta y_{s}%
\varepsilon_{s+1}+\varepsilon_{s+1}^{2}\right) \\
&  =\theta^{2}y_{s-1}\varepsilon_{s}y_{s}^{2}+2\theta y_{s-1}\varepsilon
_{s}y_{s}\varepsilon_{s+1}+y_{s-1}\varepsilon_{s}\varepsilon_{s+1}^{2}%
\end{align*}
so%
\[
E\left[  y_{s-1}\varepsilon_{s}y_{s+1}^{2}\right]  =\theta^{2}E\left[
y_{s-1}\varepsilon_{s}y_{s}^{2}\right]
\]
In generl, we have%
\begin{align*}
y_{s-1}\varepsilon_{s}y_{s+j+1}^{2}  &  =y_{s-1}\varepsilon_{s}\left(  \theta
y_{s+j+1}+\varepsilon_{s+j+1}\right)  ^{2}\\
&  =y_{s-1}\varepsilon_{s}\left(  \theta^{2}y_{s+j}^{2}+2\theta y_{s+j}%
\varepsilon_{s+j+1}+\varepsilon_{s+j+1}^{2}\right) \\
&  =\theta^{2}y_{s-1}\varepsilon_{s}y_{s+j}^{2}+2\theta y_{s-1}\varepsilon
_{s}y_{s+j}\varepsilon_{s+j+1}+y_{s-1}\varepsilon_{s}\varepsilon_{s+j+1}^{2}%
\end{align*}
so%
\[
E\left[  y_{s-1}\varepsilon_{s}y_{s+j+1}^{2}\right]  =\theta^{2}E\left[
y_{s-1}\varepsilon_{s}y_{s+j}^{2}\right]
\]
so we obtain
\[
E\left[  y_{s-1}\varepsilon_{s}y_{s+j}^{2}\right]  =\frac{2\theta^{2j+1}%
\sigma^{4}}{1-\theta^{2}}%
\]
from which we obtain%
\[
E\left[  y_{s-1}\varepsilon_{s}\left(  \frac{y_{s+j}^{2}}{\sigma^{2}}-\frac
{1}{1-\theta^{2}}\right)  \right]  =E\left[  \frac{y_{s-1}\varepsilon
_{s}y_{s+j}^{2}}{\sigma^{2}}-\frac{y_{s-1}\varepsilon_{s}}{1-\theta^{2}%
}\right]  =\frac{2\theta^{2j-1}\sigma^{2}}{1-\theta^{2}}%
\]

\end{proof}

\begin{lemma}
\label{lem-prop1-step3-prep}%
\[
\sum_{t=s+1}^{T}E\left[  \left(  \frac{y_{t-1}^{2}}{\sigma^{2}}-\frac
{1}{1-\theta^{2}}\right)  \frac{y_{s-1}\varepsilon_{s}}{\sigma^{2}}\right]
=\frac{2\theta}{\left(  1-\theta^{2}\right)  ^{2}}\left(  1-\left(  \theta
^{2}\right)  ^{T-s}\right)
\]

\end{lemma}

\begin{proof}
It follows from%
\begin{align*}
&  \sum_{t=s+1}^{T}E\left[  \left(  \frac{y_{t-1}^{2}}{\sigma^{2}}-\frac
{1}{1-\theta^{2}}\right)  \frac{y_{s-1}\varepsilon_{s}}{\sigma^{2}}\right] \\
&  =E\left[  \frac{y_{s-1}\varepsilon_{s}}{\sigma^{2}}\left(  \frac{y_{s}^{2}%
}{\sigma^{2}}-\frac{1}{1-\theta^{2}}\right)  \right]  +\cdots+E\left[
\frac{y_{s-1}\varepsilon_{s}}{\sigma^{2}}\left(  \frac{y_{T-1}^{2}}{\sigma
^{2}}-\frac{1}{1-\theta^{2}}\right)  \right] \\
&  =\frac{2\theta}{1-\theta^{2}}+\frac{2\theta^{3}}{1-\theta^{2}}%
+\frac{2\theta^{2\left(  T-1-s\right)  +1}}{1-\theta^{2}}\\
&  =\frac{2\theta}{1-\theta^{2}}\left(  1+\theta^{2}+\cdots+\left(  \theta
^{2}\right)  ^{T-1-s}\right) \\
&  =\frac{2\theta}{1-\theta^{2}}\frac{1-\left(  \theta^{2}\right)  ^{T-s}%
}{1-\theta^{2}}%
\end{align*}
where we used Lemma \ref{lem-prop1-step2} for the second equality.
\end{proof}

\begin{lemma}
\label{lem-prop1-step3}%
\[
\frac{1}{T}\sum_{s=1}^{T}\sum_{t=s+1}^{T}E\left[  \left(  \frac{y_{t-1}^{2}%
}{\sigma^{2}}-\frac{1}{1-\theta^{2}}\right)  \frac{y_{s-1}\varepsilon_{s}%
}{\sigma^{2}}\right]  =\frac{2\theta}{\left(  1-\theta^{2}\right)  ^{2}}%
-\frac{1}{T}\frac{2\theta\left(  1-\left(  \theta^{2}\right)  ^{T}\right)
}{\left(  1-\theta^{2}\right)  ^{3}}%
\]

\end{lemma}

\begin{proof}%
\begin{align*}
\frac{1}{T}\sum_{s=1}^{T}\sum_{t=s+1}^{T}E\left[  \left(  \frac{y_{t-1}^{2}%
}{\sigma^{2}}-\frac{1}{1-\theta^{2}}\right)  \frac{y_{s-1}\varepsilon_{s}%
}{\sigma^{2}}\right]   &  =\frac{1}{T}\sum_{s=1}^{T}\frac{2\theta}{\left(
1-\theta^{2}\right)  ^{2}}\left(  1-\left(  \theta^{2}\right)  ^{T-s}\right)
\\
&  =\frac{2\theta}{\left(  1-\theta^{2}\right)  ^{2}}-\frac{2\theta}{\left(
1-\theta^{2}\right)  ^{2}}\frac{1}{T}\sum_{s=1}^{T}\left(  \theta^{2}\right)
^{T-s}\\
&  =\frac{2\theta}{\left(  1-\theta^{2}\right)  ^{2}}-\frac{2\theta}{\left(
1-\theta^{2}\right)  ^{2}}\frac{1}{T}\sum_{j=0}^{T-1}\left(  \theta
^{2}\right)  ^{j}\\
&  =\frac{2\theta}{\left(  1-\theta^{2}\right)  ^{2}}-\frac{2\theta}{\left(
1-\theta^{2}\right)  ^{2}}\frac{1}{T}\frac{1-\left(  \theta^{2}\right)  ^{T}%
}{1-\theta^{2}}\\
&  =\frac{2\theta}{\left(  1-\theta^{2}\right)  ^{2}}-\frac{1}{T}\frac
{2\theta\left(  1-\left(  \theta^{2}\right)  ^{T}\right)  }{\left(
1-\theta^{2}\right)  ^{3}}%
\end{align*}
where we used Lemma \ref{lem-prop1-step3-prep} for the first equality.
\end{proof}

\begin{lemma}
\label{lem-prop2-stepI-summary}%
\[
E\left[  \left(  y_{t-1}^{2}-\frac{\sigma^{2}}{1-\theta^{2}}\right)  \left(
y_{t-1}\varepsilon_{t}\right)  ^{2}\right]  =\frac{2\left(  \sigma^{2}\right)
^{3}}{\left(  1-\theta^{2}\right)  ^{2}}%
\]
and therefore%
\[
\frac{1}{T}E\left[  \sum_{t=1}^{T}\left(  y_{t-1}^{2}-\frac{\sigma^{2}%
}{1-\theta^{2}}\right)  \left(  y_{t-1}\varepsilon_{t}\right)  ^{2}\right]
=\frac{2\left(  \sigma^{2}\right)  ^{3}}{\left(  1-\theta^{2}\right)  ^{2}}%
\]

\end{lemma}

\begin{proof}
We have%
\begin{align*}
E\left[  \left(  y_{t-1}^{2}-\frac{\sigma^{2}}{1-\theta^{2}}\right)  \left(
y_{t-1}\varepsilon_{t}\right)  ^{2}\right]   &  =E\left[  \left(  y_{t-1}%
^{2}-\frac{\sigma^{2}}{1-\theta^{2}}\right)  y_{t-1}^{2}\varepsilon_{t}%
^{2}\right] \\
&  =E\left[  \left(  y_{t-1}^{2}-\frac{\sigma^{2}}{1-\theta^{2}}\right)
y_{t-1}^{2}\right]  \sigma^{2}\\
&  =\left(  E\left[  y_{t-1}^{4}\right]  -\frac{\sigma^{2}}{1-\theta^{2}%
}E\left[  y_{t-1}^{2}\right]  \right)  \sigma^{2}%
\end{align*}
Using
\[
y_{0}\sim N\left(  0,\frac{\sigma^{2}}{1-\theta^{2}}\right)
\]
we get%
\begin{align*}
E\left[  y_{t-1}^{4}\right]   &  =3\left(  \frac{\sigma^{2}}{1-\theta^{2}%
}\right)  ^{2}\\
E\left[  y_{t-1}^{2}\right]   &  =\frac{\sigma^{2}}{1-\theta^{2}}%
\end{align*}
so we conclude%
\[
E\left[  \left(  y_{t-1}^{2}-\frac{\sigma^{2}}{1-\theta^{2}}\right)  \left(
y_{t-1}\varepsilon_{t}\right)  ^{2}\right]  =\frac{2\left(  \sigma^{2}\right)
^{3}}{\left(  1-\theta^{2}\right)  ^{2}}%
\]

\end{proof}

\begin{lemma}
\label{Eys2yt2}%
\[
E\left[  y_{s}^{2}y_{s+j}^{2}\right]  =\left(  1+2\left(  \theta^{2}\right)
^{j}\right)  \left(  \frac{\sigma^{2}}{1-\theta^{2}}\right)  ^{2}%
\]

\end{lemma}

\begin{proof}
We note that%
\begin{align*}
E\left[  y_{s}^{2}y_{s+j+1}^{2}\right]   &  =E\left[  y_{s}^{2}\left(  \theta
y_{s+j}+\varepsilon_{s+j+1}\right)  ^{2}\right] \\
&  =\theta^{2}E\left[  y_{s}^{2}y_{s+j}^{2}\right]  +2\theta E\left[
y_{s}^{2}y_{s+j}\varepsilon_{s+j+1}\right]  +E\left[  y_{s}^{2}\varepsilon
_{s+j+1}^{2}\right] \\
&  =\theta^{2}E\left[  y_{s}^{2}y_{s+j}^{2}\right]  +E\left[  y_{s}%
^{2}\right]  \sigma^{2}\\
&  =\theta^{2}E\left[  y_{s}^{2}y_{s+j}^{2}\right]  +\frac{\left(  \sigma
^{2}\right)  ^{2}}{1-\theta^{2}}%
\end{align*}
and%
\[
E\left[  y_{s}^{2}y_{s}^{2}\right]  =E\left[  y_{s}^{4}\right]  =3\left(
\frac{\sigma^{2}}{1-\theta^{2}}\right)  ^{2}%
\]
It follows that%
\begin{align*}
E\left[  y_{s}^{2}y_{s+1}^{2}\right]   &  =3\theta^{2}\left(  \frac{\sigma
^{2}}{1-\theta^{2}}\right)  ^{2}+\frac{\left(  \sigma^{2}\right)  ^{2}%
}{1-\theta^{2}},\\
E\left[  y_{s}^{2}y_{s+2}^{2}\right]   &  =\theta^{2}\left(  3\theta
^{2}\left(  \frac{\sigma^{2}}{1-\theta^{2}}\right)  ^{2}+\frac{\left(
\sigma^{2}\right)  ^{2}}{1-\theta^{2}}\right)  +\frac{\left(  \sigma
^{2}\right)  ^{2}}{1-\theta^{2}}\\
&  =3\left(  \theta^{2}\right)  ^{2}\left(  \frac{\sigma^{2}}{1-\theta^{2}%
}\right)  ^{2}+\left(  \theta^{2}+1\right)  \frac{\left(  \sigma^{2}\right)
^{2}}{1-\theta^{2}},\\
E\left[  y_{s}^{2}y_{s+3}^{2}\right]   &  =\theta^{2}\left(  3\left(
\theta^{2}\right)  ^{2}\left(  \frac{\sigma^{2}}{1-\theta^{2}}\right)
^{2}+\left(  \theta^{2}+1\right)  \frac{\left(  \sigma^{2}\right)  ^{2}%
}{1-\theta^{2}}\right)  +\frac{\left(  \sigma^{2}\right)  ^{2}}{1-\theta^{2}%
}\\
&  =3\left(  \theta^{2}\right)  ^{3}\left(  \frac{\sigma^{2}}{1-\theta^{2}%
}\right)  ^{2}+\left(  \left(  \theta^{2}\right)  ^{2}+\theta^{2}+1\right)
\frac{\left(  \sigma^{2}\right)  ^{2}}{1-\theta^{2}},
\end{align*}
and we can conclude that%
\begin{align*}
E\left[  y_{s}^{2}y_{s+j}^{2}\right]   &  =3\left(  \theta^{2}\right)
^{j}\left(  \frac{\sigma^{2}}{1-\theta^{2}}\right)  ^{2}+\left(  \left(
\theta^{2}\right)  ^{j-1}+\left(  \theta^{2}\right)  ^{j-2}+\cdots+1\right)
\frac{\left(  \sigma^{2}\right)  ^{2}}{1-\theta^{2}}\\
&  =3\left(  \theta^{2}\right)  ^{j}\left(  \frac{\sigma^{2}}{1-\theta^{2}%
}\right)  ^{2}+\frac{1-\left(  \theta^{2}\right)  ^{j}}{1-\theta^{2}}%
\frac{\left(  \sigma^{2}\right)  ^{2}}{1-\theta^{2}}\\
&  =3\left(  \theta^{2}\right)  ^{j}\left(  \frac{\sigma^{2}}{1-\theta^{2}%
}\right)  ^{2}+\left(  1-\left(  \theta^{2}\right)  ^{j}\right)  \left(
\frac{\sigma^{2}}{1-\theta^{2}}\right)  ^{2}\\
&  =\left(  1+2\left(  \theta^{2}\right)  ^{j}\right)  \left(  \frac
{\sigma^{2}}{1-\theta^{2}}\right)  ^{2}%
\end{align*}

\end{proof}

\begin{lemma}
\label{lem-prop2-step2}%
\[
E\left[  \left(  y_{s-1}^{2}-\frac{\sigma^{2}}{1-\theta^{2}}\right)  \left(
y_{t-1}\varepsilon_{t}\right)  ^{2}\right]  =2\left(  \theta^{2}\right)
^{t-s}\frac{\left(  \sigma^{2}\right)  ^{3}}{\left(  1-\theta^{2}\right)
^{2}}%
\]

\end{lemma}

\begin{proof}
We have%
\begin{align*}
E\left[  \left(  y_{s-1}^{2}-\frac{\sigma^{2}}{1-\theta^{2}}\right)
y_{t-1}^{2}\varepsilon_{t}^{2}\right]   &  =E\left[  \left(  y_{s-1}^{2}%
-\frac{\sigma^{2}}{1-\theta^{2}}\right)  y_{t-1}^{2}\right]  \sigma^{2}\\
&  =\left(  E\left[  y_{s-1}^{2}y_{t-1}^{2}\right]  -\frac{\sigma^{2}%
}{1-\theta^{2}}E\left[  y_{t-1}^{2}\right]  \right)  \sigma^{2}\\
&  =\left(  \left(  1+2\left(  \theta^{2}\right)  ^{t-s}\right)  \left(
\frac{\sigma^{2}}{1-\theta^{2}}\right)  ^{2}-\frac{\sigma^{2}}{1-\theta^{2}%
}\frac{\sigma^{2}}{1-\theta^{2}}\right)  \sigma^{2}\\
&  =2\left(  \theta^{2}\right)  ^{t-s}\frac{\left(  \sigma^{2}\right)  ^{3}%
}{\left(  1-\theta^{2}\right)  ^{2}}%
\end{align*}
where we used Lemma \ref{Eys2yt2} for the third equality.
\end{proof}

\begin{lemma}
\label{lem-prop2-step3&5}%
\begin{align*}
E\left[  \left(  y_{t-1}^{2}-\frac{\sigma^{2}}{1-\theta^{2}}\right)  \left(
y_{s-1}\varepsilon_{s}\right)  \left(  y_{t-1}\varepsilon_{t}\right)  \right]
&  =0\\
E\left[  \left(  y_{s-1}^{2}-\frac{\sigma^{2}}{1-\theta^{2}}\right)  \left(
y_{s-1}\varepsilon_{s}\right)  \left(  y_{t-1}\varepsilon_{t}\right)  \right]
&  =0
\end{align*}

\end{lemma}

\begin{proof}
Conditioning on up to the random variables up to $t-1$, we get%
\[
E_{s}\left[  \left(  y_{t-1}^{2}-\frac{\sigma^{2}}{1-\theta^{2}}\right)
\left(  y_{s-1}\varepsilon_{s}\right)  \left(  y_{t-1}\varepsilon_{t}\right)
\right]  =\left(  y_{s-1}\varepsilon_{s}\right)  \left(  y_{t-1}^{2}%
-\frac{\sigma^{2}}{1-\theta^{2}}\right)  y_{t-1}E_{t-1}\left[  \varepsilon
_{t}\right]  =0
\]
from which we get the first result. The second result is obtained similarly.
\end{proof}

\begin{lemma}
\label{lem-prop2-step4-prep}%
\[
E\left[  \left(  y_{s-1}\varepsilon_{s}\right)  ^{2}y_{t-1}^{2}\right]
=\frac{\left(  \sigma^{2}\right)  ^{3}}{\left(  1-\theta^{2}\right)  ^{2}%
}+\frac{2\left(  \sigma^{2}\right)  ^{3}}{\left(  1-\theta^{2}\right)  ^{2}%
}\left(  \theta^{2}\right)  ^{t-s-1}%
\]

\end{lemma}

\begin{proof}
We have%
\begin{align*}
\left(  y_{s-1}\varepsilon_{s}\right)  ^{2}y_{s}^{2}  &  =\left(
y_{s-1}\varepsilon_{s}\right)  ^{2}\left(  \theta y_{s-1}+\varepsilon
_{s}\right)  ^{2}\\
&  =y_{s-1}^{2}\varepsilon_{s}^{2}\left(  \theta^{2}y_{s-1}^{2}+2\theta
y_{s-1}\varepsilon_{s}+\varepsilon_{s}^{2}\right) \\
&  =\theta^{2}y_{s-1}^{4}\varepsilon_{s}^{2}+2\theta y_{s-1}^{3}%
\varepsilon_{s}^{3}+y_{s-1}^{2}\varepsilon_{s}^{4}%
\end{align*}
so
\begin{align*}
E\left[  \left(  y_{s-1}\varepsilon_{s}\right)  ^{2}y_{s}^{2}\right]   &
=\theta^{2}E\left[  y_{s-1}^{4}\right]  \sigma^{2}+E\left[  y_{s-1}%
^{2}\right]  3\left(  \sigma^{2}\right)  ^{2}\\
&  =\theta^{2}\left(  3\left(  \frac{\sigma^{2}}{1-\theta^{2}}\right)
^{2}\right)  \sigma^{2}+\left(  \frac{\sigma^{2}}{1-\theta^{2}}\right)
3\left(  \sigma^{2}\right)  ^{2}\\
&  =\frac{3\left(  \sigma^{2}\right)  ^{3}}{\left(  1-\theta^{2}\right)  ^{2}}%
\end{align*}
We also have%
\begin{align*}
\left(  y_{s-1}\varepsilon_{s}\right)  ^{2}y_{s+j+1}^{2}  &  =\left(
y_{s-1}\varepsilon_{s}\right)  ^{2}\left(  \theta y_{s+j}+\varepsilon
_{s+j+1}\right)  ^{2}\\
&  =\left(  y_{s-1}\varepsilon_{s}\right)  ^{2}\left(  \theta^{2}y_{s+j}%
^{2}+2\theta y_{s+j}\varepsilon_{s+j+1}+\varepsilon_{s+j+1}^{2}\right) \\
&  =\theta^{2}\left(  y_{s-1}\varepsilon_{s}\right)  ^{2}y_{s+j}^{2}%
+2\theta\left(  y_{s-1}\varepsilon_{s}\right)  ^{2}y_{s+j}\varepsilon
_{s+j+1}+\left(  y_{s-1}\varepsilon_{s}\right)  ^{2}\varepsilon_{s+j+1}^{2}%
\end{align*}
so%
\begin{align*}
E\left[  \left(  y_{s-1}\varepsilon_{s}\right)  ^{2}y_{s+j+1}^{2}\right]   &
=\theta^{2}E\left[  \left(  y_{s-1}\varepsilon_{s}\right)  ^{2}y_{s+j}%
^{2}\right]  +E\left[  \left(  y_{s-1}\varepsilon_{s}\right)  ^{2}\right]
\sigma^{2}\\
&  =\theta^{2}E\left[  \left(  y_{s-1}\varepsilon_{s}\right)  ^{2}y_{s+j}%
^{2}\right]  +\left(  \frac{\sigma^{2}}{1-\theta^{2}}\sigma^{2}\right)
\sigma^{2}\\
&  =\theta^{2}E\left[  \left(  y_{s-1}\varepsilon_{s}\right)  ^{2}y_{s+j}%
^{2}\right]  +\frac{\left(  \sigma^{2}\right)  ^{3}}{1-\theta^{2}}%
\end{align*}
so in general, we have%
\begin{align*}
E\left[  \left(  y_{s-1}\varepsilon_{s}\right)  ^{2}y_{s+j}^{2}\right]   &
=\frac{3\left(  \sigma^{2}\right)  ^{3}}{\left(  1-\theta^{2}\right)  ^{2}%
}\left(  \theta^{2}\right)  ^{j}+\frac{\left(  \sigma^{2}\right)  ^{3}%
}{1-\theta^{2}}\left(  \left(  \theta^{2}\right)  ^{j-1}+\cdots+1\right) \\
&  =\frac{3\left(  \sigma^{2}\right)  ^{3}}{\left(  1-\theta^{2}\right)  ^{2}%
}\left(  \theta^{2}\right)  ^{j}+\frac{\left(  \sigma^{2}\right)  ^{3}%
}{1-\theta^{2}}\frac{1-\left(  \theta^{2}\right)  ^{j}}{1-\theta^{2}}\\
&  =\frac{\left(  \sigma^{2}\right)  ^{3}}{\left(  1-\theta^{2}\right)  ^{2}%
}+\frac{2\left(  \sigma^{2}\right)  ^{3}}{\left(  1-\theta^{2}\right)  ^{2}%
}\left(  \theta^{2}\right)  ^{j}%
\end{align*}
from which the conclusion follows.
\end{proof}

\begin{lemma}
\label{lem-prop2-step4}%
\[
E\left[  \left(  y_{t-1}^{2}-\frac{\sigma^{2}}{1-\theta^{2}}\right)  \left(
y_{s-1}\varepsilon_{s}\right)  \left(  y_{s-1}\varepsilon_{s}\right)  \right]
=\frac{2\left(  \sigma^{2}\right)  ^{3}}{\left(  1-\theta^{2}\right)  ^{2}%
}\left(  \theta^{2}\right)  ^{t-s-1}%
\]

\end{lemma}

\begin{proof}%
\begin{align*}
&  E\left[  \left(  y_{t-1}^{2}-\frac{\sigma^{2}}{1-\theta^{2}}\right)
\left(  y_{s-1}\varepsilon_{s}\right)  \left(  y_{s-1}\varepsilon_{s}\right)
\right] \\
&  =E\left[  \left(  y_{s-1}\varepsilon_{s}\right)  ^{2}y_{t-1}^{2}\right]
-\frac{\sigma^{2}}{1-\theta^{2}}E\left[  \left(  y_{s-1}\varepsilon
_{s}\right)  ^{2}\right] \\
&  =\frac{\left(  \sigma^{2}\right)  ^{3}}{\left(  1-\theta^{2}\right)  ^{2}%
}+\frac{2\left(  \sigma^{2}\right)  ^{3}}{\left(  1-\theta^{2}\right)  ^{2}%
}\left(  \theta^{2}\right)  ^{t-s-1}-\frac{\sigma^{2}}{1-\theta^{2}}\left(
\frac{\sigma^{2}}{1-\theta^{2}}\sigma^{2}\right) \\
&  =\frac{2\left(  \sigma^{2}\right)  ^{3}}{\left(  1-\theta^{2}\right)  ^{2}%
}\left(  \theta^{2}\right)  ^{t-s-1}%
\end{align*}
where we used Lemma \ref{lem-prop2-step4-prep} for the second equality.
\end{proof}

\begin{lemma}
\label{lem-prop2-stepII-summary}%
\begin{align*}
&  \frac{1}{T}\sum_{s=1}^{T}\sum_{t>s}E\left(
\begin{array}
[c]{c}%
\left(  y_{s-1}^{2}-\frac{\sigma^{2}}{1-\theta^{2}}\right)  \left(
y_{t-1}\varepsilon_{t}\right)  \left(  y_{t-1}\varepsilon_{t}\right)
+2\left(  y_{t-1}^{2}-\frac{\sigma^{2}}{1-\theta^{2}}\right)  \left(
y_{s-1}\varepsilon_{s}\right)  \left(  y_{t-1}\varepsilon_{t}\right) \\
+\left(  y_{t-1}^{2}-\frac{\sigma^{2}}{1-\theta^{2}}\right)  \left(
y_{s-1}\varepsilon_{s}\right)  \left(  y_{s-1}\varepsilon_{s}\right)
+2\left(  y_{s-1}^{2}-\frac{\sigma^{2}}{1-\theta^{2}}\right)  \left(
y_{s-1}\varepsilon_{s}\right)  \left(  y_{t-1}\varepsilon_{t}\right)
\end{array}
\right) \\
&  =\frac{2\left(  \sigma^{2}\right)  ^{3}}{\left(  1-\theta^{2}\right)  ^{3}%
}\left(  \theta^{2}+1\right)  +O\left(  \frac{1}{T}\right)
\end{align*}

\end{lemma}

\begin{proof}
From Lemmas \ref{lem-prop2-step2}, \ref{lem-prop2-step3&5}, and
\ref{lem-prop2-step4}, we see the sum of the four expectations in the lemma is
equal to%
\begin{align*}
&  E\left[  \left(  y_{s-1}^{2}-\frac{\sigma^{2}}{1-\theta^{2}}\right)
\left(  y_{t-1}\varepsilon_{t}\right)  ^{2}\right]  +2E\left[  \left(
y_{s-1}\varepsilon_{s}\right)  \left(  y_{t-1}\varepsilon_{t}\right)  \left(
y_{t-1}^{2}-\frac{\sigma^{2}}{1-\theta^{2}}\right)  \right] \\
&  +E\left[  \left(  y_{t-1}^{2}-\frac{\sigma^{2}}{1-\theta^{2}}\right)
\left(  y_{s-1}\varepsilon_{s}\right)  ^{2}\right]  +2E\left[  \left(
y_{s-1}^{2}-\frac{\sigma^{2}}{1-\theta^{2}}\right)  \left(  y_{s-1}%
\varepsilon_{s}\right)  \left(  y_{t-1}\varepsilon_{t}\right)  \right] \\
&  =2\left(  \theta^{2}\right)  ^{t-s}\frac{\left(  \sigma^{2}\right)  ^{3}%
}{\left(  1-\theta^{2}\right)  ^{2}}+0+\frac{2\left(  \sigma^{2}\right)  ^{3}%
}{\left(  1-\theta^{2}\right)  ^{2}}\left(  \theta^{2}\right)  ^{t-s-1}+0\\
&  =\frac{2\left(  \sigma^{2}\right)  ^{3}}{\left(  1-\theta^{2}\right)  ^{2}%
}\left(  \theta^{2}+1\right)  \left(  \theta^{2}\right)  ^{t-s-1}%
\end{align*}
so the sum is equal to $\frac{2\left(  \sigma^{2}\right)  ^{3}}{\left(
1-\theta^{2}\right)  ^{2}}\left(  \theta^{2}+1\right)  $ times%
\begin{align*}
\frac{1}{T}\sum_{s=1}^{T}\sum_{t=s+1}^{T}\left(  \theta^{2}\right)  ^{t-s-1}
&  =\frac{1}{T}\sum_{s=1}^{T}\left(  1+\cdots+\left(  \theta^{2}\right)
^{T-s-1}\right) \\
&  =\frac{1}{T}\sum_{s=1}^{T}\frac{1-\left(  \theta^{2}\right)  ^{T-s}%
}{1-\theta^{2}}\\
&  =\frac{1}{T}\frac{1}{1-\theta^{2}}\sum_{s=1}^{T}\left(  1-\left(
\theta^{2}\right)  ^{T-s}\right) \\
&  =\frac{1}{T}\frac{1}{1-\theta^{2}}\left(  T-\frac{1-\left(  \theta
^{2}\right)  ^{T-s}}{1-\left(  \theta^{2}\right)  }\right) \\
&  =\frac{1}{1-\theta^{2}}+O\left(  \frac{1}{T}\right)
\end{align*}

\end{proof}

\begin{lemma}
\label{lem-prop2-step6&7&10&11}%
\[
E\left[
\begin{array}
[c]{c}%
\left(  y_{s-1}^{2}-\frac{\sigma^{2}}{1-\theta^{2}}\right)  \left(
y_{t-1}\varepsilon_{t}\right)  \left(  y_{u-1}\varepsilon_{u}\right)  +\left(
y_{s-1}^{2}-\frac{\sigma^{2}}{1-\theta^{2}}\right)  \left(  y_{u-1}%
\varepsilon_{u}\right)  \left(  y_{t-1}\varepsilon_{t}\right) \\
+\left(  y_{t-1}^{2}-\frac{\sigma^{2}}{1-\theta^{2}}\right)  \left(
y_{s-1}\varepsilon_{s}\right)  \left(  y_{u-1}\varepsilon_{u}\right)  +\left(
y_{t-1}^{2}-\frac{\sigma^{2}}{1-\theta^{2}}\right)  \left(  y_{u-1}%
\varepsilon_{u}\right)  \left(  y_{s-1}\varepsilon_{s}\right)
\end{array}
\right]  =0
\]

\end{lemma}

\begin{proof}
Law of iterated expectations
\end{proof}

\begin{lemma}
\label{lem-prop2-step8&9}%
\[
E\left[  \left(  y_{u-1}^{2}-\frac{\sigma^{2}}{1-\theta^{2}}\right)  \left(
y_{s-1}\varepsilon_{s}\right)  \left(  y_{t-1}\varepsilon_{t}\right)  \right]
=\theta^{2\left(  u-t-1\right)  }\frac{4\left(  \sigma^{2}\right)  ^{3}%
}{1-\theta^{2}}\theta^{2\left(  t-s\right)  }%
\]

\end{lemma}

\begin{proof}
We first note that%
\[
E\left[  \left(  y_{s-1}\varepsilon_{s}\right)  \left(  y_{t-1}\varepsilon
_{t}\right)  \right]  =0
\]
so it suffices to consider $E\left[  \left(  y_{s-1}\varepsilon_{s}\right)
\left(  y_{t-1}\varepsilon_{t}\right)  y_{u-1}^{2}\right]  $.

We start with $u=t+1$, i.e.,%
\begin{align*}
\left(  y_{s-1}\varepsilon_{s}\right)  \left(  y_{t-1}\varepsilon_{t}\right)
y_{t}^{2}  &  =\left(  y_{s-1}\varepsilon_{s}\right)  \left(  y_{t-1}%
\varepsilon_{t}\right)  \left(  \theta^{2}y_{t-1}^{2}+2\theta y_{t-1}%
\varepsilon_{t}+\varepsilon_{t}^{2}\right) \\
&  =\theta^{2}\left(  y_{s-1}\varepsilon_{s}\right)  y_{t-1}^{3}%
\varepsilon_{t}+2\theta\left(  y_{s-1}\varepsilon_{s}\right)  \left(
y_{t-1}\varepsilon_{t}\right)  ^{2}+\left(  y_{s-1}\varepsilon_{s}\right)
y_{t-1}\varepsilon_{t}^{3}%
\end{align*}
so%
\begin{align*}
E\left[  \left(  y_{s-1}\varepsilon_{s}\right)  \left(  y_{t-1}\varepsilon
_{t}\right)  y_{t}^{2}\right]   &  =2\theta E\left[  \left(  y_{s-1}%
\varepsilon_{s}\right)  y_{t-1}^{2}\varepsilon_{t}^{2}\right]  \\ &=2\theta
E\left[  \left(  y_{s-1}\varepsilon_{s}\right)  y_{t-1}^{2}\right]  \sigma
^{2}=2\theta\frac{2\theta^{2\left(  t-s\right)  -1}\sigma^{4}}{1-\theta^{2}%
}\sigma^{2}\\
&  =\frac{4\left(  \sigma^{2}\right)  ^{3}}{1-\theta^{2}}\theta^{2\left(
t-s\right)  }%
\end{align*}
where we used Lemma \ref{lem-prop1-step2} for the third equality. Repeating
\begin{align*}
E\left[  \left(  y_{s-1}\varepsilon_{s}\right)  \left(  y_{t-1}\varepsilon
_{t}\right)  y_{t+1}^{2}\right]   &  =E\left[  \left(  y_{s-1}\varepsilon
_{s}\right)  \left(  y_{t-1}\varepsilon_{t}\right)  \left(  \theta^{2}%
y_{t}^{2}+2\theta y_{t}\varepsilon_{t+1}+\varepsilon_{t+1}^{2}\right)  \right]
\\
&  =\theta^{2}E\left[  \left(  y_{s-1}\varepsilon_{s}\right)  \left(
y_{t-1}\varepsilon_{t}\right)  y_{t}^{2}\right]  +E\left[  \left(
y_{s-1}\varepsilon_{s}\right)  \left(  y_{t-1}\varepsilon_{t}\right)
\varepsilon_{t+1}^{2}\right] \\
&  =\theta^{2}E\left[  \left(  y_{s-1}\varepsilon_{s}\right)  \left(
y_{t-1}\varepsilon_{t}\right)  y_{t}^{2}\right]  ,
\end{align*}%
\begin{align*}
E\left[  \left(  y_{s-1}\varepsilon_{s}\right)  \left(  y_{t-1}\varepsilon
_{t}\right)  y_{t+2}^{2}\right]   &  =E\left[  \left(  y_{s-1}\varepsilon
_{s}\right)  \left(  y_{t-1}\varepsilon_{t}\right)  \left(  \theta^{2}%
y_{t+1}^{2}+2\theta y_{t+1}\varepsilon_{t+2}+\varepsilon_{t+2}^{2}\right)
\right] \\
&  =\theta^{2}E\left[  \left(  y_{s-1}\varepsilon_{s}\right)  \left(
y_{t-1}\varepsilon_{t}\right)  y_{t+1}^{2}\right]  +E\left[  \left(
y_{s-1}\varepsilon_{s}\right)  \left(  y_{t-1}\varepsilon_{t}\right)
\varepsilon_{t+2}^{2}\right] \\
&  =\theta^{2}E\left[  \left(  y_{s-1}\varepsilon_{s}\right)  \left(
y_{t-1}\varepsilon_{t}\right)  y_{t+1}^{2}\right]  ,
\end{align*}
we obtain
\[
E\left[  \left(  y_{s-1}\varepsilon_{s}\right)  \left(  y_{t-1}\varepsilon
_{t}\right)  y_{t+j}^{2}\right]  =\theta^{2j}\frac{4\left(  \sigma^{2}\right)
^{3}}{1-\theta^{2}}\theta^{2\left(  t-s\right)  }%
\]
or%
\[
E\left[  \left(  y_{s-1}\varepsilon_{s}\right)  \left(  y_{t-1}\varepsilon
_{t}\right)  y_{u-1}^{2}\right]  =\theta^{2\left(  u-t-1\right)  }%
\frac{4\left(  \sigma^{2}\right)  ^{3}}{1-\theta^{2}}\theta^{2\left(
t-s\right)  }%
\]

\end{proof}

\begin{lemma}
\label{lem-prop2-stepIII-summary}%
\begin{align*}
&  \frac{1}{T}\sum_{s<t<u}E\left(
\begin{array}
[c]{c}%
\left(  y_{s-1}^{2}-\frac{\sigma^{2}}{1-\theta^{2}}\right)  \left(
y_{t-1}\varepsilon_{t}\right)  \left(  y_{u-1}\varepsilon_{u}\right)  +\left(
y_{s-1}^{2}-\frac{\sigma^{2}}{1-\theta^{2}}\right)  \left(  y_{u-1}%
\varepsilon_{u}\right)  \left(  y_{t-1}\varepsilon_{t}\right) \\
+\left(  y_{u-1}^{2}-\frac{\sigma^{2}}{1-\theta^{2}}\right)  \left(
y_{s-1}\varepsilon_{s}\right)  \left(  y_{t-1}\varepsilon_{t}\right)  +\left(
y_{u-1}^{2}-\frac{\sigma^{2}}{1-\theta^{2}}\right)  \left(  y_{t-1}%
\varepsilon_{t}\right)  \left(  y_{s-1}\varepsilon_{s}\right) \\
+\left(  y_{t-1}^{2}-\frac{\sigma^{2}}{1-\theta^{2}}\right)  \left(
y_{s-1}\varepsilon_{s}\right)  \left(  y_{u-1}\varepsilon_{u}\right)  +\left(
y_{t-1}^{2}-\frac{\sigma^{2}}{1-\theta^{2}}\right)  \left(  y_{u-1}%
\varepsilon_{u}\right)  \left(  y_{s-1}\varepsilon_{s}\right)
\end{array}
\right) \\
&  =\frac{8\theta^{2}\left(  \sigma^{2}\right)  ^{3}}{\left(  1-\theta
^{2}\right)  ^{3}}+O\left(  \frac{1}{T}\right)
\end{align*}

\end{lemma}

\begin{proof}
From Lemmas \ref{lem-prop2-step6&7&10&11} and \ref{lem-prop2-step8&9}, we see
that the expectation in the lemma above is equal to%
\[
E\left[  2\left(  y_{u-1}^{2}-\frac{\sigma^{2}}{1-\theta^{2}}\right)  \left(
y_{s-1}\varepsilon_{s}\right)  \left(  y_{t-1}\varepsilon_{t}\right)  \right]
=\frac{8\left(  \sigma^{2}\right)  ^{3}}{1-\theta^{2}}\theta^{2\left(
u-t-1\right)  }\theta^{2\left(  t-s\right)  -1}%
\]
so the sum is equal to%
\begin{align*}
\frac{1}{T} &\sum_{s}\sum_{t>s}\sum_{u>t}\frac{8\left(  \sigma^{2}\right)  ^{3}%
}{1-\theta^{2}}\theta^{2\left(  u-t-1\right)  }\theta^{2\left(  t-s\right)  } \\
&  =\frac{1}{T}\sum_{u=3}^{T}\sum_{t=2}^{u-1}\sum_{s=1}^{t-1}\frac{8\left(
\sigma^{2}\right)  ^{3}}{1-\theta^{2}}\theta^{2\left(  u-t-1\right)  }%
\theta^{2\left(  t-s\right)  }\\
&  =\frac{1}{T}\sum_{u=3}^{T}\sum_{t=2}^{u-1}\left(  \frac{8\left(  \sigma
^{2}\right)  ^{3}}{1-\theta^{2}}\theta^{2\left(  u-t-1\right)  }\sum
_{s=1}^{t-1}\theta^{2\left(  t-s\right)  }\right) \\
&  =\frac{1}{T}\sum_{u=3}^{T}\sum_{t=2}^{u-1}\frac{8\left(  \sigma^{2}\right)
^{3}}{1-\theta^{2}}\theta^{2\left(  u-t-1\right)  }\left(  \theta^{2}%
+\cdots+\theta^{2\left(  t-1\right)  }\right) \\
&  =\frac{1}{T}\sum_{u=3}^{T}\sum_{t=2}^{u-1}\frac{8\left(  \sigma^{2}\right)
^{3}}{1-\theta^{2}}\theta^{2\left(  u-t-1\right)  }\theta^{2}\frac{1-\left(
\theta^{2}\right)  ^{t-1}}{1-\theta^{2}}\\
&  =\frac{1}{T}\sum_{u=3}^{T}\sum_{t=2}^{u-1}\frac{8\left(  \sigma^{2}\right)
^{3}}{\left(  1-\theta^{2}\right)  ^{2}}\left(  \theta^{2}\right)
^{u-t}\left(  1-\left(  \theta^{2}\right)  ^{t-1}\right) \\
&  =\frac{1}{T}\sum_{u=3}^{T}\left(  \frac{8\left(  \sigma^{2}\right)  ^{3}%
}{\left(  1-\theta^{2}\right)  ^{2}}\sum_{t=2}^{u-1}\left(  \left(  \theta
^{2}\right)  ^{u-t}-\left(  \theta^{2}\right)  ^{u-1}\right)  \right) \\
&  =\frac{1}{T}\sum_{u=3}^{T}\frac{8\left(  \sigma^{2}\right)  ^{3}}{\left(
1-\theta^{2}\right)  ^{2}}\left(  \left(  \theta^{2}+\cdots+\left(  \theta
^{2}\right)  ^{u-2}\right)  -\left(  u-2\right)  \left(  \theta^{2}\right)
^{u-1}\right) \\
&  =\frac{1}{T}\frac{8\left(  \sigma^{2}\right)  ^{3}}{\left(  1-\theta
^{2}\right)  ^{2}}\sum_{u=3}^{T}\left(  \frac{\theta^{2}\left(  1-\left(
\theta^{2}\right)  ^{u-2}\right)  }{1-\theta^{2}}-\left(  u-2\right)  \left(
\theta^{2}\right)  ^{u-1}\right)
\end{align*}

Because%
\[
\sum_{u=3}^{T}\left(  u-2\right)  \left(  \theta^{2}\right)  ^{u-1}\equiv
A=\Theta^{2}+2\Theta^{3}+3\Theta^{4}+\cdots+\left(  T-2\right)  \Theta^{T-1}%
\]
for $\Theta=\theta^{2}$, we have%
\[
\Theta A=\Theta^{3}+2\Theta^{4}+3\Theta^{5}+\cdots+\left(  T-2\right)
\Theta^{T}%
\]
and%
\begin{align*}
\left(  1-\Theta\right)  A  &  =\Theta^{2}+\Theta^{3}+\Theta^{4}+\cdots
+\Theta^{T-1}-\left(  T-2\right)  \Theta^{T}\\
&  =\frac{\Theta^{2}\left(  1-\Theta^{T-2}\right)  }{1-\Theta}-\left(
T-2\right)  \Theta^{T}\\
&  =\frac{\theta^{4}\left(  1-\left(  \theta^{2}\right)  ^{T-2}\right)
}{1-\theta^{2}}-\left(  T-2\right)  \left(  \theta^{2}\right)  ^{T}%
\end{align*}
and therefore,%
\[
\sum_{u=3}^{T}\left(  u-2\right)  \left(  \theta^{2}\right)  ^{u-1}%
=\frac{\theta^{4}\left(  1-\left(  \theta^{2}\right)  ^{T-2}\right)
}{1-\theta^{2}}-\left(  T-2\right)  \left(  \theta^{2}\right)  ^{T}%
\]

It follows that%
\begin{align*}
&  \frac{1}{T}\sum_{s}\sum_{t>s}\sum_{u>t}\frac{8\left(  \sigma^{2}\right)
^{3}}{1-\theta^{2}}\theta^{2\left(  u-t-1\right)  }\theta^{2\left(
t-s\right)  -1}\\
&  =\frac{1}{T}\frac{8\left(  \sigma^{2}\right)  ^{3}}{\left(  1-\theta
^{2}\right)  ^{2}}\sum_{u=3}^{T}\left(  \frac{\theta^{2}\left(  1-\left(
\theta^{2}\right)  ^{u-2}\right)  }{1-\theta^{2}}-\left(  u-2\right)  \left(
\theta^{2}\right)  ^{u-1}\right) \\
&  =\frac{1}{T}\frac{8\left(  \sigma^{2}\right)  ^{3}}{\left(  1-\theta
^{2}\right)  ^{2}}\frac{\theta^{2}}{1-\theta^{2}}\sum_{u=3}^{T}\left(
1-\left(  \theta^{2}\right)  ^{u-2}\right)  -\frac{1}{T}\frac{8\left(
\sigma^{2}\right)  ^{3}}{\left(  1-\theta^{2}\right)  ^{2}}\sum_{u=3}%
^{T}\left(  u-2\right)  \left(  \theta^{2}\right)  ^{u-1}\\
&  =\frac{1}{T}\frac{8\left(  \sigma^{2}\right)  ^{3}}{\left(  1-\theta
^{2}\right)  ^{2}}\frac{\theta^{2}}{1-\theta^{2}}\left(  T-2-\frac{\theta
^{2}\left(  1-\left(  \theta^{2}\right)  ^{T-2}\right)  }{1-\theta^{2}}\right)
\\
&  -\frac{1}{T}\frac{8\left(  \sigma^{2}\right)  ^{3}}{\left(  1-\theta
^{2}\right)  ^{2}}\left(  \frac{\theta^{4}\left(  1-\left(  \theta^{2}\right)
^{T-2}\right)  }{1-\theta^{2}}-\left(  T-2\right)  \left(  \theta^{2}\right)
^{T}\right) \\
&  =\frac{8\theta^{2}\left(  \sigma^{2}\right)  ^{3}}{\left(  1-\theta
^{2}\right)  ^{3}}+O\left(  \frac{1}{T}\right)
\end{align*}

\end{proof}

\begin{lemma}
\label{lem-var-Y2}%
\[
E\left[  \left(  \frac{1}{\sqrt{T}}\sum_{t=1}^{T}\left(  \frac{y_{t-1}^{2}%
}{\sigma^{2}}-\frac{1}{1-\theta^{2}}\right)  \right)  ^{2}\right]
=\frac{2\left(  1+\theta^{2}\right)  }{\left(  1-\theta^{2}\right)  ^{3}%
}+O\left(  T^{-1}\right)
\]

\end{lemma}

\begin{proof}
We have%
\begin{align*}
\left(  \frac{1}{\sqrt{T}}\sum_{t=1}^{T}\left(  \frac{y_{t-1}^{2}}{\sigma^{2}%
}-\frac{1}{1-\theta^{2}}\right)  \right)  ^{2}  &  =\frac{1}{T}\sum_{t=1}%
^{T}\left(  \frac{y_{t-1}^{2}}{\sigma^{2}}-\frac{1}{1-\theta^{2}}\right)
^{2}\\
&  +2\sum_{s=1}^{T-1}\sum_{t=s+1}^{T}\left(  \frac{y_{s-1}^{2}}{\sigma^{2}%
}-\frac{1}{1-\theta^{2}}\right)  \left(  \frac{y_{t-1}^{2}}{\sigma^{2}}%
-\frac{1}{1-\theta^{2}}\right)
\end{align*}
Note that%
\begin{align*}
E\left[  \left(  \frac{y_{t-1}^{2}}{\sigma^{2}}-\frac{1}{1-\theta^{2}}\right)
^{2}\right]   &  =\frac{1}{\left(  \sigma^{2}\right)  ^{2}}E\left[  \left(
y_{t-1}^{2}-E\left[  y_{t-1}^{2}\right]  \right)  ^{2}\right] \\
&  =\frac{1}{\left(  \sigma^{2}\right)  ^{2}}\left(  E\left[  y_{t-1}%
^{4}\right]  -\left(  E\left[  y_{t-1}^{2}\right]  \right)  ^{2}\right) \\
&  =\frac{1}{\left(  \sigma^{2}\right)  ^{2}}\left(  3\left(  \frac{\sigma
^{2}}{1-\theta^{2}}\right)  ^{2}-\left(  \frac{\sigma^{2}}{1-\theta^{2}%
}\right)  ^{2}\right) \\
&  =\frac{2}{\left(  1-\theta^{2}\right)  ^{2}}%
\end{align*}
Also note that%
\begin{align*}
E&\left[  \left(  \frac{y_{s-1}^{2}}{\sigma^{2}}-\frac{1}{1-\theta^{2}}\right)
\left(  \frac{y_{t-1}^{2}}{\sigma^{2}}-\frac{1}{1-\theta^{2}}\right)  \right] \\
&  =\frac{1}{\left(  \sigma^{2}\right)  ^{2}}\left(  E\left[  y_{s-1}%
^{2}y_{t-1}^{2}\right]  -\left(  \frac{\sigma^{2}}{1-\theta^{2}}\right)
^{2}\right) \\
&  =\frac{1}{\left(  \sigma^{2}\right)  ^{2}}\left(  \left(  1+2\left(
\theta^{2}\right)  ^{t-s}\right)  \left(  \frac{\sigma^{2}}{1-\theta^{2}%
}\right)  ^{2}-\left(  \frac{\sigma^{2}}{1-\theta^{2}}\right)  ^{2}\right) \\
&  =\frac{2\left(  \theta^{2}\right)  ^{t-s}}{\left(  1-\theta^{2}\right)
^{2}}%
\end{align*}
using Lemma \ref{Eys2yt2} for the second equality. So%
\begin{align*}
\sum_{t=s+1}^{T}E\left[  \left(  \frac{y_{s-1}^{2}}{\sigma^{2}}-\frac
{1}{1-\theta^{2}}\right)  \left(  \frac{y_{t-1}^{2}}{\sigma^{2}}-\frac
{1}{1-\theta^{2}}\right)  \right]   &  =\frac{2}{\left(  1-\theta^{2}\right)
^{2}}\sum_{t=s+1}^{T}\left(  \theta^{2}\right)  ^{t-s}\\
&  =\frac{2}{\left(  1-\theta^{2}\right)  ^{2}}\left(  \theta^{2}%
+\cdots+\left(  \theta^{2}\right)  ^{T-s}\right) \\
&  =\frac{2\theta^{2}}{\left(  1-\theta^{2}\right)  ^{3}}\left(  1-\left(
\theta^{2}\right)  ^{T-s}\right)
\end{align*}
implying that%
\begin{align*}
&  \sum_{s=1}^{T-1}\sum_{t=s+1}^{T}E\left[  \left(  \frac{y_{s-1}^{2}}%
{\sigma^{2}}-\frac{1}{1-\theta^{2}}\right)  \left(  \frac{y_{t-1}^{2}}%
{\sigma^{2}}-\frac{1}{1-\theta^{2}}\right)  \right] \\
&  =\frac{2\theta^{2}}{\left(  1-\theta^{2}\right)  ^{3}}\sum_{s=1}%
^{T-1}\left(  1-\left(  \theta^{2}\right)  ^{T-s}\right) \\
&  =\frac{2\theta^{2}}{\left(  1-\theta^{2}\right)  ^{3}}\left(  T-2-\left(
\theta^{2}+\cdots+\left(  \theta^{2}\right)  ^{T-1}\right)  \right) \\
&  =\frac{2\theta^{2}}{\left(  1-\theta^{2}\right)  ^{3}}\left(
T-2-\frac{\theta^{2}\left(  1-\left(  \theta^{2}\right)  ^{T-1}\right)
}{1-\theta^{2}}\right)
\end{align*}
To conclude,%
\begin{align*}
&  E\left[  \left(  \frac{1}{\sqrt{T}}\sum_{t=1}^{T}\left(  \frac{y_{t-1}^{2}%
}{\sigma^{2}}-\frac{1}{1-\theta^{2}}\right)  \right)  ^{2}\right] \\
&  =\frac{1}{T}\sum_{t=1}^{T}\frac{2}{\left(  1-\theta^{2}\right)  ^{2}}\\
&  +\frac{2}{T}\frac{2\theta^{2}}{\left(  1-\theta^{2}\right)  ^{3}}\left(
T-2-\frac{\theta^{2}\left(  1-\left(  \theta^{2}\right)  ^{T-1}\right)
}{1-\theta^{2}}\right) \\
&  =\frac{2}{\left(  1-\theta^{2}\right)  ^{2}}+\frac{4\theta^{2}}{\left(
1-\theta^{2}\right)  ^{3}}+O\left(  T^{-1}\right)
\end{align*}

\end{proof}

\section{Details for the Proof of Lemma \ref{var-intuition}%
\label{sec-var-intuition}}%

\begin{lemma}
\label{Chanda}For the random sequence $\left\{  \varepsilon_{t}\right\}
_{t=-\infty}^{\infty}$ defined in (\ref{AR1}) let $\xi_{t}=\sum_{j=0}^{\infty
}\theta^{j}\varepsilon_{t-j}$ for $t=0,\pm1,\pm2,...$ Define $\mathcal{M}%
_{a}^{b}=\sigma\left(  \xi_{a},...\xi_{b}\right)  .$ Then, for all
$t\in\left\{  1,...,T\right\}  $, $y_{t}=\xi_{t},$ and for $\lambda\in\left(
0,1\right)  $ and for all sets $A,B$ such that $A\in\mathcal{M}_{-\infty}%
^{t},$\textbf{ }$B\in\mathcal{M}_{t+k}^{\infty}$ for $t=0,\pm1,\pm
2,...$\textbf{ }it follows that
\begin{equation}
\left\vert P\left(  AB\right)  -P\left(  A\right)  P\left(  B\right)
\right\vert \leq M\alpha\left(  k\right)  \label{alpha-mix}%
\end{equation}
\textbf{ }where $M$ is a positive and finite constant and $\alpha\left(
k\right)  =\sum_{j=k}^{\infty}\left\vert j\right\vert \left\vert \theta
^{j}\right\vert ^{\lambda}$.
\end{lemma}

\begin{proof}
First note that by (\ref{AR1}) $y_{t}=\sum_{j=0}^{t-1}\theta^{j}%
\varepsilon_{t-j}+\theta^{t}y_{0}.$ Since $y_{0}=\sum_{l=0}^{\infty}\theta
^{l}\varepsilon_{-l}$ it follows that $y_{t}=\sum_{j=0}^{t-1}\theta
^{j}\varepsilon_{t-j}+\sum_{l=0}^{\infty}\theta^{t+l}\varepsilon_{-l}.$ By a
change of variables $u=l+t$ it follows that $\sum_{l=0}^{\infty}\theta
^{t+l}\varepsilon_{-l}=\sum_{u=t}^{\infty}\theta^{u}\varepsilon_{t-u}$ such
that $\xi_{t}=y_{t}$ as claimed. Next note that $\xi_{t}$ is striclty
stationary. This means that $P\left(  AB\right)  $ and $P\left(  A\right)
P\left(  B\right)  $ are invariant to translations $t$ of $\mathcal{M}%
_{t+a}^{t+b}=\sigma\left(  \xi_{t+a},...\xi_{t+b}\right)  $. Therefore choose
$t=0$ in $\mathcal{M}_{-\infty}^{t}$ and $\mathcal{M}_{t+k}^{\infty}.$ To
apply Chanda (1974, Theorem 2.1), note that $E\left[  \left\vert
\varepsilon_{t}\right\vert ^{\delta}\right]  <\infty$ for any $\delta>0$ and
$\delta$ bounded. By the discussion of Chanda (1974, p.402) it follows that
the Gaussian distribution satisfies the integrability condition for the
characteristc function. Finally, for any $\lambda>0,$ we have $\sum
_{j=0}^{\infty}j\left(  \theta^{j}\right)  ^{\lambda}<\infty$ for $\theta
\in\left(  -1,1\right)  $. Then, by Chanda (1974), Theorem 2.1 it follows that
(\ref{alpha-mix}) holds.
\end{proof}

The next result establishes summability conditions for covariances and higher
oder cumulants related to our stochastic expansions. For this define
\begin{equation}
Y_{1,t}=\left(  1-\theta^{2}\right)  \frac{y_{t-1}\varepsilon_{t}}{\sigma^{2}}
\label{Def_Y1t}%
\end{equation}
and%
\begin{equation}
Y_{2,t}=-\left(  1-\theta^{2}\right)  \left(  \frac{y_{t-1}^{2}}{\sigma^{2}%
}-\frac{1}{1-\theta^{2}}\right)  \label{Def_Y2t}%
\end{equation}
Let $\gamma_{a,b}\left(  t_{1},t_{2}\right)  =\operatorname*{Cov}\left(
Y_{a,t_{1}},Y_{b,t_{2}}\right)  $ for $a,b\in\left\{  1,2\right\}  .$ Let
$\gamma_{a_{1}a_{2},a_{3}}\left(  t_{1},t_{2},t_{3}\right)
=\operatorname*{Cov}\left(  Y_{a_{1},t_{1}}Y_{a_{2},t_{2}},Y_{a_{3},t_{3}%
}\right)  $ be a third order moment of $Y_{a_{1},t_{1}}Y_{a_{2},t_{2}}%
Y_{a_{3},t_{3}}$ for $a_{j}\in\left\{  1,2\right\}  $ and $j=1,...,3.$ Let
$cum_{a_{1},a_{2},a_{3},a_{4}}\left(  t_{1},t_{2},t_{3},t_{4}\right)  $ be the
fourth order cumulant of $Y_{a_{1},t_{1}}Y_{a_{2},t_{2}}Y_{a_{3},t_{3}%
}Y_{a_{4},t_{4}}$ for $a_{j}\in\left\{  1,2\right\}  $ and $j=1,...,4.$ Then,
the following summability conditions hold.

\begin{lemma}
\label{cumsum}Assume that $y_{t}$ satisfies (\ref{AR1}) and that $Y_{1,t}$ and
$Y_{2,t}$ are as defined in (\ref{Def_Y1t}) and (\ref{Def_Y2t}). Then,
\newline i) $\sum_{u=-\infty}^{\infty}\left\vert u\right\vert \sup
_{t}\left\vert \gamma_{a,b}\left(  t,t+u\right)  \right\vert \leq K_{1}%
<\infty.$\newline ii) $\sum_{u=-\infty}^{\infty}\left\vert u\right\vert
\sup_{t}\left\vert \gamma_{a_{1}a_{2},a_{3}}\left(  t,t,t+u\right)
\right\vert \leq K_{2}<\infty,$\newline iii) $\sum_{u=1}^{\infty}u\sup
_{t\leq0}\left\vert \gamma_{a_{1}a_{2},a_{3}}\left(  t,0,u\right)  \right\vert
\leq K_{2}<\infty,$\newline iv) $\sum_{u=1}^{\infty}u\sup_{t\leq0}\left\vert
\gamma_{a_{1}a_{2},a_{3}}\left(  0,t,u\right)  \right\vert \leq K_{2}<\infty
,$\newline v) $\sum_{u=1}^{\infty}u\sup_{t>0}\left\vert \gamma_{a_{1}%
a_{2},a_{3}}\left(  u,t+u,0\right)  \right\vert \leq K_{2}<\infty,$\newline
vi) $\sum_{t_{2},t_{3},t_{4}=-\infty}^{\infty}\sup_{t_{1}}\left\vert
cum_{a_{1},a_{2},a_{3},a_{4}}\left(  t_{1},t_{2},t_{3},t_{4}\right)
\right\vert \leq K_{3}<\infty$\newline vii) $\sum_{T>\left\vert v\right\vert
\geq T/2}\left\vert \gamma_{a,b}\left(  0,v\right)  \right\vert =o\left(
T^{-1}\right)  .$
\end{lemma}

\begin{proof}
For the first result note that $Y_{1,t}$ and $Y_{2,t}$ are both strictly
stationary. Thus, $\gamma_{a,b}\left(  t_{1},t_{2}\right)  =\gamma
_{a,b}\left(  ,t_{2}-t_{1}\right)  $ and thus $\sup_{t}\left\vert \gamma
_{a,b}\left(  t,t+u\right)  \right\vert =\left\vert \gamma_{a,b}\left(
0,u\right)  \right\vert .$ Because $\varepsilon_{t}=y_{t}-\theta y_{t-1}$ for
$t\in\left\{  1,...,T\right\}  $ it follows that $Y_{1,t}$ is measurable
w.r.t. $\mathcal{M}_{-\infty}^{t}$ and $\mathcal{M}_{t+1}^{\infty}$ and
$Y_{t,2}$ is measurable w.r.t $\mathcal{M}_{-\infty}^{t}$ and $\mathcal{M}%
_{t}^{\infty}$. By Doukhan (1994, Theorem 3, p.9), it follows that
\begin{equation}
\left\vert \gamma_{a,b}\left(  t_{1},t_{2}\right)  \right\vert \leq
8\alpha\left(  t_{1}-t_{2}-1\right)  ^{1/2}E\left[  \left\vert Y_{a,t_{1}%
}\right\vert ^{4}\right]  ^{1/4}E\left[  \left\vert Y_{b,t_{2}}\right\vert
^{4}\right]  ^{1/4}, \label{Cov_mix_bound}%
\end{equation}
and where the moments on the RHS are bounded by the fact that $y_{t}$ is
jointly Gaussian. Then,
\begin{align*}
\sum_{u=-\infty}^{\infty}\left\vert u\right\vert \sup_{t}\left\vert
\gamma_{a,b}\left(  t,t+u\right)  \right\vert  &  =\sum_{u=-\infty}^{\infty
}\left\vert u\right\vert \left\vert \gamma_{a,b}\left(  0,u\right)
\right\vert \\
&  \leq8E\left[  \left\vert Y_{a,t_{1}}\right\vert ^{4}\right]  ^{1/4}E\left[
\left\vert Y_{b,t_{2}}\right\vert ^{4}\right]  ^{1/4}\sum_{u=-\infty}^{\infty
}\left\vert u\right\vert \alpha\left(  u-1\right)  ^{1/2}.
\end{align*}
Note that
\[
\sum_{u=1}^{\infty}u\alpha\left(  u-1\right)  ^{1/2}\leq\left(  \sum
_{u=1}^{\infty}u^{-2}\right)  ^{1/2}\left(  \sum_{u=1}^{\infty}u^{4}%
\alpha\left(  u-1\right)  \right)  ^{1/2}=\frac{\pi}{\sqrt{6}}\left(
\sum_{u=1}^{\infty}u^{4}\alpha\left(  u-1\right)  \right)  ^{1/2}%
\]
by the Cauchy-Schwartz inequality. In addition
\[
\sum_{u=1}^{\infty}u^{4}\alpha\left(  u-1\right)  =\sum_{v=0}^{\infty}\left(
v+1\right)  ^{4}\alpha\left(  v\right)  \leq16\sum_{v=1}^{\infty}v^{4}%
\alpha\left(  v\right)  +\alpha\left(  0\right)  <\infty,
\]
where the fact that the expression is bounded holds for $\theta\in\left(
-1,1\right)  $ can be checked by direct calculation using the result in Lemma
\ref{Chanda}. It follows that $\sum_{u=-\infty}^{\infty}\left\vert
u\right\vert \sup_{t}\left\vert \gamma_{a,b}\left(  t,t+u\right)  \right\vert
<\infty$.

For (ii) consider $\operatorname*{Cov}\left(  Y_{a_{1},t_{1}}Y_{a_{2}t_{1}%
},Y_{a_{3},t_{3}}\right)  =\gamma_{11,2}\left(  t_{1},t_{1},t_{3}\right)  $
where $Y_{a_{1},t_{1}}Y_{a_{2}t_{1}}$ is measurable $\mathcal{M}_{-\infty
}^{t_{1}}$ and $Y_{a_{3},t_{3}}$ is measurable $\mathcal{M}_{t_{3}-1}^{\infty
}.$ Jensen's inequality yields $E\left[  \left\vert Y_{a,t_{1}}\right\vert
^{4}\right]  ^{1/4}\leq E\left[  \left\vert Y_{a,t_{1}}\right\vert
^{8}\right]  ^{1/8}\leq K^{1/8}$ for a $K$ such that $E\left[  \left\vert
Y_{a,t}\right\vert ^{8}\right]  \leq K<\infty$ for $a\in\left\{  1,2\right\}
$ and all $t$. Note that such a $K$ exists because $y_{t}$ is Gaussian and
strictly stationary. Then,
\begin{align*}
E\left[  \left\vert Y_{a_{1},t_{1}}Y_{a_{2}t_{1}}\right\vert ^{4}\right]
^{1/4}E\left[  \left\vert Y_{a_{3},t_{3}}\right\vert ^{4}\right]  ^{1/4}  &
\leq E\left[  \left\vert Y_{a_{1},t_{1}}\right\vert ^{8}\right]
^{1/8}E\left[  \left\vert Y_{a_{2},t_{1}}\right\vert ^{8}\right]
^{1/8}E\left[  \left\vert Y_{a_{3},t_{3}}\right\vert ^{8}\right]  ^{1/8}\\
&  \leq K^{3/8}.
\end{align*}
Then, it follows from Doukhan (1994, Theorem 3, p.9) and strict stationarity
that
\[
\left\vert \gamma_{a_{1}a_{2},a_{3}}\left(  t_{1},t_{1},t_{3}\right)
\right\vert =\left\vert \gamma_{a_{1}a_{2},a_{3}}\left(  0,0,t_{3}%
-t_{1}\right)  \right\vert \leq8K^{3/8}\alpha\left(  t_{3}-t_{1}-1\right)  .
\]
Then, using strict stationarity again such that $\sup_{t}\left\vert
\gamma_{a_{1}a_{2},a_{3}}\left(  t,t,t+u\right)  \right\vert =\gamma
_{a_{1}a_{2},a_{3}}\left(  0,0,u\right)  $%
\[
\sum_{u=-\infty}^{\infty}\left\vert u\right\vert \sup_{t}\left\vert
\gamma_{a_{1}a_{2},a_{3}}\left(  t,t,t+u\right)  \right\vert \leq8K^{3/8}%
\sum_{u=-\infty}^{\infty}\left\vert u\right\vert \alpha\left(  u-1\right)
\leq K_{2}<\infty,
\]
where the summability of $\alpha\left(  u\right)  $ follows from Lemma
\ref{Chanda}.

For (iii) consider $\operatorname*{Cov}\left(  Y_{a_{1},t}Y_{a_{2},0}%
,Y_{a_{3},u}\right)  =\gamma_{a_{1}a_{2},a_{3}}\left(  t,0,u\right)  $ with
$t\leq0,$ $u>0.$ Then, $Y_{a_{1},t}Y_{a_{2},0}$ is measurable $\mathcal{M}%
_{-\infty}^{0}$ and $Y_{a_{3},u}$ is measurable $\mathcal{M}_{u-1}^{\infty}.$
It then follows from Doukhan (1994, Theorem 3, p.9) that
\[
\left\vert \gamma_{a_{1}a_{2},a_{3}}\left(  t,0,u\right)  \right\vert
\leq8K^{3/8}\alpha\left(  u-1\right)
\]
such that
\[
\sum_{u=1}^{\infty}u\sup_{t\leq0}\left\vert \gamma_{a_{1}a_{2},a_{3}}\left(
t,0,u\right)  \right\vert \leq8K^{3/8}\sum_{u=1}^{\infty}u\alpha\left(
u-1\right)  \leq K_{2}<\infty,
\]
where the upper bound for the sum over $u\alpha\left(  u-1\right)  $ follows
from Lemma \ref{Chanda}.

For (iv) consider $\operatorname*{Cov}\left(  Y_{a_{1},0}Y_{a_{2},t}%
,Y_{a_{3},u}\right)  =\gamma_{a_{1}a_{2},a_{3}}\left(  0,t,u\right)  $ for
$t\leq0,$ $u>0.$ Then the proof proceeds in exactly the same way as for (iii).

For (v) consider $\operatorname*{Cov}\left(  Y_{a_{1},u}Y_{a_{2},t+u}%
,Y_{a_{3},0}\right)  =\gamma_{a_{1}a_{2},a_{3}}\left(  u,t+u,0\right)  $ for
$t>0,$ $u>0$. Then, $Y_{a_{1},u}Y_{a_{2},t+u}$ is measurable w.r.t.
$\mathcal{M}_{u-1}^{\infty}$ and $Y_{a_{3},0}$ is measurable w.r.t.
$\mathcal{M}_{-\infty}^{0}$. Then, it follows from Doukhan (1994, Theorem 3,
p.9) that
\[
\left\vert \gamma_{a_{1}a_{2},a_{3}}\left(  u,t+u,0\right)  \right\vert
\leq8K^{3/8}\alpha\left(  u-1\right)
\]
and
\[
\sum_{u=1}^{\infty}\sup_{t>0}\left\vert \gamma_{a_{1}a_{2},a_{3}}\left(
u,t+u,0\right)  \right\vert \leq8K^{3/8}\sum_{u=1}^{\infty}u\alpha\left(
u-1\right)  \leq K_{2}<\infty,
\]
where the sum over $u\alpha\left(  u-1\right)  $ is finite by Lemma
\ref{Chanda}.

For the result in (vi) verify that the conditions of Lemma 1 of Andrews (1991)
hold. For some $\nu>1$ we have
\[
E\left[  \left\vert Y_{1,t}\right\vert ^{8\nu}\right]  =\sigma^{-16\nu
}E\left[  \left\vert y_{t-1}\varepsilon_{t}\right\vert ^{8\nu}\right]
\leq\sigma^{-16\nu}E\left[  \left\vert y_{t-1}\right\vert ^{16\nu}\right]
^{1/2}E\left[  \left\vert \varepsilon_{t}\right\vert ^{16\nu}\right]
^{1/2}<\infty
\]
where the inequality is the Cauchy-Schwartz inequality and the moments are
bounded because both $y_{t-1}$ and $\varepsilon_{t}$ are Gaussian. Also,
\[
E\left[  \left\vert Y_{2,t}\right\vert ^{8\nu}\right]  =\frac{1}{\left(
1-\theta^{2}\right)  ^{8\nu}}E\left[  \left\vert y_{t-1}^{2}\left(
1-\theta^{2}\right)  \sigma^{-2}-1\right\vert ^{8\nu}\right]  .
\]
Letting $\xi=y_{t-1}^{2}\left(  1-\theta^{2}\right)  \sigma^{-2}$ and noting
that $\xi$ is $\chi_{1}^{2}$ distributed it follows that
\[
E\left[  \left\vert Y_{2,t}\right\vert ^{8\nu}\right]  \leq\frac{1}{\left(
1-\theta^{2}\right)  ^{8\nu}}E\left[  \left(  \xi+1\right)  ^{8\nu}\right]
\leq\frac{2^{8\nu-1}}{\left(  1-\theta^{2}\right)  ^{8\nu}}\left(  E\left[
\xi^{8\nu}\right]  ^{1/8\nu}+1\right)  ^{8\nu}<\infty
\]
by the Minkowski inequality and where the bound holds by the fact that the
finite moments of the $\chi_{1}^{2}$ distribution are bounded. The second
condition of Lemma 1 in Andrews is that $\sum_{j=1}^{\infty}j^{2}\alpha\left(
j-1\right)  ^{\left(  \nu-1\right)  /\nu}<\infty\ $which holds after
subsitution of the expression for the $\alpha$ coefficient from Lemma
\ref{Chanda} as long as $\left\vert \theta\right\vert <1.$ Note that we need
to adjust the $\alpha$-mixing term relative to Andrews (1991) because
$Y_{1,t}$ is measurable w.r.t to $\mathcal{M}_{t-1}^{\infty}$ but not
$\mathcal{M}_{t}^{\infty}.$

For (vii) use the bound in (\ref{Cov_mix_bound}) such that
\[
\sum_{T>\left\vert v\right\vert \geq T/2}\left\vert \gamma_{a,b}\left(
0,v\right)  \right\vert \leq8K^{1/4}\sum_{T>\left\vert v\right\vert \geq
T/2}\alpha\left(  v\right)  ^{1/2}\leq16K^{1/4}\sum_{v=T/2}^{\infty}%
\alpha\left(  v\right)  ^{1/2}=o\left(  T^{-1}\right)
\]
where the second inequality uses $\alpha\left(  v\right)  =\alpha\left(
-v\right)  $ which holds by to strict stationarity. Finally, $\sum
_{v=T/2}^{\infty}\alpha\left(  v\right)  =o\left(  T^{-1}\right)  $ follows
from Lemma \ref{Chanda} and the fact that
\[
\alpha\left(  k\right)  =\sum_{j=k}^{\infty}j\left\vert \theta\right\vert
^{\lambda j}=\left(  1-\left\vert \theta\right\vert ^{\lambda}\right)
^{-2}\left(  \left\vert \theta\right\vert ^{\lambda}\right)  ^{k}\left(
k+\left\vert \theta\right\vert ^{\lambda}\left(  1-k\right)  \right)  .
\]

\end{proof}

We now turn to the analysis of terms that appear in the split sample
estimator. Let
\begin{align*}
\mathcal{Y}_{1}^{\left(  1\right)  }  &  =\frac{1}{\sqrt{T/2}}\sum_{t=1}%
^{T/2}Y_{1,t},\quad\mathcal{Y}_{1}^{\left(  2\right)  }=\frac{1}{\sqrt{T/2}%
}\sum_{t=T/2+1}^{T}Y_{1,t},\\
\mathcal{Y}_{2}^{\left(  1\right)  }  &  =\frac{1}{\sqrt{T/2}}\sum_{t=1}%
^{T/2}Y_{2,t},\quad\mathcal{Y}_{2}^{\left(  2\right)  }=\frac{1}{\sqrt{T/2}%
}\sum_{t=T/2+1}^{T}Y_{2,t},
\end{align*}
and write%
\[
\operatorname{Var}\left(  \mathcal{Y}_{1}\mathcal{Y}_{2}\right)  =E\left[
\mathcal{Y}_{1}^{2}\mathcal{Y}_{2}^{2}\right]  -\left(  E\left[
\mathcal{Y}_{1}\mathcal{Y}_{2}\right]  \right)  ^{2}%
\]
Further write%
\begin{align}
&  E\left[  \mathcal{Y}_{1}^{2}\mathcal{Y}_{2}^{2}\right] \nonumber\\
&  =E\left[  \left(  \frac{1}{\sqrt{2}}\left(  \mathcal{Y}_{1}^{\left(
1\right)  }+\mathcal{Y}_{1}^{\left(  2\right)  }\right)  \right)  ^{2}\left(
\frac{1}{\sqrt{2}}\left(  \mathcal{Y}_{2}^{\left(  1\right)  }+\mathcal{Y}%
_{2}^{\left(  2\right)  }\right)  \right)  ^{2}\right] \nonumber\\
&  =\frac{1}{4}E\left[  \left(  \left(  \mathcal{Y}_{1}^{\left(  1\right)
}\right)  ^{2}+2\mathcal{Y}_{1}^{\left(  1\right)  }\mathcal{Y}_{1}^{\left(
2\right)  }+\left(  \mathcal{Y}_{1}^{\left(  2\right)  }\right)  ^{2}\right)
\left(  \left(  \mathcal{Y}_{2}^{\left(  1\right)  }\right)  ^{2}%
+2\mathcal{Y}_{2}^{\left(  1\right)  }\mathcal{Y}_{2}^{\left(  2\right)
}+\left(  \mathcal{Y}_{2}^{\left(  2\right)  }\right)  ^{2}\right)  \right]
\nonumber\\
&  =\frac{1}{4}\left(
\begin{array}
[c]{c}%
E\left[  \left(  \mathcal{Y}_{1}^{\left(  1\right)  }\right)  ^{2}\left(
\mathcal{Y}_{2}^{\left(  1\right)  }\right)  ^{2}\right]  +E\left[  \left(
\mathcal{Y}_{1}^{\left(  2\right)  }\right)  ^{2}\left(  \mathcal{Y}%
_{2}^{\left(  1\right)  }\right)  ^{2}\right] \\
+4E\left[  \mathcal{Y}_{1}^{\left(  1\right)  }\mathcal{Y}_{1}^{\left(
2\right)  }\mathcal{Y}_{2}^{\left(  1\right)  }\mathcal{Y}_{2}^{\left(
2\right)  }\right] \\
+E\left[  \left(  \mathcal{Y}_{1}^{\left(  1\right)  }\right)  ^{2}\left(
\mathcal{Y}_{2}^{\left(  2\right)  }\right)  ^{2}\right]  +E\left[  \left(
\mathcal{Y}_{1}^{\left(  2\right)  }\right)  ^{2}\left(  \mathcal{Y}%
_{2}^{\left(  2\right)  }\right)  ^{2}\right]
\end{array}
\right) \nonumber\\
&  +\frac{1}{4}\left(
\begin{array}
[c]{c}%
2E\left[  \left(  \mathcal{Y}_{1}^{\left(  1\right)  }\right)  ^{2}%
\mathcal{Y}_{2}^{\left(  1\right)  }\mathcal{Y}_{2}^{\left(  2\right)
}\right] \\
+2E\left[  \mathcal{Y}_{1}^{\left(  1\right)  }\mathcal{Y}_{1}^{\left(
2\right)  }\left(  \mathcal{Y}_{2}^{\left(  1\right)  }\right)  ^{2}\right]
+2E\left[  \left(  \mathcal{Y}_{1}^{\left(  2\right)  }\right)  ^{2}%
\mathcal{Y}_{2}^{\left(  1\right)  }\mathcal{Y}_{2}^{\left(  2\right)
}\right] \\
+2E\left[  \mathcal{Y}_{1}^{\left(  1\right)  }\mathcal{Y}_{1}^{\left(
2\right)  }\left(  \mathcal{Y}_{2}^{\left(  2\right)  }\right)  ^{2}\right]
\end{array}
\right)  , \label{eq-AR1-split}%
\end{align}
where we note that the sum of the terms in the second round bracket on the far
RHS are zero under the IID setting. Under fairly general dependence conditions
it follows from Equation (\ref{eq-AR1-split}) and Lemma \ref{lem-AR1-4-terms},
that%
\begin{equation}
E\left[  \mathcal{Y}_{1}^{2}\mathcal{Y}_{2}^{2}\right]  =\frac{1}{4}\left(
\begin{array}
[c]{c}%
E\left[  \left(  \mathcal{Y}_{1}^{\left(  1\right)  }\right)  ^{2}\left(
\mathcal{Y}_{2}^{\left(  1\right)  }\right)  ^{2}\right]  +E\left[  \left(
\mathcal{Y}_{1}^{\left(  2\right)  }\right)  ^{2}\left(  \mathcal{Y}%
_{2}^{\left(  1\right)  }\right)  ^{2}\right] \\
+4E\left[  \mathcal{Y}_{1}^{\left(  1\right)  }\mathcal{Y}_{2}^{\left(
1\right)  }\mathcal{Y}_{1}^{\left(  2\right)  }\mathcal{Y}_{2}^{\left(
2\right)  }\right] \\
+E\left[  \left(  \mathcal{Y}_{1}^{\left(  1\right)  }\right)  ^{2}\left(
\mathcal{Y}_{2}^{\left(  2\right)  }\right)  ^{2}\right]  +E\left[  \left(
\mathcal{Y}_{1}^{\left(  2\right)  }\right)  ^{2}\left(  \mathcal{Y}%
_{2}^{\left(  2\right)  }\right)  ^{2}\right]
\end{array}
\right)  +o\left(  1\right)  . \label{eq-AR1-split-alt}%
\end{equation}
The conditions of Lemma \ref{lem-AR1-4-terms} are satisfied by the stationary
AR(1) model as showin in Lemmas \ref{Chanda} and \ref{cumsum}.
Furthermore\textbf{ }%
\[
E\left[  \mathcal{Y}_{1}^{2}\mathcal{Y}_{2}^{2}\right]  =\frac{1}{2}E\left[
\mathcal{Y}_{1}^{2}\mathcal{Y}_{2}^{2}\right]  +\frac{1}{2}E\left[
\mathcal{Y}_{1}^{2}\right]  E\left[  \mathcal{Y}_{2}^{2}\right]  +E\left[
\mathcal{Y}_{1}\mathcal{Y}_{2}\right]  ^{2}+o\left(  1\right)
\]
by Lemma \ref{lem-y1y2_T/2} and (\ref{y1y2_T/2_1})-(\ref{y1y2_T/2_4}). We
state the following general dependence assumptions.

\begin{condition}
\label{Assume_CumSum}Let $Y_{a_{j},t_{j}}$ be random variables with $E\left[
Y_{a_{j},t_{j}}\right]  =0$ and $E\left[  \left\vert Y_{a_{j},t_{k}%
}\right\vert ^{4}\right]  \leq K_{0}<\infty$ for all $a_{j}\in\left\{
1,2\right\}  ,t_{j}\in%
\mathbb{N}
$ and $j=1,...,4.$ Assume that $T$ is even. Assume that
\begin{equation}
\sum_{u=-\infty}^{\infty}\left\vert u\right\vert \sup_{t}\left\vert
\gamma_{a,b}\left(  t,t+u\right)  \right\vert \leq K_{1}<\infty,
\label{Cond_Cov_Sum}%
\end{equation}%
\begin{align}
\sum_{u=-\infty}^{\infty}\left\vert u\right\vert \sup_{t}\left\vert
\gamma_{a_{1}a_{2},a_{3}}\left(  t,t,t+u\right)  \right\vert  &  \leq
K_{2}<\infty,\label{Cond_3rdCum_Sum1}\\
\sum_{u=1}^{\infty}u\sup_{t\leq0}\left\vert \gamma_{a_{1}a_{2},a_{3}}\left(
t,0,u\right)  \right\vert  &  \leq K_{2}<\infty,\label{Cond_3rdCum_Sum2}\\
\sum_{u=1}^{\infty}u\sup_{t\leq0}\left\vert \gamma_{a_{1}a_{2},a_{3}}\left(
0,t,u\right)  \right\vert  &  \leq K_{2}<\infty,\label{Cond_3rdCum_Sum3}\\
\sum_{u=1}^{\infty}u\sup_{t>0}\left\vert \gamma_{a_{1}a_{2},a_{3}}\left(
u,t+u,0\right)  \right\vert  &  \leq K_{2}<\infty, \label{Cond_3rdCum_Sum4}%
\end{align}
and
\begin{equation}
\sum_{t_{2},t_{3},t_{4}=-\infty}^{\infty}\sup_{t_{1}}\left\vert cum_{a_{1}%
,a_{2},a_{3},a_{4}}\left(  t_{1},t_{2},t_{3},t_{4}\right)  \right\vert \leq
K_{3}<\infty. \label{Cond_Cum_Sum}%
\end{equation}
In addition, assume
\begin{equation}
\sum_{T>\left\vert v\right\vert \geq T/2}\left\vert \gamma_{a,b}\left(
0,v\right)  \right\vert =o\left(  T^{-1}\right)  \label{Cond_gam_Tail}%
\end{equation}

\end{condition}

\begin{remark}
In the stationary case $\gamma_{a,b}\left(  t_{1},t_{1}+u\right)
=\gamma_{a,b}\left(  0,u\right)  \equiv\gamma_{a,b}\left(  u\right)  ,$
$\gamma_{a_{1}a_{2},a_{3}}\left(  t_{1},t,t_{1}+u\right)  =\gamma_{a_{1}%
a_{2},a_{3}}\left(  0,t,u\right)  $ and $cum_{a_{1},a_{2},a_{3},a_{4}}\left(
t_{1},t_{2},t_{3},t_{4}\right)  \equiv cum_{a_{1},a_{2},a_{3},a_{4}}\left(
t_{2},t_{3},t_{4}\right)  $. In the stationary case the conditions of the
Lemma simplify to
\[
\sum_{u=-\infty}^{\infty}\left\vert u\right\vert \left\vert \gamma
_{a,b}\left(  u\right)  \right\vert \leq K_{1}<\infty,
\]%
\[
\sum_{u=-\infty}^{\infty}\left\vert u\right\vert \left\vert \gamma
_{\gamma_{a_{1}a_{2},a_{3}}}\left(  0,0,u\right)  \right\vert \leq
K_{2}<\infty,,
\]
and
\[
\sum_{t_{2},t_{3},t_{4}=-\infty}^{\infty}\left\vert cum_{a_{1},a_{2}%
,a_{3},a_{4}}\left(  t_{2},t_{3},t_{4}\right)  \right\vert \leq K_{2}<\infty.
\]

\end{remark}

\begin{remark}
When $\gamma_{a,a}\left(  u\right)  $ is the covariance function of a
stationary AR(1)\ model then $\gamma_{a,a}\left(  u\right)  =\theta
^{u}/\left(  1-\theta^{2}\right)  $ with $\left\vert \theta\right\vert <1$ it
which case it follows immediately that Condition (\ref{Cond_Cov_Sum}) is satisfied.
\end{remark}

\begin{remark}
Cumulant summability conditions are standard in the time series literature,
see for example Brillinger (1981, Section 2.6). Andrews (1991, Lemma 1) shows
that (\ref{Cond_Cum_Sum}) holds for mixing processes under regularity
conditions on the existence of moments and restrictions on the decay rates of
the mixing coefficients.
\end{remark}

The next two lemmas are established when Condition \ref{Assume_CumSum} holds
which is satisfied for the AR(1) model we consider.

\begin{lemma}
\label{lem-AR1-4-terms}Assume that Condition \ref{Assume_CumSum} holds. Then,
\newline(i) $E\left[  \left(  \mathcal{Y}_{1}^{\left(  1\right)  }\right)
^{2}\mathcal{Y}_{2}^{\left(  1\right)  }\mathcal{Y}_{2}^{\left(  2\right)
}\right]  =O\left(  T^{-1}\right)  $, \newline(ii) $E\left[  \mathcal{Y}%
_{1}^{\left(  1\right)  }\mathcal{Y}_{1}^{\left(  2\right)  }\left(
\mathcal{Y}_{2}^{\left(  1\right)  }\right)  ^{2}\right]  =O\left(
T^{-1}\right)  $, \newline(iii) $E\left[  \left(  \mathcal{Y}_{1}^{\left(
2\right)  }\right)  ^{2}\mathcal{Y}_{2}^{\left(  1\right)  }\mathcal{Y}%
_{2}^{\left(  2\right)  }\right]  =O\left(  T^{-1}\right)  $, and \newline(iv)
$E\left[  \mathcal{Y}_{1}^{\left(  1\right)  }\mathcal{Y}_{1}^{\left(
2\right)  }\left(  \mathcal{Y}_{2}^{\left(  2\right)  }\right)  ^{2}\right]
=O\left(  T^{-1}\right)  $.
\end{lemma}

\begin{proof}
[Proof of Lemma \ref{lem-AR1-4-terms}]For (i) write
\begin{equation}
\left(  \mathcal{Y}_{1}^{\left(  1\right)  }\right)  ^{2}\mathcal{Y}%
_{2}^{\left(  1\right)  }\mathcal{Y}_{2}^{\left(  2\right)  }=\frac{1}{\left(
\sqrt{T/2}\right)  ^{4}}\sum_{t_{1},t_{2},t_{3}=1}^{T/2}\sum_{t_{4}=T/2+1}%
^{T}Y_{1,t_{1}}Y_{1,t_{2}}Y_{2,t_{3}}Y_{2,t_{4}} \label{Q1(i)}%
\end{equation}
where $E\left[  Y_{1,t_{1}}\right]  =E\left[  Y_{1,t_{2}}\right]  =E\left[
Y_{2,t_{3}}\right]  =E\left[  Y_{2,t_{4}}\right]  =0.$ By Brillinger (1981,
Theorem 2.3.2) as before it follows that
\begin{align*}
E\left[  Y_{1,t_{1}}Y_{1,t_{2}}Y_{2,t_{3}}Y_{2,t_{4}}\right]   &
=\operatorname*{Cov}\left(  Y_{1,t_{1}},Y_{1,t_{2}}\right)
\operatorname*{Cov}\left(  Y_{2,t_{3}},Y_{2,t_{4}}\right) \\
&  +\operatorname*{Cov}\left(  Y_{1,t_{1}},Y_{2,t_{3}}\right)
\operatorname*{Cov}\left(  Y_{1,t_{2}},Y_{2,t_{4}}\right) \\
&  +\operatorname*{Cov}\left(  Y_{1,t_{1}},Y_{2,t_{4}}\right)
\operatorname*{Cov}\left(  Y_{2,t_{3}},Y_{1,t_{2}}\right) \\
&  +cum\left(  Y_{1,t_{1}},Y_{1,t_{2}},Y_{2,t_{3}},Y_{2,t_{4}}\right) \\
&  =\gamma_{1,1}\left(  t_{1},t_{2}\right)  \gamma_{2,2}\left(  t_{3}%
,t_{4}\right)  +\gamma_{1,2}\left(  t_{1},t_{3}\right)  \gamma_{1,2}\left(
t_{2},t_{4}\right) \\
&  +\gamma_{1,2}\left(  t_{1},t_{4}\right)  \gamma_{1,2}\left(  t_{2}%
,t_{3}\right)  +cum_{1,1,2,2}\left(  t_{1},t_{2},t_{3},t_{4}\right)
\end{align*}
Then,
\begin{align*}
&  \frac{1}{\left(  \sqrt{T/2}\right)  ^{4}}\sum_{t_{1},t_{2},t_{3}=1}%
^{T/2}\sum_{t_{4}=T/2+1}^{T}\gamma_{1,1}\left(  t_{1},t_{2}\right)
\gamma_{2,2}\left(  t_{3},t_{4}\right) \\
&  =\frac{2}{T}\sum_{t_{1},t_{2}=1}^{T/2}\gamma_{1,1}\left(  t_{1}%
,t_{2}\right)  \frac{2}{T}\sum_{t_{3}=1}^{T/2}\sum_{t_{4}=T/2+1}^{T}%
\gamma_{2,2}\left(  t_{3},t_{4}\right)  ,
\end{align*}
where the first term is bounded by
\begin{align}
\frac{2}{T}\left\vert \sum_{t_{1}=1}^{T/2}\sum_{v=-t_{1}+1}^{T/2-t_{1}}%
\gamma_{1,1}\left(  t_{1},t_{1}+v\right)  \right\vert  &  \leq\frac{2}{T}%
\sum_{t_{1}=1}^{T/2}\sum_{v=-\infty}^{\infty}\left\vert \gamma_{1,1}\left(
t_{1},t_{1}+v\right)  \right\vert \leq\sum_{v=-\infty}^{\infty}\sup_{t_{1}%
}\left\vert \gamma_{1,1}\left(  t_{1},t_{1}+v\right)  \right\vert \nonumber\\
&  =\sup_{t_{1}}\left\vert \gamma_{1,1}\left(  t_{1},t_{1}\right)  \right\vert
+\sum_{v\neq0}\sup_{t_{1}}\left\vert \gamma_{1,1}\left(  t_{1},t_{1}+v\right)
\right\vert \nonumber\\
&  \leq\sup_{t_{1}}\operatorname*{Cov}\left(  Y_{1,t_{1}},Y_{1,t_{1}}\right)
+\sum_{v\neq0}\left\vert v\right\vert \sup_{t_{1}}\left\vert \gamma
_{1,1}\left(  t_{1},t_{1}+v\right)  \right\vert \nonumber\\
&  \leq\sup_{t_{1}}E\left[  \left\vert Y_{1,t_{1}}\right\vert ^{2}\right]
+\sum_{v\neq0}\left\vert v\right\vert \sup_{t_{1}}\left\vert \gamma
_{1,1}\left(  t_{1},t_{1}+v\right)  \right\vert \leq K_{0}^{1/2}%
+K_{1}\mathbf{\ } \label{TS_Cov_Bound_1}%
\end{align}
where (\ref{Cond_Cov_Sum}) was used in the last inequaltiy. Use the change of
variables $t_{4}=t_{3}+u$ such that $u=t_{3}-t_{4}$. The second term becomes
\begin{align*}
\frac{2}{T}\sum_{t_{3}=1}^{T/2}\sum_{t_{4}=T/2+1}^{T}\gamma_{2,2}\left(
t_{3},t_{4}\right)   &  =\frac{2}{T}\sum_{t_{3}=1}^{T/2}\sum_{u=T/2+1-t_{3}%
}^{T-t_{3}}\gamma_{2,2}\left(  t_{3},t_{3}+u\right) \\
&  =\frac{2}{T}\sum_{u=1}^{T/2}\sum_{t_{3}=T/2+1-u}^{T/2}\gamma_{2,2}\left(
t_{3},t_{3}+u\right)
\end{align*}
such that we can bound it by
\begin{align}
\frac{2}{T}\left\vert \sum_{t_{3}=1}^{T/2}\sum_{t_{4}=T/2+1}^{T}\gamma
_{2,2}\left(  t_{3},t_{4}\right)  \right\vert  &  \leq\frac{2}{T}\sum
_{u=1}^{T/2}\sum_{t_{3}=T/2+1-u}^{T/2}\sup_{t_{3}}\left\vert \gamma
_{2,2}\left(  t_{3},t_{3}+u\right)  \right\vert \nonumber\\
&  =\frac{2}{T}\sum_{u=1}^{T/2}\left\vert u\right\vert \sup_{t_{3}}\left\vert
\gamma_{2,2}\left(  t_{3},t_{3}+u\right)  \right\vert \nonumber\\
&  \leq\frac{2}{T}\sum_{u=-\infty}^{\infty}\left\vert u\right\vert \sup
_{t_{3}}\left\vert \gamma_{2,2}\left(  t_{3},t_{3}+u\right)  \right\vert
\leq\frac{2K_{1}}{T}=O\left(  T^{-1}\right)  , \label{TS_Cov_Bound_2}%
\end{align}
where (\ref{Cond_Cov_Sum}) was used in the last inequaltiy. Similarly,
\begin{align*}
&  \frac{1}{\left(  \sqrt{T/2}\right)  ^{4}}\left\vert \sum_{t_{1},t_{2}%
,t_{3}=1}^{T/2}\sum_{t_{4}=T/2+1}^{T}\gamma_{1,2}\left(  t_{1},t_{3}\right)
\gamma_{1,2}\left(  t_{2},t_{4}\right)  \right\vert \\
&  \leq\frac{2}{T}\sum_{t_{1},t_{3}=1}^{T/2}\left\vert \gamma_{1,2}\left(
t_{1},t_{3}\right)  \right\vert \frac{2}{T}\sum_{t_{2}=1}^{T/2}\sum
_{t_{4}=T/2+1}^{T}\left\vert \gamma_{1,2}\left(  t_{2},t_{4}\right)
\right\vert =O\left(  T^{-1}\right)
\end{align*}
by the same arguments. The third term is
\begin{align*}
&  \frac{1}{\left(  \sqrt{T/2}\right)  ^{4}}\left\vert \sum_{t_{1},t_{2}%
,t_{3}=1}^{T/2}\sum_{t_{4}=T/2+1}^{T}\gamma_{1,2}\left(  t_{1},t_{4}\right)
\gamma_{1,2}\left(  t_{2},t_{3}\right)  \right\vert \\
&  \leq\frac{2}{T}\sum_{t_{2},t_{3}=1}^{T/2}\left\vert \gamma_{1,2}\left(
t_{2},t_{3}\right)  \right\vert \frac{2}{T}\sum_{t_{1}=1}^{T/2}\sum
_{t_{4}=T/2+1}^{T}\left\vert \gamma_{1,2}\left(  t_{1},t_{4}\right)
\right\vert =O\left(  T^{-1}\right)
\end{align*}
again by the same arguments. Finally, note that
\begin{align}
\frac{1}{\left(  \sqrt{T/2}\right)  ^{4}}&\left\vert \sum_{t_{1},t_{2},t_{3}%
=1}^{T/2}\sum_{t_{4}=T/2+1}^{T}cum_{1,1,2,2}\left(  t_{1},t_{2},t_{3}%
,t_{4}\right)  \right\vert  \\
&  \leq\frac{4}{T^{2}}\sum_{t_{1},t_{2}%
,t_{3},t_{4}=1}^{T}\left\vert cum_{1,1,2,2}\left(  t_{1},t_{2},t_{3}%
,t_{4}\right)  \right\vert \nonumber\\
&  \leq\frac{4}{T^{2}}\sum_{t_{1}=1}^{T}\sum_{t_{2},t_{3},t_{4}=-\infty
}^{\infty}\sup_{t_{1}}\left\vert cum_{1,1,2,2}\left(  t_{1},t_{2},t_{3}%
,t_{4}\right)  \right\vert \nonumber\\
&  =O\left(  T^{-1}\right)  , \label{TS_Cum_Bound}%
\end{align}
where (\ref{Cond_Cum_Sum}) was used in the last inequality. Combining
(\ref{TS_Cov_Bound_1}), (\ref{TS_Cov_Bound_2}) and (\ref{TS_Cum_Bound}) shows
that
\begin{align*}
& \left\vert E\left[  \left(  \mathcal{Y}_{1}^{\left(  1\right)  }\right)
^{2}\mathcal{Y}_{2}^{\left(  1\right)  }\mathcal{Y}_{2}^{\left(  2\right)
}\right]  \right\vert  \\
&  =\frac{1}{\left(  \sqrt{T/2}\right)  ^{4}}\left\vert
\sum_{t_{1},t_{2},t_{3}=1}^{T/2}\sum_{t_{4}=T/2+1}^{T}E\left[  Y_{1,t_{1}%
}Y_{1,t_{2}}Y_{2,t_{3}}Y_{2,t_{4}}\right]  \right\vert \\
&  =\frac{1}{\left(  \sqrt{T/2}\right)  ^{4}}\left\vert \sum_{t_{1}%
,t_{2},t_{3}=1}^{T/2}\sum_{t_{4}=T/2+1}^{T}%
\begin{array}
[c]{c}%
\gamma_{1,1}\left(  t_{1},t_{2}\right)  \gamma_{2,2}\left(  t_{3}%
,t_{4}\right)  +\gamma_{1,2}\left(  t_{1},t_{3}\right)  \gamma_{1,2}\left(
t_{2},t_{4}\right) \\
+\gamma_{1,2}\left(  t_{1},t_{4}\right)  \gamma_{1,2}\left(  t_{2}%
,t_{3}\right)  +cum_{1,1,2,2}\left(  t_{1},t_{2},t_{3},t_{4}\right)
\end{array}
\right\vert \\
&  \leq\sum_{v=-\infty}^{\infty}\sup_{t_{1}}\left\vert \gamma_{1,1}\left(
t_{1},t_{1}+v\right)  \right\vert \frac{2}{T}\sum_{u=-\infty}^{\infty
}\left\vert u\right\vert \sup_{t_{3}}\left\vert \gamma_{2,2}\left(
t_{3},t_{3}+u\right)  \right\vert \\
&  +2\sum_{v=-\infty}^{\infty}\sup_{t_{1}}\left\vert \gamma_{1,2}\left(
t_{1},t_{1}+v\right)  \right\vert \frac{2}{T}\sum_{u=-\infty}^{\infty
}\left\vert u\right\vert \sup_{t_{3}}\left\vert \gamma_{2,1}\left(
t_{3},t_{3}+u\right)  \right\vert \\
&  +\frac{4}{T^{2}}\sum_{t_{1}=1}^{T}\sum_{t_{2},t_{3},t_{4}=-\infty}^{\infty
}\sup_{t_{1}}\left\vert cum_{1,1,2,1}\left(  t_{1},t_{2},t_{3},t_{4}\right)
\right\vert \\
&  =O\left(  T^{-1}\right)  .
\end{align*}

For (ii) write%
\[
\mathcal{Y}_{1}^{\left(  1\right)  }\mathcal{Y}_{1}^{\left(  2\right)
}\left(  \mathcal{Y}_{2}^{\left(  1\right)  }\right)  ^{2}=\frac{1}{\left(
\sqrt{T/2}\right)  ^{4}}\sum_{t_{1},t_{2},t_{3}=1}^{T/2}\sum_{t_{4}=T/2+1}%
^{T}Y_{1,t_{1}}Y_{2,t_{2}}Y_{2,t_{3}}Y_{1,t_{4}}%
\]
and
\begin{align*}
E\left[  Y_{1,t_{1}}Y_{2,t_{2}}Y_{2,t_{3}}Y_{1,t_{4}}\right]   &
=\gamma_{1,2}\left(  t_{1},t_{2}\right)  \gamma_{2,1}\left(  t_{3}%
,t_{4}\right)  +\gamma_{1,2}\left(  t_{1},t_{3}\right)  \gamma_{2,1}\left(
t_{2},t_{4}\right) \\
&  +\gamma_{1,2}\left(  t_{1},t_{4}\right)  \gamma_{2,1}\left(  t_{2}%
,t_{3}\right)  +cum_{1,2,2,1}\left(  t_{1},t_{2},t_{3},t_{4}\right)  .
\end{align*}
It follows that
\[
E\left[  \mathcal{Y}_{1}^{\left(  1\right)  }\mathcal{Y}_{1}^{\left(
2\right)  }\left(  \mathcal{Y}_{2}^{\left(  2\right)  }\right)  ^{2}\right]
=\frac{1}{\left(  \sqrt{T/2}\right)  ^{4}}\sum_{t_{1},t_{2},t_{3}=1}^{T/2}%
\sum_{t_{4}=T/2+1}^{T}E\left[  Y_{1,t_{1}}Y_{2,t_{2}}Y_{2,t_{3}}Y_{1,t_{4}%
}\right]  =O\left(  T^{-1}\right)
\]
because the range of summation as well as the summability of the covariance
and cumulant terms is identical to (\ref{Q1(i)}).

For case (iii)\ consider
\begin{equation}
\left(  \mathcal{Y}_{1}^{\left(  2\right)  }\right)  ^{2}\mathcal{Y}%
_{2}^{\left(  1\right)  }\mathcal{Y}_{2}^{\left(  2\right)  }=\frac{1}{\left(
\sqrt{T/2}\right)  ^{4}}\sum_{t_{1}=1}^{T/2}\sum_{t_{2},t_{3},t_{4}=T/2+1}%
^{T}Y_{2,t_{1}}Y_{1,t_{2}}Y_{1,t_{3}}Y_{2,t_{4}} \label{Q1(ii)}%
\end{equation}
where $E\left[  Y_{1,t_{1}}\right]  =E\left[  Y_{2,t_{2}}\right]  =E\left[
Y_{2,t_{3}}\right]  =E\left[  Y_{1,t_{4}}\right]  =0.$ By Brillinger (1981,
Theorem 2.3.2) as before it follows that
\begin{align*}
E\left[  Y_{1,t_{1}}Y_{1,t_{2}}Y_{2,t_{3}}Y_{1,t_{4}}\right]   &
=\gamma_{2,1}\left(  t_{1},t_{2}\right)  \gamma_{1,2}\left(  t_{3}%
,t_{4}\right)  +\gamma_{2,1}\left(  t_{1},t_{3}\right)  \gamma_{1,2}\left(
t_{2},t_{4}\right) \\
&  +\gamma_{2,2}\left(  t_{1},t_{4}\right)  \gamma_{1,1}\left(  t_{2}%
,t_{3}\right)  +cum_{2,1,1,2}\left(  t_{1},t_{2},t_{3},t_{4}\right)  .
\end{align*}
Again considering the three covariance terms in turn one has
\begin{align*}
&  \frac{1}{\left(  \sqrt{T/2}\right)  ^{4}}\sum_{t_{1}=1}^{T/2}\sum
_{t_{2},t_{3},t_{4}=T/2+1}^{T}\gamma_{2,1}\left(  t_{1},t_{2}\right)
\gamma_{1,2}\left(  t_{3},t_{4}\right) \\
&  =\frac{2}{T}\sum_{t_{1}=1}^{T/2}\sum_{t_{2}=T/2+1}^{T}\gamma_{2,1}\left(
t_{1},t_{2}\right)  \frac{2}{T}\sum_{t_{3},t_{4}=T/2+1}^{T}\gamma_{1,2}\left(
t_{3},t_{4}\right)
\end{align*}
where
\begin{align*}
\left\vert \frac{2}{T}\sum_{t_{3},t_{4}=T/2+1}^{T}\gamma_{1,2}\left(
t_{3},t_{4}\right)  \right\vert &\leq\frac{2}{T}\sum_{t_{3}=T/2+1}^{T}%
\sum_{v=T/2+-t_{3}}^{T-t_{3}}\left\vert \gamma_{1,2}\left(  t_{3}%
,t_{3}+v\right)  \right\vert \\
&\leq\frac{2}{T}\sum_{t_{3}=T/2+1}^{T}%
\sum_{v=-\infty}^{\infty}\left\vert \gamma_{1,2}\left(  t_{3},t_{3}+v\right)
\right\vert \leq\frac{K_{1}}{T}%
\end{align*}
and
\[
\left\vert \frac{2}{T}\sum_{t_{1}=1}^{T/2}\sum_{t_{2}=T/2+1}^{T}\gamma
_{2,1}\left(  t_{1},t_{2}\right)  \right\vert \leq\frac{2}{T}\sum_{u=1}%
^{T/2}\sum_{t_{1}=T/2+1-u}^{T/2}\left\vert \gamma_{2,1}\left(  t_{1}%
,t_{1}+u\right)  \right\vert \leq\frac{K_{2}}{T}.
\]
Similarly,
\begin{align*}
&  \frac{1}{\left(  \sqrt{T/2}\right)  ^{4}}\sum_{t_{1}=1}^{T/2}\sum
_{t_{2},t_{3},t_{4}=T/2+1}^{T}\gamma_{2,1}\left(  t_{1},t_{3}\right)
\gamma_{1,2}\left(  t_{2},t_{4}\right) \\
&  =\frac{2}{T}\sum_{t_{1}=1}^{T/2}\sum_{t_{3}=T/2+1}^{T}\gamma_{2,1}\left(
t_{1},t_{3}\right)  \frac{2}{T}\sum_{t_{2},t_{4}=T/2+1}^{T}\gamma_{1,2}\left(
t_{2},t_{4}\right)  =O\left(  T^{-1}\right)
\end{align*}
by the same analysis. Finally,
\begin{align*}
&  \frac{1}{\left(  \sqrt{T/2}\right)  ^{4}}\sum_{t_{1}=1}^{T/2}\sum
_{t_{2},t_{3},t_{4}=T/2+1}^{T}\gamma_{2,2}\left(  t_{1},t_{4}\right)
\gamma_{1,1}\left(  t_{2},t_{3}\right) \\
&  =\frac{2}{T}\sum_{t_{1}=1}^{T/2}\sum_{t_{4}=T/2+1}^{T}\gamma_{2,2}\left(
t_{1},t_{4}\right)  \frac{2}{T}\sum_{t_{2},t_{3}=T/2+1}^{T}\gamma_{1,1}\left(
t_{2},t_{3}\right)  =O\left(  T^{-1}\right)  .
\end{align*}
For case (iv) consider
\[
\mathcal{Y}_{1}^{\left(  1\right)  }\mathcal{Y}_{1}^{\left(  2\right)
}\left(  \mathcal{Y}_{2}^{\left(  2\right)  }\right)  ^{2}=\frac{1}{\left(
\sqrt{T/2}\right)  ^{4}}\sum_{t_{1}=1}^{T/2}\sum_{t_{2},t_{3},t_{4}=T/2+1}%
^{T}Y_{1,t_{1}}Y_{1,t_{2}}Y_{2,t_{3}}Y_{2,t_{4}}%
\]
and
\begin{align*}
E\left[  Y_{1,t_{1}}Y_{1,t_{2}}Y_{2,t_{3}}Y_{2,t_{4}}\right]   &
=\gamma_{1,1}\left(  t_{1},t_{2}\right)  \gamma_{2,2}\left(  t_{3}%
,t_{4}\right)  +\gamma_{1,2}\left(  t_{1},t_{3}\right)  \gamma_{1,2}\left(
t_{2},t_{4}\right) \\
&  +\gamma_{1,2}\left(  t_{1},t_{4}\right)  \gamma_{1,2}\left(  t_{2}%
,t_{3}\right)  +cum_{1,1,2,2}\left(  t_{1},t_{2},t_{3},t_{4}\right)  .
\end{align*}
It follows that
\[
\mathcal{Y}_{1}^{\left(  1\right)  }\mathcal{Y}_{1}^{\left(  2\right)
}\left(  \mathcal{Y}_{2}^{\left(  2\right)  }\right)  ^{2}=\frac{1}{\left(
\sqrt{T/2}\right)  ^{4}}\sum_{t_{1}=1}^{T/2}\sum_{t_{2},t_{3},t_{4}=T/2+1}%
^{T}E\left[  Y_{1,t_{1}}Y_{1,t_{2}}Y_{2,t_{3}}Y_{2,t_{4}}\right]  =O\left(
T^{-1}\right)
\]
because the range of summation as well as the summability of the covariance
and cumulant terms is identical to (\ref{Q1(ii)}).
\end{proof}

\begin{lemma}
\label{lem-y1y2_T/2}Assume that Condition \ref{Assume_CumSum} holds. In
addition assume that $Y_{1,t}$ and $Y_{2,t}$ are strictly stationary.
Then,\textbf{ }%
\begin{align}
E\left[  \left(  \mathcal{Y}_{1}^{\left(  1\right)  }\right)  ^{2}\left(
\mathcal{Y}_{2}^{\left(  1\right)  }\right)  ^{2}\right]   &  =E\left[
\left(  \mathcal{Y}_{1}^{\left(  2\right)  }\right)  ^{2}\left(
\mathcal{Y}_{2}^{\left(  2\right)  }\right)  ^{2}\right]  =E\left[
\mathcal{Y}_{1}^{2}\mathcal{Y}_{2}^{2}\right]  +o\left(  1\right)
\label{y1y2_T/2_1}\\
E\left[  \left(  \mathcal{Y}_{1}^{\left(  1\right)  }\right)  ^{2}\left(
\mathcal{Y}_{2}^{\left(  2\right)  }\right)  ^{2}\right]   &  =E\left[
\left(  \mathcal{Y}_{1}^{\left(  2\right)  }\right)  ^{2}\left(
\mathcal{Y}_{2}^{\left(  1\right)  }\right)  ^{2}\right]  =E\left[
\mathcal{Y}_{1}^{2}\right]  E\left[  \mathcal{Y}_{2}^{2}\right]  +o\left(
1\right) \label{y1y2_T/2_2}\\
E\left[  \mathcal{Y}_{1}^{\left(  1\right)  }\mathcal{Y}_{2}^{\left(
1\right)  }\right]   &  =E\left[  \mathcal{Y}_{1}^{\left(  2\right)
}\mathcal{Y}_{2}^{\left(  2\right)  }\right]  =E\left[  \mathcal{Y}%
_{1}\mathcal{Y}_{2}\right]  +o\left(  1\right)  . \label{y1y2_T/2_3}%
\end{align}%
\begin{equation}
E\left[  \mathcal{Y}_{1}^{\left(  1\right)  }\mathcal{Y}_{2}^{\left(
1\right)  }\mathcal{Y}_{1}^{\left(  2\right)  }\mathcal{Y}_{2}^{\left(
2\right)  }\right]  =E\left[  \mathcal{Y}_{1}\mathcal{Y}_{2}\right]
^{2}+o\left(  1\right)  . \label{y1y2_T/2_4}%
\end{equation}%
\begin{align}
E\left[  \left(  \mathcal{Y}_{1}^{\left(  1\right)  }\right)  ^{2}%
\mathcal{Y}_{2}^{\left(  2\right)  }\right]   &  =E\left[  \mathcal{Y}_{1}%
^{2}\right]  E\left[  \mathcal{Y}_{2}\right]  +o\left(  1\right)  =o\left(
1\right)  ,\label{y1y2_T/2_5}\\
E\left[  \mathcal{Y}_{1}^{\left(  1\right)  }\mathcal{Y}_{1}^{\left(
2\right)  }\mathcal{Y}_{2}^{\left(  2\right)  }\right]   &  =E\left[
\mathcal{Y}_{1}\right]  E\left[  \mathcal{Y}_{1}\mathcal{Y}_{2}\right]
=o\left(  1\right)  . \label{y1y2_T/2_6}%
\end{align}

\end{lemma}

\begin{proof}
First consider
\begin{align}
E\left[  \mathcal{Y}_{1}^{2}\mathcal{Y}_{2}^{2}\right]   &  =\frac{1}{\left(
\sqrt{T}\right)  ^{4}}\sum_{t_{1},t_{2},t_{3},t_{4}=1}^{T}E\left[  Y_{1,t_{1}%
}Y_{1,t_{2}}Y_{2,t_{3}}Y_{2,t_{4}}\right]  \label{y1y2_T}\\
&  =\frac{1}{\left(  \sqrt{T}\right)  ^{4}}\sum_{t_{1},t_{2}=1}^{T}%
\sum_{t_{3,}t_{4}=1}^{T}\gamma_{1,1}\left(  t_{1},t_{2}\right)  \gamma
_{2,2}\left(  t_{3},t_{4}\right)  \label{y1y2_T1}\\
&  +\frac{1}{\left(  \sqrt{T}\right)  ^{4}}\sum_{t_{1},t_{2}=1}^{T}%
\sum_{t_{3,}t_{4}=1}^{T}\gamma_{1,2}\left(  t_{1},t_{3}\right)  \gamma
_{1,2}\left(  t_{2},t_{4}\right)  \label{y1y2_T2}\\
&  +\frac{1}{\left(  \sqrt{T}\right)  ^{4}}\sum_{t_{1},t_{2}=1}^{T}%
\sum_{t_{3,}t_{4}=1}^{T}\gamma_{1,2}\left(  t_{1},t_{4}\right)  \gamma
_{1,2}\left(  t_{2},t_{3}\right)  \label{y1y2_T3}\\
&  +\frac{1}{\left(  \sqrt{T}\right)  ^{4}}\sum_{t_{1},t_{2}=1}^{T}%
\sum_{t_{3,}t_{4}=1}^{T}cum_{1,1,2,2}\left(  t_{1},t_{2},t_{3},t_{4}\right)
\nonumber
\end{align}
where again the cumulant term is $O\left(  T^{-1}\right)  .$ For the first
term%
\begin{equation}
\frac{1}{\left(  \sqrt{T}\right)  ^{4}}\sum_{t_{1},t_{2}=1}^{T}\sum
_{t_{3,}t_{4}=1}^{T}\gamma_{1,1}\left(  t_{1},t_{2}\right)  \gamma
_{2,2}\left(  t_{3},t_{4}\right)  =\frac{1}{T}\sum_{t_{1},t_{2}=1}^{T}%
\gamma_{1,1}\left(  t_{1},t_{2}\right)  \frac{1}{T}\sum_{t_{3,}t_{4}=1}%
^{T}\gamma_{2,2}\left(  t_{3},t_{4}\right)  \label{gam_11_T}%
\end{equation}
where
\begin{align*}
\frac{1}{T}\sum_{t_{1},t_{2}=1}^{T}\gamma_{1,1}\left(  t_{1},t_{2}\right)   &
=\sum_{v=-T+1}^{T-1}\left(  1-\frac{\left\vert v\right\vert }{T}\right)
\gamma_{1,1}\left(  0,v\right)  \\
&  =\sum_{v=-T+1}^{T-1}\gamma_{1,1}\left(  0,v\right)  +\frac{1}{T}%
\sum_{v=-T+1}^{T-1}\left\vert v\right\vert \gamma_{1,1}\left(  0,v\right)  \\
&  =\sum_{v=-T+1}^{T-1}\gamma_{1,1}\left(  0,v\right)  +O\left(
T^{-1}\right)
\end{align*}
where the first equality follows from stationarity of the AR(1) model and the
third equality follows from the fact that $\sum_{v=-T+1}^{T-1}\left\vert
v\right\vert \gamma_{1,1}\left(  0,v\right)  $ is uniformly bounded in $T$ by
Assumption \ref{Cond_Cov_Sum}. Similarly, it follows that the second term in
(\ref{gam_11_T}) satisfies
\[
\frac{1}{T}\sum_{t_{3,}t_{4}=1}^{T}\gamma_{2,2}\left(  t_{3},t_{4}\right)
=\sum_{v=-T+1}^{T-1}\gamma_{2,2}\left(  0,v\right)  +O\left(  T^{-1}\right)  .
\]
The terms in (\ref{y1y2_T2}) can be analyzed in the same way, noting in
particular that stationarity implies that for any $h,$ $\gamma_{1,2}\left(
t_{1},t_{3}\right)  =\gamma_{1,2}\left(  t_{1}+h,t_{3}+h\right)  =\gamma
_{1,2}\left(  0,t_{3}-t_{1}\right)  $ such that
\begin{equation}
\frac{1}{\left(  \sqrt{T}\right)  ^{4}}\sum_{t_{1},t_{2}=1}^{T}\sum
_{t_{3,}t_{4}=1}^{T}\gamma_{1,2}\left(  t_{1},t_{3}\right)  \gamma
_{1,2}\left(  t_{2},t_{4}\right)  =\sum_{v_{1},\nu_{2}=-T+1}^{T-1}\gamma
_{1,2}\left(  0,v_{1}\right)  \gamma_{1,2}\left(  0,v_{2}\right)  +O\left(
T^{-1}\right)  .\label{gam_12_T}%
\end{equation}
The same analysis applies to (\ref{y1y2_T3}) such that (\ref{y1y2_T}) can be
written as
\begin{align*}
E\left[  \mathcal{Y}_{1}^{2}\mathcal{Y}_{2}^{2}\right]   &  =\left(
\sum_{v=-T+1}^{T-1}\gamma_{1,1}\left(  0,v\right)  \right)  \left(
\sum_{v=-T+1}^{T-1}\gamma_{2,2}\left(  0,v\right)  \right)  \\
&  +2\left(  \sum_{v=-T+1}^{T-1}\gamma_{1,2}\left(  0,v_{1}\right)  \right)
^{2}+O\left(  T^{-1}\right)  .
\end{align*}
Now turn to the split sample moments. We start with (\ref{y1y2_T/2_1})
\begin{align}
E\left[  \left(  \mathcal{Y}_{1}^{\left(  1\right)  }\right)  ^{2}\left(
\mathcal{Y}_{2}^{\left(  1\right)  }\right)  ^{2}\right]   &  =\frac
{1}{\left(  \sqrt{T/2}\right)  ^{4}}\sum_{t_{1},t_{2},t_{3},t_{4}=1}%
^{T/2}E\left[  Y_{1,t_{1}}Y_{1,t_{2}}Y_{2,t_{3}}Y_{2,t_{4}}\right]
\label{y1y2_T/2_a}\\
&  =\frac{1}{\left(  \sqrt{T/2}\right)  ^{4}}\sum_{t_{1},t_{2},t_{3},t_{4}%
=1}^{T/2}\gamma_{1,1}\left(  t_{1},t_{2}\right)  \gamma_{2,2}\left(
t_{3},t_{4}\right)  \label{y1y2_T/2_b}\\
&  +\frac{1}{\left(  \sqrt{T/2}\right)  ^{4}}\sum_{t_{1},t_{2},t_{3},t_{4}%
=1}^{T/2}\gamma_{1,2}\left(  t_{1},t_{3}\right)  \gamma_{1,2}\left(
t_{2},t_{4}\right)  \label{y1y2_T/2_c}\\
&  +\frac{1}{\left(  \sqrt{T/2}\right)  ^{4}}\sum_{t_{1},t_{2},t_{3},t_{4}%
=1}^{T/2}\gamma_{1,2}\left(  t_{1},t_{4}\right)  \gamma_{1,2}\left(
t_{2},t_{3}\right)  \label{y1y2_T/2_d}\\
&  +\frac{1}{\left(  \sqrt{T/2}\right)  ^{4}}\sum_{t_{1},t_{2},t_{3},t_{4}%
=1}^{T/2}cum_{1,1,2,2}\left(  t_{1},t_{2},t_{3},t_{4}\right)  \nonumber
\end{align}
where the cumulant term is again $O\left(  T^{-1}\right)  $ by Assumption
\ref{Cond_Cum_Sum}. Now turn to (\ref{y1y2_T/2_b})%
\begin{equation}
\frac{1}{\left(  \sqrt{T/2}\right)  ^{4}}\sum_{t_{1},t_{2}=1}^{T/2}%
\sum_{t_{3,}t_{4}=1}^{T/2}\gamma_{1,1}\left(  t_{1},t_{2}\right)  \gamma
_{2,2}\left(  t_{3},t_{4}\right)  =\frac{2}{T}\sum_{t_{1},t_{2}=1}^{T/2}%
\gamma_{1,1}\left(  t_{1},t_{2}\right)  \frac{2}{T}\sum_{t_{3,}t_{4}=1}%
^{T/2}\gamma_{2,2}\left(  t_{3},t_{4}\right)  \label{y1y2_T/2_b1}%
\end{equation}
where the first term on the RHS of (\ref{y1y2_T/2_b1}) is
\begin{align}
\frac{2}{T}\sum_{t_{1},t_{2}=1}^{T/2}\gamma_{1,1}\left(  t_{1},t_{2}\right)
&  =\sum_{v=-T/2+1}^{T/2-1}\left(  1-\frac{2\left\vert v\right\vert }%
{T}\right)  \gamma_{1,1}\left(  0,v\right)  \label{gam_11_rep}\\
&  =\sum_{v=-T/2+1}^{T/2-1}\gamma_{1,1}\left(  0,v\right)  +O\left(
T^{-1}\right)  \nonumber
\end{align}
by the same arguments as in the full sample case. For the second term on the
RHS of (\ref{y1y2_T/2_b1}) we have
\begin{equation}
\frac{2}{T}\sum_{t_{3,}t_{4}=1}^{T}\gamma_{2,2}\left(  t_{3},t_{4}\right)
=\sum_{v=-T/2+1}^{T/2-1}\gamma_{2,2}\left(  0,v\right)  +O\left(
T^{-1}\right)  \label{gam_22_rep}%
\end{equation}
by the same arguments. Next turn to (\ref{y1y2_T/2_c}) which is
\[
\frac{1}{\left(  \sqrt{T/2}\right)  ^{4}}\sum_{t_{1},t_{2},t_{3},t_{4}%
=1}^{T/2}\gamma_{1,2}\left(  t_{1},t_{3}\right)  \gamma_{1,2}\left(
t_{2},t_{4}\right)  =\left(  \frac{2}{T}\sum_{t_{1},t_{2}=1}^{T/2}\gamma
_{1,2}\left(  t_{1},t_{2}\right)  \right)  ^{2}%
\]
and using stationarity such that for any integer $h,$ $\gamma_{1,2}\left(
t_{1},t_{2}\right)  =\gamma_{1,2}\left(  t_{1}+h,t_{2}+h\right)  $ leads to
\begin{equation}
\frac{2}{T}\sum_{t_{1},t_{2}=1}^{T/2}\gamma_{1,2}\left(  t_{1},t_{2}\right)
=\sum_{v=-T/2+1}^{T/2-1}\left(  1-\frac{2\left\vert v\right\vert }{T}\right)
\gamma_{1,2}\left(  0,v\right)  =\sum_{v=-T/2+1}^{T/2-1}\gamma_{1,2}\left(
0,v\right)  +O\left(  T^{-1}\right)  \label{gam_12_T/2}%
\end{equation}
such that
\[
\frac{1}{\left(  \sqrt{T/2}\right)  ^{4}}\sum_{t_{1},t_{2},t_{3},t_{4}%
=1}^{T/2}\gamma_{1,2}\left(  t_{1},t_{3}\right)  \gamma_{1,2}\left(
t_{2},t_{4}\right)  =\left(  \sum_{v=-T/2+1}^{T/2-1}\gamma_{1,2}\left(
0,v\right)  \right)  ^{2}+O\left(  T^{-1}\right)  .
\]
The same argument shows that (\ref{y1y2_T/2_d}) is
\[
\frac{1}{\left(  \sqrt{T/2}\right)  ^{4}}\sum_{t_{1},t_{2},t_{3},t_{4}%
=1}^{T/2}\gamma_{1,2}\left(  t_{1},t_{4}\right)  \gamma_{1,2}\left(
t_{2},t_{3}\right)  =\left(  \sum_{v=-T/2+1}^{T/2-1}\gamma_{1,2}\left(
0,v\right)  \right)  ^{2}+O\left(  T^{-1}\right)
\]
Now consider the difference
\begin{align}
E\left[  \left(  \mathcal{Y}_{1}^{\left(  1\right)  }\right)  ^{2}\left(
\mathcal{Y}_{2}^{\left(  1\right)  }\right)  ^{2}\right]  -E\left[
\mathcal{Y}_{1}^{2}\mathcal{Y}_{2}^{2}\right]   &  =\left(  \sum
_{v=-T/2+1}^{T/2-1}\gamma_{1,1}\left(  0,v\right)  \right)  \left(
\sum_{v=-T/2+1}^{T/2-1}\gamma_{2,2}\left(  0,v\right)  \right)
\label{Diff_Ey1y2}\\
&  +2\left(  \sum_{v=-T/2+1}^{T/2-1}\gamma_{1,2}\left(  0,v\right)  \right)
^{2}\nonumber\\
&  -\left(  \sum_{v=-T+1}^{T-1}\gamma_{1,1}\left(  0,v\right)  \right)
\left(  \sum_{v=-T+1}^{T-1}\gamma_{2,2}\left(  0,v\right)  \right)
\nonumber\\
&  -2\left(  \sum_{v=-T+1}^{T-1}\gamma_{1,2}\left(  0,v\right)  \right)
^{2}+O\left(  T^{-1}\right)  .\nonumber
\end{align}
The order of the difference of the product term in (\ref{Diff_Ey1y2}) then
depends on
\[
\sum_{v=-T/2+1}^{T/2-1}\gamma_{1,1}\left(  0,v\right)  -\sum_{v=-T+1}%
^{T-1}\gamma_{1,1}\left(  0,v\right)  =\sum_{T>\left\vert v\right\vert \geq
T/2}\gamma_{1,1}\left(  0,v\right)
\]
and $\sum_{T>\left\vert v\right\vert =T/2}\gamma_{2,2}\left(  0,v\right)  $.
By the condition in (\ref{Cond_gam_Tail}) it follows that
\[
\sum_{T>\left\vert v\right\vert \geq T/2}\gamma_{1,1}\left(  0,v\right)
=o\left(  T^{-1}\right)  ,\text{ }\sum_{T>\left\vert v\right\vert \geq
T/2}\gamma_{2,2}\left(  0,v\right)  =o\left(  T^{-1}\right)  .
\]
This shows that
\begin{align}
&  \left(  \sum_{v=-T/2+1}^{T/2-1}\gamma_{1,1}\left(  0,v\right)  \right)
\left(  \sum_{v=-T/2+1}^{T/2-1}\gamma_{2,2}\left(  0,v\right)  \right)
\label{Diff_Ey1y2_a}\\
&  -\left(  \sum_{v=-T+1}^{T-1}\gamma_{1,1}\left(  0,v\right)  \right)
\left(  \sum_{v=-T+1}^{T-1}\gamma_{2,2}\left(  0,v\right)  \right)
\nonumber\\
&  =o\left(  T^{-1}\right)  .\nonumber
\end{align}
Similarly, for the quadratic term in (\ref{Diff_Ey1y2}) one obtains by the
condition in (\ref{Cond_gam_Tail}) that
\begin{equation}
\left(  \sum_{v=-T/2+1}^{T/2-1}\gamma_{1,2}\left(  0,v\right)  \right)
^{2}-\left(  \sum_{v=-T+1}^{T-1}\gamma_{1,2}\left(  0,v\right)  \right)
^{2}=o\left(  T^{-1}\right)  .\label{Diff_Ey1y2_b}%
\end{equation}
Substituting (\ref{Diff_Ey1y2_a}) and (\ref{Diff_Ey1y2_b}) in
(\ref{Diff_Ey1y2}) shows that
\[
E\left[  \left(  \mathcal{Y}_{1}^{\left(  1\right)  }\right)  ^{2}\left(
\mathcal{Y}_{2}^{\left(  1\right)  }\right)  ^{2}\right]  -E\left[
\mathcal{Y}_{1}^{2}\mathcal{Y}_{2}^{2}\right]  =o\left(  T^{-1}\right)
+O\left(  T^{-1}\right)  =O\left(  T^{-1}\right)  .
\]
The term $E\left[  \left(  \mathcal{Y}_{1}^{\left(  2\right)  }\right)
^{2}\left(  \mathcal{Y}_{2}^{\left(  2\right)  }\right)  ^{2}\right]  $ in
(\ref{y1y2_T/2_1}) can be analyzed exactly in the same way. We therefore
conclude that
\[
E\left[  \left(  \mathcal{Y}_{1}^{\left(  1\right)  }\right)  ^{2}\left(
\mathcal{Y}_{2}^{\left(  1\right)  }\right)  ^{2}\right]  =E\left[  \left(
\mathcal{Y}_{1}^{\left(  2\right)  }\right)  ^{2}\left(  \mathcal{Y}%
_{2}^{\left(  2\right)  }\right)  ^{2}\right]  =E\left[  \mathcal{Y}_{1}%
^{2}\mathcal{Y}_{2}^{2}\right]  +O\left(  T^{-1}\right)  .
\]
Next turn to (\ref{y1y2_T/2_2}). First consider $E\left[  \mathcal{Y}_{1}%
^{2}\right]  $ and $E\left[  \mathcal{Y}_{2}^{2}\right]  $ where
\begin{align}
E\left[  \mathcal{Y}_{1}^{2}\right]   &  =\frac{1}{\left(  \sqrt{T}\right)
^{2}}\sum_{t_{1},t_{2}=1}^{T}E\left[  Y_{1,t_{1}}Y_{1,t_{2}}\right]
\label{EY1}\\
&  =\sum_{v=-T+1}^{T-1}\gamma_{1,1}\left(  0,v\right)  +O\left(
T^{-1}\right)  \nonumber
\end{align}
by previous arguments. Similarly,
\begin{equation}
E\left[  \mathcal{Y}_{2}^{2}\right]  =\sum_{v=-T+1}^{T-1}\gamma_{2,2}\left(
0,v\right)  +O\left(  T^{-1}\right)  .\label{EY2}%
\end{equation}
Now consider
\begin{align}
E\left[  \left(  \mathcal{Y}_{1}^{\left(  1\right)  }\right)  ^{2}\left(
\mathcal{Y}_{2}^{\left(  2\right)  }\right)  ^{2}\right]   &  =\frac
{1}{\left(  \sqrt{T/2}\right)  ^{4}}\sum_{t_{1},t_{2}=1}^{T/2}\sum
_{t_{3},t_{4}=T/2+1}^{T}E\left[  Y_{1,t_{1}}Y_{1,t_{2}}Y_{2,t_{3}}Y_{2,t_{4}%
}\right]  \label{Ey1y2_T/2_i}\\
&  =\frac{1}{\left(  \sqrt{T/2}\right)  ^{4}}\sum_{t_{1},t_{2}=1}^{T/2}%
\sum_{t_{3},t_{4}=T/2+1}^{T}\gamma_{1,1}\left(  t_{1},t_{2}\right)
\gamma_{2,2}\left(  t_{3},t_{4}\right)  \label{Ey1y2_T/2_ii}\\
&  +\frac{1}{\left(  \sqrt{T/2}\right)  ^{4}}\sum_{t_{1},t_{2}=1}^{T/2}%
\sum_{t_{3},t_{4}=T/2+1}^{T}\gamma_{1,2}\left(  t_{1},t_{3}\right)
\gamma_{1,2}\left(  t_{2},t_{4}\right)  \label{Ey1y2_T/2_iii}\\
&  +\frac{1}{\left(  \sqrt{T/2}\right)  ^{4}}\sum_{t_{1},t_{2}=1}^{T/2}%
\sum_{t_{3},t_{4}=T/2+1}^{T}\gamma_{1,2}\left(  t_{1},t_{4}\right)
\gamma_{1,2}\left(  t_{2},t_{3}\right)  \label{Ey1y2_T/2_iv}\\
&  +\frac{1}{\left(  \sqrt{T/2}\right)  ^{4}}\sum_{t_{1},t_{2}=1}^{T/2}%
\sum_{t_{3},t_{4}=T/2+1}^{T}cum_{1,1,2,2}\left(  t_{1},t_{2},t_{3}%
,t_{4}\right)  \nonumber
\end{align}
where the cumulant term is $O\left(  T^{-1}\right)  $ by Assumption
\ref{Cond_Cum_Sum}. The term (\ref{Ey1y2_T/2_ii}) is
\begin{align}
&  \frac{2}{T}\sum_{t_{1},t_{2}=1}^{T/2}\gamma_{1,1}\left(  t_{1}%
,t_{2}\right)  \frac{2}{T}\sum_{t_{3},t_{4}=T/2+1}^{T}\gamma_{2,2}\left(
t_{3},t_{4}\right)  \label{gam_11-22_T/2}\\
&  =\left(  \sum_{v=-T/2+1}^{T/2-1}\gamma_{1,1}\left(  0,v\right)  \right)
\left(  \sum_{v=-T/2+1}^{T/2-1}\gamma_{2,2}\left(  0,v\right)  \right)
+O\left(  T^{-1}\right)  \nonumber
\end{align}
by (\ref{gam_11_rep}) and (\ref{gam_22_rep}). For the term
(\ref{Ey1y2_T/2_iii}) it follows that
\begin{align*}
&  \frac{1}{\left(  \sqrt{T/2}\right)  ^{4}}\sum_{t_{1},t_{2}=1}^{T/2}%
\sum_{t_{3},t_{4}=T/2+1}^{T}\gamma_{1,2}\left(  t_{1},t_{3}\right)
\gamma_{1,2}\left(  t_{2},t_{4}\right)  \\
&  =\frac{2}{T}\sum_{t_{1}=1}^{T/2}\sum_{t_{3}=T/2+1}^{T}\gamma_{1,2}\left(
t_{1},t_{3}\right)  \frac{2}{T}\sum_{t_{2}=1}^{T/2}\sum_{t_{3}=T/2+1}%
^{T}\gamma_{1,2}\left(  t_{2},t_{4}\right)
\end{align*}
where with $\nu=t_{3}-t_{1}$ and $\gamma_{1,2}\left(  t_{1},t_{3}\right)
=\gamma_{1,2}\left(  0,t_{3}-t_{1}\right)  $ by stationarity%
\[
\left\vert \frac{2}{T}\sum_{t_{1}=1}^{T/2}\sum_{t_{3}=T/2+1}^{T}\gamma
_{1,2}\left(  t_{1},t_{3}\right)  \right\vert \leq\frac{2}{T}\sum_{\nu
=1}^{T/2}\left\vert \nu\right\vert \left\vert \gamma_{1,2}\left(
0,\nu\right)  \right\vert =O\left(  T^{-1}\right)
\]
by Assumption \ref{Cond_Cov_Sum}. This implies that (\ref{Ey1y2_T/2_iii}) and
(\ref{Ey1y2_T/2_iv}) are $O\left(  T^{-1}\right)  .$ It follows that
\begin{equation}
E\left[  \left(  \mathcal{Y}_{1}^{\left(  1\right)  }\right)  ^{2}\left(
\mathcal{Y}_{2}^{\left(  2\right)  }\right)  ^{2}\right]  =O\left(
T^{-1}\right)  .\label{y1y2_T/2_2a}%
\end{equation}
By the same arguments it follows that in (\ref{y1y2_T/2_2})%
\begin{equation}
E\left[  \left(  \mathcal{Y}_{1}^{\left(  2\right)  }\right)  ^{2}\left(
\mathcal{Y}_{2}^{\left(  1\right)  }\right)  ^{2}\right]  =O\left(
T^{-1}\right)  .\label{y1y2_T/2_2b}%
\end{equation}
Combining (\ref{EY1}), (\ref{EY2}), (\ref{gam_11-22_T/2}), (\ref{y1y2_T/2_2a})
and (\ref{y1y2_T/2_2b}) in (\ref{y1y2_T/2_2}) leads to
\begin{align*}
E\left[  \left(  \mathcal{Y}_{1}^{\left(  1\right)  }\right)  ^{2}\left(
\mathcal{Y}_{2}^{\left(  2\right)  }\right)  ^{2}\right]  -E\left[
\mathcal{Y}_{1}^{2}\right]  E\left[  \mathcal{Y}_{2}^{2}\right]   &  =\left(
\sum_{v=-T/2+1}^{T/2-1}\gamma_{1,1}\left(  0,v\right)  \right)  \left(
\sum_{v=-T/2+1}^{T/2-1}\gamma_{2,2}\left(  0,v\right)  \right)  \\
&  -\left(  \sum_{v=-T/2+1}^{T/2-1}\gamma_{1,1}\left(  0,v\right)  \right)
\left(  \sum_{v=-T/2+1}^{T/2-1}\gamma_{2,2}\left(  0,v\right)  \right)  \\
&  +O\left(  T^{-1}\right)  \\
&  =O\left(  T^{-1}\right)
\end{align*}
where the second equality follows from (\ref{Diff_Ey1y2_a}). Finally turn to
(\ref{y1y2_T/2_3}). We have
\begin{align}
E\left[  \mathcal{Y}_{1}\mathcal{Y}_{2}\right]   &  =\frac{1}{\left(  \sqrt
{T}\right)  ^{2}}\sum_{t_{1},t_{2}=1}^{T}E\left[  Y_{1,t_{1}}Y_{2,t_{2}%
}\right]  \nonumber\\
&  =\frac{1}{\left(  \sqrt{T}\right)  ^{2}}\sum_{t_{1},t_{2}=1}^{T}\gamma
_{12}\left(  t_{1},t_{2}\right)  \nonumber\\
&  =\sum_{v=-T+1}^{T-1}\gamma_{1,2}\left(  0,v\right)  +O\left(
T^{-1}\right)  \label{EY1Y2-alt}%
\end{align}
by (\ref{gam_12_T}). Similarly,
\begin{align*}
E\left[  \mathcal{Y}_{1}^{\left(  1\right)  }\mathcal{Y}_{2}^{\left(
1\right)  }\right]   &  =\frac{1}{\left(  \sqrt{T/2}\right)  ^{2}}\sum
_{t_{1},t_{2}=1}^{T/2}E\left[  Y_{1,t_{1}}Y_{2,t_{2}}\right]  \\
&  =\sum_{v=-T/2+1}^{T/2-1}\gamma_{1,2}\left(  0,v\right)  +O\left(
T^{-1}\right)
\end{align*}
by (\ref{gam_12_T/2}). By the same argument
\[
E\left[  \mathcal{Y}_{1}^{\left(  2\right)  }\mathcal{Y}_{2}^{\left(
2\right)  }\right]  =\sum_{v=-T/2+1}^{T/2-1}\gamma_{1,2}\left(  0,v\right)
+O\left(  T^{-1}\right)
\]
It then follows that
\begin{align*}
E\left[  \mathcal{Y}_{1}^{\left(  1\right)  }\mathcal{Y}_{2}^{\left(
1\right)  }\right]  -E\left[  \mathcal{Y}_{1}\mathcal{Y}_{2}\right]   &
=\sum_{T>\left\vert v\right\vert \geq T/2}\gamma_{1,2}\left(  0,v\right)
+O\left(  T^{-1}\right)  \\
&  =O\left(  T^{-1}\right)
\end{align*}
where the last equality follows from the condition in (\ref{Cond_gam_Tail}).

To establish (\ref{y1y2_T/2_4}) consider
\begin{align}
E\left[  \mathcal{Y}_{1}^{\left(  1\right)  }\mathcal{Y}_{2}^{\left(
1\right)  }\mathcal{Y}_{1}^{\left(  2\right)  }\mathcal{Y}_{2}^{\left(
2\right)  }\right]   &  =\frac{1}{\left(  \sqrt{T/2}\right)  ^{4}}\sum
_{t_{1},t_{2}=1}^{T/2}\sum_{t_{3},t_{4}=T/2+1}^{T}E\left[  Y_{1,t_{1}%
}Y_{2,t_{2}}Y_{1,t_{3}}Y_{2,t_{4}}\right]  \label{Ey1y2_T/2_i'}\\
&  =\frac{1}{\left(  \sqrt{T/2}\right)  ^{4}}\sum_{t_{1},t_{2}=1}^{T/2}%
\sum_{t_{3},t_{4}=T/2+1}^{T}\gamma_{1,2}\left(  t_{1},t_{2}\right)
\gamma_{1,2}\left(  t_{3},t_{4}\right)  \label{Ey1y2_T/2_ii'}\\
&  +\frac{1}{\left(  \sqrt{T/2}\right)  ^{4}}\sum_{t_{1},t_{2}=1}^{T/2}%
\sum_{t_{3},t_{4}=T/2+1}^{T}\gamma_{1,1}\left(  t_{1},t_{3}\right)
\gamma_{2,2}\left(  t_{2},t_{4}\right)  \label{Ey1y2_T/2_iii'}\\
&  +\frac{1}{\left(  \sqrt{T/2}\right)  ^{4}}\sum_{t_{1},t_{2}=1}^{T/2}%
\sum_{t_{3},t_{4}=T/2+1}^{T}\gamma_{1,2}\left(  t_{1},t_{4}\right)
\gamma_{2,1}\left(  t_{2},t_{3}\right)  \label{Ey1y2_T/2_iv'}\\
&  +\frac{1}{\left(  \sqrt{T/2}\right)  ^{4}}\sum_{t_{1},t_{2}=1}^{T/2}%
\sum_{t_{3},t_{4}=T/2+1}^{T}cum_{1,2,1,2}\left(  t_{1},t_{2},t_{3}%
,t_{4}\right)  \nonumber
\end{align}
where the cumulant term is $O\left(  T^{-1}\right)  $ by Assumption
\ref{Cond_Cum_Sum}. The term (\ref{Ey1y2_T/2_ii'}) is
\begin{equation}
\frac{2}{T}\sum_{t_{1},t_{2}=1}^{T/2}\gamma_{1,2}\left(  t_{1},t_{2}\right)
\frac{2}{T}\sum_{t_{3},t_{4}=T/2+1}^{T}\gamma_{1,2}\left(  t_{3},t_{4}\right)
=\left(  \sum_{v=-T/2+1}^{T/2-1}\gamma_{1,2}\left(  0,v\right)  \right)
^{2}+O\left(  T^{-1}\right)  \nonumber
\end{equation}
by (\ref{gam_12_T/2}). By (\ref{Diff_Ey1y2_b}), we further have%
\begin{equation}
\frac{2}{T}\sum_{t_{1},t_{2}=1}^{T/2}\gamma_{1,2}\left(  t_{1},t_{2}\right)
\frac{2}{T}\sum_{t_{3},t_{4}=T/2+1}^{T}\gamma_{1,2}\left(  t_{3},t_{4}\right)
=\left(  \sum_{v=-T+1}^{T-1}\gamma_{1,2}\left(  0,v\right)  \right)
^{2}+O\left(  T^{-1}\right)  .\nonumber
\end{equation}
By (\ref{EY1Y2-alt}), we further have%
\[
\frac{2}{T}\sum_{t_{1},t_{2}=1}^{T/2}\gamma_{1,2}\left(  t_{1},t_{2}\right)
\frac{2}{T}\sum_{t_{3},t_{4}=T/2+1}^{T}\gamma_{1,2}\left(  t_{3},t_{4}\right)
=E\left[  \mathcal{Y}_{1}\mathcal{Y}_{2}\right]  ^{2}+O\left(  T^{-1}\right)
.
\]
For the term (\ref{Ey1y2_T/2_iii'}) it follows that
\begin{align*}
&  \frac{1}{\left(  \sqrt{T/2}\right)  ^{4}}\sum_{t_{1},t_{2}=1}^{T/2}%
\sum_{t_{3},t_{4}=T/2+1}^{T}\gamma_{1,1}\left(  t_{1},t_{3}\right)
\gamma_{2,2}\left(  t_{2},t_{4}\right)  \\
&  =\frac{2}{T}\sum_{t_{1}=1}^{T/2}\sum_{t_{3}=T/2+1}^{T}\gamma_{1,1}\left(
t_{1},t_{3}\right)  \frac{2}{T}\sum_{t_{2}=1}^{T/2}\sum_{t_{3}=T/2+1}%
^{T}\gamma_{2,2}\left(  t_{2},t_{4}\right)
\end{align*}
where with $\nu=t_{3}-t_{1}$ and $\gamma_{1,1}\left(  t_{1},t_{3}\right)
=\gamma_{1,1}\left(  0,t_{3}-t_{1}\right)  $ by stationarity%
\[
\left\vert \frac{2}{T}\sum_{t_{1}=1}^{T/2}\sum_{t_{3}=T/2+1}^{T}\gamma
_{1,1}\left(  t_{1},t_{3}\right)  \right\vert \leq\frac{2}{T}\sum_{\nu
=1}^{T/2}\left\vert \nu\right\vert \left\vert \gamma_{1,1}\left(
0,\nu\right)  \right\vert =O\left(  T^{-1}\right)
\]
by Assumption \ref{Cond_Cov_Sum}. Similarly, we have
\[
\left\vert \frac{2}{T}\sum_{t_{2}=1}^{T/2}\sum_{t_{4}=T/2+1}^{T}\gamma
_{2,2}\left(  t_{2},t_{4}\right)  \right\vert \leq\frac{2}{T}\sum_{\nu
=1}^{T-1}\left\vert \nu\right\vert \left\vert \gamma_{2,2}\left(
0,\nu\right)  \right\vert =O\left(  T^{-1}\right)  .
\]
This implies that (\ref{Ey1y2_T/2_iii'}) is $O\left(  T^{-1}\right)  $.
Likewise (\ref{Ey1y2_T/2_iv'}) is $O\left(  T^{-1}\right)  $. It follows that
\[
E\left[  \left(  \mathcal{Y}_{1}^{\left(  1\right)  }\right)  ^{2}\left(
\mathcal{Y}_{2}^{\left(  2\right)  }\right)  ^{2}\right]  =E\left[
\mathcal{Y}_{1}\mathcal{Y}_{2}\right]  ^{2}+O\left(  T^{-1}\right)  .
\]
To establish (\ref{y1y2_T/2_5}) recall that $E\left[  Y_{1,t}\right]
=E\left[  Y_{2,t}\right]  =0$ and note that because $E\left[  Y_{a,t}\right]
=0$ for all $a\in\left\{  1,2\right\}  $ and all $t$ it follows that
\[
E\left[  Y_{a_{1},t_{1}}Y_{a_{2},t_{2}}Y_{a_{3},t_{3}}\right]
=\operatorname*{Cov}\left(  Y_{a_{1},t_{1}}Y_{a_{2},t_{2}},Y_{a_{3},t_{3}%
}\right)  =\operatorname*{Cov}\left(  Y_{a_{1},t_{1}},Y_{a_{2},t_{2}}%
Y_{a_{3},t_{3}}\right)  =\operatorname*{Cov}\left(  Y_{a_{1},t_{1}}%
Y_{a_{3},t_{3}},Y_{a_{2},t_{2}}\right)  .
\]
Then,%
\begin{align}
E\left[  \left(  \mathcal{Y}_{1}^{\left(  1\right)  }\right)  ^{2}%
\mathcal{Y}_{2}^{\left(  2\right)  }\right]   &  =\frac{1}{\left(  \sqrt
{T/2}\right)  ^{3}}\sum_{t_{1},t_{2}=1}^{T/2}\sum_{t_{3}=T/2+1}^{T}E\left[
Y_{1,t_{1}}Y_{1,t_{2}}Y_{2,t_{3}}\right]  \nonumber\\
&  =\frac{1}{\left(  \sqrt{T/2}\right)  ^{3}}\sum_{t_{1},t_{2}=1}^{T/2}%
\sum_{t_{3}=T/2+1}^{T}\operatorname*{Cov}\left(  Y_{1,t_{1}}Y_{1,t_{2}%
},Y_{2,t_{3}}\right)  \nonumber\\
&  =\frac{1}{\left(  \sqrt{T/2}\right)  ^{3}}\sum_{t_{1}=1}^{T/2}\sum
_{t_{3}=T/2+1}^{T}\operatorname*{Cov}\left(  Y_{1,t_{1}}^{2},Y_{2,t_{3}%
}\right)  \label{y1y1y2_T/2_a}\\
&  +\frac{1}{\left(  \sqrt{T/2}\right)  ^{3}}\sum_{t_{1}=1}^{T/2}\sum
_{t_{2}=t_{1}+1}^{T/2}\sum_{t_{3}=T/2+1}^{T}\operatorname*{Cov}\left(
Y_{1,t_{1}}Y_{1,t_{2}},Y_{2,t_{3}}\right)  \label{y1y1y2_T/2_b}\\
&  +\frac{1}{\left(  \sqrt{T/2}\right)  ^{3}}\sum_{t_{2}=1}^{T/2}\sum
_{t_{1}=t_{2}+1}^{T/2}\sum_{t_{3}=T/2+1}^{T}\operatorname*{Cov}\left(
Y_{1,t_{1}}Y_{1,t_{2}},Y_{2,t_{3}}\right)  \label{y1y1y2_T/2_b1}%
\end{align}
where $Y_{1,t_{1}}Y_{1,t_{2}}$ is measurable $\mathcal{M}_{-\infty}%
^{\max\left(  t_{1,}t_{2}\right)  }$ and $Y_{2,t_{3}}$ is measurable
$\mathcal{M}_{t_{3}}^{\infty}.$ Then, for (\ref{y1y1y2_T/2_a})%
\begin{equation}
\left\vert \frac{1}{\left(  \sqrt{T/2}\right)  ^{3}}\sum_{t_{1}=1}^{T/2}%
\sum_{t_{3}=T/2+1}^{T}\operatorname*{Cov}\left(  Y_{1,t_{1}}^{2},Y_{2,t_{3}%
}\right)  \right\vert \leq\frac{8K_{0}^{3/8}}{\left(  \sqrt{T/2}\right)  ^{3}%
}\sum_{t_{1}=1}^{T/2}\sum_{t_{3}=T/2+1}^{T}\left\vert \gamma_{11,2}\left(
t_{1},t_{1},t_{3}\right)  \right\vert \label{y1y1y2_T/2_c}%
\end{equation}
and letting $u=t_{3}-t_{1}$ one obtains
\begin{align}
\sum_{t_{1}=1}^{T/2}\sum_{t_{3}=T/2+1}^{T}\left\vert \gamma_{11,2}\left(
0,0,t_{3}-t_{1}\right)  \right\vert  &  =\sum_{u=1}^{T/2-1}u\left\vert
\gamma_{11,2}\left(  0,0,u\right)  \right\vert +\sum_{u=T/2}^{T-1}\left(
T-u\right)  \left\vert \gamma_{11,2}\left(  0,0,u\right)  \right\vert
\label{sum_alpha}\\
&  \leq\sum_{u=1}^{T-1}u\left\vert \gamma_{11,2}\left(  0,0,u\right)
\right\vert <\infty,\nonumber
\end{align}
where the inequality uses $T-u\leq u$ for $u\geq T/2$ and the fact
that$\left\vert \gamma_{11,2}\left(  0,0,u\right)  \right\vert \geq0$ for all
$u$. Then, $\sum_{u=1}^{T-1}u\left\vert \gamma_{11,2}\left(  0,0,u\right)
\right\vert <\infty$ holds by the condition in (\ref{Cond_3rdCum_Sum1}). This
shows that (\ref{y1y1y2_T/2_c}) is $O\left(  T^{-3/2}\right)  .$ For
(\ref{y1y1y2_T/2_b}) proceed similarly by bounding
\begin{align*}
\frac{2}{\left(  \sqrt{T/2}\right)  ^{3}}&\left\vert \sum_{t_{1}=1}^{T/2}%
\sum_{t_{2}=t_{1}+1}^{T/2}\sum_{t_{3}=T/2+1}^{T}\operatorname*{Cov}\left(
Y_{1,t_{1}}Y_{1,t_{2}},Y_{2,t_{3}}\right)  \right\vert  \\
&  \leq\frac
{8K_{0}^{3/8}}{\left(  \sqrt{T/2}\right)  ^{3}}\sum_{t_{1}=1}^{T/2}\sum
_{t_{2}=t_{1}+1}^{T/2}\sum_{t_{3}=T/2+1}^{T}\left\vert \gamma_{11,2}\left(
t_{1},t_{2},t_{3}\right)  \right\vert \\
&  \leq\frac{8K_{0}^{3/8}}{\left(  \sqrt{T/2}\right)  ^{3}}\sum_{t_{3}%
=T/2+1}^{T}\sum_{t_{2}=2}^{T/2}\sum_{t_{1}=1}^{t_{2}-1}\left\vert
\gamma_{11,2}\left(  t_{1}-t_{2},0,t_{3}-t_{2}\right)  \right\vert \\
&  =\frac{8K_{0}^{3/8}}{\left(  \sqrt{T/2}\right)  ^{3}}\sum_{t_{3}=T/2+1}%
^{T}\sum_{t_{2}=2}^{T/2}\sum_{t=1-t_{2}}^{-1}\left\vert \gamma_{11,2}\left(
t,0,t_{3}-t_{2}\right)  \right\vert \\
&  \leq\frac{8K_{0}^{3/8}}{\left(  \sqrt{T/2}\right)  ^{3}}\sum_{t_{3}%
=T/2+1}^{T}\sum_{t_{2}=2}^{T/2}\frac{T}{2}\sup_{s\leq0}\left\vert
\gamma_{11,2}\left(  s,0,t_{3}-t_{2}\right)  \right\vert \\
&  \leq\frac{8K_{0}^{3/8}}{\sqrt{T/2}}\sum_{u=1}^{T-1}u\sup_{s\leq0}\left\vert
\gamma_{11,2}\left(  s,0,u\right)  \right\vert =O\left(  T^{-1/2}\right)  ,
\end{align*}
because $\sum_{u=1}^{T-1}u\sup_{s<0}\left\vert \gamma_{11,2}\left(
s,0,u\right)  \right\vert \leq\sum_{u=1}^{\infty}u\sup_{s<0}\left\vert
\gamma_{11,2}\left(  s,0,u\right)  \right\vert <\infty$ is bounded by the
condition in (\ref{Cond_3rdCum_Sum2}).  Finally, for (\ref{y1y1y2_T/2_b1}) it
follows that $Y_{1,t_{1}}Y_{1,t_{2}}$ is measurable w.r.t. $\mathcal{M}%
_{-\infty}^{t_{1}}$ and $Y_{2,t_{3}}$ is measurable $\mathcal{M}_{t_{3}%
}^{\infty}$ such that
\begin{align*}
\frac{1}{\left(  \sqrt{T/2}\right)  ^{3}}&\left\vert \sum_{t_{2}=1}^{T/2}%
\sum_{t_{1}=t_{2}+1}^{T/2}\sum_{t_{3}=T/2+1}^{T}\operatorname*{Cov}\left(
Y_{1,t_{1}}Y_{1,t_{2}},Y_{2,t_{3}}\right)  \right\vert \\
&  \leq\frac
{8K_{0}^{3/8}}{\left(  \sqrt{T/2}\right)  ^{3}}\sum_{t_{2}=1}^{T/2}\sum
_{t_{1}=t_{2}+1}^{T/2}\sum_{t_{3}=T/2+1}^{T}\left\vert \gamma_{11,2}\left(
t_{1},t_{2},t_{3}\right)  \right\vert \\
&  \leq\frac{8K_{0}^{3/8}}{\sqrt{T/2}}\sum_{t_{1}=1}^{T/2}\sum_{t_{3}%
=T/2+1}^{T}\sup_{t\leq0}\left\vert \gamma_{11,2}\left(  0,t,t_{3}%
-t_{1}\right)  \right\vert \\
&  \leq\frac{8K_{0}^{3/8}}{\sqrt{T/2}}\sum_{u=1}^{T}u\sup_{t\leq0}\left\vert
\gamma_{11,2}\left(  0,t,u\right)  \right\vert =O\left(  T^{-1/2}\right)  .
\end{align*}
where $\sum_{u=1}^{T}u\sup_{t\leq0}\left\vert \gamma_{11,2}\left(
0,t,u\right)  \right\vert <\infty$ is bounded by the condition in
(\ref{Cond_3rdCum_Sum3}). It then follows that
\[
E\left[  \left(  \mathcal{Y}_{1}^{\left(  1\right)  }\right)  ^{2}%
\mathcal{Y}_{2}^{\left(  2\right)  }\right]  =O\left(  T^{-3/2}\right)
+O\left(  T^{-1/2}\right)  =O\left(  T^{-1/2}\right)  =o\left(  1\right)  .
\]
Finally, to establish (\ref{y1y2_T/2_6}) consider%
\begin{align}
E\left[  \mathcal{Y}_{1}^{\left(  1\right)  }\mathcal{Y}_{1}^{\left(
2\right)  }\mathcal{Y}_{2}^{\left(  2\right)  }\right]   &  =\frac{1}{\left(
\sqrt{T/2}\right)  ^{3}}\sum_{t_{1}=1}^{T/2}\sum_{t_{2},t_{3}=T/2+1}%
^{T}E\left[  Y_{1,t_{1}}Y_{1,t_{2}}Y_{2,t_{3}}\right]  \label{y1y1y2_T/2_d}\\
&  =\frac{1}{\left(  \sqrt{T/2}\right)  ^{3}}\sum_{t_{1}=1}^{T/2}\sum
_{t_{2},t_{3}=T/2+1}^{T}\operatorname*{Cov}\left(  Y_{1,t_{1}},Y_{1,t_{2}%
}Y_{2,t_{3}}\right)  \nonumber\\
&  =\frac{1}{\left(  \sqrt{T/2}\right)  ^{3}}\sum_{t_{1}=1}^{T/2}\sum
_{t_{2}=T/2+1}^{T}\operatorname*{Cov}\left(  Y_{1,t_{1}},Y_{1,t_{2}}%
Y_{2,t_{2}}\right)  \label{y1y1y2_T/2_d1}\\
&  +\frac{1}{\left(  \sqrt{T/2}\right)  ^{3}}\sum_{t_{1}=1}^{T/2}\sum
_{t_{2}=T/2+1}^{T}\sum_{t_{3}=t_{2}+1}^{T}\operatorname*{Cov}\left(
Y_{1,t_{1}},Y_{1,t_{2}}Y_{2,t_{3}}\right)  .\label{y1y1y2_T/2_d2}\\
&  +\frac{1}{\left(  \sqrt{T/2}\right)  ^{3}}\sum_{t_{1}=1}^{T/2}\sum
_{t_{2}=T/2+1}^{T}\sum_{t_{3}=t_{2}+1}^{T}\operatorname*{Cov}\left(
Y_{1,t_{1}},Y_{1,t_{3}}Y_{2,t_{2}}\right)  \label{y1y1y2_T/2_d3}%
\end{align}
Then, for (\ref{y1y1y2_T/2_d1}) we have%
\begin{align}
\frac{1}{\left(  \sqrt{T/2}\right)  ^{3}}&\left\vert \sum_{t_{1}=1}^{T/2}%
\sum_{t_{2}=T/2+1}^{T}\operatorname*{Cov}\left(  Y_{1,t_{1}},Y_{1,t_{2}%
}Y_{2,t_{3}}\right)  \right\vert \\
&  \leq\frac{1}{\left(  \sqrt{T/2}\right)
^{3}}\sum_{t_{1}=1}^{T/2}\sum_{t_{2}=T/2+1}^{T}\left\vert \gamma_{1,12}\left(
0,t_{2}-t_{1},t_{2}-t_{1}\right)  \right\vert \label{y1y1y2_T/2_e}\\
&  =\frac{1}{\left(  \sqrt{T/2}\right)  ^{3}}\sum_{u=1}^{T-1}u\left\vert
\gamma_{1,12}\left(  0,u,u\right)  \right\vert =O\left(  T^{-3/2}\right)
\nonumber
\end{align}
where $\sum_{u=1}^{\infty}u\left\vert \gamma_{1,12}\left(  0,u,u\right)
\right\vert <\infty$ \ by (\ref{Cond_3rdCum_Sum1}) because $\gamma
_{1,12}\left(  0,u,u\right)  =\gamma_{12,1}\left(  u,u,0\right)  $ and by
stationarity $\gamma_{a_{1}a_{2},a_{3}}\left(  t,t,t+u\right)  =\gamma
_{a_{1}a_{2},a_{3}}\left(  -u,-u,0\right)  .$ For (\ref{y1y1y2_T/2_d2}),
consider the second term in (\ref{y1y1y2_T/2_d}) by setting $u=t_{2}-t_{1}$
\begin{align}
&\left\vert \frac{1}{\left(  \sqrt{T/2}\right)  ^{3}}\sum_{t_{1}=1}^{T/2}%
\sum_{t_{2}=T/2+1}^{T}\sum_{t_{3}=t_{2}+1}^{T}\operatorname*{Cov}\left(
Y_{1,t_{1}},Y_{1,t_{2}}Y_{2,t_{3}}\right)  \right\vert \\
&  \leq\frac
{8K_{0}^{3/8}}{\left(  \sqrt{T/2}\right)  ^{3}}\sum_{t_{1}=1}^{T/2}\sum
_{t_{2}=T/2+1}^{T}\sum_{t_{3}=t_{2}+1}^{T}\gamma_{1,12}\left(  0,u,t_{3}%
-t_{1}\right)  \label{y1y1y2_T/2_f}\\
&  \leq\frac{8K_{0}^{3/8}}{\sqrt{T/2}}\sum_{t_{1}=1}^{T/2}\sum_{t_{2}%
=T/2+1}^{T}\sup_{t>0}\left\vert \gamma_{1,12}\left(  0,u,t+u\right)
\right\vert \nonumber\\
&  \leq\frac{8K_{0}^{3/8}}{\sqrt{T/2}}\sum_{u=1}^{T-1}u\sup_{t>0}\left\vert
\gamma_{1,12}\left(  0,u,t+u\right)  \right\vert =O\left(  T^{-1/2}\right)
\nonumber
\end{align}
where $\sum_{u=1}^{T-1}u\sup_{t>0}\left\vert \gamma_{1,12}\left(
0,u,t+u\right)  \right\vert $ is bounded by the condition in
(\ref{Cond_3rdCum_Sum4}). Finally, (\ref{y1y1y2_T/2_d3}) is analyzed exactly
the same way as (\ref{y1y1y2_T/2_d2}) and is also $O\left(  T^{-1/2}\right)
.$ Then, combining (\ref{y1y1y2_T/2_e}) and (\ref{y1y1y2_T/2_f}) and
substituting in (\ref{y1y1y2_T/2_d}) shows that
\[
E\left[  \mathcal{Y}_{1}^{\left(  1\right)  }\mathcal{Y}_{1}^{\left(
2\right)  }\mathcal{Y}_{2}^{\left(  2\right)  }\right]  =O\left(
T^{-3/2}\right)  +O\left(  T^{-1/2}\right)  =o\left(  1\right)  .
\]

\end{proof}




\newpage
\begin{center}
{\LARGE Supplementary Appendix IV: Efficient Bias Correction for Cross-section and Panel Data}
\end{center}

\setcounter{section}{0}

\section{V-statistics}

\subsection{Properties of normalized V-statistics}

Consider a statistic of the form

\begin{align*}
W_{i,T,m} & \equiv\frac{1}{T^{m/2}}\sum_{t_{1}=1}^{T}\sum_{t_{2}=1}^{T}\cdots\sum_{t_{m}=1}^{T}k_{1}\left(x_{i,t_{1}}\right)k_{2}\left(x_{i,t_{2}}\right)\cdots k_{m}\left(x_{i,t_{m}}\right)\\
 & \equiv T^{m/2}\overline{k}_{i,1}\overline{k}_{i,2}\cdots\overline{k}_{i,m}\\
 & \equiv T^{m/2}K_{i,T}
\end{align*}

where $E\left[k_{j}\left(x_{i,t}\right)\right]=0$. We will call the
average
\[
W_{T,(m)}=\frac{1}{n}\sum_{i=1}^{n}W_{i,T,m}
\]

the normalized V-statistic of order $m$.
\begin{condition}
\label{About-k}(i) $n=O\left(T\right)$; (ii) $E\left[k_{j}\left(x_{i,t}\right)\right]=0$;
(iii) $\left|k_{j}\left(x_{i,t}\right)\right|\leq CM\left(x_{i,t}\right)$
such that $\sup_{i}E\left[M\left(x_{i,t}\right)^{12}\right]<\infty$,
where $C$ denotes a generic constant.
\end{condition}
\begin{lemma}
\label{WitEOp(1)}Suppose that Condition \ref{About-k} holds. Then,
for $m=1,2,3,4,5,6$, $E[W_{i,T,m}]=O\left(1\right)$.
\end{lemma}
\begin{proof}
For $m=1$ the result is immediate from $E[\sum_{t}k_{1}(x_{i,t})/\sqrt{T}]=0$.
We next prove the result for $m=6$

\begin{eqnarray*}
E\left[\left(\frac{1}{\sqrt{T}}\sum_{t=1}^{T}k_{j}\left(x_{i,t}\right)\right)^{6}\right] & = & \frac{TE\left[k_{j}\left(x_{i,t}\right)^{6}\right]+15\frac{T\left(T-1\right)}{2}E\left[k_{j}\left(x_{i,t}\right)^{2}\right]E\left[k_{j}\left(x_{i,t}\right)^{4}\right]}{T^{3}}\\
 & \quad & +\frac{20\frac{T\left(T-1\right)}{2}E\left[k_{j}\left(x_{i,t}\right)^{3}\right]^{2}+90\frac{T\left(T-1\right)(T-2)}{6}E\left[k_{j}\left(x_{i,t}\right)^{2}\right]^{3}}{T^{3}}\\
 & = & \frac{1}{T^{2}}E\left[k_{j}\left(x_{i,t}\right)^{6}\right]+\frac{T-1}{T^{2}}\Big(\frac{15}{2}E\left[k_{1}\left(x_{i,t}\right)^{2}\right]E\left[k_{j}\left(x_{i,t}\right)^{4}\right]\\
 &  & +10E\left[k_{j}\left(x_{i,t}\right)^{3}\right]^{2}\Big)+15\frac{(T-1)(T-2)}{T^{2}}E\left[k_{j}\left(x_{i,t}\right)^{2}\right]^{3}\\
 & \leq & \Big(\frac{1}{T^{2}}+\frac{35}{2}\frac{T-1}{T^{2}}+15\frac{(T-1)(T-2)}{T^{2}}\Big)E\left[M\left(x_{i,t}\right)^{6}\right]\\
 & < & \infty
\end{eqnarray*}

Repeated application of Holder's inequality then gives
\begin{align*}
E\Big[\Big|\prod_{j=1}^{6}\Big(\frac{1}{\sqrt{T}}\sum_{t=1}^{T}k_{j}(x_{i,t})\Big)\Big|\Big] & \leq\prod_{j=1}^{6}E\Big[\Big|\Big(\frac{1}{\sqrt{T}}\sum_{t=1}^{T}k_{j}(x_{i,t})\Big)^{6}\Big|\Big]^{1/6}\\
 & =O(1)
\end{align*}

the result for smaller $m$ then follows by Jensen's inequality.
\end{proof}
\begin{lemma}
\label{WitOp(1)}Suppose that Condition \ref{About-k} holds. Then,
for $m=1,2,3,4,5,6$, $W_{i,T,m}=O_{p}\left(1\right)$.
\end{lemma}
\begin{proof}
For $m=6$ Markov's inequality gives
\[
Pr\big(\left(\frac{1}{\sqrt{T}}\sum_{t=1}^{T}k_{j}\left(x_{i,t}\right)\right)^{6}>t\big)\leq\frac{E\big[\left(\frac{1}{\sqrt{T}}\sum_{t=1}^{T}k_{j}\left(x_{i,t}\right)\right)^{6}\big]}{t}
\]

which, with Lemma \ref{WitEOp(1)}, implies $\big(\frac{1}{\sqrt{T}}\sum_{t=1}^{T}k_{j}(x_{i,t})\big)^{6}=O_{p}(1)$.
For other $m<6$ applying Jensen's inequality gives $\big|\big(\frac{1}{\sqrt{T}}\sum_{t=1}^{T}k_{j}(x_{i,t})\big)^{m}\big|=O_{p}(1)$.
\end{proof}
\begin{lemma}
\label{WitVOp(1)}Suppose that Condition \ref{About-k} holds. Then,
for $m=1,\dots,6$, $Var(W_{i,T,m})=O\left(1\right)$.
\end{lemma}
\begin{proof}
We first prove the result for $m=6$. To reduce notation, we write
$K_{i}(t_{1},\dots,t_{m})=k_{1}(x_{i,t_{1}})\cdots k_{m}(x_{i,t_{m}})$.
First, note that for the term
\[
E\Big[\big(K_{i}(s_{1},\dots,s_{m})-E[K_{i}(s_{1},\dots,s_{m})]\big)\big(K_{i}(t_{1},\dots,t_{m})-E[K_{i}(t_{1},\dots,t_{m})]\big)\Big]
\]

to be non-zero, we must have at last two of each time period so that
there are at most $m$ unique time periods in $(s_{1},\dots,s_{m},t_{1},\dots,t_{m})$.
Using this observation, and letting $\sum_{\mathcal{T}_{6}}$ be the
summation over sets of $(s_{1},\dots,s_{m},t_{1},\dots,t_{m})$ with
at most six unique time periods,

\begin{align*}
Var\bigg( & \Big(\frac{1}{\sqrt{T}}\sum_{t=1}^{T}k_{1}(x_{i,t})\Big)\cdots\Big(\frac{1}{\sqrt{T}}\sum_{t=1}^{T}k_{6}(x_{i,t})\Big)\bigg)\\
 & =E\bigg[\bigg(\frac{1}{T^{3}}\sum_{t_{1},\dots,t_{6}}\big(K_{i}(t_{1},\dots,t_{6})-E[K_{i}(t_{1},\dots,t_{6})]\big)\bigg)^{2}\bigg]\\
 & =\frac{1}{T^{6}}\sum_{s_{1}\dots,s_{6}}\sum_{t_{1},\dots,t_{6}}E\Big[\big(K_{i}(s_{1},\dots,s_{6})-E[K_{i}(s_{1},\dots,s_{6})]\big) \times \\
 &\qquad\big(K_{i}(t_{1},\dots,t_{6})-E[K_{i}(t_{1},\dots,t_{6})]\big)\Big]\\
 & =\frac{1}{T^{6}}\sum_{\mathcal{T}_{6}}E\Big[K_{i}(s_{1},\dots,s_{6})K_{i}(t_{1},\dots,t_{6})\Big]-E[K_{i}(s_{1},\dots,s_{6})]E[K_{i}(t_{1},\dots,t_{6})]\\
 & \leq\frac{1}{T^{6}}\sum_{\mathcal{T}_{6}}CE\big[M(x_{it})^{12}\big]+O(1)\\
 & =O(1)
\end{align*}

where the inequality uses Condition \ref{About-k}, and the fact that
(from repeated application of Holder's inequality)
\[
E\Big[K_{i}(t_{1},\dots,t_{6})K_{i}(s_{1},\dots,s_{6})\Big]\leq\prod_{j=1}^{6}(E[|k_{j}(x_{is_{j}})|^{12}]E[|k_{j}(x_{it_{j}})|^{12}])^{1/12}
\]

It follows similarly, that $Var(W_{i,T,m})=O\left(1\right)$ for $m=1,\dots,5$.
\end{proof}

\subsection{Functions of normalized V-statistics}\label{subsec:nomV-def}

Now, consider a statistic of the form 
\[
W_{T,(m_{1},\dots,m_{L})}=W_{\left[1\right],T}W_{\left[2\right],T}\cdots W_{\left[L\right],T}
\]
where 
\begin{eqnarray*}
W_{\left[\ell\right],T} & \equiv & \frac{1}{n}\sum_{i=1}^{n}W_{i,T,m_{l}}
\end{eqnarray*}
Therefore, $W_{T}$ is the product of $m_{1},\cdots,m_{L}$ normalized
V-statistics. We will also call them normalized V-statistic of order
$\sum_{\ell=1}^{L}m_{\ell}$.
\begin{lemma}
\label{Wit_prodV}Suppose that Condition \ref{About-k} holds. Then,
for $\sum_{l=1}^{L}m_{l}\leq6$

\begin{align*}
Var(W_{i_{1},T,m_{1}}\cdots W_{i_{L},T,m_{L}}) & =O(1)
\end{align*}
\end{lemma}
\begin{proof}
Firstly, note that it is sufficient to consider the case $i_{1}\ne i_{2}\ne\dots\ne i_{L}$
since if any two V-statistics share the same $i$ (say $i_{1}=i_{2}$)
we may consider them as a single V-statistic with $m=m_{1}+m_{2}$.
Then,
\begin{align*}
Var\big(W_{i_{1},T,m_{1}}\cdots W_{i_{L},T,m_{L}}\big) & =E[W_{i_{1},T,m_{1}}^{2}]\cdots E[W_{i_{L},T,m_{L}}^{2}]-\big(E[W_{i_{1},T,m_{1}}]\cdots E[W_{i_{L},T,m_{L}}]\big)^{2}\\
 & =O(1)
\end{align*}

from Lemmas \ref{WitEOp(1)} and \ref{WitVOp(1)}.
\end{proof}
\begin{lemma}
\label{WtVOp()}Suppose that Condition \ref{About-k} holds. Then,
for $\bar{m}=\sum_{l=1}^{L}m_{l}=1,2,3,4,5,6$, we have $W_{T,(m_{1},\dots,m_{L})}=O_{p}(1)$,
and $Var(W_{T,(m_{1},\dots,m_{L})})=O(n^{-1})$.
\end{lemma}
\begin{proof}
We drop the subscript $(m_{1},\dots,m_{L})$ to reduce notation, and
prove the lemma for $L=1,2$ and $6$; the remaining values can be
shown in a similar manner. For $L=1$ we have

\begin{align*}
Var\big(W_{T}\big) & =Var\big(\frac{1}{n}\sum_{i=1}^{n}W_{i,T}\big)\\
 & =\frac{1}{n^{2}}\sum_{i=1}^{n}Var\big(W_{iT}\big)\\
 & =O(\frac{1}{n})
\end{align*}

from Lemma \ref{WitVOp(1)}. Next, for $L=2$, using Lemma \ref{Wit_prodV}

\begin{align*}
Var\big(W_{T}\big) & =Var\big(\frac{1}{n^{2}}\sum_{i,j=1}^{n}W_{i,T,m_{1}}W_{j,T,m_{2}}\big)\\
 & =\frac{1}{n^{4}}\Big\{\sum_{i}Var(W_{i,T,m_{1}}W_{i,T,m_{2}})+\sum_{i_{1}\ne i_{2}}Var(W_{i_{1},T,m_{1}}W_{i_{2},T,m_{2}})\\
 & \quad+\sum_{i_{1}\ne i_{2}}Cov(W_{i_{1},T,m_{1}}W_{i_{2},T,m_{2}},W_{i_{2},T,m_{1}}W_{i_{1},T,m_{2}})\\
 & \quad+\sum_{i_{1}\ne i_{2}\ne i_{3}}\Big(Cov(W_{i_{1},T,m_{1}}W_{i_{2},T,m_{2}},W_{i_{1},T,m_{1}}W_{i_{3},T,m_{2}})\\
 &\qquad +Cov(W_{i_{1},T,m_{1}}W_{i_{2},T,m_{2}},W_{i_{3},T,m_{1}}W_{i_{2},T,m_{2}})\Big)\Big\}\\
 & \leq\frac{1}{n^{4}}\sum_{i_{1}\ne i_{2}\ne i_{3}}\Big(\big(Var(W_{i_{1},T,m_{1}}W_{i_{2},T,m_{2}})Var(W_{i_{1},T,m_{1}}W_{i_{3},T,m_{2}})\big)^{1/2}\\
 & \quad+\big(Var(W_{i_{1},T,m_{1}}W_{i_{2},T,m_{2}})Var(W_{i_{3},T,m_{1}}W_{i_{2},T,m_{2}})\big)^{1/2}\Big)+O_{p}(\frac{1}{n^{2}})\\
 & =O_{p}(\frac{1}{n})
\end{align*}

Similarly, for $L=6$ , we write $W_{\{i\}}=W_{i_{1}T,m_{1}}\cdots W_{i_{6},T,m_{6}}$
for the product over the set for the set $\{i\}=\{i_{1},\dots,i_{6}\}$
and begin by noting that $Cov(W_{\{i\}},W_{\{j\}})$ is equal to zero
whenever $\{i\}$ and $\{j\}$ have no elements in common. Then, there
are at most $2m-1=11$ unique $i$ values, and $O(n^{11})$ possible
combinations of $\{i\},\{j\}$ giving non-zero covariances. By applying
Cauchy-Schwarz and Lemma \ref{Wit_prodV}, these covariances are $O(1)$
for any $\{i\},\{j\}$.

\begin{align*}
Var\big(W_{T}\big) & =Var\big(\frac{1}{n^{6}}\sum_{i_{1},\dots i_{6}}W_{\{i_{1},\dots,i_{6}\}}\big)\\
 & =\frac{1}{n^{12}}\sum_{i_{1},\dots i_{6}}\sum_{j_{1},\dots j_{6}}Cov\big(W_{\{i_{1},\dots,i_{6}\}},W_{\{j_{1},\dots,j_{6}\}}\big)\\
 & =\frac{1}{n^{12}}O(n^{11})=O(n^{-1})
\end{align*}

the remaining $L$ are proven similarly. Now, since $Var(W_{T})=o(1)$,
by Chebyshev's inequality we have $W_{T}-E[W_{T}]=o_{p}(1)$. By Lemma
\ref{WitEOp(1)}, $E[W_{T}]=O(1)$ so that we have $W_{T}=O_{p}(1)$
as required.
\end{proof}

\section{Jackknifed normalized V-statistics\label{sec:V-jack}}

Consider the average of delete-$t$ estimators 
\[
\frac{1}{T}\sum_{t=1}^{T}W_{T,(m_{1,\dots,}m_{L})}^{\left(-t\right)}=\frac{1}{T}\sum_{t=1}^{T}W_{\left[1\right],T,\left(m_{1}\right)}^{(-t)}\cdots W_{\left[L\right],T,\left(m_{L}\right)}^{(-t)}
\]
where 
\begin{eqnarray*}
W_{\left[\ell\right],T,(m_{l})}^{(-t)} & \equiv & \frac{1}{n}\sum_{i=1}^{n}\frac{1}{\left(T-1\right)^{m_{\ell}/2}}\sum_{t_{1}\neq t}^{T}\sum_{t_{2}\neq t}^{T}\cdots\sum_{t_{m_{\ell}}\neq t}^{T}k_{\left[\ell,1\right]}\left(x_{i,t_{1}}\right)k_{\left[\ell,2\right]}\left(x_{i,t_{2}}\right)\cdots k_{\left[\ell,m_{\ell}\right]}\left(x_{i,t_{m_{l}}}\right)\\
 & = & \left(T-1\right)^{m_{\ell}/2}\frac{1}{n}\sum_{i=1}^{n}\frac{T\overline{k}_{i,\left[\ell,1\right]}-k_{\left[\ell,1\right]}\left(x_{i,t}\right)}{T-1}\frac{T\overline{k}_{i,\left[\ell,2\right]}-k_{\left[\ell,2\right]}\left(x_{i,t}\right)}{T-1}\\
 &  & \quad\cdots\frac{T\overline{k}_{i,\left[\ell,m_{\ell}\right]}-k_{\left[\ell,m_{\ell}\right]}\left(x_{i,t}\right)}{T-1}\\
 & = & \frac{1}{\left(T-1\right)^{m_{\ell}/2}}\frac{1}{n}\sum_{i=1}^{n}\left(T\overline{k}_{i,\left[\ell,1\right]}-k_{\left[\ell,1\right]}\left(x_{i,t}\right)\right)\left(T\overline{k}_{i,\left[\ell,2\right]}-k_{\left[\ell,2\right]}\left(x_{i,t}\right)\right)\\
 &  & \quad\cdots\left(T\overline{k}_{i,\left[\ell,m_{\ell}\right]}-k_{\left[\ell,m_{\ell}\right]}\left(x_{i,t}\right)\right)
\end{eqnarray*}

The jackknifed V-statistic is
\[
\tilde{W}_{T,(m_{1,\dots,}m_{L})}=TW_{T,(m_{1,\dots,}m_{L})}-(T-1)\left(\frac{T}{T-1}\right)^{\bar{m}/2}\frac{1}{T}\sum_{t=1}^{T}W_{T,,(m_{1,\dots,}m_{L})}^{(-t)}
\]

where $\bar{m}=\sum_{l}m_{l}$.

We begin by characterizing the effect of the jackknife on these statistics.
\begin{lemma}
\label{WiT_jack}Let $\tilde{W}_{i,T,m}=TW_{i,T,m}-(T-1)(\frac{T}{T-1})^{m/2}\frac{1}{T}\sum_{t}W_{i,T,m}^{(-t)}$.
Then we have

For $m=1$
\[
\tilde{W}_{i,T,1}=W_{i,T,1}
\]

For $m=2$
\[
\tilde{W}_{i,T,2}=\frac{1}{T-1}\sum_{s\ne t}k_{1}(x_{i,s})k_{2}(x_{i,t})
\]

For $m=3$
\begin{align*}
\tilde{W}_{i,T,3} & =W_{i,T,3}+\frac{1-2T}{T^{1/2}(T-1)}\frac{1}{T}\sum_{t}k_{1}(x_{i,t})k_{2}(x_{i,t})k_{3}(x_{i,t})\\
 & \quad+\frac{T^{2}+3T-1}{(T-1)^{2}}\frac{1}{T^{3/2}}\sum_{t_{1}\ne t_{2}}\big(k_{1}(x_{i,t_{1}})k_{2}(x_{i,t_{2}})k_{3}(x_{i,t_{2}}) \\
 &\qquad +k_{1}(x_{i,t_{2}})k_{2}(x_{i,t_{1}})k_{3}(x_{i,t_{2}})+k_{1}(x_{i,t_{2}})k_{2}(x_{i,t_{2}})k_{3}(x_{i,t_{!}})\big)\\
 & \quad+\frac{3T-1}{(T-1)^{2}}\frac{1}{T^{3/2}}\sum_{t_{1}\neq t_{2}\ne t_{3}}k_{1}(x_{i,t_{1}})k_{2}(x_{i,t_{2}})k_{3}(x_{i,t_{3}})
\end{align*}

For $m=4$
\begin{align*}
\tilde{W}_{i,T,4} & =W_{i,T,4}-\frac{3T^{2}-3T+1}{T(T-1)^{2}}\frac{1}{T}\sum_{t}k_{1}(x_{i,t})k_{2}(x_{i,t})k_{3}(x_{i,t})k_{4}(x_{i,t})\\
 & \quad-\frac{2T^{3}-6T^{2}+4T-1}{T(T-1)^{3}}\Big\{\frac{1}{T}\sum_{t_{1}\ne t_{2}}\big(k_{1}(x_{i,t_{1}})k_{2}(x_{i,t_{2}})k_{3}(x_{i,t_{2}})k_{4}(x_{i,t_{2}}) \\
 & \quad\quad +k_{1}(x_{i,t_{2}})k_{2}(x_{i,t_{1}})k_{3}(x_{i,t_{2}})k_{4}(x_{i,t_{2}})+k_{1}(x_{i,t_{2}})k_{2}(x_{i,t_{2}})k_{3}(x_{i,t_{1}})k_{4}(x_{i,t_{2}})\\
 & \quad\quad +k_{1}(x_{i,t_{2}})k_{2}(x_{i,t_{2}})k_{3}(x_{i,t_{2}})k_{4}(x_{i,t_{1}})\big)\Big\}\\
 & \quad-\frac{2T^{3}-6T^{2}+4T-1}{T^{2}(T-1)^{3}}\sum_{t_{1}\ne t_{2}}\big(k_{1}(x_{i,t_{1}})k_{2}(x_{i,t_{1}})k_{3}(x_{i,t_{2}})k_{4}(x_{i,t_{2}})\\
 & \quad\quad +k_{1}(x_{i,t_{1}})k_{2}(x_{i,t_{2}})k_{3}(x_{i,t_{1}})k_{4}(x_{i,t_{2}})+k_{1}(x_{i,t_{1}})k_{2}(x_{i,t_{2}})k_{3}(x_{i,t_{2}})k_{4}(x_{i,t_{1}})\big)\\
 & \quad-\frac{T^{3}-6T^{2}+4T-1}{T^{2}(T-1)^{3}}\sum_{t_{1}\neq t_{2}\ne t_{3}}\big(k_{1}(x_{i,t_{1}})k_{2}(x_{i,t_{1}})k_{3}(x_{i,t_{2}})k_{4}(x_{i,t_{3}})\\
 & \quad\quad +k_{1}(x_{i,t_{1}})k_{2}(x_{i,t_{2}})k_{3}(x_{i,t_{1}})k_{4}(x_{i,t_{3}})+k_{1}(x_{i,t_{1}})k_{2}(x_{i,t_{2}})k_{3}(x_{i,t_{3}})k_{4}(x_{i,t_{1}})\\
 & \quad\quad +k_{1}(x_{i,t_{2}})k_{2}(x_{i,t_{1}})k_{3}(x_{i,t_{1}})k_{4}(x_{i,t_{3}})+k_{1}(x_{i,t_{2}})k_{2}(x_{i,t_{1}})k_{3}(x_{i,t_{3}})k_{4}(x_{i,t_{1}})\\
 & \quad\quad +k_{1}(x_{i,t_{2}})k_{2}(x_{i,t_{3}})k_{3}(x_{i,t_{1}})k_{4}(x_{i,t_{1}})\big)\\
 & \quad+\frac{6T^{2}-4T+1}{(T-1)^{3}}\frac{1}{T^{2}}\sum_{t_{1}\neq t_{2}\ne t_{3}\ne t_{4}}k_{1}(x_{i,t_{1}})k_{2}(x_{i,t_{1}})k_{3}(x_{i,t_{2}})k_{4}(x_{i,t_{4}})\Big\}
\end{align*}
\end{lemma}
\begin{proof}
For $m=1$
\begin{align*}
(T-1)(\frac{T}{T-1})^{1/2}\frac{1}{T}\sum_{t}W_{i,T,1}^{(-t)} & =(T-1)\frac{1}{\sqrt{T}}\sum_{t}\frac{1}{T-1}(T\bar{k}_{1,i}-k_{1}(x_{i,t}))\\
 & =T\sqrt{T}\bar{k}_{1,i}-\sqrt{T}\bar{k}_{1,i}\\
 & =(T-1)W_{i,T,1}
\end{align*}

Next, for $m=2$

\begin{align*}
\tilde{W}_{i,T,2} & =T\big(T\bar{k}_{1,i}\bar{k}_{2,i}\big)-(T-1)\frac{T}{T-1}\frac{1}{T}\sum_{t}\frac{1}{T-1}(\bar{k}_{1,i}-k_{1}(x_{i,t}))(\bar{k}_{2,i}-k_{2}(x_{i,t}))\\
 & =T^{2}\bar{k}_{1,i}\bar{k}_{2,i}-\frac{1}{T-1}\sum_{t}(\bar{k}_{1,i}-k_{1}(x_{i,t}))(\bar{k}_{2,i}-k_{2}(x_{i,t}))\\
 & =\frac{T^{2}}{T-1}\bar{k}_{1,i}\bar{k}_{2,i}-\frac{1}{T-1}\sum_{t}k_{1}(x_{i,t})k_{2}(x_{i,t})\\
 & =\frac{1}{T-1}\sum_{s\ne t}k_{1}(x_{i,s})k_{2}(x_{i,t})
\end{align*}

For $m=3$

\begin{align*}
\frac{1}{T^{3/2}}\tilde{W}_{i,T,3} & =\frac{1}{T^{2}}\sum_{t_{1}}\sum_{t_{2}}\sum_{t_{3}}k_{1}(x_{i,t_{1}})k_{2}(x_{i,t_{2}})k_{3}(x_{i,t_{3}}) \\
&-\frac{1}{(T-1)^{2}}\frac{1}{T}\Big\{(T-1)\sum_{t}k_{1}(x_{i,t})k_{2}(x_{i,t})k_{3}(x_{i,t})\\
 & \quad+(T-2)\sum_{t_{1}\ne t_{2}}\big(k_{1}(x_{i,t_{1}})k_{2}(x_{i,t_{2}})k_{3}(x_{i,t_{2}})+k_{1}(x_{i,t_{2}})k_{2}(x_{i,t_{1}})k_{3}(x_{i,t_{2}})\\
 & \quad\quad +k_{1}(x_{i,t_{2}})k_{2}(x_{i,t_{2}})k_{3}(x_{i,t_{1}})\big)\\
 & \quad+(T-3)\sum_{t_{1}\neq t_{2}\ne t_{3}}k_{1}(x_{i,t_{1}})k_{2}(x_{i,t_{2}})k_{3}(x_{i,t_{3}})\Big\}\\
 & =\frac{1}{T^{3}}\sum_{t_{1}}\sum_{t_{2}}\sum_{t_{3}}k_{1}(x_{i,t_{1}})k_{2}(x_{i,t_{2}})k_{3}(x_{i,t_{3}}) \\
 &+\big(\frac{T-1}{T^{3}}-\frac{1}{T(T-1)}\big)\sum_{t}k_{1}(x_{i,t})k_{2}(x_{i,t})k_{3}(x_{i,t})\\
 & \quad+\big(\frac{T-1}{T^{3}}-\frac{T-2}{T(T-1)^{2}}\big)\sum_{t_{1}\ne t_{2}}\big(k_{1}(x_{i,t_{1}})k_{2}(x_{i,t_{2}})k_{3}(x_{i,t_{2}})\\
 & \quad+k_{1}(x_{i,t_{2}})k_{2}(x_{i,t_{1}})k_{3}(x_{i,t_{2}})+k_{1}(x_{i,t_{2}})k_{2}(x_{i,t_{2}})k_{3}(x_{i,t_{1}})\big)\\
 & \quad+\big(\frac{T-1}{T^{3}}-\frac{T-3}{T(T-1)^{2}}\big)\sum_{t_{1}\neq t_{2}\ne t_{3}}k_{1}(x_{i,t_{1}})k_{2}(x_{i,t_{2}})k_{3}(x_{i,t_{3}})\\
 & =\frac{1}{T^{3/2}}W_{T}+\frac{1-2T}{T^{2}(T-1)}\frac{1}{T}\sum_{t}k_{1}(x_{i,t})k_{2}(x_{i,t})k_{3}(x_{i,t})\\
 & \quad+\frac{T^{2}+3T-1}{T^{3/2}(T-1)^{2}}\frac{1}{T^{3/2}}\sum_{t_{1}\ne t_{2}}\big(k_{1}(x_{i,t_{1}})k_{2}(x_{i,t_{2}})k_{3}(x_{i,t_{2}})+k_{1}(x_{i,t_{2}})k_{2}(x_{i,t_{1}})k_{3}(x_{i,t_{2}})\\
 & \quad+k_{1}(x_{i,t_{2}})k_{2}(x_{i,t_{2}})k_{3}(x_{i,t_{1}})\big)\\
 & \quad+\frac{3T-1}{T^{3/2}(T-1)^{2}}\frac{1}{T^{3/2}}\sum_{t_{1}\neq t_{2}\ne t_{3}}k_{1}(x_{i,t_{1}})k_{2}(x_{i,t_{2}})k_{3}(x_{i,t_{3}})
\end{align*}

For $m=4$

\begin{align*}
&\frac{1}{T^{2}}\tilde{W}_{i,T,4} \\
& =\frac{1}{T^{3}}\sum_{t_{1}}\sum_{t_{2}}\sum_{t_{3}}\sum_{t_{4}}k_{1}(x_{i,t_{1}})k_{2}(x_{i,t_{2}})k_{3}(x_{i,t_{3}})k_{4}(x_{i,t_{4}})\\
 & \quad-\frac{1}{(T-1)^{3}}\frac{1}{T}\Big\{(T-1)\sum_{t}k_{1}(x_{i,t})k_{2}(x_{i,t})k_{3}(x_{i,t})k_{4}(x_{i,t})\\
 & \quad+(T-2)\sum_{t_{1}\ne t_{2}}\big(k_{1}(x_{i,t_{1}})k_{2}(x_{i,t_{2}})k_{3}(x_{i,t_{2}})k_{4}(x_{i,t_{2}})+k_{1}(x_{i,t_{2}})k_{2}(x_{i,t_{1}})k_{3}(x_{i,t_{2}})k_{4}(x_{i,t_{2}})\\
 & \quad\quad+k_{1}(x_{i,t_{2}})k_{2}(x_{i,t_{2}})k_{3}(x_{i,t_{1}})k_{4}(x_{i,t_{2}})+k_{1}(x_{i,t_{2}})k_{2}(x_{i,t_{2}})k_{3}(x_{i,t_{2}})k_{4}(x_{i,t_{1}})\big)\\
 & \quad+(T-2)\sum_{t_{1}\ne t_{2}}\big(k_{1}(x_{i,t_{1}})k_{2}(x_{i,t_{1}})k_{3}(x_{i,t_{2}})k_{4}(x_{i,t_{2}})+k_{1}(x_{i,t_{1}})k_{2}(x_{i,t_{2}})k_{3}(x_{i,t_{1}})k_{4}(x_{i,t_{2}})\\
 & \quad\quad+k_{1}(x_{i,t_{1}})k_{2}(x_{i,t_{2}})k_{3}(x_{i,t_{2}})k_{4}(x_{i,t_{1}})\big)\\
 & \quad+(T-3)\sum_{t_{1}\neq t_{2}\ne t_{3}}\big(k_{1}(x_{i,t_{1}})k_{2}(x_{i,t_{1}})k_{3}(x_{i,t_{2}})k_{4}(x_{i,t_{3}})+k_{1}(x_{i,t_{1}})k_{2}(x_{i,t_{2}})k_{3}(x_{i,t_{1}})k_{4}(x_{i,t_{3}})\\
 & \quad\quad+k_{1}(x_{i,t_{1}})k_{2}(x_{i,t_{2}})k_{3}(x_{i,t_{3}})k_{4}(x_{i,t_{1}})+k_{1}(x_{i,t_{2}})k_{2}(x_{i,t_{1}})k_{3}(x_{i,t_{1}})k_{4}(x_{i,t_{3}})\\
 & \quad\quad+k_{1}(x_{i,t_{2}})k_{2}(x_{i,t_{1}})k_{3}(x_{i,t_{3}})k_{4}(x_{i,t_{1}})+k_{1}(x_{i,t_{2}})k_{2}(x_{i,t_{3}})k_{3}(x_{i,t_{1}})k_{4}(x_{i,t_{1}})\big)\\
 & +(T-4)\sum_{t_{1}\neq t_{2}\ne t_{3}\ne t_{4}}k_{1}(x_{i,t_{1}})k_{2}(x_{i,t_{1}})k_{3}(x_{i,t_{2}})k_{4}(x_{i,t_{4}})\Big\}\\
 & =\frac{1}{T^{4}}\sum_{t_{1}}\sum_{t_{2}}\sum_{t_{3}}\sum_{t_{4}}k_{1}(x_{i,t_{1}})k_{2}(x_{i,t_{2}})k_{3}(x_{i,t_{3}})k_{4}(x_{i,t_{4}})\\
 & \quad+\big(\frac{T-1}{T^{4}}-\frac{1}{T(T-1)^{2}}\big)\sum_{t}k_{1}(x_{i,t})k_{2}(x_{i,t})k_{3}(x_{i,t})k_{4}(x_{i,t})\\
 & \quad+\big(\frac{T-1}{T^{4}}-\frac{T-2}{T(T-1)^{3}}\big)\Big\{\sum_{t_{1}\ne t_{2}}\big(k_{1}(x_{i,t_{1}})k_{2}(x_{i,t_{2}})k_{3}(x_{i,t_{2}})k_{4}(x_{i,t_{2}})\\
 & \quad\quad +k_{1}(x_{i,t_{2}})k_{2}(x_{i,t_{1}})k_{3}(x_{i,t_{2}})k_{4}(x_{i,t_{2}})+k_{1}(x_{i,t_{2}})k_{2}(x_{i,t_{2}})k_{3}(x_{i,t_{1}})k_{4}(x_{i,t_{2}})\\
 & \quad\quad +k_{1}(x_{i,t_{2}})k_{2}(x_{i,t_{2}})k_{3}(x_{i,t_{2}})k_{4}(x_{i,t_{1}})\big)\\
 & \quad+\sum_{t_{1}\ne t_{2}}\big(k_{1}(x_{i,t_{1}})k_{2}(x_{i,t_{1}})k_{3}(x_{i,t_{2}})k_{4}(x_{i,t_{2}})+k_{1}(x_{i,t_{1}})k_{2}(x_{i,t_{2}})k_{3}(x_{i,t_{1}})k_{4}(x_{i,t_{2}})\\
 & \quad\quad+k_{1}(x_{i,t_{1}})k_{2}(x_{i,t_{2}})k_{3}(x_{i,t_{2}})k_{4}(x_{i,t_{1}})
 \end{align*}
 \begin{align*}
 & \quad+\big(\frac{T-1}{T^{4}}-\frac{T-3}{T(T-1)^{3}}\big)\sum_{t_{1}\neq t_{2}\ne t_{3}}\big(k_{1}(x_{i,t_{1}})k_{2}(x_{i,t_{1}})k_{3}(x_{i,t_{2}})k_{4}(x_{i,t_{3}})\\
 & \quad\quad +k_{1}(x_{i,t_{1}})k_{2}(x_{i,t_{2}})k_{3}(x_{i,t_{1}})k_{4}(x_{i,t_{3}})+k_{1}(x_{i,t_{1}})k_{2}(x_{i,t_{2}})k_{3}(x_{i,t_{3}})k_{4}(x_{i,t_{1}})\\
 & \quad\quad +k_{1}(x_{i,t_{2}})k_{2}(x_{i,t_{1}})k_{3}(x_{i,t_{1}})k_{4}(x_{i,t_{3}})+k_{1}(x_{i,t_{2}})k_{2}(x_{i,t_{1}})k_{3}(x_{i,t_{3}})k_{4}(x_{i,t_{1}})\\
 & \quad\quad +k_{1}(x_{i,t_{2}})k_{2}(x_{i,t_{3}})k_{3}(x_{i,t_{1}})k_{4}(x_{i,t_{1}})\big)\\
 & +\big(\frac{T-1}{T^{4}}-\frac{T-4}{T(T-1)^{3}}\big)\sum_{t_{1}\neq t_{2}\ne t_{3}\ne t_{4}}k_{1}(x_{i,t_{1}})k_{2}(x_{i,t_{1}})k_{3}(x_{i,t_{2}})k_{4}(x_{i,t_{4}})\Big\}\\
 & =\frac{1}{T^{2}}W_{iT}-\frac{3T^{2}-3T+1}{T^{4}(T-1)^{2}}\sum_{t}k_{1}(x_{i,t})k_{2}(x_{i,t})k_{3}(x_{i,t})k_{4}(x_{i,t})\\
 & \quad-\frac{2T^{3}-6T^{2}+4T-1}{T^{4}(T-1)^{3}}\Big\{\sum_{t_{1}\ne t_{2}}\big(k_{1}(x_{i,t_{1}})k_{2}(x_{i,t_{2}})k_{3}(x_{i,t_{2}})k_{4}(x_{i,t_{2}})\\
 & \quad\quad +k_{1}(x_{i,t_{2}})k_{2}(x_{i,t_{1}})k_{3}(x_{i,t_{2}})k_{4}(x_{i,t_{2}})\\
 & \quad\quad+k_{1}(x_{i,t_{2}})k_{2}(x_{i,t_{2}})k_{3}(x_{i,t_{1}})k_{4}(x_{i,t_{2}})+k_{1}(x_{i,t_{2}})k_{2}(x_{i,t_{2}})k_{3}(x_{i,t_{2}})k_{4}(x_{i,t_{1}})\big)\\
 & \quad+\sum_{t_{1}\ne t_{2}}\big(k_{1}(x_{i,t_{1}})k_{2}(x_{i,t_{1}})k_{3}(x_{i,t_{2}})k_{4}(x_{i,t_{2}})+k_{1}(x_{i,t_{1}})k_{2}(x_{i,t_{2}})k_{3}(x_{i,t_{1}})k_{4}(x_{i,t_{2}})\\
 & \quad\quad+k_{1}(x_{i,t_{1}})k_{2}(x_{i,t_{2}})k_{3}(x_{i,t_{2}})k_{4}(x_{i,t_{1}})\Big\}\\
 & \quad-\frac{T^{3}-6T^{2}+4T-1}{T^{4}(T-1)^{3}}\sum_{t_{1}\neq t_{2}\ne t_{3}}\big(k_{1}(x_{i,t_{1}})k_{2}(x_{i,t_{1}})k_{3}(x_{i,t_{2}})k_{4}(x_{i,t_{3}})\\
 & \quad\quad +k_{1}(x_{i,t_{1}})k_{2}(x_{i,t_{2}})k_{3}(x_{i,t_{1}})k_{4}(x_{i,t_{3}})\\
 & \quad\quad+k_{1}(x_{i,t_{1}})k_{2}(x_{i,t_{2}})k_{3}(x_{i,t_{3}})k_{4}(x_{i,t_{1}})+k_{1}(x_{i,t_{2}})k_{2}(x_{i,t_{1}})k_{3}(x_{i,t_{1}})k_{4}(x_{i,t_{3}})\\
 & \quad\quad+k_{1}(x_{i,t_{2}})k_{2}(x_{i,t_{1}})k_{3}(x_{i,t_{3}})k_{4}(x_{i,t_{1}})+k_{1}(x_{i,t_{2}})k_{2}(x_{i,t_{3}})k_{3}(x_{i,t_{1}})k_{4}(x_{i,t_{1}})\big)\\
 & +\frac{6T^{2}-4T+1}{T^{4}(T-1)^{3}}\sum_{t_{1}\neq t_{2}\ne t_{3}\ne t_{4}}k_{1}(x_{i,t_{1}})k_{2}(x_{i,t_{1}})k_{3}(x_{i,t_{2}})k_{4}(x_{i,t_{4}})\Big\}
\end{align*}
\end{proof}
The following Lemma is used to show that $\tilde{W}_{i,T}=O_{p}(1)$
for $\sum_{l}m_{l}\leq6$.
\begin{lemma}
\label{lem:WitJackOp(1)}Suppose that Condition \ref{About-k} holds.
Then for $1\leq p\leq5$ and $p\leq m\leq6$, 
\[
\big(T^{1/2}\bar{k}_{i_{1},1}\big)\cdots\big(T^{1/2}\bar{k}_{i_{p},p}\big)\frac{1}{T}\sum_{t}k_{p+1}(x_{i_{p+1},t})\cdots k_{m}(x_{i_{m},t})=O_{p}(1)
\]
\end{lemma}
\begin{proof}
Firstly,
\begin{align*}
E\Big[ & \Big|\big(T^{1/2}\bar{k}_{i_{1},1}\big)\cdots\big(T^{1/2}\bar{k}_{i_{p},p}\big)\frac{1}{T}\sum_{t}k_{p+1}(x_{i_{p+1},t})\cdots k_{m}(x_{i_{m},t})\Big|\Big]\\
 & \leq\frac{1}{T}\sum_{t}E\big[\big(\big|T^{1/2}\bar{k}_{i_{1},1}\big|\cdots\big|T^{1/2}\bar{k}_{i_{p},p}\big|\big)^{2}\big]^{1/2}E\big[M(x_{i_{p+1},t})^{2}\cdots M(x_{i_{m},t})^{2}\big]^{1/2}\\
 & \leq\frac{1}{T}\sum_{t}E\big[\big|T^{1/2}\bar{k}_{i_{1},1}\big|^{2p}\big]^{1/2p}\cdots E\big[\big|T^{1/2}\bar{k}_{i_{p},p}\big|^{2p}\big]^{1/2p}\big(\max_{i}E\big[M(x_{i,t})^{2(m-p)}\big]\big)^{1/2}\\
 & =O(1)
\end{align*}

where the final result follows since (1) $p\leq5$ implies by Lemma
\ref{WitEOp(1)} that $E\big[\big|T^{1/2}\bar{k}_{i_{1},1}\big|^{2p}\big]=O_{p}(1)$,
and (2) $m-p\leq6$ implies by Condition \ref{About-k} that $\max_{i}E\big[M(x_{i,t})^{2(m-p)}\big]<\infty$.
Markov's inequality then gives the result.
\end{proof}
\begin{lemma}
\label{lem:jack:all}

For all $\tilde{W}_{T,(m_{1},\dots,m_{L})}$, with $\sum_{l=1}^{L}m_{l}\leq6$,
we have
\[
\tilde{W}_{T,(m_{1},\dots,m_{L})}=O_{p}(1)
\]
\end{lemma}
\begin{proof}
Here we show the result for $L=1$ and $m=6$ to simplify notation.
The same steps apply for any $L\leq6$ with $\sum_{l=1}^{L}m_{l}\leq6$
to give the results stated in the lemma. Firstly, $\tilde{W}_{T,(6)}=\frac{1}{n}\sum_{i}\tilde{W}_{i,T,6}$,
where

\begin{eqnarray*}
\tilde{W}_{i,T,6} & = & TW_{i,T,6}-(T-1)\frac{1}{T}\sum_{t=1}^{T}W_{i,T,(6)}^{(-t)}\\
 & = & TW_{i,T,6}-(T-1)\frac{1}{T}\sum_{t=1}^{T}\frac{1}{\left(T-1\right)^{3}}\left(T\overline{k}_{i,1}-k_{1}\left(x_{i,t}\right)\right)\\
 &  & \quad\left(T\overline{k}_{i,2}-k_{2}\left(x_{i,t}\right)\right)\cdots\left(T\overline{k}_{i,6}-k_{6}\left(x_{i,t}\right)\right)\\
 & = & TW_{i,T,6}-\frac{1}{\left(T-1\right)^{2}}\bigg[T^{6}\overline{k}_{i,1}\cdots\overline{k}_{i,6}-T^{5}\overline{k}_{i,1}\cdots\overline{k}_{i,6}\\
 &  & +T^{4}\big\{\frac{1}{T}\sum_{t=1}^{T}k_{1}\left(x_{i,t}\right)k_{2}\left(x_{i,t}\right)\overline{k}_{i,3}\cdots\overline{k}_{i,6}+\frac{1}{T}\sum_{t=1}^{T}k_{5}\left(x_{i,t}\right)k_{6}\left(x_{i,t}\right)\overline{k}_{i,1}\cdots\overline{k}_{i,4}\Big\}\\
 &  & -T^{3}\big\{\frac{1}{T}\sum_{t=1}^{T}k_{1}\left(x_{i,t}\right)\dots k_{3}\left(x_{i,t}\right)\overline{k}_{i,4}\cdots\overline{k}_{i,6}+\frac{1}{T}\sum_{t=1}^{T}k_{4}\left(x_{i,t}\right)\cdots k_{6}\left(x_{i,t}\right)\overline{k}_{i,1}\cdots\overline{k}_{i,3}\Big\}\\
 &  & +T^{2}\big\{\frac{1}{T}\sum_{t=1}^{T}k_{1}\left(x_{i,t}\right)\dots k_{4}\left(x_{i,t}\right)\overline{k}_{i,5}\overline{k}_{i,6}+\frac{1}{T}\sum_{t=1}^{T}k_{3}\left(x_{i,t}\right)\cdots k_{6}\left(x_{i,t}\right)\overline{k}_{i,1}\overline{k}_{i,2}\Big\}\\
 &  & -T\big\{\frac{1}{T}\sum_{t=1}^{T}k_{1}\left(x_{i,t}\right)\dots k_{5}\left(x_{i,t}\right)\overline{k}_{i,6}+\frac{1}{T}\sum_{t=1}^{T}k_{2}\left(x_{i,t}\right)\cdots k_{6}\left(x_{i,t}\right)\overline{k}_{i,1}\Big\}\\
 &  & +\frac{1}{T}\sum_{t=1}^{T}k_{1}\left(x_{i,t}\right)\dots k_{6}\left(x_{i,t}\right)\bigg]
\end{eqnarray*}

Now, using the result in Lemma \ref{lem:WitJackOp(1)}, we can write
this as
\begin{align*}
\tilde{W}_{i,T,6} & =TW_{i,T,6}-\frac{T^{2}}{(T-1)^{2}}\big(TW_{i,T,6}-W_{i,T,6}\big)\\
 & \quad-\frac{1}{(T-1)^{2}}\Big[T^{2}\times O_{p}(1)-T^{3/2}\times O_{p}(1)\\
 & \quad+T\times O_{p}(1)-\sqrt{T}\times O_{p}(1)+O_{p}(1)\Big]\\
 & =-\frac{T}{T-1}W_{i,T,6}+O_{p}(1)\\
 & =O_{p}(1)
\end{align*}

Then, $\tilde{W}_{T,(6)}=\frac{1}{n}\sum_{i}\tilde{W}_{i,T,6}=O_{p}(1)$
also. Similarly, other forms of the V-statistic $\tilde{W}_{T,(m_{1},\dots,m_{L})}$
are simply cross-sectional averages over such bounded (in probability)
objects.
\end{proof}

\section{Split-sample V-statistics\label{sec:V-split}}

The following lemmas contain results analogous to those in the previous
section, but for the split-sample jackknife. Here we consider the
half-sample version and, for simplicity and to avoid notational complication,
assume $T$ is divisible by 2. Let $\tau_{1}=\{1,\dots,T/2\}$ and
$\tau_{2}=\{T/2+1,\dots,T\}$ denote the sets of time periods in the
first and second halves of the sample.

Define the half-sample jackknife statistic $\check{W}_{T,(m_{1},\dots,m_{L})}$
as
\[
\check{W}_{T,(m_{1},\dots,m_{L})}=2W_{T(m_{1},\dots,m_{L})}-2^{\bar{m}/2}\frac{1}{2}(W_{T,(m_{1},\dots,m_{L})}^{(1)}+W_{T,(m_{1},\dots,m_{L})}^{(2)})
\]

where $W_{T,(m_{1},\dots,m_{L})}^{(1)}$ and $W_{T,(m_{1},\dots,m_{L})}^{(2)}$
are the equivalent terms using only observations from time periods
in $\tau_{1}$ and $\tau_{2}$ respectively.
\begin{lemma}
\label{lem:Wit_half}Let $\check{W}_{T,(m_{1},\dots,m_{L})}=2W_{T(m_{1},\dots,m_{L})}-2^{\bar{m}/2}\frac{1}{2}(W_{T,(m_{1},\dots,m_{L})}^{(1)}+W_{T,(m_{1},\dots,m_{L})}^{(2)})$.
Then

for $m=1$

\[
\check{W}_{i,T,1}=W_{i,T,1}
\]

for $m=2$ 
\[
\check{W}_{i,T,2}=2\Big(\frac{1}{T}\sum_{t_{1}\in\tau_{1}}\sum_{t_{2}\in\tau_{2}}k_{1}(x_{it_{1}})k_{2}(x_{it_{2}})+\frac{1}{T}\sum_{t_{1}\in\tau_{2}}\sum_{t_{2}\in\tau_{1}}k_{1}(x_{it_{1}})k_{2}(x_{it_{2}})\Big)
\]

for $m=3$
\begin{align*}
\check{W}_{i,T,3} & =2W_{i,T,3}\\
 & -4\Big(\frac{1}{T^{3/2}}\sum_{t_{1},t_{2},t_{3}\in\tau_{1}}k_{1}(x_{it_{1}})k_{2}(x_{it_{2}})k_{3}(x_{it_{3}})+\frac{1}{T^{3/2}}\sum_{t_{1},t_{2},t_{3}\in\tau_{2}}k_{1}(x_{it_{1}})k_{2}(x_{it_{2}})k_{3}(x_{it_{3}})\Big)
\end{align*}

for $m=4$
\begin{align*}
\check{W}_{i,T,4} & =2W_{i,T,4}\\
 & -8\Big(\frac{1}{T^{2}}\sum_{t_{1},t_{2},t_{3},t_{4}\in\tau_{1}}k_{1}(x_{it_{1}})k_{2}(x_{it_{2}})k_{3}(x_{it_{3}})k_{4}(x_{it_{4}})\\
 &+\frac{1}{T^{2}}\sum_{t_{1},t_{2},t_{3},t_{4}\in\tau_{2}}k_{1}(x_{it_{1}})k_{2}(x_{it_{2}})k_{3}(x_{it_{3}})k_{4}(x_{it_{4}})\Big)
\end{align*}
\end{lemma}
\begin{proof}
For $m=1$ we have
\begin{align*}
\check{W}_{i,T,1} & =2\frac{1}{\sqrt{T}}\sum_{t}k_{1}(x_{it})-\frac{1}{\sqrt{2}}\Big(\sqrt{\frac{2}{T}}\sum_{\tau_{1}}k_{1}(x_{it})+\sqrt{\frac{2}{T}}\sum_{\tau_{2}}k_{1}(x_{it})\Big)\\
 & =\frac{1}{\sqrt{T}}\sum_{t}k_{1}(x_{it})
\end{align*}

For $m=2$ we have
\begin{align*}
\check{W}_{i,T,2} & =2\frac{1}{T}\sum_{t_{1}}\sum_{t_{2}}k_{1}(x_{it_{1}})k_{2}(x_{it_{2}})\\
&-\Big(\frac{2}{T}\sum_{t_{1,}t_{2}\in\tau_{1}}k_{1}(x_{it_{1}})k_{2}(x_{it_{2}})+\frac{2}{T}\sum_{t_{1,}t_{2}\in\tau_{2}}k_{1}(x_{it_{1}})k_{2}(x_{it_{2}})\Big)\\
 & =2\Big(\frac{1}{T}\sum_{t_{1}\in\tau_{1}}\sum_{t_{2}\in\tau_{2}}k_{1}(x_{it_{1}})k_{2}(x_{it_{2}})+\frac{1}{T}\sum_{t_{1}\in\tau_{2}}\sum_{t_{2}\in\tau_{1}}k_{1}(x_{it_{1}})k_{2}(x_{it_{2}})\Big)
\end{align*}

The remaining results simply follow from their definitions.
\end{proof}
\begin{lemma}
\label{lem:half:all}Suppose that Condition \ref{About-k} holds.
Then for all $\check{W}_{T,(m_{1},\dots,m_{L})}$ for which $\sum_{l}m_{l}\leq6$,
we have
\[
\check{W}_{T,(m_{1},\dots,m_{L})}=O_{p}(1)
\]
\end{lemma}
\begin{proof}
Since $W_{T,(m_{1},\dots,m_{L})}=O_{p}(1)$ holds for any $T$, the
two half-sample versions of the statistic are also bounded in probability.
The result then follows by the definition of $\check{W}_{T,(m_{1},\dots,m_{L})}$.
\end{proof}

\section{Expansions for panel data models}

This section establishes the expansions that are used in the panel
data results of the main paper. As described in the main appendix, the
expansions are generated from a Taylor expansion of $\theta(F(\epsilon))$
, where $F(\epsilon)=F+\epsilon\sqrt{T}(\hat{F}-F)$ for $\epsilon\in[0,T^{-1/2}]$,
and $\alpha_{i}(\theta,F_{i}(\epsilon))$ and $\theta(F(\epsilon))$
are the solutions to the moment equations
\begin{align*}
0 & =\int V\big(z_{it};\theta,\alpha_{i}(\theta,F_{i}(\epsilon))\big)dF_{i}(\epsilon),\quad\text{for all }i\\
0 & =\sum_{i=1}^{n}\int U\Big(z_{it};\theta(F(\epsilon)),\alpha_{i}\big(\theta(F(\epsilon)),F_{i}(\epsilon)\big)\Big)dF_{i}(\epsilon)
\end{align*}

This process is described in detail in Hahn and Newey (2004), and
the precise form of a number of terms are derived in the supplementary
appendix to that paper. This gives an expansion of the form 
\begin{align*}
\sqrt{nT}(\hat{\theta}-\theta_{0}) & =\sqrt{n}\theta^{\epsilon}(0)+\frac{1}{2}\frac{\sqrt{n}}{\sqrt{T}}\theta^{\epsilon\epsilon}(0)+\frac{1}{6}\frac{\sqrt{n}}{T}\theta^{\epsilon\epsilon\epsilon}(0)+\frac{1}{24}\frac{\sqrt{n}}{T\sqrt{T}}\theta^{4\epsilon}(0)\\
 & +\frac{1}{120}\frac{\sqrt{n}}{T^{2}}\theta^{5\epsilon}(0)+\frac{1}{720}\frac{\sqrt{n}}{T^{5/2}}\theta^{6\epsilon}(0)+\frac{1}{5040}\frac{\sqrt{n}}{T^{3}}\theta^{7\epsilon}(\tilde{\epsilon})
\end{align*}

for some $\tilde{\epsilon}\in[0,T^{-1/2}]$. The key feature of this
expansion for our purposes, is that the $m$-th order term in the
expansion is made of normalized V-statistics of order $m$. We study
the properties of the expansion by looking at the properties of these
V-statistics.

Write the $m$-th expansion term as $\theta^{m\epsilon}(0)$. From
Lemma \ref{WitOp(1)}, we have that (for $m\leq6$) $\theta^{m\epsilon}(0)=O_{p}(1)$
and so,
\begin{align*}
\frac{\sqrt{n}}{T\sqrt{T}}\theta^{4\epsilon}(0) & =O_{p}(T^{-1})\\
\frac{\sqrt{n}}{T^{2}}\theta^{5\epsilon}(0) & =o_{p}(T^{-1})\\
\frac{\sqrt{n}}{T^{5/2}}\theta^{6\epsilon}(0) & =o_{p}(T^{-1})
\end{align*}

The next set of results extends two existing results in Hahn and Newey
(2004) which bound $\theta^{5\epsilon}(\tilde{\epsilon})$, allowing
us to bound the remainder term $\theta^{7\epsilon}(\tilde{\epsilon})$.
\begin{lemma}
\label{lem:lem5_new}\textbf{ }Suppose that, for each $i$, $\{\xi_{it}(\phi),t=1,2,\dots\}$
is a sequence of zero mean i.i.d. variables indexed by some parameter
$\phi\in\Phi$. We assume that the $\xi_{it}(\phi)$ are independent
across $i$, and that $\sup_{\phi\in\Phi}|\xi_{it}(\phi)|\leq B_{it}$
for some sequence of variables $B_{it}$ that is i.i.d. across $t$
and independent across $i$., with $\max_{i}E[|B_{it}|^{64}]<\infty$,
and $n=o(T^{3})$. Then
\[
Pr\Big[\max_{i}\Big|\frac{1}{\sqrt{T}}\sum_{t=1}^{T}\xi_{it}(\phi_{i})\Big|>T^{\frac{1}{15}-\nu}\Big]=o(T^{-1})
\]

for every $0\leq v<\frac{1}{160}$, with $\{\phi_{i}\}$ some arbitrary
sequence in $\Phi$.
\end{lemma}
\begin{proof}
We have
\begin{align*}
Pr\Big[\sup_{\phi\in\Phi} & \Big|\frac{1}{\sqrt{T}}\sum_{t=1}^{T}\xi_{it}(\phi_{i})\Big|>T^{\frac{1}{15}-\nu}\Big]=Pr\Big[\sup_{\phi\in\Phi}\Big|\sum_{t=1}^{T}\xi_{it}(\phi_{i})\Big|>T^{\frac{17}{30}-\nu}\Big]\\
 & \leq\frac{E\Big[\sup_{\phi\in\Phi}\Big|\sum_{t=1}^{T}\xi_{it}(\phi_{i})\Big|^{64}\Big]}{T^{64\times\frac{17}{30}-64\nu}}\\
 & =\frac{\sup_{\phi\in\Phi}E\Big[\Big|\sum_{t=1}^{T}\xi_{it}(\phi_{i})\Big|^{64}\Big]}{T^{64\times\frac{17}{30}-64\nu}}\\
 & \leq\frac{\sup_{\phi\in\Phi}CT^{31}E\Big[\xi_{it}(\phi_{i})^{64}\Big]}{T^{\frac{1088}{30}-64\nu}}\\
 & \leq\frac{CT^{31}E\Big[B_{it}^{64}\Big]}{T^{\frac{1088}{30}-64\nu}}
\end{align*}

for some constant $C>0$. The first inequality above is the Markov
inequality, while the second to last inequality comes from Loeve (1977,
p. 276). Therefore, we have
\begin{align*}
TPr\Big[\max_{i}\Big|\frac{1}{\sqrt{T}}\sum_{t=1}^{T}\xi_{it}(\phi_{i})\Big|>T^{\frac{1}{15}-\nu}\Big] & \leq T\sum_{i=1}^{n}Pr\Big[\Big|\frac{1}{\sqrt{T}}\sum_{t=1}^{T}\xi_{it}(\phi_{i})\Big|>T^{\frac{1}{15}-\nu}\Big]\\
 & =o(T^{4})\frac{CT^{31}}{T^{\frac{1088}{30}-64\nu}}\\
 & =o(T^{-\frac{38}{30}+64\nu})\\
 & =o(1)
\end{align*}

for $0\leq\nu<\frac{1}{240}$.
\end{proof}
\begin{lemma}
\label{thm:7bound}Suppose that Conditions 1, 2, 3, and 4 hold. Then
we have
\begin{align*}
Pr\Big[\max_{0\leq\epsilon\leq\frac{1}{\sqrt{T}}}\big|\theta^{\epsilon\epsilon\epsilon\epsilon\epsilon\epsilon}(\epsilon)\big| & >C\big(T^{\frac{1}{15}-\nu}\big)^{6}\Big]=o(T^{-1})\\
Pr\Big[\max_{0\leq\epsilon\leq\frac{1}{\sqrt{T}}}\big|\theta^{\epsilon\epsilon\epsilon\epsilon\epsilon\epsilon\epsilon}(\epsilon)\big| & >C\big(T^{\frac{1}{15}-\nu}\big)^{7}\Big]=o(T^{-1})
\end{align*}
for some constant $C$ and $0\leq\nu<\frac{1}{240}$.
\end{lemma}
\begin{proof}
The proof is virtually identical to the related proofs in Hahn and
Newey (2004), with Lemma \ref{lem:lem5_new} replacing their Lemma
5.
\end{proof}
From Theorem \ref{thm:7bound}, we have that 
\[
Pr\big(\big|\frac{1}{\sqrt{T}}\theta^{7\epsilon}(\tilde{\epsilon})\big|>C\big)=o(1)
\]

so that $\frac{\sqrt{n}}{T^{3}}\theta^{7\epsilon}(\tilde{\epsilon})=\frac{\sqrt{n}}{\sqrt{T}}\frac{1}{T^{2}}\frac{1}{\sqrt{T}}\theta^{7\epsilon}(\tilde{\epsilon})=o(T^{-1})$.
We now have
\begin{align*}
\sqrt{nT}(\hat{\theta}-\theta_{0}) & =\sqrt{n}\theta^{\epsilon}(0)+\frac{1}{2}\frac{\sqrt{n}}{\sqrt{T}}\theta^{\epsilon\epsilon}(0)+\frac{1}{6}\frac{\sqrt{n}}{T}\theta^{\epsilon\epsilon\epsilon}(0)+\frac{1}{24}\frac{\sqrt{n}}{T\sqrt{T}}\theta^{4\epsilon}(0)+o_{p}(T^{-1})\\
 & =\sqrt{n}\theta^{\epsilon}(0)+\frac{1}{2}\frac{\sqrt{n}}{\sqrt{T}}\theta^{\epsilon\epsilon}(0)+\frac{1}{6}\frac{\sqrt{n}}{T}\theta^{\epsilon\epsilon\epsilon}(0)+O_{p}(T^{-1})
\end{align*}

which is the expansion given in the main appendix. Note that, by Lemma
\ref{WtVOp()}, we have $Var(\sqrt{n}\theta^{\epsilon}(0))=O(1)$,
$Var(\sqrt{n}\theta^{\epsilon\epsilon}(0))=O(1)$, and $Var(\sqrt{n}\theta^{\epsilon\epsilon\epsilon}(0))=O(1)$.

\subsection{Jackknife and split-sample expansions}

We now show that the same expansion applies to the jackknife and split-sample
bias-corrected estimators. Let $\tilde{\theta}_{J}^{m\epsilon}(0)$
and $\tilde{\theta}_{1/2}^{m\epsilon}$ be the $m$-th order terms
in the expansions for the jackknife and split-sample estimators. Specifically,
\[
T^{-m/2}\tilde{\theta}_{J}^{m\epsilon}(0)=T^{-m/2}\theta^{m\epsilon}(0)-\frac{1}{T}\sum_{t=1}^{T}(T-1)^{-m/2}\theta_{(t)}^{m\epsilon}(0)
\]

where $\theta_{(t)}^{m\epsilon}(0)$ is the $m$-th order term for
the estimator that excludes period $t$. Similarly,
\[
T^{-m/2}\tilde{\theta}_{1/2}^{m\epsilon}(0)=2T^{-m/2}\theta^{m\epsilon}(0)-\frac{1}{2}(T/2)^{-m/2}\big(\theta_{1}^{m\epsilon}(0)+\theta_{2}^{m\epsilon}(0)\big)
\]

where $\theta_{1}^{m\epsilon}(0)$ and $\theta_{2}^{m\epsilon}(0)$
are the $m$-th order terms for the estimators that use only the first
and second halves of time periods.

As for the MLE expansion, the results of Sections \ref{sec:V-jack}
and \ref{sec:V-split} show that that $\tilde{\theta}_{J}^{m\epsilon}$
and $\tilde{\theta}_{1/2}^{m\epsilon}$ are both $O_{p}(1)$ for $m\leq6$.
It remains to establish that the remainder terms are negligible. Lemma
\ref{thm:7bound} immediately gives so that $\frac{\sqrt{n}}{T^{3}}\theta_{1}^{7\epsilon}(\tilde{\epsilon}_{1})=o(T^{-1})$
and $\frac{\sqrt{n}}{T^{3}}\theta_{2}^{7\epsilon}(\tilde{\epsilon}_{2})=o(T^{-1})$.
For the jackknife term
\begin{align*}
Pr\Big[\max_{t}\Big|\frac{1}{\sqrt{T}}\theta_{(t)}^{7\epsilon}(\tilde{\epsilon}_{(t)})\Big|>\eta\Big] & \leq\sum_{t}Pr\Big[\Big|\frac{1}{\sqrt{T}}\theta_{(t)}^{7\epsilon}(\tilde{\epsilon}_{(t)})\Big|>\eta\Big]\\
 & =T\cdot o(T^{-1})=o(1)
\end{align*}

which implies

\begin{align*}
T\times\Big|\frac{1}{T}\sum_{t=1}^{T}\sqrt{nT}\frac{1}{(T-1)^{7/2}}\theta_{(t)}^{\epsilon\epsilon\epsilon\epsilon\epsilon\epsilon}(\tilde{\epsilon}_{(t)})\Big| & \leq C\sqrt{nT}\frac{T\sqrt{T}}{(T-1)^{7/2}}\max_{t}\Big|\frac{1}{\sqrt{T}}\theta_{(t)}^{\epsilon\epsilon\epsilon\epsilon\epsilon\epsilon}(\tilde{\epsilon}_{(t)})\Big|\\
 & =\sqrt{\frac{n}{T}}\frac{T^{5/2}}{(T-1)^{7/2}}O_{p}(1)\\
 & =o_{p}(1)
\end{align*}

Therefore,
\begin{align*}
\sqrt{nT}T^{-7/2}\tilde{\theta}_{J}^{7\epsilon}(0) & =\sqrt{\frac{n}{T}}T^{-2}\frac{1}{\sqrt{T}}\theta^{7\epsilon}(0)-\frac{1}{T}\sum_{t=1}^{T}\sqrt{nT}(T-1)^{-m/2}\theta_{(t)}^{m\epsilon}(0)\\
 & =o_{p}(T^{-1})
\end{align*}

Given this, we may write the expansions used in the paper as 
\begin{align*}
\sqrt{nT}(\tilde{\theta}_{J}-\theta_{0}) & =\sqrt{n}\tilde{\theta}_{J}^{\epsilon}(0)+\frac{1}{2}\frac{\sqrt{n}}{\sqrt{T}}\tilde{\theta}_{J}^{\epsilon\epsilon}(0)+\frac{1}{6}\frac{\sqrt{n}}{T}\tilde{\theta}_{J}^{\epsilon\epsilon\epsilon}(0)+\frac{1}{24}\frac{\sqrt{n}}{T\sqrt{T}}\tilde{\theta}_{J}^{4\epsilon}(0)+o_{p}(T^{-1})\\
 & =\sqrt{n}\tilde{\theta}_{J}^{\epsilon}(0)+\frac{1}{2}\frac{\sqrt{n}}{\sqrt{T}}\tilde{\theta}_{J}^{\epsilon\epsilon}(0)+\frac{1}{6}\frac{\sqrt{n}}{T}\tilde{\theta}_{J}^{\epsilon\epsilon\epsilon}(0)+O_{p}(T^{-1})
\end{align*}

\begin{align*}
\sqrt{nT}(\tilde{\theta}_{1/2}-\theta_{0}) & =\sqrt{n}\tilde{\theta}_{1/2}^{\epsilon}(0)+\frac{1}{2}\frac{\sqrt{n}}{\sqrt{T}}\tilde{\theta}_{1/2}^{\epsilon\epsilon}(0)+\frac{1}{6}\frac{\sqrt{n}}{T}\tilde{\theta}_{1/2}^{\epsilon\epsilon\epsilon}(0)+\frac{1}{24}\frac{\sqrt{n}}{T\sqrt{T}}\tilde{\theta}_{1/2}^{4\epsilon}(0)+o_{p}(T^{-1})\\
 & =\sqrt{n}\tilde{\theta}_{1/2}^{\epsilon}(0)+\frac{1}{2}\frac{\sqrt{n}}{\sqrt{T}}\tilde{\theta}_{1/2}^{\epsilon\epsilon}(0)+\frac{1}{6}\frac{\sqrt{n}}{T}\tilde{\theta}_{1/2}^{\epsilon\epsilon\epsilon}(0)+O_{p}(T^{-1})
\end{align*}

The proof of Proposition 8 establishes: $Var(\sqrt{n}\tilde{\theta}_{J}^{\epsilon}(0))=O(1)$,
$Var(\sqrt{n}\tilde{\theta}_{J}^{\epsilon\epsilon}(0))=O(1)$, and
$Var(\sqrt{n}\tilde{\theta}_{J}^{\epsilon\epsilon\epsilon}(0))=O(1)$,
as well as $Var(\sqrt{n}\tilde{\theta}_{1/2}^{\epsilon}(0))=O(1)$,
$Var(\sqrt{n}\tilde{\theta}_{1/2}^{\epsilon\epsilon}(0))=O(1)$, and
$Var(\sqrt{n}\tilde{\theta}_{1/2}^{\epsilon\epsilon\epsilon}(0))=O(1)$.

\section{Additional results for the proof of Proposition 9}

For completion of the proof of Proposition 9, we show here that $\frac{\sqrt{n}}{T}(C_{n}-\tilde{C}_{n,J})=O_{p}(T^{-1})$
and $\frac{\sqrt{n}}{T}(C_{n}-\tilde{C}_{n,1/2})=O_{p}(T^{-1/2})$,
where, following the notation of this appendix, $C_{n}=\frac{1}{6}\theta^{\epsilon\epsilon\epsilon}(0)$,
$\tilde{C}_{n,J}=\frac{1}{6}\tilde{\theta}_{J}^{\epsilon\epsilon\epsilon}(0)$,
and $\tilde{C}_{n,1/2}=\frac{1}{6}\tilde{\theta}_{1/2}^{\epsilon\epsilon\epsilon}(0)$.

For the jackknife case, we have from Lemma \ref{WiT_jack} that
\begin{align*}
\frac{\sqrt{n}}{T}\big(\tilde{W}_{T,(3)}-W_{T,(3)}\big) & =\frac{1}{\sqrt{n}}\sum_{i}\Big\{\frac{1-2T}{T^{5/2}(T-1)}\sum_{t}k_{1}(x_{i,t})k_{2}(x_{i,t})k_{3}(x_{i,t})\\
 & \quad+\frac{T^{2}+3T-1}{T^{5/2}(T-1)^{2}}\sum_{t_{1}\ne t_{2}}\big(k_{1}(x_{i,t_{1}})k_{2}(x_{i,t_{2}})k_{3}(x_{i,t_{2}})\\
 & \quad+k_{1}(x_{i,t_{2}})k_{2}(x_{i,t_{1}})k_{3}(x_{i,t_{2}})+k_{1}(x_{i,t_{2}})k_{2}(x_{i,t_{2}})k_{3}(x_{i,t_{1}})\big)\\
 & \quad+\frac{3T-1}{T^{5/2}(T-1)^{2}}\sum_{t_{1}\neq t_{2}\ne t_{3}}k_{1}(x_{i,t_{1}})k_{2}(x_{i,t_{2}})k_{3}(x_{i,t_{3}})\Big\}
\end{align*}

The first term is $\sqrt{\frac{n}{T}}\frac{1-2T}{T(T-1)}\frac{1}{nT}\sum_{i}\sum_{t}k_{1}(x_{i,t})k_{2}(x_{i,t})k_{3}(x_{i,t})=O_{p}(T^{-1})$.
The second term is
\[
\frac{T^{2}+3T-1}{T(T-1)^{2}}\frac{1}{\sqrt{n}}\sum_{i}\frac{1}{\sqrt{T}}\sum_{t_{1}}k_{1}(x_{i,t_{1}})\frac{1}{T}\sum_{t_{2}\ne t_{1}}k_{2}(x_{i,t_{2}})k_{3}(x_{i,t_{2}})=O_{p}(T^{-1})
\]
while the final term is
\[
\frac{3T-1}{T(T-1)^{2}}\frac{1}{\sqrt{n}}\sum_{i}\frac{1}{T\sqrt{T}}\sum_{t_{1}\neq t_{2}\ne t_{3}}k_{1}(x_{i,t_{1}})k_{2}(x_{i,t_{2}})k_{3}(x_{i,t_{3}})=o_{p}(T^{-1})
\]

Hence we have $\frac{\sqrt{n}}{T}\big(\tilde{W}_{T,(3)}-W_{T,(3)}\big)=O_{p}(T^{-1})$.
Similar results for $\tilde{W}_{T,(2,1)}$ and $\tilde{W}_{T,(1,1,1)}$
are easily shown, giving $\frac{\sqrt{n}}{T}(C_{n}-\tilde{C}_{n,J})=O_{p}(T^{-1})$
as required.

In the split-sample case, Lemma \ref{lem:Wit_half} gives
\begin{align*}
\frac{\sqrt{n}}{T}\big(\check{W}_{T,(3)}-W_{T,(3)}\big) & =\frac{\sqrt{n}}{T}W_{T,(3)}-4\frac{\sqrt{n}}{T}\Big(\frac{1}{n}\sum_{i}\frac{1}{T\sqrt{T}}\sum_{t_{1},t_{2},t_{3}\in\tau_{1}}k_{1}(x_{it_{1}})k_{2}(x_{it_{2}})k_{3}(x_{it_{3}})\\
 & +\frac{1}{T\sqrt{T}}\sum_{t_{1},t_{2},t_{3}\in\tau_{2}}k_{1}(x_{it_{1}})k_{2}(x_{it_{2}})k_{3}(x_{it_{3}})\Big)\\
 & =\frac{1}{\sqrt{T}}\sqrt{\frac{n}{T}}\Big(W_{T,(3)}-4\big(W_{T,(3)}^{(1)}+W_{T,(3)}^{(2)}\big)\Big)\\
 & =O_{p}(T^{-1/2})
\end{align*}

where the final line follows form $W_{T,(3)}=O_{p}(1)$, $W_{T,(3)}^{(1)}=O_{p}(1)$,
and $W_{T,(3)}^{(2)}=O_{p}(1)$.

\section{Higher-order Bias}

Here we demonstrate the effect of the two bias corrections on the
$O(T^{-2})$ bias of the estimators, which remains even after bias
correction. To clarify this, it is useful to extend the expansion
in the main paper to include a further term 
\[
\sqrt{nT}(\hat{\theta}-\theta_{0})=\sqrt{n}A_{n}+\frac{\sqrt{n}}{\sqrt{T}}B_{n}+\frac{\sqrt{n}}{T}C_{n}+\frac{\sqrt{n}}{T\sqrt{T}}D_{n}+o_{p}(T^{-1})
\]
 where again, $Var(\sqrt{n}D_{n})=O(1)$. Using this expansion, we
can define \\ $\frac{1}{T}\sqrt{\rho}\mathbf{D}=\frac{1}{T}\sqrt{\rho}\lim_{n\to\infty}\{\sqrt{T}E[C_{n}]+E[D_{n}]\}$
as the $O(T^{-2})$ bias. While both bias-corrections remove the asymptotic
bias of the MLE, they do not remove the $O(T^{-2})$ bias, and in
fact will inflate this term. The following theorem establishes the
higher-order biases of the two bias-corrected estimators, relative
to that of the MLE $\hat{\theta}$.
\begin{theorem}
Let $\tilde{\theta}_{J}$ and $\tilde{\theta}_{1/2}$ be the jackknife
and split-sample bias-corrected estimates. The $O(T^{-2})$ biases
of these estimators are $\frac{1}{T}\sqrt{\rho}\mathbf{D}_{J}$ and
$\frac{1}{T}\sqrt{\rho}\mathbf{D}_{1/2}$ respectively, where 
\begin{align*}
\mathbf{D}_{J} & =\lim_{n\to\infty}\{\sqrt{T}E[\tilde{C}_{n,J}]+E[\tilde{D}_{n,J}]\}=-\frac{T}{T-1}\mathbf{D}+o(1)\\
\mathbf{D}_{1/2} & =\lim_{n\to\infty}\{\sqrt{T}E[\tilde{C}_{n,1/2}]+E[\tilde{D}_{n,1/2}]\}=-2\mathbf{D}+o(1)
\end{align*}
\end{theorem}
The jackknife correction inflates the $O(T^{-2})$ bias by a factor
of $\frac{T}{T-1}$, while in the split-sample case the bias is inflated
by a factor of 2, leading to strictly larger bias for $T>2$. As shown
by Dhaene and Jochmans (2015), this factor of 2 is the smallest inflation
of higher-order bias among a variety of forms of split-sample correction,
for example, splits into more than two subsets and/or splits into
subsets of unequal length. It can also be shown that the general leave-$k$-out
jackknife inflates the $O(T^{-2})$ bias by a factor of $\frac{T}{T-k}$,
so that the choice $k=1$ has the smallest higher-order bias among
this class. It is of course possible to remove this $O(T^{-2})$ bias
also. For example, the average of all leave-$k$-out and all leave-$j$-out
estimates (with $j<k$) can be combined with the MLE to give $\tilde{\theta}_{2J}=\frac{T^{2}}{jk}\hat{\theta}-\frac{(T-j)^{2}}{j(k-j)}\bar{\theta}_{j}+\frac{(T-k)^{2}}{k(k-j)}\bar{\theta}_{k}$,
which has bias of order $O(T^{-3})$. Similarly, Dhaene and Jochmans
(2015) suggest $\tilde{\theta}_{1/(2,3)}=3\hat{\theta}-3\bar{\theta}_{1/2}+\bar{\theta}_{1/3}$,
where $\bar{\theta}_{1/2}$ and $\bar{\theta}_{1/3}$ are the averages
of the two half-sample and three third-sample estimates. The properties
of these higher-order bias corrections, in particular their impact
on the higher-order variance is a subject of future research.

\subsection{Proof of Theorem}

As previously, the proof proceeds in terms of general V-statistics
that match the order of those in the relevant expansion terms. We
begin by deriving the form of $E[C_{n}]$ and $E[D_{n}]$, by taking
the expectations of third and fourth-order V-statistics. For a third-order
statistic we have 
\begin{align*}
E\big[\sqrt{T}W_{T,(3)}\big] & =\frac{1}{n}\sum_{i}\frac{1}{T}\sum_{t_{1},t_{2},t_{3}}E\big[k_{1}(x_{i,t_{1}})k_{2}(x_{i,t_{2}})k_{3}(x_{i,t_{3}})\big]\\
 & =\frac{1}{n}\sum_{i}E\big[k_{1}(x_{i,t})k_{2}(x_{i,t})k_{3}(x_{i,t})\big]E\big[\sqrt{T}W_{T,(2,1)}\big]\\
 & =\frac{1}{n^{2}}\sum_{i}E\big[k_{1}(x_{i,t})k_{2}(x_{i,t})k_{3}(x_{i,t})\big]E\big[\sqrt{T}W_{T,(1,1,1)}\big]\\
 & =\frac{1}{n^{3}}\sum_{i}E\big[k_{1}(x_{i,t})k_{2}(x_{i,t})k_{3}(x_{i,t})\big]
\end{align*}
 Clearly $\lim_{n\to\infty}E\big[W_{T,(2,1)}\big]=\lim_{n\to\infty}E\big[W_{T,(1,1,1)}\big]=0$
and we may ignore these terms. From Lemmas 1 and 2 in Section B3 of
the Appendix, this is true of the jackknife and split-sample versions
of these statistics as well, so that we need only focus on $W_{T,(3)}$
and its bias-corrected counterparts. Computing these expectations
gives 
\begin{align*}
E\big[\sqrt{T}\tilde{W}_{T,(3)}\big] & =-\frac{T}{(T-1)}\frac{1}{n}\sum_{i}E\big[k_{1}(x_{i,t})k_{2}(x_{i,t})k_{3}(x_{i,t})\big]\\
 & =-\frac{T}{T-1}E\big[\sqrt{T}W_{T,(3)}\big]E\big[\sqrt{T}\check{W}_{T,(3)}\big]\\
 & =2E\big[\sqrt{T}W_{T,(3)}\big]-4\frac{1}{T}\sum_{t\in\tau_{1}}E[k_{1}(x_{it})k_{2}(x_{it})k_{3}(x_{it})]\\
 &+\sum_{t\in\tau_{2}}E[k_{1}(x_{it})k_{2}(x_{it})k_{3}(x_{it})]\\
 & =-2E\big[\sqrt{T}W_{T,(3)}\big]
\end{align*}
 Next we consider the fourth-order V-statistics in $D_{n}$. Taking
the expectation of $W_{T,(4)}$ gives 
\begin{align*}
E[W_{T,(4)}] & =\frac{1}{n}\sum_{i}\frac{1}{T^{2}}\sum_{t_{1},t_{2},t_{3},t_{4}}E\big[k_{1}(x_{i,t_{1}})k_{2}(x_{i,t_{1}})k_{3}(x_{i,t_{2}})k_{4}(x_{i,t_{4}})\big]\\
 & =\frac{1}{n}\sum_{i}\frac{1}{T^{2}}\sum_{t}E\big[k_{1}(x_{i,t})k_{2}(x_{i,t})k_{3}(x_{i,t})k_{4}(x_{i,t})\big]\\
 & \quad+\frac{1}{n}\sum_{i}\frac{1}{T^{2}}\sum_{t_{1}\ne t_{2}}E\Big[k_{1}(x_{i,t_{1}})k_{2}(x_{i,t_{1}})k_{3}(x_{i,t_{2}})k_{4}(x_{i,t_{2}})\\
 & \quad+k_{1}(x_{i,t_{1}})k_{2}(x_{i,t_{2}})k_{3}(x_{i,t_{1}})k_{4}(x_{i,t_{2}})\\
 & \quad+k_{1}(x_{i,t_{1}})k_{2}(x_{i,t_{2}})k_{3}(x_{i,t_{2}})k_{4}(x_{i,t_{1}})\Big]\\
 & =\frac{T-1}{T}\frac{1}{n}\sum_{i}\Big(E[k_{1}(x_{i,t})k_{2}(x_{i,t})]E[k_{3}(x_{i,t})k_{4}(x_{i,t})]\\
 & \quad+E[k_{1}(x_{i,t})k_{3}(x_{i,t})]E[k_{2}(x_{i,t})k_{4}(x_{i,t})]\\
 & \quad+E[k_{1}(x_{i,t})k_{4}(x_{i,t})]E[k_{2}(x_{i,t})k_{3}(x_{i,t})]\Big)+o(1)
\end{align*}
 It has been shown that the fourth-order jackknifed V-statistic has
the form 
\begin{align*}
\tilde{W}_{i,T,4} & =\frac{-2T+1}{T(T-1)^{2}}\sum_{t}k_{1}(x_{i,t})k_{2}(x_{i,t})k_{3}(x_{i,t})k_{4}(x_{i,t})\\
 & -\frac{T^{2}-3T+1}{T(T-1)^{3}}\sum_{t_{1}\ne t_{2}}\Big(k_{1}(x_{i,t_{1}})k_{2}(x_{i,t_{2}})k_{3}(x_{i,t_{2}})k_{4}(x_{i,t_{2}})\\
 & +k_{1}(x_{i,t_{2}})k_{2}(x_{i,t_{1}})k_{3}(x_{i,t_{2}})k_{4}(x_{i,t_{2}})+k_{1}(x_{i,t_{2}})k_{2}(x_{i,t_{2}})k_{3}(x_{i,t_{1}})k_{4}(x_{i,t_{2}})\\
 & +k_{1}(x_{i,t_{2}})k_{2}(x_{i,t_{2}})k_{3}(x_{i,t_{2}})k_{4}(x_{i,t_{1}})\Big)\\
 & -\frac{T^{2}-3T+1}{T(T-1)^{3}}\sum_{t_{1}\ne t_{2}}\Big(k_{1}(x_{i,t_{1}})k_{2}(x_{i,t_{1}})k_{3}(x_{i,t_{2}})k_{4}(x_{i,t_{2}})\\
 & +k_{1}(x_{i,t_{1}})k_{2}(x_{i,t_{2}})k_{3}(x_{i,t_{1}})k_{4}(x_{i,t_{2}})+k_{1}(x_{i,t_{1}})k_{2}(x_{i,t_{2}})k_{3}(x_{i,t_{2}})k_{4}(x_{i,t_{1}})\Big)\\
 & +\frac{3T-1}{T(T-1)^{3}}\sum_{t_{1}\neq t_{2}\ne t_{3}}\Big(k_{1}(x_{i,t_{1}})k_{2}(x_{i,t_{1}})k_{3}(x_{i,t_{2}})k_{4}(x_{i,t_{3}})\\
 & +k_{1}(x_{i,t_{1}})k_{2}(x_{i,t_{2}})k_{3}(x_{i,t_{1}})k_{4}(x_{i,t_{3}})+k_{1}(x_{i,t_{1}})k_{2}(x_{i,t_{2}})k_{3}(x_{i,t_{3}})k_{4}(x_{i,t_{1}})\\
 & +k_{1}(x_{i,t_{2}})k_{2}(x_{i,t_{1}})k_{3}(x_{i,t_{1}})k_{4}(x_{i,t_{3}})+k_{1}(x_{i,t_{2}})k_{2}(x_{i,t_{1}})k_{3}(x_{i,t_{3}})k_{4}(x_{i,t_{1}})\\
 & +k_{1}(x_{i,t_{2}})k_{2}(x_{i,t_{3}})k_{3}(x_{i,t_{1}})k_{4}(x_{i,t_{1}})\Big)\\
 & +\frac{T^{2}+3T-1}{T(T-1)^{3}}\sum_{t_{1}\neq t_{2}\ne t_{3}\ne t_{4}}k_{1}(x_{i,t_{1}})k_{2}(x_{i,t_{1}})k_{3}(x_{i,t_{2}})k_{4}(x_{i,t_{4}})
\end{align*}
 Taking expectations gives 
\begin{align*}
E[\tilde{W}_{T,4}] & =-\frac{T^{2}-3T+1}{(T-1)^{2}}\frac{1}{n}\sum_{i}\Big(E[k_{1}(x_{i,t})k_{2}(x_{i,t})]E[k_{3}(x_{i,t})k_{4}(x_{i,t})]\\
 & +E[k_{1}(x_{i,t})k_{3}(x_{i,t})]E[k_{2}(x_{i,t})k_{4}(x_{i,t})]\\
 & +E[k_{1}(x_{i,t})k_{4}(x_{i,t})]E[k_{2}(x_{i,t})k_{3}(x_{i,t})]\Big)+o(1)\\
 & =-\frac{T}{T-1}E[W_{T,4}]+o(1)
\end{align*}
 The fourth order term $\tilde{D}_{n,J}$ also contains fourth-order
V-statistics of the form $\tilde{W}_{T,(3,1)}$, $\tilde{W}_{T,(2,2)}$,
$\tilde{W}_{T,(2,1,1)}$, and $\tilde{W}_{T,(1,1,1,1)}$, and the
same relationship holds for these terms.

The fourth-order split-sample statistic has the form 
\begin{align*}
\check{W}_{i,T,4} & =2W_{i,T,4}\\
 & -8\Big(\frac{1}{T^{2}}\sum_{t_{1},t_{2},t_{3},t_{4}\in\tau_{1}}k_{1}(x_{it_{1}})k_{2}(x_{it_{2}})k_{3}(x_{it_{3}})k_{4}(x_{it_{4}})\\
 & +\frac{1}{T^{2}}\sum_{t_{1},t_{2},t_{3},t_{4}\in\tau_{2}}k_{1}(x_{it_{1}})k_{2}(x_{it_{2}})k_{3}(x_{it_{3}})k_{4}(x_{it_{4}})\Big)
\end{align*}
 and taking expectations it is straightforward to see that 
\[
E[\check{W}_{T,(4)}]=2E[W_{T,(4)}]=-2E[W_{T,(4)}]+o(1)
\]
 As before, identical results can be shown for the other forms of
fourth-order statistics. Combining these results gives the statement
in the theorem.

\section{Derivation of bias for average effects}

As with the expansions for other parameters we first let $F_{i}(\epsilon)=F_{i}+\epsilon\sqrt{T}(\hat{F}_{i}-F_{i})=F_{i}+\epsilon\Delta_{i}$,
and then consider an expansion of $\mu(\epsilon)$ (for $\epsilon=1/\sqrt{T}$)
around $\epsilon=0$.
\[
\mu(\epsilon)=\mu_{0}+\frac{1}{\sqrt{T}}\frac{\partial\mu(0)}{\partial\epsilon}+\frac{1}{2}\frac{1}{T}\frac{\partial^{2}\mu(0)}{\partial\epsilon^{2}}+o_{p}(\frac{1}{T})
\]

Using the definition $\mu(\epsilon)=\frac{1}{n}\sum_{i}\int m\big(\theta(\epsilon),\alpha_{i}(\theta(\epsilon),\epsilon)\big)dF_{i}(\epsilon)$,
and taking two derivatives with respect to $\epsilon$ gives
\begin{align}
\mu^{\epsilon}=\frac{\partial\mu(0)}{\partial\epsilon} & =\frac{1}{n}\sum_{i}\int\frac{\partial m_{it}}{\partial\epsilon}dF_{i}+\frac{1}{n}\sum_{i}\int m_{it}d\Delta_{i}\nonumber \\
\mu^{\epsilon\epsilon}=\frac{\partial^{2}\mu(0)}{\partial\epsilon^{2}} & =\frac{1}{n}\sum_{i}\int\frac{\partial^{2}m_{it}}{\partial\epsilon^{2}}dF_{i}+\frac{2}{n}\sum_{i}\int\frac{\partial m_{it}}{\partial\epsilon}d\Delta_{i}\label{eq:mu_eps}
\end{align}

The first expression is given by

\begin{align*}
\mu^{\epsilon}=\frac{\partial\mu(0)}{\partial\epsilon} & =\frac{1}{n}\sum_{i}\int\frac{\partial m_{it}}{\partial\theta'}dF_{i}\theta^{\epsilon}(0)+\frac{1}{n}\sum_{i}\int\frac{\partial m_{it}}{\partial\alpha_{i}}dF_{i}\alpha_{i}^{\theta'}\theta^{\epsilon}(0)\\
 & \quad+\frac{1}{n}\sum_{i}\int\frac{\partial m_{it}}{\partial\alpha_{i}}\alpha_{i}^{\epsilon}(0)dF_{i}+\frac{1}{n}\sum_{i}\frac{1}{\sqrt{T}}\sum_{t}\tilde{m}_{it}\\
 & =\frac{1}{n}\sum_{i}\Big\{\big(E[m_{it}^{\theta'}]+E[m_{it}^{\alpha_{i}}]\alpha_{i}^{\theta'}\big)\theta^{\epsilon}+E[m_{it}^{\alpha_{i}}]\alpha_{i}^{\epsilon}+\frac{1}{\sqrt{T}}\sum_{t}\tilde{m}_{it}\Big\}
\end{align*}

with $m_{it}=m(w,z_{it},\theta_{0},\alpha_{i}(\theta_{0}))$ and $\tilde{m}_{it}=m_{it}-E[m_{it}]$.

Taking a second derivative of $m$ gives

\begin{align*}
\frac{1}{n}\sum_{i}\int\frac{\partial^{2}m_{it}}{\partial\epsilon^{2}}dF_{i} & =\frac{1}{n}\sum_{i}\Big\{(\theta^{\epsilon})'E[m^{\theta\theta'}]\theta^{\epsilon}+2E[m^{\alpha_{i}\theta'}]\theta^{\epsilon}\alpha_{i}^{\theta'}\theta^{\epsilon}+2E[m^{\alpha_{i}\theta'}]\alpha_{i}^{\epsilon}\theta^{\epsilon}\\
 & +E[m^{\alpha_{i}\alpha_{i}}](\alpha_{i}^{\theta'}\theta^{\epsilon})^{2}+2E[m^{\alpha_{i}\alpha_{i}}]\alpha_{i}^{\epsilon}\alpha_{i}^{\theta'}\theta^{\epsilon}+E[m^{\alpha_{i}\alpha_{i}}](\alpha_{i}^{\epsilon})^{2}\\
 & +\big(E[m^{\theta'}]+E[m^{\alpha_{i}}]\alpha_{i}^{\theta'}\big)\theta^{\epsilon\epsilon}+E[m^{\alpha_{i}}](\theta^{\epsilon})'\alpha_{i}^{\theta\theta'}\theta^{\epsilon}+2E[m^{\alpha_{i}}]\alpha_{i}^{\epsilon\theta'}\theta^{\epsilon}\\
 & +E[m^{\alpha_{i}}]\alpha_{i}^{\epsilon\epsilon}\Big\}\\
 & =\frac{1}{n}\sum_{i}\Big\{(\theta^{\epsilon})'\big(E[m^{\theta\theta'}]+E[m^{\alpha_{i}}]\alpha_{i}^{\theta\theta'}\big)\theta^{\epsilon}\\
 & +2E[m^{\alpha_{i}\theta'}]\theta^{\epsilon}\alpha_{i}^{\theta'}\theta^{\epsilon}+2\big(E[m^{\alpha_{i}\theta'}]+E[m^{\alpha_{i}\alpha_{i}}]\alpha_{i}^{\theta'}\big)\alpha_{i}^{\epsilon}\theta^{\epsilon}\\
 & +E[m^{\alpha_{i}\alpha_{i}}](\alpha_{i}^{\theta'}\theta^{\epsilon})^{2}+E[m^{\alpha_{i}\alpha_{i}}]\alpha_{i}^{\epsilon}+E[m^{\alpha_{i}}]\alpha_{i}^{\epsilon\epsilon}\\
 & +\big(E[m^{\theta'}]+E[m^{\alpha_{i}}]\alpha_{i}^{\theta'}\big)\theta^{\epsilon\epsilon}+2E[m^{\alpha_{i}}]\alpha_{i}^{\epsilon\theta'}\theta^{\epsilon}\Big\}
\end{align*}

It is straightforward to then show that

\begin{align*}
\frac{1}{n}\sum_{i}E[m^{\alpha_{i}\alpha_{i}}](\alpha_{i}^{\epsilon})^{2} & =O_{p}(1)\\
\frac{1}{n}\sum_{i}\big(E[m^{\theta'}]+E[m^{\alpha_{i}}]\alpha_{i}^{\theta'}\big)\theta^{\epsilon\epsilon} & =O_{p}(1)\\
\frac{1}{n}\sum_{i}E[m^{\alpha_{i}}]\alpha_{i}^{\epsilon\epsilon} & =O_{p}(1)
\end{align*}

while the remaining terms are all $O_{p}(n^{-1})$. By differentiating
the first order condition for $\alpha_{i}$ twice with respect to
$\epsilon$ we find that 
\begin{align*}
\alpha_{i}^{\epsilon\epsilon} & =-\frac{1}{E[V_{it}^{\alpha_{i}}]}\Big\{ E[V_{it}^{\alpha\alpha}](\alpha_{i}^{\epsilon})^{2}+\alpha_{i}^{\epsilon}\frac{2}{\sqrt{T}}\sum_{t}\tilde{V}_{it}^{\alpha_{i}}\Big\}=O_{p}(1)
\end{align*}

Finally, we evaluate the remaining term in Equation \ref{eq:mu_eps},
which is (letting $\tilde{m}=m-E[m]$)
\begin{align*}
\frac{1}{n}\sum_{i}\int\frac{\partial m_{it}}{\partial\epsilon}d\Delta_{i} & =\frac{1}{n}\sum_{i}\Big\{\Big(\frac{1}{\sqrt{T}}\sum_{t}(\tilde{m}_{it}^{\theta}+\tilde{m}_{it}^{\alpha_{i}}\alpha_{i}^{\theta'})\Big)\theta^{\epsilon}+\Big(\frac{1}{\sqrt{T}}\sum_{t}\tilde{m}_{it}^{\alpha_{i}}\Big)\alpha_{i}^{\epsilon}\Big\}\\
 & =\frac{1}{n}\sum_{i}\frac{1}{\sqrt{T}}\sum_{t}\tilde{m}_{it}^{\alpha_{i}}\alpha_{i}^{\epsilon}+O_{p}(n^{-1})
\end{align*}

Combining the above results gives
\begin{align*}
\mu^{\epsilon\epsilon} & =\frac{1}{n}\sum_{i}E[m^{\alpha_{i}\alpha_{i}}](\alpha_{i}^{\epsilon})^{2}+\frac{1}{n}\sum_{i}E[m^{\alpha_{i}}]\alpha_{i}^{\epsilon\epsilon}+\frac{1}{n}\sum_{i}\big(E[m^{\theta'}]+E[m^{\alpha_{i}}]\alpha_{i}^{\theta'}\big)\theta^{\epsilon\epsilon}\\
 & \quad+\frac{2}{n}\sum_{i}\Big(\frac{1}{\sqrt{T}}\sum_{t}\tilde{m}_{it}^{\alpha_{i}}\Big)\alpha_{i}^{\epsilon}+O_{p}(n^{-1})
\end{align*}

We can take the expectation of the first and second-order terms in
the expansion to compute the first-order bias of the estimator. Taking
the expectation of the first term gives
\begin{align*}
E[\mu^{\epsilon}] & =\frac{1}{n}\sum_{i}\Big\{\big(E[m_{it}^{\theta'}]+E[m_{it}^{\alpha_{i}}]\alpha_{i}^{\theta'}\big)E[\theta^{\epsilon}]+E[m_{it}^{\alpha_{i}}\alpha_{i}^{\epsilon}]+\frac{1}{\sqrt{T}}\sum_{t}E[\tilde{m}_{it}]\Big\}\\
 & =0
\end{align*}
Looking next at the components of the second term
\begin{align*}
E[(\alpha_{i}^{\epsilon})^{2}] & =\frac{1}{E[V_{it}^{2}]^{2}}\frac{1}{T}\sum_{s,t}E[V_{is}V_{it}]=\frac{1}{E[V_{it}^{2}]}\\
E[\alpha_{i}^{\epsilon\epsilon}] & =\frac{1}{E[V_{it}^{2}]^{2}}\Big(E[V_{it}^{\alpha\alpha}]+2E[V_{it}V_{it}^{\alpha_{i}}]\Big)\\
E[\theta^{\epsilon\epsilon}] & =\big(\frac{1}{n}\sum_{i}\mathcal{I}_{i}\big)^{-1}\frac{1}{n}\sum_{i}\frac{E[V_{it}U_{it}^{\alpha_{i}}]}{E[V_{it}^{2}]}+o(1)\\
E[\alpha_{i}^{\epsilon}\frac{1}{\sqrt{T}}\sum_{t}\tilde{m}_{it}^{\alpha_{i}}] & =\frac{1}{E[V_{it}^{2}]}\frac{1}{T}\sum_{s,t}E[m_{it}^{\alpha_{i}}V_{is}]=-\frac{E[m_{it}^{\alpha_{i}}V_{it}]}{E[V_{it}^{2}]}
\end{align*}

Combining gives
\begin{align*}
E[ & \sqrt{nT}(\hat{\mu}-\mu_{0})]\\
 & =\frac{1}{2}\sqrt{\frac{n}{T}}\frac{1}{n}\sum_{i}\Big\{\frac{E[m_{it}^{\alpha_{i}\alpha_{i}}]}{E[V_{it}^{2}]}+E[m_{it}^{\alpha_{i}}]\frac{E[V_{it}^{\alpha\alpha}]+2E[V_{it}V_{it}^{\alpha_{i}}]}{E[V_{it}^{2}]^{2}}+\frac{2E[m_{it}^{\alpha_{i}}V_{it}]}{E[V_{it}^{2}]}\\
 & +\big(E[m_{it}^{\theta'}]+E[m_{it}^{\alpha_{i}}]\alpha_{i}^{\theta'}\big)\big(\frac{1}{n}\sum_{j}\mathcal{I}_{j}\big)^{-1}\Big(\frac{1}{n}\sum_{j}\frac{E[V_{jt}U_{jt}^{\alpha_{j}}]+\frac{1}{2}E[U_{it}^{\alpha\alpha}]}{E[V_{jt}^{2}]}\Big)\Big\}+o(1)
\end{align*}

The first three terms capture the bias that comes directly from the
estimation of the fixed effects and match the formula for $\Delta(w)$
in Hahn and Newey (2004). The terms in the second line capture the
effect of the bias in the parameter $\theta$ (and hence an indirect
impact of the fixed effects).

\end{document}